\documentclass[acmtog,table]{acmart}
\acmSubmissionID{876}

\usepackage{booktabs} 

\setcitestyle{nosort,square} 

\usepackage{picinpar,moresize,xfrac,graphpap,dcolumn,wrapfig,graphicx}

\microtypesetup{stretch=30,shrink=30}
\DeclareGraphicsExtensions{.png,.pdf,.jpg}
\usepackage{subcaption}
\usepackage{stackengine}

\usepackage{tikz,pgfplots}
\usepackage{tikz-cd}
\usepackage{pgf}
\usepgfplotslibrary{colormaps}
\usetikzlibrary{pgfplots.colormaps}
\usetikzlibrary{arrows,automata,positioning,backgrounds}
\usepackage{rotating}
\pgfplotsset{compat=newest}
\pgfplotsset{plot coordinates/math parser=false}
\newlength\figureheight
\newlength\figurewidth

\usepackage{soul,array,calc,url,ragged2e,graphpap}
\usepackage{booktabs} 
\usepackage{tabularx}
\usepackage{colortbl}
\usepackage{multirow}
\usepackage{hyphenat}
\usepackage{transparent}
\usepackage{enumerate}
\usepackage{mathrsfs}
\usepackage{mdframed}
\usepackage{acro}

\usepackage[figure]{hypcap}
\usepackage{mathtools}
\mathtoolsset{centercolon}
\usepackage{nicefrac}
\usepackage{units}
\usepackage[normalem]{ulem}
\usepackage{cancel}
\usepackage{pbox}
\usepackage{multicol}
\usepackage{cases}
\usepackage[all]{xy}
\usepackage{nicematrix}

\usepackage{adjustbox}
\usepackage{wrapfig}
\AfterEndEnvironment{wrapfigure}{\setlength{\intextsep}{0mm}}

\usepackage{cleveref}
\crefmultiformat{subequation}%
{\edef\crefstripprefixinfo{#1}(#2#1#3}%
{,#2\crefstripprefix{\crefstripprefixinfo}{#1}#3)}%
{,#2\crefstripprefix{\crefstripprefixinfo}{#1}#3}%
{,#2\crefstripprefix{\crefstripprefixinfo}{#1}#3)}

\makeatletter
\DeclareFontFamily{OMX}{MnSymbolE}{}
\DeclareSymbolFont{MnLargeSymbols}{OMX}{MnSymbolE}{m}{n}
\SetSymbolFont{MnLargeSymbols}{bold}{OMX}{MnSymbolE}{b}{n}
\DeclareFontShape{OMX}{MnSymbolE}{m}{n}{
    <-6>  MnSymbolE5
   <6-7>  MnSymbolE6
   <7-8>  MnSymbolE7
   <8-9>  MnSymbolE8
   <9-10> MnSymbolE9
  <10-12> MnSymbolE10
  <12->   MnSymbolE12
}{}
\DeclareFontShape{OMX}{MnSymbolE}{b}{n}{
    <-6>  MnSymbolE-Bold5
   <6-7>  MnSymbolE-Bold6
   <7-8>  MnSymbolE-Bold7
   <8-9>  MnSymbolE-Bold8
   <9-10> MnSymbolE-Bold9
  <10-12> MnSymbolE-Bold10
  <12->   MnSymbolE-Bold12
}{}
\let\llangle\@undefined
\let\rrangle\@undefined
\DeclareMathDelimiter{\llangle}{\mathopen}%
                     {MnLargeSymbols}{'164}{MnLargeSymbols}{'164}
\DeclareMathDelimiter{\rrangle}{\mathclose}%
                     {MnLargeSymbols}{'171}{MnLargeSymbols}{'171}
\makeatother

\usepackage{fancyvrb,listings}
\usepackage{algorithm}
\usepackage{algorithmicx}
\usepackage[noend]{algpseudocode}

\algrenewcommand\alglinenumber[1]{\footnotesize #1:}
\makeatletter  
 \renewcommand{\ALG@name}{\small Algorithm} 
\makeatother 
\algdef{SE}[DOWHILE]{Do}{doWhile}{\algorithmicdo}[1]{\algorithmicwhile\ #1}%

\newtheoremstyle{mine}{3pt}{3pt}{\itshape}{}{\bfseries}{.}{.5em}{}
\theoremstyle{mine}
\newtheorem{theorem}{Theorem}[section]
\newtheorem{lemma}{Lemma}[section]
\newtheorem{proposition}{Proposition}[section]
\newtheorem{corollary}{Corollary}[section]
\newtheorem{definition}{Definition}[section]
\newtheorem{example}{Example}[section]

\newtheorem{remark}{Remark}[section]

\newcommand{\figref}[1]{\textup{Fig.~\ref{#1}}}
\newcommand{\tabref}[1]{\textup{Table~\ref{#1}}}

\newcommand{\secref}[1]{\textup{Section~\ref{#1}}}
\newcommand{\appref}[1]{\textup{Appendix~\ref{#1}}}

\newcommand{\thmref}[1]{\textup{Theorem~\ref{#1}}}
\newcommand{\lemref}[1]{\textup{Lemma~\ref{#1}}}
\newcommand{\remref}[1]{\textup{Remark~\ref{#1}}}
\newcommand{\defref}[1]{\textup{Definition~\ref{#1}}}
\newcommand{\propref}[1]{\textup{Proposition~\ref{#1}}}
\newcommand{\corref}[1]{\textup{Corollary~\ref{#1}}}
\newcommand{\exref}[1]{\textup{Example~\ref{#1}}}

\def\etal{\emph{et al.}}
\def\cf{\emph{cf.}}
\def\ie{\emph{i.e.}}
\def\eg{\emph{e.g.}}

\def\RR{\mathbb{R}}

\def\bA{\mathbf{A}}
\def\bB{\mathbf{B}}

\def\bP{\mathbf{P}}

\def\bV{\mathbf{V}}

\def\cB{\mathcal{B}}
\def\cC{\mathcal{C}}

\def\cE{\mathcal{E}}
\def\cF{\mathcal{F}}

\def\cH{\mathcal{H}}
\def\cI{\mathcal{I}}
\def\cJ{\mathcal{J}}

\def\cO{\mathcal{O}}
\def\cP{\mathcal{P}}

\def\cW{\mathcal{W}}

\def\sfi{\mathsf{i}}
\def\sfj{\mathsf{j}}
\def\sfk{\mathsf{k}}

\def\sfn{\mathsf{n}}

\def\sfp{\mathsf{p}}

\def\ba{\mathbf{a}}
\def\bb{\mathbf{b}}
\def\bc{\mathbf{c}}
\def\bd{\mathbf{d}}

\def\bff{\mathbf{f}}
\def\bg{\mathbf{g}}
\def\bh{\mathbf{h}}

\def\bp{\mathbf{p}}
\def\bq{\mathbf{q}}
\def\br{\mathbf{r}}

\def\bw{\mathbf{w}}
\def\bx{\mathbf{x}}
\def\by{\mathbf{y}}

\def\bLambda{\boldsymbol{\Lambda}}

\def\btau{\boldsymbol{\tau}}

\def\bpsi{\boldsymbol{\psi}}

\def\fB{\mathfrak{B}}
\def\fD{\mathfrak{D}}

\def\fX{\mathfrak{X}}

\def\bzero{\mathbf{0}}
\def\bpartial{\boldsymbol{\partial}}

\def\continuity{\mathcal{C}}
\def\dt{{\Deltait t}}
\def\adjoint{\intercal}

\usepackage{bbm}
\DeclareSymbolFont{bbold}{U}{bbold}{m}{n}
\DeclareSymbolFontAlphabet{\mathbbold}{bbold}

\def\tr {\operatorname{tr}}
\def\det{\operatorname{det}}

\def\im {\operatorname{im}}
\def\ker{\operatorname{ker}}

\def\id{\operatorname{id}}

\def\cof{\operatorname{cof}}
\def\J{\cJ}

\DeclareMathOperator*{\argmin}{argmin}

\DeclareMathOperator{\SL}{SL}

\renewcommand{\so}{\mathfrak{so}} 

\DeclareMathOperator{\gl}{\mathfrak{gl}}
\DeclareMathOperator{\slla}{\mathfrak{sl}}

\DeclareMathOperator{\SDiff}{SDiff}

\DeclareMathOperator{\ad}{ad}

\DeclareMathOperator{\ham}{\mathfrak{ham}}

\DeclareMathOperator{\Der}{Der}

\DeclareMathOperator{\Adv}{Adv}
\DeclareMathOperator{\adv}{adv}

\DeclarePairedDelimiterX\braket[2]{\langle}{\rangle}{#1\,\delimsize\vert\,\mathopen{}#2}

\newcommand{\grad}{\mathop{\mathrm{grad}}\nolimits}
\newcommand{\sgrad}{\mathop{\mathrm{sgrad}}\nolimits}
\newcommand{\curl}{\mathop{\mathrm{curl}}\nolimits}

\renewcommand{\div}{\mathop{\mathrm{div}}\nolimits}
\newcommand{\LD}{\mathop{\mathscr{L}}\nolimits}

\def\ip{i}

\def\bpsi{\boldsymbol{\psi}}

\setcopyright{rightsretained}
\acmJournal{TOG}
\acmYear{2024} \acmVolume{43} \acmNumber{6} \acmArticle{270} \acmMonth{12}\acmDOI{10.1145/3687970}

\begin{document}
\title{Fluid Implicit Particles on Coadjoint Orbits}

\author{Mohammad Sina Nabizadeh}
\affiliation{%
  \institution{University of California San Diego}
  \streetaddress{9500 Gilman Dr, MC 0404}
  \city{La Jolla}
  \state{CA}
  \postcode{92093}
  \country{USA}
}
\email{mnabizad@ucsd.edu}

\author{Ritoban Roy-Chowdhury}
\affiliation{%
  \institution{University of California San Diego}
  \streetaddress{9500 Gilman Dr, MC 0404}
  \city{La Jolla}
  \state{CA}
  \postcode{92093}
  \country{USA}
}
\email{rroychowdhury@ucsd.edu}

\author{Hang Yin}
\affiliation{%
  \institution{University of California San Diego}
  \streetaddress{9500 Gilman Dr, MC 0404}
  \city{La Jolla}
  \state{CA}
  \postcode{92093}
  \country{USA}
}
\email{h7yin@ucsd.edu}

\author{Ravi Ramamoorthi}
\affiliation{%
  \institution{University of California San Diego}
  \streetaddress{9500 Gilman Dr, MC 0404}
  \city{La Jolla}
  \state{CA}
  \postcode{92093}
  \country{USA}
}
\email{ravir@ucsd.edu}
\author{Albert Chern}
\affiliation{%
  \institution{University of California San Diego}
  \streetaddress{9500 Gilman Dr, MC 0404}
  \city{La Jolla}
  \state{CA}
  \postcode{92093}
  \country{USA}
}
\email{alchern@ucsd.edu}

\begin{abstract}
We propose Coadjoint Orbit FLIP (CO-FLIP), a high order accurate, structure preserving fluid simulation method in the hybrid Eulerian-Lagrangian framework. We start with a Hamiltonian formulation of the incompressible Euler Equations, and then, using a local, explicit, and high order divergence free interpolation, construct a modified Hamiltonian system that governs our discrete Euler flow.  The resulting discretization, when paired with a geometric time integration scheme, is energy and circulation preserving (formally the flow evolves on a coadjoint orbit) and is similar to the Fluid Implicit Particle (FLIP) method.  CO-FLIP enjoys multiple additional properties including that the pressure projection is exact in the weak sense, and the particle-to-grid transfer is an exact inverse of the grid-to-particle interpolation. The method is demonstrated numerically with outstanding stability, energy, and Casimir preservation.  We show that the method produces benchmarks and turbulent visual effects even at low grid resolutions.
\end{abstract}

%

\ccsdesc[500]{Computing methodologies~Physical simulation}
\ccsdesc[300]{Mathematics of computing~Interpolation}
\ccsdesc[300]{Mathematics of computing~Discretization}
\ccsdesc[500]{Mathematics of computing~Partial differential equations}

\keywords{Geometric fluid mechanics, Hamiltonian mechanics, Structure preserving discretizations, Mimetic interpolation}

\maketitle

\section{Introduction}
\label{sec:Introduction}

\begin{figure} \centering
\begin{picture}(240,280)
    \put(0,0){\includegraphics[{trim=85px 0 95px 0px},clip,width=\columnwidth]{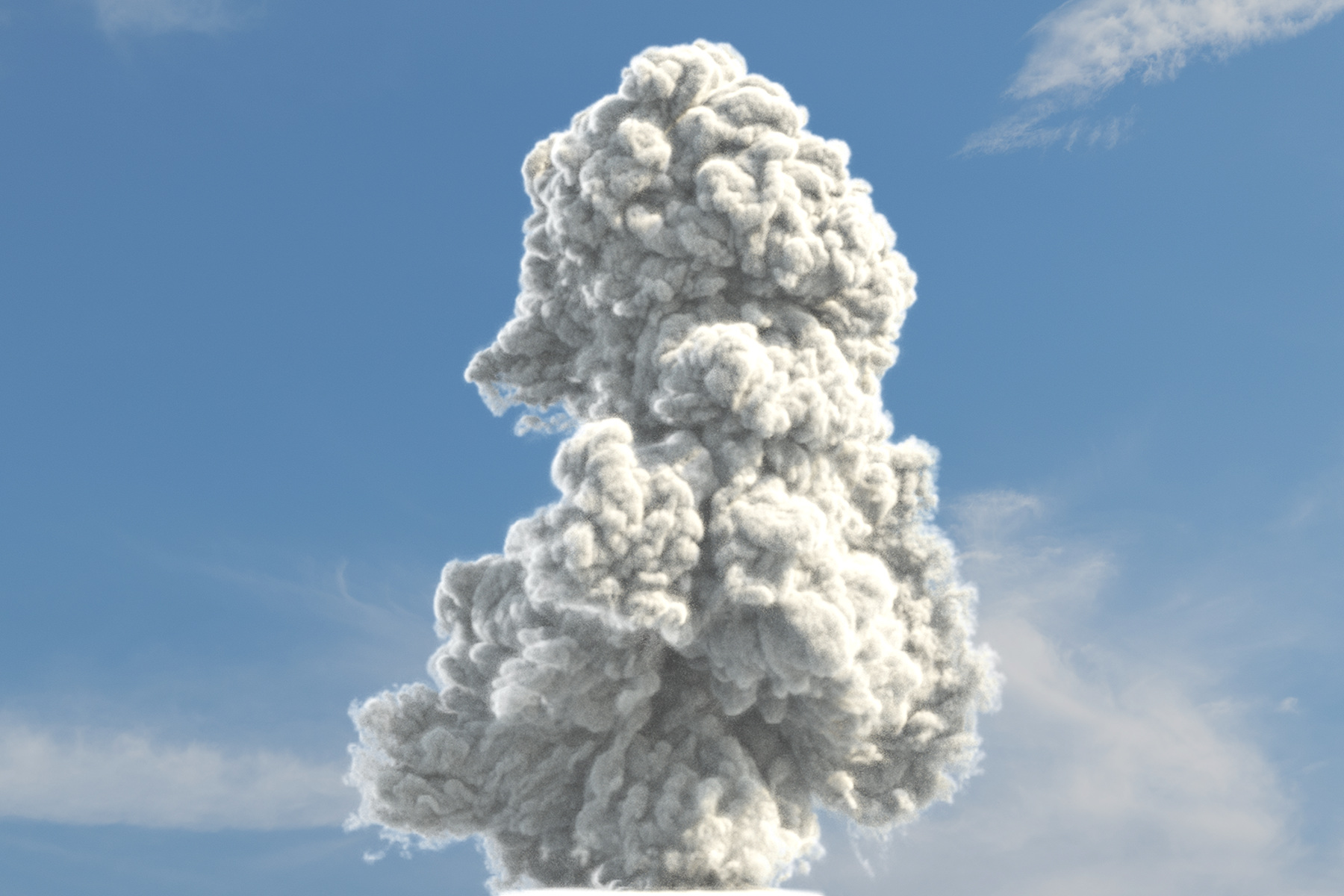}}
    \put(0.5,201.5){\frame{\scalebox{-1}[1]{\includegraphics[trim={320px 1100px 340px 414px},clip,width=0.3\columnwidth]{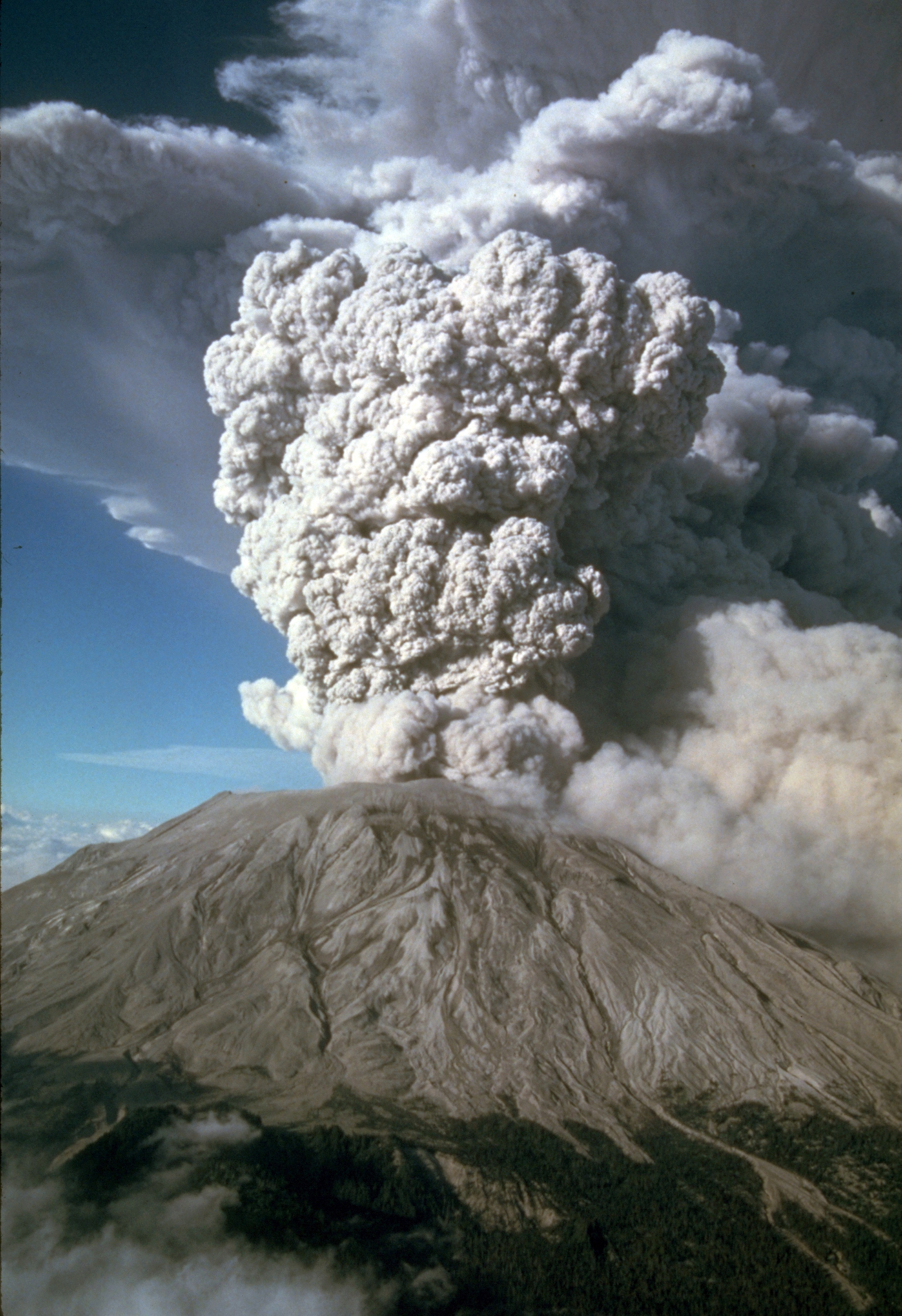}}}}
\end{picture}
\caption{
    Pyroclastic cloud simulated using our method (CO-FLIP).
    Note the intricate vortical structures present on the plume despite the low simulation grid resolution of $96\times192\times96$. 
    Inset: Image of volcano eruption from Mount St. Helens, Washington, USA, on May 18th, 1980 (Photograph by Harris \& Ewing, Inc. [Public domain], via Wikimedia Commons).
    }
    \label{fig:pyroclastic}
\end{figure}

Discretizing incompressible fluid dynamical systems has been a central focus in the development of computational fluid dynamics and fluid animations in computer graphics.
The goal is to find a computationally tractable finite dimensional analogue of the infinite dimensional continuous fluid system that properly represents the simulated measurements and visual phenomena. 
Among these fluid phenomena, the one that has received considerable attention is the conservation of \emph{vortices} in fluids at low or zero viscosity.
They constitute the main visual elements in smoke visual effects as well as ingredients in turbulent flows.
This conservation of vortices, better known as the \emph{circulation conservation}, is, however, challenging to reproduce accurately and stably after discretizations.

In this paper, we revisit the mathematical foundation of the vortex-related conservation laws.  
Specifically, the theory involves describing incompressible fluid dynamics using Lie algebras and Hamiltonian formulations.
The payoff of the theory is that it provides a concrete characterization for the circulation conservation: ``\emph{the circulation state stays on a coadjoint orbit},'' a proposition that will be explained later (\secref{sec:Theory1}).
We propose a discretization framework that satisfies this proposition.
It turns out the natural algorithm derived from the framework is reminiscent of a Fluid Implicit Particles (FLIP) method \cite{Brackbill:1988:FLIP}, a familiar method in the fluid simulation community in graphics.
Our discrete fluid system is therefore referred to as Coadjoint Orbit Fluid Implicit Particles (CO-FLIP).
Inheriting from the mathematical properties of the Hamiltonian formulation of fluid dynamics, CO-FLIP preserves energy and circulation even on a low resolution grid.
Despite such low resolutions, our method produces intricate and fractal-shaped vortical structures (see \figref{fig:pyroclastic}).
In addition, CO-FLIP exhibits outstanding stability compared to similar methods.

While the mathematical properties of CO-FLIP involve a deeper study in geometric fluid mechanics, our CO-FLIP method can be described using the familiar language of FLIP-based algorithms. See \tabref{tab:methods} for a list of acronyms used throughout the text.

\begin{table}
    \centering
    \caption{Method acronyms used throughout the paper. }
    \vspace*{-3ex}
    \label{tab:methods}
    \footnotesize
    \setlength{\tabcolsep}{1pt}
\begin{tabularx}{\columnwidth}{p{100pt}p{60pt}p{30pt}}
    \toprule
    \rowcolor{white}
    Method & Reference & Acronym \\
    \midrule
    Polynomial Particle-in-Cell & \scriptsize \cite{Fu:2017:PolyPIC} & PolyPIC \\
    Fluid Implicit Particles & \scriptsize \cite{Brackbill:1988:FLIP} & FLIP \\
    Polynomial FLIP & \scriptsize \cite{Fei:2021:ASFLIP} & PolyFLIP \\
    Reflection & \scriptsize \cite{Zehnder:2018:ARS} & R \\
    Covector Fluids & \scriptsize \cite{Nabizadeh:2022:CF} & CF \\
    Neural Flow Map & \scriptsize \cite{Deng:2023:FSN} & NFM \\
    \textbf{Coadjoint Orbit FLIP} & \scriptsize \textbf{Our method} & \textbf{CO-FLIP} \\
    \bottomrule
\end{tabularx}
\end{table}

\begin{figure}
\centering
\setlength{\unitlength}{1pt}
\begin{picture}(240,180)
\put(0,0){
\includegraphics[width=240pt]{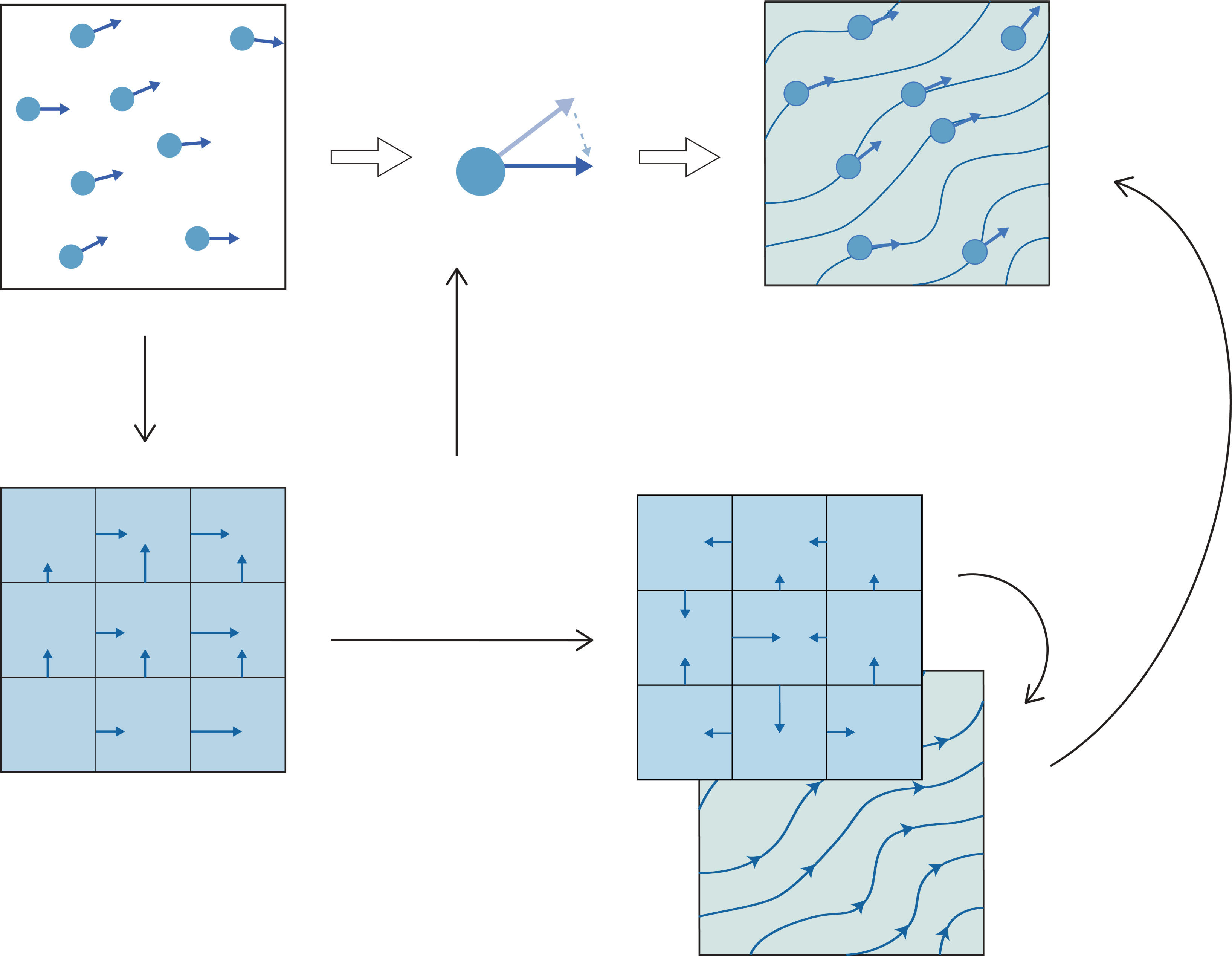}
}
\put(10,105){\sffamily P2G}
\put(75,70){\parbox{10ex}{\linespread{0.6} \small \sffamily Pressure projection}}
\put(97,110){\parbox{30pt}{\linespread{0.6} \small \sffamily Interpolate pressure force}}
\put(190,73){\parbox{30pt}{\linespread{0.6} \small \sffamily Interpolate velocity}}
\put(180,105){ \sffamily G2P}
\put(218,100){\parbox{30pt}{\linespread{0.6} \small \sffamily Advect particles}}
\end{picture}
\caption{The classical FLIP method for simulating incompressible fluids. Particles (top row) are moved, and their momentum updated, using a background velocity field which is reconstructed from the pressure projection of the transferred particle data to the grid (bottom row).  The CO-FLIP algorithm modifies each stage: The particles move and their velocity covector is transformed (i) using a point-wise div-free interpolation (ii) of grid velocity data transferred from particles using the pseudoinverse of the div-free interpolation (iii), and, further, Galerkin pressure projected (iv). This time evolution is now a Hamiltonian ODE which we integrate geometrically (v). }
\label{fig:flip_algorithm}
\end{figure}

\paragraph{Quick Recall of the FLIP Method}
FLIP \cite{Brackbill:1988:FLIP} is one of the family of Particle in Cell (PIC) \cite{Harlow:1962:PIC} methods or Material Point Methods (MPMs) \cite{Sulsky:1995:MPM}. 
They model incompressible and inviscid fluids governed by the Euler equations using a hybrid of Lagrangian and Eulerian coordinates.  
The fluid state is given by the positions and velocities of a large number of particles.  There is also an auxiliary grid that discretizes the fluid domain.  A so-called \emph{particle-to-grid transfer} (P2G) interpolates the particle velocities to the grid as an Eulerian velocity field.  
This grid velocity is typically stored in a Marker-And-Cell (MAC) staggered grid, with fluxes annotated on the grid's cell walls.
The grid velocity data is then \emph{pressure-projected} (PP), via solving a central-difference Poisson equation, to become a discrete divergence-free velocity field.
The divergence-free grid velocity is interpolated (\emph{grid-to-particle} (G2P)) into a continuous velocity field over the fluid domain.  
The particle positions are updated by flowing along this interpolated velocity field, and the particle velocity is updated by the interpolated pressure force computed in the pressure projection step.
This concludes the entire time step of the FLIP algorithm. See \figref{fig:flip_algorithm}.

Popular cousin methods include the classical PIC and Affine-PIC (APIC) \cite{Jiang:2015:APIC} methods.
The classical PIC method replaces the particle velocity by the G2P velocity at every time step. 
APIC is a PIC method that additionally stores spatial gradients of velocities on particles for better quality P2G and G2P.  
FLIP is more advantageous to PIC or APIC as FLIP does not overwrite the particle velocity with a lower-resolution grid velocity.  
Thus FLIP is seen as a long-time method of characteristic, while PIC or APIC has their velocity state limited to a grid resolution comparable to purely grid-based semi-Lagrangian \cite{Stam:1999:SF} or finite volume methods \cite{Mullen:2009:EPI}.
The downside of FLIP is that it has been observed to be more unstable compared to PIC.
Common implementations of FLIP blend FLIP's particle velocity with PIC or APIC velocity for stability.

\paragraph{Our CO-FLIP Method}
From a numerical viewpoint, to obtain CO-FLIP, we modify the FLIP method to satisfy the following properties. See \figref{fig:flip_algorithm} for a summary of the modifications.
\begin{enumerate}[(i)]
    \item\label{item:intro-covector} The particle velocity is stored as a \emph{velocity covector} \cite{Nabizadeh:2022:CF}, also known as an \emph{impulse} \cite{Feng:2022:IFS,Deng:2023:FSN}.  As a particle position is updated by flowing along a G2P velocity, its velocity covector is transformed by the transpose of the G2P velocity gradient. As depicted in the inset, vectors are thought of as arrows with magnitude and direction, whereas covectors are taken as local level sets with magnitude (its compactness) and direction (its orientation).
    
        \setlength{\unitlength}{1pt}
        \begin{picture}(240,90)
        \put(20,15){
        \includegraphics[width=60pt]{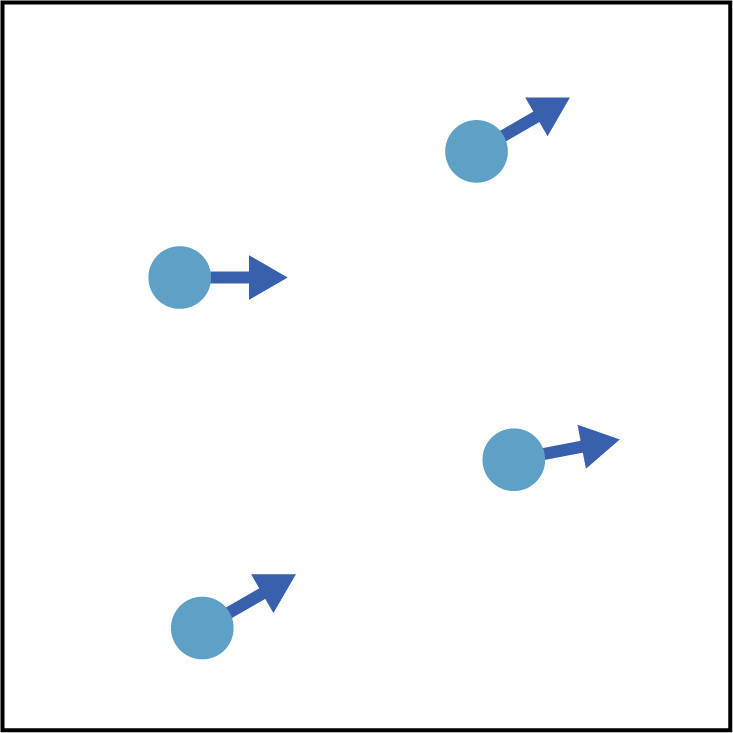}}
        \put(125,15){
        \includegraphics[width=60pt]{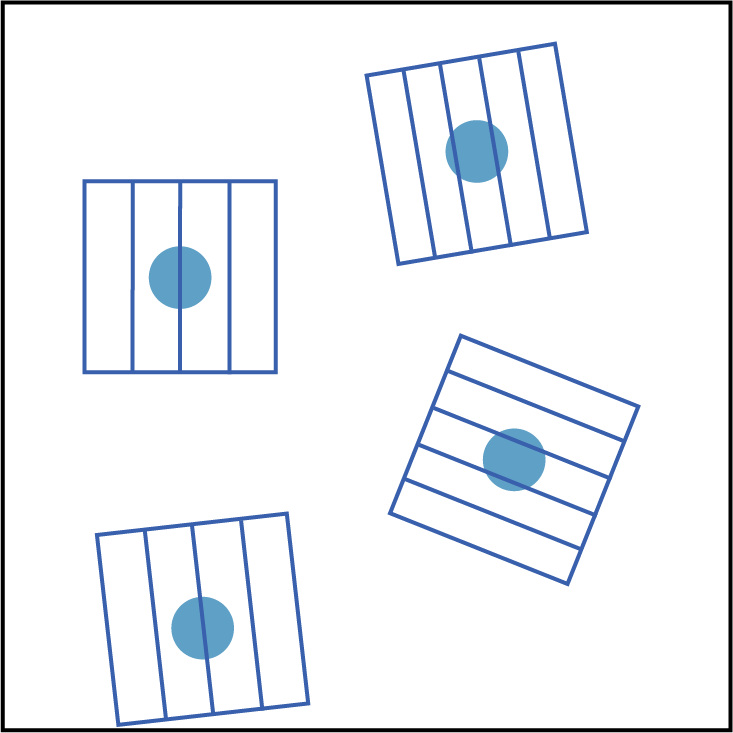}}
        \put(28,5){\small \sffamily Velocity vectors}
        \put(113,5){\small \sffamily Velocity covectors (impulse)}
        \put(51,80){\makebox(0,0)[c]{\small\sffamily FLIP}}
        \put(156,80){\makebox(0,0)[c]{\small\sffamily CO-FLIP}}
        \end{picture}

    \item\label{item:intro-mimetic} The G2P interpolation is \emph{mimetic}.  That is, a discretely divergence-free MAC-grid velocity will interpolate into a continuously divergence-free vector field.  This G2P interpolation is, therefore, an embedding from the finite dimensional space of grid velocity data into the infinite dimensional space of vector fields so that the discrete divergence-free subspace maps into the continuous divergence-free subspace.  Using a specific usage of B-splines \cite{Buffa:2010:IGA,Schroeder:2022:LDP,RoyChowdhury:2024:HOD}, we have an explicit local interpolation scheme that is mimetic and of arbitrarily high order. As shown in the inset, this results in point-wise divergence-free interpolated velocity fields for the CO-FLIP method; as opposed to the traditional FLIP methods which show particle clumping.
    
        \setlength{\unitlength}{1pt}
        \begin{picture}(240,126)
        \put(-5,-3){
        \includegraphics[width=230pt]{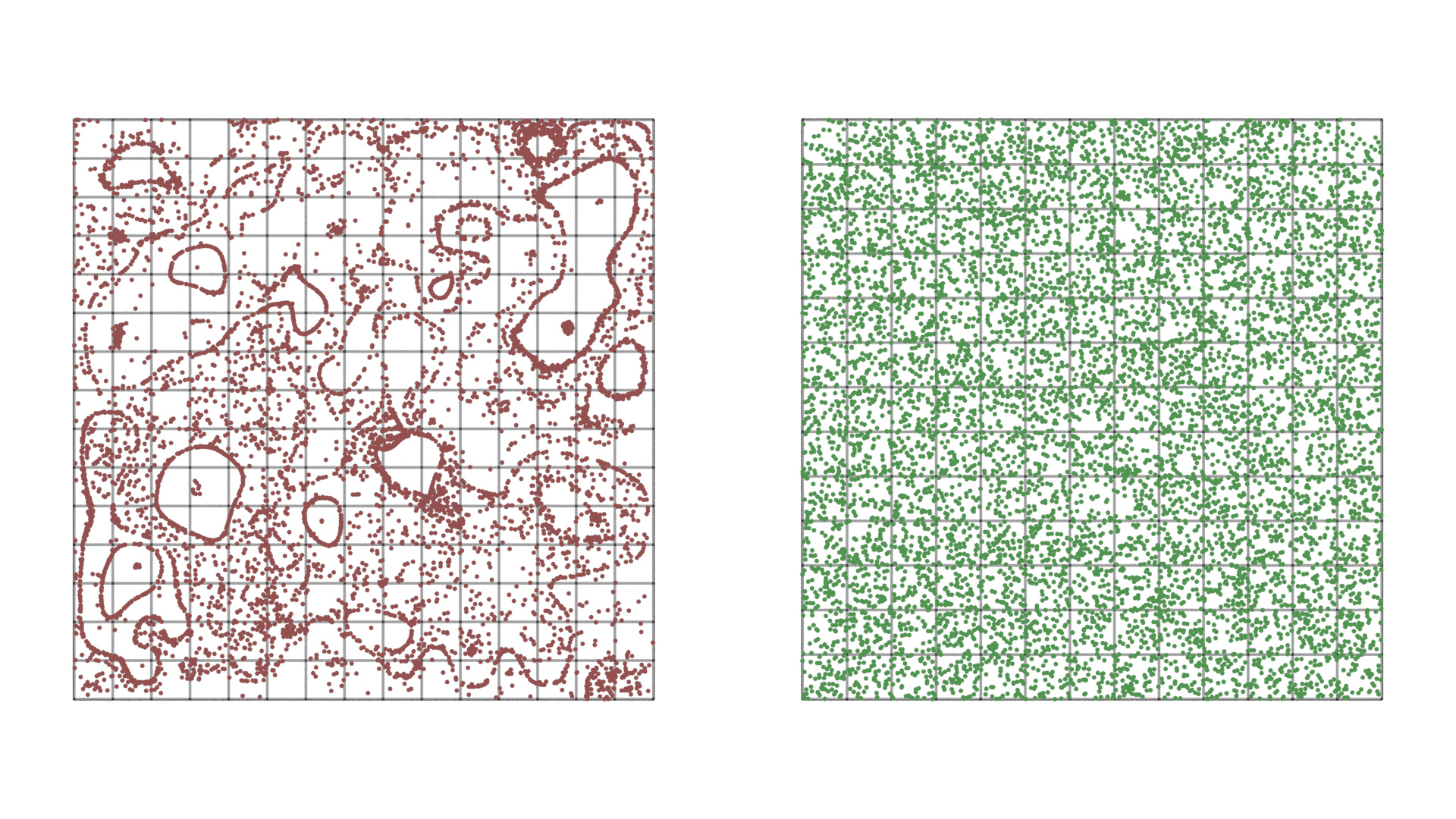}}
        \put(22,5){\small \sffamily Not divergence-free}
        \put(146,5){\small \sffamily Divergence-free}
        \put(0,63){\makebox(0,0)[c]{\small\sffamily \rotatebox{90}{FLIP}}}
        \put(119,66){\makebox(0,0)[c]{\small\sffamily \rotatebox{90}{CO-FLIP}}}
        \end{picture}

    \begin{adjustbox}{minipage={\linewidth}, valign=t}    
    \begin{wrapfigure}{t}{0.4\linewidth}
        \vspace{10pt}
        \hspace{-15pt}
        \includegraphics[width=\linewidth+10pt]{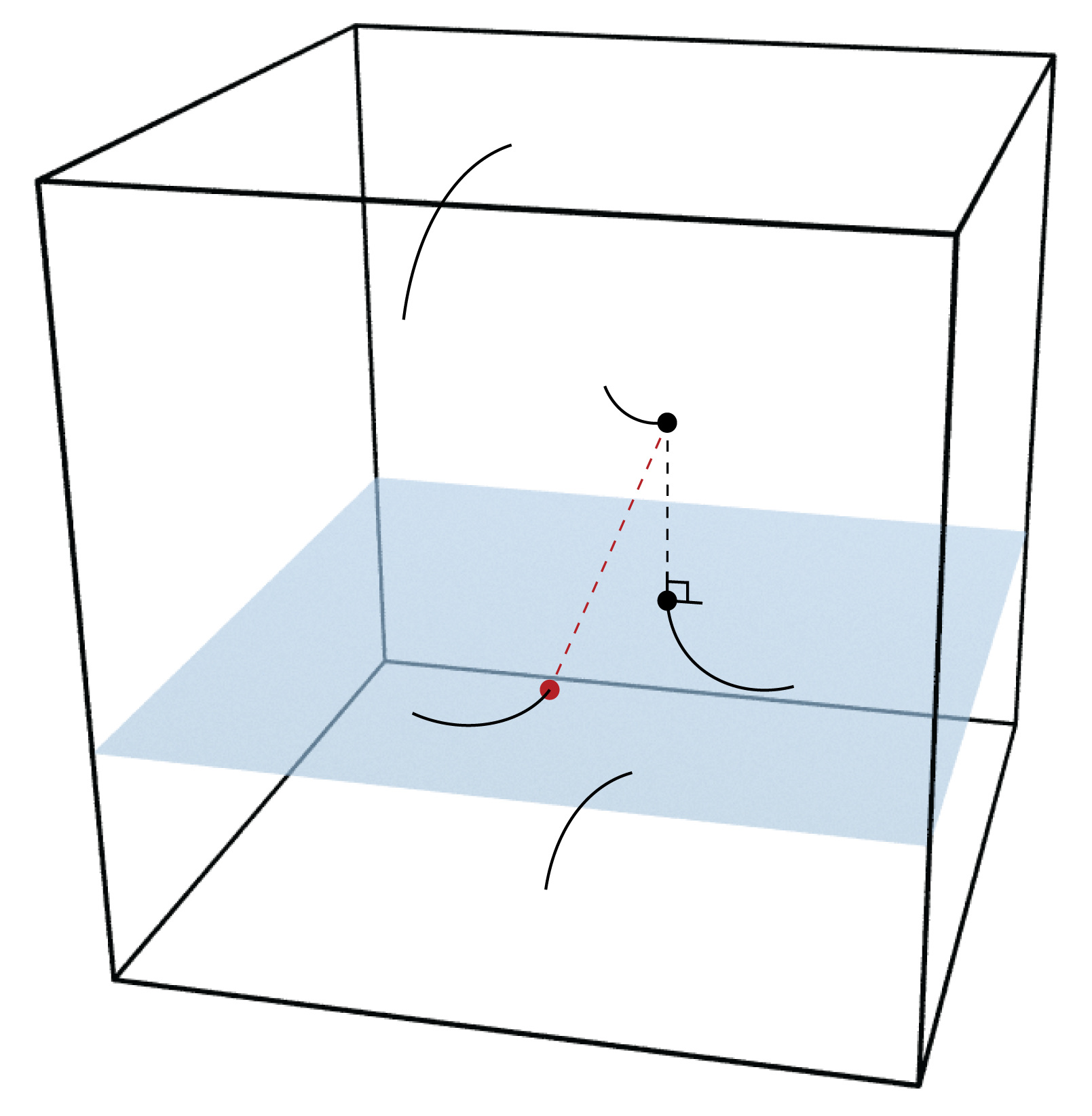}
        \begin{picture}(0,0)(0,0)
            \put(-55,88){\sffamily \scriptsize \parbox{7em}{\linespread{0.6} Space of smooth vector fields}}
            \put(-67,14){\sffamily \scriptsize \parbox{8em}{\linespread{0.6} Subspace of interpolated vector fields}}
            \put(-57,68){\sffamily \scriptsize {Particle velocity}}
            \put(-27, 37){\sffamily \scriptsize \parbox{6em}{\linespread{0.6} CO-FLIP \\ grid velocity}}
            \put(-88, 34){\sffamily \scriptsize \parbox{6em}{\linespread{0.6} Traditional grid velocity}}
        \end{picture}
        \vspace{-2\baselineskip}
    \end{wrapfigure}
    \vspace*{0.15em}
    \item\label{item:intro-galerkin-P2G} The P2G interpolation is derived as the \emph{pseudoinverse} of the mimetic G2P interpolation of (\ref{item:intro-mimetic}).  If a particle velocity data is assigned by an interpolation from a grid velocity data, its P2G interpolation exactly reconstructs the grid velocity data (P2G\({}\mathbin{\circ}{}\)G2P is the identity map on the space of grid velocities).  If the particle velocity is not in the range of G2P, the P2G operator orthogonally projects it on the image of G2P.
    In particular, the order of accuracy of the interpolation is directly governed by the B-spline order of (\ref{item:intro-mimetic}). 
    There is no need for the APIC or PolyPIC \cite{Fu:2017:PolyPIC} technique to include additional Taylor expansion modes on particles for improving the P2G accuracy.  As illustrated in the inset, this means we orthogonally project to the subspace of interpolated velocities. In contrast, previous methods find a less accurate, non-orthogonal projection.
    \end{adjustbox}

    \begin{adjustbox}{minipage={\linewidth}, valign=t}    
    \begin{wrapfigure}{t}{0.4\linewidth}
        \vspace{10pt}
        \hspace{-15pt}
        \includegraphics[width=\linewidth+10pt]{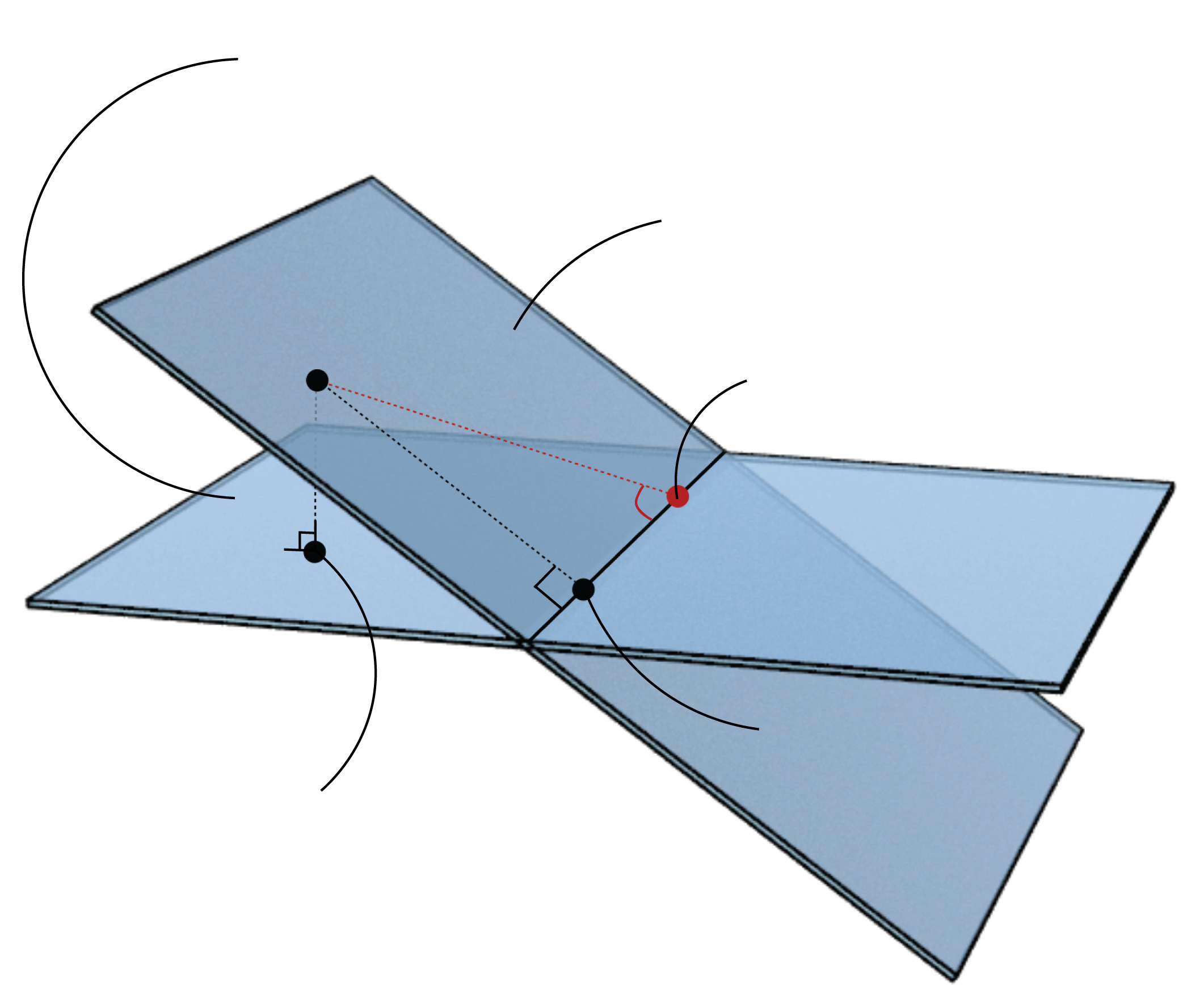}
        \begin{picture}(0,0)(0,0)
            \put(-81,78){\sffamily \scriptsize \parbox{8em}{\linespread{0.6} Subspace of smooth div-free vector fields}}
            \put(-47,63){\sffamily \scriptsize \parbox{8em}{\linespread{0.6} Subspace of interpolated vector fields}}
            \put(-40,46){\sffamily \scriptsize \parbox{8em}{\linespread{0.6} Traditional \\ pressure projection}}
            \put(-39, 16){\sffamily \scriptsize \parbox{8em}{\linespread{0.6} CO-FLIP \\ pressure projection}}
            \put(-90,9){\sffamily \scriptsize \parbox{6em}{\linespread{0.6} Exact pressure projection}}
        \end{picture}
        \vspace{-2\baselineskip}
    \end{wrapfigure}
    \vspace*{0.15em}
    \item\label{item:intro-galerkin-PP} The discrete pressure projection step over the grid is \emph{Galerkin}.  Instead of solving the simplified central difference Poisson equation, we derive the pressure projection respecting the inner product structure induced by the continuous \(L^2\) structure through the embedding (\ref{item:intro-mimetic}).
    In particular, this pressure projection is the exact \(L^2\) projection to the divergence-free subspace of B-splines when it is mapped by the mimetic interpolation of (\ref{item:intro-mimetic}).
    Compared to traditional methods, this means we have a high order pressure projection with minimal error as demonstrated in the inset. 
    \end{adjustbox}

    \item\label{item:intro-Hamiltonian}
    The resulting time-evolution equation for the particles is a Hamiltonian ODE, allowing us to employ structure preserving geometric integrators.
\end{enumerate}
Some of the above numerical ideas were known in the literature but have not been integrated into fluid simulations in a practical manner. In that regard, we highlight a few additional numerical contributions.
\begin{itemize}
    \item The divergence-free interpolation methods were recently introduced to computer graphics. Unlike these recent techniques requiring global linear solves \cite{Chang:2022:CurlFlow,Lyu:2024:WPE}, our mimetic interpolation (\ref{item:intro-mimetic}) is local and explicit and does not require any linear solve.
    \item Pseudoinverse-based P2G (\ref{item:intro-galerkin-P2G}) discussed above is introduced in the original MPM work of \cite{Sulsky:1995:MPM} but was immediately simplified to the common PIC interpolations through mass lumping.  Existing work XPIC \cite{Hammerquist:2017:XPIC} and FMPM \cite{Nairn:2021:FMPM} attempts to better approximate the pseudoinverse-based P2G through a truncated Neumann series.  We show that the exact pseudoinverse-based P2G can be computed efficiently by a standard preconditioned conjugate gradient (PCG) solver.

    \item Structure-preserving time integrators \cite{Mullen:2009:EPI,Pavlov:2011:SPD,Azencot:2014:FFS} for (\ref{item:intro-Hamiltonian}) are usually implicit, involving expensive high-dimensional root finding with Newton solves.  We show that these implicit methods can be bootstrapped from simple explicit integrators via a few fixed point iterations.
\end{itemize}

In addition to commonly tested fluid simulation benchmarks, we demonstrate the following conservation laws quantitatively.  Some of these conservation laws have rarely been tested in previous numerical work.  In addition to energy, momentum, and angular momentum conservation,  
in 2D \(\iint (\curl\vec u)^k dA\) for every \(k=1,2,3,\ldots\) is conserved under the Euler equation.  In 3D, the helicity \(\iiint \vec u\cdot(\curl\vec u) dV\) is conserved.  Here \(\vec u\) denotes the fluid velocity.  These invariants of the Euler flow are known as the \emph{Casimir invariants}.
We point out that the conservation of these Casimir invariants is the quantitative measure for the circulation conservation.

Besides being a higher order method based on (\ref{item:intro-mimetic})--(\ref{item:intro-galerkin-PP}), the CO-FLIP algorithm shows superior invariant conservation, both in energy and Casimirs, and stability compared to a classical FLIP method.  Moreover, CO-FLIP is capable of producing benchmark vortex phenomena at low grid resolutions (see \figref{fig:unknot_evolution}). For instance, $64^3$, or even lower resolutions of $32^3$.

\begin{figure*}
    \centering
    \includegraphics[trim={450px 0 600px 0},clip,width=0.16\linewidth]{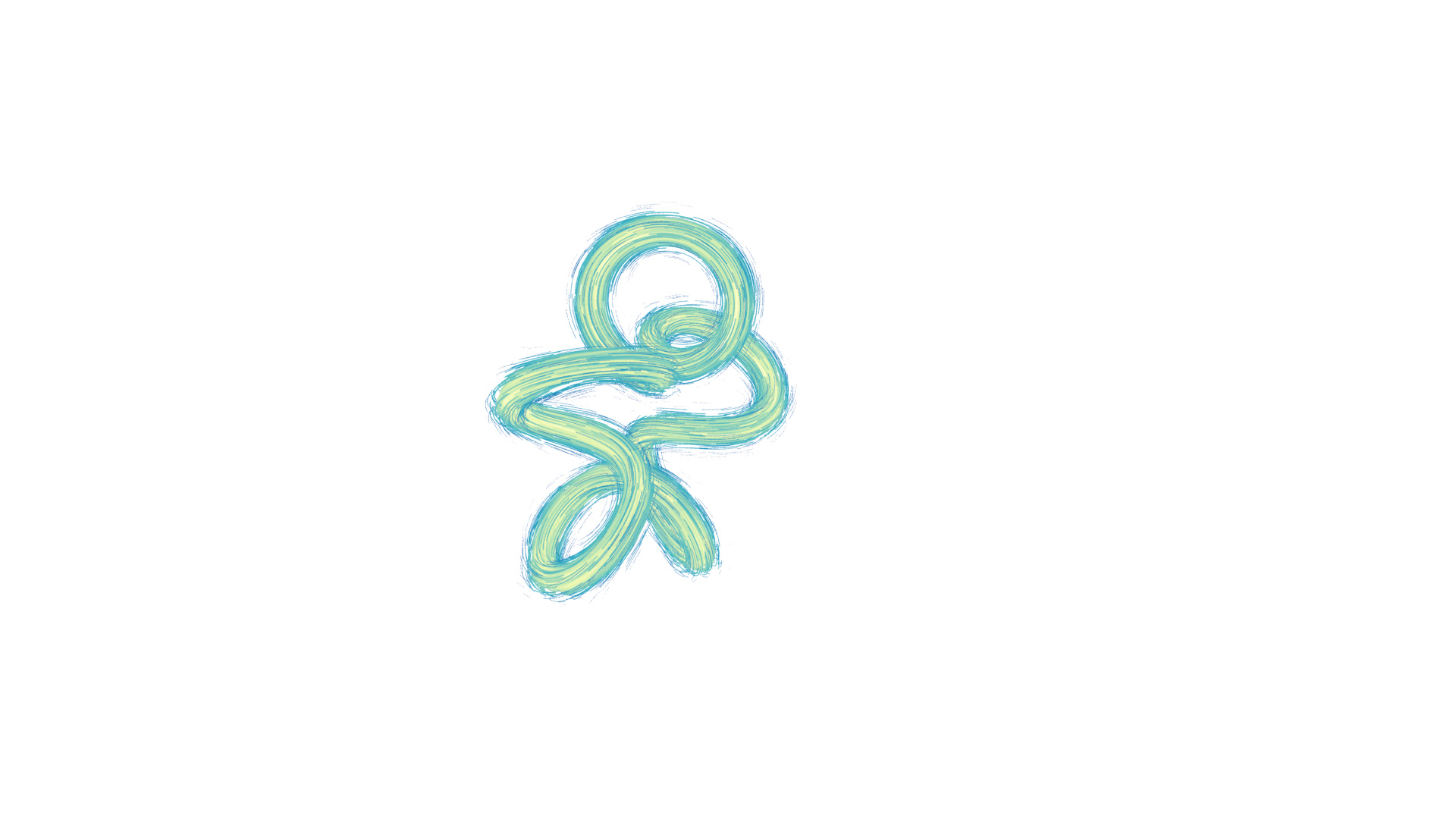}
    \includegraphics[trim={450px 0 600px 0},clip,width=0.16\linewidth]{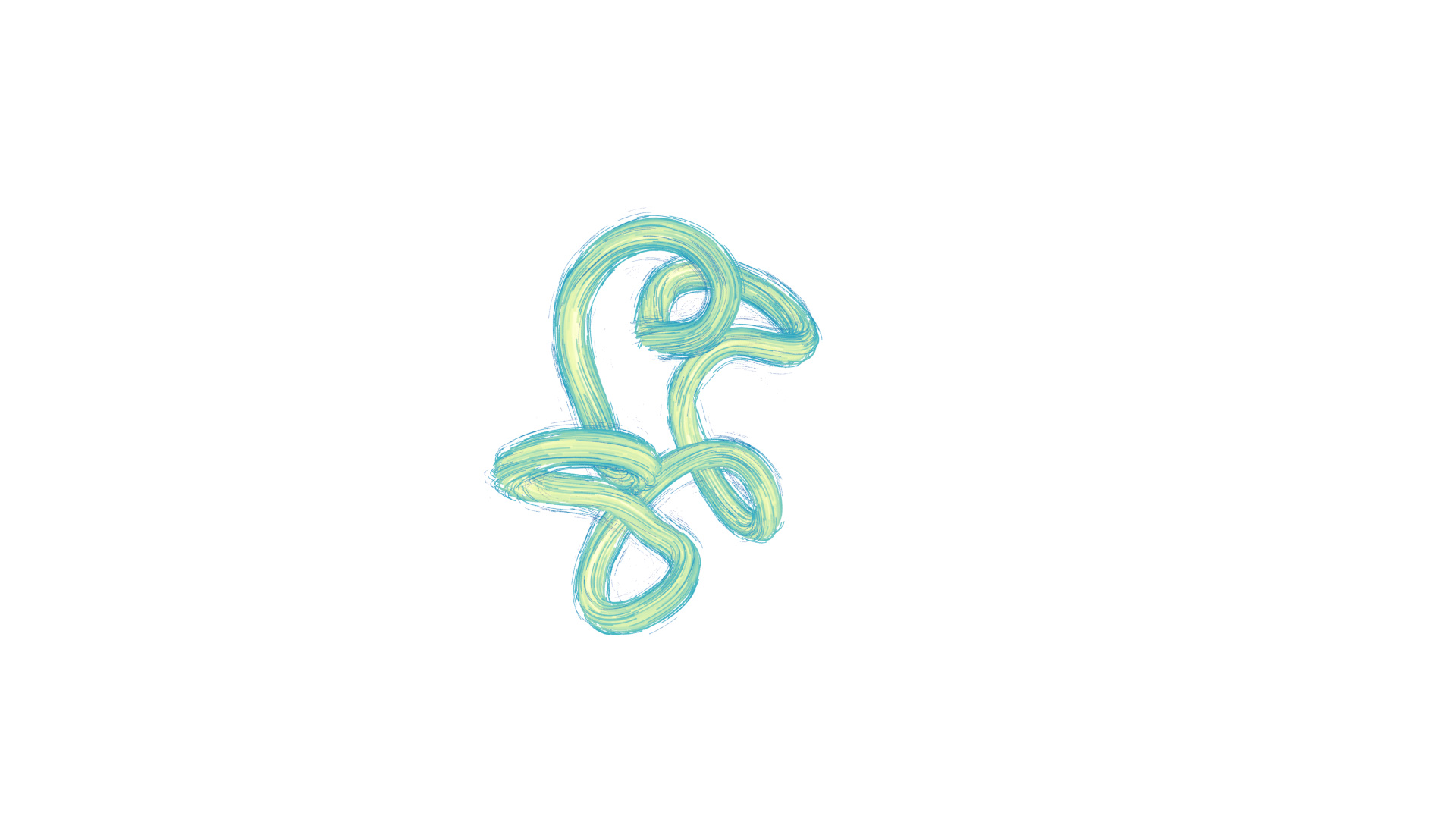}
    \includegraphics[trim={450px 0 600px 0},clip,width=0.16\linewidth]{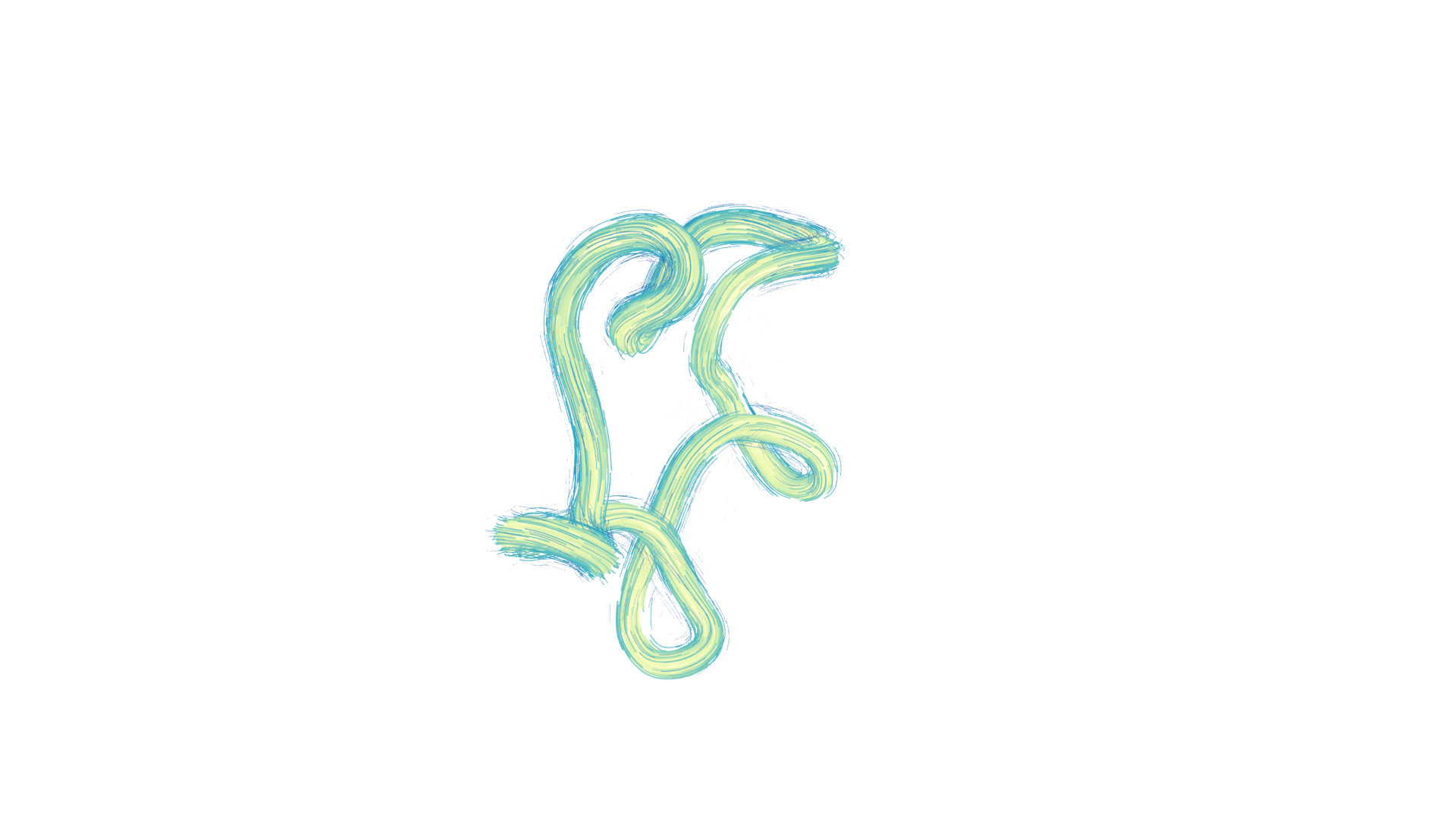}
    \includegraphics[trim={450px 0 600px 0},clip,width=0.16\linewidth]{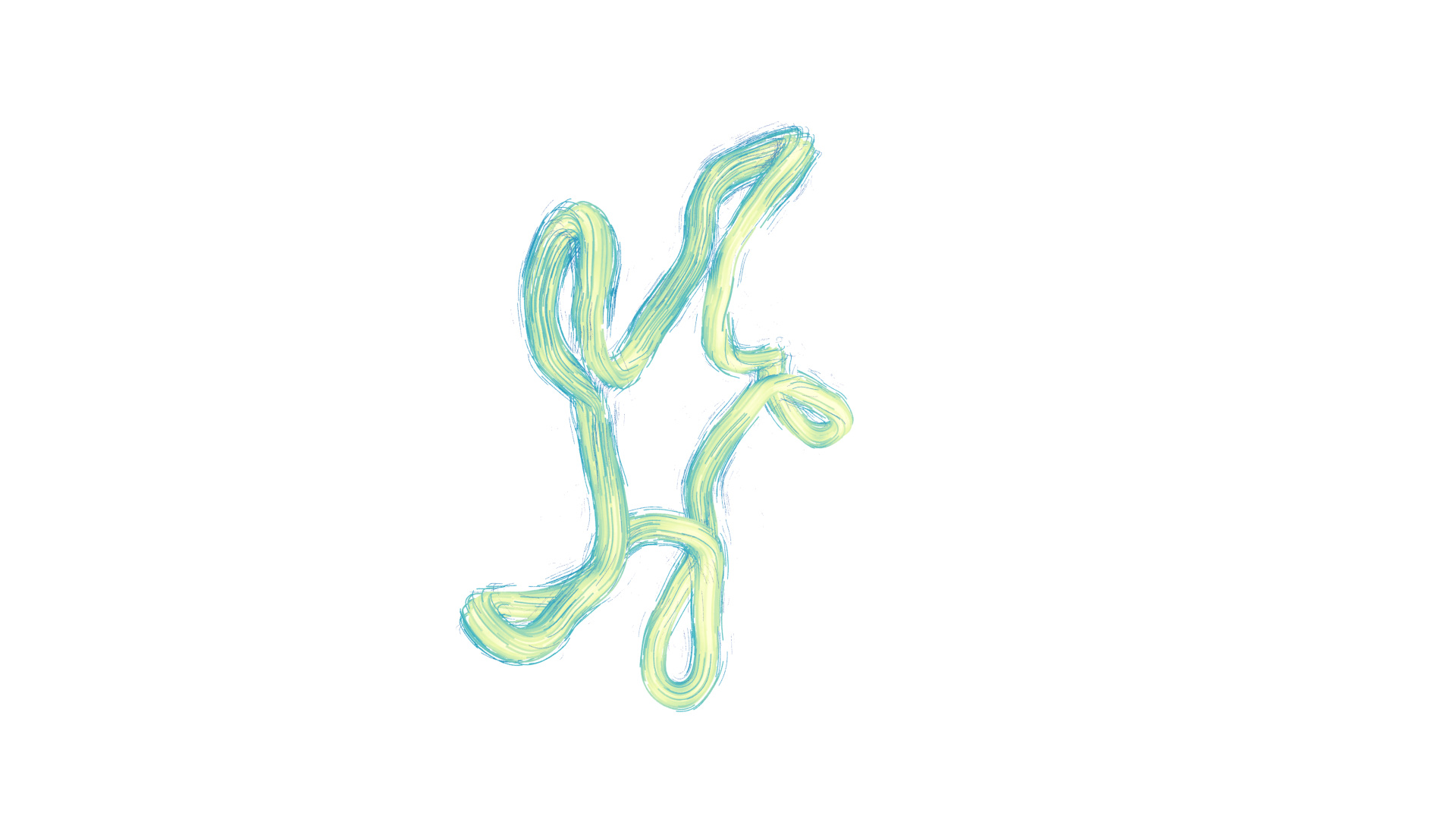}
    \includegraphics[trim={450px 0 600px 0},clip,width=0.16\linewidth]{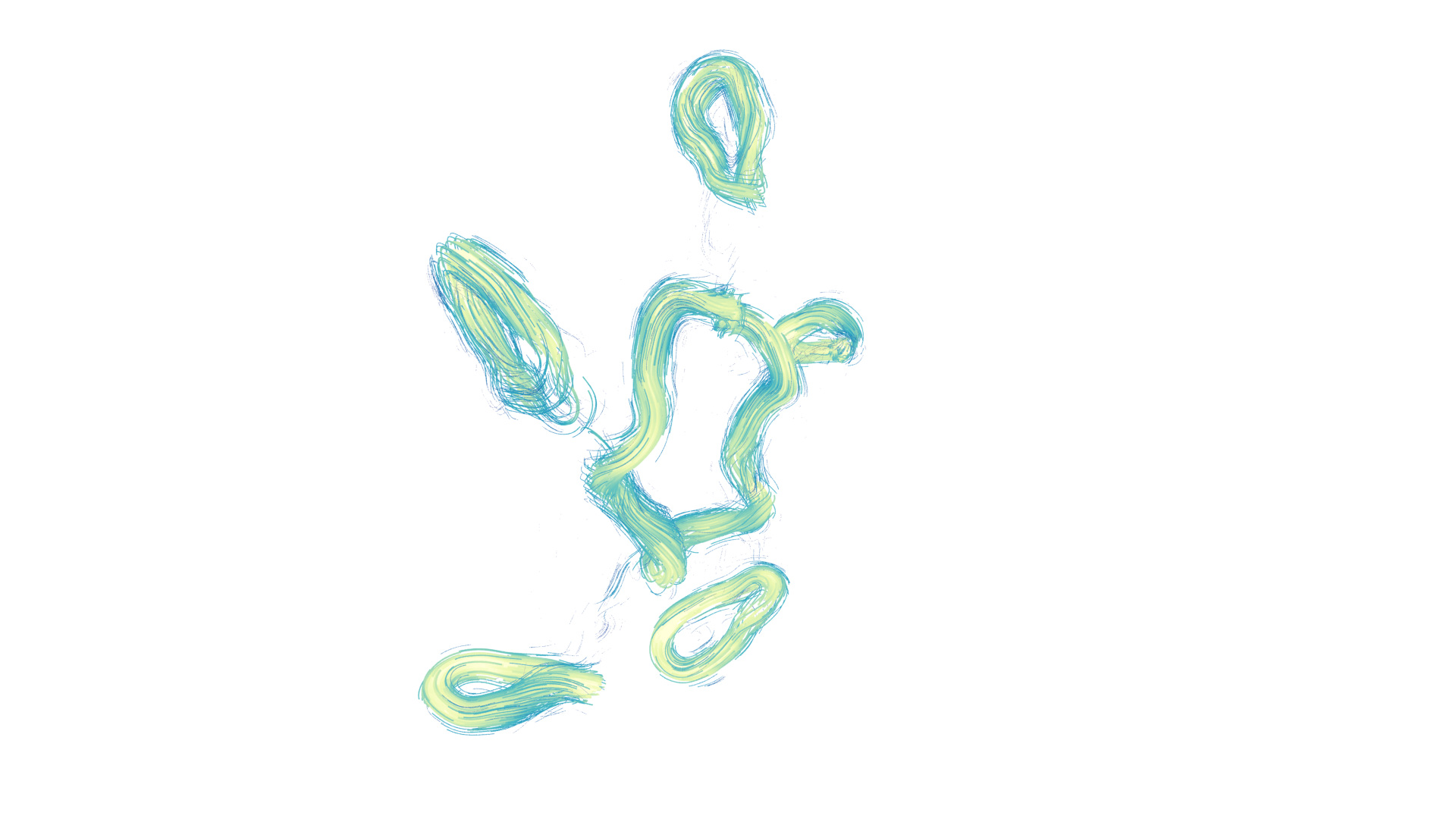}
    \includegraphics[trim={450px 0 600px 0},clip,width=0.16\linewidth]{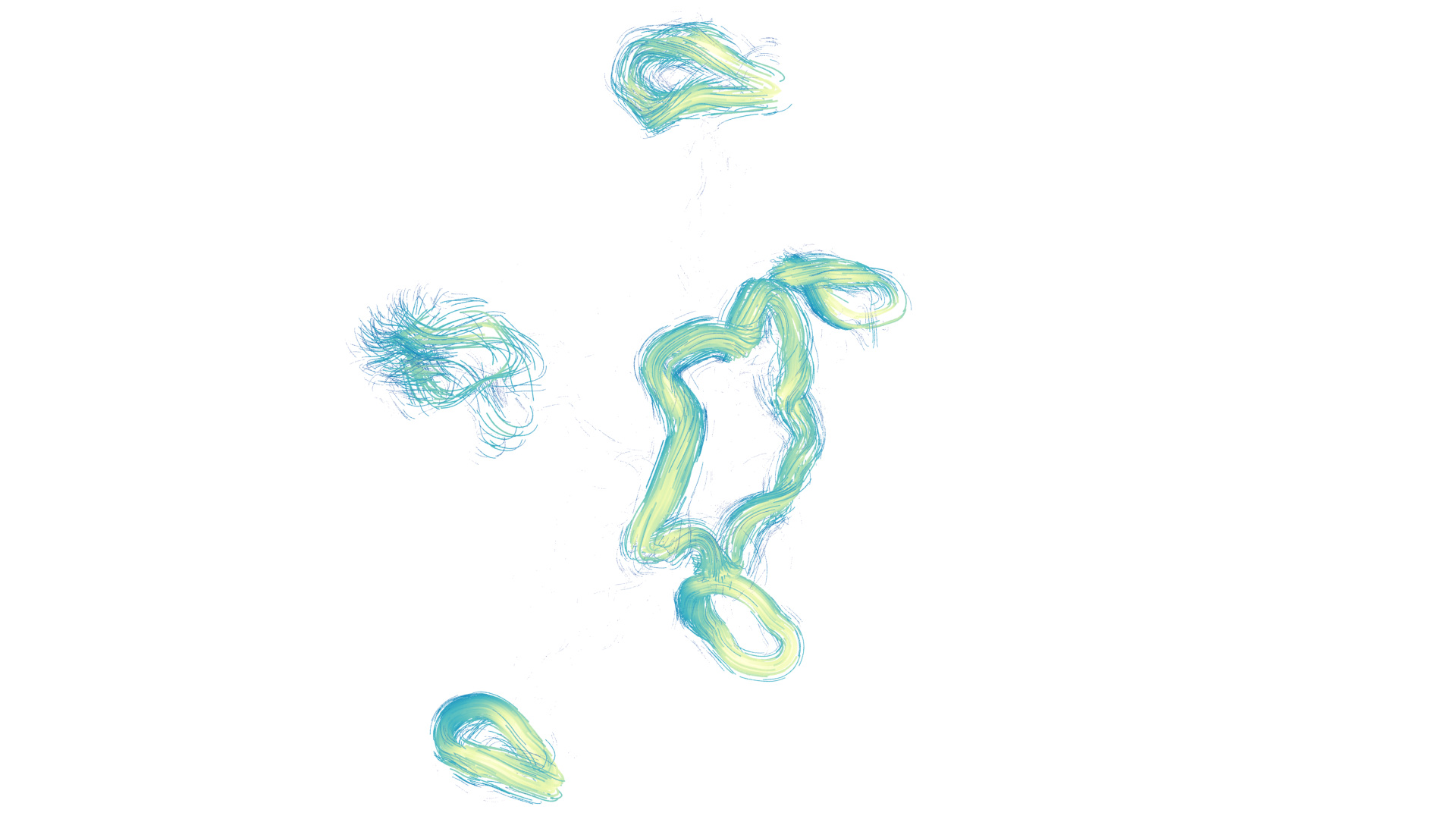}
    \vspace{10pt}
    \includegraphics[trim={450px 0 400px 0},clip,width=0.16\linewidth]{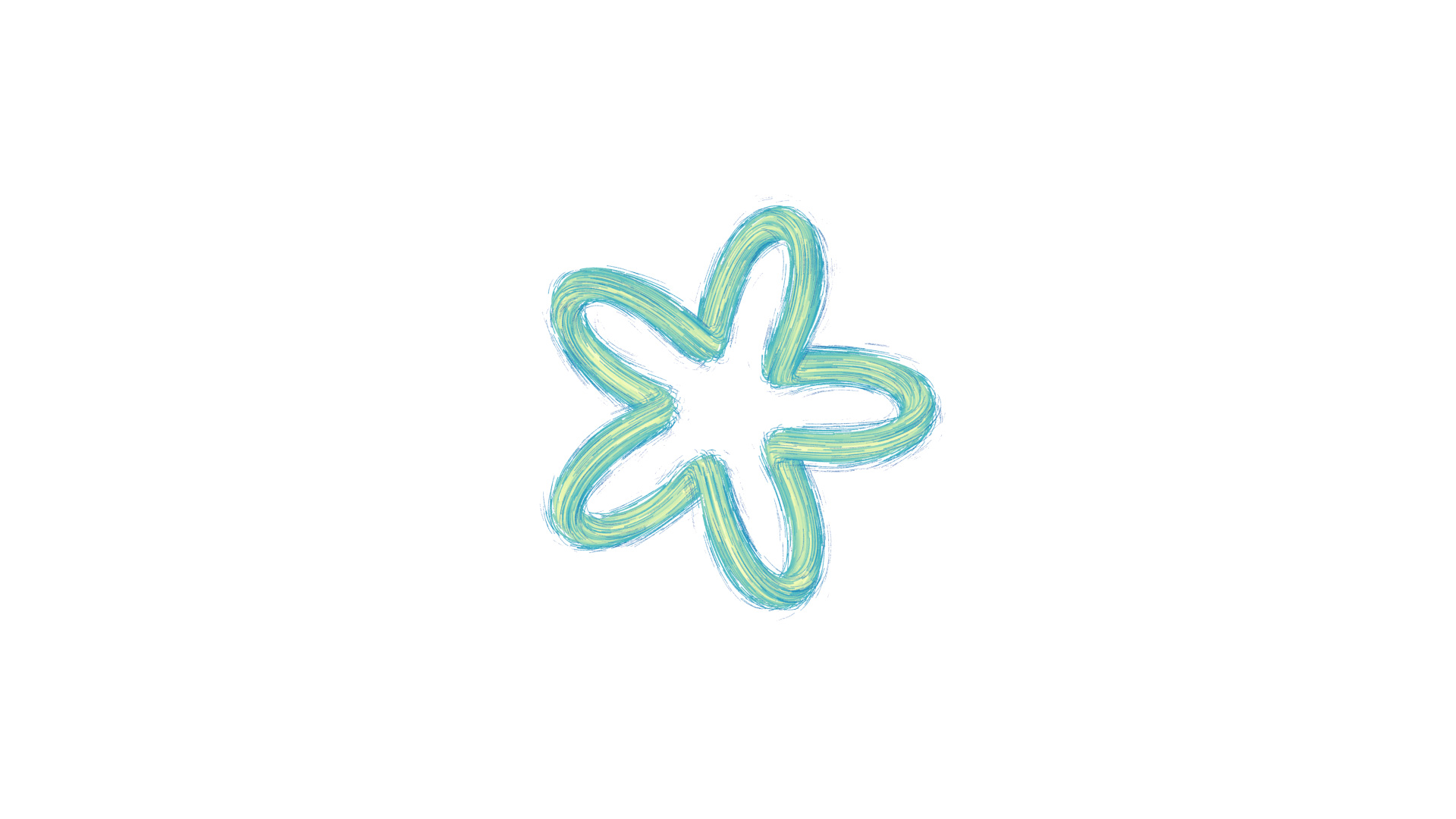}
    \includegraphics[trim={450px 0 400px 0},clip,width=0.16\linewidth]{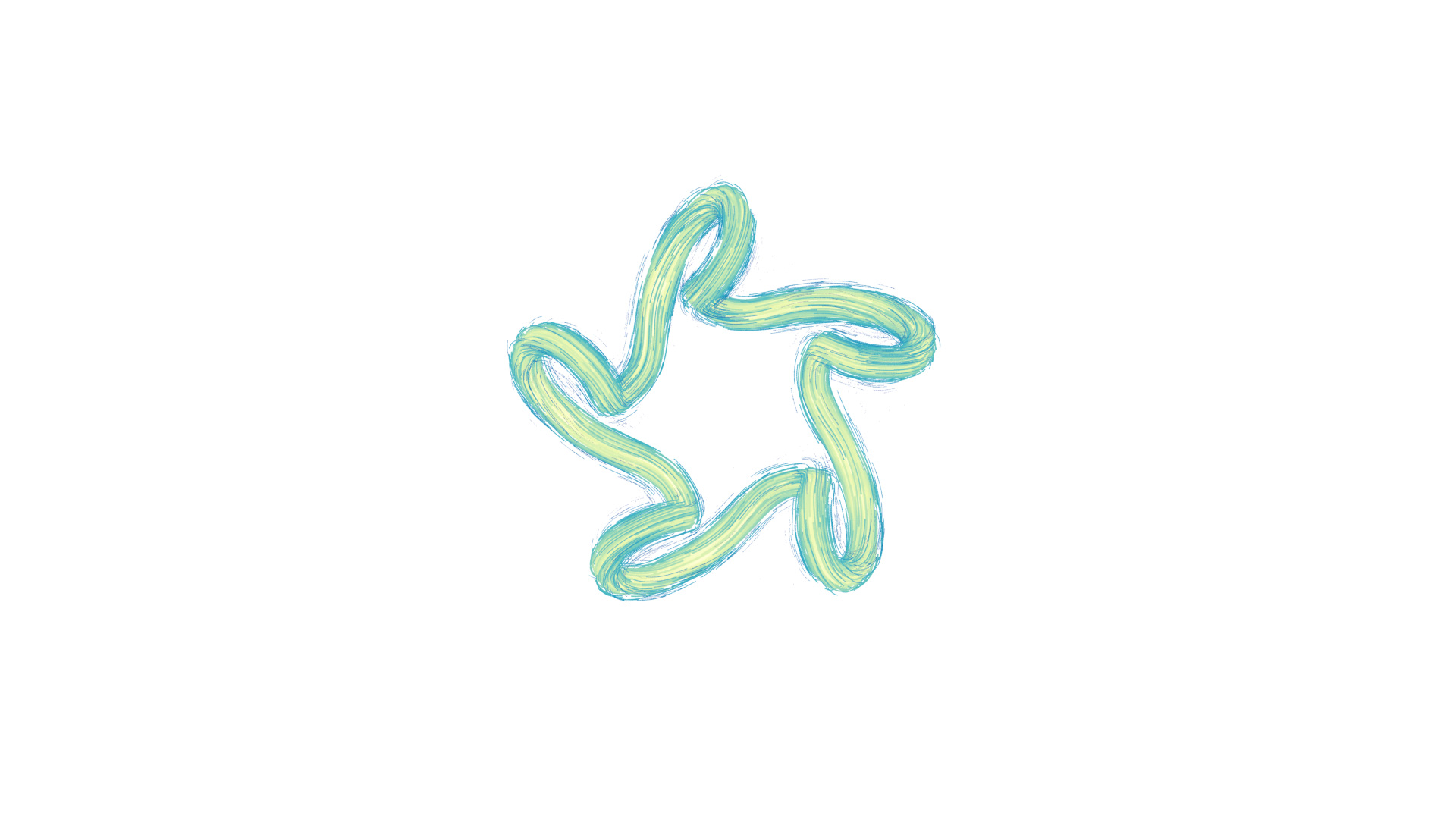}
    \includegraphics[trim={450px 0 400px 0},clip,width=0.16\linewidth]{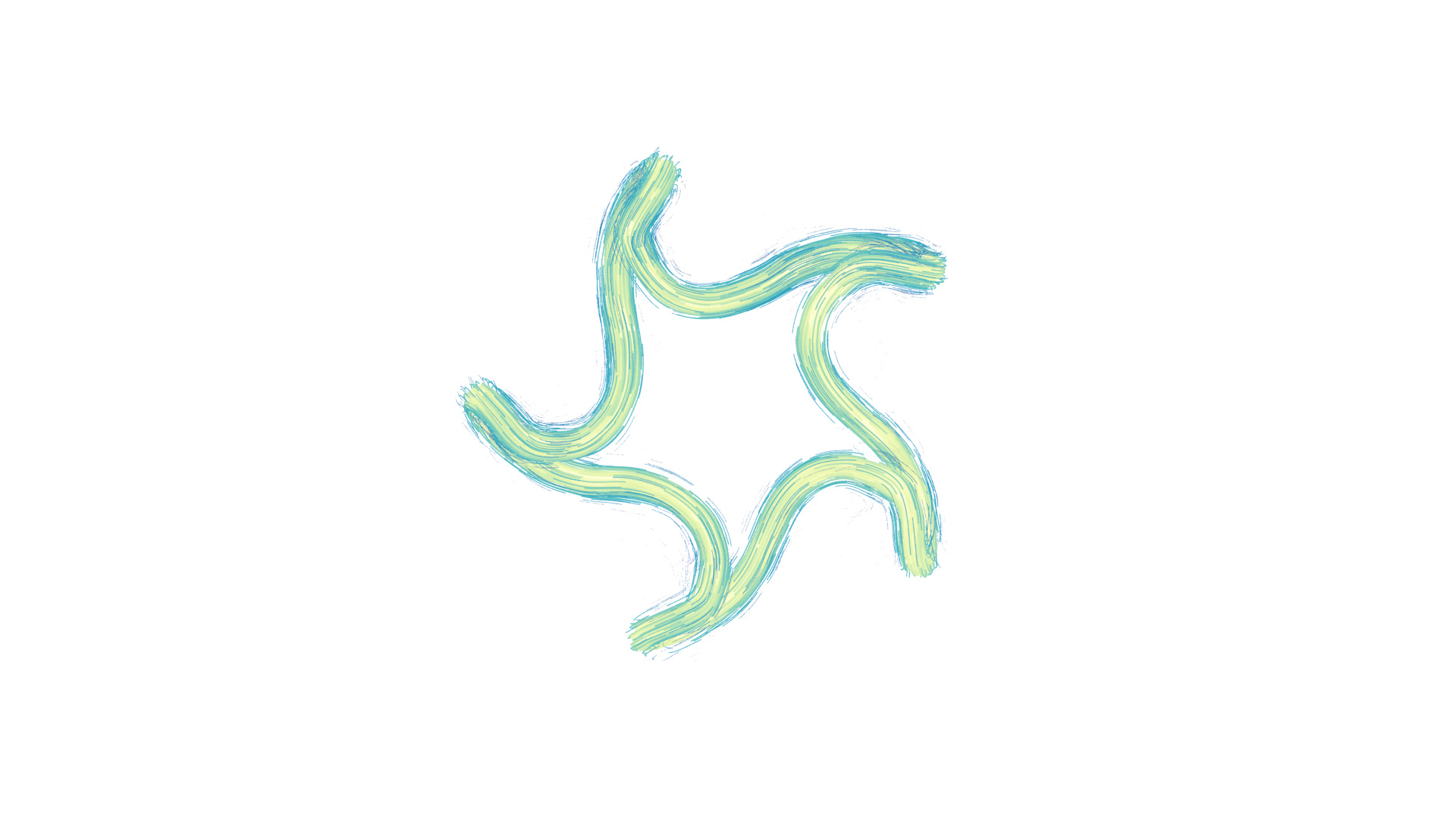}
    \includegraphics[trim={450px 0 400px 0},clip,width=0.16\linewidth]{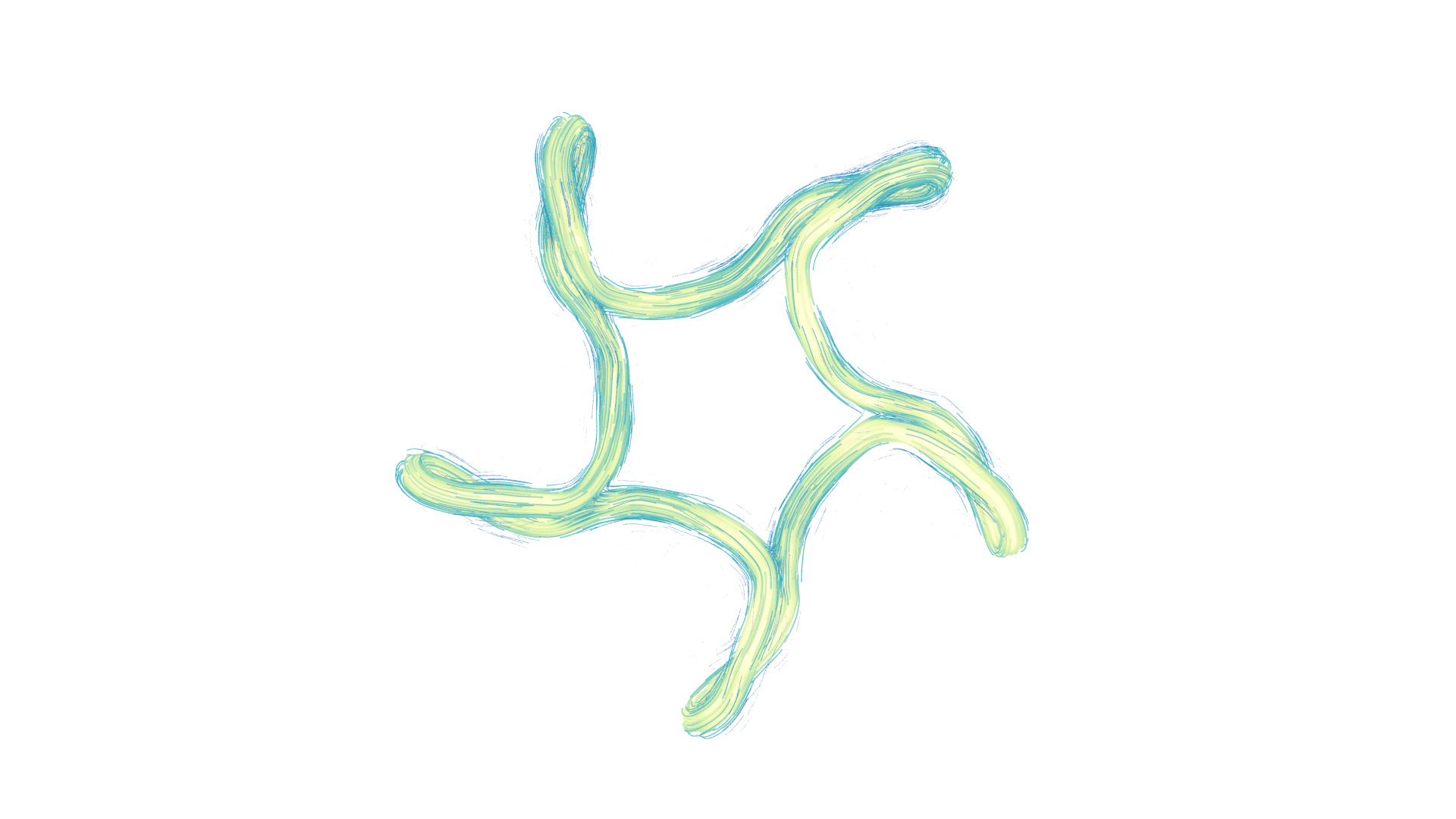}
    \includegraphics[trim={450px 0 400px 0},clip,width=0.16\linewidth]{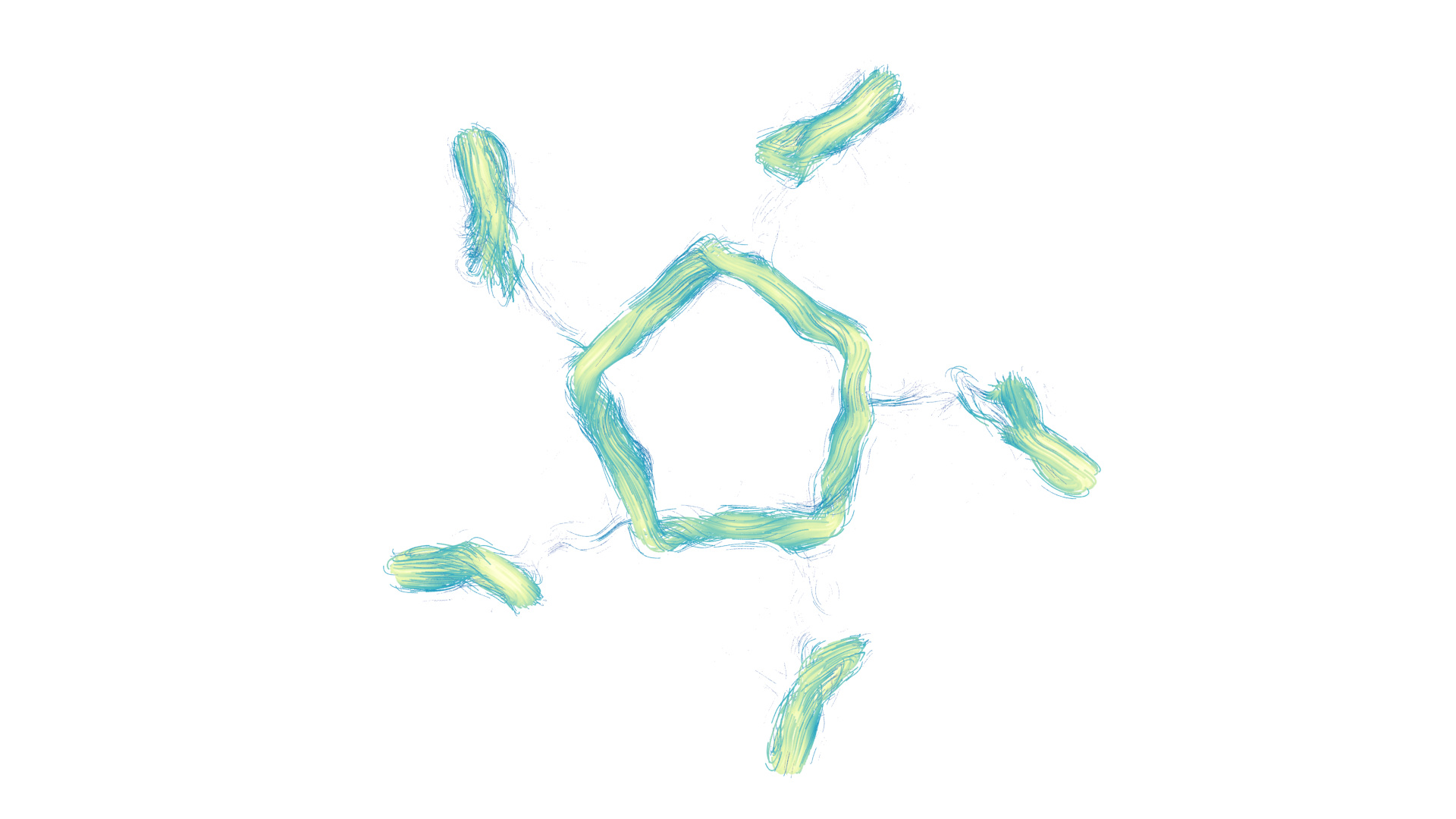}
    \includegraphics[trim={450px 0 400px 0},clip,width=0.16\linewidth]{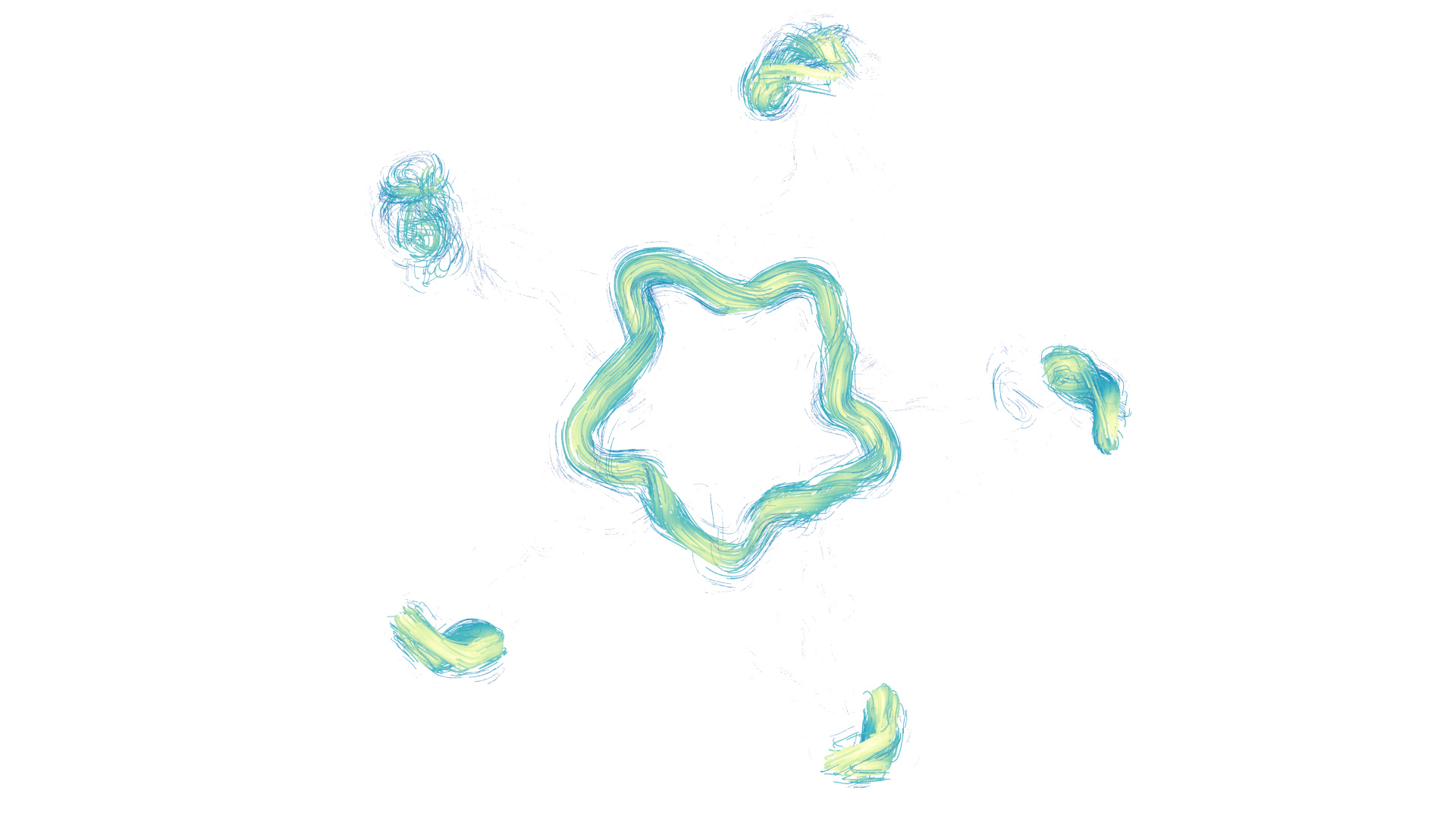}
    \\
    \begin{picture}(0,0)(0,0)
        \put(-260,65){\begin{tikzpicture}
\pgfplotscolorbardrawstandalone[ 
    colormap={myProteinColor}{
        rgb255=(8, 29, 88)
        rgb255=(37, 52, 148)
        rgb255=(34, 94, 168)
        rgb255=(29, 145, 192)
        rgb255=(65, 182, 196)
        rgb255=(127, 205, 187)
        rgb255=(199, 233, 180)
        rgb255=(237, 248, 177)
        rgb255=(255,255, 217)
    },
    point meta min=-1,
    point meta max=1,
    colorbar style={
        width=4pt,
        height=35pt,
        ytick={-1,-0.5,0,0.5,1},
        ytick style={draw=none},
        yticklabels={{0},{},{},{},{10}},
        yticklabel style={font=\scriptsize, xshift=-0.5ex},
        ylabel={\sffamily vorticity norm ($\nicefrac{1}{\text{s}}$)},
        ylabel style={font=\tiny, yshift=25pt}
        }]
\end{tikzpicture}}
        \put(-22,98){\sffamily\scriptsize (a) Side view.}
        \put(-25,0){\sffamily\scriptsize (b) Head-on view.}
        \put(-220,10){\sffamily \scriptsize Frame 0}
        \put(-140,10){\sffamily \scriptsize Frame 72}
        \put(-50,10){\sffamily \scriptsize Frame 144}
        \put(25,10){\sffamily \scriptsize Frame 216}
        \put(115,10){\sffamily \scriptsize Frame 288}
        \put(197,10){\sffamily \scriptsize Frame 360}
    \end{picture}
    \caption{The (1,5)-torus unknot.
    They are referred to as \emph{unknot} as the numbers 1 and 5 are no longer coprime \cite{Maggioni:2010:VEH}.
    With time, the unknot moves forward and stretches apart until the vortex reconnection event happens.
    Our method at a low resolution of $64\times64\times64$ captures this phenomenon.
    }
    \label{fig:unknot_evolution}
\end{figure*}
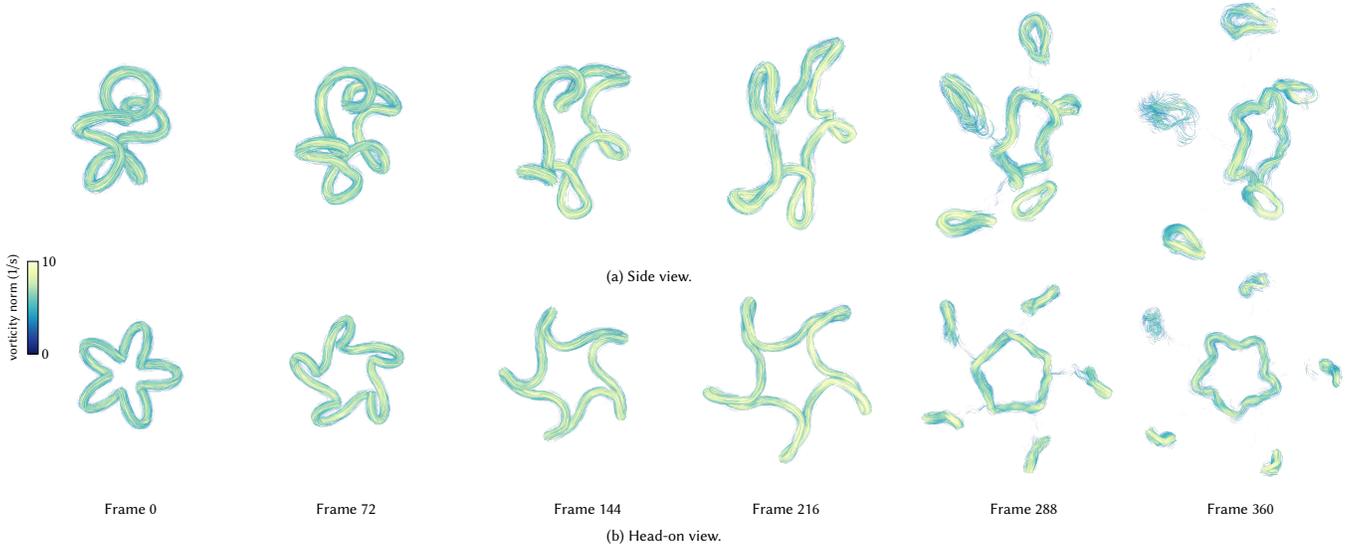

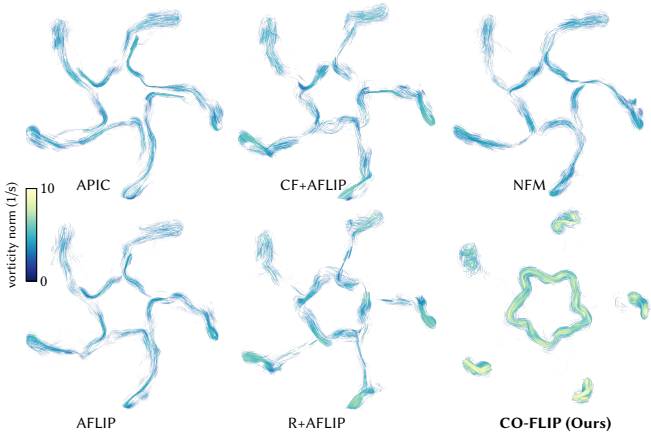
\begin{figure}
    \centering
    \includegraphics[trim={450px 0 400px 0},clip,width=0.32\columnwidth]{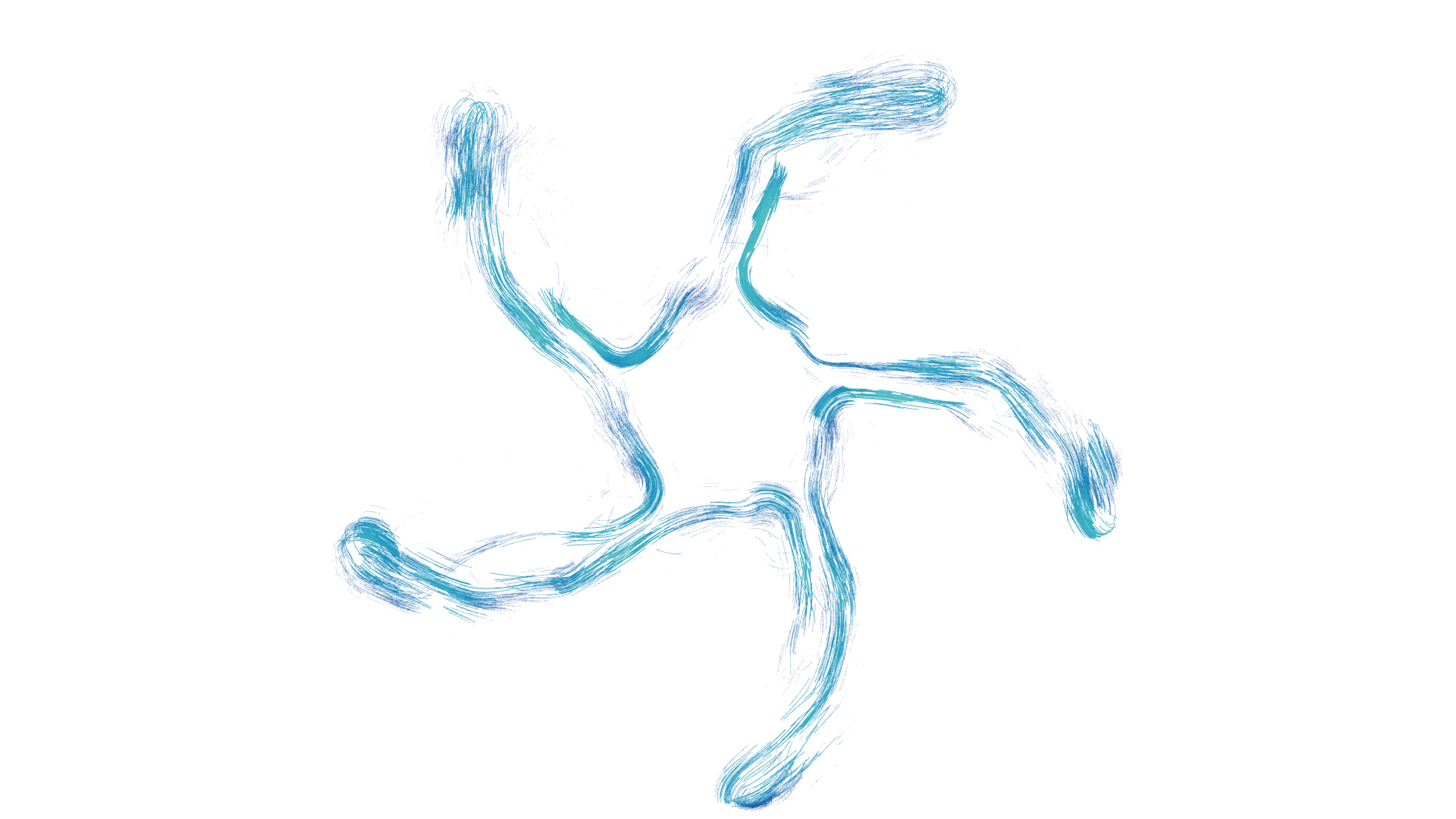}
    \includegraphics[trim={450px 0 400px 0},clip,width=0.32\columnwidth]{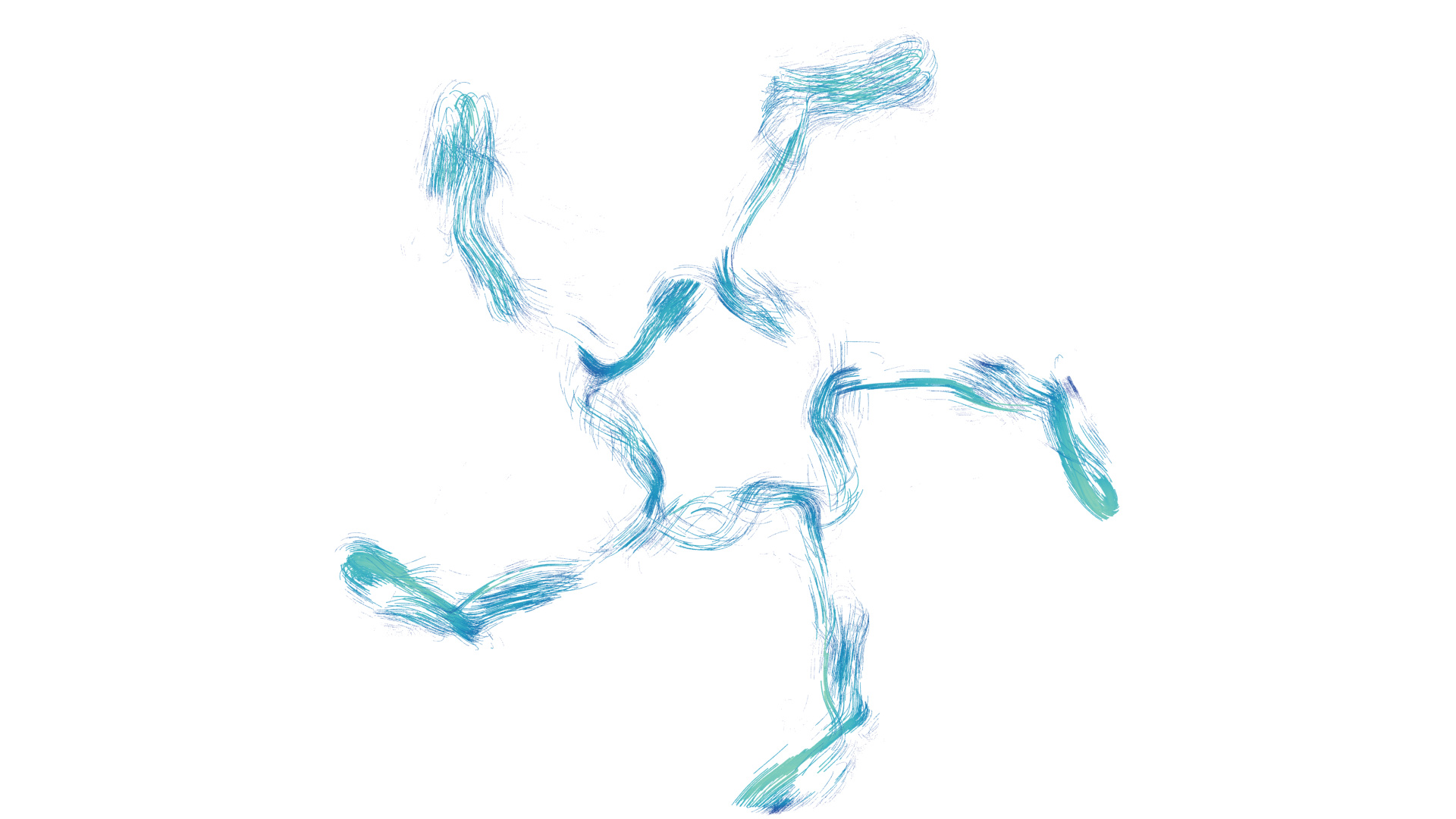}
    \includegraphics[trim={450px 0 400px 0},clip,width=0.32\columnwidth]{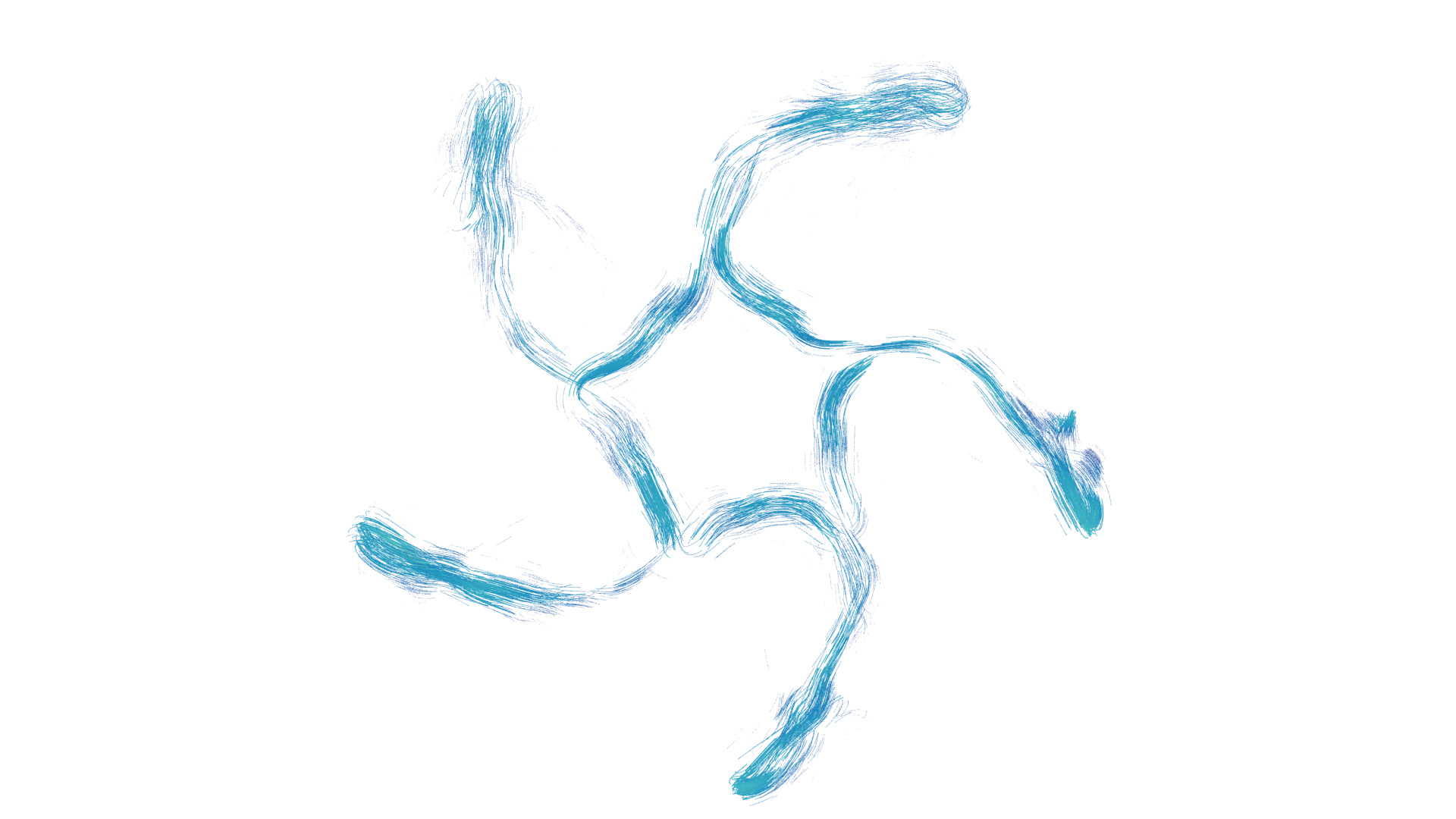}
    \includegraphics[trim={450px 0 400px 0},clip,width=0.32\columnwidth]{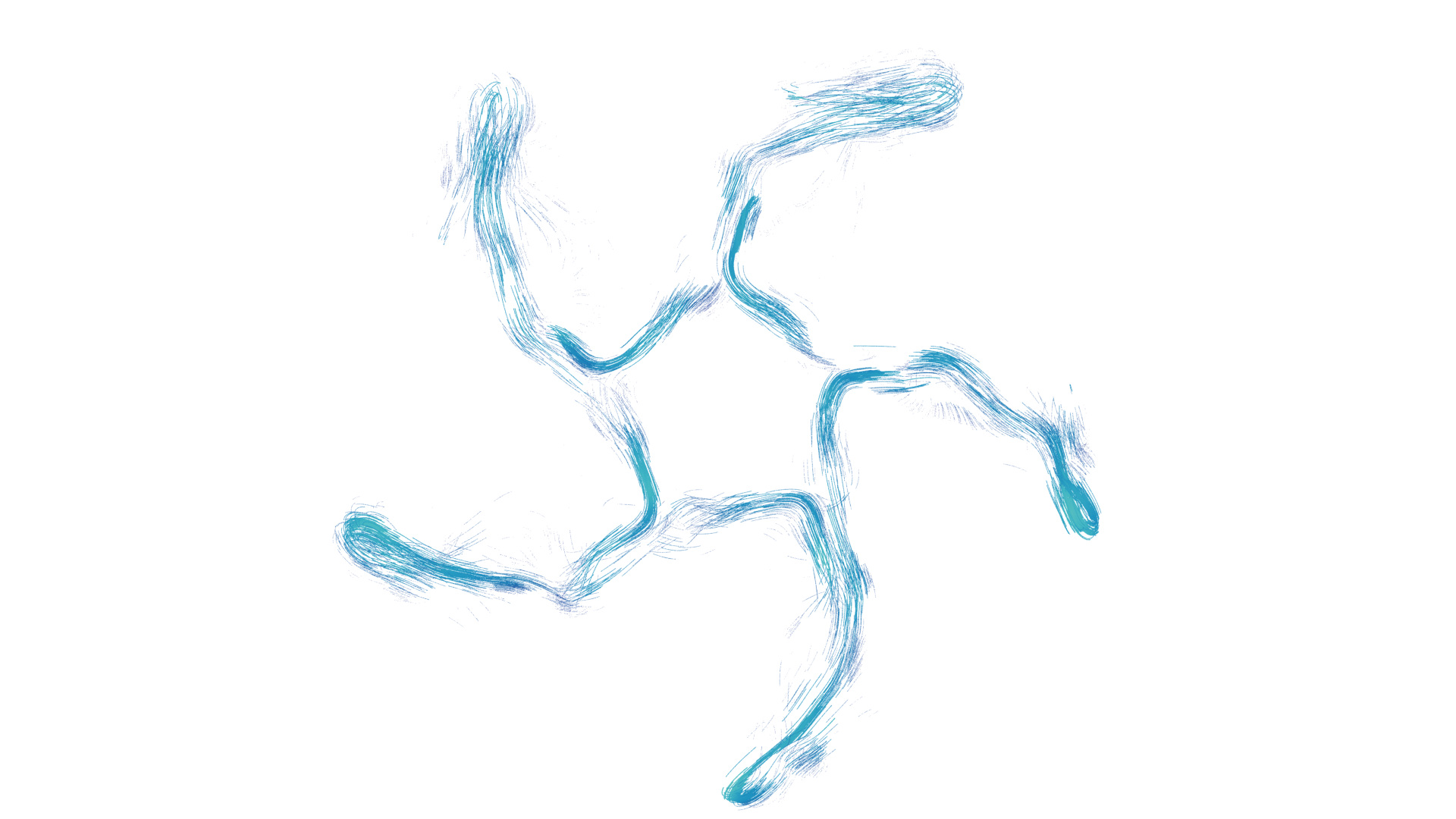}
    \includegraphics[trim={450px 0 400px 0},clip,width=0.32\columnwidth]{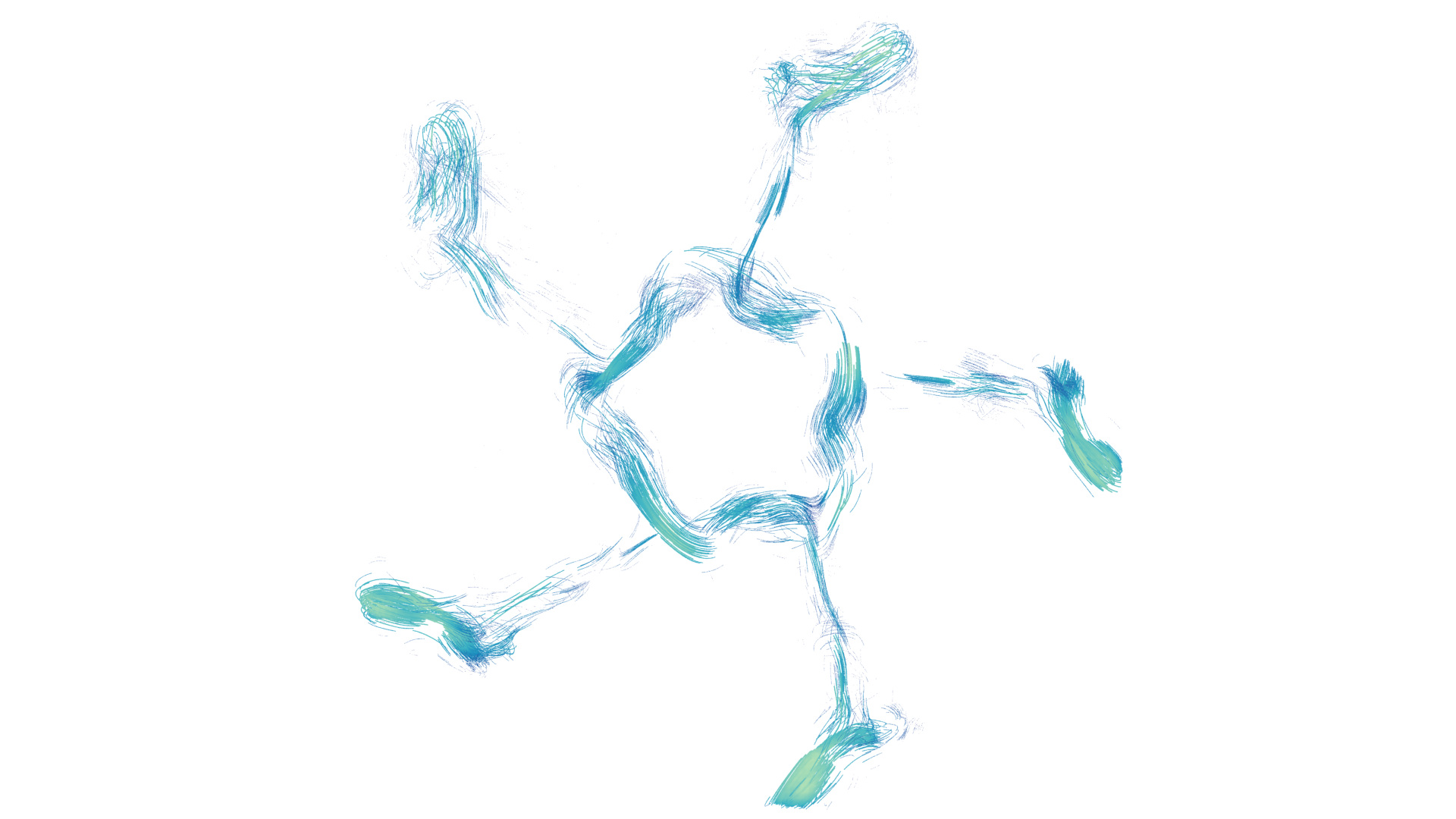}
    \includegraphics[trim={450px 0 400px 0},clip,width=0.32\columnwidth]{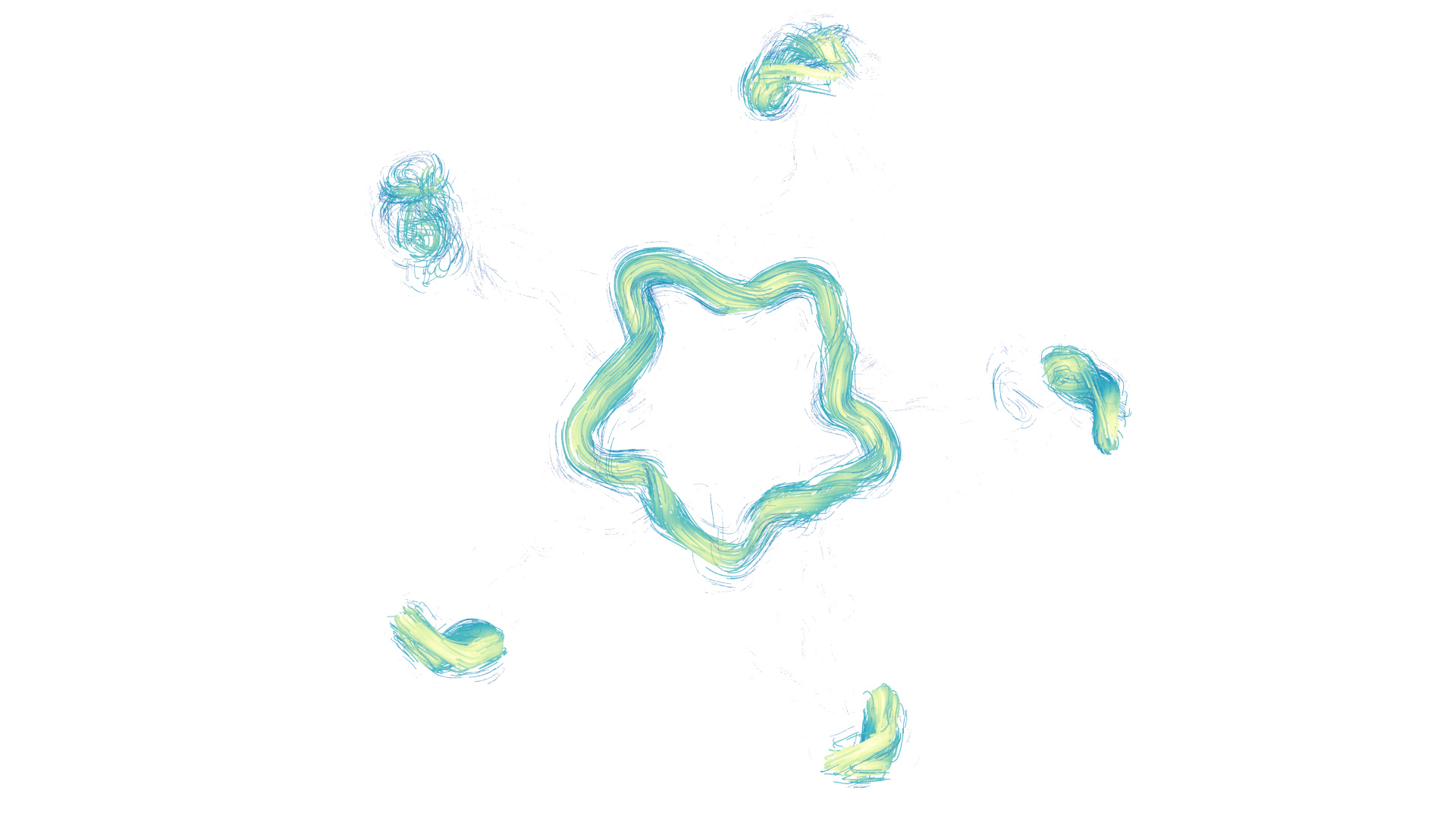}
    \\
    \begin{picture}(0,0)(0,0)
        \put(-130,55){\begin{tikzpicture}
\pgfplotscolorbardrawstandalone[ 
    colormap={myProteinColor}{
        rgb255=(8, 29, 88)
        rgb255=(37, 52, 148)
        rgb255=(34, 94, 168)
        rgb255=(29, 145, 192)
        rgb255=(65, 182, 196)
        rgb255=(127, 205, 187)
        rgb255=(199, 233, 180)
        rgb255=(237, 248, 177)
        rgb255=(255,255, 217)
    },
    point meta min=-1,
    point meta max=1,
    colorbar style={
        width=4pt,
        height=35pt,
        ytick={-1,-0.5,0,0.5,1},
        ytick style={draw=none},
        yticklabels={{0},{},{},{},{10}},
        yticklabel style={font=\scriptsize, xshift=-0.5ex},
        ylabel={\sffamily vorticity norm ($\nicefrac{1}{\text{s}}$)},
        ylabel style={font=\tiny, yshift=25pt}
        }]
\end{tikzpicture}}
        \put(-101,95){\sffamily \scriptsize PolyPIC}
        \put(-28,95){\sffamily \scriptsize CF+PolyFLIP}
        \put(65,95){\sffamily \scriptsize NFM}
        \put(-100,5){\sffamily \scriptsize PolyFLIP}
        \put(-20,5){\sffamily \scriptsize R+PolyFLIP}
        \put(60,5){\sffamily \scriptsize \textbf{CO-FLIP (Ours)}}
    \end{picture}
    \caption{Comparison of torus (1,5)-unknot experiment.
    Note that our method conserves energy and vortical structures throughout the simulation, while traditional methods lose both energy and helicity.}
    \label{fig:unknot}
\end{figure}

\paragraph{Glimpse of the Mathematics behind CO-FLIP}
Beyond improving the results of FLIP, we elucidate that the
CO-FLIP algorithm naturally arises when considering a structure-preserving discretization of the geometric mechanics formulation of the incompressible Euler equations.

In the continuous theory, the Euler equation is a Hamiltonian system. The phase space is given by the dual space of divergence-free vector fields, which is a Poisson space, and the Hamiltonian is given by the kinetic energy of the velocity field.

Towards the discrete theory, a mimetic interpolation (\ref{item:intro-mimetic}) gives us a finite-dimensional subspace of divergence-free vector fields to work with, which constitute grid velocities. We define a discrete Euler equation on the original continuous dual space of divergence-free vector fields, using the finite-dimensional Hamiltonian function defined over grid velocities.
This discrete Euler flow has energy conservation and circulation conservation in the same way the continuous Euler flows obey these laws.
Using a technique involving momentum maps, we can emulate this dynamical system on the Lagrangian coordinates, which leads to the particle-in-cell hybrid algorithm of CO-FLIP.

\paragraph{Overview}
After a literature survey in \secref{sec:relatedworks}, we first provide a high-level overview of the CO-FLIP algorithm in 
\secref{sec:Method}.
Next, we elucidate the theoretical foundation of CO-FLIP in \secref{sec:Theory1}.
In \secref{sec:Theory2} we focus on the mimetic interpolations, their constructions, and their implications. 
The remainder of the paper consists of implementation details (\secref{sec:Implementation}) and numerical examples (\secref{sec:Results}).
\section{Related Work}
\label{sec:relatedworks}

\begin{figure*}
    \centering
    \includegraphics[trim={150px 100px 150px 0},clip,width=0.33\linewidth]{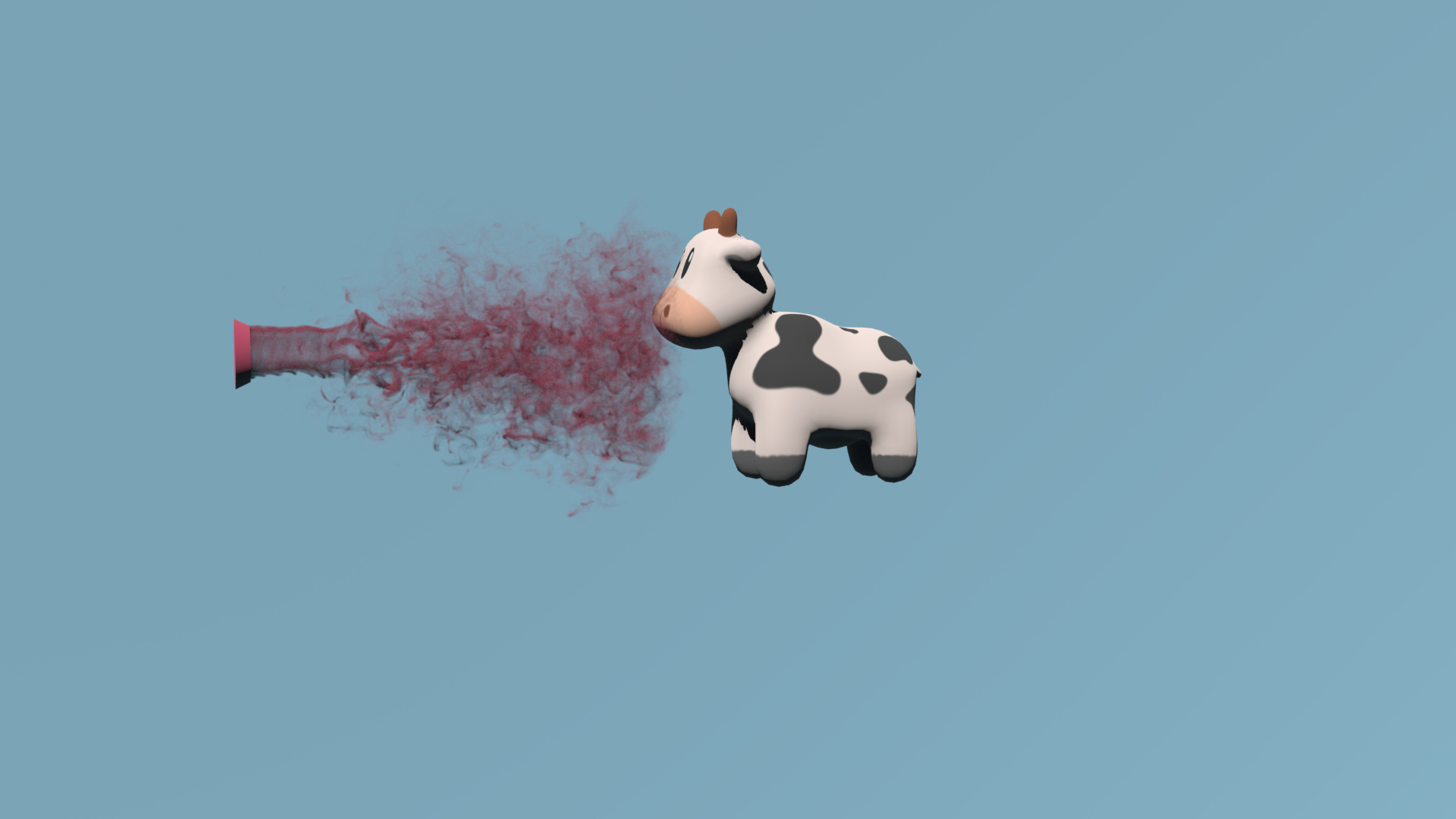}
    \includegraphics[trim={150px 100px 150px 0},clip,width=0.33\linewidth]{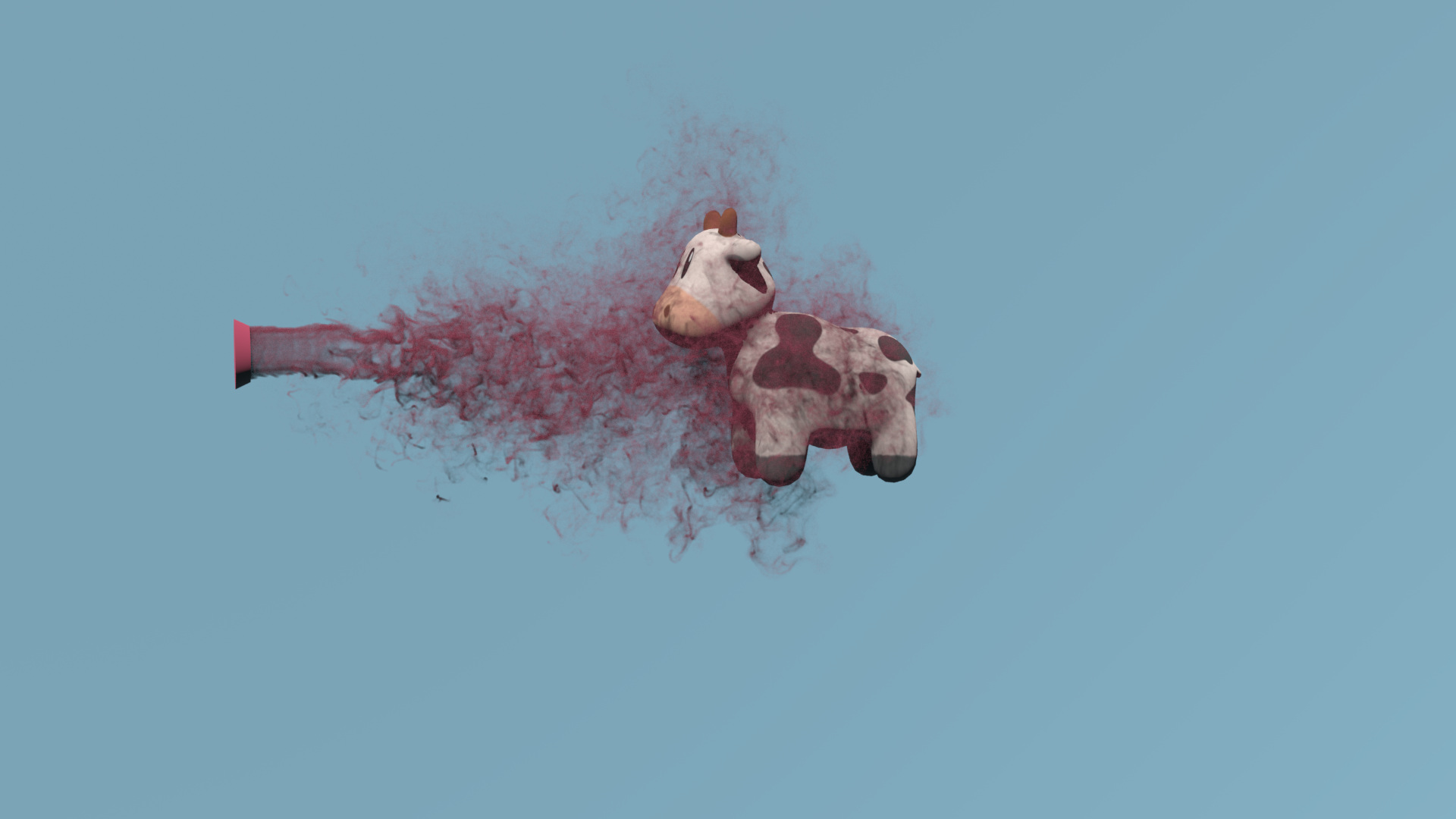}
    \includegraphics[trim={150px 100px 150px 0},clip,width=0.33\linewidth]{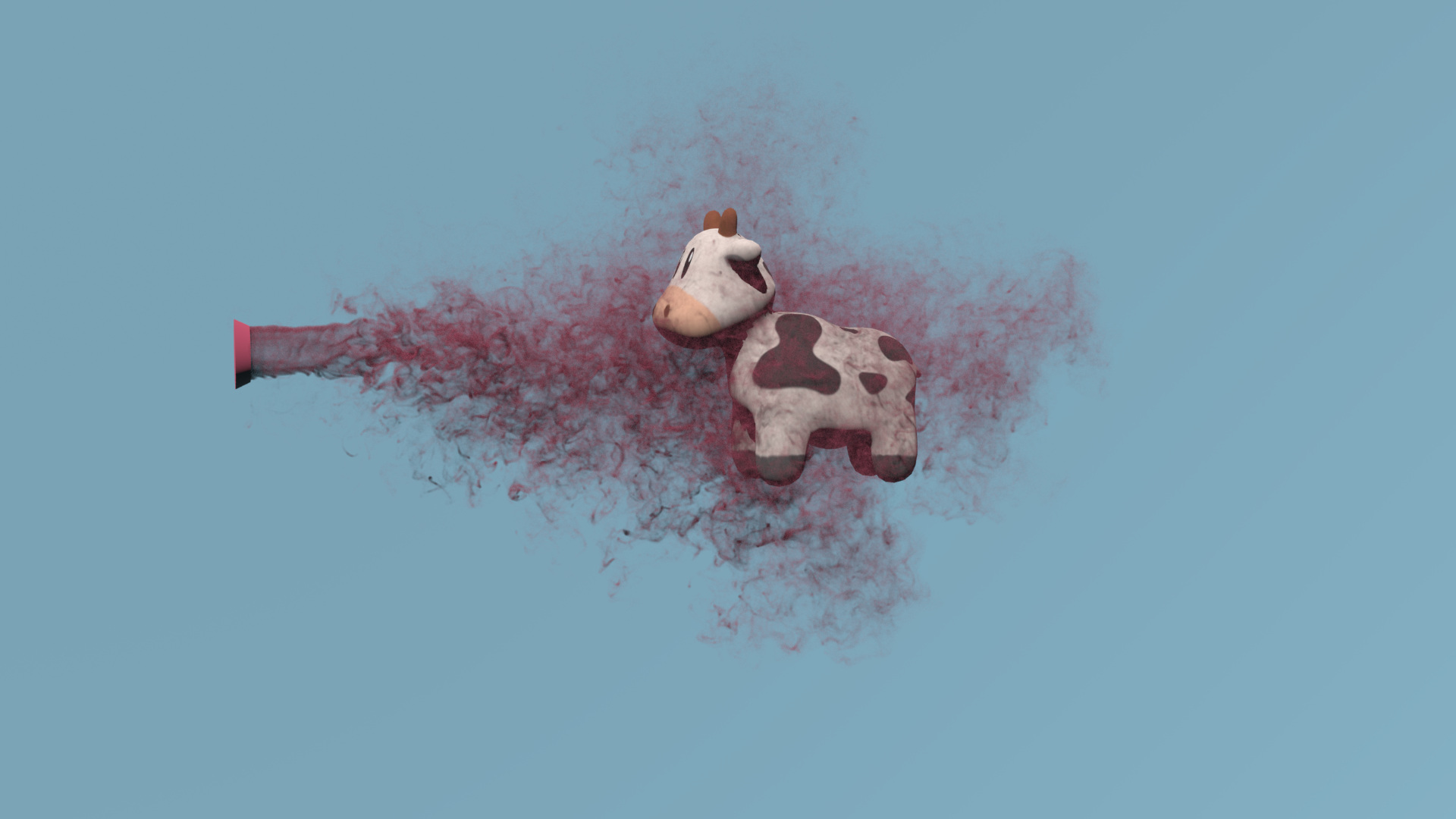}
    \caption{A nozzle shooting out smoke and hitting a Spot-shaped \cite{Crane:2013:RFC} obstacle.
    Note the intricate vortical structures our method (CO-FLIP) can create at a low resolution of $128\times64\times64$.}
    \label{fig:spot}
\end{figure*}

\subsection{Vorticity Conservation in Eulerian Fluid Simulation}
Since the early work of \citet{Stam:1999:SF} in the computer graphics field, the need for maintaining vorticity in the fluid simulations has received much attention. Indeed, the nonlinear advection operator present in the Navier-Stokes and Euler equations, when discretized on a grid, introduces much artificial viscosity. A variety of techniques have been developed to mitigate this dissipation. In particular, back-and-forth error compensation and correction (BFECC) and modified MacCormack methods \citep{Dupont:2003:BFECC, Kim:2005:Flowfixer, Selle:2008:MCM}, vorticity confinement method \cite{Fedkiw:2001:VSS}, and advection-reflection \citep{Zehnder:2018:ARS} are few methods which tackle energy loss and restore the simulation vorticity otherwise lost. 

Another class of approaches simulate the \textit{vorticity equation}, which describes the evolution of the vorticity, instead of velocity \citep{Selle:2005:VPM, Elcott:2007:SCP, Zhang:2015:RMV}. By removing the splitting error between the projection and advection steps, these methods generally do well at conserving vorticity. However, they run into instability problems in three dimensions due to the stretching term, and further complicate handling of harmonic components  \citep{Yin:2023:FC}. 

Finally, a recent line of work has observed that the Euler equations can be recast as the Lie advection of the velocity covector \citep{Nabizadeh:2022:CF} or impulse \citep{Feng:2022:IFS, Saye:2016:IGM, Yang:2021:CGF}. Such methods provide a velocity-based formulation which similarly removes the splitting error without the complication of handling harmonic components. 
 
In these semi-Lagrangian methods, Lagrangian features of the fluid are tracked over time, along the characteristics of the fluid flow. As such, an effective method of further removing dissipation is to track characteristics further back in time.  This idea of the long-term method of characteristic mappings (MCM) has been developed in multiple works \citep{Tessendorf:2011:CMF, Sato:2018:SALS,Qu:2019:ECF}. 
A recent work \cite{Deng:2023:FSN} uses a neural network to compress the large amount of velocity data needed to evaluate these long-term characteristics. 
In this work, we advocate an alternative that requires no compression but produces similarly accurate long-term characteristics.
We, instead, opt to use particles to track Lagrangian features of the fluid. In particular, we track velocity covectors and deformation gradients on Lagrangian particles, so that the Lie advection of the velocity covector becomes solving an ODE on particles. This both minimizes splitting error, and tracks the long term features of fluids along characteristics, preserving circulation as well as energy.

\subsection{Hybrid Eulerian-Lagrangian Methods}
Particle-in-cell  (PIC) algorithms for simulating fluids were first introduced by \cite{Harlow:1962:PIC}. Since then, they have been used in a variety of scientific computing and computer graphics applications. In particular, Fluid Implicit Particles (FLIP) was developed by \cite{Brackbill:1986:FLIP,Brackbill:1988:FLIP} for use in magnetohydrodynamic calculations, and \cite{Sulsky:1995:MPM} developed the Material Point Method (MPM), which also tracks deformation gradients on particles to handle more sophisticated materials. 

The material point method gained prominence in computer graphics after \citet{Stomakhin:2013:SMPM} used it for snow simulation, and subsequently has been extended to a wide variety of materials \cite{Stomakhin:2014:AMPM, Ram:2015:MPM-VE, Klar:2016:DPE, Jiang:MPM-cloth, Tampubolon:2017:MSS} as well as scenarios like contact \cite{Han:MPM-contact} and fracture \cite{Wolpher:CDMPM}. There has also been significant work in improving the properties of the P2G transfer \cite{Jiang:2015:APIC,Fu:2017:PolyPIC,Hu:2018:MLS, Fei:2021:ASFLIP} in order to reduce the dissipation involved in PIC, and the instabilities and noise in FLIP \citep{Hammerquist:2017:XPIC}.
The least-squares based P2G transfer proposed in this paper (\eqref{eq:PressureSolve}) uses a simple preconditioned conjugate gradient solver to alleviate both issues. First, the P2G transfer is directly defined as the pseudoinverse of G2P interpolation and, as such, matches its arbitrarily high order of accuracy. Unlike high-order PIC \citep{Edwards:2012:HOPIC}, which used moving least squares reconstruction and WENO interpolation \cite{Jiang:1996:WENO}, our transfers canonically follow from the choice of B-spline subspace. 
Second, the least-squares solve minimizes the instability problem involved with FLIP-based methods, which are due to divergent velocity modes on the particles that are in the kernel of the particle to P2G transfer, and therefore never get projected out \citep{Hammerquist:2017:XPIC, Ding:2020:APIC-MAC}.  

Particle-in-cell methods have also previously been used in graphics for fluid simulation. In graphics, FLIP for fluids was first suggested by \cite{Zhu:2005:ASF}, and subsequent work has focused on maintaining incompressibility and volume conservation in the flow \cite{Cornelius:2014:IISPH-FLIP, Kugelstadt:2019:IDP, Qu:2022:PPIC}. These approaches have all involved a more sophisticated pressure projection that directly involves particle positions in some capacity -- they thus also ensure particles remain distributed evenly. In contrast, our Galerkin Pressure projection (Eq. \eqref{eq:PressureSolve}), does not involve particles directly and improves the incompressibility of fluids by finding the best divergence-free representative in the B-spline in the $L^2$ sense (See Appendix \secref{app:DiscretePressureProjection}). Additionally, the particles remain evenly distributed throughout the domain as they are advected by a point-wise divergence free interpolated velocity field, as discussed in \secref{sec:DivFreeInterpolation}.

\subsection{Geometric Fluid Mechanics}
Our algorithm is derived from physical principles arising in geometric fluid mechanics: the study of fluid mechanics from Lagrangian and Hamiltonian perspectives. 

\begin{figure*}
    \centering
    \begin{subfigure}{0.19\linewidth}
        \centering
        \includegraphics[trim={0 120px 0 100px},clip,width=\linewidth]{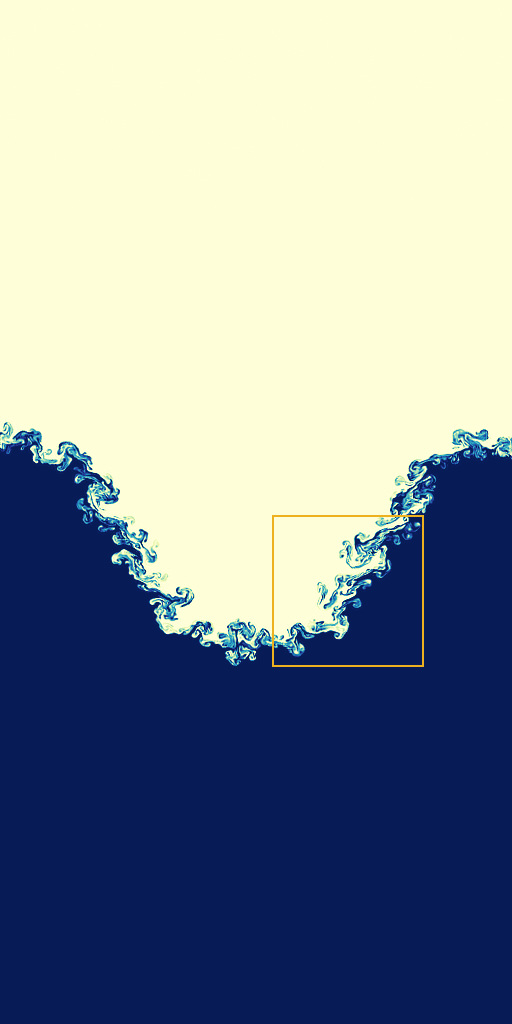}
    \end{subfigure}
    \begin{subfigure}{0.19\linewidth}
        \centering
        \includegraphics[width=0.38\columnwidth]{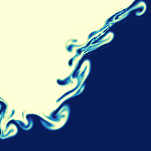}
        \includegraphics[width=0.38\columnwidth]{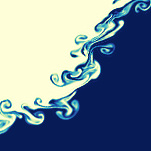}
        \includegraphics[width=0.38\columnwidth]{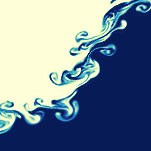}
        \includegraphics[width=0.38\columnwidth]{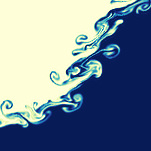}
        \includegraphics[width=.78\columnwidth]{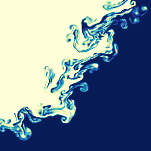}
     \end{subfigure}
     \begin{subfigure}{0.19\linewidth}
         \centering
         \includegraphics[trim={0 120px 0 100px},clip,width=\linewidth]{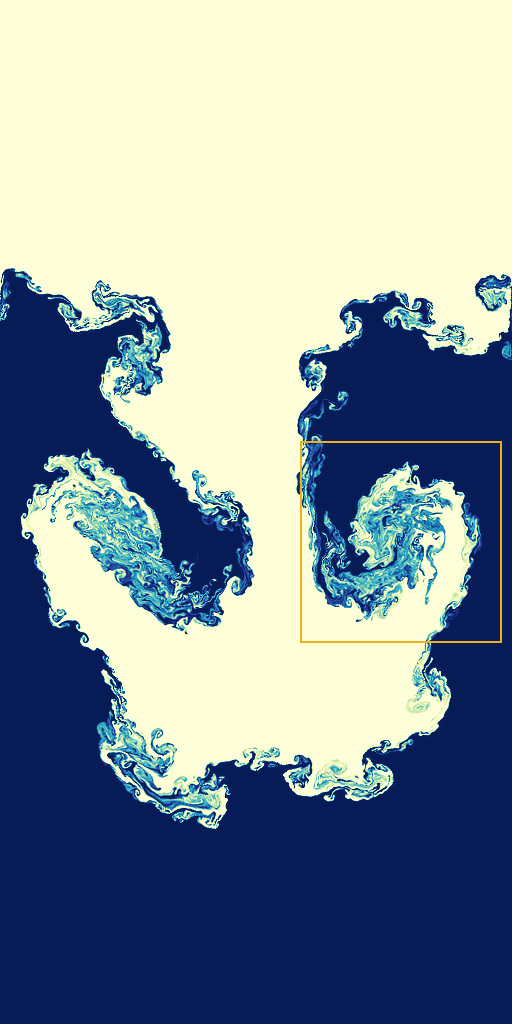}
     \end{subfigure}
     \begin{subfigure}{0.19\linewidth}
         \centering
         \includegraphics[width=0.38\columnwidth]{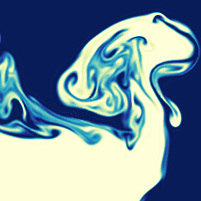}
         \includegraphics[width=0.38\columnwidth]{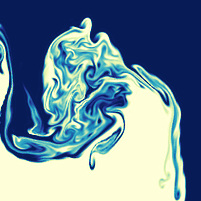}
         \includegraphics[width=0.38\columnwidth]{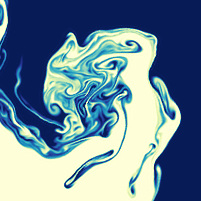}
         \includegraphics[width=0.38\columnwidth]{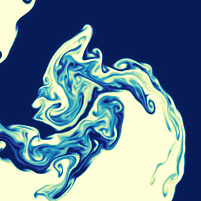}
         \includegraphics[width=.78\columnwidth]{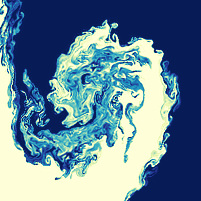}
     \end{subfigure}
     \begin{subfigure}{0.19\linewidth}
         \centering
         \includegraphics[trim={0 120px 0 100px},clip,width=\linewidth]{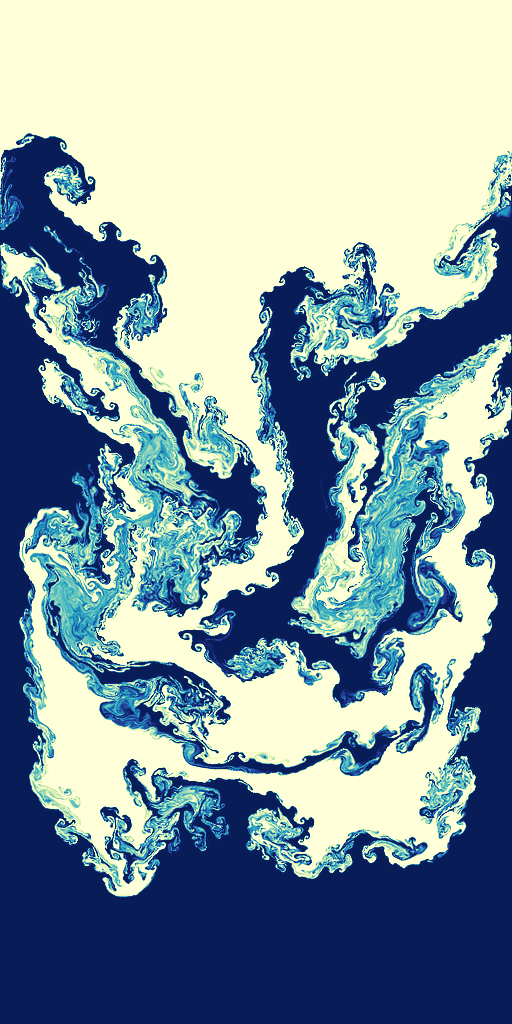}
     \end{subfigure}
     \\
    \begin{picture}(0,0)(0,0)
        \put(-178,14){\sffamily \scriptsize \textcolor{white}{Frame 150}}
        \put(20,14){\sffamily \scriptsize \textcolor{white}{Frame 300}}
        \put(220,14){\sffamily \scriptsize \textcolor{white}{Frame 400}}
        \put(-119,127){\sffamily \tiny \textcolor{white}{PolyPIC}}
        \put(-120,89){\sffamily \tiny \textcolor{white}{PolyFLIP}}
        \put(-89,127){\sffamily \tiny \textcolor{white}{CF+PolyFLIP}}
        \put(-86,89){\sffamily \tiny \textcolor{white}{R+PolyFLIP}}
        \put(-105,14){\sffamily \scriptsize \textcolor{white}{\textbf{CO-FLIP (Ours)}}}
        \put(63,157){\sffamily \tiny \textcolor{white}{PolyPIC}}
        \put(79,120){\sffamily \tiny \textcolor{white}{PolyFLIP}}
        \put(109.5,157.5){\sffamily \tiny \textcolor{white}{CF+PolyFLIP}}
        \put(112,120){\sffamily \tiny \textcolor{white}{R+PolyFLIP}}
        \put(93,81){\sffamily \scriptsize \textcolor{white}{\textbf{CO-FLIP (Ours)}}}
    \end{picture}
     
    \caption{Rayleigh-Taylor instability at frames 150 (left), 300 (middle), and 400 (right).
    The insets provide a closer look at how our method (CO-FLIP) produces and maintains an increasing amount of vortical structures compared with traditional methods.
    This behavior matches better the infinite fractal look results expected from buoyant fluids evolving with zero viscosity.}
    \label{fig:RT_instability}
\end{figure*}
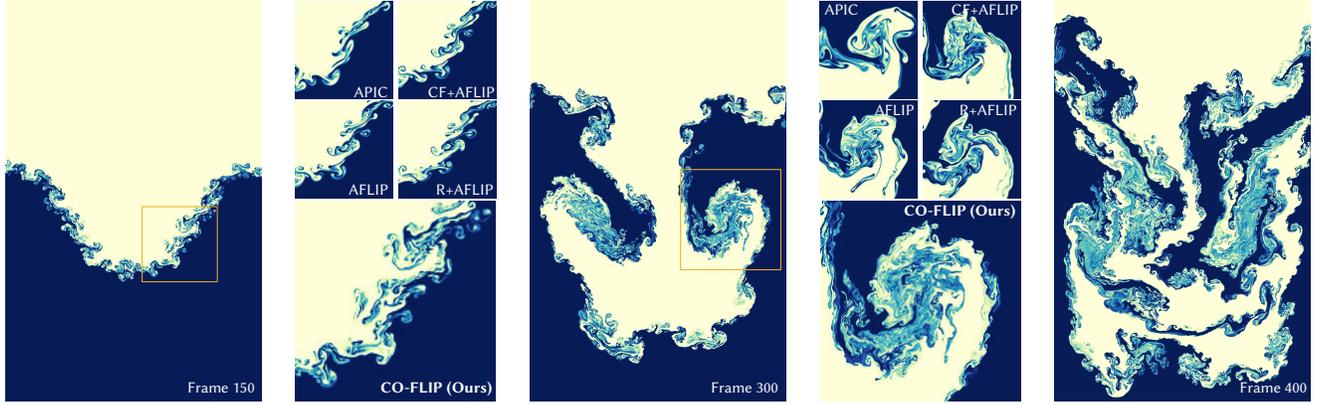

The story of geometric fluid mechanics begins in the 19th century, with \citet{Thomson:1868:OVM} circulation over closed loops is conserved, \citet{Helmholtz:1858:VD} observing that vortex rings form stable structures under the Euler flow, and \citet{Clebsch:1859:IHG} showing that ideal fluids (at least with zero helicity) can be represented as a Lagrangian/Hamiltonian system. A truly general theory, however, only emerged after a breakthrough from \citet{Arnold:1966:GDG}, who observed that the Euler flow is geodesic motion on the Lie group of volume-preserving diffeomorphisms. From this insight, much work has occurred in geometric fluid mechanics. In particular, the term helicity was coined by \citet{Moffatt:1969:DKT} and shown to be a topological invariant related to the knottedness of vortex lines \cite{Moreau:1961:CTF, Moffatt:2014:HSS}. Significant advances were also made by \citet{Ebin:1970:GOD}, who addressed functional analytic issues in this representation and \citet{Marsden:1983:COV} who unified this picture with Clebsch representations using the notion of \textit{reduction} of a Hamiltonian system, and observed that Kelvin circulation conservation is a Noether current associated with particle relabling symmetry.

The Hamiltonian view also suggests the covector or impulse formulation of Euler equation \cite{Oseledets:1989:NWW,Cortez:1995:IBM}, which have recently become influential in graphics following work by \cite{Nabizadeh:2022:CF, Feng:2022:IFS, Saye:2016:IGM}. 
The geometric formulation also gives a close relation between Euler flows and Schr\"odinger flows \cite{Chern:2016:SS,Khesin:2019:GMT} and Clebsch-based formulations, which have been used in graphics for vortical flow simulations \cite{Yang:2021:CGF, Xiong:2022:CMF} and visualizations \cite{Chern:2017:IF}. Recent developments in Geometric Fluid Mechanics include the classification of Casimir invariants in the Euler flow \cite{Izosimov:2017:CC2,Khesin:2022:HUC}. These Casimirs are constant on coadjoint orbits and thus enable us to perform measurements that show coadjoint orbit conservation in \secref{sec:CasimirMeasurement}. 

For a broader discussion of geometric fluid mechanics, we refer the reader to the books by \citet{Arnold:1998:TMH, Holm:2011:GMD} or \citet{Morrison:1998:HD}.

\subsection{Structure-Preserving Discretizations}
The ideas of geometric fluid mechanics have inspired a wide array of both spatial and temporal discretizations, including ours, that aim to more faithfully reflect features of the underlying mathematical equations. 

In computer graphics, \citet{Mullen:2009:EPI} developed a reversible, energy preserving time integrator involving a Newton iteration. More generally, there are a wide variety of approaches to geometric numerical integration \citep{Hairer:2006:GNI} that can be used to conserve energy, symplectic or Poisson structures \citep{Marsden:1999:DEP}, or remain on a Lie Group \citep{Celledoni:2014:LGI}. Such a Lie group integrator was applied in computer graphics by \citet{Azencot:2014:FFS} for fluid simulation on surfaces. We adopt a method similar to \citet{Engo:2001:NILP} for the time integration of Lie-Poisson systems that preserves energy and remains on coadjoint orbits, see \secref{sec:time-integration}. 

On the spatial discretization side, \citet{Pavlov:2011:SPD} approached the problem by viewing discrete flow maps as stochastic, orthogonal matrices. The resulting discretization showed significantly better energy behavior than previous methods. Another line of structure-preserving spatial discretizations involves Nambu brackets, a generalization of Poisson brackets \citep{Nambu:1973:GHD, Nevir:1993:NRI}, which have been extensively used in numerical weather prediction \citep{Gassmann:2008:TCN, Zangl:2013:ICON}.

\subsubsection{Mimetic Interpolation \& Isogemetric Analysis}
Our method takes advantage of a \textit{mimetic interpolation} \citep{Pletzer:2015:CIE,Pletzer:2019:MIV} -- that is, an interpolation that is consistent with a discrete derivative operator. For example, if grid data is discretely divergence free (such as that obtained after a pressure projection), then a mimetic interpolant should be pointwise divergence-free everywhere in the domain. Such interpolants were introduced to computer graphics by \cite{Chang:2022:CurlFlow}, but they have a long history in scientific computing. The particular mimetic interpolants that we adopt are based on B-splines, and have been used extensively in isogeometric analysis \citet{Buffa:2010:IGA, Buffa:2011:IGD, Evans:2013:IGA} and finite element exterior calculus (FEEC) \citet{Arnold:2006:FEEC}. The idea is to observe that the mimetic interpolation properties of B-splines yield a discrete de Rham complex, which in turn yield natural finite element spaces for all of the involved quantities. Thus, FEEC is a generalization of discrete exterior calculus \citep{Hirani:2003:DEC} using higher order function spaces. FEEC has given rise to several conservative simulation algorithms. For example, \citet{Zhang:2022:MEH} developed a dual-field formulation of the Navier-Stokes equations that naturally discretizes to exactly conserve energy and helicity. We use a B-spline to represent grid data (see \secref{sec:Theory2}), as it guarantees that our particle advection action is volume preserving and that our pressure projection is weakly exact (Thm.~\ref{thm:PressureProjectionIsExact}). 

For a more thorough discussion of alternative approaches to mimetic interpolation, we refer the interested reader to \cite{Schroeder:2022:LDP, RoyChowdhury:2024:HOD, Chang:2022:CurlFlow}.

\section{Method}
\label{sec:Method}

In this section, we describe the CO-FLIP method for simulating incompressible and inviscid fluids.
After a brief background of the governing PDE in \secref{sec:EulerEquation}, we describe a particle-in-cell framework, involving a finite dimensional velocity space as the grid structure in \secref{sec:DivFreeInterpolation}, and the particle dynamical system in \secref{sec:ParticleEquationsOfMotion}. See \tabref{tab:notations} for notations used across the paper.

\begin{table}
    \centering
    \caption{Notations and their meanings.}
    \vspace*{-3ex}
    \label{tab:notations}
    \footnotesize
    \setlength{\tabcolsep}{1pt}
\begin{tabularx}{\columnwidth}{p{80pt}p{150pt}}
    \toprule
    \rowcolor{white}
    Notation & Meaning \\
    \midrule
    $M$ & Material space (Lagrangian viewpoint) \\
    $W$ & World space (Eulerian viewpoint) \\
    $(\Sigma,\sigma)$ &  Symplectic manifold, and its associated 2-form \\
    $\cP$ & Index set on the particles \\
    $\fX(W),\fX^*(W)$ & Space of vector fields on $W$, and its dual\\
    $\fX_{\div}(W),\fX_{\div}^*(W)$ & Space of div-free vector fields on $W$, and its dual\\
    $\fB,\fB^*$ & Space of discrete vector fields, and its dual\\
    $\fB_{\div},\fB_{\div}^*$ & Space of discrete div-free vector fields, and its dual\\
    $\cI $ & Interpolation operator (aka G2P map)\\
    $\cI^+$ &  Pseudoinverse of $\cI$ using continuous metric\\
    $\hat{\cI}^+ $ & $\cI^+$ using particle-based metric (aka P2G map)\\
    $\cI^\adjoint $ & Adjoint of $\cI$\\
    $\bP$ & Pressure projection operator\\
    $\bd$ & Exterior derivative operator\\
    $*$ & Continuous Hodge star\\
    $\star$ & Galerkin Hodge star\\
    $\flat, \sharp$ & Musical isomorphisms\\
    $\cW_p $ &$p$-th moment of the vorticity (2D Casimirs)\\
    $\cH$ & Helicity (3D Casimir) \\
    $H$ & Hamiltonian function \\
    SDiff($W$) & Lie group of vol-preserving diffeomorphisms on $W$ \\
    $\Adv $ & Advection action on $\Sigma$ by SDiff($W$)\\
    $\adv $ & Infinitesimal advection action induced by $\Adv$\\
    $J_{\adv}$ & Momentum map induced by $\adv$\\
    \bottomrule
\end{tabularx}
\end{table}

\subsection{Euler Equation}
\label{sec:EulerEquation}
Let \(W \subset \RR^n\), \(n=2\text{ or }3\), be a fluid domain in a Euclidean space.
Let 
\begin{align}
    &\fX(W)\coloneqq\left\{\text{$\continuity^1$ vector fields on $W$}\right\}\\
    &\fX_{\div}(W)\coloneqq\left\{\vec v\in \fX(W)\,\middle\vert\,
    \nabla\cdot\vec v= 0,
    \vec v\cdot\vec n|_{\partial W} = 0
    \right\}
\end{align}
respectively
denote the space of smooth (\(\continuity^1\)) vector fields and the space of divergence-free vector fields satisfying the no-through boundary condition.  Here \(\vec n\) denotes the normal vector of the domain boundary \(\partial W\).
The space \(\fX_{\div}\) models the space of velocity fields of a smooth incompressible flow.
The space \(\fX\) is equipped with an \(L^2\) inner product denoted by \(\llangle\vec v,\vec w\rrangle\coloneqq \int_W\langle\vec v,\vec w\rangle\, d\mu\), where \(d\mu\) is the Euclidean volume form on \(W\).

The incompressible and inviscid fluid with uniform density is governed by the \emph{Euler equation} that describes the time evolution of a divergence-free velocity field \(\vec v(t)\in\fX_{\div}(W)\):
\begin{align}
\label{eq:EulerEquation}
    \tfrac{\partial}{\partial t}\vec v + \vec v\cdot\nabla\vec v = -\nabla p,\quad \vec v(t)\in\fX_{\div}(W),
\end{align}
where the pressure force ``\(-\nabla p\),'' \(p\) being a scalar function, is a general element of the orthogonal complement \(\fX_{\div}(W)^\bot\) in \(\fX(W)\) with respect to \(\llangle\cdot,\cdot\rrangle\) that keeps \(\vec v\) on \(\fX_{\div}(W)\).%
\footnote{We acknowledge that the long-time existence of the Euler equation in the \(\continuity^1\) class is not guaranteed \cite{Elgindi:2019:SSS,Chen:2022:SSS}.}
The equivalent \emph{covector formulation} of the Euler equation \eqref{eq:EulerEquation} describes the time evolution of the velocity covector (1-form) \(\vec v^\flat\in\Omega^1(W)\) \cite{Nabizadeh:2022:CF}:
\begin{align}
\label{eq:EulerEquationCovector}
    \tfrac{\partial}{\partial t}\vec v^\flat + \LD_{\vec v}\vec v^\flat = -d(p-\tfrac{1}{2}|\vec v|^2),\quad \vec v(t)\in\fX_{\div}(W),
\end{align}
where \(\LD_{\vec v}\) is the Lie derivative operator for 1-forms and \(d\) is the exterior derivative.
One can further rewrite \eqref{eq:EulerEquationCovector} in a more general form for 1-forms \(\vec u(t)^\flat\in\Omega(W)\) with \(\vec u(t)\in\fX(W)\) not necessarily divergence-free:
\begin{align}
    \label{eq:ImpulseEquation}
    \tfrac{\partial}{\partial t}\vec u^\flat + \LD_{\vec v}\vec u^\flat = -dq,\quad \vec v = \bP_{\fX_{\div}}\vec u,
\end{align}
where \(\bP_{\fX_{\div}}\colon \fX(W)\to\fX_{\div}(W)\) is the orthogonal projection, and the scalar function \(q\) is arbitrary.
Note that \(\vec u\) and \(\vec v\) differ only by a gradient of an arbitrary scalar function, which is known as a \emph{gauge degree of freedom} in \(\vec u\).  The field
\(\vec u\) in \eqref{eq:ImpulseEquation} is referred to as an \emph{impulse} in the literature \cite{Cortez:1995:IBM,Feng:2022:IFS}.
The 1-form \(\vec u^\flat\) is also regarded as the \emph{circulation field} since the circulation \(\oint_C\vec u^\flat\) on any closed curve \(C\) is well-defined exactly with the gauge degree of freedom in \(\vec u\).
Finally, in terms of vectors \(\vec u\), \eqref{eq:ImpulseEquation} is equivalent to
\begin{align}
    \label{eq:ImpulseEquationVector}
    \tfrac{\partial}{\partial t}\vec u + \vec v\cdot\nabla\vec u + (\nabla\vec v)^\intercal\vec u = -\nabla q,\quad \vec v = \bP_{\fX_{\div}}\vec u.
\end{align}

The CO-FLIP method is a particle--grid hybrid method for simulating \eqref{eq:ImpulseEquation} (or \eqref{eq:ImpulseEquationVector}).
The advected quantity \(\vec u\) is stored on particles, and the velocity \(\vec v\) generating the advection is stored on the grid.  
We first describe the grid structure in \secref{sec:DivFreeInterpolation}, which allows us to write down the ODE system for the particles in \secref{sec:ParticleEquationsOfMotion}. 

\begin{figure}
    \centering
    \begin{picture}(400,180)(0,0)
        \put(0.01,10){\includegraphics[trim={0 0 0 0},clip,width=0.999\columnwidth]{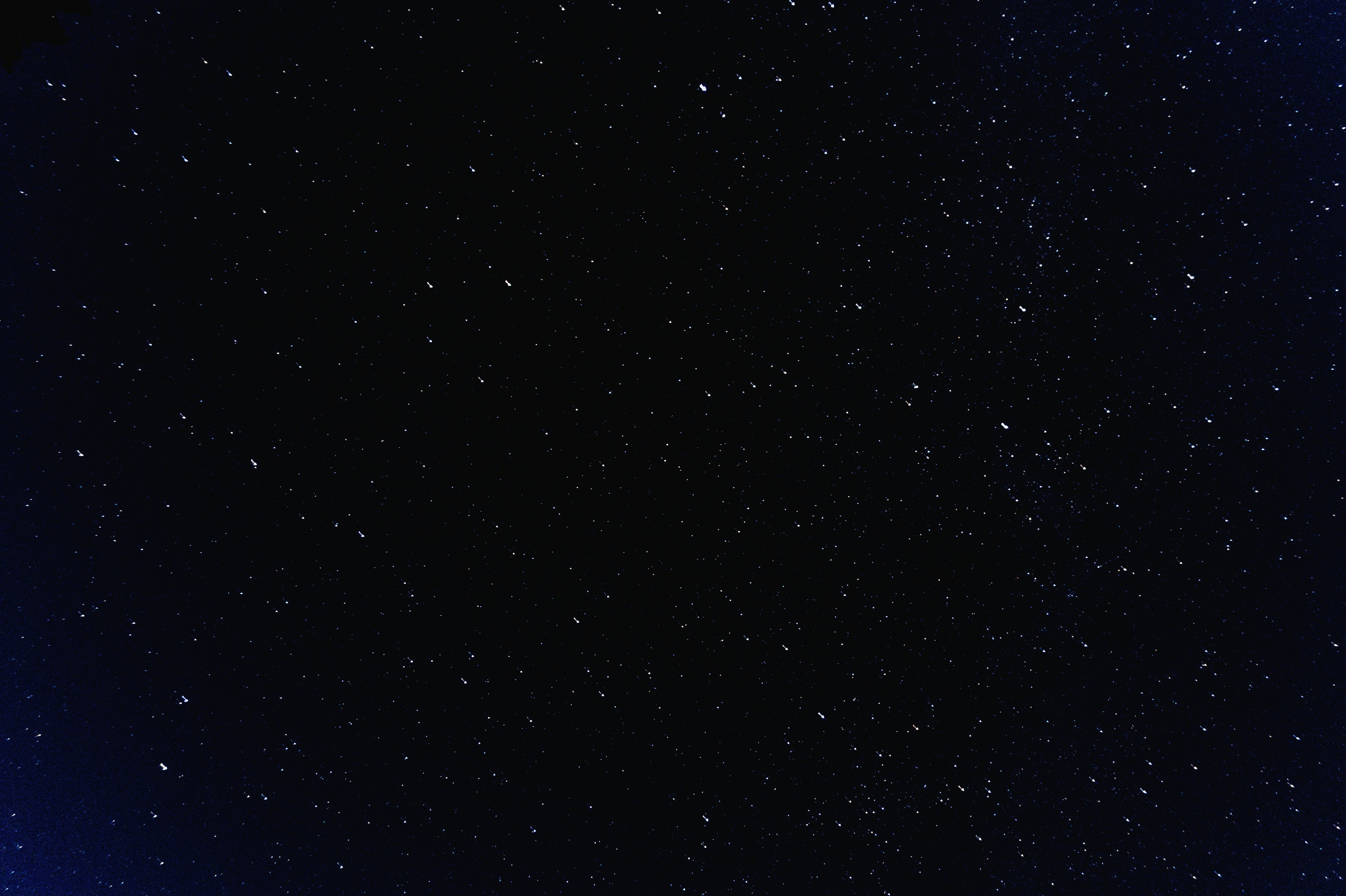}}
        \put(0,0){\includegraphics[trim={0 0 300px 0},clip,width=\columnwidth]{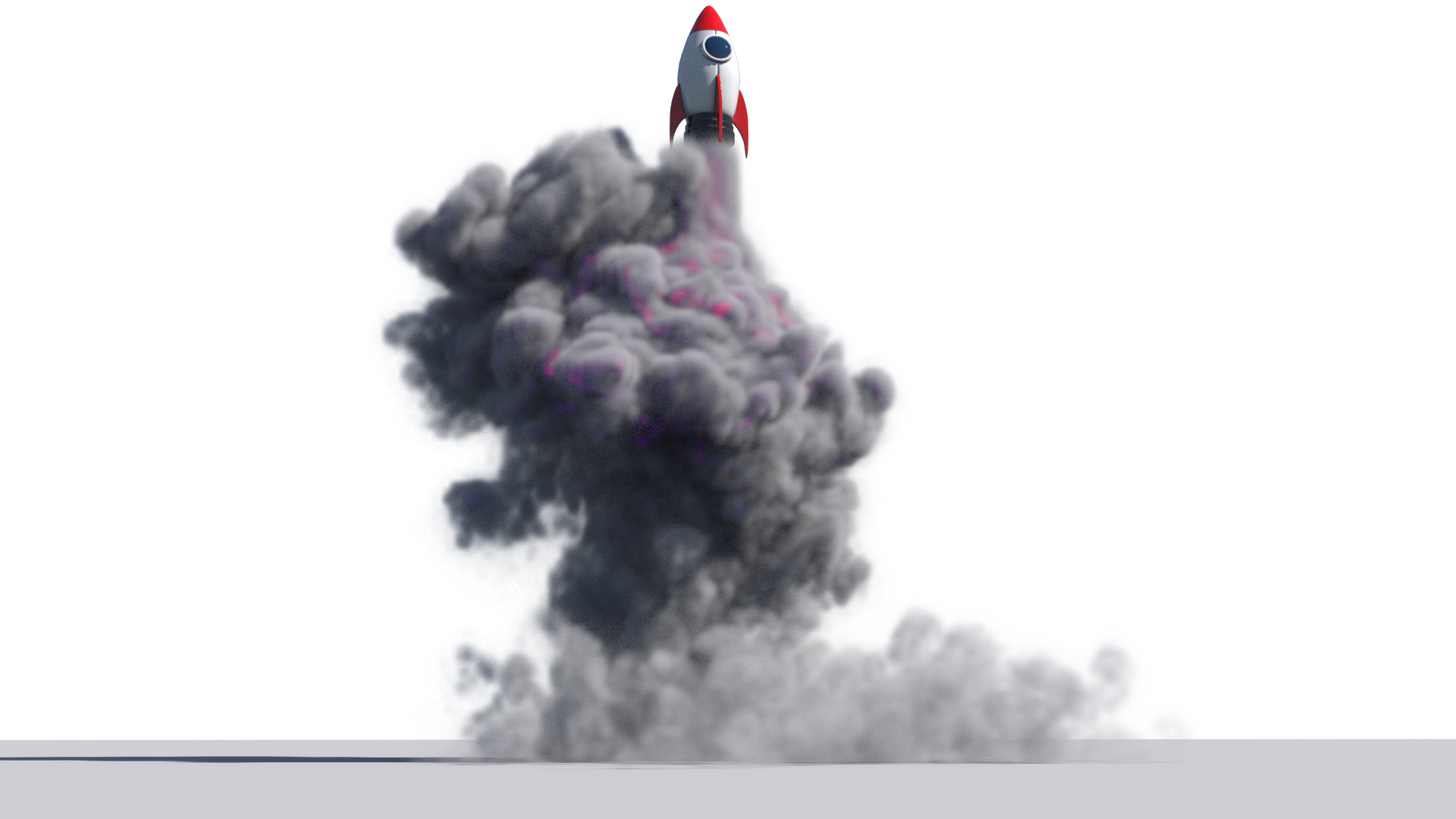}}
    \end{picture}
    \caption{Rocket launching out into space with hot fuel smoke in its wake. Smoke ejection is simulated using our method (CO-FLIP).
    Note the turbulent features of the fuel plume captured despite the low resolution of $64\times64\times64$. The colors represent an artistic visualization of rocket engine exhaust.}
    \label{fig:rocket}
\end{figure}

\subsection{Divergence-Consistent Grid}
\label{sec:DivFreeInterpolation}

As part of the setup, we suppose we have a \emph{divergence-free interpolation scheme} to represent divergence-free vector fields.
Precisely, there is a finite \(N\)-dimensional vector space \(\fB\), representing a space of discrete vector field data (\eg\@ over a grid covering \(W\)).
The space \(\fB\) is equipped with a discrete divergence operator 
\begin{align}
\label{eq:DiscreteDivergenceOperator}
    \bd\colon\fB\xrightarrow{\rm linear}\fD
\end{align}
where \(\fD\) denotes some space of discrete divergence data.
The discrete divergence operator \eqref{eq:DiscreteDivergenceOperator} 
defines the space
\begin{align}
    \fB_{\div}\coloneqq \ker(\bd)\subset\fB
\end{align}
of discrete divergence-free vector fields.
What is referred to as
an \emph{interpolation scheme} is an operator that realizes each element of \(\fB\) as a smooth vector field
\begin{align}
\label{eq:cIAsInterpolationScheme}
    \cI\colon\fB\xrightarrow{\substack{\text{linear,}\\\text{injective}}}\fX(W),\quad
    \bff = 
    \begin{bsmallmatrix}
        f_1\\
        \vert\\
        f_{N}
    \end{bsmallmatrix}
    \mapsto
    \cI(\bff)=
    \sum_{\sfi = 1}^{N}f_{\sfi}\vec b_{\sfi}
\end{align}
via a set of given basis vector fields \(\vec b_1,\ldots,\vec b_N\in\fX(W)\).  
A \emph{divergence-free interpolation scheme} is an interpolation \(\cI\)
so that discrete divergence free vector fields are realized as smooth divergence free vector fields
\begin{align}
\label{eq:DivergenceFreeInterpolation}
    \cI(\fB_{\div})\subset\fX_{\div}(W).
\end{align}
In other words, the discrete and continuous divergence operators are consistent through the interpolation.

In \secref{sec:Theory2} we construct a divergence-consistent interpolation method from a discrete grid using B-splines.
In general, such interpolation methods that are consistent with the exterior derivatives (\(\grad\), \(\curl\), \(\div\)) are known as a \textbf{mimetic interpolation}, also detailed in \secref{sec:Theory2}.
In this section, we assume a divergence-free interpolation scheme \(\cI\) is given, so that we can describe the algorithm in terms of elementary linear algebra.

Once \((\fB,\fB_{\div},\cI)\) is given, we have an induced inner product structure on \(\fB\) from the \(L^2\) inner product structure of \(\fX(W)\) through the map \(\cI\colon\fB\to\fX(W)\):
\begin{align}
\label{eq:InnerProductOnB}
    \langle \bff,\bg\rangle_{\fB} = \bff^\intercal\star\bg \coloneqq \llangle \cI(\bff),\cI(\bg)\rrangle = \int_W\langle\cI(\bff),\cI(\bg)\rangle\,d\mu
\end{align}
for all $\bff, \bg \in \cB$, where \(\star\) denotes the inner product matrix, which is an \(N\)-by-\(N\) symmetric positive definite matrix explicitly given by \(\star_{\sfi\sfj} = \llangle\vec b_\sfi,\vec b_\sfj\rrangle\).
Let \(\bP_{\fB_{\div}}\colon\fB\to\fB_{\div}\) denote the orthogonal projection to \(\fB_{\div}=\ker(\bd)\) with respect to \(\langle\cdot,\cdot\rangle_{\fB}\), given by
\begin{align}
\label{eq:DiscretePressureProjection}
    \bP_{\fB_{\div}} = \id {}-{} \star^{-1}\bd^\intercal (\bd\star^{-1}\bd^\intercal)^{-1} \bd,
\end{align}
derived in \appref{app:DiscretePressureProjection}.
The discrete projection \eqref{eq:DiscretePressureProjection} is analogous to the continuous pressure projection \((\id - \nabla\Delta^{-1}(\nabla\cdot)\)).
One may equivalently write \eqref{eq:DiscretePressureProjection} in terms of streamfunctions, which is detailed in \secref{sec:StreamformSolve}.
We refer \(\bP_{\fB_{\div}}\)  as the \textbf{discrete pressure projection operator}.
Note that the divergence-consistency \eqref{eq:DivergenceFreeInterpolation} of the interpolation \(\cI\) implies that the discrete pressure projection is \emph{exact} in the weak sense.  That is, it results in the closest-point approximation to the exact pressure projection:
\begin{theorem}
\label{thm:PressureProjectionIsExact}
The resulting vector field \(\cI(\bP_{\fB_{\div}}\bff)\) by the discrete pressure projection is, among \(\cI(\fB_{\div})\), the closest element to  the continuous pressure projection \(\bP_{\fX_{\div}}\cI(\bff)\) performed in \(\fX(W)\).  That is, 
    \begin{align}
    \left(\cI(\bP_{\fB_{\div}}\bff) - \bP_{\fX_{\div}}\cI(\bff)\right)\mathrel{\bot_{L^2}}\cI(\fB_{\div})\quad\text{for all \(\bff\in\fB\)}.
\end{align}
\end{theorem}
\begin{proof}
    \appref{app:PressureProjectionIsExact}.
\end{proof}
\subsection{Equations of Motion of CO-FLIP}
\label{sec:ParticleEquationsOfMotion}
CO-FLIP is a dynamical system for particles.  Now, we describe the equations of motion for this particle system.
Let \(\cP\) denote the index set for the particles.
The time-dependent state \((\vec x_{\sfp}(t),\vec u_{\sfp}(t))_{\sfp\in\cP}\) of the system consists of the \emph{position}  \(\vec x(t)\colon \cP\to W\) of each particle, and the \emph{impulse} \(\vec u(t)\colon\cP\to \RR^n\) at each particle.

The equations of motion for the particles are given by 
\begin{subequations}\label{eq:COFLIP-ODE}
\begin{numcases}{}
\label{eq:COFLIP-ODEa}
     \tfrac{d}{dt}\vec x_{\sfp}(t) = \vec v|_{\vec x_\sfp(t)}, &\text{$\sfp\in\cP$} \\
\label{eq:COFLIP-ODEb}
    \tfrac{d}{dt}\vec u_{\sfp}(t) = -(\nabla\vec v|_{\vec x_\sfp(t)})^\intercal\vec u_\sfp(t),&\text{$\sfp\in\cP$}\\
\label{eq:COFLIP-ODEc}
    \vec v = \cI\bP_{\fB_{\div}}\hat\cI^+(\vec x_{\sfp},\vec u_{\sfp})_{\sfp\in\cP}\in\fX_{\div}(W),
\end{numcases}
\end{subequations}
as a discretization of \eqref{eq:ImpulseEquationVector} (with \(q=0\)) along advected particles.
Here \(\hat\cI^+\) is the pseudoinverse of \(\cI\) that takes in a state \((\vec x_{\sfp},\vec u_{\sfp})_{\sfp\in\cP}\) and returns 
an element in \(\fB\) by
\begin{align}
\label{eq:ArgminProblemForIDagger}
    \hat\cI^+(\vec x_{\sfp},\vec u_{\sfp})_{\sfp\in\cP}\coloneqq\argmin_{\bff\in\fB}\sum_{\sfp\in\cP}|(\cI(\bff))_{\vec x_\sfp} - \vec u_{\sfp}|^2\mu_\sfp,
\end{align}
where \(\mu_\sfp\) is a particle volume preserved over time.
\begin{theorem}
\label{thm:IDaggerIsLeftInverse}
    \(\hat\cI^+\) is an \emph{exact left-inverse} of \(\cI\) in the sense that
\begin{align}
\label{eq:IDaggerIsLeftInverse}
    \bff = \hat\cI^+\left(\vec x_{\sfp},\cI(\bff)_{\vec x_{\sfp}}\right)_{\sfp\in\cP}\quad
    \parbox{10em}{
    for all \(\bff\in\fB\) and particle positions \((\vec x_\sfp)_{\sfp\in\cP}\),}
\end{align}
provided that \((\vec x_\sfp)_{\sfp\in\cP}\) is sufficiently dense so that the minimizer for \eqref{eq:ArgminProblemForIDagger} is unique.
\end{theorem}
\begin{proof}
    \appref{app:IDaggerIsLeftInverse}.
\end{proof}
In the language of PIC and FLIP, the interpolation \(\bff\in\fB\mapsto(\vec x_{\sfp},\cI(\bff)_{\vec x_{\sfp}})_{\sfp\in\cP}\) is our grid-to-particle (G2P) transfer, while 
the least squares solve \(\hat\cI^+\) is our particle-to-grid (P2G) transfer.
In particular, \eqref{eq:IDaggerIsLeftInverse} states that P2G perfectly reconstructs the grid data if the particle data comes from a G2P transfer.
We detail, in \secref{sec:pseudoinversep2gsolve}, an efficient manner of solving this problem using a preconditioned conjugate gradient (PCG) solver.

Also note that the \(\continuity^0\)-continuous gradient of velocity \(\nabla\vec v\) in \eqref{eq:COFLIP-ODEb} can also be evaluated exactly since \(\vec v\) is in the image of \(\cI\).  For example if \(\vec v = \cI(\bff) = \sum_{\sfi=1}^Nf_\sfi\vec b_\sfi\) then \(\nabla\vec v = \sum_{\sfi = 1}^Nf_\sfi\nabla\vec b_{\sfi}\).

This concludes the CO-FLIP dynamical system as an ODE \eqref{eq:COFLIP-ODE}.

\subsection{Time Integration}
\label{sec:time-integration}
The time integration of \eqref{eq:COFLIP-ODE} is dissected into two levels.  
At the first level, we describe the time integration of \eqref{eq:COFLIP-ODEa} and \eqref{eq:COFLIP-ODEb} for a time-independent velocity field \(\vec v\).  
Second, we describe the coupling with \eqref{eq:COFLIP-ODEc}.

\subsubsection{Advection action}
When a smooth divergence-free vector field \(\vec v\in\fX_{\div}(W)\) is given and fixed over time, the ODE system \eqref{eq:COFLIP-ODEa} and \eqref{eq:COFLIP-ODEb} for \((\vec x_{\sfp}(t),\vec u_\sfp(t))_{\sfp\in\cP}\) can be integrated accurately, efficiently, and in parallel over \(\sfp\in\cP\), using the 4th order Runge--Kutta (RK4) method.
Conceptually, this first level ODE solve is the evaluation of a Lie group action on the state \((\vec x_{\sfp},\vec u_\sfp)_{\sfp\in\cP}\) by an element of the Lie algebra (i.e\acdot\, generator) \(\vec v\in\fX_{\div}(W)\).  Let us denote this advection by a fixed \(\vec v\in\fX_{\div}(W)\) over a given time span \(\tau\geq 0\) by
\begin{align}
\label{eq:AdvDefinition}
\begin{aligned}
    &\Adv_{\vec v}^{\tau}\colon W\times\RR^n\to W\times\RR^n,\\
    &\Adv_{\vec v}^\tau (\vec x^0,\vec u^0)\coloneqq (\vec x(\tau),\vec u(\tau)),\\
    &\text{where }
    \begin{cases}
     {d\over dt}\vec x(t)=\vec v|_{\vec x(t)}\\
     {d\over dt}\vec u(t) = -(\nabla\vec v|_{\vec x(t)})^\intercal\vec u(t)\\
     \vec x(0) = \vec x^0, \vec u(0) = \vec u^0.
    \end{cases}
\end{aligned}
\end{align}
The operator \(\Adv_{\vec v}^\tau\) also acts on the particle state space \((W\times\RR^n)^{|\cP|}\) by mapping the operation in parallel to each \((\vec x_\sfp,\vec u_\sfp)\) for \(\sfp\in\cP\).
Note that we have group structure \(\Adv_{\vec v}^{\tau_1}\Adv_{\vec v}^{\tau_2} = \Adv_{\vec v}^{\tau_1 + \tau_2}\) and \(\Adv_{a\vec v}^{\tau} = \Adv_{\vec v}^{a\tau}\).

\subsubsection{Lie group integrator}
At the second level, we couple the above group action \eqref{eq:AdvDefinition} with \eqref{eq:COFLIP-ODEc} using a Lie group integrator.  That is, we construct a higher order time integration method only by compositions of group actions.
For simplicity in notation, we abbreviate each state by \(\by = (\vec x_{\sfp},\vec u_{\sfp})_{\sfp\in\cP}\in (W\times \RR^n)^{|\cP|}\), and we write \eqref{eq:COFLIP-ODEc} as \(\vec v =  \cI\cF(\by)\) where
\begin{align}
\label{eq:cFDefinition}
    \cF\colon(W\times\RR^n)^{|\cP|}\to\fB_{\div},\quad\cF\coloneqq\bP_{\fB_{\div}}\hat \cI^+.
\end{align}
We adopt a trapezoidal rule.  On a time axis \((\sfn\dt)_{\sfn=0,1,2,\ldots}\) discretized by a uniform size \(\dt>0\), a given state \(\by^{(\sfn)}\) at the \(\sfn\)-th time step determines the state \(\by^{(\sfn+1)}\) at the next time step by solving
\begin{subequations}
\label{eq:ImplicitMidpoint}
\begin{numcases}{}
\label{eq:ImplicitMidpointA}
    \by^{(\sfn+1)} = 
    \Adv^{\dt}_{\cI\bff^*}\by^{(\sfn)}\\
    \label{eq:ImplicitMidpointB}
    \bff^* = \tfrac{1}{2}\left(
    \cF(\by^{(\sfn)})
    +
    \cF(\by^{(\sfn+1)})
    \right).
\end{numcases}
\end{subequations}
The intermediate grid velocity data \(\bff^{*}\in\fB_{\div}\) represents the grid velocity at the midpoint, \ie\@ the \((\sfn + \nicefrac{1}{2})\)-th time step.

In \secref{sec:PropertiesOfCOFLIP} (\corref{cor:ImplicitMidpointOnMIsEnergyPreserving}) we show that \eqref{eq:ImplicitMidpoint} becomes an energy preserving integrator in the limit where the particle density approaches a continuum (the grid and temporal resolution stay fixed during this limit).

A simple method to evaluate the solution to \eqref{eq:ImplicitMidpoint} is by a fixed point iteration:
\begin{algorithm}[H]
\caption{CO-FLIP Trapezoidal Integrator}
\begin{algorithmic}[1]
\Require \(\by^{(\sfn)} = (\vec x_\sfp,\vec u_\sfp)_{\sfp\in\cP}^{(\sfn)}\)\Comment{Particle position and impulse.}
    \State \(\bff^*\gets\cF(\by^{(\sfn)})\);
    \Comment{Projected particle-to-grid}
    \Repeat 
    \State \(\by^{(\sfn+1)}\gets \) \eqref{eq:ImplicitMidpointA};
    \Comment{Advect the particles.}
    \State \(\bff^*\gets\) \eqref{eq:ImplicitMidpointB};
    \Comment{Averaged particle-to-grid velocity.}
    \Until {$\cF(\by^{(\sfn+1)})$ not converged.}
\Ensure \(\by^{(\sfn+1)}\).
\end{algorithmic}
\end{algorithm}
The stop criteria is determined either by a max number of iterations, or when $\nicefrac{\vert\vert\cF(\by^{(\sfn+1)}) - \cF(\by^{(\sfn)})\vert\vert}{\vert\vert\cF(\by^{(0)})\vert\vert}$ has converged to a desired tolerance.
This algorithm is implementation-wise attractive since it only requires calling the explicit advection procedure \eqref{eq:AdvDefinition} and evaluating \(\cF\) \eqref{eq:cFDefinition}, which are familiar subroutines in any PIC or FLIP implementation. It is in contrast to a full Newton's method for \eqref{eq:ImplicitMidpoint} that would involve computing the system's Hessian \cite{Mullen:2009:EPI,Azencot:2014:FFS}.
When only one iteration is performed, this is the forward Euler method or time-splitting method adopted in most particle-in-cell methods.
Performing two iterations corresponds to the higher order time integration, such as explicit trapezoidal or midpoint schemes, adopted in \cite{Narain:2019:SAR,Nabizadeh:2022:CF}.
We run the fixed point iteration until convergence, which typically takes up to 4--6 iterations (note that later iterations cost less, see \figref{fig:timings}).

\subsection{Energy-Based Correction}
\label{sec:EnergyBasedCorrection}
The \(\hat\cI^+\) \eqref{eq:ArgminProblemForIDagger} and the \(\cF\) \eqref{eq:cFDefinition} operators are projection operators, which take a vector field represented by the particles, and send to the image of \(\cI\colon \fB\to\fX(W)\) and \(\cI\colon\fB_{\div}\to\fX_{\div}(W)\).  However, the projection \(\hat\cI^+\) (which is a part of \(\cF\)) is orthogonal only with respect to the metric \(\sum_{\sfp\in\cP}|\cdot|^2\mu_\sfp\) instead of the continuous \(L^2\) metric \(\int_W|\cdot|^2\,d\mu\).
Increasing the particle density essentially improves the approximation of \(\sum_{\sfp\in\cP}|\cdot|^2\mu_\sfp\) to \(\int_W|\cdot|^2\,d\mu\).  In other words, the difference between the two norms reflects a discretization error due to the particles not being infinitely dense.

\begin{figure}
    \centering
%
%
\definecolor{mycolor1}{rgb}{0.00000,0.44700,0.74100}%
\definecolor{mycolor2}{rgb}{0.85000,0.32500,0.09800}%
\definecolor{mycolor3}{rgb}{0.92900,0.69400,0.12500}%
\begin{tikzpicture}

\begin{axis}[%
width=0.55\linewidth,
height=0.4\linewidth,
at={(0,0)},
scale only axis,
bar width=0.8,
xmin=0,
xmax=7,
xticklabel style = {font=\sffamily \scriptsize},
yticklabel style = {font=\sffamily \scriptsize},
xtick={1, 2, 3, 4, 5, 6},
xtick style={draw=none},
xlabel style={font=\color{white!15!black}},
xlabel={\sffamily \scriptsize Number of fixed point iterations},
ymin=0,
ymax=30000,
ylabel style={font=\color{white!15!black}},
ylabel={\sffamily \scriptsize Time taken (ms)},
axis background/.style={fill=white},
legend style={legend cell align=left, align=left, legend plot pos=left, draw=white!15!black}
]
\addplot[ybar stacked, fill=mycolor1, draw=black] table[row sep=crcr] {%
1	14118.4\\
2	9582.4\\
3	5780.4\\
4	2588.6\\
5	1020\\
6	28.2\\
};
\addlegendentry{\sffamily \tiny Streamfunction-Vorticity Solve}

\addplot [color=black, forget plot]
  table[row sep=crcr]{%
0	0\\
7	0\\
};
\addplot[ybar stacked, fill=mycolor2, draw=black] table[row sep=crcr] {%
1	10992.4\\
2	6124.2\\
3	4057.2\\
4	1758.4\\
5	1012.6\\
6	562.2\\
};
\addlegendentry{\sffamily \tiny Momentum-Map Least Squares}

\addplot[ybar stacked, fill=mycolor3, draw=black] table[row sep=crcr] {%
1	409.2\\
2	412.4\\
3	409.2\\
4	408\\
5	410\\
6	411.4\\
};
\addlegendentry{\sffamily \tiny Advection}

\end{axis}
\end{tikzpicture}%
\begin{tikzpicture}

\begin{axis}[%
width=0.3\linewidth,
height=0.4\linewidth,
at={(1.011in,0.642in)},
scale only axis,
bar shift auto,
log origin=infty,
xmin=0.514285714285714,
xmax=2.48571428571429,
xtick={1,2},
xtick style={draw=none},
xticklabels={{CO-FLIP},{CF+PolyFLIP}},
xticklabel style = {font=\sffamily \scriptsize},
yticklabel style = {font=\sffamily \scriptsize},
xlabel={\sffamily \scriptsize Method},
ymode=log,
ymin=1,
ymax=1000000,
yminorticks=true,,
yticklabel pos=right,
ylabel style={font=\color{white!15!black}, at=(current axis.below origin), anchor=south},
ylabel={\sffamily \scriptsize 
},
axis background/.style={fill=white},
legend style={legend cell align=left, align=left, legend plot pos=left, draw=white!15!black},
area legend/.style={
    legend image code/.code={
        \draw[#1](0cm,-0.1cm)rectangle
            (0.25cm,0.1cm) 
        ;
    }
}
]
\addplot[ybar, fill=mycolor1, draw=black, area legend] table[row sep=crcr] {%
1	60085.2\\
2	898.75\\
};
\addlegendentry{\sffamily \tiny Res=$64^3$}

\addplot[ybar, fill=mycolor2, draw=black, area legend] table[row sep=crcr] {%
1	146411\\
2	3203.75\\
};
\addlegendentry{\sffamily \tiny Res=$96^3$}

\end{axis}
\end{tikzpicture}%
    \caption{The average time taken per fixed point iteration during the trapezoidal implicit time integration; for the Trefoil knot experiment in \figref{fig:trefoilknot_evolution} at resolution $64^3$ (left), and a comparison of the total time taken for different resolutions of CO-FLIP and CF+PolyFLIP (right).
    Note that the exponential drop in time taken to run the later iterations.
    Compared to an explicit midpoint solver (first two iterations only), these additional steps only incur 50\% additional computational time.  In comparison to traditional methods, our method takes longer to run, and it scales linearly as number of grid cells increase.}
    \label{fig:timings}
\end{figure}

Here, we propose a simple method to include a correction term that compensates this discretization error.
We will utilize the facts that (i) we do have access to the exact \(L^2\) norm when the data is on \(\fB\) (\cf\@ \eqref{eq:InnerProductOnB}), and (ii) at the limit where particles become a continuum the method \eqref{eq:ImplicitMidpoint} is energy preserving (\corref{cor:ImplicitMidpointOnMIsEnergyPreserving} in \secref{sec:PropertiesOfCOFLIP}).

The general idea goes as follows.
Let \(\bff^{(\sfn)}\approx \cF(\by^{(\sfn)})\), \(\bff^{(\sfn+1)}\approx \cF(\by^{(\sfn+1)})\) represent modified grid velocities projected from the particle data \(\by^{(\sfn)},\by^{(\sfn+1)}\).  This modification is made minimally so that the kinetic energy is conserved:
\(|\bff^{({\sfn+1})}|_{\fB}^2 = |\bff^{({\sfn})}|_{\fB}^2\) for all \(\sfn\), where \(|\cdot|_{\fB}^2\) is given by \eqref{eq:InnerProductOnB} which yields the exact \(L^2\) norm of the interpolated grid velocities.
Note that by the difference of squares formula, the condition \(|\bff^{({\sfn+1})}|_{\fB}^2 - |\bff^{({\sfn})}|_{\fB}^2=0\) translates to 
\begin{align}
    \langle \underbrace{\bff^{(\sfn+1)}-\bff^{(\sfn)}}_{\Deltait \bff},\underbrace{\bff^{(\sfn+1)}+\bff^{(\sfn)}}_{2\tilde\bff^*}\rangle_{\fB} = 0.
\end{align}
To impose this energy conservation condition, we project an original difference \((\Deltait\bff)_{\rm original} = \cF(\by^{(n+1)}) - \bff^{(n)}\) to the orthogonal complement of \(\tilde\bff^{*}\), the intermediate grid velocity projected from the particle data.  This treatment corresponds to the following implicit scheme.

Initialize \(\bff^{(0)} = \cF(\by^{(0)})\).
Given \(\by^{(\sfn)}, \bff^{(\sfn)}\), determine \(\by^{(\sfn+1)}, \bff^{(\sfn+1)}\) by solving
\begin{subequations}
\label{eq:EnergyCorrection}
    \begin{numcases}{}
    \label{eq:EnergyCorrectionA}
        \by^{(\sfn+1)} = \Adv^{\Deltait t}_{\cI\tilde\bff^*}\by^{(\sfn)}\\
    \label{eq:EnergyCorrectionB}
        \tilde\bff^* = \textstyle{1\over 2}\left(\bff^{(\sfn)} + \bff^{(\sfn+1)}\right)\\
    \label{eq:EnergyCorrectionC}
        \bff^{(\sfn+1)} = \bff^{(\sfn)} + \bP_{\bff^{*\bot}}\left(\cF(\by^{(\sfn+1)})-\bff^{(\sfn)}\right)
    \end{numcases}
\end{subequations}
where \(\bP_{\tilde\bff^{*\bot}}\bff = \bff - \tilde\bff^*\langle \tilde\bff^*,\bff\rangle_{\fB}/|\tilde\bff^*|^2_{\fB}\) is the orthogonal projection to the orthogonal complement of \(\tilde\bff^*\).

Similar to \eqref{eq:ImplicitMidpoint}, the solution to \eqref{eq:EnergyCorrection} can be evaluated using a fixed point iteration with an implicit trapezoidal integrator:
\begin{algorithm}[H]
\caption{CO-FLIP Integrator with Energy-Based Correction}
\begin{algorithmic}[1]
\Require \(\by^{(\sfn)},\bff^{(\sfn)}\)
\Comment{Particle state and grid state}
    \State \(\tilde\bff^*\gets\cF(\by^{(\sfn)})\);
    \Comment{Projected particle-to-grid}
    \Repeat
    \State \(\by^{(\sfn+1)}\gets \) \eqref{eq:EnergyCorrectionA};
    \Comment{Advect the particles.}
    \State \(\bff^{(\sfn+1)}\gets\cF(\by^{(\sfn+1)})\);
    \Comment{Projected particle-to-grid.}
    \State \(\tilde\bff^*\gets\) \eqref{eq:EnergyCorrectionB};
    \Comment{Average.}
    \State \(\bff^{(\sfn+1)}\gets \bff^{(\sfn)} + \bP_{\tilde\bff^{*\bot}}(\bff^{(\sfn+1)}-\bff^{(\sfn)})\);
    \Comment{Correction \eqref{eq:EnergyCorrectionC}.}
    \Until {$\cF(\by^{(\sfn+1)})$ not converged.}
\Ensure \(\by^{(\sfn+1)}, \bff^{(\sfn+1)}\).
\end{algorithmic}
\end{algorithm}
We let the fixed point iteration run until convergence with the same criteria as Algorithm 1.
For our experiments, we keep the CFL number close to 0.5, which results in 4-6 iterations to exit for a tolerance of $10^{-9}$ in 2D ($10^{-7}$ in 3D).
As shown in \figref{fig:timings}, by reusing the previous iterations' results as guesses to the solvers for the next iteration, each step of the fixed-point iteration runs progressively faster than the one before.
In practice, we observe an increase in computation time of roughly 50\% compared to an explicit second-order method.
This is a relatively small price to pay for the conservation of energy.
Preserving energy informs us quantitatively that the algorithm is correctly simulating the $\cI$-discrete Euler flow (\secref{sec:DiscreteEulerFlow}). Note that energy conservation is only possible when $\tilde\bff^*$ has sufficiently converged using our implicit integration; otherwise, the explicit scheme cannot benefit from the energy correction (see \figref{fig:implicit_vs_explicit_with_eps_proj}).  We use Algorithm 2 for all of our experiments.

\begin{figure}
    \centering
    \input{figures/Results/2D/ablations/ablation_implicit_vs_explicit_with_eps_proj}
    \caption{Comparison of time integration method used for our method.
    This study compares energy change within a timestep depending on whether our method uses either an implicit or explicit trapezoidal integrator, with or without energy-based correction.
    Note that the energy-based correction only works as expected with an implicit integrator since $\tilde\bff^*$ is correctly evaluated with convergence of the fixed point iteration.}
    \label{fig:implicit_vs_explicit_with_eps_proj}
\end{figure}

\subsection{Discussion}
\label{sec:MethodDiscussion}

We have described the essence of the CO-FLIP method.  
One only needs to provide a divergence-consistent interpolation scheme \(\cI\colon\fB_{\div}\to\fX_{\div}(W)\). Given such an interpolation method, the rest of the template description canonically unfolds into a fluid simulation method.

In particular, the spatial accuracy of the method depends only on the choice of the grid interpolation \(\cI\) (assuming the typical setup that the particles are sufficiently dense compared to the grid).
A high-order interpolation scheme yields a globally high-order fluid solver.  This is in contrast to most fluid solvers, where achieving a certain global accuracy requires designing multiple high-order subroutines, such as particle-to-grid transfer, grid-to-particle transfer, and pressure solves.
We provide concrete constructions of divergence-consistent interpolations of arbitrary order in \secref{sec:Theory2}.
In addition to having a high order of accuracy, the method preserves energy and circulations.  This is a consequence of the associated Hamiltonian system governing the CO-FLIP dynamics.

Some practical implementation details are postponed to \secref{sec:Implementation}.
We discuss an efficient construction of our high-order pressure projection, an adaptive resetting strategy, and details regarding minimizing the instability of FLIP-based methods.
Additional details regarding the construction of the pseudoinverse and streamform-vorticity solves are presented in \secref{sec:performanceconsiderations}.

In the following \secref{sec:Theory1}, we detail the mathematical foundation of the CO-FLIP algorithm.  In particular, we use the framework of geometric fluid mechanics to show that CO-FLIP satisfies the conservation of circulation and energy in a precise manner (\secref{sec:PropertiesOfCOFLIP}).  
The theory relies on the background on Poisson structure and Lie algebras developed in \appref{app:Preliminaries}. 
Readers primarily interested in the implementation may skip to \secref{sec:Theory2}. 
\section{Theory Part 1: A Structure Preserving Fluid Simulator}
\label{sec:Theory1}

In this section, we show that the CO-FLIP method of \secref{sec:Method} naturally arises when considering discretizing the geometric theory of the incompressible Euler equation in a structure-preserving manner.  For further preliminary material regarding geometric mechanics, we refer the readers to \appref{app:Preliminaries}, and for a brief exposition on exterior calculus we refer the readers to \cite{Wang:2023:EC, Yin:2023:FC, Crane:2013:DEC}.

\subsection{Background}
In geometric mechanics, the incompressible Euler equation is formulated as a Hamiltonian system.

The phase space of the Hamiltonian system is given by the dual space \(\fX_{\div}^*(W)\) of the space \(\fX_{\div}(W)\) of divergence-free vector fields.
Note that the dual space of any Lie algebra is equipped with a natural Poisson bracket, called the \textbf{Lie--Poisson bracket} (\secref{sec:DualLieAlgArePoisson}).
Since \(\fX_{\div}(W)\) is a Lie algebra, with the standard Lie bracket for vector fields, our phase space \(\fX_{\div}^*(W)\) is equipped with a Poisson bracket.
The purpose of the Poisson bracket is that it can turn any smooth function defined over the phase space, called a Hamiltonian, into a vector field that describes a Hamiltonian dynamical system on the phase space.
The incompressible Euler equation is the Hamiltonian system whose phase space is \(\fX^*_{\div}(W)\) with the Lie--Poisson bracket, and whose Hamiltonian is the kinetic energy of each state in \(\fX_{\div}^*(W)\).

The fact that the incompressible Euler flow is a Hamiltonian system using the Lie--Poisson bracket structure has some significant implications.
A dual Lie algebra is foliated into lower dimensional \textbf{coadjoint orbits} (see \defref{def:CoadjointOrbits}), and any Hamiltonian vector field on a dual Lie algebra is tangential to each coadjoint orbit and therefore the flow preserves the coadjoint orbits.
In particular, a state in \(\fX_{\div}^*(W)\) following the incompressible Euler flow must stay on the same coadjoint orbit in \(\fX_{\div}^*(W)\).
This conservation of coadjoint orbit, besides energy conservation, is an important structure that we aim to preserve in a structure-preserving discretization.

Let us unpack the above Hamiltonian description with concrete expressions in terms of fluid mechanical variables. 

\subsubsection{Dual space of divergence free vector fields}
\label{sec:DualSpaceOfDivFree}
\begin{proposition}\label{prop:DualSpaceOfDivFree}
\(\fX_{\div}^*(W) = \Omega^1(W)/d\Omega^0(W)\).
\end{proposition}
\begin{proof}
    \appref{app:DualSpaceOfDivFree}.
\end{proof}
That is, each element of \(\fX_{\div}^*(W)\) is an equivalence class of 1-forms, where two 1-forms \(\eta,\eta'\in\Omega^1(W)\) are defined to be equivalent if their difference is exact \(\eta-\eta'\in\im(d)\).
Recall the usual interpretation that a 1-form is an object to be line-integrated along oriented curves in the domain \(W\).  Here, an element \([\eta]\in \fX_{\div}^*(W) = \Omega^1(W)/d\Omega^0(W)\) can be evaluated only along \emph{closed} oriented curves.  Therefore, we call \([\eta]\) a \emph{circulation field}.

The dual pairing between a circulation field \([\eta]\in\fX_{\div}^*(W)\) and a divergence-free vector field \(\vec v\in\fX_{\div}(W)\) is denoted and given by
\begin{align}
\label{eq:CirculationVectorDualPairing}
    \left\llangle[\eta]\,\middle|\,\vec v\right\rrangle = \int_W \langle\eta|\vec v\rangle\,d\mu
\end{align}
where \(\langle\cdot|\cdot\rangle\) is the pointwise dual pairing between a 1-form and a vector.  
One can verify that the right-hand side expression of \eqref{eq:CirculationVectorDualPairing} is independent of the representative \(\eta\in[\eta]\) for the equivalence class \([\eta]\).
The dual pairing expression is relevant when we perform calculus of variations.  For each functional \(F\colon\fX_{\div}^*(W)\to\RR\), the variation \((\nicefrac{\delta F}{\delta[\eta]})\) is a divergence-free vector field \((\nicefrac{\delta F}{\delta[\eta]})\in\fX_{\div}(W)\) given so that
\begin{align}
    \left.{d\over d\varepsilon}\right|_{\varepsilon=0}F([\eta]+\varepsilon[\mathring\eta]) = \left\llangle[\mathring\eta]\,\middle|\,{\delta F\over\delta[\eta]} \right\rrangle\quad\text{for all \([\mathring\eta]\in\fX_{\div}^*(W)\)}.
\end{align}

Note that the \(L^2\) inner product structure \(\llangle\cdot,\cdot\rrangle = \int\langle\cdot,\cdot\rangle\, d\mu\) on \(\fX_{\div}(W)\) induces an isomorphism between \(\fX_{\div}(W)\) and \(\fX_{\div}^*(W)\)
\begin{align}
    \flat_{\fX_{\div}}\colon \fX_{\div}(W)\xrightarrow{\rm linear\, \cong}\fX_{\div}^*(W),\quad \flat_{\fX_{\div}}\vec u = [\vec u^\flat]
\end{align}
where \((\cdot)^\flat\) is the pointwise flat operator that converts vectors to covectors using the metric on \(W\), so that we have \(\llangle\vec u,\vec v\rrangle = \llangle\flat_{\fX_{\div}}\vec u|\vec v\rrangle\) for all \(\vec u,\vec v\in \fX_{\div}(W)\).
The inverse map
\begin{align}
\label{eq:SharpXDiv}
    \sharp_{\fX_{\div}} \coloneqq \flat_{\fX_{\div}}^{-1}\colon \fX_{\div}^*(W)\xrightarrow{\rm linear} \fX_{\div}(W)
\end{align}
is given by that \(\sharp_{\fX_{\div}}[\eta] \) is the unique divergence free vector field \(\vec v\in\fX_{\div}(W)\) such that \([\vec v^\flat] = [\eta]\).  
Explicitly, this \(\vec v\) is the pressure projection of \(\eta^\sharp\) of any representative \(\eta\in[\eta]\).

\begin{figure}
    \centering
    \includegraphics[trim={0 0 0 0},clip,width=1.2\columnwidth,angle=90]{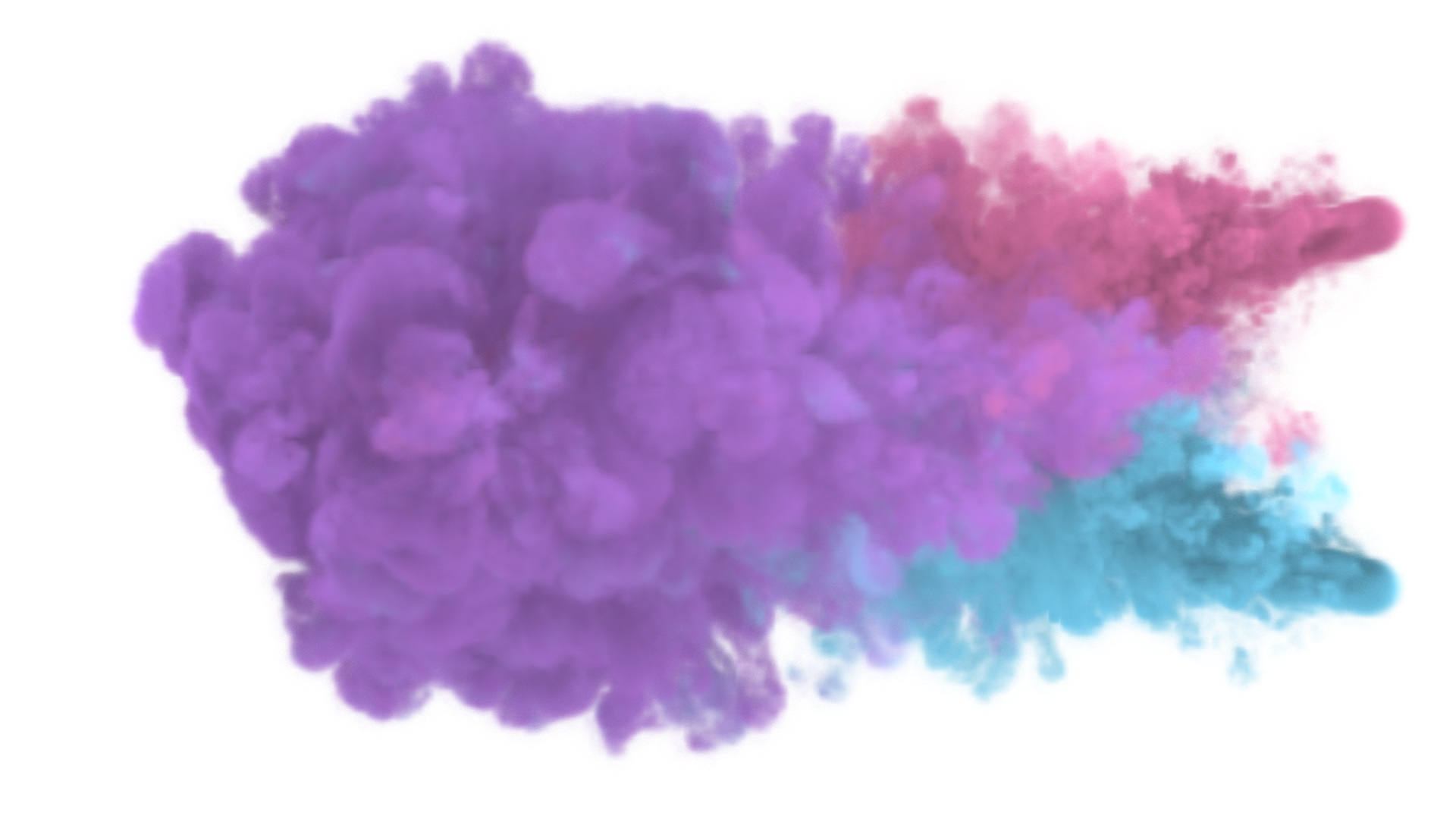}
    \caption{Two nozzles shooting out inks of different colors collide with one another.
    This highly turbulent behavior produces complex vortical structures that are captured using our method (CO-FLIP) at a low resolution of $64\times128\times64$. }
    \label{fig:inkjet}
\end{figure}

\subsubsection{Coadjoint orbits in \(\fX_{\div}^*(W)\)}

\begin{proposition}\label{prop:HamiltonianFlowOnXDivStar}
    The Hamiltonian flow on \(\fX_{\div}^*(W)\) for any Hamiltonian function \(H\colon\fX_{\div}^*(W)\to\RR\) is given by
    \begin{align}
    \label{eq:HamiltonianFlowAsLieAdvection}
        {\partial\over \partial t}[\eta] + \LD_{\frac{\delta H}{\delta[\eta]}}[\eta] = [0].
    \end{align}
\end{proposition}
\begin{proof}
    \appref{app:HamiltonianFlowOnXDivStar}.
\end{proof}
Every Hamiltonian flow will lie on a coadjoint orbit (Definition \ref{def:CoadjointOrbits}).  The next proposition characterizes the coadjoint orbits in terms of circulation fields.
\begin{proposition}
\label{prop:CoadjointOrbitOnCirculations}
    Two states \([\eta_0],[\eta_1]\in\fX_{\div}^*(W)\) lie on the same coadjoint orbit if and only if there exists a volume-preserving map \(\varphi\colon W\to W\) (isotopic to the identity map) so that \([\varphi^*\eta_1] = [\eta_0]\).  Here \(\varphi^*\) denotes the pullback operator.
\end{proposition}
\begin{proof}
    \appref{app:CoadjointOrbitOnCirculations}
\end{proof}
In other words, two circulation data share the same coadjoint orbit if and only if they are related by a volume-preserving transportation with Kelvin circulation preservation.

The significance about this proposition is that the Kelvin circulation theorem, known to be associated to the qualitative statement of ``vorticity conservation'' in Euler flows, is now concretely characterized as that the state should lie on a submanifold, being the coadjoint orbit, which is a structure solely determined by the Poisson structure equipped in every dual Lie algebra (\appref{sec:PoissionSymplectic}).

Furthermore, we can quantitatively detect whether a family of states lie in the same coadjoint orbit.  A Casimir function (\defref{def:Casimir}) is a function on \(\fX_{\div}^*(W)\) that is constant along each coadjoint orbit (using \propref{prop:CoadjointOrbitsCoincideWithSymplecticLeaves}).  
Elements on the same coadjoint orbit will  take the same value for each Casimir function.  
For 2D fluids (\(n=2\)) a basis for the Casimir functions are the \(p\)-th moment of the vorticity function
\begin{align}
    \cW_p([\eta]) \coloneqq \int_W w^p\, d\mu,\quad w = *(d\eta).
\end{align}
In 3D the regular Casimir is the helicity
\begin{align}
    \cH([\eta]) \coloneqq \int_W \eta\wedge\omega,\quad \omega = d\eta.
\end{align}
In particular, these Casimir functions are conserved for any Hamiltonian flow \eqref{eq:HamiltonianFlowAsLieAdvection} on \(\fX^*_{\div}(W)\).

One example for a Hamiltonian flow on \(\fX_{\div}^*(W)\) is the incompressible Euler equation:
\begin{proposition}\label{prop:EulerEquationHamiltonian}
    The incompressible Euler equation is the Hamiltonian flow \((\fX_{\div}^*(W),H_{\rm Euler})\) for the Hamiltonian function
    \begin{align}
    \label{eq:H_Euler}
        H_{\rm Euler}([\eta]) = {1\over 2}\llangle\sharp_{\fX_{\div}}[\eta],\sharp_{\fX_{\div}}[\eta]\rrangle.
    \end{align}
\end{proposition}
\begin{proof}
    \appref{app:EulerEquationHamiltonian}
\end{proof}

Next, as we consider a discretized Euler equation, we modify the Hamiltonian function to a lower dimensional function, but we keep the Poisson structure on \(\fX_{\div}^*(W)\) the same.  By doing so, the modified equation will still stay on a coadjoint orbit in the true phase space \(\fX_{\div}^*(W)\) and therefore obey Kelvin's circulation conservation.

\subsection{Discrete Euler Flow}
\label{sec:DiscreteEulerFlow}
Our goal is to develop a finite dimensional (discrete) analog of the Hamiltonian system \((\fX_{\div}^*(W),H_{\rm Euler})\).

Our discretization starts with a finite dimensional subspace in the space of divergence-free vector fields given by
a divergence-consistent interpolation
\begin{align}
    &\cI\colon\fB\hookrightarrow \fX(W),\\
    &\cI\colon\fB_{\div}\hookrightarrow \fX_{\div}(W).
    \label{eq:cIOnBdiv}
\end{align}
Let us go towards replacing \(\fX_{\div}(W)\) with this finite dimensional representation \(\fB_{\div}\).
Like the continuous theory, the dual space \(\fB_{\div}^*\) is the quotient space \(\fB_{\div}^* = \fB^*/\im(\bd^\intercal)\) where \(\bd\) is the discrete divergence operator defining \(\fB_{\div} = \ker(\bd)\).
As described in \eqref{eq:InnerProductOnB}
the embedding induces an inner product structure on \(\fB_{\div}\) denoted by \(\langle\bff,\bg\rangle_{\fB} = \llangle\cI(\bff),\cI(\bg)\rrangle= \bff^\intercal\star\bg\).
Similar to the continuous theory the inner product \(\star\) induces a map \(\sharp_{\fB_{\div}}\colon \fB_{\div}^*\to\fB_{\div}\) given by \(\sharp_{\fB_{\div}} = \bP_{\fB_{\div}}\star^{-1}\), the pressure projection of the metric dual of any representative in \(\fB^*/\im(\bd^\intercal)\).
Using the metric we define the \emph{discrete} Hamiltonian function analogous to the Hamiltonian for the Euler equation:
\begin{definition}
    The \(\cI\)-discrete Hamiltonian for Euler fluid is given by
    \begin{align}
\label{eq:DiscreteHamiltonian}
    H_{\rm D}\colon\fB_{\div}^* \to\RR,\quad H_{\rm D}([\bc])\coloneqq {1\over 2}\langle\sharp_{\fB_{\div}}\bc,\sharp_{\fB_{\div}}\bc\rangle_{\fB}.
\end{align}
\end{definition}
Here \([\bc]\) represents a discrete circulation data.

\subsubsection{Lack of Poisson Structure for \(\fB_{\div}^*\)}
\label{sec:LackOfPoissonStructure}
To construct a discrete analogue of the Hamiltonian system \((\fX_{\div}^*(W),H_{\rm Euler})\), it is tempting to consider a system \((\fB_{\div}^*,H_{\rm D})\) where the phase space is \(\fB_{\div}^*\) on which the discrete Hamiltonian \(H_{\rm D}\) is given.  Unfortunately, 
unlike the continuous theory, the space \(\fB_{\div}\) of divergence-free discrete vector fields is in general not a Lie algebra like \(\fX_{\div}\).
It is unlikely that there exists a Lie algebraically closed set of vector fields suitable for fluid simulation, given that, at least in 2D, all Lie algebras of vector fields have been completely classified \citep{Gonzalez-Lopez:1992:LIE}.
As a consequence, in general, there is no natural Poisson structure for \(\fB_{\div}^*\) to be a phase space.  
The core problem that the space of computational velocities not being Lie algebraically closed is referred to as that the velocity is \emph{non-holonomically constrained} \cite{Mullen:2009:EPI, Pavlov:2011:SPD, Liu:2015:MVFS}.  We discuss the non-holonomicity in the literature of structure-preserving fluid simulation in \appref{app:RelationToNonHolonomicity}.

\subsubsection{Our dynamical system}
The lack of Poisson structure for \(\fB_{\div}^*\) does not prevent us from considering a Hamiltonian system on the larger space \(\fX_{\div}^*(W)\), the original dual Lie algebra of the continuous theory.
The adjoint of the embedding \eqref{eq:cIOnBdiv} is a map of type
\begin{align}
\label{eq:AdjointOfI}
    \cI^\adjoint\colon \fX_{\div}^*(W)\xrightarrow{\rm linear} \fB_{\div}^*.
\end{align}
This allows us to pullback the discrete Hamiltonian function \eqref{eq:DiscreteHamiltonian} from \(\fB_{\div}^*\) to \(\fX_{\div}^*(W)\) by composition
\begin{align}
\label{eq:DiscreteHamiltonianIT}
    (H_{\rm D}\circ\cI^\adjoint)\colon \fX_{\rm div}^*(W)\to\RR.
\end{align}
\begin{definition}
    The \(\cI\)-discrete Euler flow is given by the Hamiltonian system \((\fX_{\div}^*(W),H_{\rm D}\circ\cI^\adjoint)\).
\end{definition}
\begin{theorem}\label{thm:IDiscreteEulerFlow}.
    The equation of motion for the \(\cI\)-discrete Euler flow on \(\fX_{\div}^*(W)\) is given by 
    \begin{subequations}
    \label{eq:IDiscreteEulerEquation}
    \begin{align}
    \label{eq:IDiscreteEulerEquationA}
        &{\partial\over\partial t}[\eta] + \LD_{\vec v}[\eta] = [0],\quad \text{where}\\
    \label{eq:IDiscreteEulerEquationB}
        &\vec v = \cI\sharp_{\fB_{\div}}\cI^\adjoint[\eta] = \cI\bP_{\fB_{\div}}\star^{-1}\cI^\adjoint\eta.
    \end{align}
    \end{subequations}
\end{theorem}
\begin{proof}
    \appref{app:IDiscreteEulerFlow}.
\end{proof}
\begin{corollary}
    The \(\cI\)-discrete Euler flow preserves the energy \(H_{\rm D}\circ\cI^\adjoint\) and the coadjoint orbits in \(\fX_{\div}^*(W)\).
\end{corollary}

\begin{figure}
\centering
\setlength{\unitlength}{1pt}
\begin{picture}(240,110)
\put(0,0){
\includegraphics[width=100pt]{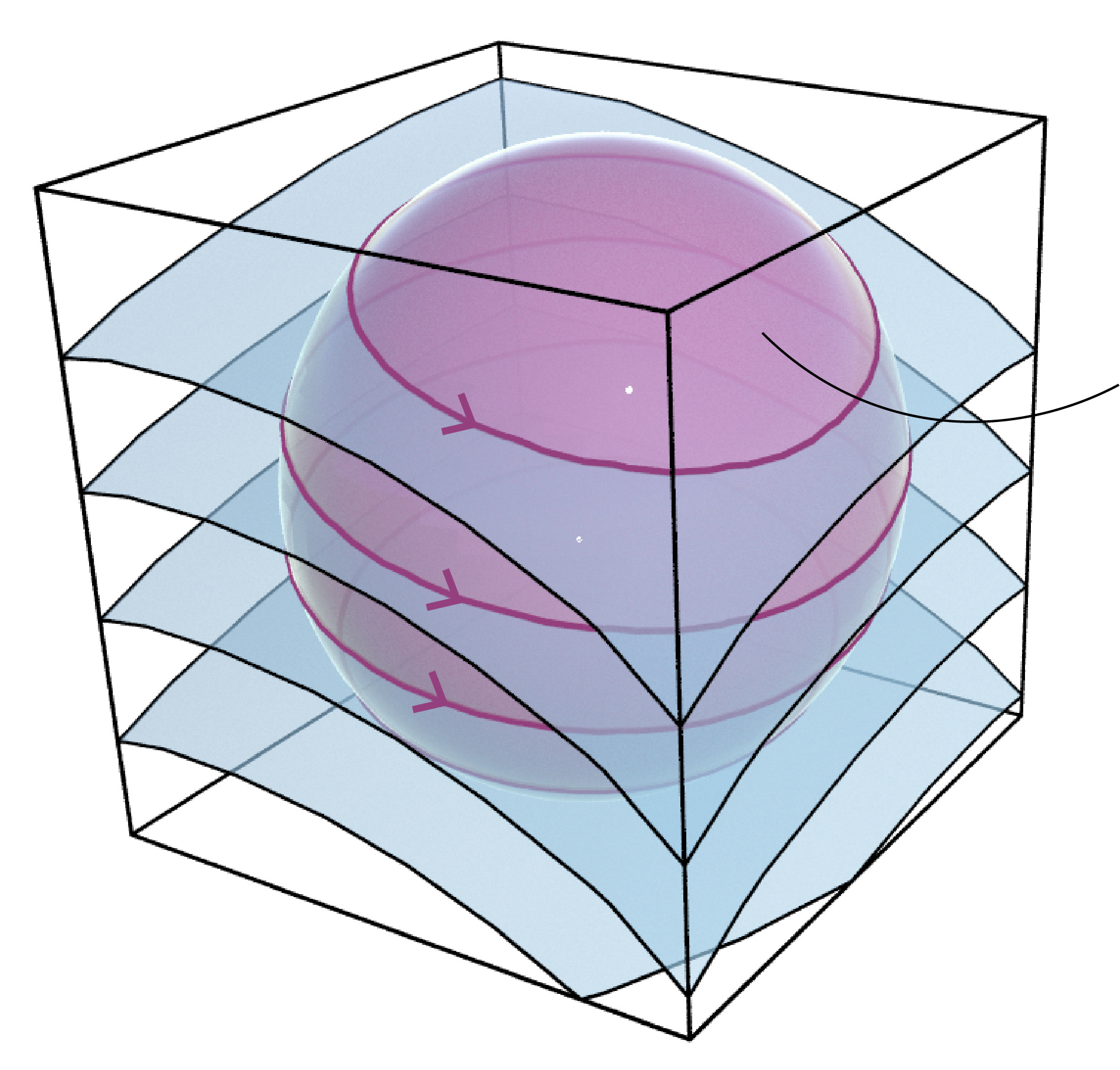}
}
\put(120, 0){
\includegraphics[width=100pt]{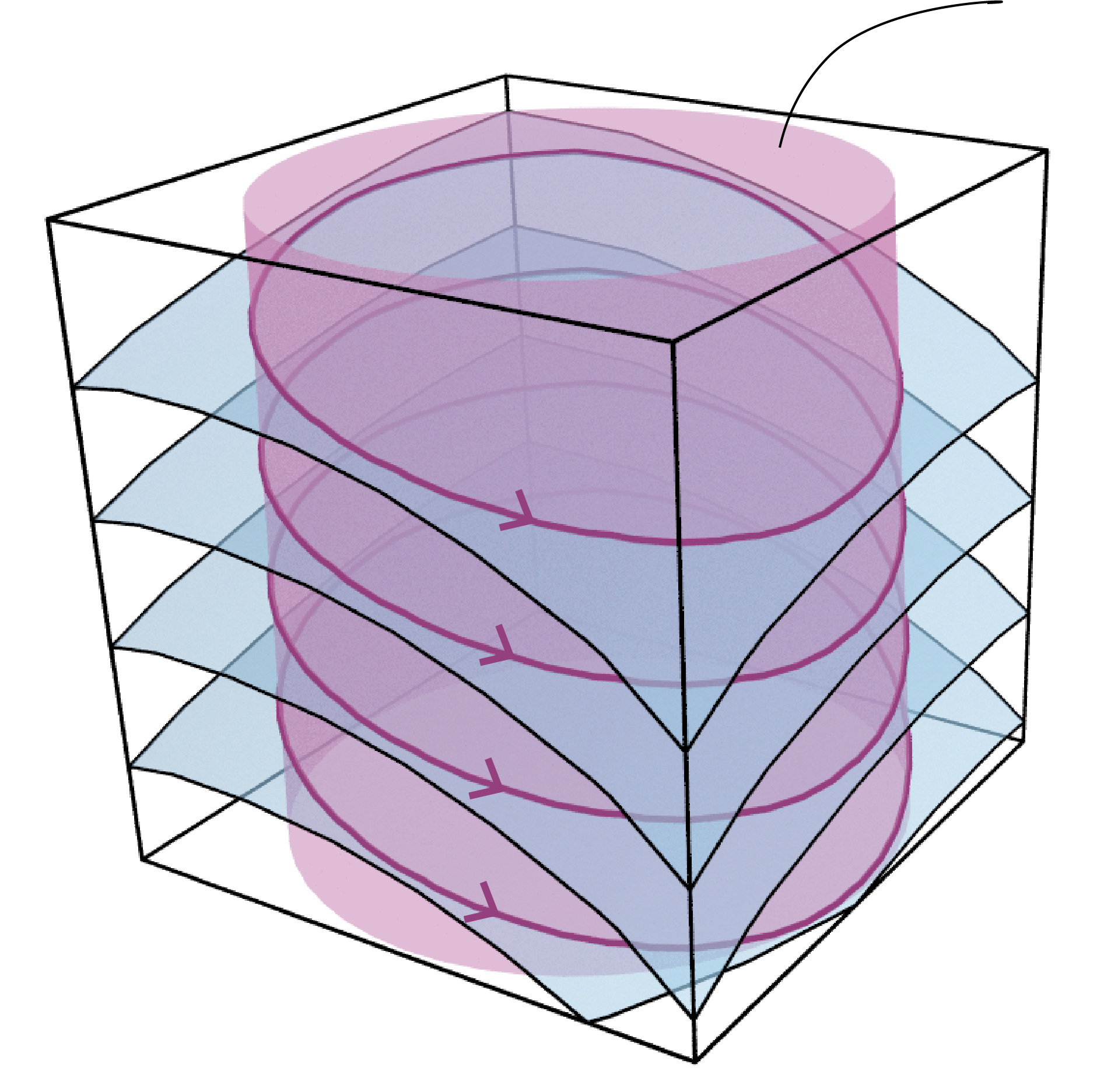}
}
\put(212,95){ $H_{\text{D}} \circ \cI^\adjoint$}
\put(100,63){ $H_{\text{Euler}}$}
\put(50,-5){ \(\fX_{\div}^*(W)\)}
\put(180,-5){ \(\fX_{\div}^*(W)\)}
\put(88,15){\parbox{20pt}{\linespread{0.6} \footnotesize \sffamily coadjoint orbits}}
\put(88,15){\line(-1,2){10pt}}
\put(218,40){\parbox{20pt}{\linespread{0.6} \footnotesize \sffamily coadjoint orbits}}
\put(218,47){\line(-1,2){10pt}}
\end{picture}
\caption{Left: The Hamiltonian for Euler fluid $H_{\text{Euler}}$ on $\fX_{\div}^*(W)$. Right: The pullback of \(\cI\)-discrete Hamiltonian for Euler fluid $H_{\text{D}}$ on $\fX_{\div}^*(W)$. 
Note that Euler flow and the \(\cI\)-discrete Euler flow stay on the same coadjoint orbits (illustrated in blue).
Their intersections with the coadjoint orbits of $\fX_{\div}^*(W)$ are marked with pink lines with arrows.
}
\label{fig:hamiltonian}
\end{figure}

The \(\cI\)-discrete Euler flow only modifies the continuous Euler flow by changing the Hamiltonian function from \(H_{\rm Euler}\) to \(H_{\rm D}\circ\cI^\adjoint\), a function of ``less complexity'' filtered through a projection \(\cI^\adjoint\) to a finite dimensional space \(\fB_{\div}^*\). However, the \(\cI\)-discrete Euler equation expressed for elements \([\eta](t)\in\fX_{\div}^*(W)\) is still a continuous PDE \eqref{eq:IDiscreteEulerEquation} (see \figref{fig:hamiltonian}). How can we compute solutions to such PDEs?
Next, we describe a technique using momentum maps to answer this question.
The procedure leads to an interpretation of the \(\cI\)-discrete Euler flow in a manner closer to a Lagrangian--Eulerian hybrid simulation algorithm.

\subsection{Method of Characteristics using an Auxiliary Symplectic Space}
\label{sec:AuxiliarySymplecticSpace}
We now introduce a technique that states that we can emulate the Hamiltonian system \((\fX_{\div}^*(W),H_{\rm D}\circ\cI^\adjoint)\) on an auxiliary symplectic space (See \ref{def:SymplecticManifold} on which the dynamical system may be easier to compute).
A special case of this technique is the method of characteristics for advective PDEs, leading to the hybrid formulation.
Readers less familiar with  Symplectic Manifolds could refer to the appendix \secref{sec:SymplecticManifoldsArePoisson} for a brief introduction.

Note that \(\fX_{\div}(W)\) is the Lie algebra (tangent space at the identity) of 
    the Lie group of volume-preserving diffeomorphisms on \(W\) 
    \begin{align*}
        \SDiff(W) = \{\varphi\in\cC^1(W;W)\,|\, \operatorname{vol} U = \operatorname{vol} \varphi(U) \text{ for all } U \subset W.\}
    \end{align*}

Consider any symplectic manifold \((\Sigma,\sigma)\) that is acted symplectomorphically by the Lie group \(\SDiff(W)\).
That is, there is a group action, which we call \textbf{advection}, on \(\Sigma\)
\begin{align}
    \Adv\colon\SDiff(W)\times\Sigma\to\Sigma
\end{align}
that preserves the symplectic 2-form \(\sigma\in\Omega^2(\Sigma)\)
\begin{align}
    (\Adv_\varphi)^*\sigma = \sigma,\quad\text{for all \(\varphi\in\SDiff(W)\).}
\end{align}
Here \((\Adv_\varphi)^*\) is the pullback operator by the map \(\Adv_{\varphi}\colon\Sigma\to\Sigma\).
By considering the infinitesimal actions of \(\Adv\), we obtain an induced Lie algebra (anti)homomorphism between spaces of vector fields
\begin{align}
\label{eq:advLieAlgHom}
    \adv\colon\fX_{\div}(W)\xrightarrow{\substack{\text{Lie alg.}\\\text{(anti)hom}}}\fX_{\sigma}(\Sigma)
\end{align}
where \(\fX_{\sigma}(\Sigma)\) is the Lie algebra of symplectic-form-preserving vector fields \(\fX_{\sigma}(\Sigma)\coloneqq\{\vec X\in\fX(\Sigma)\,|\,\LD_{\vec X}\sigma = 0\}\).
The symplectomorphic action of \(\adv\) induces a \textbf{momentum map} (\exref{ex:MomentumMapForSymplectomorphicGroupAction})
\begin{align}
    J_{\adv}\colon \Sigma \to \fX_{\div}^*(W),
\end{align}
which, by \propref{prop:MomentumMapsArePoisson}, is an (anti-)Poisson map (\defref{def:PoissonMap}).  In other words, \(\Sigma\) and \(\fX_{\div}^*(W)\) share the same Poisson structure through the change of variables of \(J_{\adv}\).  Therefore, the Hamiltonian system \((\fX_{\div}^*(W),H_{\rm D}\circ\cI^\adjoint)\) is equivalent to \((\Sigma,H_{\rm D}\circ\cI^\adjoint\circ J_{\adv})\) in the following sense.
\begin{definition}
    \(H_{\Sigma}\coloneqq (H_{\rm D}\circ\cI^\adjoint\circ J_{\adv})\colon\Sigma\to\RR\).
\end{definition}
\begin{theorem}\label{thm:SigmaEmulates}
    If a path \(s\colon[0,T]\to\Sigma\) is a solution of the Hamiltonian dynamical system \((\Sigma,H_\Sigma)\), then \((J_{\adv}\circ s)\colon[0,T]\to\fX_{\div}^*(W)\) is a solution to the \(\cI\)-discrete Euler flow \((\fX_{\div}^*(W), H_{\rm D}\circ\cI^\adjoint)\).
\end{theorem}
\begin{proof}
    \appref{app:SigmaEmulates}.
\end{proof}
In other words, we can solve the \(\cI\)-discrete Euler equation  \eqref{eq:IDiscreteEulerEquation} by solving the following equation of motion for \(s(t)\in\Sigma\).
\begin{theorem}\label{thm:IDiscreteEulerEquationSigma}
    The equation of motion for the Hamiltonian system \((\Sigma,H_\Sigma)\) is given by
    \begin{subequations}
    \label{eq:IDiscreteEulerEquationSigma}
        \begin{numcases}{}
        \label{eq:IDiscreteEulerEquationSigmaA}
            \textstyle{d\over dt}s(t) = \adv_{\vec v(t)}s(t),\\
            \label{eq:IDiscreteEulerEquationSigmaB}
            \vec v(t) = \cI\bP_{\fB_{\div}}\star^{-1}\cI^\adjoint J_{\adv}s(t).
        \end{numcases}
    \end{subequations}
\end{theorem}
\begin{proof}
    \appref{app:IDiscreteEulerEquationSigma}.
\end{proof}
\begin{corollary}
\label{cor:JsOnCoadjointOrbit}
    Any solution \(s\) to \eqref{eq:IDiscreteEulerEquationSigma} has the property that \(J_{\adv}s(t)\) and \(J_{\adv}s(0)\) lie on the same coadjoint orbit in \(\fX_{\div}^*(W)\)  for all \(t\).
\end{corollary}

    In the work of Marsden and Weinstein \shortcite{Marsden:1983:COV}, the variable \(s\in\Sigma\) is called a \textbf{Clebsch variable}, and the Poisson map \(J_{\adv}\) is a Clebsch representation \cite{Chern:2016:SS,Chern:2017:IF,Yang:2021:CGF}.  
    We will see in the following concrete instance of \(\Sigma\) that \(s\) serves as a Lagrangian marker.  The change of variables from  \eqref{eq:IDiscreteEulerEquationA} to  \eqref{eq:IDiscreteEulerEquationSigma} is similar to changing an Eulerian coordinate to a Lagrangian coordinate in the method of characteristics.  
    We do emphasize that the construction of \eqref{eq:IDiscreteEulerEquationSigma} is more delicate than just devising a Lagrangian coordinate.  Properties such as \corref{cor:JsOnCoadjointOrbit} require the advection action on \(\Sigma\) to be symplectomorphic, and the transfer function \(J_{\adv}\) to be the derived momentum map associated to the symplectomorphic action.

\subsection{A Choice for the Symplectic Space \(\Sigma\)}
\label{sec:ChoiceForSymplecticSpace}
Here we describe a construction of \(\Sigma\) and the actions \(\Adv\), \(\adv\) on \(\Sigma\), followed by deriving the associated momentum map \(J_{\adv}\). 

\subsubsection{Flow Action on Positions}
Consider the space of material positions
\begin{align}
    Q = \{\vec x\colon M\to W\,|\,\vec x\text{ is volume preserving}\}
\end{align}
where \(M\) represents the material space (Lagrangian coordinate) equipped with a volume form.  One may identify \(M=W\) as the fluid domain at \(t=0\).  We distinguish the notation for the space \(M\) from \(W\) for clarity.  The Lie group \(\SDiff(W)\) of volume preserving maps on \(W\) acts on \(Q\) by flowing the particle
\begin{align}
    \Adv^Q\colon\SDiff(W)\times Q\to Q,\quad \Adv^Q_\varphi(\vec x) = \varphi\circ\vec x.
\end{align}

\subsubsection{Lifted Action on the Cotangent Bundle}
Now, neither \(Q\) is a symplectic space nor \(\Adv^Q\) is a symplectomorphic action.
But starting from such a group action on the position space \(Q\), we can always lift the action to an action on the cotangent bundle \(\Sigma \coloneqq T^*Q\).  It turns out that the cotangent bundle is symplectic, and the lifted action is symplectomorphic.
Precisely, the cotangent bundle is given by
\begin{align}
    \Sigma = T^*Q = \left\{(\vec x,\vec u)\,\middle|\, \vec x\colon M \to W, \vec u\colon M\to T^*W, \vec u_\sfp\in T_{\vec x_\sfp}^*W\right\}.
\end{align}
That is, each element of the cotangent bundle is an assignment of particle position \(\vec x_{\sfp}\in W\), and an assignment of a covector \(\vec u_{\sfp}\in T_{\vec x_{\sfp}}^*W\) based at \(\vec x_{\sfp}\) for each \(\sfp\in M\).
The cotangent bundle has a canonical symplectic 2-form \(\sigma\) given by the exterior derivative \(d\vartheta\) of the \textbf{Liouville 1-form} \(\vartheta\in \Omega^1(\Sigma)\) given by
\begin{align}
    \vartheta_{(\vec x,\vec u)}\llbracket(\mathring{\vec x},\mathring{\vec u})\rrbracket = \int_M \langle \vec u | \mathring{\vec x}\rangle\, d\mu,\quad \text{for \((\mathring{\vec x},\mathring{\vec u})\in T_{(\vec x,\vec u)}\Sigma\).}
\end{align}
The lifted action \(\Adv\colon\SDiff(W)\times \Sigma\to\Sigma\) from \(\Adv^Q\) is given by
\begin{align}
\label{eq:LiftedAdvAction}
    \Adv_{\varphi}(\vec x,\vec u) = \left(\varphi\circ\vec x, (d\varphi^{-1})^*\vec u\right)
\end{align}
where \((d\varphi^{-1})^*\) is the adjoint of the inverse of the flow map \(\varphi\).
One may check that \(\Adv\) preserves the symplectic structure by checking that it preserves the Liouville 1-form \(\Adv_\varphi^*\vartheta = \vartheta\) and that the exterior derivative and pullbacks commute.

By taking the direction derivative of \eqref{eq:LiftedAdvAction} with respect to \(\varphi\) at the identity along a Lie algebra element \(\vec v\in\fX_{\div}\), we obtain the induced infinitesimal action 
\begin{align}
\label{eq:LiftedAdvActionLieAlgebra}
    \adv_{\vec v}(\vec x,\vec u) = \left(\vec v|_{\vec x}, -(\nabla\vec v)_{\vec x}^\intercal \vec u\right)\in T_{(\vec x,\vec u)}\Sigma.
\end{align}

\subsubsection{The Momentum Map}
The momentum map \(J_{\adv}\colon \Sigma\to\fX_{\div}^*(W)\) associated to this symplectomorphic action can be expressed explicitly using the Liouville form.  At each \((\vec x,\vec u)\in\Sigma\), the covector field \(J_{\adv}(\vec x,\vec u)\in \fX_{\div}^*(W)\), when paired with an arbitrary \(\vec v\in\fX_{\div}(W)\), is given by
\begin{align}
    \llangle J_{\adv}(\vec x,\vec u)|\vec v\rrangle = \vartheta_{(\vec x,\vec u)}\llbracket\adv_{\vec v}(\vec x,\vec u)\rrbracket = \int_M\langle \vec u|\vec v\circ\vec x\rangle\, d\mu.
\end{align}
From this expression, we conclude that
    \begin{align}
         J_{\adv}(\vec x,\vec u) = [\eta] \in\fX_{\div}^*(W) = \Omega^1(W)/d\Omega^0(W)
    \end{align}
    where \(\eta\) is uniquely defined by that \(\vec u = \eta\circ\vec x\).
    See \figref{fig:TstarQ} for illustration.

\begin{figure}
\centering
\setlength{\unitlength}{1pt}
\begin{picture}(240,150)
\put(0,0){
\includegraphics[width=250pt]{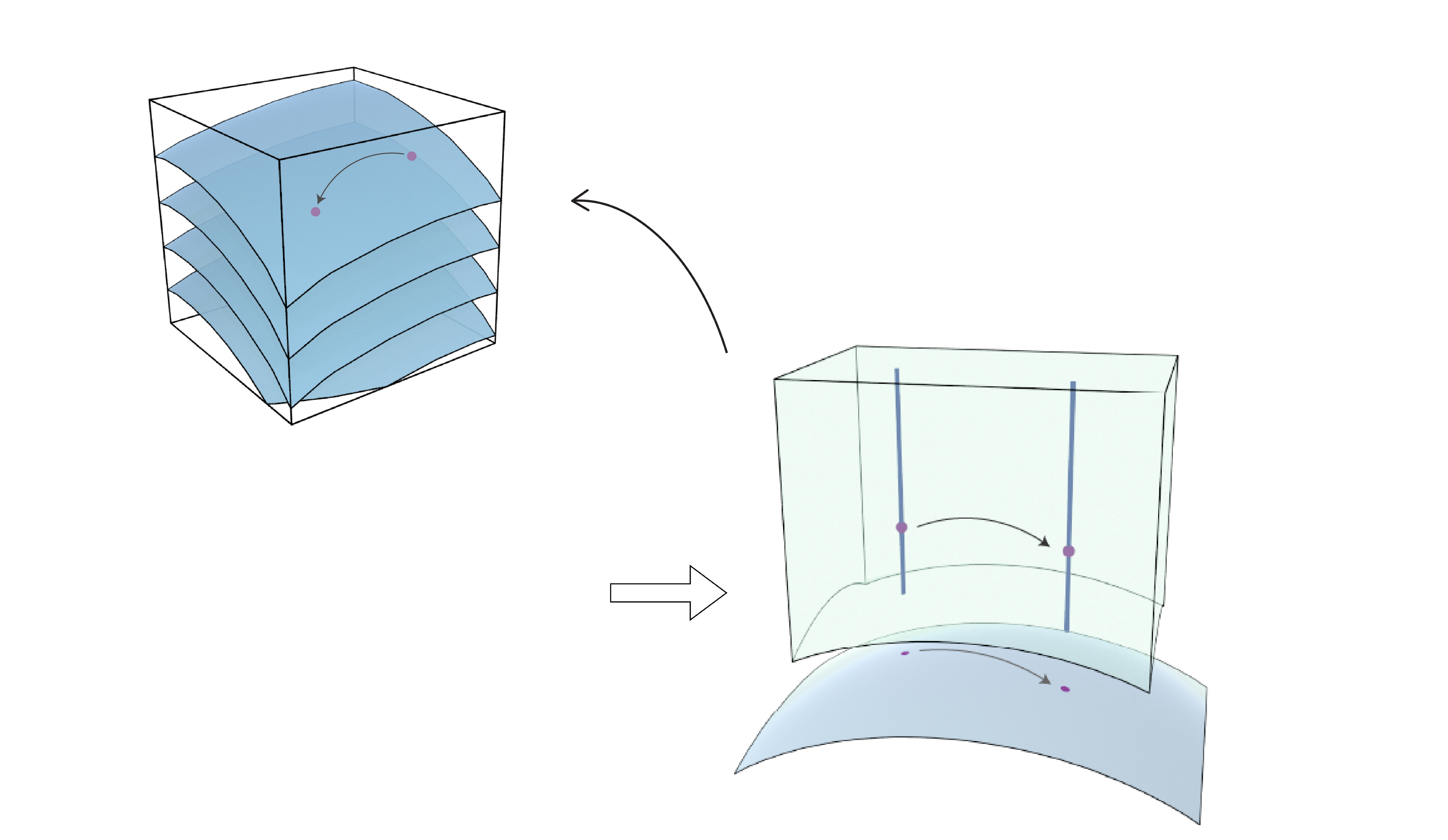}
}
\put(36,41){\LARGE $\varphi \in $ SDiff($W$)}
\put(106,48){\small \sffamily acts on}
\put(118,105){ $J_{\adv}$}
\put(50,140){ $\fX_{\div}^* (W)$}
\put(215,12){ $Q$}
\put(205,50){ $\Sigma = T^*Q$}
\put(160,22){\small $\Adv_{\varphi}^Q$}
\put(160,58){\small $\Adv_{\varphi}$}
\end{picture}
\caption{Mapping between $\Sigma = T^*Q$, SDiff($W$), and $\fX_{\div}^*(W)$.} 
\label{fig:TstarQ}
\end{figure}

\subsection{Arrival at the CO-FLIP Equations}
The CO-FLIP equation arises by substituting the setup of \secref{sec:ChoiceForSymplecticSpace} in \eqref{eq:IDiscreteEulerEquationSigma}.
To see it, let us unpack one final operator, \(\star^{-1}\cI^\adjoint\), in \eqref{eq:IDiscreteEulerEquationSigmaB}.
\subsubsection{The Map \(\star^{-1}\cI^\adjoint J_{\adv}\)}
In \eqref{eq:IDiscreteEulerEquationSigmaB} the momentum map \(J_{\adv}\) is used only in conjunction with subsequent operator
\begin{align}
    \star^{-1}\cI^\adjoint \colon\Omega^1(W)\to\fB.
\end{align}
Here \(\cI^\adjoint\) is the adjoint of \(\cI\colon\fB\to\fX(W)\).  The map \(\star^{-1}\cI^\adjoint \) becomes well-defined on \(\fX_{\div}^*(W)= \Omega^1(W)/d\Omega^0(W)\) if it is always followed by the pressure projection \(\bP_{\fB_{\div}}\), which is indeed the case in \eqref{eq:IDiscreteEulerEquationSigmaB}.
 
\begin{remark}
    In fact, \(\cI^\adjoint J_{\adv}\colon\Sigma\to\fB_{\div}^*\), which is the map from the Lagrangian data \(\Sigma\) to the grid, is the momentum map associated to the action \((\adv\circ \cI)\colon\fB_{\div}\to\fX_\sigma(\Sigma)\).
\end{remark}

The following lemmas characterizes \(\star^{-1}\cI^\adjoint\) and \(\star^{-1}\cI^\adjoint\J_{\adv}\) as the pseudoinverse of \(\cI\).
\begin{lemma}\label{lem:IDagger1}
    The map \(\star^{-1}\cI^\adjoint\colon\Omega^1(W)\to\fB\) is the pseudoinverse of \(\flat\circ\cI\colon \fB\to\Omega^1(W)\) with respect to the \(L^2\) structure on \(\fX(W)\) and its dual \(\Omega^1(W)\):
    \begin{align}
    \label{eq:IDaggerContinuum1}
        \star^{-1}\cI^\adjoint\eta = \argmin_{\bff\in\fB}\int_W|\cI(\bff) - \eta^\sharp|^2\, d\mu.
    \end{align}
\end{lemma}
\begin{lemma}\label{lem:IDagger2}
    The map \(\star^{-1}\cI^\adjoint J_{\adv}\colon \Sigma\to\fB\) is given by
    \begin{align}
    \label{eq:IDaggerContinuum}
        \star^{-1}\cI^\adjoint J_{\adv}(\vec x,\vec u) = \argmin_{\bff\in\fB}\int_M|\cI(\bff)\circ\vec x - \vec u|^2\, d\mu.
    \end{align}
    That is, \(\star^{-1}\cI^\adjoint J_{\adv}\) is the \(L^2\)-projection of a vector field, represented by the position--impulse field \((\vec x,\vec u)\), to the subspace \(\fB\) of interpolatable vector fields.
\end{lemma}
\begin{proof}
    \appref{app:IDagger}.
\end{proof}

\begin{definition}
    We define \(\cI^+ \coloneqq \star^{-1}\cI^\adjoint\) and with slight abuse of notation \(\cI^+\coloneqq\star^{-1}\cI^\adjoint J_{\adv}\).  In both cases, \(\cI^+\colon(\cdot)\to\fB\) is an operator that \(L^2\)-orthogonally projects a smooth vector field to the subspace \(\fB\) of interpolatable vector fields.  The input vector field is either represented as a vector or covector (through identification of \(\flat\)/\(\sharp\)) on either the world \(W\) or material \(M\) coordinate (through identification of \(J_{\adv}\))
\end{definition}

\subsubsection{CO-FLIP on Continuous Material Space}
Combining \eqref{eq:IDiscreteEulerEquationSigma}, \eqref{eq:LiftedAdvActionLieAlgebra}, and \eqref{eq:IDaggerContinuum}, we obtain an equation of motion for \(\vec x(t)\colon M\to W\) and \(\vec u(t)\colon M\to T^*W\) (with \(\vec u_{\sfp}\in T_{\vec x_{\sfp}}^*W\))
\begin{subequations}
\label{eq:COFLIPonM}
    \begin{numcases}{}
    \label{eq:COFLIPonMA}
        \textstyle{d\over dt}\vec x_{\sfp}(t) = \vec v|_{\vec x_{\sfp}(t)}&\text{$\sfp\in M$}\\
        \label{eq:COFLIPonMB}
        \textstyle{d\over dt}\vec u_{\sfp}(t) = -(\nabla\vec v|_{\vec x_{\sfp}(t)})^\intercal\vec u_{\sfp}(t)&\text{$\sfp\in M$}\\
        \label{eq:COFLIPonMC}
        \vec v = \cI\bP_{\fB_{\div}}\cI^+(\vec x,\vec u).
    \end{numcases}
\end{subequations}

\subsubsection{Time Discretization of CO-FLIP on Continuous Material Space}
Following the trapezoidal rule \eqref{eq:ImplicitMidpoint} we consider the time discretization of \eqref{eq:COFLIPonM} as the following Lie group integrator:
\begin{subequations}
\label{eq:ImplicitMidpointOnM}
    \begin{numcases}{}
    \label{eq:ImplicitMidpointOnMA}
        \textstyle(\vec x,\vec u)^{(\sfn + 1)} = \Adv_{\vec v}^{\Deltait t}(\vec x,\vec u)^{(\sfn)}\\
        \label{eq:ImplicitMidpointOnMB}
        \textstyle\vec v = \cI\left({1\over 2}(\bP_{\fB_{\div}}\cI^+(\vec x,\vec u)^{(\sfn)} + \bP_{\fB_{\div}}\cI^+(\vec x,\vec u)^{(\sfn+1)})\right)
    \end{numcases}
\end{subequations}
which is a map \((\vec x,\vec u)^{(\sfn)}\mapsto (\vec x,\vec u)^{(\sfn+1)}\) on \(\Sigma\).
Here,
\begin{align}
    \Adv^{\Deltait t}_{\vec v} \coloneqq \exp(\Deltait t\adv_{\vec v})=\Adv_{\varphi(\Deltait t)}
\end{align}
is the symplectomorphic group action \eqref{eq:LiftedAdvAction}
where \(\varphi(\Deltait t)\in\SDiff(W)\) is the flow map generated by \(\vec v\) over \(\Delta t\) time span.

One may also strip away the instantiation of the symplectic space \(\Sigma\) in \secref{sec:ChoiceForSymplecticSpace}, and observe that the time discretization \eqref{eq:ImplicitMidpointOnM} is a special case of the following Lie group integrator for \eqref{eq:IDiscreteEulerEquationSigma}:
\begin{subequations}
\label{eq:ImplicitMidpointSigma}
    \begin{numcases}{}
        s^{(\sfn+1)} = \exp(\Deltait t \adv_{\vec v})s^{(\sfn)} = \Adv_{\varphi(\Deltait t)}s^{(\sfn)}\\
        \textstyle \vec v = \cI\bP_{\fB_{\div}}\cI^+\left({1\over 2}(J_{\adv}s^{(\sfn)}+J_{\adv}s^{(\sfn+1)})\right)
    \end{numcases}
\end{subequations}
which is a map \(s^{(\sfn)}\mapsto s^{(\sfn+1)}\) generating a discrete sequence on \(\Sigma\).

\begin{figure}
\centering
\setlength{\unitlength}{1pt}
\begin{picture}(240,180)
\put(0,0){
\includegraphics[width=240pt]{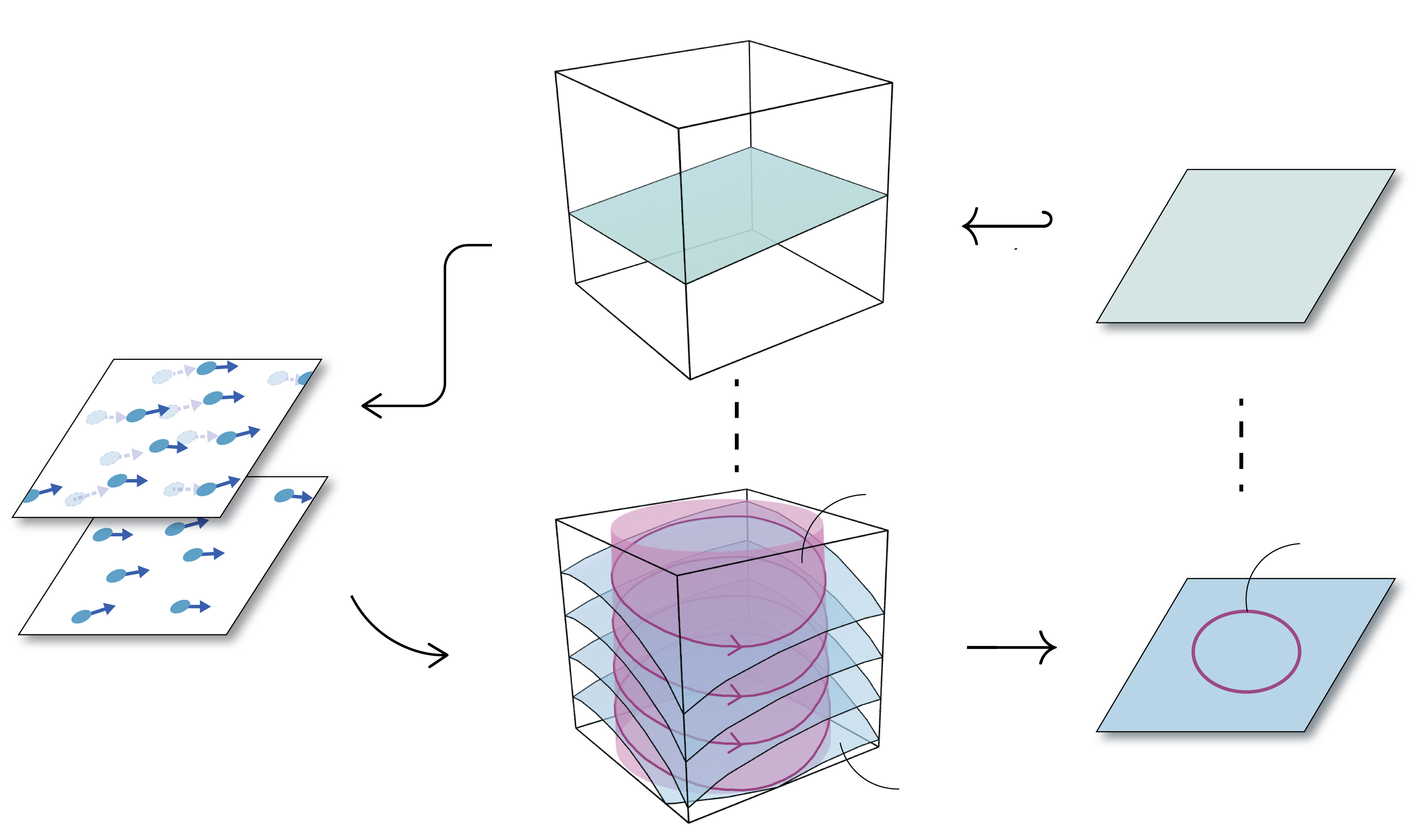}
}

\put(15,89){\parbox{10ex}{\linespread{0.6} \small \sffamily Particle state space}}
\put(106,-5){ \sffamily $\fX_{\div}^*(W)$}
\put(110,140){ \sffamily $\fX_{\div}(W)$}
\put(110,105){ \footnotesize \sffamily $\cI(\fB_{\div})$}
\put(170,108){ \sffamily $\cI$}
\put(191.5,90.5){ \sffamily \small $\fB_{\div}$}
\put(192,23){ \sffamily \small $\fB^*_{\div}$}
\put(169,37){ \sffamily $\cI^\adjoint$}
\put(60,90){ \sffamily $\adv$}
\put(62,41){ \sffamily $J_{\adv}$}
\put(12,20){ \sffamily $\Sigma = T^*Q$}
\put(221,48){ \sffamily \footnotesize $H_D$}
\put(148,57){ \sffamily \footnotesize $H_D \circ \cI^\adjoint$}
\put(157,7){\parbox{20pt}{\linespread{0.6} \footnotesize \sffamily coadjoint orbits}}
\end{picture}
\caption{Illustration of our CO-FLIP dynamical system solving the $\cI$-discrete Euler equation; with finite particles approximating the phase space.}
\label{fig:method}
\end{figure}

\subsubsection{CO-FLIP on Finite Particles}
In \eqref{eq:COFLIPonM}, the advection component \eqref{eq:COFLIPonMA} and \eqref{eq:COFLIPonMB} are ODEs that hold for \(\sfp\in \cP\) for any subset \(\cP\subset M\).
Let \(\cP\subset M\) be a finite set of points, representing computational particles.
The restrictions \((\vec x,\vec u)\colon\cP\to T^*W\) of \((\vec x,\vec u)\colon M \to T^*W\) are samples of the function \((\vec x,\vec u)\colon M \to T^*W\) on the fixed set of points \(\cP\).
The only modification on the system \eqref{eq:COFLIPonM} when restricting the material space \(M\) to the finite set \(\cP\) is replacing \(\cI^+ = \star^{-1}\cI^\adjoint J_{\div}\)  by the following approximation of \eqref{eq:IDaggerContinuum}
\begin{align}
\label{eq:IDaggerDiscrete}
    \hat\cI^+(\vec x,\vec u)\coloneqq\argmin_{\bff\in\fB}\sum_{\sfp\in\cP}|\cI(\bff)\circ\vec x_{\sfp} - \vec u_{\sfp}^\sharp|^2\mu_{\sfp}.
\end{align}
The weight \(\mu_{\sfp}\) represents the particle volume.
This modification results in the finite dimensional ODE system \eqref{eq:COFLIP-ODE} and the algorithm \eqref{eq:ImplicitMidpoint} in \secref{sec:Method}, as illustrated in \figref{fig:method}.

In the discretization from \eqref{eq:IDaggerContinuum} to \eqref{eq:IDaggerDiscrete}, the \(L^2\)-integral \(\int_M|\cdot|^2\, d\mu\) is approximated by a finite sum \(\sum_{\sfp\in\cP}|(\cdot)_{\sfp}|^2\mu_{\sfp}\).
Other linear algebraic properties about \(\cI^+\) remain true for \(\hat\cI^+\).
\begin{theorem}
Suppose \(\cP\subset M\) is sufficiently dense so that the minimization \eqref{eq:IDaggerDiscrete} is unique.  Then
    \(\hat\cI^+(\vec x,\vec u)\) is a projection on the subspace \(\fB\) from a smooth vector field represented by \((\vec x,\vec u)\).  In particular, \(\hat\cI^+\) is an exact left-inverse of \(\cI\).
\end{theorem}
\begin{proof}
    See \thmref{thm:IDaggerIsLeftInverse}.
\end{proof}
Note that the projection \(\hat\cI^+\) is an orthogonal projection not with respect to \(\int_M|\cdot|^2\, d\mu\) but with respect to a perturbed metric \(\sum_{\sfp\in\cP}|(\cdot)_{\sfp}|^2\mu_\sfp\).
In the infinite limit of \(\cP\to M\), the projection \(\hat\cI^+\) to \(\fB\) is progressively closer to being \(L^2\)-orthogonal.

\begin{figure}
    \centering
    \includegraphics[trim={150px 140px 0 140px},clip,width=0.4\columnwidth,angle=90]{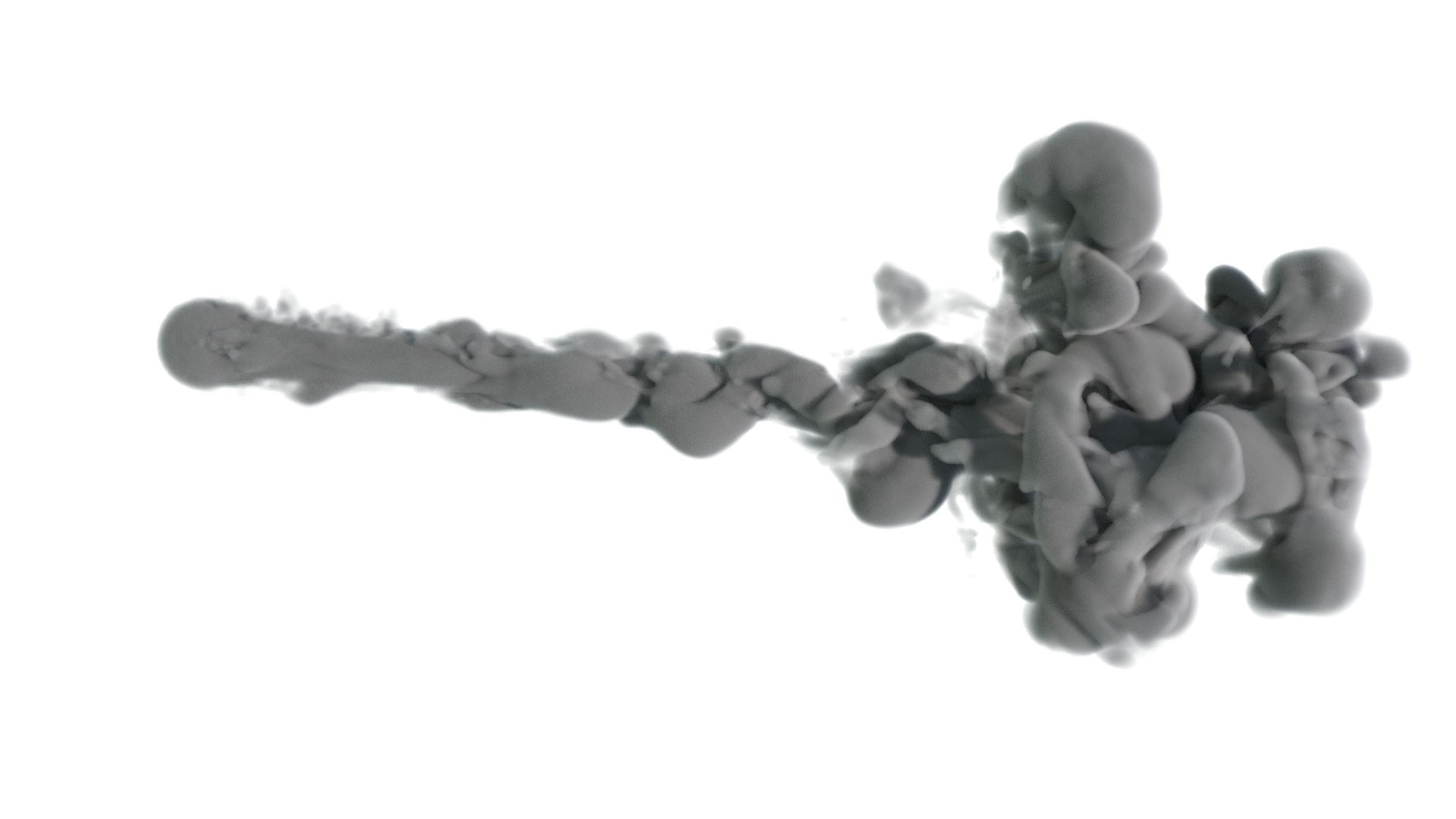}
    \includegraphics[trim={150px 140px 0 140px},clip,width=0.4\columnwidth,angle=90]{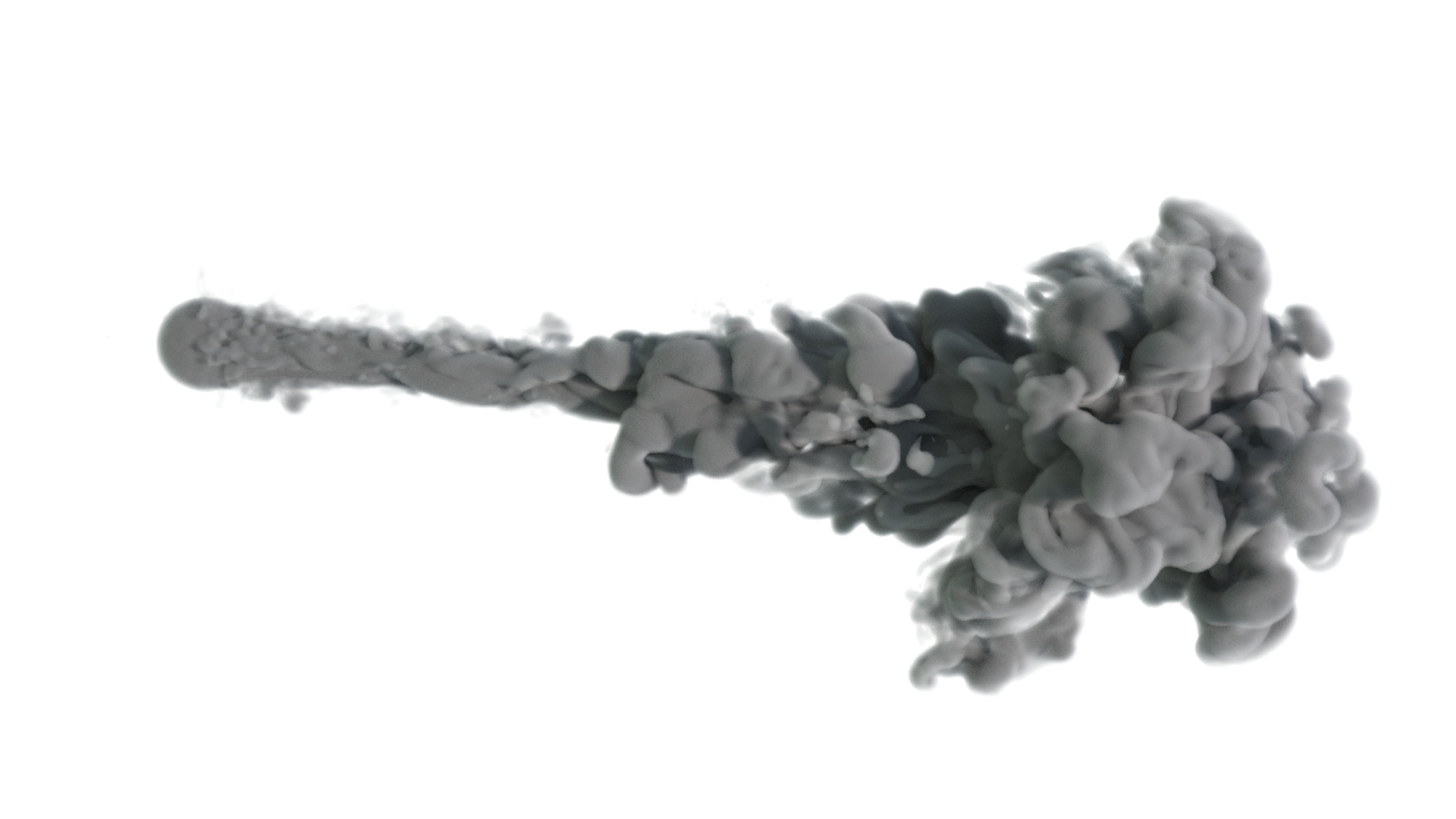}
    \includegraphics[trim={150px 140px 0 140px},clip,width=0.4\columnwidth,angle=90]{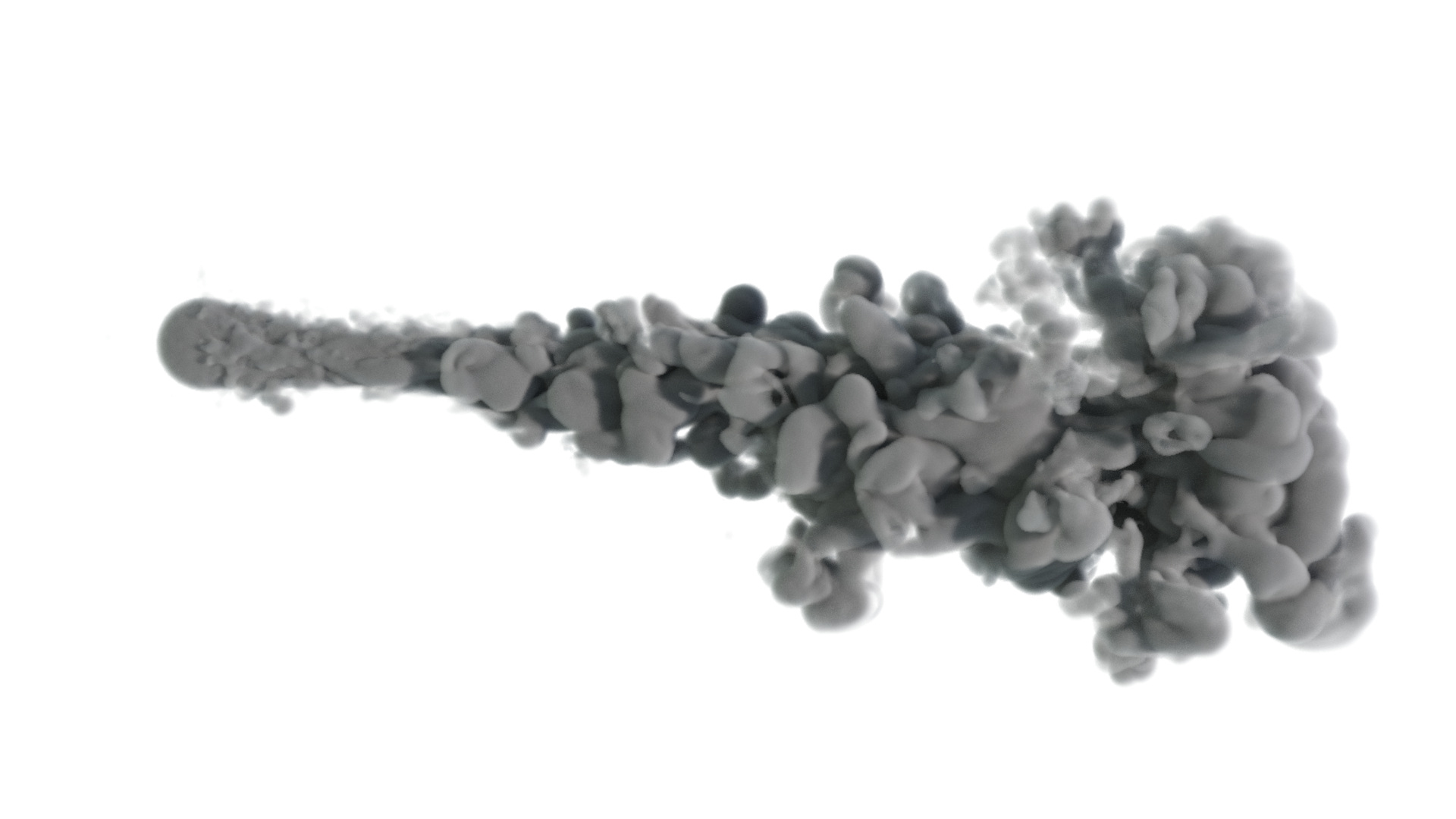}
    \includegraphics[trim={150px 140px 0 140px},clip,width=0.4\columnwidth,angle=90]{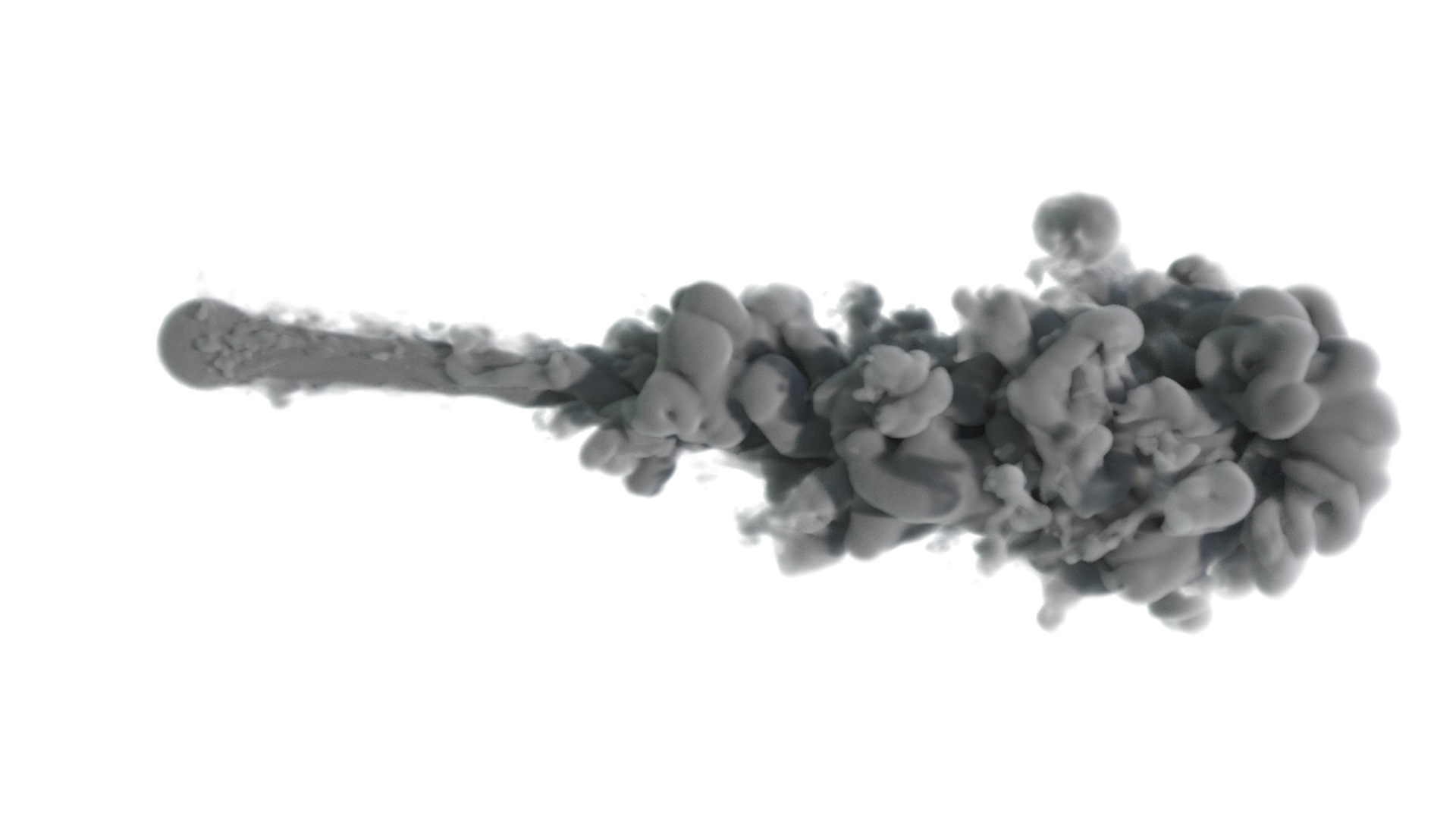}
    \includegraphics[trim={150px 140px 0 140px},clip,width=0.4\columnwidth,angle=90]{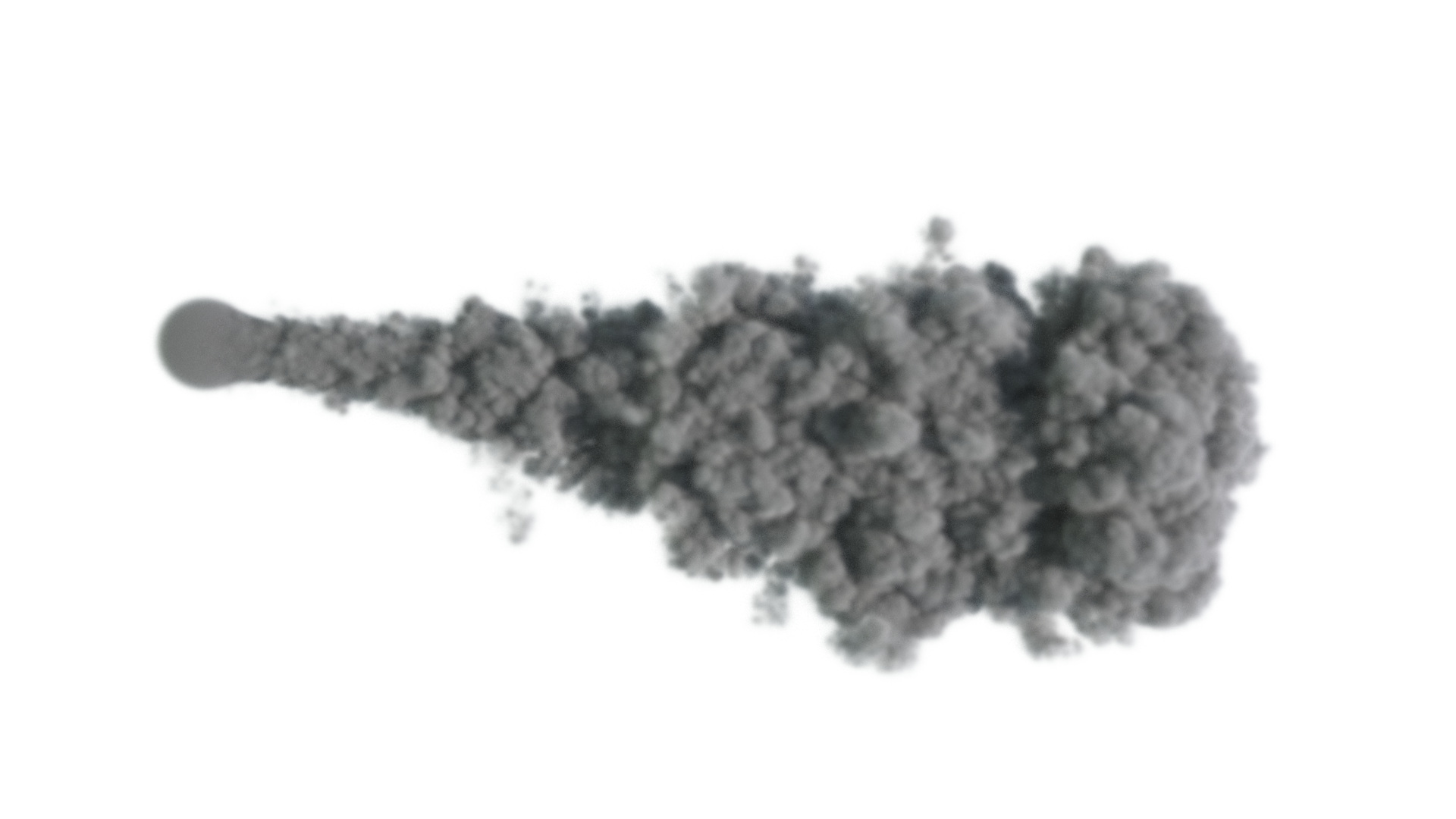}
    \\
    \includegraphics[trim={150px 140px 0 140px},clip,width=0.4\columnwidth,angle=90]{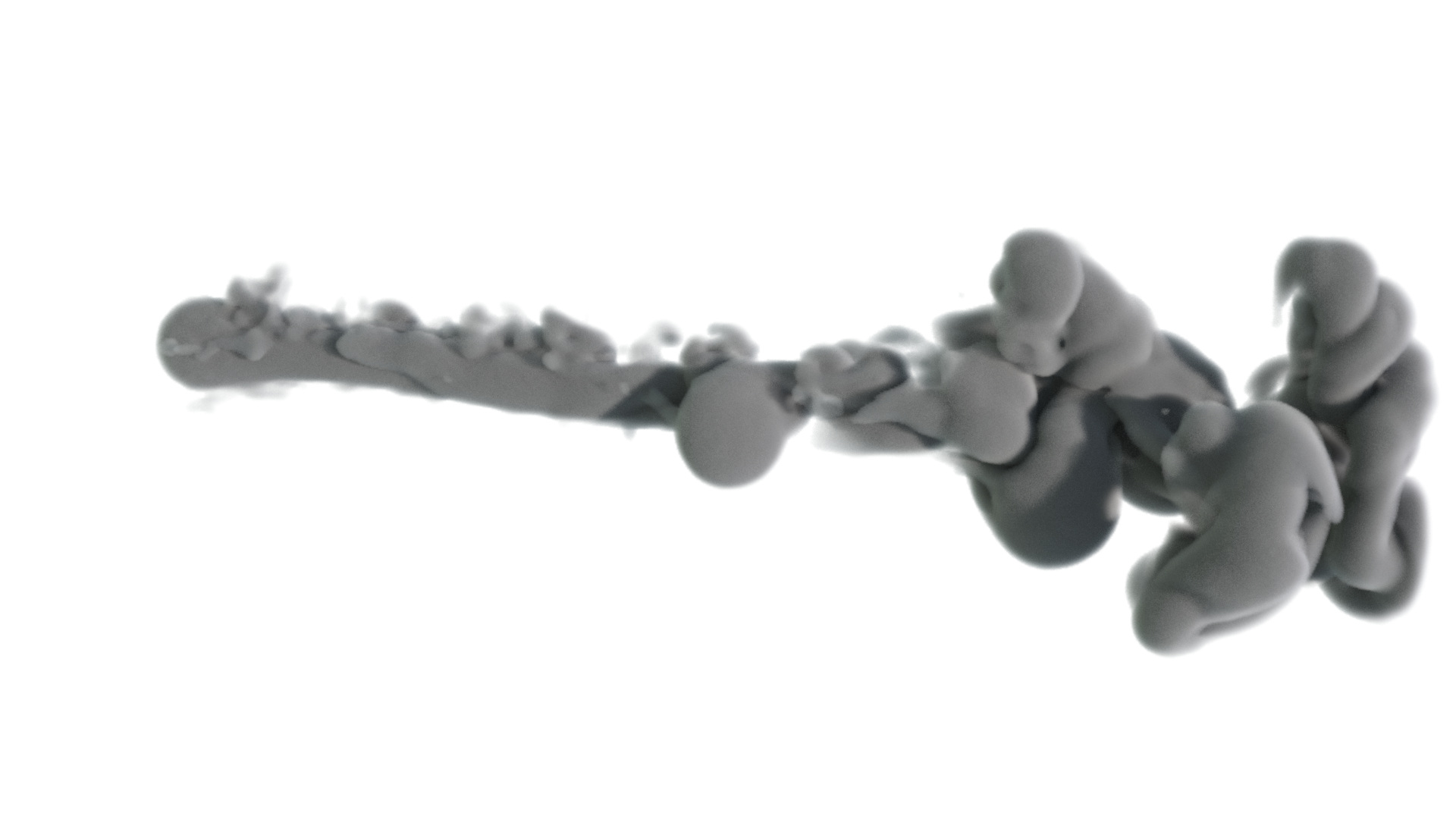}
    \includegraphics[trim={150px 140px 0 140px},clip,width=0.4\columnwidth,angle=90]{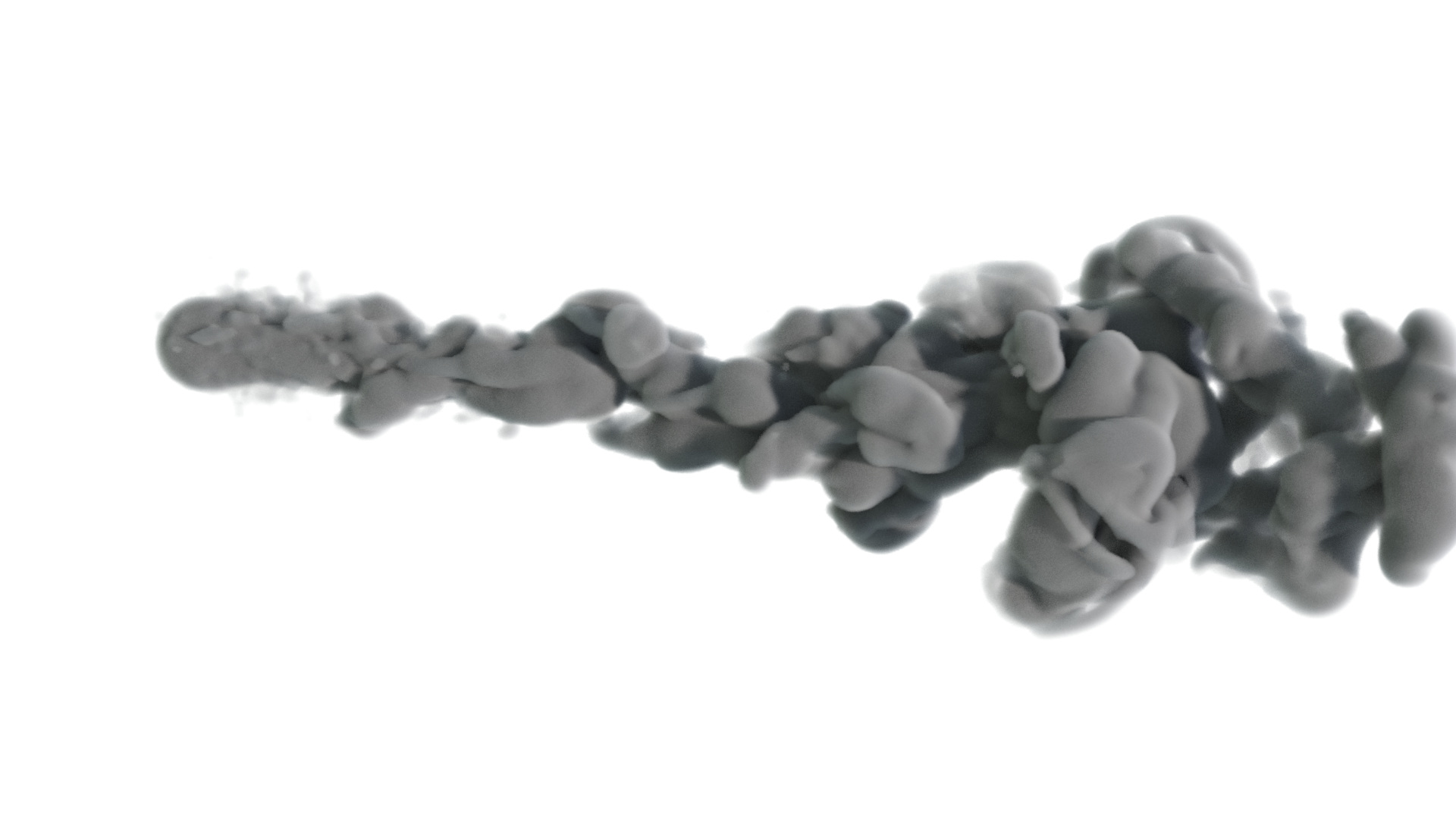}
    \includegraphics[trim={150px 140px 0 140px},clip,width=0.4\columnwidth,angle=90]{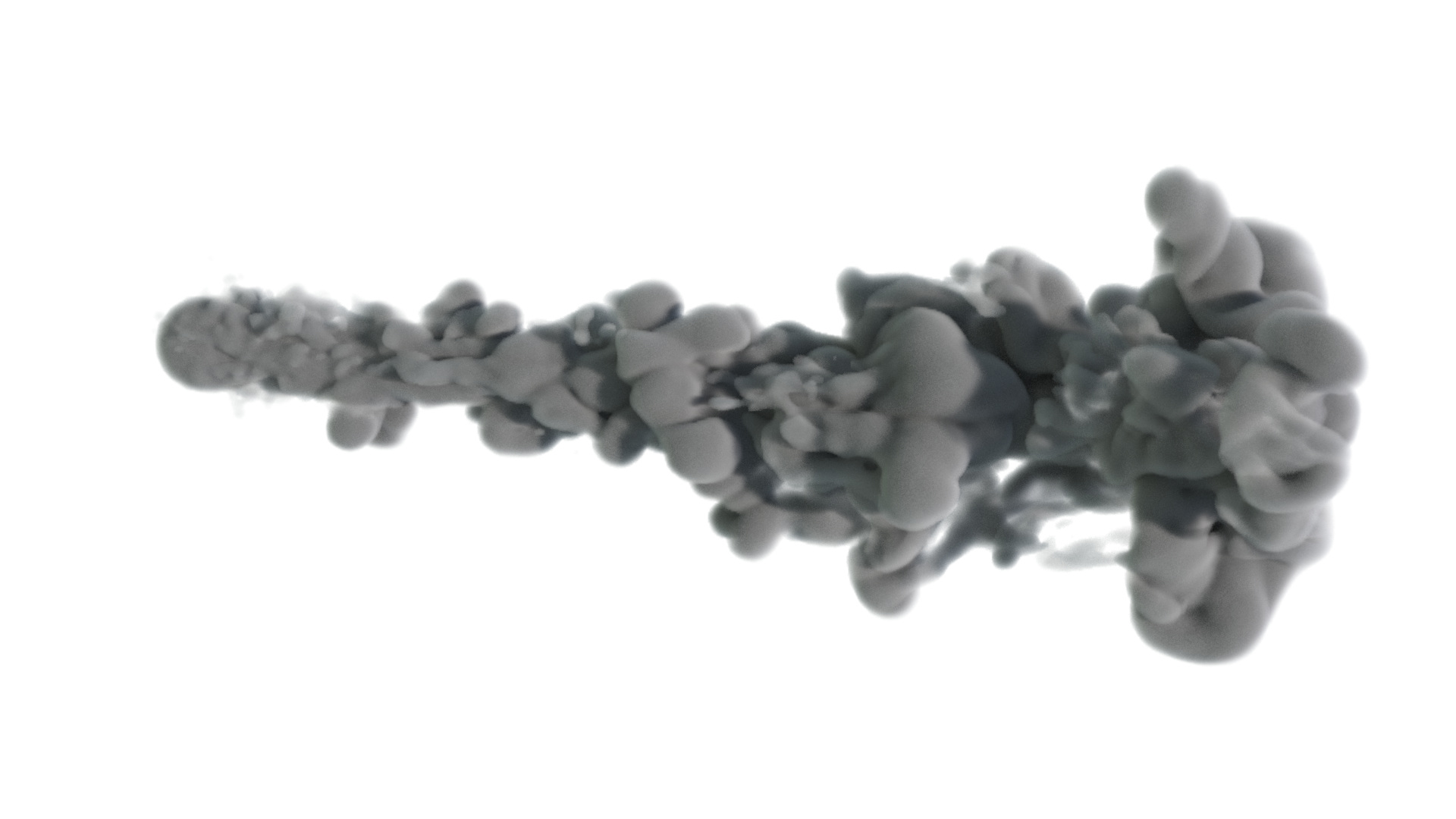}
    \includegraphics[trim={150px 140px 0 140px},clip,width=0.4\columnwidth,angle=90]{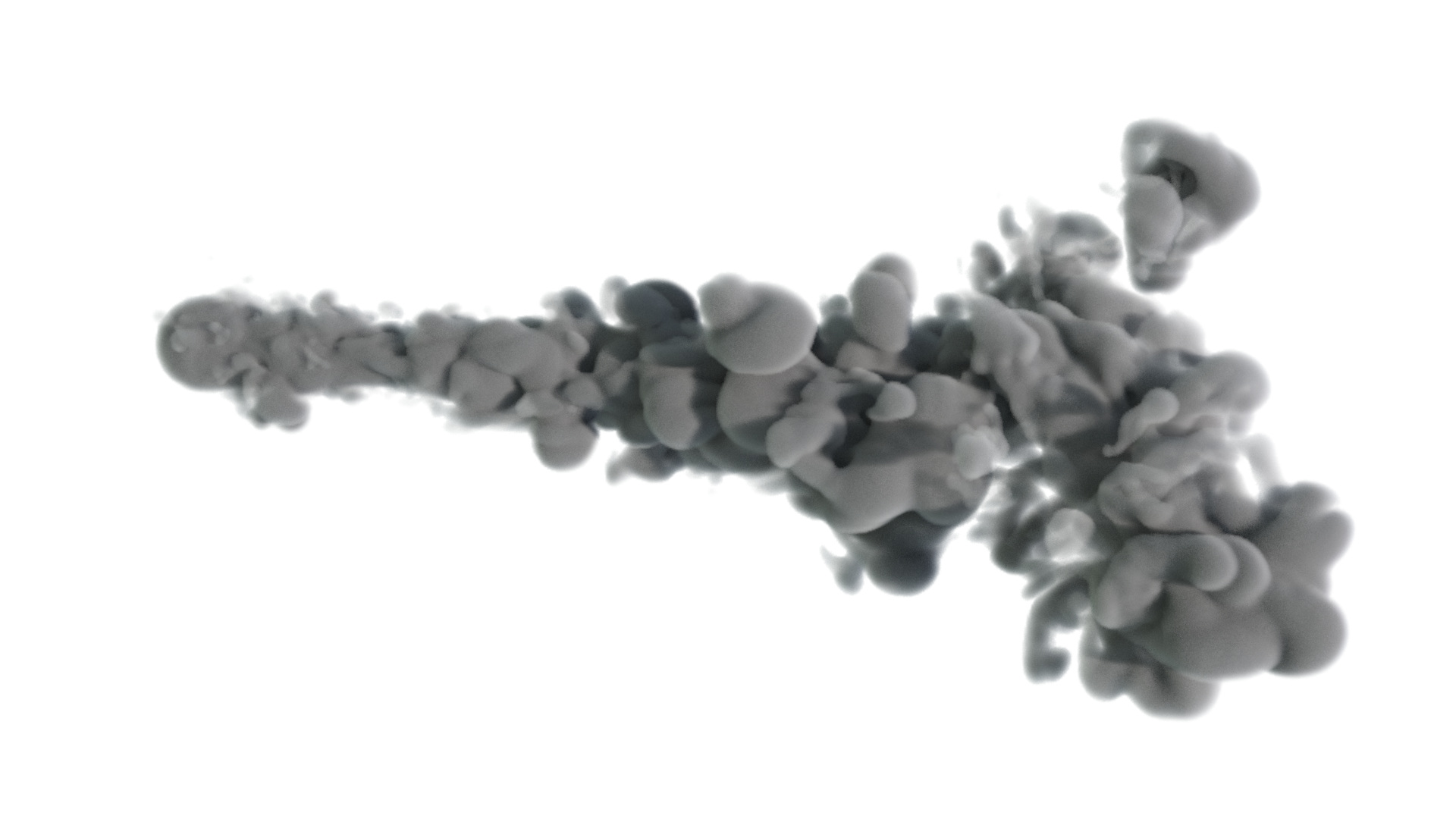}
    \includegraphics[trim={150px 140px 0 140px},clip,width=0.4\columnwidth,angle=90]{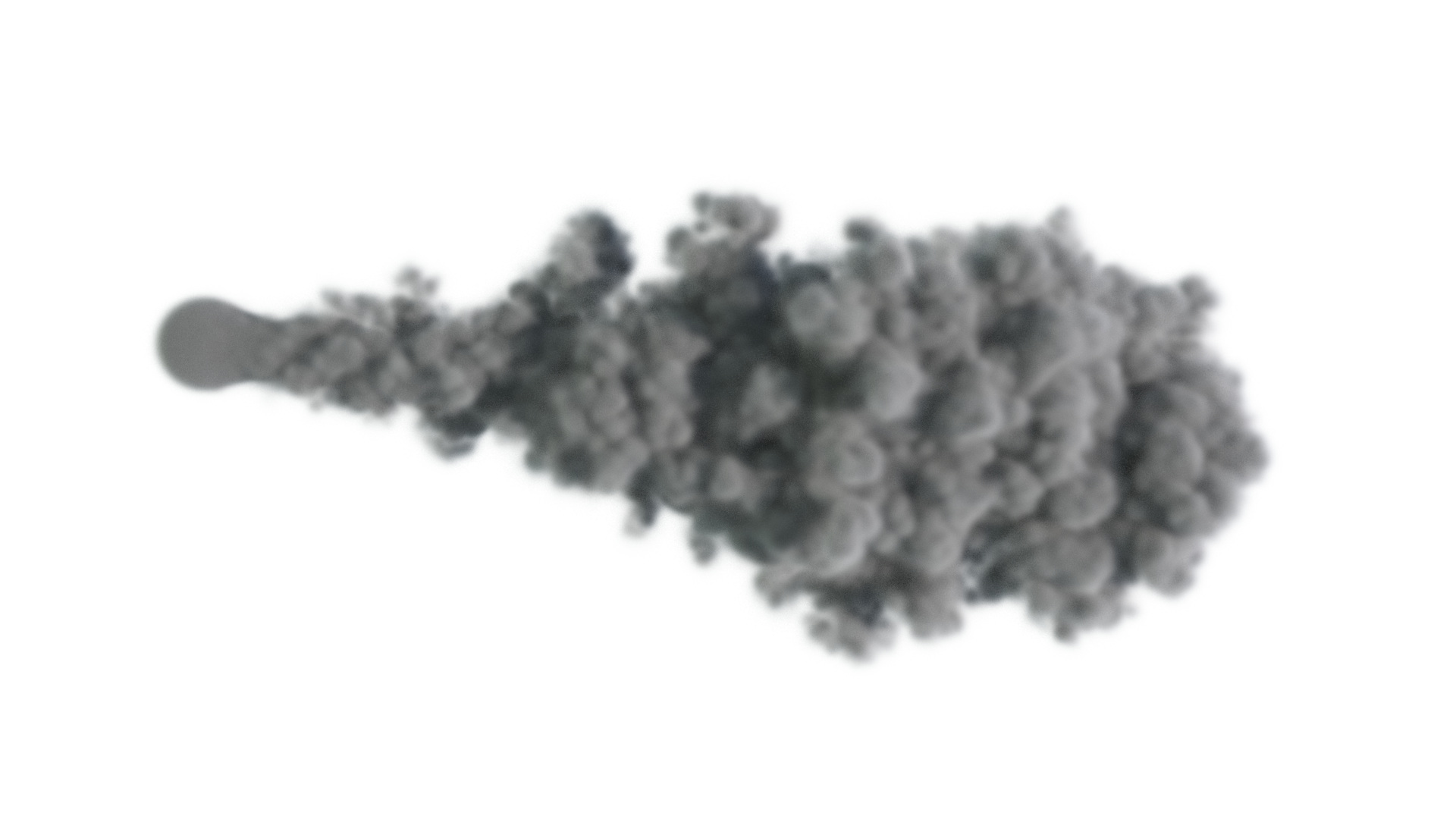}
    \\
    \includegraphics[trim={150px 140px 0 140px},clip,width=0.4\columnwidth,angle=90]{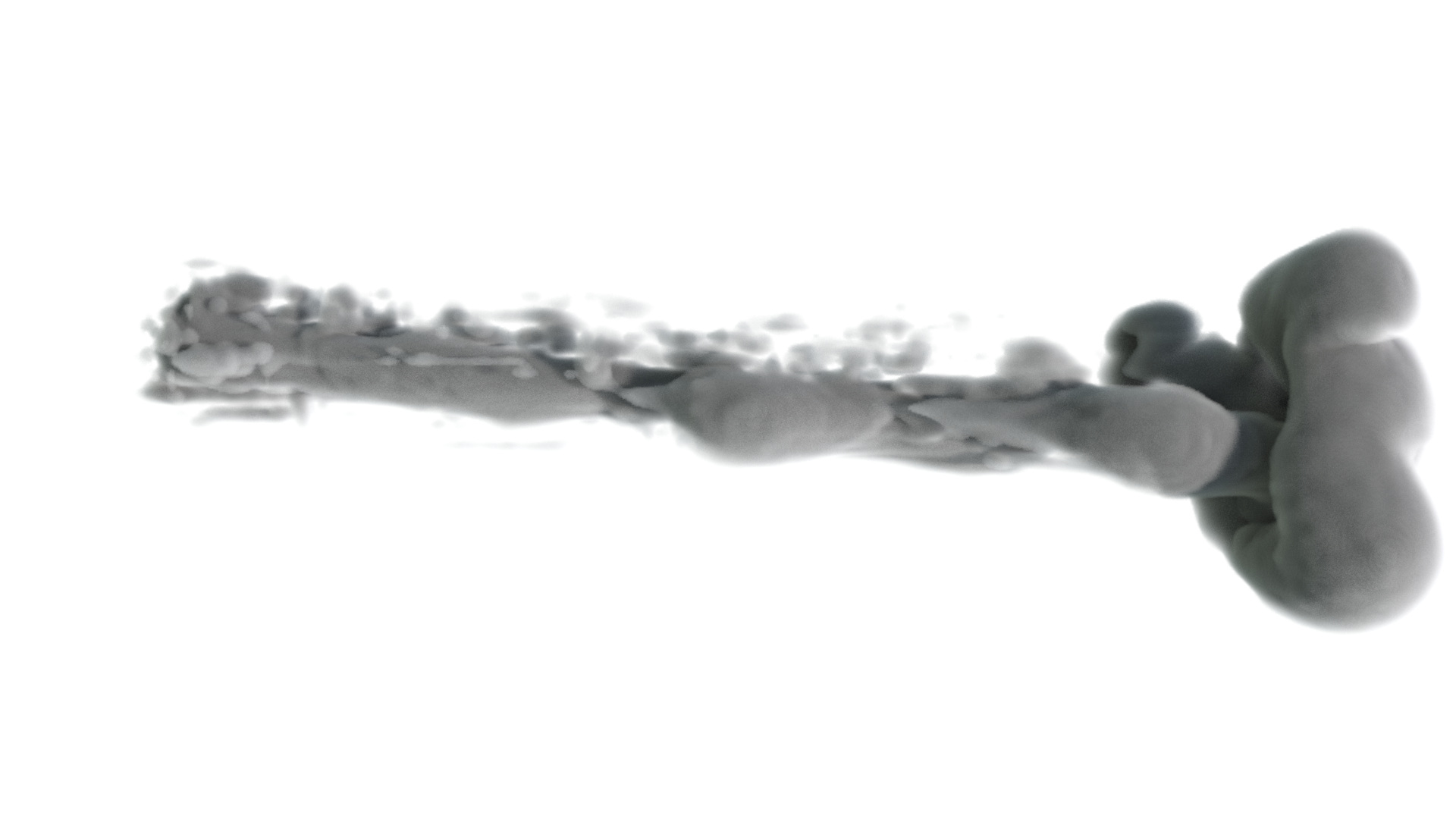}
    \includegraphics[trim={150px 140px 0 140px},clip,width=0.4\columnwidth,angle=90]{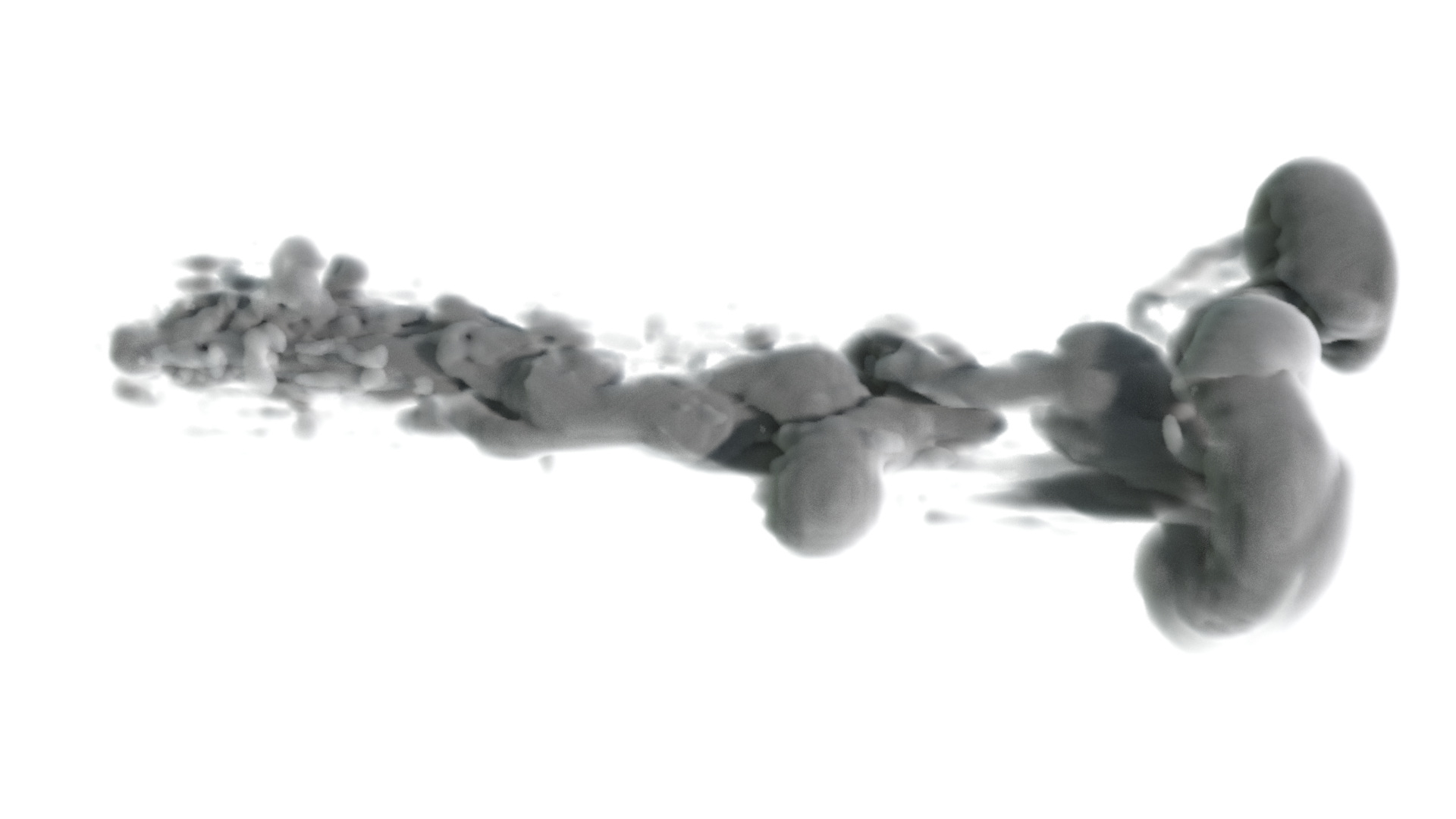}
    \includegraphics[trim={150px 140px 0 140px},clip,width=0.4\columnwidth,angle=90]{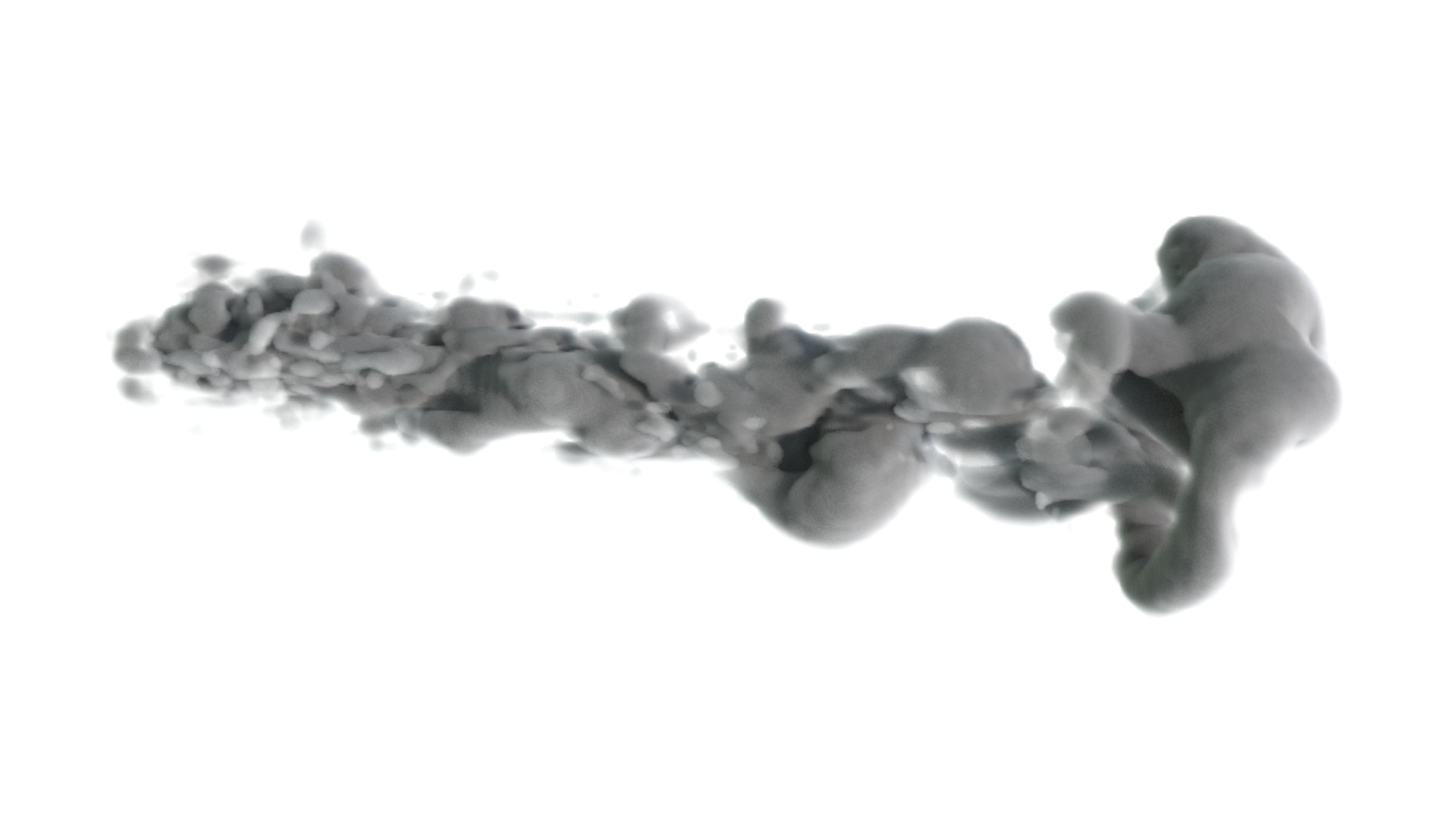}
    \includegraphics[trim={150px 140px 0 140px},clip,width=0.4\columnwidth,angle=90]{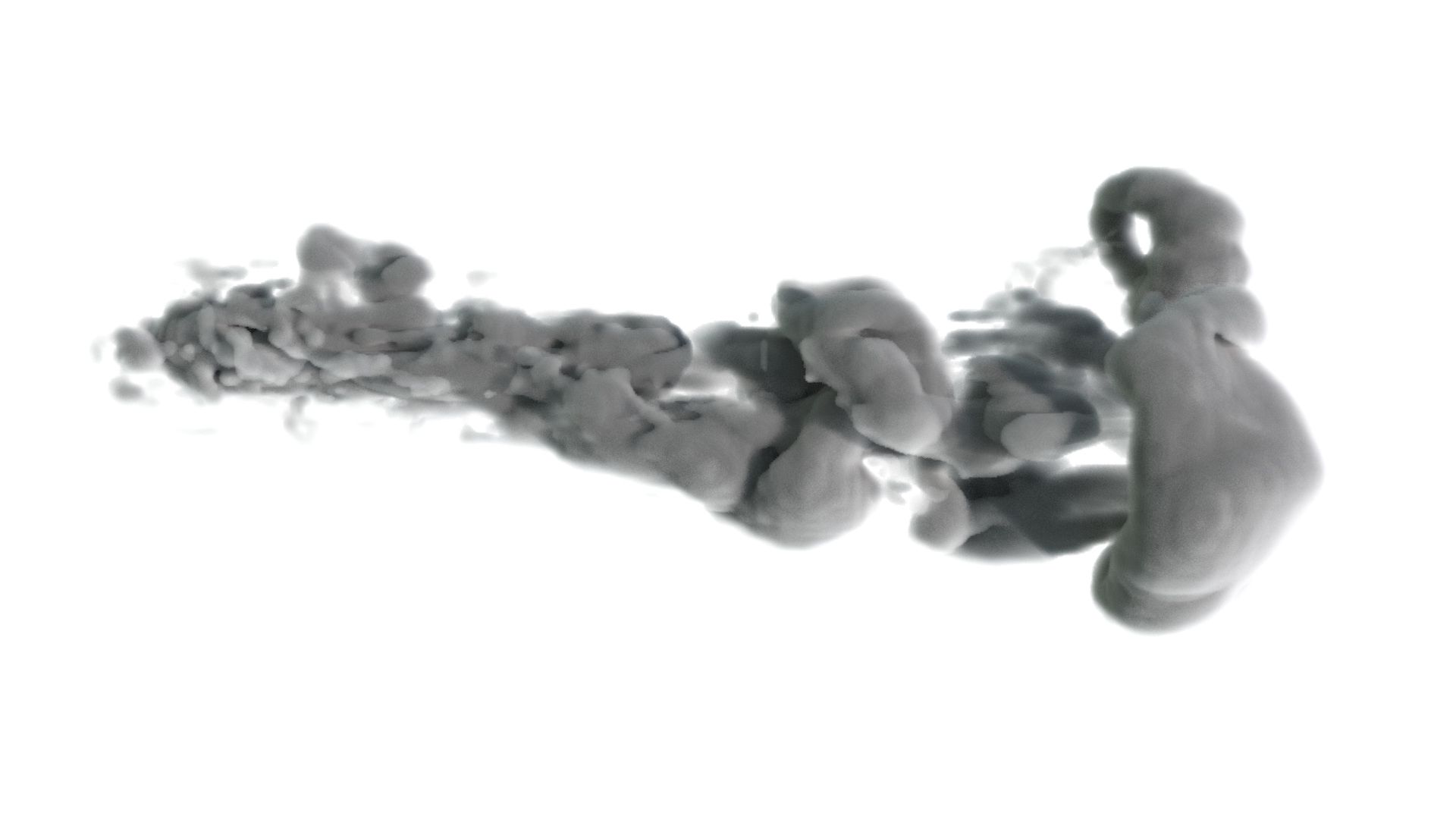}
    \includegraphics[trim={150px 140px 0 140px},clip,width=0.4\columnwidth,angle=90]{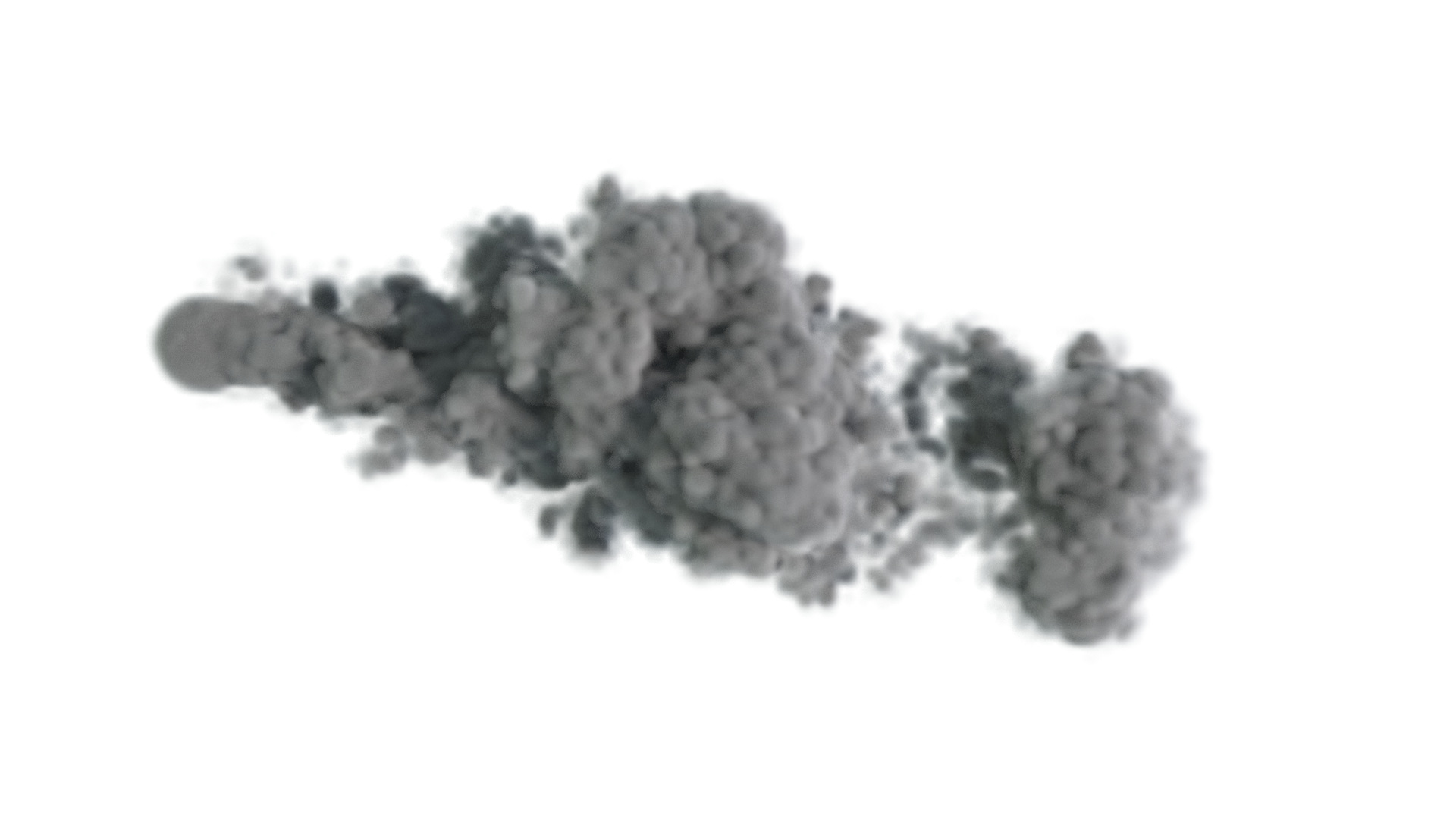}
    \\
    \begin{picture}(0,0)(0,0)
        \put(-113,225){\sffamily \footnotesize \rotatebox{90}{$96\times192\times96$}}
        \put(-113,140){\sffamily \footnotesize \rotatebox{90}{$64\times128\times64$}}
        \put(-113,5){\sffamily \footnotesize PolyPIC}
        \put(-67,5){\sffamily \footnotesize PolyFLIP}
        \put(-26,5){\sffamily \footnotesize CF+PolyFLIP}
        \put(24,5){\sffamily \footnotesize R+PolyFLIP}
        \put(68,5){\sffamily \footnotesize \textbf{CO-FLIP (Ours)}}
        \put(-113,40){\sffamily \footnotesize \rotatebox{90}{$32\times64\times32$}}
    \end{picture}

    \caption{Smoke plume 3D.
    At a fixed CFL number, with increasing spatiotemporal resolution, the vortical structures present in the smoke get denser and more abundant.
    Note that our method produces increased vortical features, reminiscing the fractal features seen in pyroclastic clouds.
    This behavior is present even at a very low resolution of $32\times64\times32$.
    }
    \label{fig:smokeplumes}
\end{figure}

\subsection{Properties of CO-FLIP}
\label{sec:PropertiesOfCOFLIP}
\begin{corollary}
    Solutions  \((\vec x(t),\vec u(t))\) to the CO-FLIP equation  \eqref{eq:COFLIPonM} on \(M\) are \(\cI\)-discrete Euler flow.  In particular it preserves the energy \((H_{\rm D}\circ\cI^\adjoint\circ J_{\adv})(\vec x(t),\vec u(t))\) and \(J_{\adv}(\vec x(t),\vec u(t))\) stays on a coadjoint orbit.
\end{corollary}
\begin{theorem}\label{thm:EnergyPreservingSigma}
    The integrator \(s^{(\sfn)}\to s^{(\sfn+1)}\) by \eqref{eq:ImplicitMidpointSigma} exactly preserves the energy \(\cE(s)\coloneqq H_{\rm D}\circ\cI^\adjoint\circ J_{\adv}(s)\) and exactly preserves the coadjoint orbit for \(J_{\div}(s)\).
\end{theorem}
\begin{proof}
    \appref{app:EnergyPreservingSigma}
\end{proof}
\begin{corollary}
\label{cor:ImplicitMidpointOnMIsEnergyPreserving}
    The integrator \eqref{eq:ImplicitMidpointOnM} exactly preserves the energy and the coadjoint orbit.
\end{corollary}
\begin{theorem}\label{thm:COFLIPCoadjointPreserving}
    The finite dimensional CO-FLIP ODE system \eqref{eq:COFLIP-ODE} and the integrator \eqref{eq:ImplicitMidpoint} preserves the coadjoint orbit.
\end{theorem}
\begin{proof}
    \appref{app:COFLIPCoadjointPreserving}
\end{proof}
    The systems \eqref{eq:COFLIP-ODE} and \eqref{eq:ImplicitMidpoint} preserves the coadjoint orbit.  However the energy can drift depending on the approximation of the integral \(\int_M|\cdot|^2\,d\mu\) by the summation \(\sum_{\sfp\in\cP}|(\cdot)_{\sfp}|^2\mu_{\sfp}\), as described in \secref{sec:EnergyBasedCorrection}.
\begin{theorem}\label{thm:EnergyCorrection}
    The method \eqref{eq:EnergyCorrection} exactly preserves the coadjoint orbit and the energy.
\end{theorem}
\begin{proof}
    \appref{app:EnergyCorrection}.
\end{proof}

\section{Theory Part 2: Mimetic Interpolation}
\label{sec:Theory2}

The theory of \secref{sec:Theory1} and the resulting method of \secref{sec:Method} rely heavily on a divergence-free interpolation scheme \(\cI\colon\fB\hookrightarrow\fX(W)\), \(\cI\colon\fB_{\div}\hookrightarrow\fX_{\div}(W)\).
In this section, we describe the properties and construction of \(\cI\) using B-splines.

\paragraph{Highlights} A special quality of our divergence-free interpolation is that the evaluation is explicit and local.  This is in contrast to the recent divergence-free interpolation methods \cite{Chang:2022:CurlFlow,Lyu:2024:WPE} introduced in fluid simulations in computer graphics that require a global streamfunction integration.  
Moreover, our interpolation scheme is accompanied with a family of interpolation schemes that are curl-consistent and gradient-consistent.  These interpolations are useful for implementation details such as building various versions of pressure solvers.

\paragraph{Overview} We employ a family of \textbf{mimetic interpolation} schemes, widely studied in isogeometric analysis and finite element exterior calculus.  These interpolation methods \(\cI_k\) interpolate differential \(k\) forms and commute with the exterior derivative operators.  
Our divergence-free interpolation is the interpolation \(\cI_{n-1}\) for \((n-1)\)-forms (flux forms), followed by the canonical conversion between a flux form and a vector field \(\ip_{\vec v}(d\mu)\in\Omega^{n-1}(W)\leftrightarrow \vec v\in\fX(W)\).
Closed flux form translates to divergence-free vector fields.

\subsection{Mimetic interpolation}
Let 
\begin{align}
    C^0\xrightarrow{\bd_0}C^1\xrightarrow{\bd_1}C^2\xrightarrow{\bd_2}\cdots\xrightarrow{\bd_{n-1}}C^n
\end{align}
be a real-coefficient cochain complex, where each \(C^k\) is a finite dimensional real vector space.%
\footnote{Do not confuse the cochain space \(C^k\) with the space \(\continuity^k\) of \(k\)-time continuously differentiable functions.}
Each cochain space \(C^k\) represents the space of discrete \(k\)-forms, and the operators \(\bd_k\) that satisfy the cochain condition \(\bd_{k+1}\bd_k = \bzero\) are discrete exterior derivatives.
\begin{definition}[Mimetic interpolation]
    A family of interpolation schemes  $\cI_k:C^k\to\Omega^k(W)$ is \emph{mimetic} if it commutes with the exterior derivatives:
    \begin{align}
        d_k(\cI_k\ba) = \cI_{k+1}(\bd_k\ba)\quad\text{for all \(\ba\in C^k\), for all \(k=0,\ldots,n-1\).}
    \end{align}
    Here $ d_k:\Omega^k(W)\to\Omega^{k+1}(W)$ is the exterior derivative in the smooth theory.
\end{definition}

Diagrammatically, a mimetic interpolation relates a finite dimensional cochain complex to the de Rham complex in the following commutative diagram
\begin{align}
\begin{aligned}
    \xymatrix{
    C^0 
    \ar[r]^{\bd_0}
    \ar[d]^{\cI_0}
    &
    C^1
    \ar[r]^{\bd_1}
    \ar[d]^{\cI_1}
    &
    C^2
    \ar[r]^{\bd_2}
    \ar[d]^{\cI_2}
    &
    \cdots
    \ar[r]^{\bd_{n-1}}
    &
    C^n
    \ar[d]^{\cI_n}
    \\
    \Omega^0
    \ar[r]^{d_0}
    &
    \Omega^1
    \ar[r]^{d_1}
     &
    \Omega^2
    \ar[r]^{d_2}
    &
    \cdots
    \ar[r]^{d_{n-1}}
    &
    \Omega^n
    }
\end{aligned}
\end{align}

\subsubsection{Basis Representation}
Let \(
    \beta^k_{\sfi}\in\Omega^k(W), \sfi = 1,\ldots,\dim(C^k)\),
be a basis \(k\)-forms for \(\cI_k(C^k)\subset\Omega^k(W)\) so that \(\cI_k\) can be expressed as a map that takes in a list of coefficients \(\ba = (a_1,\ldots,a_{\dim(C^k)})\), representing an element of \(C^k\), to a \(k\)-form \(\cI_k(\ba)\in\Omega^k(W)\) by
\begin{align}
\textstyle
    \cI_k(\ba) = \sum_{\sfi = 1}^{\dim(C^k)} a_\sfi \beta^k_{\sfi}.
\end{align}
\subsubsection{Galerkin Hodge Stars}
The mimetic interpolations \(\cI_k\colon C^k\to\Omega^k(W)\) induces a metric on \(C^k\) from the \(L^2\) structure on \(\Omega^k(W)\).  The \(L^2\) inner product between \(\omega_1,\omega_2\in\Omega^k(W)\) is given by
\begin{align}
    \llangle\omega_1,\omega_2\rrangle_{\Omega^k(W)} = \int_W\omega_1\wedge \ast \omega_2,
\end{align}
where \(\ast\) is the Hodge star in the smooth theory.  The induced inner product structure on \(C^k\) is a positive definite self-adjoint operator 
\(\star_k\colon C^k\to C^{k*}\), where \(C^{k*}\) is the dual space of \(C^k\), so that
\begin{align}
    \ba^\intercal\star_k\bb=\langle \ba|\star_k\bb\rangle = \llangle \cI_k\ba,\cI_k\bb\rrangle_{\Omega^k(W)},\quad\forall\,\ba,\bb\in C^k.
\end{align}
In basis expression, \(\star_k\) is the inner product matrix with entries 
\begin{align}
    (\star_{k})_{\sfi\sfj} = \llangle\beta_\sfi^k,\beta_\sfj^k\rrangle_{\Omega^k(W)}.
\end{align}

\subsubsection{General Construction}
There are general constructions for mimetic interpolations starting with a scalar function interpolant \(\cI_0\colon C^0\to\Omega^0(W)\).
An interpolant $\cI_0:C^0\to\Omega^0(W)$ is said to satisfy \textbf{partition of unity} if there exists \(\ba\in C^0\) such that $\cI_0\ba =  1$.
Hiemstra \etal~\shortcite{Hiemstra:2014:HOG} showed that any interpolant $\cI_0$ which satisfies partition of unity would create a series of interpolants $\cI_k$ for all $k$-forms that satisfies the mimetic property.
As such, there are many families of mimetic interpolations, such as the ones based on B-splines \cite{Buffa:2011:IGD}, Lagrange interpolating polynomials \cite{Gerritsma:2010:EFS}, and radial basis functions \cite{Narcowich:1994:GHI,Drake:2021:RBF}.

We take B-splines as our choice of mimetic interpolation, similar to the works of  \cite{Buffa:2011:IGD}.
B-splines are piece-wise polynomial, local and compactly supported, non-negative, and are common in the MPM works in computer graphics \citep{Jiang:2016:MPM}.

\subsection{Univariate B-Spline Mimetic Interpolation}
\label{sec:UnivariateBSplineMimeticInterpolation}
Mimetic interpolation for B-splines in multidimensions are built from mimetic univariate B-splines.  We first give an overview of basic properties of univariate B-splines.  See also \cite{Hughes:2005:IGA} for detailed theory of B-splines.

Basis splines (B-splines) are piecewise polynomial functions of a given degree with maximal smoothness and minimal support \cite{Curry:1947:SDL, Curry:1966:PFF}.
On a 1-dimensional parameter domain, consider a knot sequence \(\xi_1\leq\xi_2\leq\cdots\) partitioning the domain into pieces separated by the knots.  
\begin{definition}[B-splines]
    The \(\sfi\)-th degree-\(p\) B-spline \(B_\sfi^p\) is a piecewise degree-\(p\) polynomial supported over interval \([\xi_{\sfi},\xi_{\sfi+p+1}]\) defined recursively by the Cox--de~Boor formula \cite{Cox:1972:NEB, DeBoor:1972:CBS}
    \begin{subequations}
    \begin{align}
        &B^p_\sfi(t) \coloneqq \chi^p_\sfi(t) B^{p-1}_\sfi(t) + \big(1 - \chi^p_{\sfi+1}(t)\big) B^{p-1}_{\sfi+1}(t),\quad p\geq 1, \\
        &B^0_\sfi(t) \coloneqq \begin{cases} 1,\quad t \in [\xi_\sfi, \xi_{\sfi+1}),\\0,\quad \text{otherwise},\end{cases}
        \quad
        \chi^p_\sfi(t) \coloneqq \frac{t-\xi_\sfi}{\xi_{\sfi+p}-\xi_\sfi}.
    \end{align}
    \end{subequations}
\end{definition}
When we use B-splines with a maximal polynomial degree \(p\) over a compact parameter domain \([0,1]\), we set the knot set \(\Xi = (\xi_1\leq\xi_2\leq\cdots\leq\xi_{N+p+1})\) as
\begin{align}
\label{eq:OpenUniformKnotVector}
\begin{aligned}
    &0=\xi_1=\cdots=\xi_{p+1}<\xi_{p+2}<\cdots\\
    &\cdots<\xi_N<\xi_{N+1}=\cdots=\xi_{N+p+1} = 1.
\end{aligned}
\end{align}
The number \(N+p+1\) of knots correspond to the count that there are \(N\) B-splines of degree \(p\), \(\{B_1^p,\ldots,B_N^p\}\), using which we represent an interpolated function from \(N\) coefficients \(\bff = (f_1,\ldots,f_N)\in C^0 = \RR^N\):
\begin{align}
    (\cI_0\bff)(t) \coloneqq \sum_{\sfi=1}^N f_{\sfi}B_{\sfi}^p(t).
\end{align}
In \eqref{eq:OpenUniformKnotVector}, the multiplicity of coinciding knots at the end points ensures that the boundary coefficients \(f_1, f_N\) are interpolated: \((\cI_0\bff)(0) = f_1\) and \((\cI_0\bff)(1) = f_{N}\).

The (exterior) derivative of a B-spline is remarkably a linear combination of lower order B-splines:
\begin{equation}
\label{eq:DerivativeOfBSpline}
    d B^p_\sfi(t) = \big( D^{p-1}_\sfi(t) - D^{p-1}_{\sfi+1}(t) \big)\, dt,
\end{equation}
where \(D_{\sfi}^{p-1}(t)\) is known as the Curry--Schoenberg basis, or M-splines \cite{Curry:1966:PFF}, defined as follows.
\begin{definition}[Curry--Schoenberg M-splines]
    The \(\sfi\)-th degree-\((p-1)\) Curry--Schoenberg M-spline $D_{\sfi}^{p-1}$ is a rescale of the corresponding B-spline, given by
    \begin{align}
    \label{eq:CurrySchoenberg}
        D_{\sfi}^{p-1}(t)\coloneqq 
        \begin{cases}
            {p\over \xi_{\sfi + p} - \xi_{\sfi}}B_{\sfi}^{p-1}(t)& \xi_{\sfi+p}>\xi_{\sfi}\\
            0 & \xi_{\sfi + p} = \xi_{\sfi}.
        \end{cases}
    \end{align}
\end{definition}

In the setup of \eqref{eq:OpenUniformKnotVector}, the formula \eqref{eq:DerivativeOfBSpline} holds for all \(\sfi = 1,\ldots,N\), with \(D_1^{p-1}(t) = D_{N+1}^{p-1}(t) = 0\).

In this one-dimensional setting, we obtain the mimetic interpolation property as we now show.  Define the space \(C^1\) of 1-cochains as \(\RR^{N-1}\). Let \(\bd_0\colon C^0=\RR^N\to C^1=\RR^{N-1}\) be the discrete differential operator given by
\begin{align}
\label{eq:DiscreteD0}
    (\bd_0\bff)_{\sfi} = \bff_{\sfi}-\bff_{\sfi-1}.
\end{align}
Define the interpolation for 1-forms \(\cI_1\colon C^1\to\Omega^1([0,1])\) as 
\begin{align}
    \cI_1\ba = \sum_{\sfi=2}^N a_{\sfi}D_{\sfi}^{p-1}(t)\, dt.
\end{align}
Then we have the following mimetic interpolation property.
\begin{proposition}[Mimetic interpolation in 1D]
\label{prop:MimeticInterpolationIn1D}
    \(d(\cI_0\bff) = \cI_1(\bd_0\bff)\) for all \(\bff\in C^0 = \RR^N\).
\end{proposition}
\begin{proof}
    \(d(\cI_0\bff) = d(\sum_{\sfi=1}^Nf_\sfi B_\sfi^p) =  \sum_{\sfi=1}^N f_\sfi dB_{\sfi}^p 
     \overset{\eqref{eq:DerivativeOfBSpline}}{=} \sum_{\sfi=1}^N f_\sfi(D_{\sfi}^{p-1}-D_{\sfi+1}^{p-1})dt 
     = \allowbreak \sum_{i=2}^N(\bd_0\bff)_{\sfi}D_{\sfi}^{p-1}dt 
     = \cI_1(\bd_0\bff)\).
\end{proof}

\begin{remark}
\label{rmk:GridOfCoefficients}
    The discrete 0-forms, 1-forms, and exterior derivative \eqref{eq:DiscreteD0} are related to each other in a similar fashion as in Discrete Exterior Calculus (DEC).  However, we only think of the data \(\bff\in C^0\) and \(\ba\in C^1\) as abstract coefficients, instead of the evaluation or integration of differential forms on a particular grid.
\end{remark}
\begin{remark}
\label{rmk:GrevilleGrid}
    The parameter domain \([0,1]\) is mapped to a physical domain \([x_{\min},x_{\max}]\) linearly by an affine function \(x\colon [0,1]\to [x_{\min},x_{\max}]\).  This function \(x(t)\) can be represented using B-spline \(x(t) = (\cI_0\bx)(t)\) using a particular list of coefficients \(\bx=(x_1,\ldots,x_N)\in C^0\).  These coefficients \(x_\sfi\) can be regarded as the effective grid points in the physical space, called the \textbf{Greville grid}.  Note that \(\bx\) is in general not uniformly spaced near the boundary, contrary to common computation grids adopted in fluid simulations.
\end{remark}

\begin{figure}
    \centering
    \includegraphics[trim={0 200px 0 0},clip,width=0.32\columnwidth]{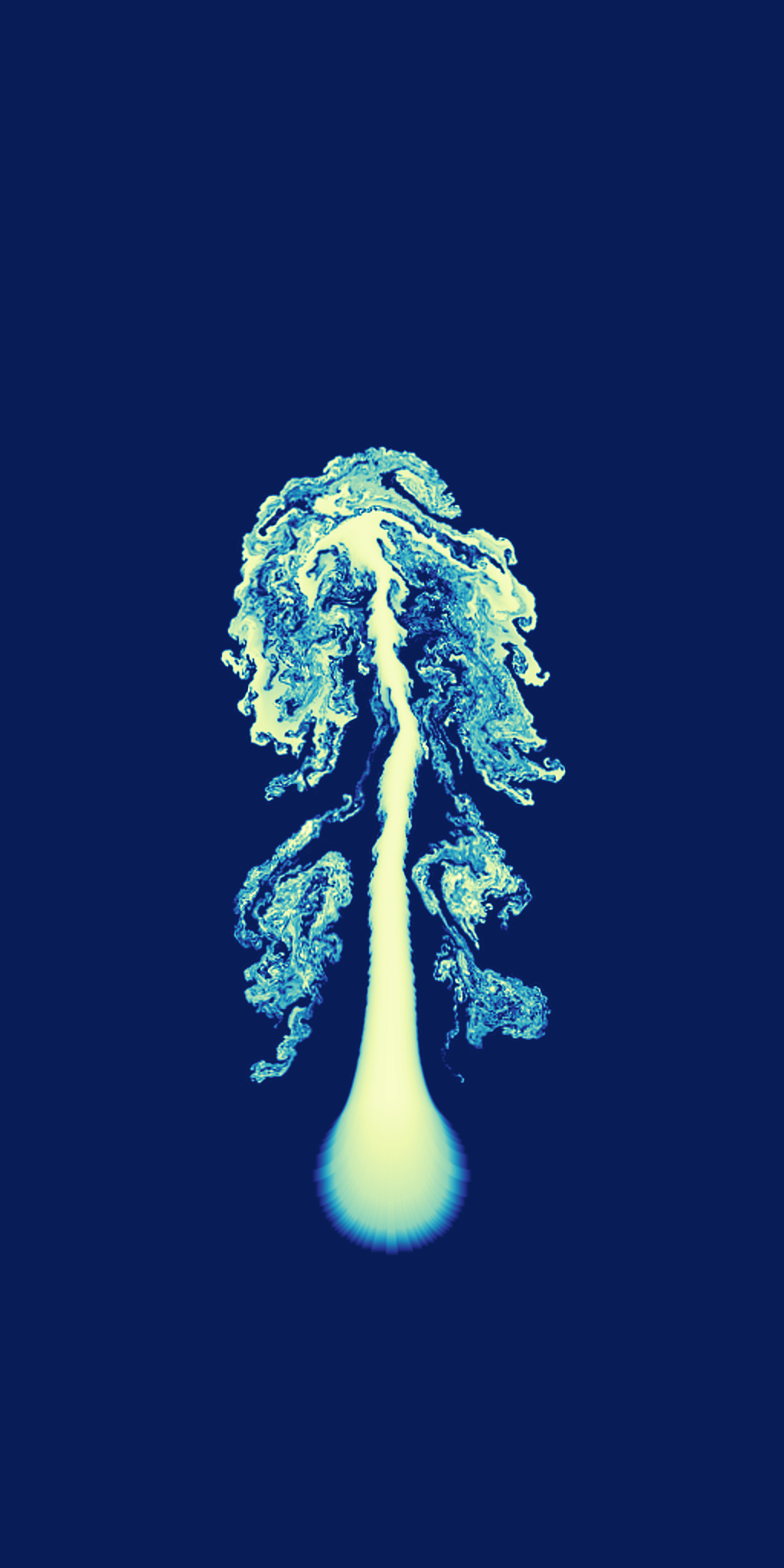}
    \includegraphics[trim={0 200px 0 0},clip,width=0.32\columnwidth]{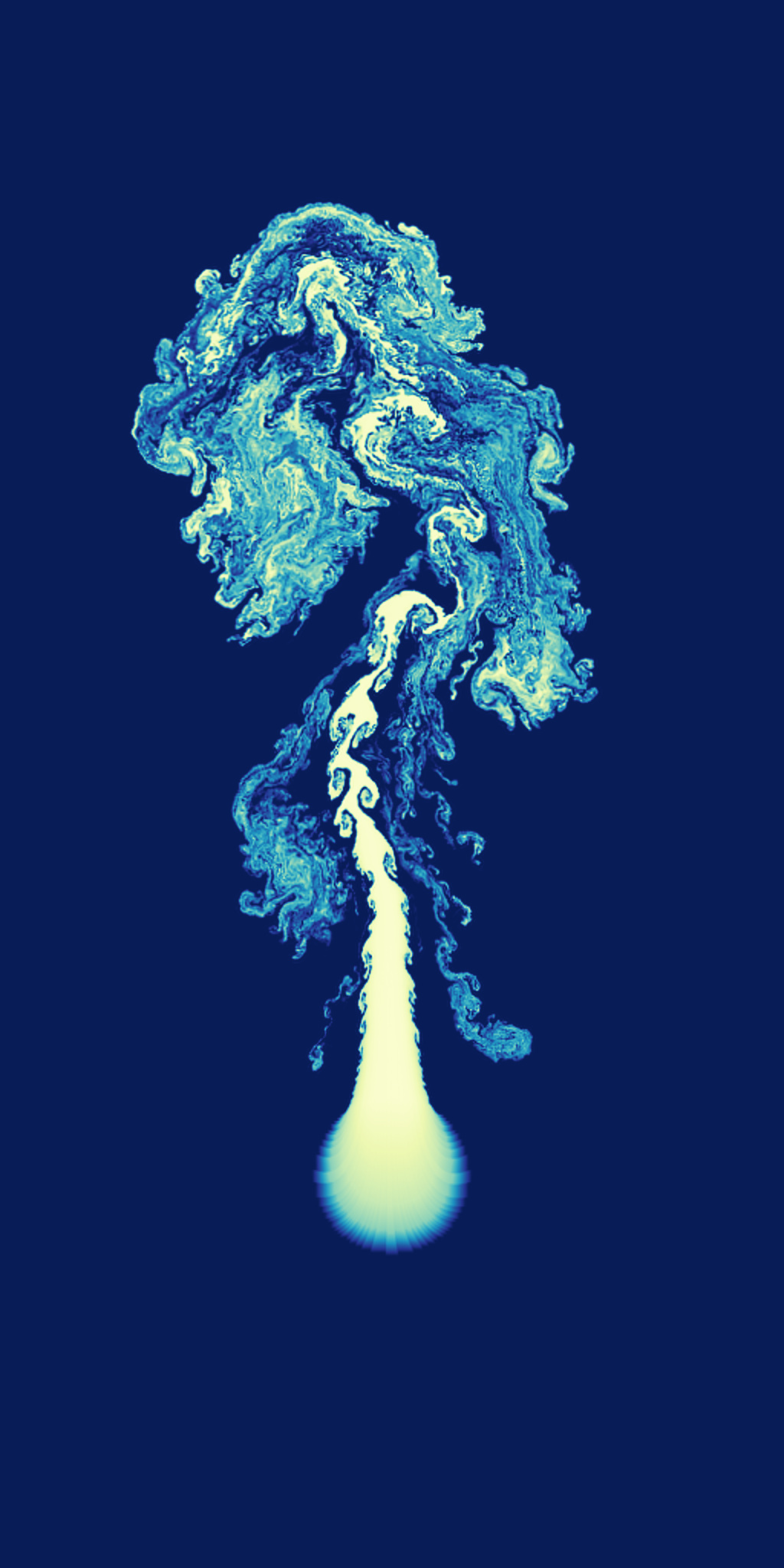}
    \includegraphics[trim={0 200px 0 0},clip,width=0.32\columnwidth]{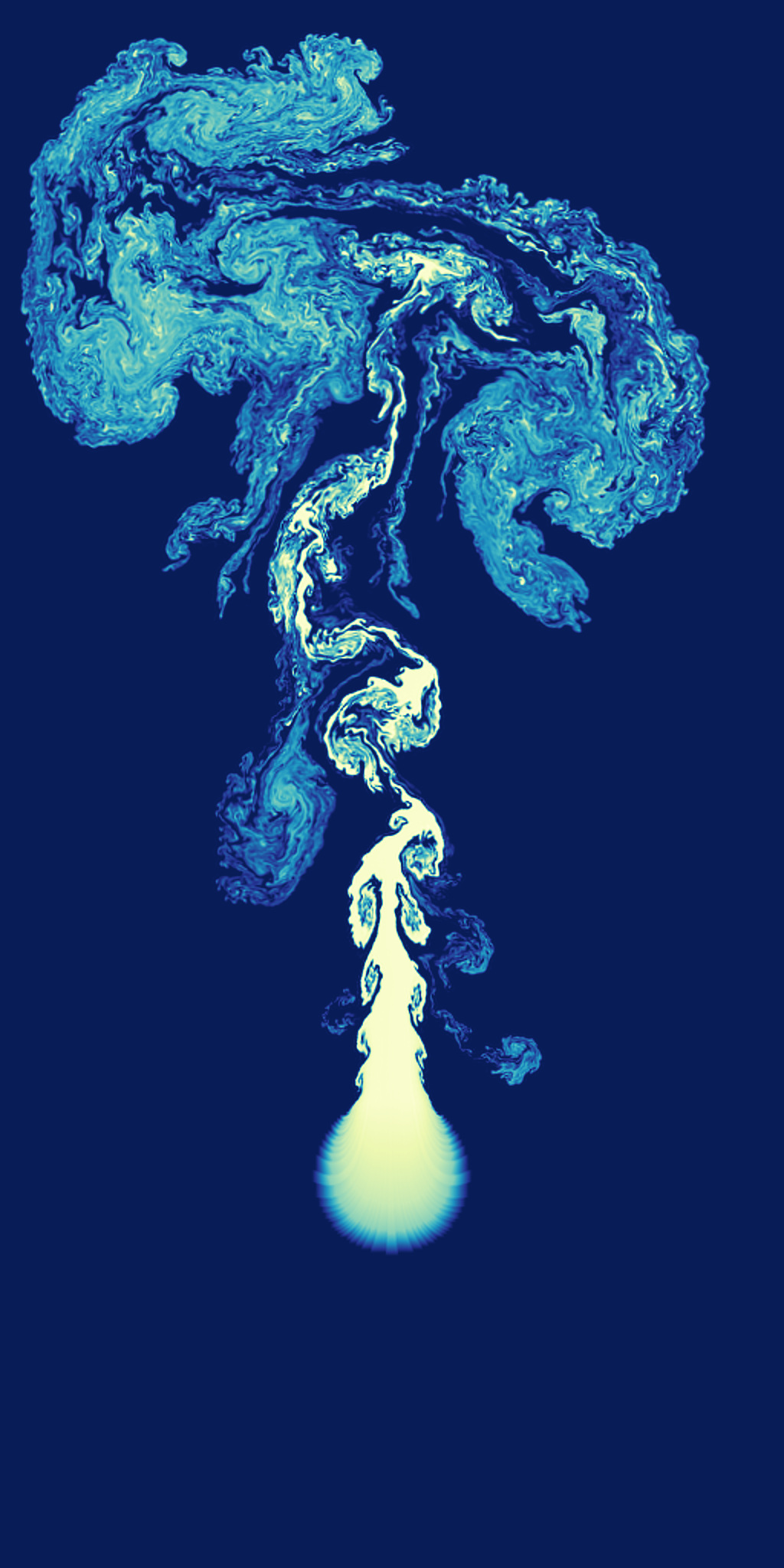}
    \caption{Smoke plume in two dimensions.
    We demonstrate an abundance of vortical structures using our method (CO-FLIP), in the evolution of a buoyant material traveling upwards in a closed box.}
    \label{fig:2dsmokeplume}
\end{figure}

\subsection{Multivariate B-Spline  Mimetic Interpolation}
We now consider B-spines in \(n\) dimension.  We will demonstrate the construction for \(n=3\), whose pattern is easy to generalize to other dimensions.  For general constructions, see \cite{Evans:2013:IGA,Kapidani:2022:HOG}.

Consider three knot sets \(\Xi_x,\Xi_y,\Xi_z\), each of which is a list like \eqref{eq:OpenUniformKnotVector} with size \(N_x+p+1,\ldots,N_z+p+1\).  By the construction depicted in \secref{sec:UnivariateBSplineMimeticInterpolation}, we obtain three families \(\{B_{x,1}^p,\ldots,B_{x,N_x}^p\}\),
\(\{B_{y,1}^p,\ldots,B_{y,N_y}^p\}\),
\(\{B_{z,1}^p,\ldots,B_{z,N_z}^p\}\)
of degree-\(p\) B-splines, and three families 
\(\{D_{x,2}^{p-1},\ldots,D_{x,N_x}^{p-1}\}\),
\(\{D_{y,2}^{p-1},\ldots,D_{y,N_y}^{p-1}\}\),
\(\{D_{z,2}^{p-1},\ldots,D_{z,N_z}^{p-1}\}\)
of degree-\((p-1)\) Curry--Schoenberg M-splines,
respectively over \(\Xi_x,\Xi_y,\Xi_z\).

Now consider an abstract grid with \(N_x\times N_y\times N_z\) vertices.
Similar to the discussion in \remref{rmk:GridOfCoefficients}, this grid is a \textbf{grid of coefficients} that admits a discrete topology identical to a DEC grid.
Define \(C^0\coloneqq\RR^{N_x N_y N_z} = \{(g_{\sfi\sfj\sfk})_{{1\leq\sfi\leq N_x, 1\leq\sfj\leq N_y, 1\leq\sfk\leq N_z}}\}\) where elements are numbers assigned to the vertices of the grid.
Likewise, define \(C^1\coloneqq \RR^{(N_x-1) N_y N_z}\oplus \RR^{N_x (N_y-1) N_z}\oplus \RR^{N_x N_y (N_z-1)}\) where elements \(((h^x_{\sfi\sfj\sfk}),(h^y_{\sfi\sfj\sfk}),(h^z_{\sfi\sfj\sfk}))\) are values assigned to edges, \(C^2\coloneqq \RR^{N_x (N_y-1) (N_z-1)}\oplus \RR^{(N_x-1)N_y(N_z-1)}\oplus\RR^{(N_x-1)(N_y-1)N_z}\) where elements \(((f^x_{\sfi\sfj\sfk}),(f^y_{\sfi\sfj\sfk}),(f^z_{\sfi\sfj\sfk}))\) are values assigned to faces, and 
\(C^3\coloneqq\RR^{(N_x-1)(N_y-1)(N_z-1)}\) where elements \((r_{\sfi\sfj\sfk})_{2\leq\sfi\leq N_x,\ldots,2\leq\sfk\leq N_z}\) are values assigned to cells.
By the boundary relationship between the vertices, edges, faces, and cells, construct the standard discrete de Rham complex in the DEC framework
\begin{subequations}
\begin{align}
&C^0\xrightarrow{\bd_0}C^1\xrightarrow{\bd_1}C^2\xrightarrow{\bd_2}C^3\\
    &\bd_0\bg = (\bpartial_x\bg,\bpartial_y\bg,\bpartial_z\bg)\\
    &\bd_1(\bh^x,\bh^y,\bh^z) = (\bpartial_y\bh^z - \bpartial_z\bh^y,\bpartial_z\bh^x-\bpartial_x\bh^z,\bpartial_x\bh^y-\bpartial_y\bh^x)\\
    &\bd_2(\bff^x,\bff^y,\bff^z) = (\bpartial_x\bff^x + \bpartial_y\bff^y + \bpartial_z\bff^z),
\end{align}
\end{subequations}
where \(\bpartial_x,\bpartial_y,\bpartial_z\) are the coordinate difference operators
\begin{subequations}
\label{eq:CoordinateDifferenceOperators}
\begin{align}
    &(\bpartial_x\cdot)_{\sfi\sfj\sfk}\coloneqq (\cdot)_{\sfi\sfj\sfk} - (\cdot)_{(\sfi-1)\sfj\sfk},\quad
    (\bpartial_y\cdot)_{\sfi\sfj\sfk}\coloneqq (\cdot)_{\sfi\sfj\sfk} - (\cdot)_{\sfi(\sfj-1)\sfk},\\
    &\hspace{7em}(\bpartial_z\cdot)_{\sfi\sfj\sfk}\coloneqq (\cdot)_{\sfi\sfj\sfk} - (\cdot)_{\sfi\sfj(\sfk-1)}.
\end{align}
\end{subequations}
Note that \(C^2\) is the MAC grid for velocity flux data on faces, and \(\bd_2\) is the discrete divergence operator. 

Define \(\cI_0\colon C^0\to\Omega^0([0,1]^3)\) by
\begin{align}
\label{eq:3DMimeticInterpolate0Form}
\textstyle
    (\cI_0\bg)_{(t_1,t_2,t_3)} \coloneqq \sum_{\sfi=1}^{N_x}\sum_{\sfj=1}^{N_y}\sum_{\sfk=1}^{N_z}g_{\sfi\sfj\sfk}B_{x,\sfi}^p(t_1)B_{y,\sfj}^p(t_2)B_{z,\sfk}^p(t_3),
\end{align}
which is the tensor product B-spline interpolation for scalar functions.
Now, instead of employing the common treatment of interpolating \(k\)-forms (\(k\geq 1\)) by the same tensor product B-spline basis such as in tri-quadratic or tri-cubic interpolation schemes, one replaces the B-spline functions by the Curry--Schoenberg spline functions for those directions where the edge/face/cell spans a width.
That is,
define \(\cI_1\colon C^1\to\Omega^1([0,1]^3)\) by
\begin{align}
\label{eq:3DMimeticInterpolate1Form}
\textstyle
\nonumber
    &(\cI_1(\bh^x,\bh^y,\bh^z))_{(t_1,t_2,t_3)}\\ \nonumber
    &\coloneqq 
    \textstyle
    \sum_{\sfi=2}^{N_x}\sum_{\sfj=1}^{N_y}\sum_{\sfk=1}^{N_z}
    h_{\sfi\sfj\sfk}^x D_{x,\sfi}^{p-1}(t_1)B_{y,\sfj}^p(t_2)B_{z,\sfk}^p(t_3)\, dt_1\\ \nonumber
    &\quad{}+
    \textstyle
    \sum_{\sfi=1}^{N_x}\sum_{\sfj=2}^{N_y}\sum_{\sfk=1}^{N_z}
    h_{\sfi\sfj\sfk}^y B_{x,\sfi}^p(t_1)D_{y,\sfj}^{p-1}(t_2)B_{z,\sfk}^p(t_3)\, dt_2\\ 
    &\quad{}+
    \textstyle
    \sum_{\sfi=1}^{N_x}\sum_{\sfj=1}^{N_y}\sum_{\sfk=2}^{N_z}
    h_{\sfi\sfj\sfk}^z B_{x,\sfi}^p(t_1)B_{y,\sfj}^p(t_2)D_{z,\sfk}^{p-1}(t_3)\, dt_3,
\end{align}
\(\cI_2\colon C^2\to\Omega^2([0,1]^3)\) by
\begin{align}
\label{eq:3DMimeticInterpolate2Form}
\textstyle
\nonumber
    &(\cI_2(\bff^x,\bff^y,\bff^z))_{(t_1,t_2,t_3)}\\ \nonumber
    &\coloneqq 
    \textstyle
    \sum_{\sfi=1}^{N_x}\sum_{\sfj=2}^{N_y}\sum_{\sfk=2}^{N_z}
    f_{\sfi\sfj\sfk}^x B_{x,\sfi}^{p}(t_1)D_{y,\sfj}^{p-1}(t_2)D_{z,\sfk}^{p-1}(t_3)\, dt_2\wedge dt_3\\ \nonumber
    &\quad{}+
    \textstyle
    \sum_{\sfi=2}^{N_x}\sum_{\sfj=1}^{N_y}\sum_{\sfk=2}^{N_z}
    f_{\sfi\sfj\sfk}^y D_{x,\sfi}^{p-1}(t_1)B_{y,\sfj}^{p}(t_2)D_{z,\sfk}^{p-1}(t_3)\, dt_3\wedge dt_1\\ 
    &\quad{}+
    \textstyle
    \sum_{\sfi=2}^{N_x}\sum_{\sfj=2}^{N_y}\sum_{\sfk=1}^{N_z}
    f_{\sfi\sfj\sfk}^z D_{x,\sfi}^{p-1}(t_1)D_{y,\sfj}^{p-1}(t_2)B_{z,\sfk}^{p}(t_3)\, dt_1\wedge dt_2,
\end{align}
and \(\cI_3\colon C^3\to\Omega^3([0,1]^3)\) by
\begin{align}
\label{eq:3DMimeticInterpolate3Form}
\textstyle
    &(\cI_3\br)_{(t_1,t_2,t_3)} \\ \nonumber
    &\textstyle\coloneqq \sum_{\sfi=2}^{N_x}\sum_{\sfj=2}^{N_y}\sum_{\sfk=2}^{N_z}r_{\sfi\sfj\sfk}D_{x,\sfi}^{p-1}(t_1)D_{y,\sfj}^{p-1}(t_2)D_{z,\sfk}^{p-1}(t_3)\, dt_1\wedge dt_2\wedge dt_3.
\end{align}
This construction ensures the mimetic property.
\begin{theorem}[Mimetic interpolation in 3D]
    \(d(\cI_k\ba)=\cI_{k+1}(\bd_k\ba)\) for all \(\ba\in C^k\), \(k=0,1,2\).
\end{theorem}
\begin{proof}
    Express the continuous exterior derivative of a differential form \(\omega = \sum_I\omega_I dt_I\), where \(I\) is some multi-index, as \(d\omega = \sum_I{\partial\omega_I\over\partial t_1}dt_1\wedge dt_I + {\partial\omega_I\over\partial t_2}dt_2\wedge dt_I+ {\partial\omega_I\over\partial t_3}dt_3\wedge dt_I\).  The partial derivatives boil down to differentiating univariate \(B\)-splines where \eqref{eq:DerivativeOfBSpline} applies.  Identical to the proof of \propref{prop:MimeticInterpolationIn1D}, one performs a summation by parts to arrive at an expression in terms of the coordinate differences \eqref{eq:CoordinateDifferenceOperators} of the coefficients.
\end{proof}

\subsubsection{Map to Physical Space}
The formulas \eqref{eq:3DMimeticInterpolate0Form}--\eqref{eq:3DMimeticInterpolate3Form} interpolate a cochain \(\in C^k\) to a differential form \(\in \Omega^k([0,1]^3)\) on the canonical parameter domain \([0,1]^3\), which is yet to be a physical domain \(W = [x_{\min},x_{\max}]\times [y_{\min},y_{\max}]\times [z_{\min},z_{\max}]\).
To interpolate cochains into the physical domain, it suffices to apply the linear change of variable as described in \remref{rmk:GrevilleGrid}.  The explicit change of variables are
\begin{align}
    &
    \textstyle
    t_1 = {x-x_{\min}\over x_{\max}-x_{\min}},\quad 
    t_2 = {y-y_{\min}\over y_{\max}-y_{\min}},\quad 
    t_3 = {z-z_{\min}\over z_{\max}-z_{\min}}, \\
    &\textstyle
    dt_1 = {dx\over x_{\max}-x_{\min}},\quad 
    dt_2 = {dy\over y_{\max}-y_{\min}},\quad 
    dt_3 = {dz\over z_{\max}-z_{\min}}.
\end{align}

\subsubsection{Divergence-Free Interpolation}
In 3D, we let \(\fB = C^2\), discrete divergence operator \(\bd = \bd_2\), and \(\cI = \sharp\circ *\circ\cI_2\).  The map \(\cI_2\colon \fB\to \Omega^2(W)\) takes values in a continuous 2-form, the Hodge star \(*\) converts it into a 1-form, and \(\sharp\) makes the final conversion into a vector field \(\fX(W)\).  The divergence-free subspace \(\fB_{\div} = \ker(\bd_2)\) maps to closed 2-forms thanks to the mimetic property; closed 2-form corresponds to divergence-free vector fields \(\in \fX_{\div}(W)\).

\subsubsection{Choice of Knots}
\label{sec:ChoiceOfKnots}
While the theory works for arbitrary knot sequence \eqref{eq:OpenUniformKnotVector}, we let the interior knots \(\xi_{p+2},\ldots,\xi_N\) be evenly spaced in \([0,1]\) for each dimension.  This choice is for the simplicity of index query.  This uniform knot grid is called the \textbf{B\'ezier grid}, which is not to be confused with the Greville grid (\cf\@ \remref{rmk:GrevilleGrid}).
B\'ezier grid is used for deciding which B-splines have support at a given particle location, whereas the concept of Greville grids do not appear in implementation.

\subsubsection{Choice of Degree}
We choose the degree of the mimetic interpolation to be \(p=3\).  That is, for the velocity interpolation \(\cI_2\) from \(C^2\) to \(\Omega^2(W)\), the interpolation basis is a mixture of cubic B-splines and quadratic B-splines.  This is the minimal degree for the corresponding interpolated velocity to have a \(\continuity^1\) continuity, which is required for computing \eqref{eq:COFLIP-ODEb}.  

\subsubsection{Inclusion of Arbitrary Obstacles}
\label{sec:Obstacle}
While the mimetic interpolation detailed in this section works for boxed domain \(W = [x_{\min},x_{\max}]\times [y_{\min},y_{\max}]\times [z_{\min},z_{\max}]\), one can include arbitrary obstacles in it.  
This can be achieved by modifying the Galerkin Hodge star \(\star\) induced by \(\cI\colon\fB\to\fX(W)\).  Set the kinetic energy \({1\over 2}\bff^{\intercal}\star\bff = \|\cI(\bff)\|^2\) to infinity if the support of the vector field \(\cI(\bff)\) overlaps with the interior of the obstacle.
This amounts to modifying the Galerkin Hodge star by surrounding it with a diagonal masking matrix
\(
    \widetilde\star = \bLambda\star\bLambda,
\)
where \(\bLambda = \operatorname{diag}(\lambda_1,\ldots,\lambda_{\dim(\fB)})\), \(\lambda_{\sfi} = \infty\) if the basis vector field \(\vec b_{\sfi}\in\fX(W)\) (\cf\@ \eqref{eq:cIAsInterpolationScheme}) overlaps with the obstacle, and otherwise \(\lambda_{\sfi} = 1\).
In the CO-FLIP dynamics, the Galerkin Hodge star appears in the pressure projection \eqref{eq:DiscretePressureProjection} in the inverse form
\(
    \widetilde\star^{-1} = \bV\star^{-1}\bV,
\)
where \(\bV = \bLambda^{-1} = \operatorname{diag}(V_1,\ldots,V_{\dim(\fB)})\), \(V_{\sfi} = 0\) in the obstacle and otherwise \(V_{\sfi} = 1\).  According to how Hodge star appears in \eqref{eq:DiscretePressureProjection}, this is effectively removing the columns of \(\bd\) that is incident to an obstacle.

In practice, to include obstacles, we keep the Galerkin Hodge star unchanged and, instead, remove the entries from $\bd$ associated with \(\fB = C^2\) grid elements incident to the obstacles.

\section{Implementation Details}
\label{sec:Implementation}
This section provides practical notes on improving CO-FLIP in terms of computational robustness and efficiency. Additionally, notes on Casimir measurements are included.

\begin{figure}
    \centering
    \includegraphics[width=0.3\columnwidth]{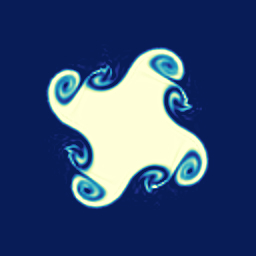}
    \includegraphics[width=0.3\columnwidth]{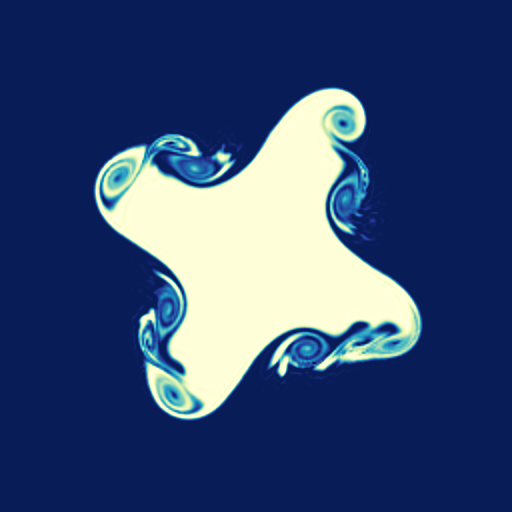}
    \includegraphics[width=0.3\columnwidth]{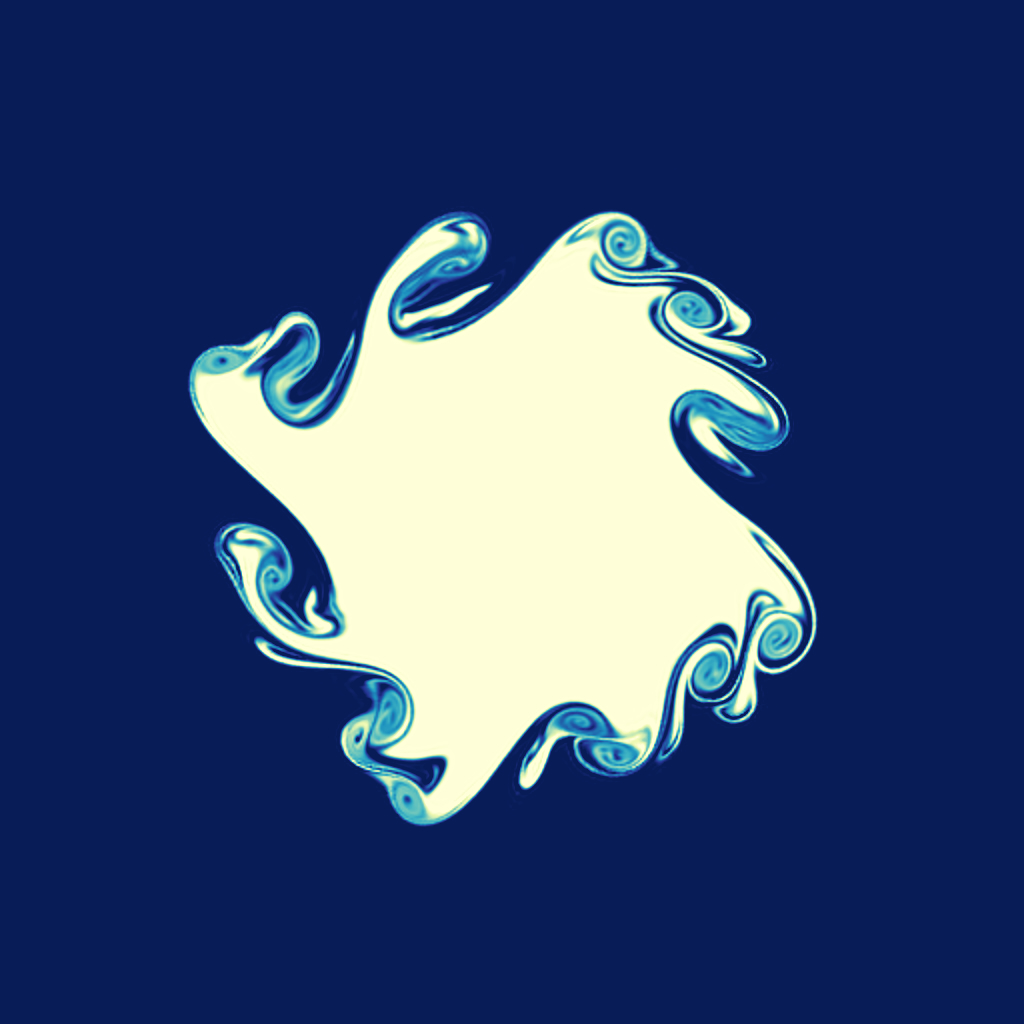}
    \includegraphics[width=0.3\columnwidth]{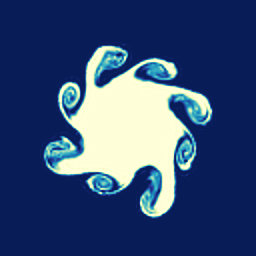}
    \includegraphics[width=0.3\columnwidth]{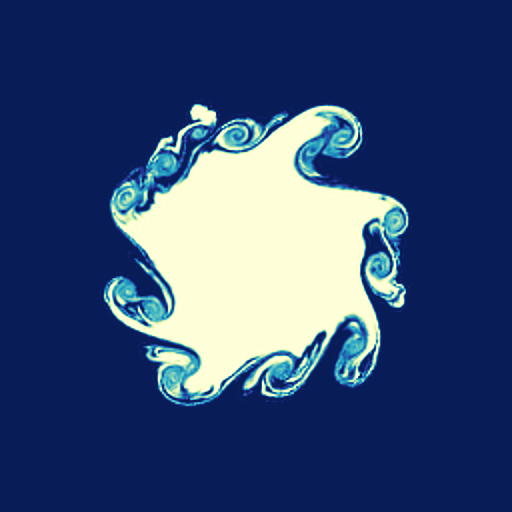}
    \includegraphics[width=0.3\columnwidth]{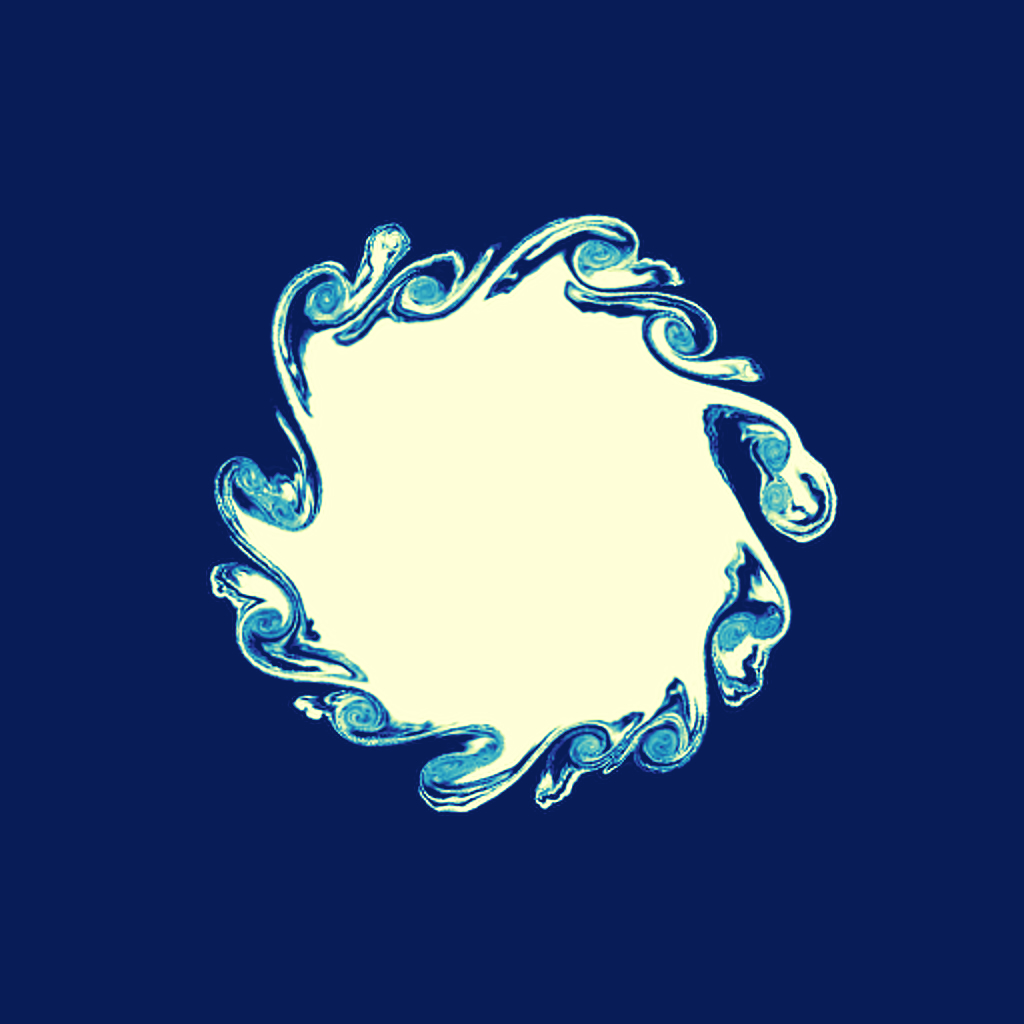}
    \\
    \begin{picture}(0,0)(0,0)
        \put(-106,13){\sffamily \scriptsize \textcolor{white}{$64\times64$ and $\Delta t=1/18s$}}
        \put(-34,13){\sffamily \scriptsize \textcolor{white}{$128\times128$ and $\Delta t=1/36s$}}
        \put(41,13){\sffamily \scriptsize \textcolor{white}{$256\times256$ and $\Delta t=1/72s$}}
        \put(-109,107){\sffamily \scriptsize \rotatebox{90}{\textcolor{white}{CF+PolyFLIP}}}
        \put(-109,25){\sffamily \scriptsize \rotatebox{90}{\textcolor{white}{\textbf{CO-FLIP (Ours)}}}}
    \end{picture}
    \caption{2D vortex sheet under spatiotemporal refinement. We compare CO-FLIP method (ours) against state-of-the-art method CF+PolyFLIP.  Note that CO-FLIP produces more turbulent vortical structures at every resolution. Further, the lower resolution CO-FLIP results show similar qualitative behavior to higher resolutions of the previous method.}
    \label{fig:vortex_sheet}
\end{figure}

\subsection{Pressure Projection}
For the clarity of exposition of this discussion, we will use the cochain notation in \secref{sec:Theory2} for \(n\)-dimension. 
 In particular \(\fB = C^{n-1}\) and \(\fB_{\div} = \ker(\bd_{n-1})\).
 
The pressure projection operator \(\bP_{\fB_{\div}}\colon C^{n-1}\to \ker(\bd_{n-1})\) (\ie\@ \(\fB\to\fB_{\div}\)) acts on the space of MAC-grid flux data \(\bff\in C^{n-1}\).
However, unlike the classical central difference Poisson solve employed for most fluid simulators on MAC-grids, the Laplacian \((\bd_{n-1}\star_{n-1}^{-1}\bd_{n-1}^\intercal)\) in \eqref{eq:DiscretePressureProjection}, or
\begin{subequations}
\label{eq:PressureSolve}
\begin{align}
\label{eq:PressureSolveA}
    &\textsc{Solve}\quad(\bd_{n-1}\star_{n-1}^{-1}\bd_{n-1}^\intercal)\bp = \bd_{n-1}\bff\quad\textsc{for}\quad \bp\in C^{n*}\\
    \label{eq:PressureSolveB}
    &\textsc{Update}\quad\bff_{\rm new}\gets\bff - \star_{n-1}^{-1}\bd_{n-1}^\intercal\bp,
\end{align}
\end{subequations}
involves a \emph{dense} inverse matrix \(\star_{n-1}^{-1}\) of a \emph{non-diagonal} Galerkin Hodge star.
Explicitly, building and storing this inverse Hodge star or the Laplace matrix would be memory infeasible.  

\secref{sec:MatrixFreePressureProjection}--\secref{sec:StreamformSolve} are three ways to avoid explicitly inverting the Galerkin Hodge star.

\subsubsection{Matrix-Free Pressure Projection}
\label{sec:MatrixFreePressureProjection}

When solving this Poisson problem \eqref{eq:PressureSolveA} using a conjugate gradient (CG) method, one could use a matrix-free operation which only involves evaluating the the Laplace operator \((\bd_{n-1}\star_{n-1}^{-1}\bd_{n-1}^\intercal)\) acting on a vector.  This at most involves a linear solve for evaluating \(\star_{n-1}^{-1}\) acting on a vector, which can once again be computed by a CG that only involves actions by \(\star_{n-1}\).

Here \(\star_{n-1}\) can also be implemented matrix-free since most entries of this inner product matrix are identical due to the regularity of the grid (\secref{sec:ChoiceOfKnots}).

\subsubsection{Fourier Domain}
\label{sec:FourierDomain}
If the fluid domain \(W\) is a periodic domain, then the Fourier spectral method applies. The Fast Fourier Transform (FFT) diagonalizes all Galerkin Hodge stars \(\star_{0},\ldots,\star_{n}\) and discrete exterior derivatives \(\bd_0,\ldots,\bd_{n-1}\).  Hence \eqref{eq:PressureSolve} only amounts to an elementwise arithmetic operation.  FFT and its inverse are the only costs for the pressure solve.

\subsubsection{Streamform Solve}
\label{sec:StreamformSolve}
The projection \eqref{eq:PressureSolve} can also be computed using a streamform (streamfunction).  
Let us assume that the fluid domain \(W\) is simply connected.  We will discuss the case of non-simply-connected domains in \remref{rmk:Cohomology}.

On simply connected domains, the projected flux \(\bff_{\rm new}\in \ker(\bd_{n-1})\) is in the image of \(\bd_{n-2}\); \ie\@ \(\bff = \bd_{n-2}\bpsi\) for some \textbf{streamform} \(\bpsi\in C^{n-1}\) that vanishes at the boundary.  For \(n=2\), the boundary condition is a zero Dirichlet boundary condition for the streamfunction \(\bpsi\in C^0\), and for \(n=3\), the stream 1-form vanishes on the edges that are contained in the boundary.  This latter condition in 3D corresponds to the fact that the streamform as a vector field is normal to the boundary.  These boundary values are easy to impose thanks to the boundary-interpolating property of the knot sequence \eqref{eq:OpenUniformKnotVector}.
We let the boundary condition be denoted by \(\bpsi|_{\partial W} = \bzero\).

Now, the equation satisfied by \(\bpsi\) is the Poisson system 
\begin{align}
\label{eq:StreamformPoissonWithoutGauge}
    (\bd_{n-2}^\intercal\star_{n-1}\bd_{n-2})\bpsi = \bd_{n-2}^\intercal\star_{n-1}\bff,\quad \bpsi|_{\partial W} = \bzero.
\end{align}
For \(n=2\) the equation \eqref{eq:StreamformPoissonWithoutGauge} becomes \((\bd_0^\intercal\star_1\bd_0)\bpsi = \bd_0^\intercal\star_1\bff\) which is a scalar Poisson equation with Dirichlet boundary, which has a unique solution.
Note that this is a sparse linear solve as there is no involvement of the inverse of any Galerkin Hodge star.

In 3D, however, \((\bd_1^\intercal\star_{2}\bd_1)\) has a kernel, being the image of \(\bd_0\).
A standard treatment in Discrete Exterior Calculus is to fix this gauge degree of freedom by replacing \((\bd_1^\intercal\star_2\bd_1)\) (discrete \(\delta d\)) by the full Hodge Laplacian \((\bd_1^\intercal\star_2\bd_1 + \star_1\bd_0\star_0^{-1}\bd_0^\intercal\star_1)\) (discrete \(\delta d + d\delta\)).  
Notice that this treatment would introduce the inverse of a Galerkin Hodge star.
However, by the following theorem, we point out that the gauge-fixing term \(\star_1\bd_0\star_0^{-1}\bd_0^\intercal\star_1\) does not require that the stars are Galerkin.  In fact, all the \(\star\)'s in the gauge-fixing term can be replaced by \emph{any} symmetric positive definite matrix (\eg\@ the identity matrix) without changing the equality \eqref{eq:StreamformPoissonWithoutGauge}.

\begin{theorem}
In 3D, the solution \(\bpsi\in C^1\) to
\begin{align}
\label{eq:StreamformPoissonWithGauge}
    (\bd_1^\intercal\star_2\bd_1 + \bA\bd_0\bB\bd_0^\intercal\bA)\bpsi = \bd_1^\intercal\star_2\bff,\quad \bpsi|_{\partial W} = \bzero
\end{align}
satisfies \((\bd_1^\intercal\star_2\bd_1)\bpsi = \bd_1^\intercal\star_2\bff\), where \(\bA,\bB\) are arbitrary symmetric positive definite matrices.
\end{theorem}
\begin{proof}
    Suppose \(\bpsi\) is a solution to \eqref{eq:StreamformPoissonWithGauge}.
    Multiplying \(\bpsi^\intercal\bA\bd_0\bB\bd_0^\intercal\) from the left on \(\eqref{eq:StreamformPoissonWithGauge}\) and noting that \(\bd_0^\intercal\bd_1^\intercal = \bzero\), one obtains 
    \begin{align}
        \bpsi^\intercal\bA\bd_0\bB\bd_0^\intercal\bA\bd_0\bB\bd_0^\intercal\bA\bpsi = \bzero.
    \end{align}
    This implies \(\bd_0\bB\bd_0^\intercal\bA\bpsi = \bzero\).  Therefore \eqref{eq:StreamformPoissonWithGauge} implies \((\bd_1^\intercal\star_2\bd_1)\bpsi = \bd_1^\intercal\star_2\bff\).
\end{proof}

In practice we solve the following system in 3D.
\begin{subequations}
\label{eq:3DStreamformSolve}
\begin{align}
\label{eq:3DStreamformSolveA}
    &\textsc{Solve}\quad(\bd_1^\intercal\star_2\bd_1 + \bd_0\bd_0^\intercal)\bpsi = \bd_1^\intercal\star_2\bff,\quad\bpsi|_{\partial W} = \bzero,\\
    \label{eq:3DStreamformSolveB}
    &\textsc{Update}\quad\bff_{\rm new}\gets\bd_1\bpsi.
\end{align}
\end{subequations}
Importantly this is a sparse linear solve as there is no more involvement of any inverse Galerkin Hodge star.  

\begin{remark}[Non-simply-connected domains]\label{rmk:Cohomology}
If the domain is non-simply-connected, then we have \(\bff_{\rm new} = \bd_{n-2}\bpsi + \bh\) for some cohomology component \(\bh\).  Here \(\bh\in \cH^{n-1} \coloneqq \{\bd_{n-2}\bpsi\,|\,\bpsi|_{\partial W} = \bzero\}^\bot \cap \ker(\bd_{n-1})\) is some divergence-free form that is orthogonal to streamform fields.  The space \(\cH^{n-1}\) is finite dimensional.  Elements of \(\cH^{n-1}\) are close to harmonic: after interpolation \(\cI_{n-1}\) they are exactly divergence-free but not exactly curl-free \cite{RoyChowdhury:2024:HOD}.
Nevertheless this space serves as a representative for the cohomology space.
To obtain this cohomology component \(\bh\in \cH^{n-1}\) in \(\bff_{\rm new} = \bd_{n-2}\bpsi + \bh\), simply orthogonally project \(\bff\) to \(\cH^{n-1}\).

Note that we do not need to evolve the cohomology component \(\bh\) as described in \cite{Yin:2023:FC}.  In contrast to a vorticity method that requires evolving the cohomology component, the field \(\bff\) comes from a particle covector field, whose advection encodes the correct dynamics in its cohomology.
\end{remark}

\subsubsection{Acceleration for Linear Solves}
\label{sec:AccelerationForLinearSolves}
To speed up the involved CG solve for \eqref{eq:3DStreamformSolveA}, we employ a preconditioner.
We use a geometric multi-grid (GMG) for the Poisson problem on 1-cochains \(C^1\) defined over edges.
Similar speed-ups have previously been proposed for solvers involving the face elements \cite{Aanjaneya:2019:EGMVL}.
For this, we require two operators:
the \emph{prolongator}, which involves the operation of assigning values to a finer resolution grid, and the \emph{restrictor}, which entails the inverse operation, where coarser edges are used to reconstruct values on a finer mesh.
These prolongator and restrictor operators can be constructed by the B-spline interpolation and its pseudoinverse, which can inherit any order of accuracy from the choice of mimetic B-spline interpolations.
This method is reminiscent of the works of
 \cite{Bell:2008:AMG} where they perform such operations for Algebraic multigrid methods, and to the work of \cite{DeGoes:2016:SEC} of extending DEC machinery to subdivision surfaces.
For simplicity and to increase the sparsity of coarser-level Laplacians, we choose the lowest accuracy for a mimetic prolongation (\ie\@ B-spline degree $p=1$).
Additionally, we perform a row-wise lumping of the restrictor (aka pseudoinverse interpolation) operation to lower the cost of restriction operator.
Note that such approximations do not lower the accuracy of the CG solve as these approximations are used only for the preconditioner.
With the preconditioner added, the number of iterations for the PCG solver to converge is roughly reduced by three times.

\subsection{Adding Pressure Force Back to Particles}

In the classical FLIP method, one adds the interpolated pressure force inferred from the pressure projection step to the particle's momentum/impulse.
We can do so 
in the continuous theory, since 
the particle impulse \(\vec u\) has a gauge degree of freedom by addition of the gradient of a scalar field (\cf\@ \eqref{eq:ImpulseEquationVector} and \eqref{eq:HamiltonianFlowAsLieAdvection}).  Adding the interpolated pressure force to particles can keep \((\vec x_{\sfp},\vec u_{\sfp})_{\sfp\in\cP}\) close to being divergence-free.  Close to being divergence-free helps maintaining the stability of the overall method.

Here, we describe two methods (Sections~\ref{sec:DivConsistentInterpolationOfPressureForce} and \ref{sec:CurlConsistentInterpolationOfPressureForce}) of adding pressure force back to particles.  Let us denote the pressure force on the grid by \(\btau\coloneqq \bff_{\rm new} - \bff\in C^{n-1}\) during a pressure solve (by \eqref{eq:PressureSolve} or \eqref{eq:3DStreamformSolve}).

\subsubsection{Divergence-Consistent Interpolation of Pressure Force}
\label{sec:DivConsistentInterpolationOfPressureForce}
We can add the pressure force on the particle impulse data \((\vec u_{\sfp})_{\sfp\in\cP}\) by interpolating the pressure force as a flux form
\begin{align}
\label{eq:AddingPressureForceFlux}
    \vec u_{\sfp}\gets\vec u_{\sfp} + (\cI_{n-1}\btau)_{\vec x_{\sfp}}.
\end{align}
Remarkably, this modification \eqref{eq:AddingPressureForceFlux} \emph{does not} modify the vector field \(\vec v\in\fX_{\div}(W)\) in \eqref{eq:COFLIP-ODEc}, using the fact \eqref{eq:IDaggerIsLeftInverse} that \(\hat\cI^+\) is a left-inverse of \(\cI\).  
Spelling it out, observe that
the modification \eqref{eq:AddingPressureForceFlux} modifies \(\hat\cI^+(\vec x_\sfp,\vec u_\sfp)_{\sfp\in\cP}\) by \(\btau = -\star_{n-1}^{-1}\bd_{n-1}^\intercal\bp\in\fB_{\div}^\bot\) (\cf\@ \eqref{eq:PressureSolveB}) which is exactly removed by \(\bP_{\fB_{\div}}\).  
As a consequence, the addition of \eqref{eq:AddingPressureForceFlux} does not affect energy conservation.

The caveat of
\eqref{eq:AddingPressureForceFlux}
is that \(\cI_{n-1}\btau\) as a continuous vector field is not an exact gradient of a continuous function (or a curl-free field).  Note that the mimetic guarantee for \(\cI_{n-1}\) is that it commutes with the exterior derivative for flux form, which is the divergence operator.  In particular \eqref{eq:AddingPressureForceFlux} may modify the state of \(J_{\adv}(\vec x,\vec u)\in \fX_{\div}^*(W) = \Omega^1(W)/d\Omega_0(W)\) in such a way that changes the coadjoint orbit.

In practice, on a simulation grid with a decent resolution, the high-order interpolation \(\cI_{n-1}\btau\) of a discrete gradient field \(\btau = -\star_{n-1}^{-1}\bd_{n-1}^\intercal\bp\) is still close to a continuous gradient field.  We observe that its effects on the coadjoint orbit preservation and Casimir conservation are minimal (see Figures \ref{fig:leapfrog_plots} and \ref{fig:twistedtorus}).

\begin{figure}
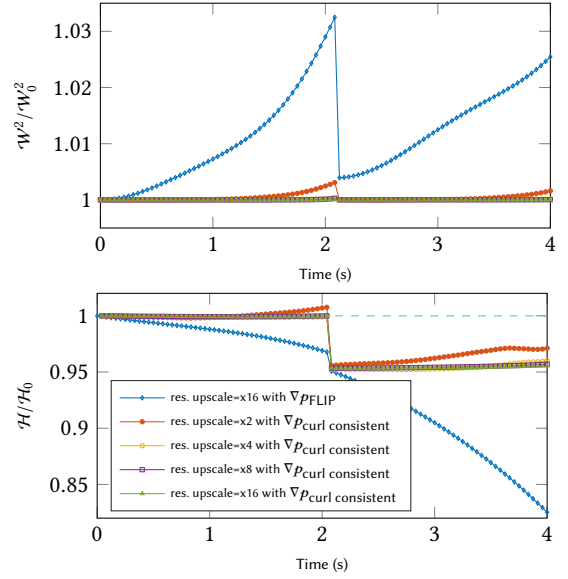

    \centering
    \input{figures/Results/2D/casimir_under_refinement/with_and_without_exact_grad_p_comparison}
    \input{figures/Results/3D/helicity_under_refinement/helicity_under_refinement_exact_grad_p}
    \caption{
        Enstrophy for two dimensions, and helicity for three dimensions under refinement.
        Note that with increasing the resolution at which the Casimir measurement takes place, we better see the preservation of them.
        This is only expected while particle information is only purturbed with the curl consistent pressure force.
        With the traditional FLIP-based pressure force, we see the Casimirs changing; though, this error is very small.
        Finally, we emphasize that the jump in the middle is due to global resetting of the particle data.
    }
    \label{fig:casimir_refinement}
\end{figure}

\subsubsection{Curl-Consistent Interpolation of Pressure Force}
\label{sec:CurlConsistentInterpolationOfPressureForce}

Another option for adding pressure force \(\btau\) back to the particles is to treat the force as a discretely exact 1-form.  In that case, its interpolation by \(\cI_1\), which commutes with curl, will be exact (and curl free). Therefore, the addition of pressure does not affect the coadjoint orbit.
To do so for a given \(\btau\in C^{n-1}\), first interpolate it to the particles \(\vec\tau_{\sfp} = (\cI_{n-1}\btau)_{\vec x_{\sfp}}\) and then inverse-interpolate back to a discrete 1-form \(\tilde\btau \coloneqq\hat\cI_1^+(\vec x_{\sfp},\vec \tau_{\sfp})_{\sfp\in\cP}\in C^1\).
This \(\tilde\btau\in C^1\) is the discrete 1-form that best shares the same interpolated vector field with the original \(\btau\in C^{n-1}\).  Note that \(\tilde\btau\) is not necessarily exact (\(\bd_0\) of a scalar field).
So, perform a Poisson reconstruction \(\tilde\bp = (\bd_0^{\intercal}\star_1\bd_0)^{-1}\bd_0^{\intercal}\star_1\tilde\btau\) to find the closest exact 1-form \(\bd_0\tilde\bp\) to \(\tilde\btau\).
In three dimensions, the reconstruction problem additionally has a Dirichlet boundary condition of $\tilde{\bp}=0$, which ensures that helicity remains unchanged.
Finally, add the pressure force 
\begin{align}
\label{eq:AddingPressureForce1Form}
    \vec u_{\sfp}\gets\vec u_{\sfp} + (\cI_1\bd_0\tilde\bp)_{\vec x_{\sfp}}.
\end{align}
The modification \eqref{eq:AddingPressureForce1Form} modifies \((\vec x,\vec u)\) by an exact form, and therefore, it does not modify the coadjoint in the limit when the particles represent a continuum.
Note, however, that the curl-consistent interpolated pressure force is less computationally efficient as it involves two additional global solves, for the pseudo-inverse and Poisson solves. As such, we use the div-consistent pressure force for all of our simulations, despite its minimal effects on Casimir conservation.

\subsection{Stability Discussions}
    The CO-FLIP algorithm is conditionally stable.  It substantially improves on the well-known instability problems of the traditional FLIP method. Here, we discuss the various factors that affects stability.

    \subsubsection{Stability Given by Pressure Force}
    \label{sec:StabilityGivenByPressureForce}
    A known problem in FLIP methods is the frequency gap between the high resolution particles and low resolution grids.  The particle data can develop high frequency patterns, whereas an interpolated pressure force feeding back to the particles can only regulate them up to the resolution of the grid.
    While this problem with FLIP is generally present, we observe that utilizing mimetic interpolation can significantly improve the regularity of particle data.
    We observe that the divergence-consistent mimetic interpolation of pressure force \eqref{eq:AddingPressureForceFlux} of \secref{sec:DivConsistentInterpolationOfPressureForce} is better at maintaining the stability of FLIP-based methods.
    In our experiments, as seen in \figref{fig:ablation_gradp_adaptivity}, we show that the addition of the pressure force allows for a far more infrequent global resetting frequency of the particle data.
    Empirically, we found the reset frequency can be set to once every 200 frames.
    In contrast, the state-of-the-art frequency in \cite{Deng:2023:FSN} is every 20 frames.
    Additionally, as discussed in \secref{sec:AdaptiveReset}, the adaptive resetting of the particles further lowers the error at very infrequent global reset frequencies.
    
    \subsubsection{Particle Distribution Over Time}
    The stability of the CO-FLIP algorithm is further improved by using the mimetic interpolation scheme to get point-wise divergence-free velocity queries which we can integrate to volume preserving flow maps for the particles. Concretely, a uniform independent and identical distribution (i.i.d\acdot) of particles would remain uniformly i.i.d\acdot\ through the dynamics of the incompressible fluid.
    However, to reduce the variance on the Monte Carlo integration needed for $\hat\cI^+$ (see \eqref{eq:IDaggerDiscrete}), one could use a stratified sampling of particles per grid cell.  Such a stratified sampling would not remain stratified throughout the advection of particles, thus becoming an ill-suited distribution for the Monte Carlo integration of $\hat\cI^+$. This would increase the error gap between the discrete operator $\hat\cI^+$ and the continuous operator $\cI^+$, which would worsen the errors in the simulation and create room for instabilities to grow.  As such, on global resets of the particle information, we additionally redistribute the particles in a per-grid-cell-stratified fashion.

    \subsubsection{Particle-per-Cell Count}
    Other than the particle distribution, the number of particles per cell (PPC) also affects the variance of the Monte Carlo integration.  In principal, one needs a dense enough sample for the least-square system to have a unique solution. However, we empirically choose higher PPC counts to ensure a lower variance in the integration error of the discrete $\hat\cI^+$ operation.  We select minimum and maximum PPC counts, and adaptively distribute particles in the domain based on the magnitude of local vorticity compared to the global maximum; this adaptive nature also helps with the performance (see \secref{sec:pseudoinversep2gsolve}).  For 3D experiments, the PPC count is in the range $[8-27]$ ($[4-16]$ for 2D).  In our experiments, this provides stable results throughout the simulated frames; however, below the mentioned range, we see unstable noisy behavior in the velocity field.

    \subsubsection{Choosing the CFL Number}
    The CFL number \cite{Courant:1928:CFL} we choose for the time integration dictates how good the fixed-point iteration algorithm's initial guess is, which relates to how fast the fixed-point iteration converges.  At times, the fixed-point iteration fails to converge due to this.  In these cases, we observe an unstable behavior from the solver. Additionally, so long as the fixed-point iteration converges in Algorithm 2, energy can be preserved by using \eqref{eq:EnergyCorrection}; otherwise, one would observe an energy drift. Note, however, than one may not want to increase CFL number too high as it increase the error in the time integration; even if the fixed-point iteration converges.  As such, an optimal CFL number is preferred where there is balance between performance and accuracy.  We empirically find CFL numbers below one to be stable, but we choose roughly 0.5 to be optimal for accuracy.

\begin{figure}
    \centering
%
%
\definecolor{mycolor1}{rgb}{0.00000,0.44700,0.74100}%
\definecolor{mycolor2}{rgb}{0.85000,0.32500,0.09800}%
\definecolor{mycolor3}{rgb}{0.92900,0.69400,0.12500}%
\begin{tikzpicture}

\begin{axis}[%
width=0.7\linewidth,
height=0.5\linewidth,
at={(0,0)},
scale only axis,
xmode=log,
xmin=1,
xmax=256,
xtick={1,2,4,8,16,32,64,128,256},
xticklabels={{1},{$\nicefrac{1}{2}$},{$\nicefrac{1}{4}$},{$\nicefrac{1}{8}$},{$\nicefrac{1}{16}$},{$\nicefrac{1}{32}$},{$\nicefrac{1}{64}$},{$\nicefrac{1}{128}$},{$\nicefrac{1}{256}$}},
xlabel style={font=\color{white!15!black}},
xlabel={\sffamily \scriptsize Reset Frequency (1 / \# frames)},
ymode=log,
ymin=1e-05,
ymax=0.01,
ytick={ 1e-05, 0.0001,  0.001,   0.01},
yminorticks=true,
ylabel style={font=\color{white!15!black}},
ylabel={\sffamily \scriptsize $L^2 \text{ Velocity Error}$},
axis background/.style={fill=white},
legend pos = south east,
legend style={legend cell align=left, align=left, legend plot pos=left, draw=white!15!black}
]
\addplot [color=mycolor1, mark=*]
  table[row sep=crcr]{%
1	0.000277578\\
2	0.000272687\\
4	0.000119266\\
8	0.00010086\\
16	0.000100664\\
32	0.000147595\\
64	0.000163341\\
128	0.000221951\\
256	0.000239441\\
};
\addlegendentry{\sffamily \tiny With $\nabla p$  \& adaptive resetting ($\alpha=1$)}

\addplot [color=mycolor2, mark=o]
  table[row sep=crcr]{%
1	0.000277578\\
2	0.000272687\\
4	0.000119266\\
8	0.00010086\\
16	0.000129486\\
32	0.000169696\\
64	0.00035325\\
128	0.000542823\\
256	0.00211288\\
};
\addlegendentry{\sffamily \tiny With $\nabla p$ \&  not adaptive resetting}

\addplot [color=mycolor3, mark=square*]
  table[row sep=crcr]{%
1	0.000277659\\
2	0.000277884\\
4	6.79e-05\\
8	6.23e-05\\
16	0.000593042\\
32	1.53574e+91\\
64	1.53574e+91\\
128	1.53574e+91\\
256	1.53574e+91\\
};
\addlegendentry{\sffamily \tiny Without $\nabla p$ \& not adaptive resetting}

\end{axis}
\end{tikzpicture}%
    \caption{
        Comparison of adding or not adding pressure force back to the particles, and whether adaptive resetting is used.
        Once pressure force is added to the particles, we can avoid blow-ups in lower global reset frequencies.
        Additionally, the error in simulation can be further minimized in highly infrequent global resets by using adaptive resets on the particles.
    }
    \label{fig:ablation_gradp_adaptivity}
\end{figure}
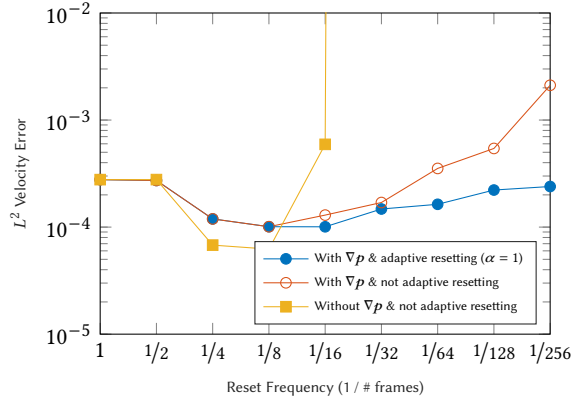

\subsection{Initialization}
To initialize the simulation the user can provide a velocity field or a vorticity field.  For a given velocity field \(\vec u_{\rm given}\in\fX_{\div}(W)\), simply assign it to particles $\vec u_{\sfp}, \sfp\in\cP$.  One may also filter this velocity through \(\cI_{n-1}\hat\cI_{n-1}^+\) so its level of detail is conformed to the resolution of the grid.  

For a given vorticity field (scalar field in 2D and vector field in 3D) \(w\), assign the vorticity on each particle, transfer it to the grid by the pseudoinverse of the \(C^{n-2}\) mimetic interpolation
\begin{align}
    \bw = \star_{n-2}\hat\cI_{n-2}^+(\vec x_\sfp,w_\sfp)_{\sfp\in\cP} \in C^{(n-2)*},
\end{align}
and solve the streamform \(\bpsi\in C^{n-2}\), \eqref{eq:StreamformPoissonWithoutGauge} for 2D and \eqref{eq:3DStreamformSolveA} for 3D, with the right-hand side set as \(\bw\in C^{(n-2)*}\).  Finally set \(\bff = \bd_{n-2}\bpsi\) followed by interpolating \(\bff\) to particle velocity using \(\cI_{n-1}\).

Note that one should not simply assign values at the grid faces since these faces may not represent the expected physical locations in the domain (\cf\@ \remref{rmk:GridOfCoefficients} and \remref{rmk:GrevilleGrid}).

\subsection{Particle Advection}
\label{sec:ParticleAdvection}
Here we detail how we execute the particle advection action \eqref{eq:AdvDefinition}.
In \eqref{eq:AdvDefinition} we assume a given vector field \(\vec v\in\fX_{\div}(W)\), and we transport the data \((\vec x,\vec u)\) by the flow generated by \(\vec v\).  While solving the ODE system \eqref{eq:AdvDefinition} is feasible, we implement the computational setup differently.  
Instead of updating \((\vec x_{\sfp},\vec u_{\sfp})\) at each time \(t\), we store a rest-frame \((X_{\sfp},\vec U_{\sfp})\) and update a flow map and its inverse Jacobian per particle.  To access \((\vec x_{\sfp},\vec u_{\sfp})\) at the current frame, we use the flow map information to transform the rest-frame \((\vec X_{\sfp},\vec U_{\sfp})\). 
The main reason for doing so is that we can access the Jacobian at each particle, which is a cumulative record of the quality of the map.  Based on the quality of the particle-wise Jacobian, we can reset the rest frame for a particle.  This procedure is detailed in \secref{sec:AdaptiveReset}.

\subsubsection{Data Structure}
Each particle \(\sfp\in\cP\) is assigned with a rest-frame position and impulse 
\((\vec X_{\sfp}\in W,\vec U_{\sfp}\in T_{\vec X}^*W=\RR^n)\), and the timestamp \(T_{\sfp}\in\RR\) of the rest-frame.
Each particle is also equipped with a position vector \(\vec x_{\sfp}\in W\) that represents the particle's current location, and an \(n\times n\) matrix \(\Psi_{\sfp}\).
The pair of positions \(\vec X_{\sfp}\) and \(\vec x_{\sfp}\) represents the flow map at this particle from time \(T_{\sfp}\) to the current time \(t\).  The matrix \(\Psi_{\sfp}\) is the inverse Jacobian of the flow map.

To access the current impulse \(\vec u_\sfp\) from this data structure, evaluate \(\vec u_\sfp = \Psi_\sfp^\intercal \vec U_\sfp\).

\subsubsection{Reset}
\label{sec:Reset}
When a reset is called for a particle \(\sfp\), such as at the initial frame, we set \(\vec X_\sfp\gets \vec x_\sfp\), \(\Psi_{\sfp} \gets \id\), \(T_\sfp\gets t\), and \(\vec U_\sfp \gets \vec u_\sfp^{\rm given}\), where \(\vec u_\sfp^{\rm given}\) is the given impulse covector to be assigned to this particle for the reset.

\subsubsection{Advection}
The advection step \eqref{eq:AdvDefinition} for \((\vec x_\sfp,\vec u_\sfp)\) is equivalent to marching \((\vec x_\sfp,\Psi_\sfp)\) (while fixing \(\vec X_\sfp,\vec U_\sfp, T_\sfp\)) by the following ODE
\begin{align}
\label{eq:AdvWithPsi}
    \begin{cases}
        {d\over dt}\vec x_\sfp(t) = \vec v|_{\vec x_\sfp(t)}\\
        {d\over dt}\Psi_\sfp(t) = -\Psi_\sfp(t)(\nabla\vec v|_{\vec x_\sfp(t)})
    \end{cases}
\end{align}
which we solve by RK4.  It is easy to verify that \(\vec u_\sfp(t) = \Psi_\sfp^\intercal(t)\vec U_\sfp\) satisfies \eqref{eq:AdvDefinition}.

\subsubsection{Trace-Free and Unimodularity}
Since the velocity field \(\vec v\) from our mimetic interpolation is pointwise divergence free, the Jacobian \(\nabla\vec v\) is pointwise trace-free \(\nabla\vec v\in\slla(n)\).%
\footnote{\(\slla(n)\) is the space of trace-free \(n\times n\) matrices.  It is the Lie algebra of \(\SL(n)\).}
Therefore, the solution \(\Psi_\sfp(t)\) to \({d\over dt}\Psi_\sfp(t) = -\Psi_\sfp(t)(\nabla\vec v|_{\vec x_\sfp(t)})\) with \(\Psi_\sfp(t=T_{\sfp}) = \id\) is a unimodular (\(\det=1\), or volume-preserving) matrix \(\Psi_\sfp(t)\in\SL(n)\).%
\footnote{The special linear group \(\SL(n)\) is the Lie group of \(n\times n\) matrices with \(\det=1\).}
The unimodularity of \(\Psi_\sfp(t)\) remains true over different frames even when the flow velocity \(\vec v\) has changed.  This is because the algebraic property that \(\det=1\) is closed under matrix compositions (\(\SL(n)\) is a group).

\paragraph{Projection to \(\SL(n)\)}
\begin{wrapfigure}{r}{70pt}
\centering
\setlength{\unitlength}{1pt}
\begin{picture}(70,70)
    \put(-10,0){\includegraphics[width=80pt]{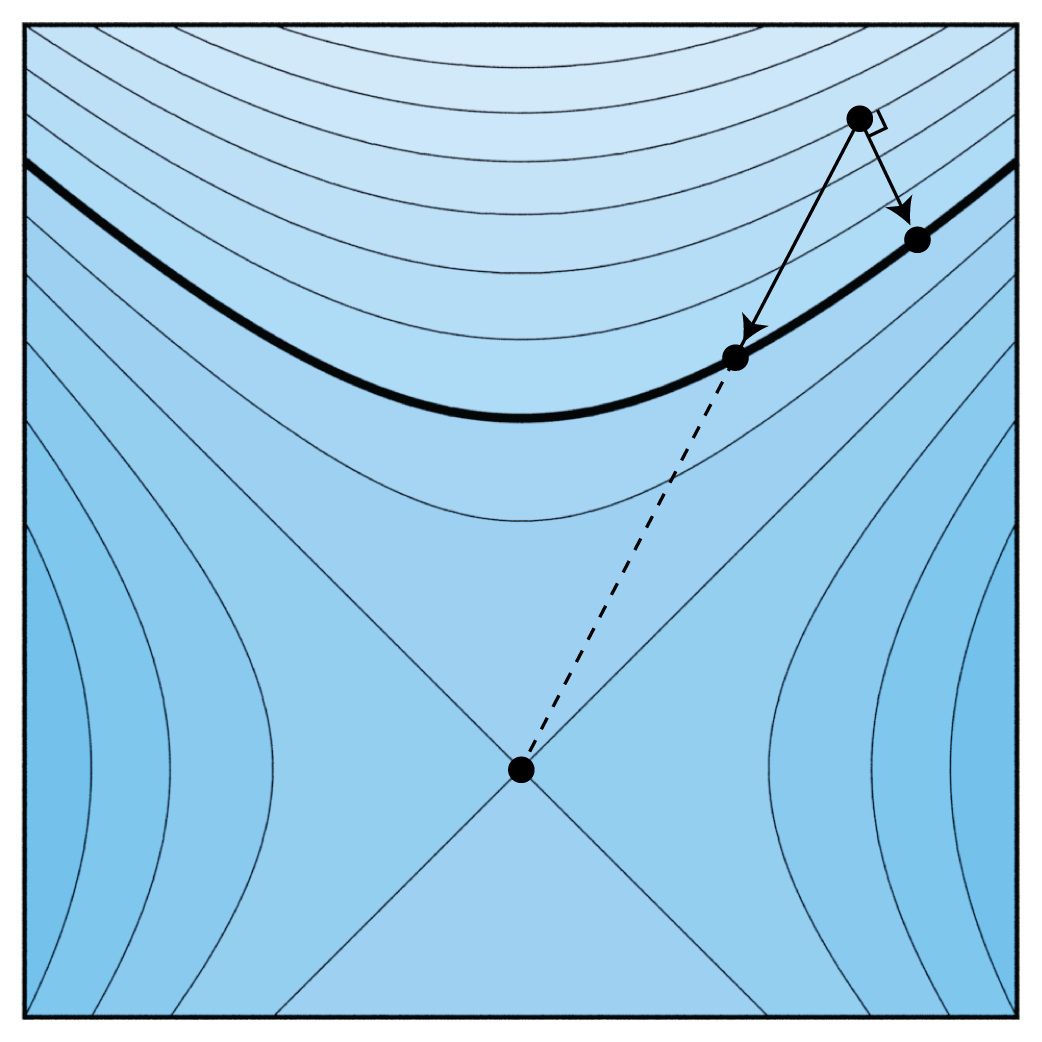}}
    \put(54,73){\tiny\(\Psi\)}
    \put(60,55){\tiny\(\Psi^{\rm new}\)}
    \put(30,15){\tiny\(\bzero\)}
    \put(25,55){\tiny\sffamily rescaled \(\Psi\)}
    \put(-8,70){\tiny\rotatebox{-40}{\(\det=1\)}}
    \put(-5,5){\tiny\(\RR^{n\times n}\)}
\end{picture}
\end{wrapfigure}
Although RK4 gives us a highly accurate integration of \eqref{eq:AdvWithPsi}, over a long time \(\det(\Psi_\sfp)\) may drift away from 1.  To enforce \(\det(\Psi_\sfp)=1\) we employ the following projection to snap \(\Psi_{\sfp}\) back to \(\SL(n)\)  after each RK4 step \cite[Example~4.5]{Hairer:2006:GNI}: Update \(\Psi_\sfp\) by \(\Psi_{\sfp}^{\rm new}\), where \(\Psi_{\sfp}^{\rm new}\) solves
\begin{subequations}
\label{eq:DeterminantProjection}
    \begin{numcases}{}
    \Psi_\sfp^{\rm new}= \Psi_\sfp + \lambda\Psi_\sfp^{-\intercal},\quad\lambda\in\RR\\
    \det(\Psi_{\sfp}^{\rm new})=1.
\end{numcases}
\end{subequations}
One can see that \eqref{eq:DeterminantProjection} is indeed an orthogonal projection to the level-set of determinant in \(\RR^{n\times n}\) by noting that
the cofactor matrix \(\cof(\Psi_\sfp) = \det(\Psi_\sfp)\Psi_\sfp^{-\intercal}\) is the gradient (with respect to the Frobenius inner product) of the determinant function at \(\Psi_\sfp\in\RR^{n\times n}\).
Compared to other alternatives, such as normalizing the matrix by its determinant, our projection is less aggressive in modifying the matrix (see inset).
Solving \eqref{eq:DeterminantProjection} amounts to finding \(\lambda\in\RR\) so that \(\det(\Psi_{\sfp} + \lambda\Psi_{\sfp}^{-\intercal}) = 1\), which is a 1D root finding problem.  We use the following fixed point iteration (quasi-Newton iteration) with initial guess \(\lambda_0 = 0\):
\begin{align}
    \lambda_{\sfk+1} = \lambda_{\sfk} - {\det(\Psi_{\sfp}+\lambda_\sfk\Psi_\sfp^{-\intercal})-1\over \tr((\Psi_\sfp^\intercal\Psi_\sfp)^{-1})\det(\Psi_{\sfp}+\lambda_\sfk\Psi_\sfp^{-\intercal})}.
\end{align}
The iteration involves \(2\times 2\) or \(3\times 3\) matrix inversion, which is computationally cheap.
In practice, we observe that 1--2 iterations are sufficient.

\subsection{Particle-wise Adaptive Reset}
\label{sec:AdaptiveReset}

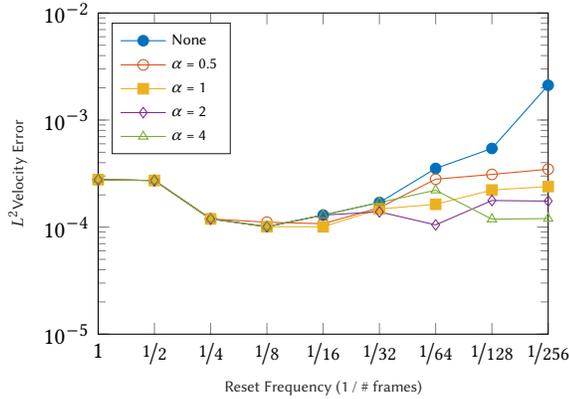
\begin{figure}
    \centering
%
%
\definecolor{mycolor1}{rgb}{0.00000,0.44700,0.74100}%
\definecolor{mycolor2}{rgb}{0.85000,0.32500,0.09800}%
\definecolor{mycolor3}{rgb}{0.92900,0.69400,0.12500}%
\definecolor{mycolor4}{rgb}{0.49400,0.18400,0.55600}%
\definecolor{mycolor5}{rgb}{0.46600,0.67400,0.18800}%
\begin{tikzpicture}

\begin{axis}[%
width=0.7\linewidth,
height=0.5\linewidth,
at={(0,0)},
scale only axis,
xmode=log,
xmin=1,
xmax=256,
xtick={1,2,4,8,16,32,64,128,256},
xticklabels={{1},{$\nicefrac{1}{2}$},{$\nicefrac{1}{4}$},{$\nicefrac{1}{8}$},{$\nicefrac{1}{16}$},{$\nicefrac{1}{32}$},{$\nicefrac{1}{64}$},{$\nicefrac{1}{128}$},{$\nicefrac{1}{256}$}},
xlabel style={font=\color{white!15!black}},
xlabel={\sffamily \scriptsize Reset Frequency (1 / \# frames)},
ymode=log,
ymin=1e-05,
ymax=0.01,
ytick={ 1e-05, 0.0001,  0.001,   0.01},
yminorticks=true,
ylabel style={font=\color{white!15!black}},
ylabel={\sffamily \scriptsize $L^2 \text{Velocity Error}$},
axis background/.style={fill=white},
legend pos = north west,
legend style={legend cell align=left, align=left, legend plot pos=left, draw=white!15!black}
]
\addplot [color=mycolor1, mark=*]
  table[row sep=crcr]{%
1	0.000277578\\
2	0.000272687\\
4	0.000119266\\
8	0.00010086\\
16	0.000129486\\
32	0.000169696\\
64	0.00035325\\
128	0.000542823\\
256	0.00211288\\
};
\addlegendentry{\sffamily \scriptsize None}

\addplot [color=mycolor2, mark=o]
  table[row sep=crcr]{%
1	0.000277578\\
2	0.000272687\\
4	0.000119266\\
8	0.000111134\\
16	0.000107475\\
32	0.000151707\\
64	0.00028026\\
128	0.000310621\\
256	0.000347334\\
};
\addlegendentry{\sffamily \scriptsize $\alpha\text{ = 0.5}$}

\addplot [color=mycolor3, mark=square*]
  table[row sep=crcr]{%
1	0.000277578\\
2	0.000272687\\
4	0.000119266\\
8	0.00010086\\
16	0.000100664\\
32	0.000147595\\
64	0.000163341\\
128	0.000221951\\
256	0.000239441\\
};
\addlegendentry{\sffamily \scriptsize $\alpha\text{ = 1}$}

\addplot [color=mycolor4, mark=diamond]
  table[row sep=crcr]{%
1	0.000277578\\
2	0.000272687\\
4	0.000119266\\
8	0.00010086\\
16	0.000129486\\
32	0.000138853\\
64	0.000104863\\
128	0.000177228\\
256	0.00017492\\
};
\addlegendentry{\sffamily \scriptsize $\alpha\text{ = 2}$}

\addplot [color=mycolor5, mark=triangle]
  table[row sep=crcr]{%
1	0.000277578\\
2	0.000272687\\
4	0.000119266\\
8	0.00010086\\
16	0.000129486\\
32	0.000169696\\
64	0.000221227\\
128	0.000118627\\
256	0.000120221\\
};
\addlegendentry{\sffamily \scriptsize $\alpha\text{ = 4}$}

\end{axis}
\end{tikzpicture}%
    \caption{ Comparison of different choices of adaptive resetting cutoff ($\alpha$) for particles.
    Note that there is little difference between the cutoff values, but all results show a lower error than those without an adaptive reset.}
    \label{fig:ablation-adaptivity-cutoffs}
\end{figure}

The impulse/covector variable \(\vec u_\sfp\) is stretched by the deformation of the flow \eqref{eq:AdvDefinition}.  However, in a fluid flow, the amount of shearing deformation in a flow map is generally unbounded.  This can lead to numerical instability.
In the literature, fluid methods that are based on using a long-time flow map to transport momentum or impulse \cite{Qu:2019:ECF, Nabizadeh:2022:CF, Deng:2023:FSN} reset the flow map globally every few frames to avoid instability. 
We call the number of frames before a global reset the \textit{reset period}.
\cite{Nabizadeh:2022:CF} has a reset period of 5.  \cite{Deng:2023:FSN} resets every 20 frames. 
Classical FLIP is also blended with APIC or PIC, with the typical blending factor empirically chosen around roughly 0.99.
In extreme cases, semi-Lagrangian, PIC, or APIC methods correspond to resetting the flow map at every frame.
Although resetting the flow map stabilizes the fluid simulator, it removes the Lagrangian nature of the method of characteristics.

We point out that the amount of stretching throughout the domain is not uniform.  
Some parts of the domain may deform heavily and require flow-map resets, whereas the flow map in the other parts is still regular and does not require resets.
Indeed, the flow map does not need to be reset globally: in \secref{sec:Reset}, the particle-wise reset can be called \emph{asynchronously}.
This motivates an adaptive strategy to the reset.

Per particle \(\sfp\in\cP\), the reset function \secref{sec:Reset} is called if the deformation gradient \(F_\sfp = \Psi_\sfp^{-1}\) of the particle has its largest singular value exceeding a threshold.
At the reset, we set the rest-frame particle impulse covector \(\vec U_{\sfp}\) (\cf\@ \secref{sec:Reset}) as the current interpolated grid velocity field at the particle's location.

Concretely, for the thresholding, we inspect the following \emph{cumulative Finite Time Lyapunov Exponent}
\begin{align}
    S_\sfp\coloneqq\ln\left(\sigma_{\max}(F_\sfp)\right),
\end{align}
where \(\sigma_{\max}\) evaluates the largest singular value.  It measures the maximal stretching of the deformation since the previous reset call.  
Note that \(S_{\sfp}\) is closely related to the Finite Time Lyapunov Exponent (FTLE), which is defined by \({\rm FTLE}_\sfp\coloneqq S_{\sfp}/T_{\sfp}\) where \(T_\sfp\) is the duration of the deformation as defined in \secref{sec:ParticleAdvection}.
Empirically, as seen in \figref{fig:ablation-adaptivity-cutoffs}, we found that limiting \(S_\sfp < 1\) (\ie\@ reset the particle's flow map if \(S_\sfp\geq1\)) leads to stable results without sacrificing accuracy.

We also include an optional global reset like previous methods.
As discussed in \secref{sec:StabilityGivenByPressureForce}, adding the pressure force back to the particles alone can improve stability so that
the reset period can be as much as 256 frames.
Our adaptive flow-map resetting further enhances the stability of CO-FLIP.
See \figref{fig:ablation_gradp_adaptivity} for comparison.

\subsection{Casimir Measurements}
\label{sec:CasimirMeasurement}

The coadjoint orbit preservation of the CO-FLIP method states that the image \(J_{\adv}(\vec x,\vec u)\in \fX_{\div}^*(W)\) of the fluid state in the infinite dimensional space \(\fX_{\div}^*(W)\) of continuous 1-forms stays on a coadjoint orbit, as previously explained in \secref{sec:DualSpaceOfDivFree}.
A quantitative measurement of this quality is measuring the Casimirs (\cf\@ \ref{def:Casimir}).  A Casimir is a function defined over \(\fX_{\div}^*(W)\) that is constant on each coadjoint orbit.  Therefore, a necessary condition for the coadjoint orbit conservation is the preservation of each Casimir function.  
In 2D a basis for the Casimir functions are the \(p\)-th moment of the vorticity function \cite{Khesin:1989:IEE}
\begin{align}
    \cW^p([\eta])\coloneqq \int_W w^p\, d\mu,\quad w = *(d\eta),\quad[\eta]\in \fX_{\div}^*(W).
\end{align}
In 3D the only regular Casimir is the helicity \citep{Khesin:2022:HUC}
\begin{align}
    \cH\coloneqq \int_W \eta\wedge \omega,\quad \omega = d\eta,\quad [\eta]\in \fX_{\div}^*(W).
\end{align}

However, the Casimirs are defined on the continuous space \(\fX_{\div}^*(W)\), whose evaluation can only be approximated.
In our validations, we measure the Casimirs using a sequence of grids with progressively higher resolutions to obtain a sequence of progressively more accurate approximations of the Casimirs.  
The conservation of the Casimirs is observed in the continuous limit of the approximated Casimirs (see \figref{fig:casimir_refinement}).

Note that the procedures of adding divergence-consistent pressure force back to the particles (\secref{sec:DivConsistentInterpolationOfPressureForce}) and the adaptive resetting (\secref{sec:AdaptiveReset}) can affect the exactness of Casimir preservation.
For validating the exact Casimir preservation in \figref{fig:casimir_refinement} we employed the curl-consistent pressure feedback (\secref{sec:CurlConsistentInterpolationOfPressureForce}) and disabled adaptive resets.  To prevent instability we included a global reset every \(\unit[2]{s}\times \unit[24]{fps} = \unit[48]{frames}\).

In a more realistic simulation setup we include the more stable divergence-consistent pressure force (\secref{sec:DivConsistentInterpolationOfPressureForce}) and the adaptive resetting (\secref{sec:AdaptiveReset}).  As shown in \figref{fig:leapfrog_plots} the 2D Casimirs drift only mildly, compared to other methods, over the course of \(\unit[500]{s}\times \unit[48]{fps} = \unit[24,000]{frames}\).  In 3D, the helicity is conserved better than competing methods over the course of \(\unit[10]{s}\times \unit[48]{fps} = \unit[480]{frames}\), as shown in \figref{fig:twistedtorus}.

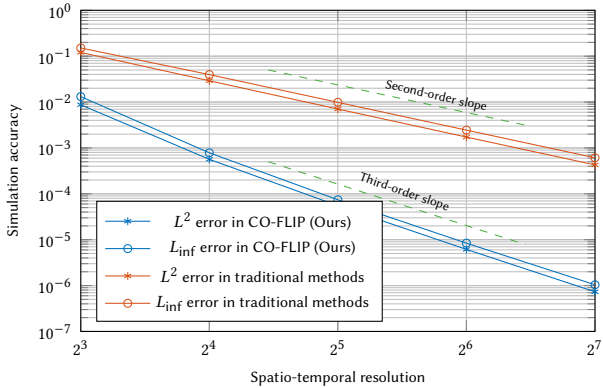
\begin{figure}
    \centering
%
%
\definecolor{mycolor1}{rgb}{0.00000,0.44700,0.74100}%
\definecolor{mycolor2}{rgb}{0.85000,0.32500,0.09800}%
\definecolor{mycolor3}{rgb}{0.2900,0.69400,0.2500}%
\begin{tikzpicture}

\begin{loglogaxis}[%
width=0.8\linewidth,
height=0.5\linewidth,
at={(0,0)},
scale only axis,
xmin=8,
xmax=128,
log basis x=2,
log basis y=10,
ymin=1e-07,
ymax=1,
xminorticks=true,
yminorticks=true,
ymajorticks=true,
xminorgrids=true,
xmajorgrids=true,
yminorgrids=true,
ymajorgrids=true,
max space between ticks=1,
xlabel={\scriptsize \sffamily Spatio-temporal resolution},
ylabel={\scriptsize \sffamily Simulation accuracy},
xticklabel style={font=\scriptsize},
yticklabel style={font=\scriptsize},
legend pos = south west,
]
\addplot [color=mycolor1, mark=asterisk, mark options={solid, mycolor1}, mark size=1.5pt]
  table[row sep=crcr]{%
8	0.00870778\\
16	0.000564607\\
32	5.50643e-05\\
64	6.14164e-06\\
128	7.31478e-07\\
};
\addlegendentry{\scriptsize \sffamily $L^2$ error in CO-FLIP (Ours)}
\addplot [color=mycolor1, mark=o, mark options={solid, mycolor1}, mark size=1.5pt]
  table[row sep=crcr]{%
8	0.0131315\\
16	0.00078327\\
32	7.41315e-05\\
64	8.40112e-06\\
128	1.04205e-06\\
};
\addlegendentry{\scriptsize \sffamily $L_\text{inf}$ error in CO-FLIP (Ours)}
\addplot [color=mycolor2, mark=asterisk, mark options={solid, mycolor2}, mark size=1.5pt]
  table[row sep=crcr]{%
8	0.120823\\
16	0.0291132\\
32	0.00704491\\
64	0.00172607\\
128	0.000427043\\
};
\addlegendentry{\scriptsize \sffamily $L^2$ error in traditional methods}
\addplot [color=mycolor2, mark=o, mark options={solid, mycolor2}, mark size=1.5pt]
  table[row sep=crcr]{%
8	0.151\\
16	0.0396795\\
32	0.00984635\\
64	0.00244569\\
128	0.000610298\\
};
\addlegendentry{\scriptsize \sffamily $L_\text{inf}$ error in traditional methods}
\addplot [color=mycolor3, dashed]
  table[row sep=crcr]{%
22	0.0005\\
44	6.25e-05\\
88	7.8125e-06\\
};
\addplot [color=mycolor3, dashed]
  table[row sep=crcr]{%
22	0.05\\
44	0.0125\\
88	0.003125\\
};
\end{loglogaxis}
\end{tikzpicture}%
    \\
    \begin{picture}(0,0)(0,0)
    \put(20,92){\sffamily \tiny \rotatebox{-19}{Third-order slope}}
    \put(30,127){\sffamily \tiny \rotatebox{-13}{Second-order slope}}
    \end{picture}
    \caption{
        The 2D convergence plot for the Taylor-Green experiment ran for one second at fixed $CFL=0.4$.
        Note that our method achieves third-order accuracy using quadratic B-splines compared to second order accuracy of previous methods.
    }
    \label{fig:convergenceplot_2d}
\end{figure}

\begin{figure*}
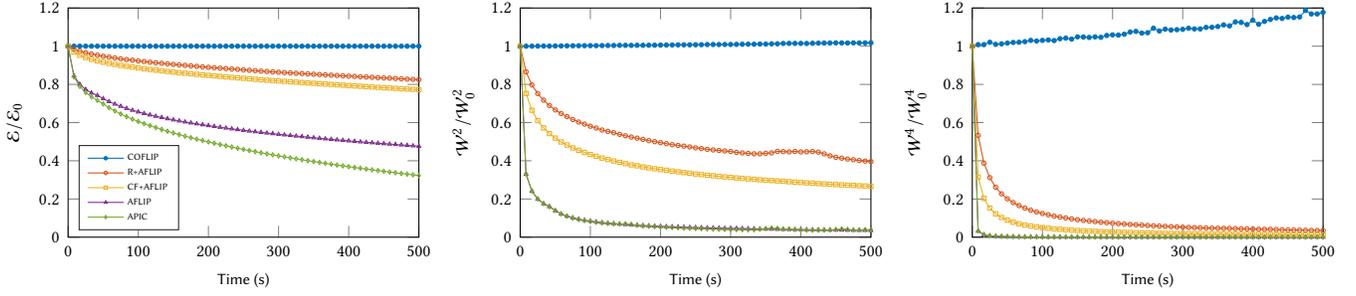

    \input{figures/Results/2D/leapfrog/leapfrog_energy2d}
    \input{figures/Results/2D/leapfrog/leapfrog_casimir2d}
    \input{figures/Results/2D/leapfrog/leapfrog_casimir2d_vort4}
    \caption{
        Energy, enstrophy (second moment of vorticity), and fourth moment of vorticity plots over time for the 2D leapfrogging vortices experiment.
        Note that our method exactly conserves energy over time, minimal relative error in the measured values for Casimirs, compared to other methods. 
        For clarity purposes, the plots show data samples at x400 times lower frequency than the simulation frame rate.
    }
    \label{fig:leapfrog_plots}
\end{figure*}

\begin{figure}
    \centering
    \includegraphics[width=0.3\columnwidth]{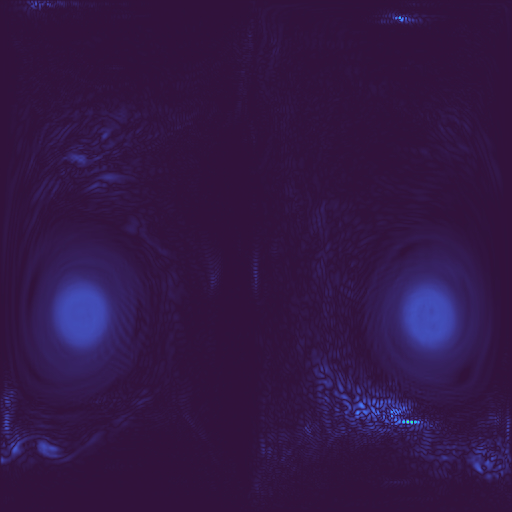}
    \includegraphics[width=0.3\columnwidth]{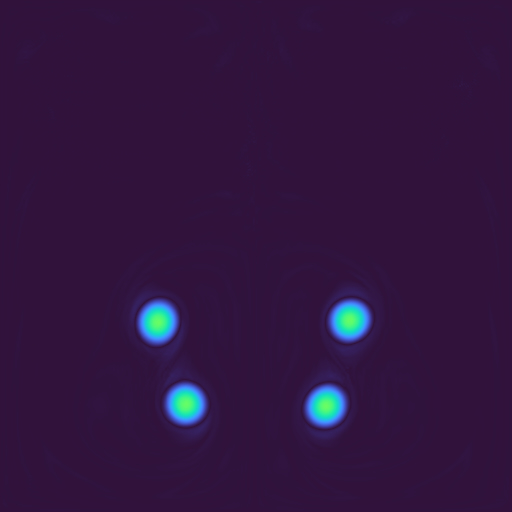}
    \includegraphics[width=0.3\columnwidth]{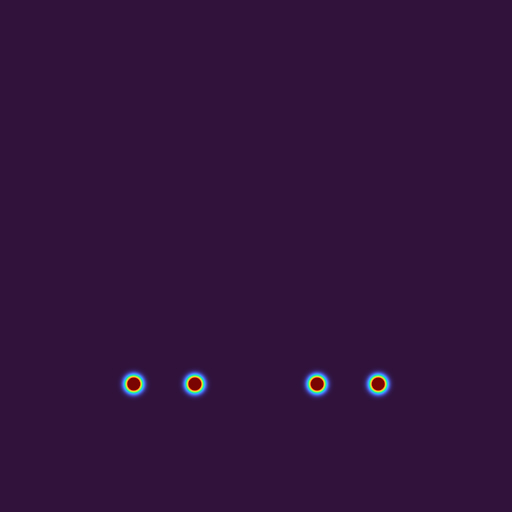}
    \includegraphics[width=0.3\columnwidth]{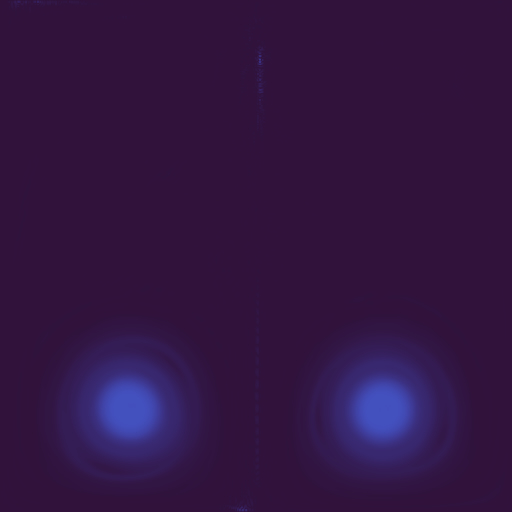}
    \includegraphics[width=0.3\columnwidth]{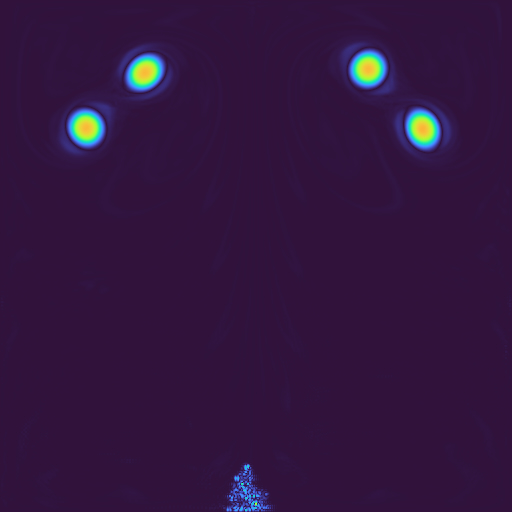}
    \includegraphics[width=0.3\columnwidth]{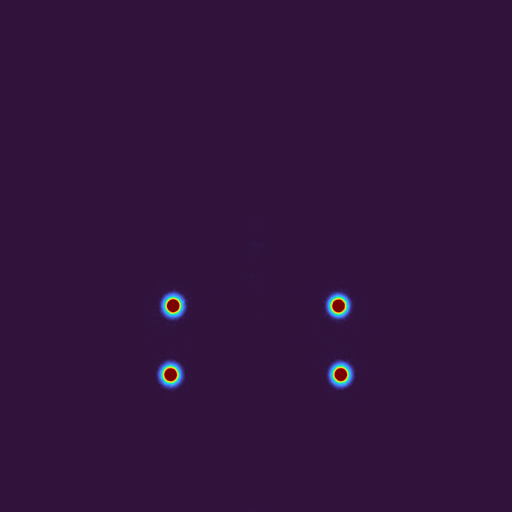}
    \\
    \begin{picture}(0,0)(0,0)
        \put(-110,140){\begin{tikzpicture}
\pgfplotscolorbardrawstandalone[ 
    colormap={myProteinColor}{
        rgb255=(48,18,59)
        rgb255=(70,105,224)
        rgb255=(42,185,238)
        rgb255=(47,241,155)
        rgb255=(161,253,61)
        rgb255=(236,209,58)
        rgb255=(251,129,34)
        rgb255=(210,49,5)
        rgb255=(122,4,3)
    },
    colorbar horizontal,
    point meta min=-1,
    point meta max=1,
    colorbar style={
        width=30pt,
        height=4pt,
        xtick={-1,-0.5,0,0.5,1},
        xtick style={draw=none},
        xticklabels={{\textcolor{white}{0}},{},{},{},{\textcolor{white}{10}}},
        xticklabel style={font=\tiny, xshift=0.0ex,yshift=0.5ex, scale=0.95},
        xlabel={\textcolor{white}{\sffamily abs. vorticity ($\nicefrac{1}{\text{s}}$)}},
        xlabel style={font=\tiny, yshift=18.5pt, scale=0.8},
        axis line style={white},
        }]
\end{tikzpicture}}
        \put(-108,88){\sffamily \scriptsize \textcolor{white}{PolyPIC}}
        \put(-33,88){\sffamily \scriptsize \textcolor{white}{CF+PolyFLIP}}
        \put(42,88){\sffamily \scriptsize \textcolor{white}{Frame 0}}
        \put(-108,14){\sffamily \scriptsize \textcolor{white}{PolyFLIP}}
        \put(-33,14){\sffamily \scriptsize \textcolor{white}{R+PolyFLIP}}
        \put(42,14){\sffamily \scriptsize \textcolor{white}{\textbf{CO-FLIP (Ours)}}}
    \end{picture}
    \caption{2D leapfrogging vortices.
    Our method (CO-FLIP) retains the core vortex structure with minimal dissipation after 500 seconds of simulation.
    Other methods lose both energy and vorticity as time progresses.}
    \label{fig:leapfrog_marathon}
\end{figure}

\section{Results}
\label{sec:Results}
In this section, we detail the experiments conducted to demonstrate our method.
We implement our method (CO-FLIP) as detailed in \secref{sec:Method}.
In addition to adding our method, we implement established methods such as PolyPIC \cite{Fu:2017:PolyPIC} and PolyFLIP \cite{Fei:2021:ASFLIP}.
We chose the PolyFLIP method over the basic FLIP algorithm because it provides better stability by providing a smaller kernel of P2G transfer.
We run these methods with second-order explicit midpoint time integration to make comparisons fair.
Additionally, we combine the second-order advection-reflection \cite{Narain:2019:SAR} and covector fluids \cite{Nabizadeh:2022:CF} with the PolyFLIP method to build Eulerian-Lagrangian methods referred to as R+PolyFLIP and CF+PolyFLIP, respectively.
These methods lower the separation error (i.e. how much advection and pressure projection operators commute) when solving for the Euler equations \cite{Zehnder:2018:ARS, Nabizadeh:2022:CF}.
We build all the above-mentioned methods on top of the codebase open-sourced by \cite{Nabizadeh:2022:CF}.
While the codebase offers many implementations of semi-Lagrangian methods, we instead focus on extending the backend toolset with various performance improvements and implementing the hybrid Eulerian-Lagrangian methods we desire.
Readers should refer to the supplementary material for all code and data.
Lastly, for the NFM experiments, we use the codebase shared by \cite{Deng:2023:FSN} with no additional modifications.

\begin{figure*}
    \centering
    \begin{subfigure}{0.35\linewidth}
        \includegraphics[trim={600px 0px 600px 150px},clip,width=0.32\linewidth]{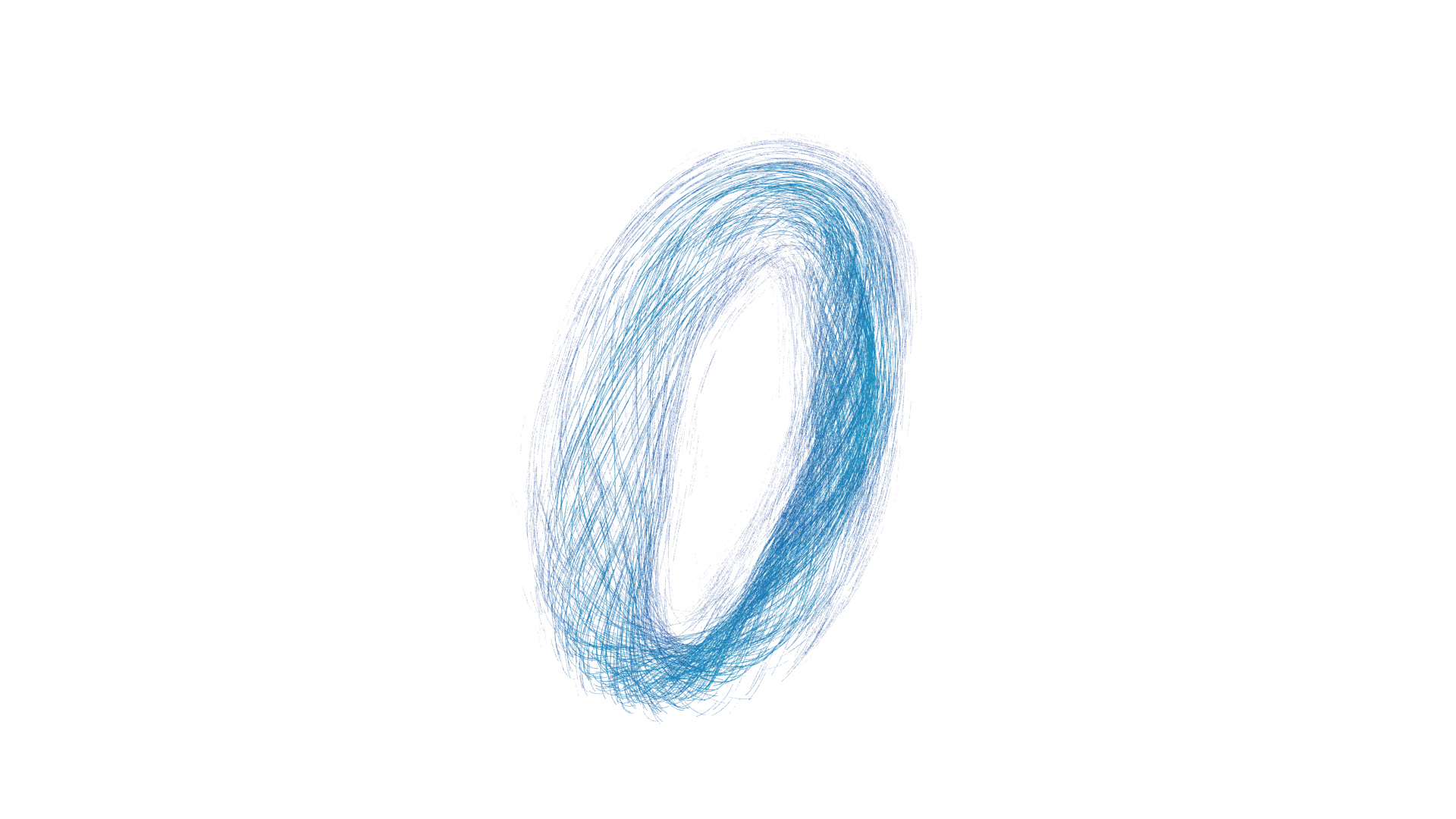}
        \hfill
        \includegraphics[trim={600px 0px 600px 150px},clip,width=0.32\linewidth]{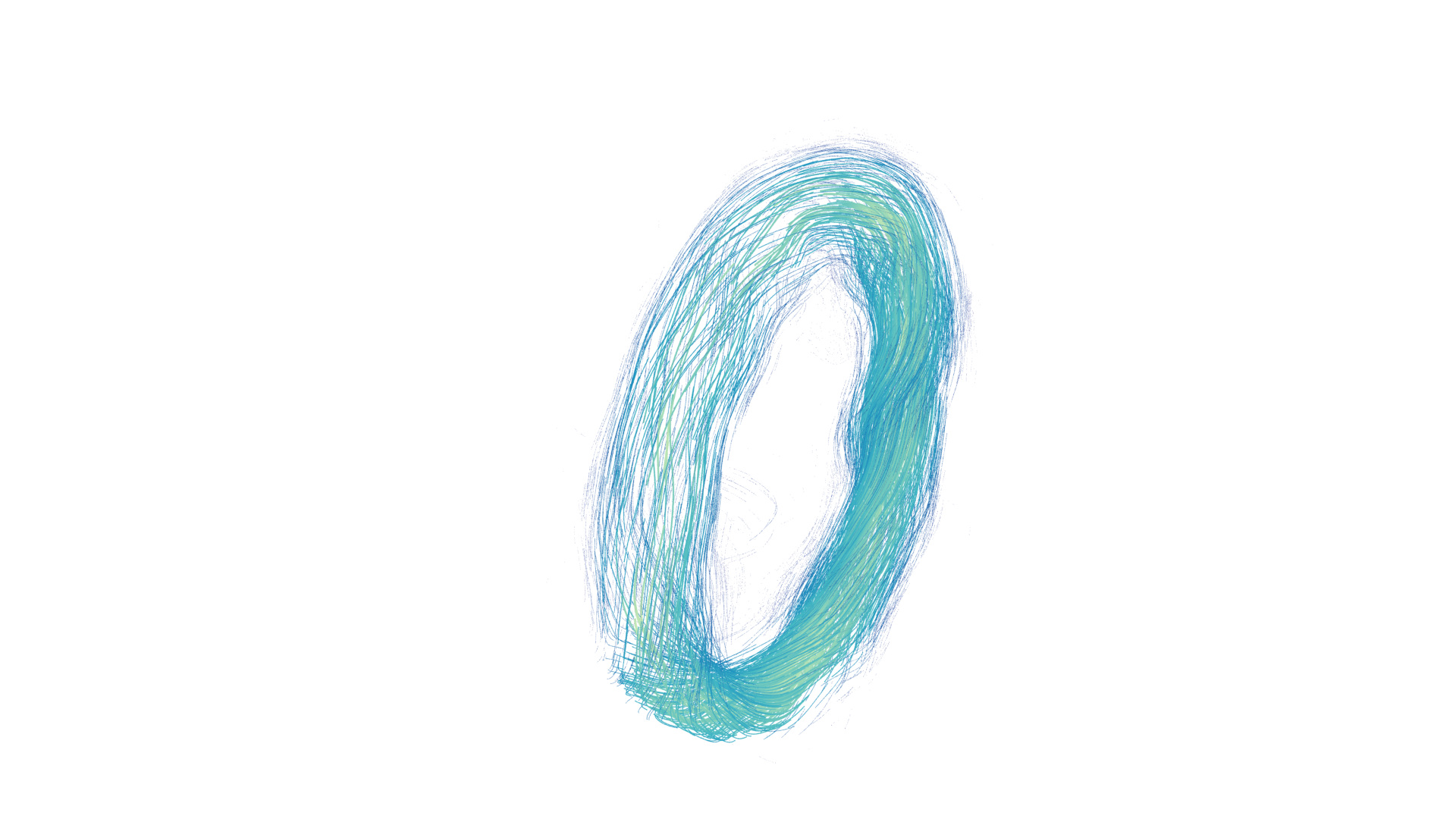}
        \hfill
        \includegraphics[trim={600px 0px 600px 150px0},clip,width=0.32\linewidth]{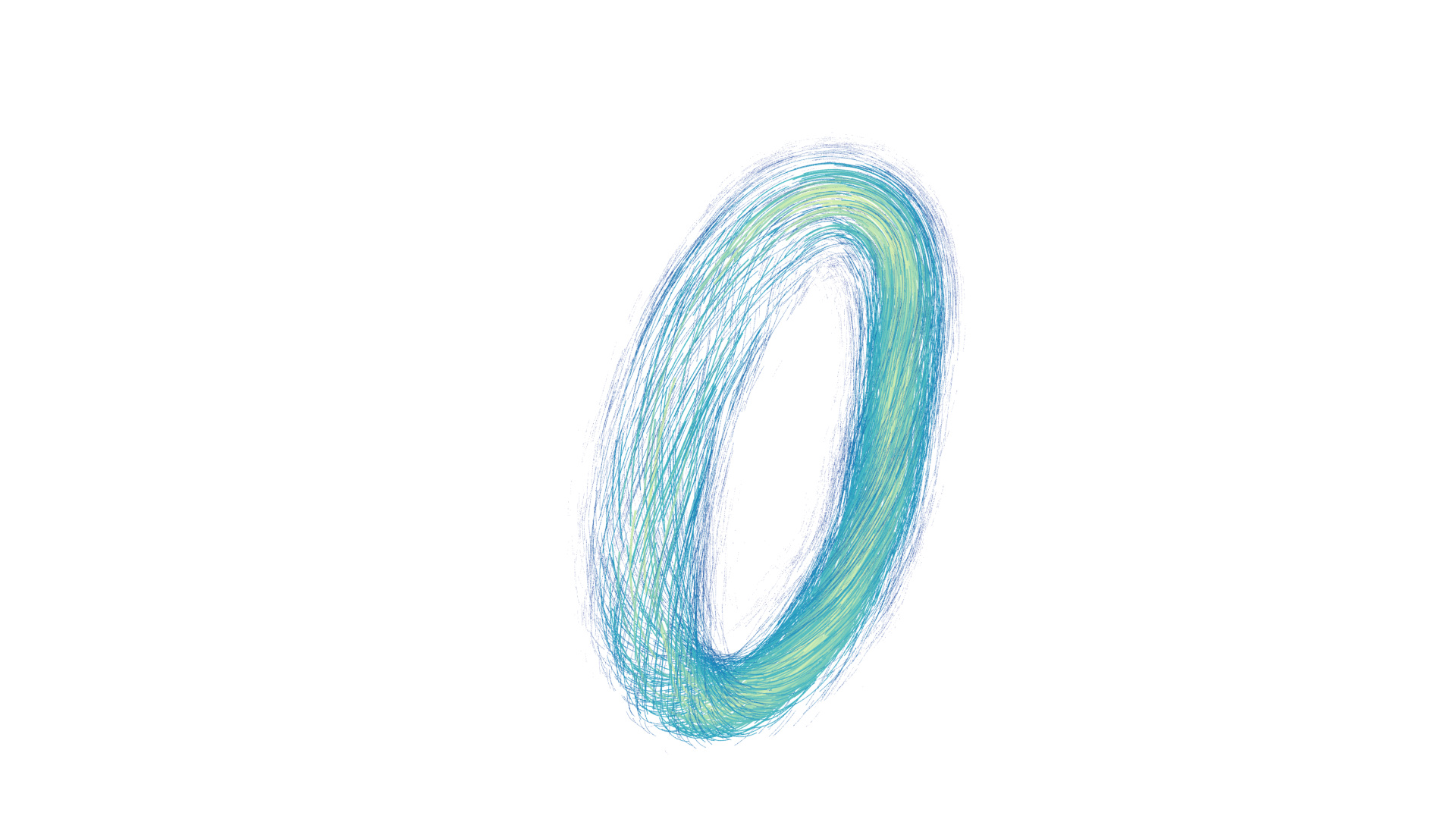}
        \includegraphics[trim={600px 0px 600px 150px},clip,width=0.32\linewidth]{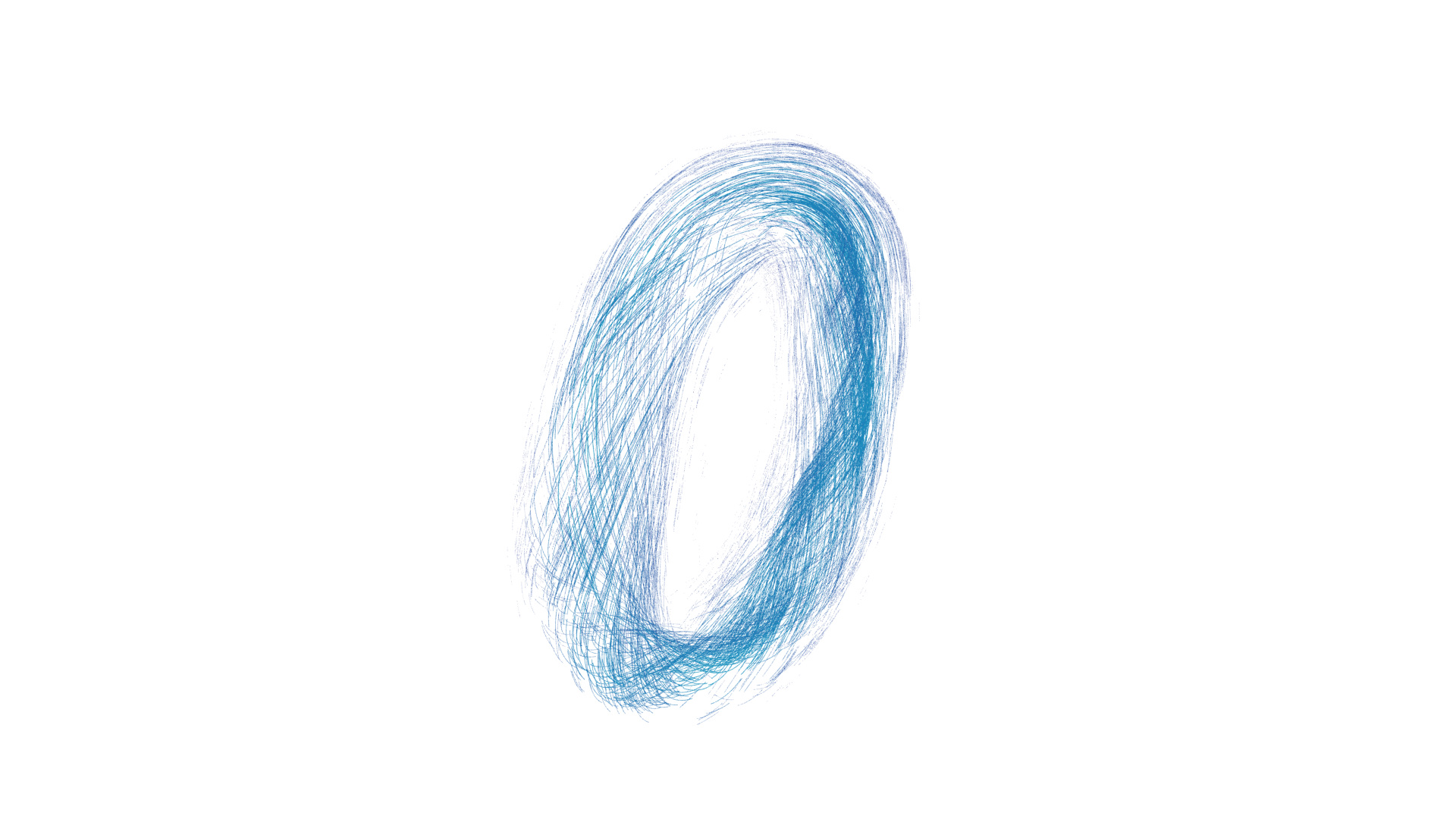}
        \hfill
        \includegraphics[trim={600px 0px 600px 150px},clip,width=0.32\linewidth]{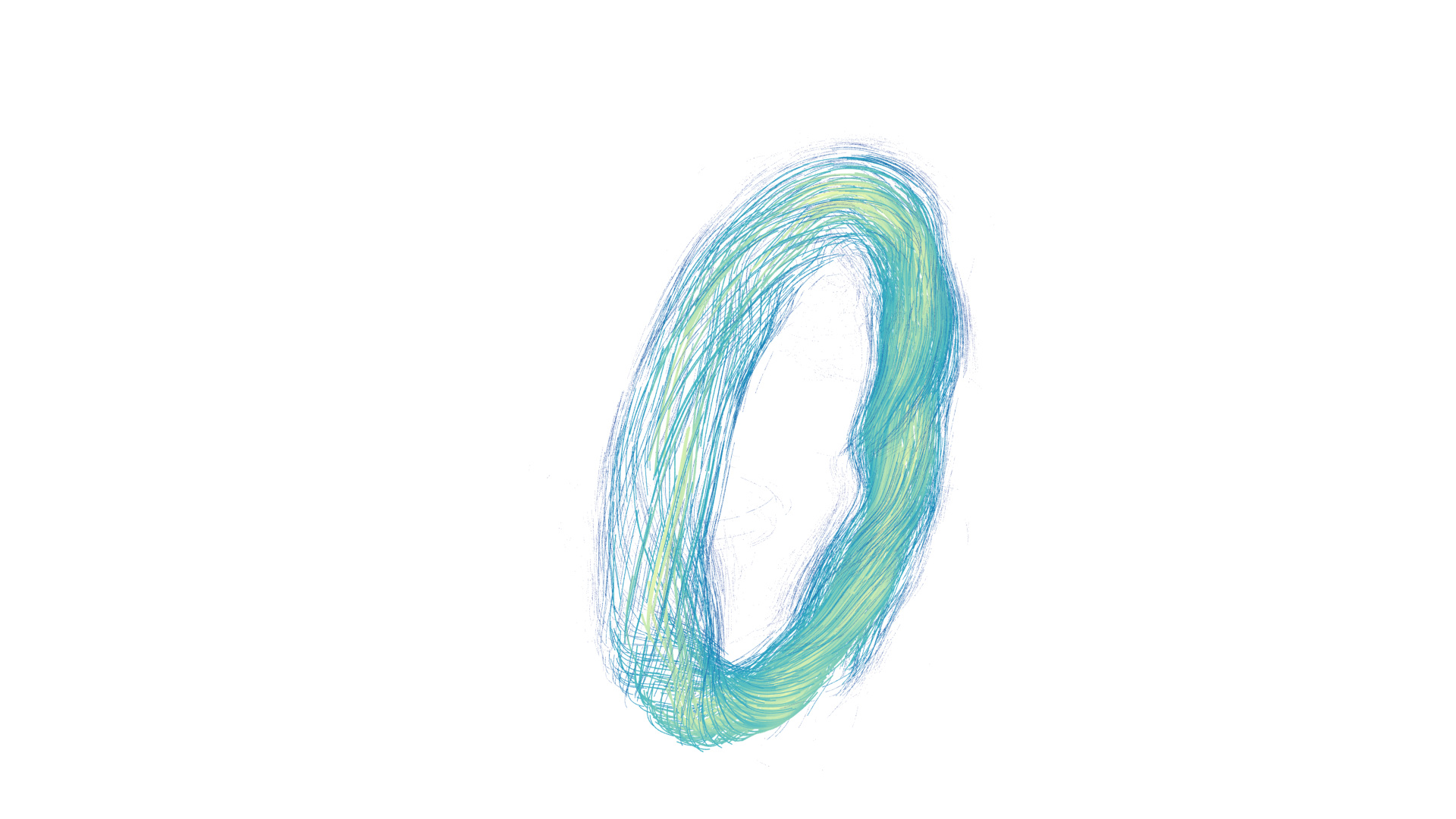}
        \hfill
        \includegraphics[trim={600px 0px 600px 150px},clip,width=0.32\linewidth]{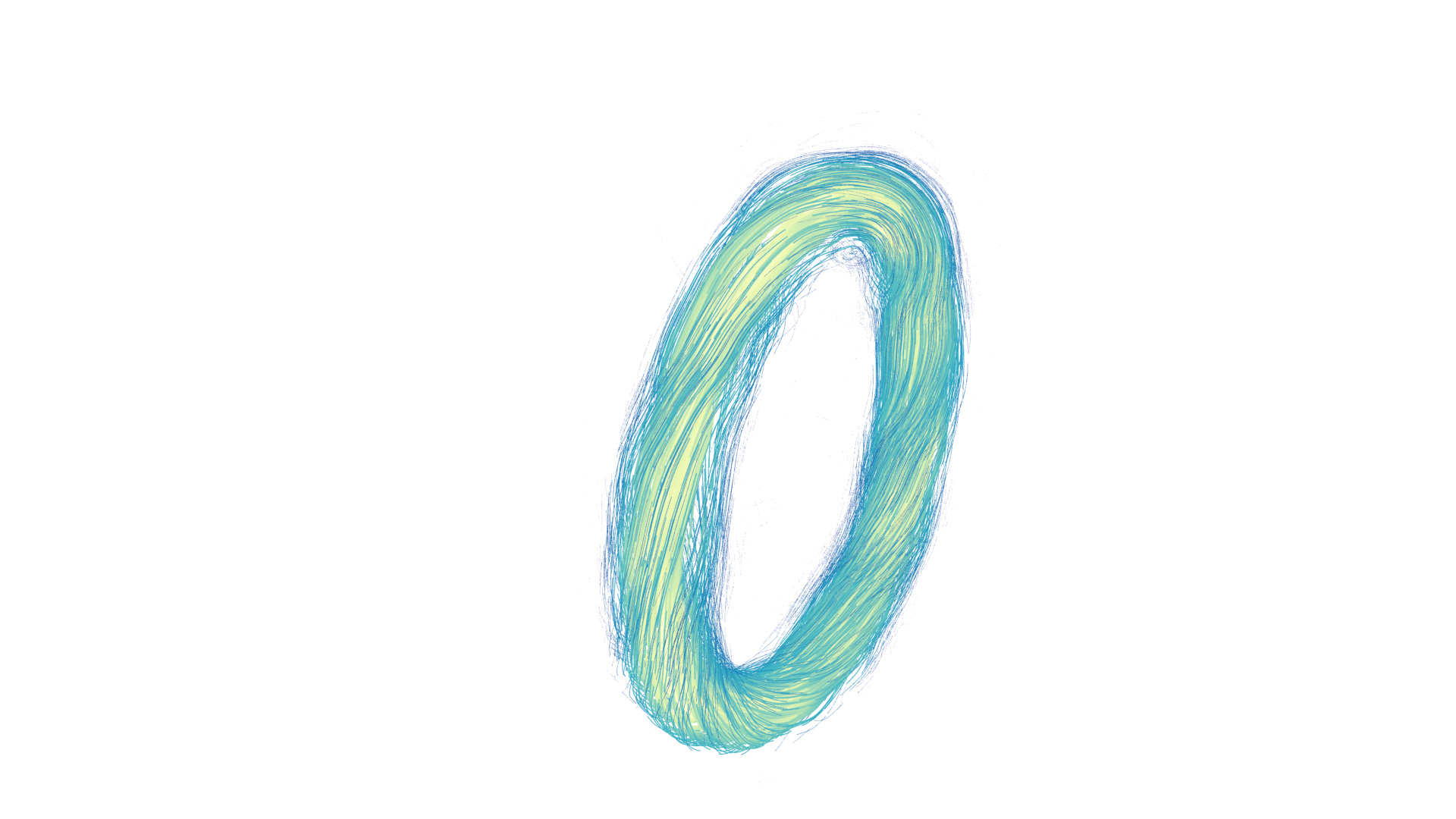}
    \end{subfigure}
    \hfill
    \begin{subfigure}{0.3\linewidth}
%
%
\definecolor{mycolor1}{rgb}{0.00000,0.44700,0.74100}%
\definecolor{mycolor2}{rgb}{0.85000,0.32500,0.09800}%
\definecolor{mycolor3}{rgb}{0.92900,0.69400,0.12500}%
\definecolor{mycolor4}{rgb}{0.49400,0.18400,0.55600}%
\definecolor{mycolor5}{rgb}{0.46600,0.67400,0.18800}%
\definecolor{mycolor6}{rgb}{0.30100,0.74500,0.93300}%
\definecolor{mycolor7}{rgb}{0.63500,0.07800,0.18400}%
\begin{tikzpicture}

\begin{axis}[%
width=0.84\linewidth,
height=0.7\linewidth,
at={(0,0)},
scale only axis,
xmin=0,
xmax=10,
ymin=0.7,
ymax=1.1,
axis background/.style={fill=white},
xlabel={\scriptsize \sffamily Time (s)},
ylabel={\scriptsize \sffamily $\mathcal{E}/\mathcal{E}_0$},
xticklabel style={font=\scriptsize},
yticklabel style={font=\scriptsize},
]
\addplot [color=mycolor1,mark=*, mark size=2pt]
  table[row sep=crcr]{%
0	1\\
0.833333	1\\
1.66667	1\\
2.5	1\\
3.33333	1\\
4.16667	1\\
5	1\\
5.83333	1\\
6.66667	1\\
7.5	1\\
8.33333	0.999998420884747\\
9.16667	0.999998420884747\\
10	0.999998420884747\\
10.8333	0.999998420884747\\
11.6667	0.999998420884747\\
12.5	0.999998420884747\\
13.3333	0.999998420884747\\
14.1667	0.999998420884747\\
};

\addplot [color=mycolor2,mark=o, mark size=2pt]
  table[row sep=crcr]{%
0	1\\
0.833333	0.996344446321787\\
1.66667	0.993072607349321\\
2.5	0.990373971826182\\
3.33333	0.987882194848117\\
4.16667	0.985496215751883\\
5	0.983288671889187\\
5.83333	0.981321147101691\\
6.66667	0.979383624698594\\
7.5	0.977545583885877\\
8.33333	0.975758073404781\\
9.16667	0.973994249016632\\
10	0.972184631515452\\
10.8333	0.970436597855935\\
11.6667	0.968492759161386\\
12.5	0.966451017949321\\
13.3333	0.964200839119319\\
14.1667	0.961825913533128\\
};

\addplot [color=mycolor3,mark=square, mark size=2pt]
  table[row sep=crcr]{%
0	1\\
0.833333	0.990656625868687\\
1.66667	0.983162346060134\\
2.5	0.977292932227772\\
3.33333	0.972126205819515\\
4.16667	0.967483731601827\\
5	0.963319716461677\\
5.83333	0.959463620529842\\
6.66667	0.955779643540092\\
7.5	0.952367267082805\\
8.33333	0.949106481620381\\
9.16667	0.94594201960261\\
10	0.942828087916461\\
10.8333	0.939799426164922\\
11.6667	0.936743920174709\\
12.5	0.933707363058854\\
13.3333	0.930547638259672\\
14.1667	0.927321592400238\\
};

\addplot [color=mycolor4,mark=triangle, mark size=2pt]
  table[row sep=crcr]{%
0	1\\
0.833333	0.936470740569382\\
1.66667	0.893990522404675\\
2.5	0.862765935608567\\
3.33333	0.838004494041368\\
4.16667	0.817152836883352\\
5	0.799609337373654\\
5.83333	0.784082313910211\\
6.66667	0.769889607016137\\
7.5	0.757099116824548\\
8.33333	0.745625573400833\\
9.16667	0.734908405878572\\
10	0.724814972137259\\
10.8333	0.71540843509142\\
11.6667	0.706541940964782\\
12.5	0.698078110418249\\
13.3333	0.689909566497127\\
14.1667	0.681932090392447\\
};

\addplot [color=mycolor5,mark=diamond, mark size=2pt]
  table[row sep=crcr]{%
0	1\\
0.833333	0.938953043110268\\
1.66667	0.898230333042258\\
2.5	0.867800019896318\\
3.33333	0.843171220449625\\
4.16667	0.822087439580725\\
5	0.804171278875321\\
5.83333	0.788566880841583\\
6.66667	0.774453127590666\\
7.5	0.761463674218319\\
8.33333	0.749658525493342\\
9.16667	0.738857667109333\\
10	0.728814763699641\\
10.8333	0.719289796189065\\
11.6667	0.710278027359016\\
12.5	0.701673659327662\\
13.3333	0.693383526796077\\
14.1667	0.685311306319607\\
};

\addplot [color=mycolor6,mark=x, mark size=2pt]
  table[row sep=crcr]{%
0	1\\
0.833333333333333	0.990275147166508\\
1.66666666666667	0.982050481373454\\
2.5	0.974976693806848\\
3.33333333333334	0.968895265990594\\
4.16666666666667	0.963691428275333\\
5	0.958412478938065\\
5.83333333333334	0.954469743765963\\
6.66666666666667	0.950320195958381\\
7.5	0.946382006118927\\
8.33333333333334	0.942429100765025\\
9.16666666666667	0.938414833420376\\
10	0.93435806721547\\
10.8333333333333	0.930563907604296\\
11.6666666666667	0.926657648726693\\
12.5	0.92362966174998\\
13.3333333333333	0.920284637821077\\
14.1666666666667	0.917444145816089\\
};

\end{axis}
\end{tikzpicture}%
        \caption{Energy plot.}
    \end{subfigure}
    \hfill
    \begin{subfigure}{0.3\linewidth}
%
%
\definecolor{mycolor1}{rgb}{0.00000,0.44700,0.74100}%
\definecolor{mycolor2}{rgb}{0.85000,0.32500,0.09800}%
\definecolor{mycolor3}{rgb}{0.92900,0.69400,0.12500}%
\definecolor{mycolor4}{rgb}{0.49400,0.18400,0.55600}%
\definecolor{mycolor5}{rgb}{0.46600,0.67400,0.18800}%
\definecolor{mycolor6}{rgb}{0.30100,0.74500,0.93300}%
\definecolor{mycolor7}{rgb}{0.63500,0.07800,0.18400}%
\begin{tikzpicture}

\begin{axis}[%
width=0.84\linewidth,
height=0.7\linewidth,
at={(0,0)},
scale only axis,
xmin=0,
xmax=10,
ymin=0.5,
ymax=1.1,
ytick={0.5, 0.6, 0.7, 0.8, 0.9, 1, 1.1},
yticklabels={{0.5},{0.6},{0.7},{0.8},{0.9},{1},{1.1}},
axis background/.style={fill=white},
xlabel={\scriptsize \sffamily Time (s)},
ylabel={\scriptsize \sffamily $\mathcal{H}/\mathcal{H}_0$},
xticklabel style={font=\scriptsize},
yticklabel style={font=\scriptsize},
legend pos = south west,
legend image post style={scale=0.8},
legend style={legend cell align=left, align=left, draw=white!15!black, nodes={scale=0.6, transform shape}}
]
\addplot [color=mycolor1, mark=*, mark size=2pt]
  table[row sep=crcr]{%
0	1\\
0.833333	0.995423551508048\\
1.66667	0.988150838782201\\
2.5	0.97921596316439\\
3.33333	0.970021062064082\\
4.16667	0.963954467537819\\
5	0.956951781274565\\
5.83333	0.953639456666753\\
6.66667	0.951609257707255\\
7.5	0.949703070900948\\
8.33333	0.947960900168217\\
9.16667	0.945262635738302\\
10	0.943664479118954\\
10.8333	0.943992511266104\\
11.6667	0.943334446775784\\
12.5	0.943856497936798\\
13.3333	0.944116523419295\\
14.1667	0.945280637502475\\
};
\addlegendentry{\scriptsize \sffamily CO-FLIP}

\addplot [color=mycolor2, mark=o, mark size=2pt]
  table[row sep=crcr]{%
0	1\\
0.833333	0.986705987553768\\
1.66667	0.967865065571818\\
2.5	0.954288459481118\\
3.33333	0.937484230265376\\
4.16667	0.926594282323162\\
5	0.917151617545836\\
5.83333	0.910660037827178\\
6.66667	0.90493549800065\\
7.5	0.899172606179516\\
8.33333	0.89345008487902\\
9.16667	0.888960882984028\\
10	0.884413143834109\\
10.8333	0.880307461885182\\
11.6667	0.875999927333063\\
12.5	0.871787263504444\\
13.3333	0.866983171548472\\
14.1667	0.86540468419151\\
};
\addlegendentry{\scriptsize \sffamily R+PolyFLIP}

\addplot [color=mycolor3, mark=square, mark size=2pt]
  table[row sep=crcr]{%
0	1\\
0.833333	0.973294900597686\\
1.66667	0.943808272323384\\
2.5	0.923340418359705\\
3.33333	0.899951757227837\\
4.16667	0.883789419290246\\
5	0.869330717323596\\
5.83333	0.858162212788977\\
6.66667	0.847589173433775\\
7.5	0.838229268223758\\
8.33333	0.828463639281324\\
9.16667	0.820056478358373\\
10	0.812273041979286\\
10.8333	0.804003140826506\\
11.6667	0.796643594914122\\
12.5	0.789106418710929\\
13.3333	0.78213039274461\\
14.1667	0.7758568138374\\
};
\addlegendentry{\scriptsize \sffamily CF+PolyFLIP}

\addplot [color=mycolor4, mark=triangle, mark size=2pt]
  table[row sep=crcr]{%
0	1\\
0.833333	0.887002912733064\\
1.66667	0.804332160569709\\
2.5	0.745988684143065\\
3.33333	0.700343754983236\\
4.16667	0.660576773628361\\
5	0.629321916550097\\
5.83333	0.603400005248168\\
6.66667	0.579284674744808\\
7.5	0.558225392653776\\
8.33333	0.539759916513763\\
9.16667	0.522739705012606\\
10	0.506763071469951\\
10.8333	0.492689907975398\\
11.6667	0.47955737761172\\
12.5	0.467240331764939\\
13.3333	0.455617658873138\\
14.1667	0.444437043182327\\
};
\addlegendentry{\scriptsize \sffamily PolyFLIP}

\addplot [color=mycolor5, mark=diamond, mark size=2pt]
  table[row sep=crcr]{%
0	1\\
0.833333	0.885555629568177\\
1.66667	0.801835243868222\\
2.5	0.744636271701678\\
3.33333	0.698983268437721\\
4.16667	0.65585745976573\\
5	0.623439931693079\\
5.83333	0.59869482106776\\
6.66667	0.574032470009749\\
7.5	0.551996221319268\\
8.33333	0.534275581285034\\
9.16667	0.517360333137536\\
10	0.501422051589488\\
10.8333	0.486692867134561\\
11.6667	0.473794485790586\\
12.5	0.462309072668956\\
13.3333	0.450892289432411\\
14.1667	0.439713692267632\\
};
\addlegendentry{\scriptsize \sffamily PolyPIC}

\addplot [color=mycolor6, mark=x, mark size=2pt]
  table[row sep=crcr]{%
0	1\\
0.833333333333333	0.973811654282471\\
1.66666666666667	0.944992271139067\\
2.5	0.921678961931503\\
3.33333333333334	0.893410505172641\\
4.16666666666667	0.872929371689889\\
5	0.851291171003049\\
5.83333333333334	0.835568396649365\\
6.66666666666667	0.817525572055186\\
7.5	0.802930865727542\\
8.33333333333334	0.788653410407735\\
9.16666666666667	0.776184923814519\\
10	0.763682218191208\\
10.8333333333333	0.752008367074844\\
11.6666666666667	0.740422499464675\\
12.5	0.729605168071116\\
13.3333333333333	0.719124053147838\\
14.1666666666667	0.711637973333284\\
};
\addlegendentry{\scriptsize \sffamily NFM}

\end{axis}
\end{tikzpicture}%
        \caption{Helicity plot.}
    \end{subfigure}
    \\
    \begin{picture}(0,0)(0,0)
        \put(-235,88){\sffamily \scriptsize PolyPIC}
        \put(-170,88){\sffamily \scriptsize CF+PolyFLIP}
        \put(-110,88){\sffamily \scriptsize NFM}
        \put(-235,10){\sffamily \scriptsize PolyFLIP}
        \put(-170,10){\sffamily \scriptsize R+PolyFLIP}
        \put(-110,10){\sffamily \scriptsize \textbf{CO-FLIP (Ours)}}
        \put(-267,50){\begin{tikzpicture}
\pgfplotscolorbardrawstandalone[ 
    colormap={myProteinColor}{
        rgb255=(8, 29, 88)
        rgb255=(37, 52, 148)
        rgb255=(34, 94, 168)
        rgb255=(29, 145, 192)
        rgb255=(65, 182, 196)
        rgb255=(127, 205, 187)
        rgb255=(199, 233, 180)
        rgb255=(237, 248, 177)
        rgb255=(255,255, 217)
    },
    point meta min=-1,
    point meta max=1,
    colorbar style={
        width=4pt,
        height=35pt,
        ytick={-1,-0.5,0,0.5,1},
        ytick style={draw=none},
        yticklabels={{0},{},{},{},{10}},
        yticklabel style={font=\scriptsize, xshift=-0.5ex},
        ylabel={\sffamily vorticity norm ($\nicefrac{1}{\text{s}}$)},
        ylabel style={font=\tiny, yshift=25pt}
        }]
\end{tikzpicture}}
    \end{picture}
    \caption{Twisted torus comparison, and its associated energy and helicity plots.
    Note that our method exactly preserves energy over time at a low resolution of $64\times64\times64$.
    The measured helicity shows minimal relative error compared with other methods.
    For clarity purposes, the plots show data samples at x20 times lower frequency than the simulation frame rate.
    }
    \label{fig:twistedtorus}
\end{figure*}

\begin{table}
    \centering
    \caption{Experiment statistics.}
    \label{tab:Statistics}
    \scriptsize
    \setlength{\tabcolsep}{3.6pt}
\begin{tabularx}{\columnwidth}{m{86pt}m{33pt}m{37pt}m{32pt}p{16pt}}
    \toprule
    \rowcolor{white}
    Figure name (figure number) & Domain size & Grid resolution & Avg. particles & $\Deltait t$ \\
    \midrule
    Pyroclastic plume (Fig.~\ref{fig:pyroclastic}) & $5\times10\times5\text{ m}^3$ & $96\times192\times96$ & $3\times10^7$ & $\nicefrac{1}{288}$ s \\
    Unknot (Fig.~\ref{fig:unknot}) & $5\times5\times5\text{ m}^3$ & $64\times64\times64$ & $2\times10^6$ & $\nicefrac{1}{48}$ s \\
    Spot obstacle (Fig.~\ref{fig:spot}) & $10\times5\times5\text{ m}^3$ & $128\times64\times64$ & $6\times10^6$ & $\nicefrac{1}{96}$ s \\
    Rayleigh-Taylor Instability (Fig.~\ref{fig:RT_instability}) & $0.2\times0.4\text{ m}^2$ & $256\times512$ & $3\times10^6$ & $\nicefrac{1}{288}$ s \\
    Rocket (Fig.~\ref{fig:rocket}) & $5\times5\times5\text{ m}^3$ & $64\times64\times64$ & $1\times10^6$ & $\nicefrac{1}{48}$ s \\
    Ink jet (Fig.~\ref{fig:inkjet}) & $5\times10\times5\text{ m}^3$ & $64\times128\times64$ & $1\times10^7$ & $\nicefrac{1}{48}$ s \\
    3D smoke plume (Fig.~\ref{fig:smokeplumes}) & $5\times10\times5\text{ m}^3$ & $64\times128\times64$ & $5\times10^6$ & $\nicefrac{1}{192}$ s \\
    2D smoke plume (Fig.~\ref{fig:2dsmokeplume}) & \makebox[0pt][l]{$0.5\times1.0\text{ m}^2$} & $256\times512$ & $4\times10^6$ & $\nicefrac{1}{288}$ s \\
    Leapfrogging vortices (Fig.~\ref{fig:leapfrog_marathon}) & $2\pi\times2\pi\text{ m}^2$ & $256\times256$ & $1\times10^6$ & $\nicefrac{1}{96}$ s \\
    Twisted torus (Fig.~\ref{fig:twistedtorus}) & $5\times5\times5\text{ m}^3$ & $64\times64\times64$ & $4\times10^6$ & $\nicefrac{1}{48}$ s \\
    Trefoil knot (Fig.~\ref{fig:trefoilknot}) & $5\times5\times5\text{ m}^3$ & $64\times64\times64$ & $9\times10^5$ & $\nicefrac{1}{72}$ s \\
    Leapfrogging rings (Fig.~\ref{fig:leapfrog_rings}) & $10\times5\times5\text{ m}^3$ & $128\times64\times64$ & $7\times10^6$ & $\nicefrac{1}{72}$ s \\
    \bottomrule
\end{tabularx}
\end{table}

\subsection{Performance Considerations}
\label{sec:performanceconsiderations}
All the methods are parallelized to run on the CPU.
We expect all parallelization to be easily transferred to the GPU for additional speed-ups.
The experiments ran on a desktop machine with an Intel i9-13900K processor with 64GB of memory.
See \tabref{tab:Statistics} for a summary of experiment statistics, and \figref{fig:timings} for how long our method takes per iteration for the Trefoil knot experiment.
Our method's additional cost compared with traditional methods is due to the use of high-order methods.
The majority of the time is spent solving the global problems of streamform-vorticity solve, and pseudoinverse P2G map.
We conduct a performance study under the same conditions such as same particle-per-cell count, and CFL number.
Compared to traditional methods, CO-FLIP is slower in wall clock time by an order of magnitude, given our modest optimizations.
However, our results provide higher accuracy, and preservation of energy and Casimirs.
Note that increasing the resolution of traditional methods, for achieving lower error, is not always practical, as they would hit system memory constraints.
This is because the PolyPIC methods require 8$\times$ additional data per particle compared to our method, and that they are only second order accurate compared to the third order accuracy of CO-FLIP.
Finally, our method's performance scales linearly, similar to traditional methods, with increasing spatiotemporal resolution.
Below, we discuss optimizations for the global solves that impact our method.

\subsubsection{Pesudoinverse P2G Solve}
\label{sec:pseudoinversep2gsolve}
We use a preconditioned conjugate gradient (PCG) solver to tackle the pseudoinverse P2G problem.
This solver has three major components: (1) a preconditioner for the solver, (2) the interpolation operator (i.e.\ G2P transfer), and (3) the interpolation transpose operator (i.e.\ weighted sum of the particles values back to the grid).
We use the \emph{incomplete Cholesky (IC) with limited memory} implementation from the Eigen library \cite{Lin:1999:IC,Gael:2010:Eigen} to speed up the solver.
The IC preconditioner requires the construction of matrix $\hat\cI^\intercal\hat\cI$, which we build explicitly, without dealing with the much larger matrix $\hat\cI$.
The sparsity pattern of $\hat\cI^\intercal\hat\cI$ is exactly the same as the one for $\star_2$, and for the IC preconditioner, we only require its lower/upper triangular part to be constructed.
For this, we run through the particles and subsequently go through the neighboring grid elements using a double for-loop to deposit their contribution to said elements.
For the purpose of this paper, we only parallelize the for-loop over particles in the three axes of the domain (in $\mathbb{R}^3$), but more advanced binning strategies could further allow for parallelization and performance boost.
Additionally, evaluating the IC preconditioner itself incurs an overhead, but in turn, it drastically reduces the number of iterations during the least-squares pseudoinverse solve.
As discussed in \cite{Gould:2017:SPS}, if the IC preconditioner is evaluated at every step, the combined time taken to evaluate the preconditioner and run the solver is no different than using a simple diagonal preconditioner.
To alleviate this problem and continue to benefit from the IC preconditioner's performance improvement, we note that the preconditioner for one step could be reused for some future steps with a minimal drop in performance.
We find empirically that re-evaluating the preconditioner every five timesteps shows an improvement in total runtime.

The second major component of the pseudoinverse solve is the interpolation operator, which is embarrassingly parallelizable because there is no data race between particles looking up their values from the grid.
The third component, the transpose of this operator, is more complicated.
This is because many particles could be in the neighborhood of a grid element, causing a data race.
One approach to this problem is through \emph{scattering} \cite{Gao:2018:GPUMPM}.
Here, each particle's information is transferred to its neighboring grid element.
A trivial method of introducing parallelization for the scattering technique is through atomic writes.
For this paper, we employ this method.
However, as the number of particles in the system grows, it is slower to resolve write conflicts.
To resolve this problem, there are a few works that optimize and alleviate this problem on both CPU \cite{Fang:2018:TAM,Hu:2018:MLS} and GPU \cite{Gao:2018:GPUMPM} by employing a method of binning using a form of data structure (e.g.\ SPGrids).
This optimization level is beyond the scope of the current paper, but we expect our method to similarly improve if such speed-up methods are used.
Finally, to further lower write conflicts and improve the performance, we opted to use an adaptive set of particles throughout the domain.
Concretely, the number of particles per cell is selected based on the relative magnitude of local vorticity compared to its maximum in the entire domain.
This would additionally lower the memory footprint of the particles.
Overall, the use of the IC preconditioner, as detailed above, improves the runtime of the P2G solve by roughly x50.

\subsubsection{Streamform-Vorticity Solve}
We employ a PCG solver based on streamform and vorticity to find divergence-free velocity fields.
Since the Laplacian matrix comes from a high-order Galerkin Hodge star, it has a more dense sparsity pattern than the simple finite-difference Laplacian matrix.
This more accurate Laplacian matrix allows for a more accurate pressure projection operation but incurs an overhead to the solver.
We detail various ways of improving this overhead.

As mentioned in \secref{sec:AccelerationForLinearSolves}, we employ a geometric multigrid (GMG) preconditioner for the solution based on the edges (nodes in 2D) of the systems.
The preconditioner takes an input matrix and applies the prolongation and restrictor operators on either side to construct the progressively smaller matrices for the coarser levels.
As a result of these GMG preconditioners, the PCG solver can converge in much fewer iterations \cite{Saad:2003:IMS}.
However, there is a tradeoff between the complexity of the GMG preconditioner and the total time the solver takes to converge.
As such, we opt to build the GMG preconditioner using a smaller Laplacian matrix from a lower-order Hodge star.
Note that the solver still uses the high-order Laplacian matrix, and only the preconditioner is approximated further.
This introduces no error in the solver as the preconditioner is an approximation anyway.
Our experiments show that a Laplacian matrix built from the second-order Galerkin Hodge star produces the best total runtime.
While the number of iterations is larger than using the GMG preconditioner built from the Laplacian matrix based on the third-order Hodge star, each iteration takes less time, reducing the total time.

Another avenue for optimization is in the inner problems that must be solved at each level of the GMG grids.
Traditionally, without the presence of obstacles or multiphase fluids, the Laplacian decouples, which results in simplified solves during the inner iterative solves of the preconditioners, may use red-black Gauss-Seidel parallelization for these solves for acceleration.
However, in the presence of Laplacian matrices made from high-order Hodge stars, they do not decouple, making the inner solves more complicated.
Inspired by the results in \cite{Adams:2003:PMS,Saad:2003:IMS}, we note that instead of using multicolor Gauss-Seidel methods, one may use Chebyshev acceleration (see \cite[algorithm 12.1]{Saad:2003:IMS}) to make fast a weighted Jacobi iteration solver by choosing an optimal weight.
For this, one needs only an approximation of the largest and smallest eigenvalues of the matrix \cite{Adams:2003:PMS} at each grid level, which we get by using the Spectra \cite{Qiu:2022:Spectra} library.
Overall, we see a x3 reduction in the number of iterations needed for the PCG solver to converge.
We observe that in some experiments, despite the reduction in the number of iterations, the solver takes a longer time to complete compared to a CG solver with fast parallelized matrix-vector multiplication and no preconditioner.
We believe the slowdown is due to an unoptimized implementation of the matrix-vector multiplication borrowed from the Eigen library, which would be resolved if one uses explicit matrix-free multiplications instead.

\subsubsection{Miscellaneous improvements}
For all the Hodge stars that must be inverted, we evaluate an IC preconditioner on the first simulation frame to speed up the solver.
This lowers the number of iterations needed to converge by roughly five times.
Additionally, we use parallelized matrix-free operations for all matrix-vector multiplications involving the Hodge star operators, which further enhance the performance over sparse matrix-vector multiplication solutions offered by the Eigen library.

\subsection{Validation}
We validate our method through a collection of experiments.
We showcase how our energy conserves energy and the Casimirs throughout the simulation.

\subsubsection{2D Taylor-Green Vortices Convergence}
Here, we validate the order of accuracy of our method using the Taylor-Green vortices experiment \cite{Taylor:1937:MPS}, an analytical solution to the Euler equations.
At a fixed $CFL=0.4$, we run the experiment using our method and traditional methods for one second.
We measure the relative $L^2$ and $L_\text{inf}$ errors under progressively higher spatio-temporal resolutions.
As seen in \figref{fig:convergenceplot_2d}, our method shows third-order convergence when using quadratic B-splines as the interpolation operator.
The traditional methods remain only at second-order, as expected.
Studying the plot, our method requires a lower spatio-temporal resolution to achieve a fixed error rate.
For instance, if the user is looking for an error of $10^{-5}$, our method can achieve this error at a $64^2$ resolution, while traditional methods must run at $1024^2$ to match this accuracy.

\begin{figure*}
    \centering
    \includegraphics[trim={300px 160px 800px 40px},clip,width=0.16\linewidth]{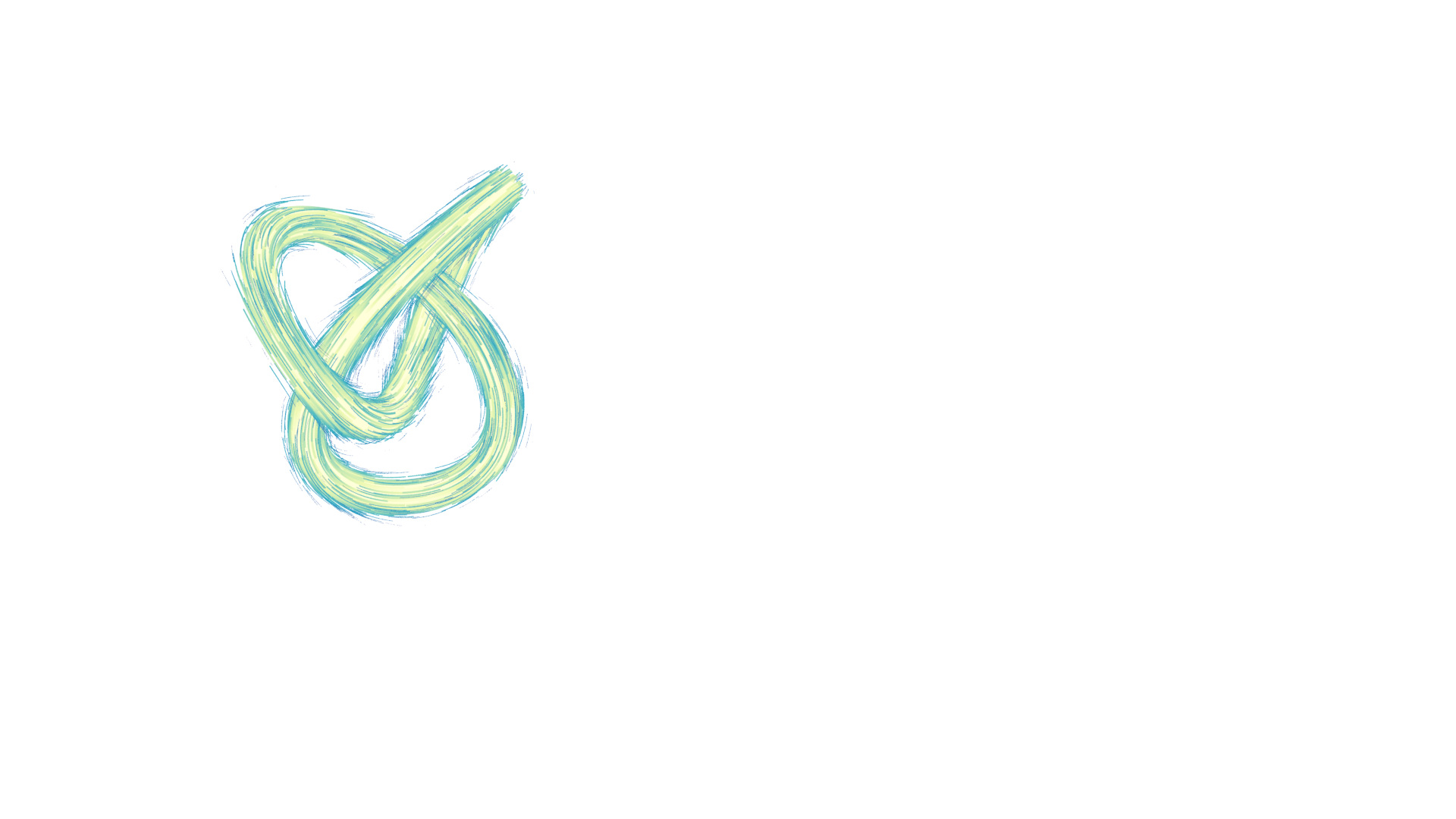}
    \includegraphics[trim={300px 160px 800px 40px},clip,width=0.16\linewidth]{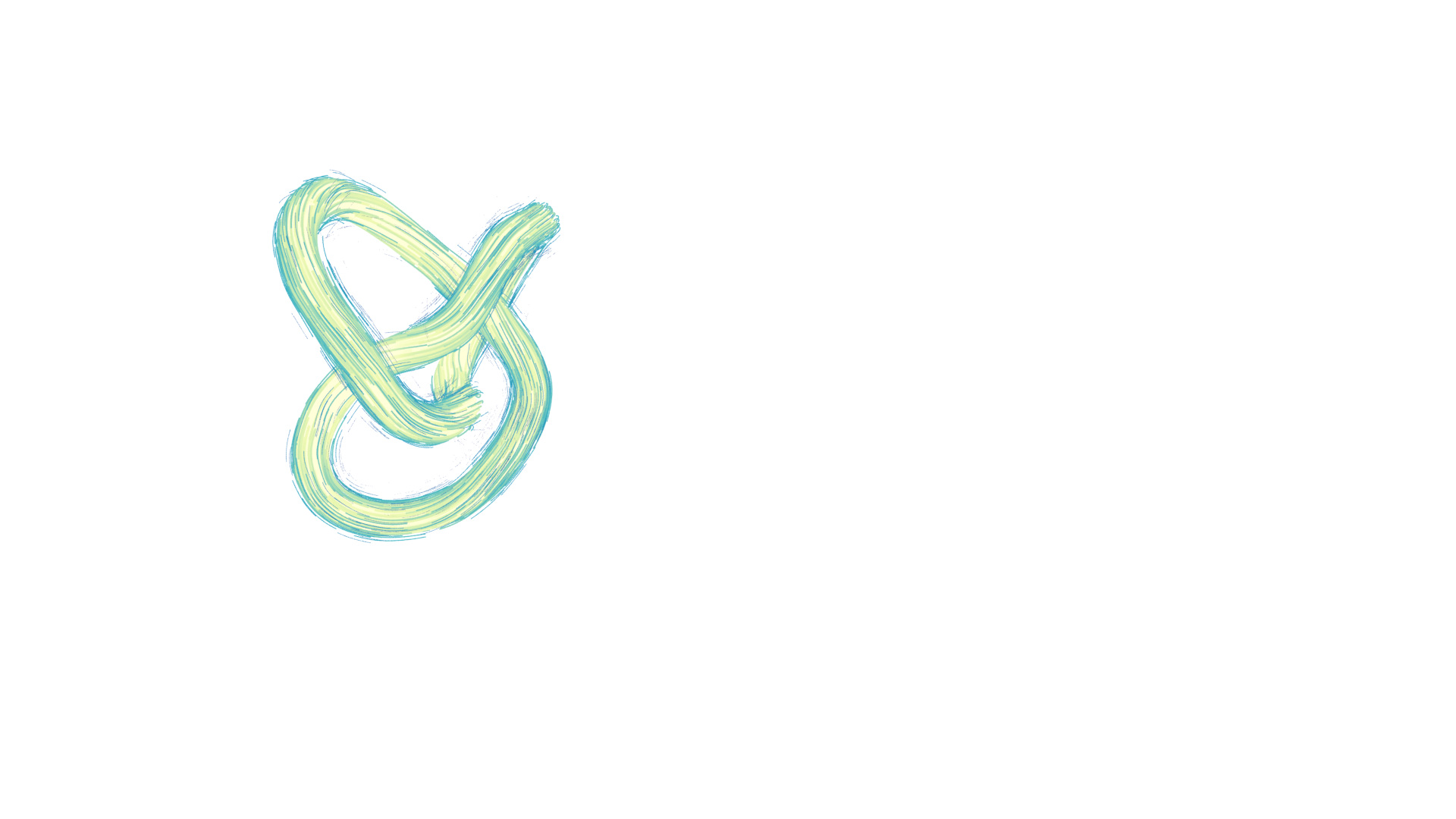}
    \includegraphics[trim={300px 160px 800px 40px},clip,width=0.16\linewidth]{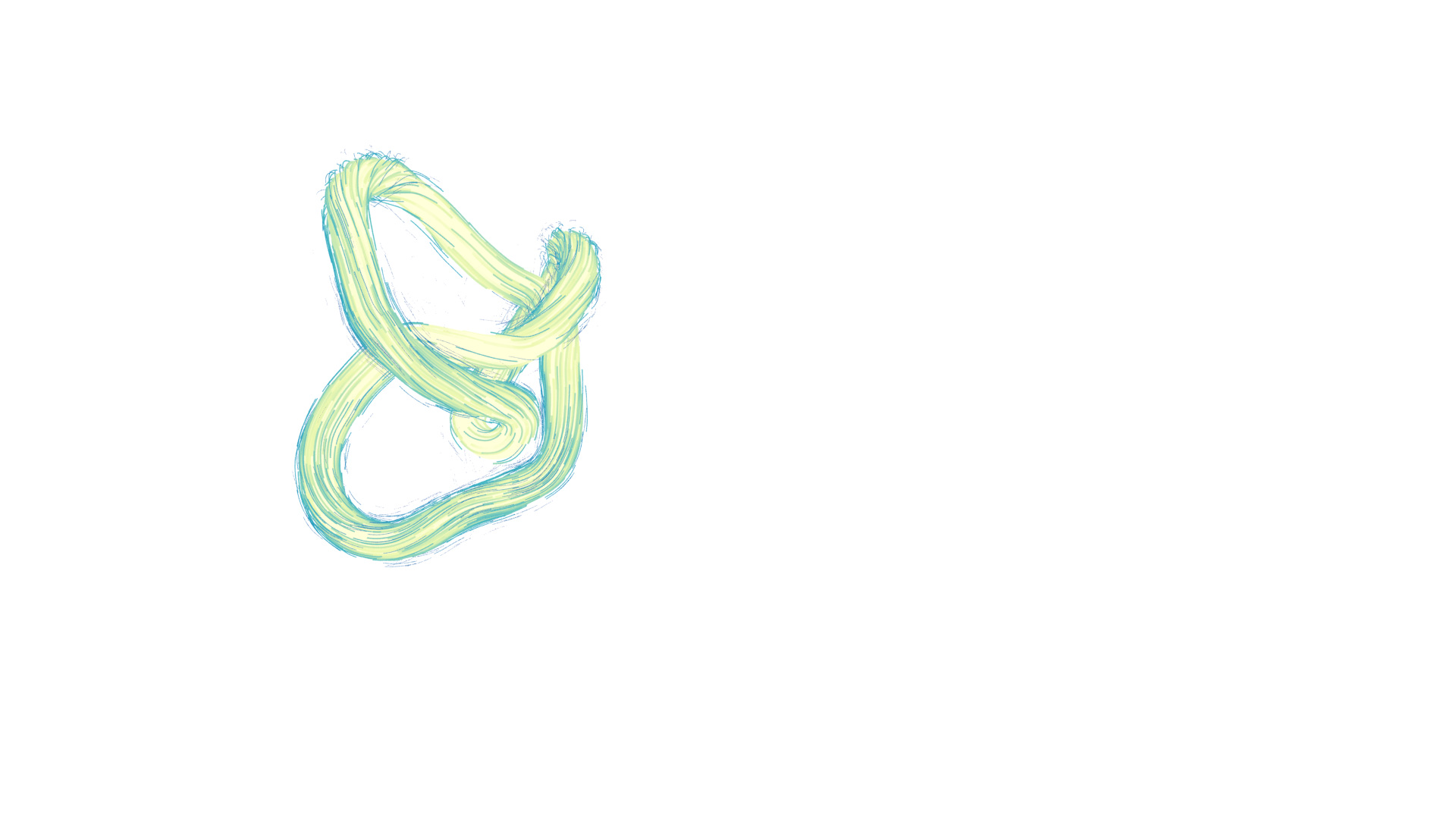}
    \includegraphics[trim={300px 160px 800px 40px},clip,width=0.16\linewidth]{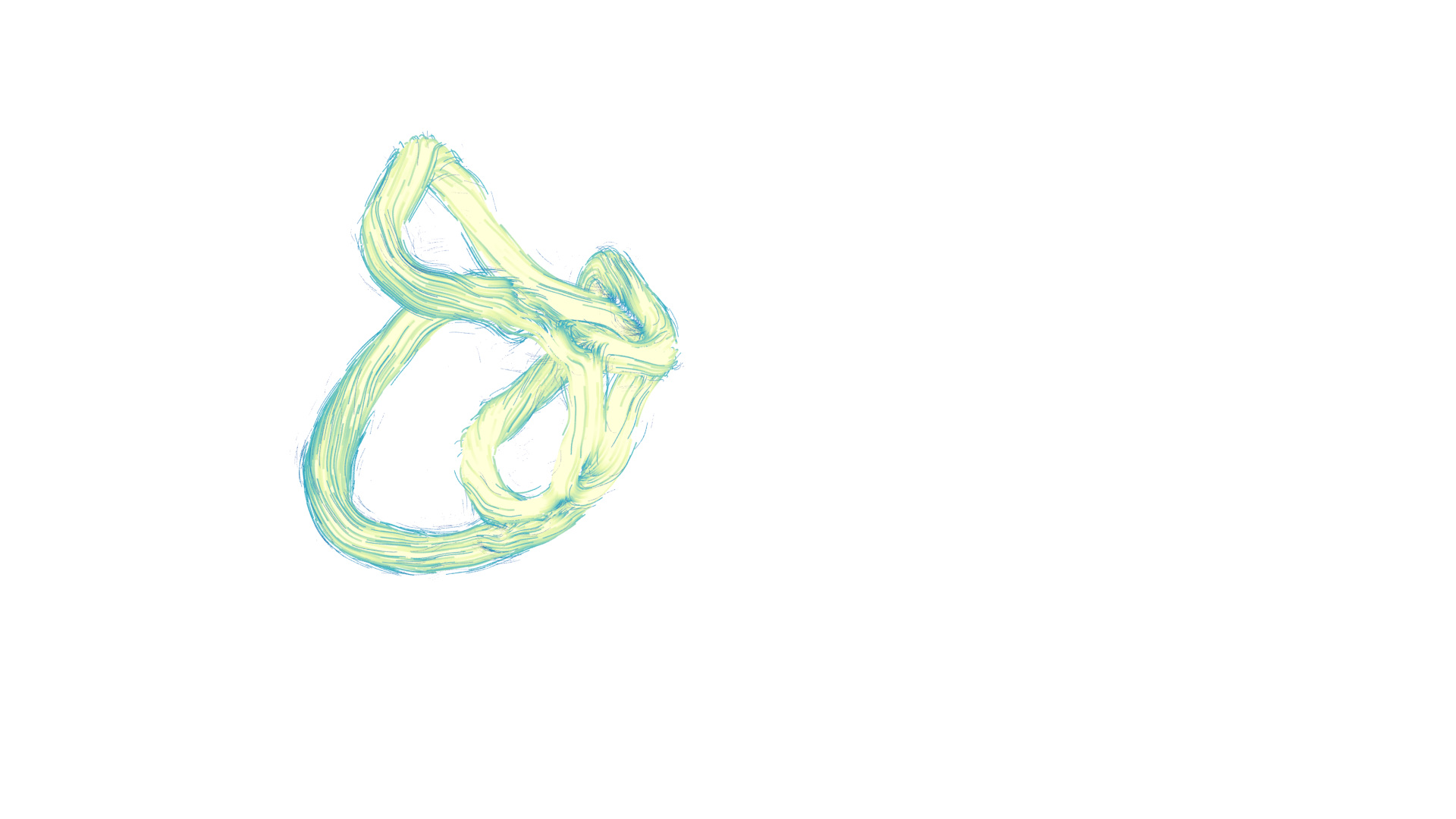}
    \includegraphics[trim={300px 160px 800px 40px},clip,width=0.16\linewidth]{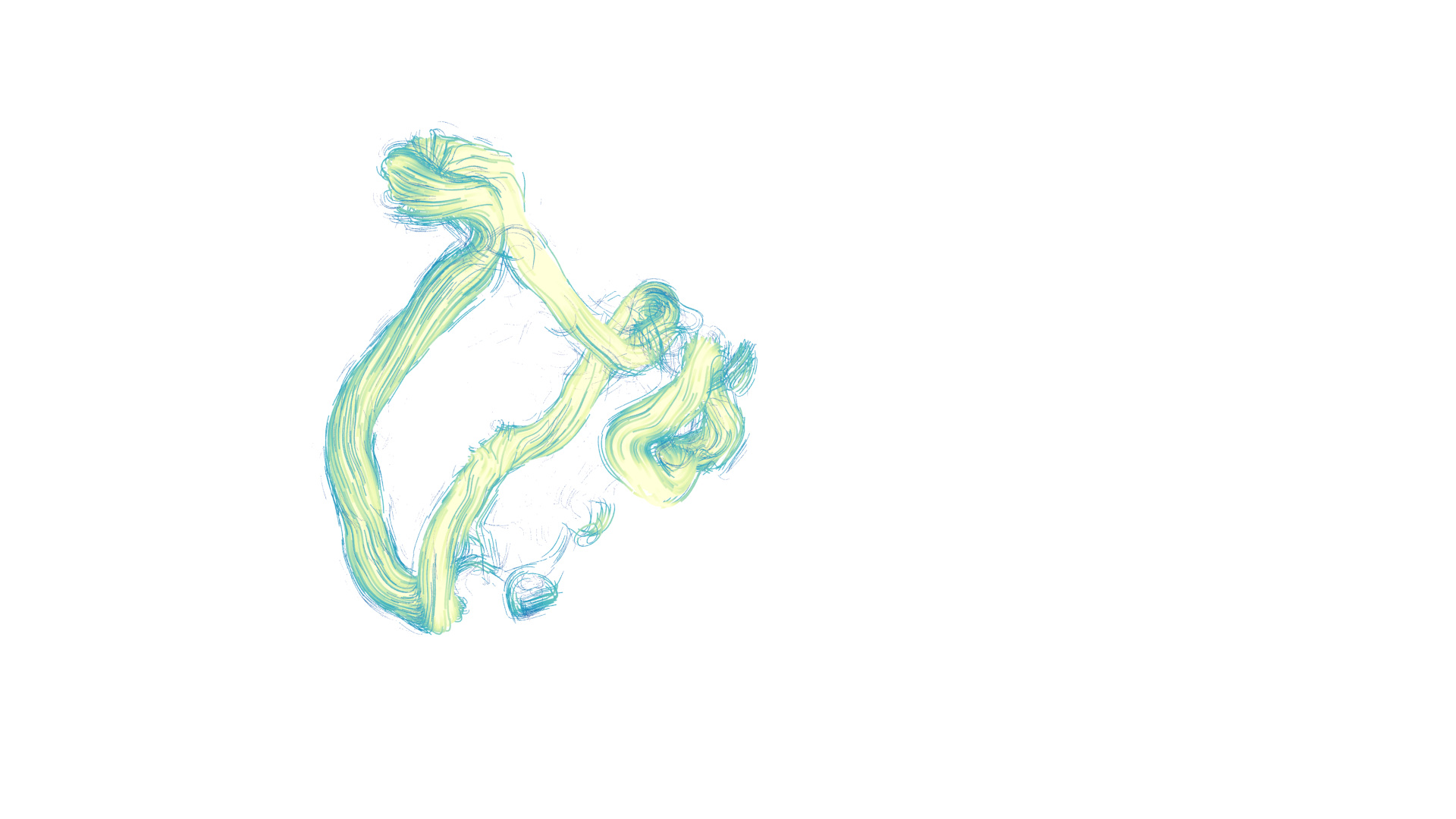}
    \includegraphics[trim={300px 160px 800px 40px},clip,width=0.16\linewidth]{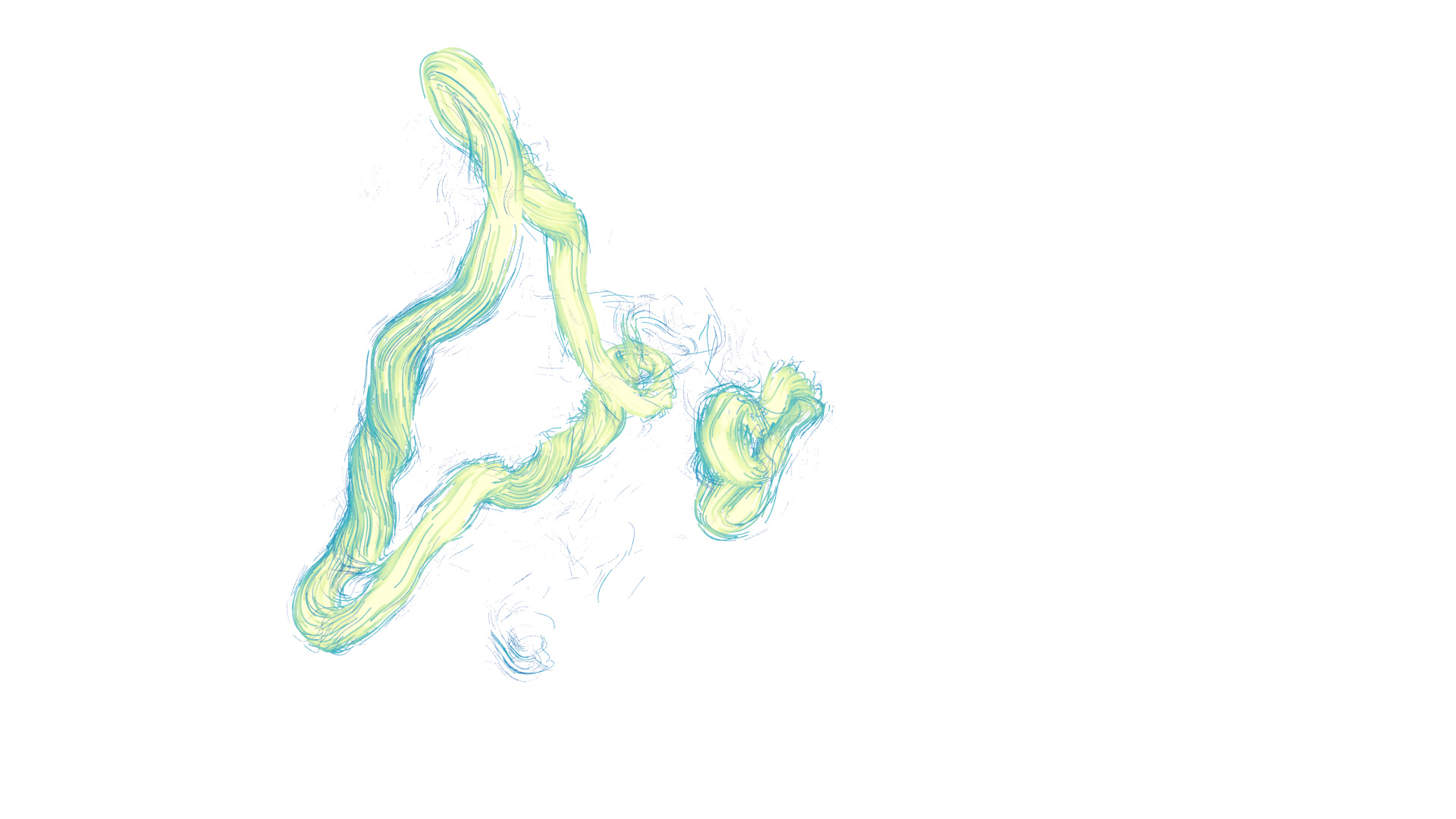}
    \vspace{10pt}
    \includegraphics[trim={550px 210px 550px 90px},clip,width=0.16\linewidth]{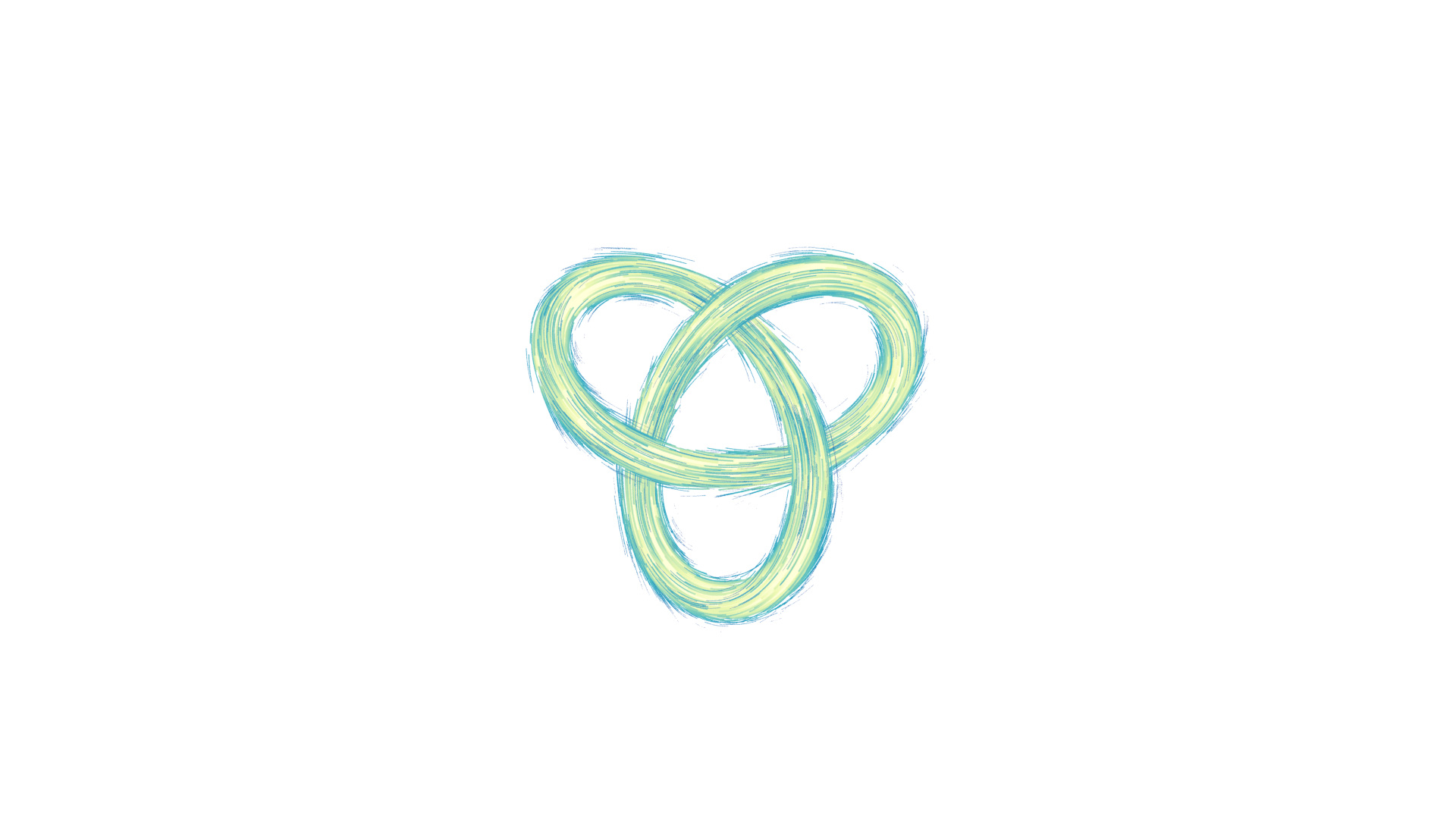}
    \includegraphics[trim={550px 210px 550px 90px},clip,width=0.16\linewidth]{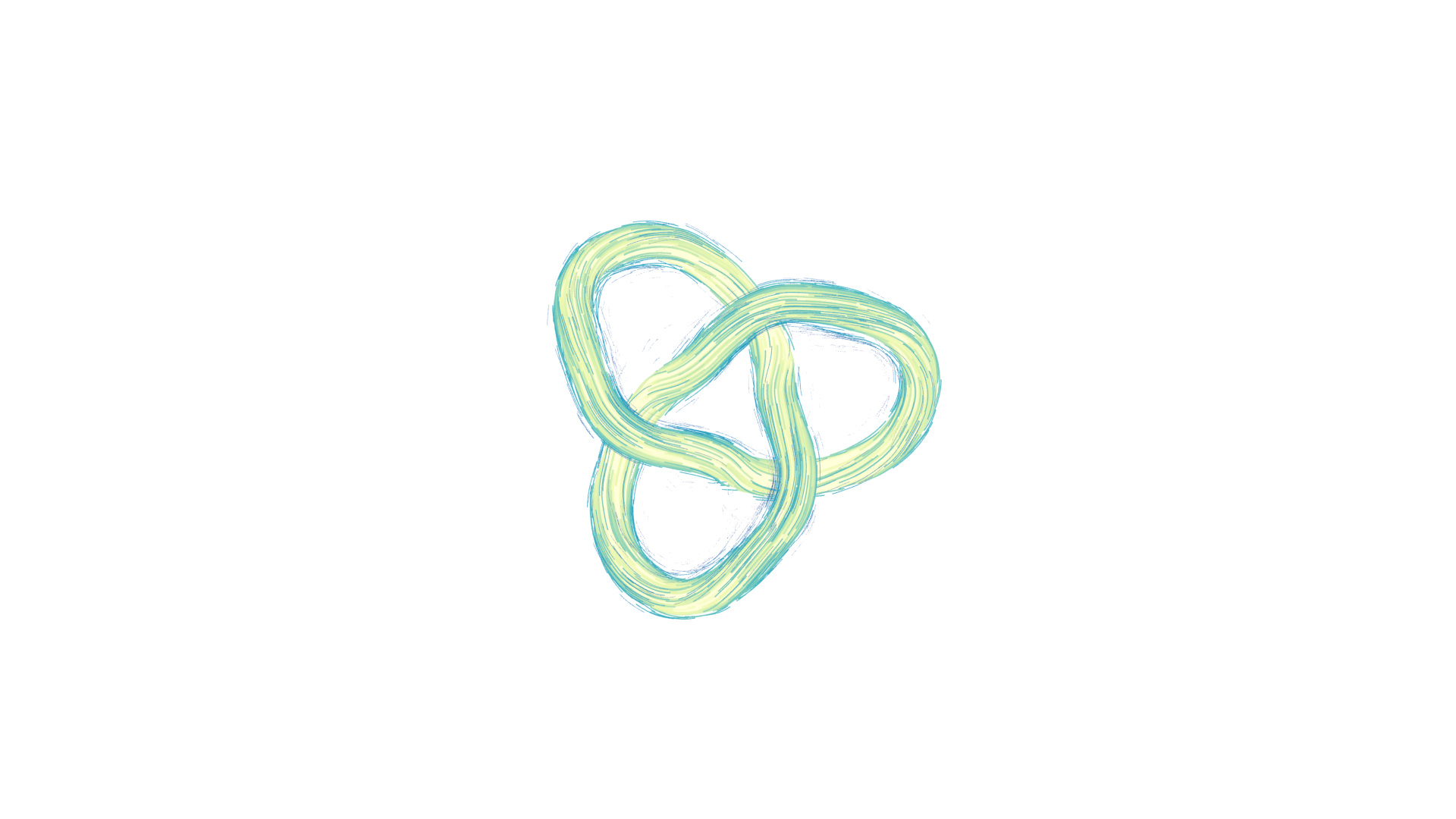}
    \includegraphics[trim={550px 210px 550px 90px},clip,width=0.16\linewidth]{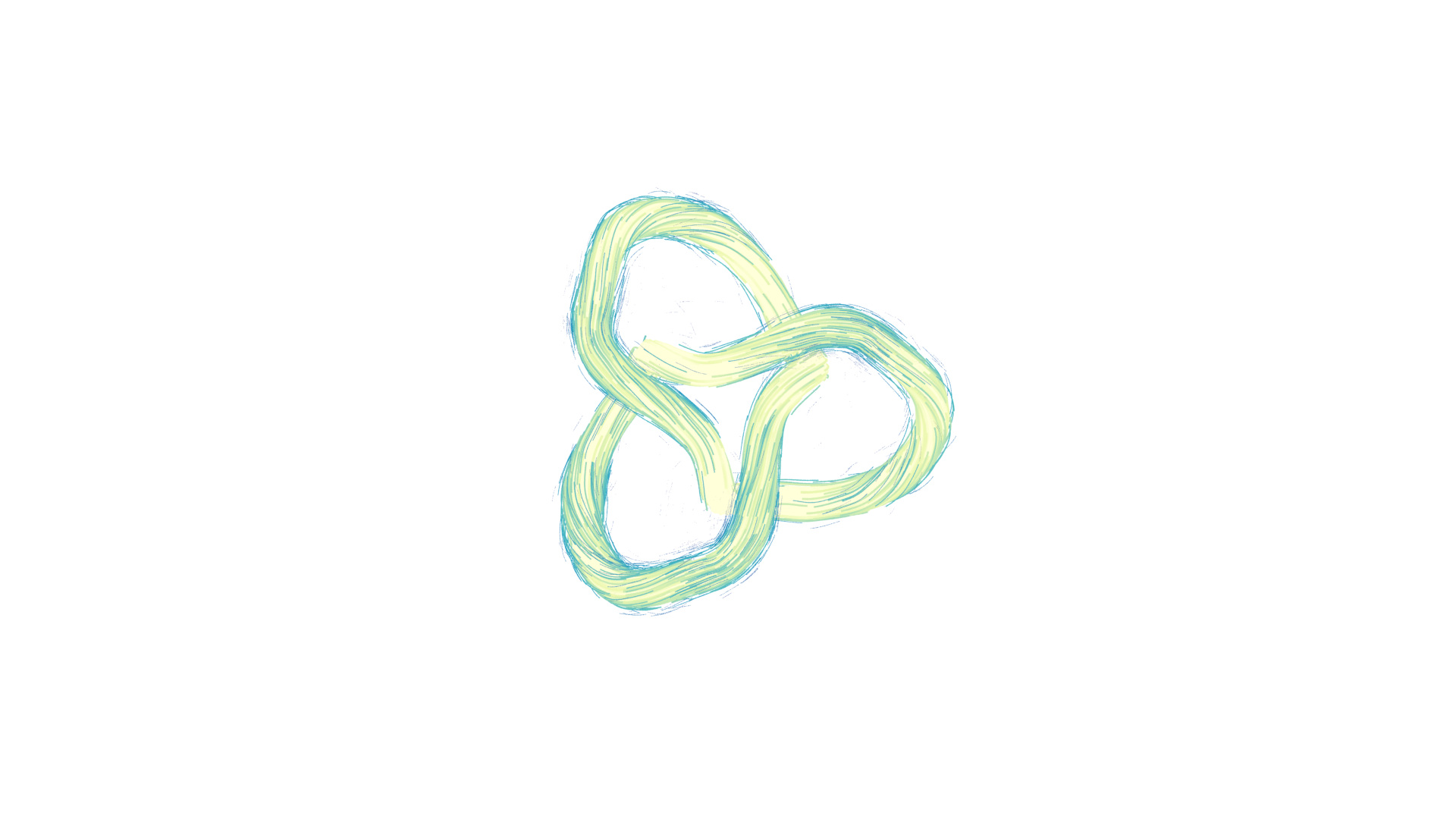}
    \includegraphics[trim={550px 210px 550px 90px},clip,width=0.16\linewidth]{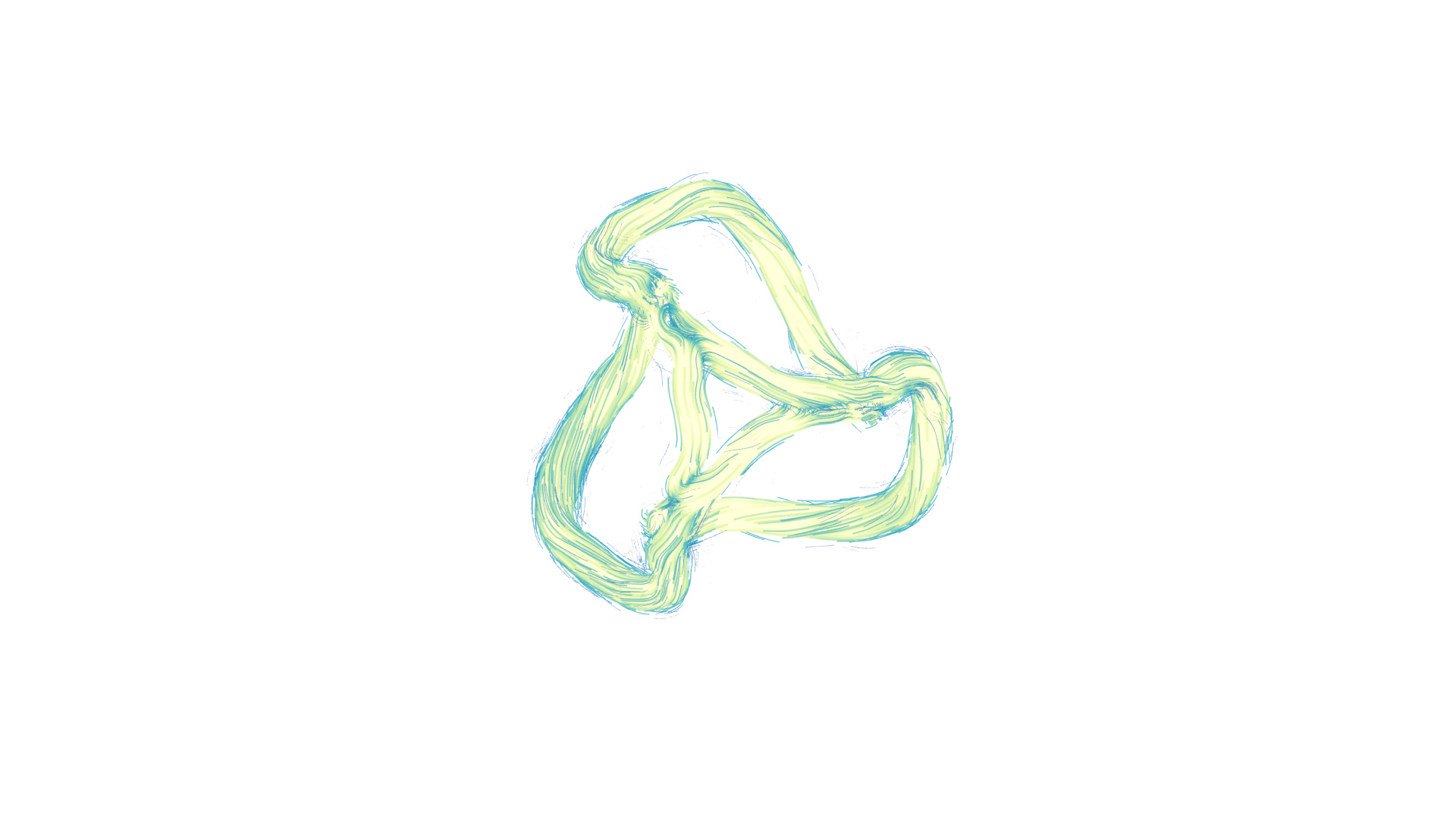}
    \includegraphics[trim={550px 210px 550px 90px},clip,width=0.16\linewidth]{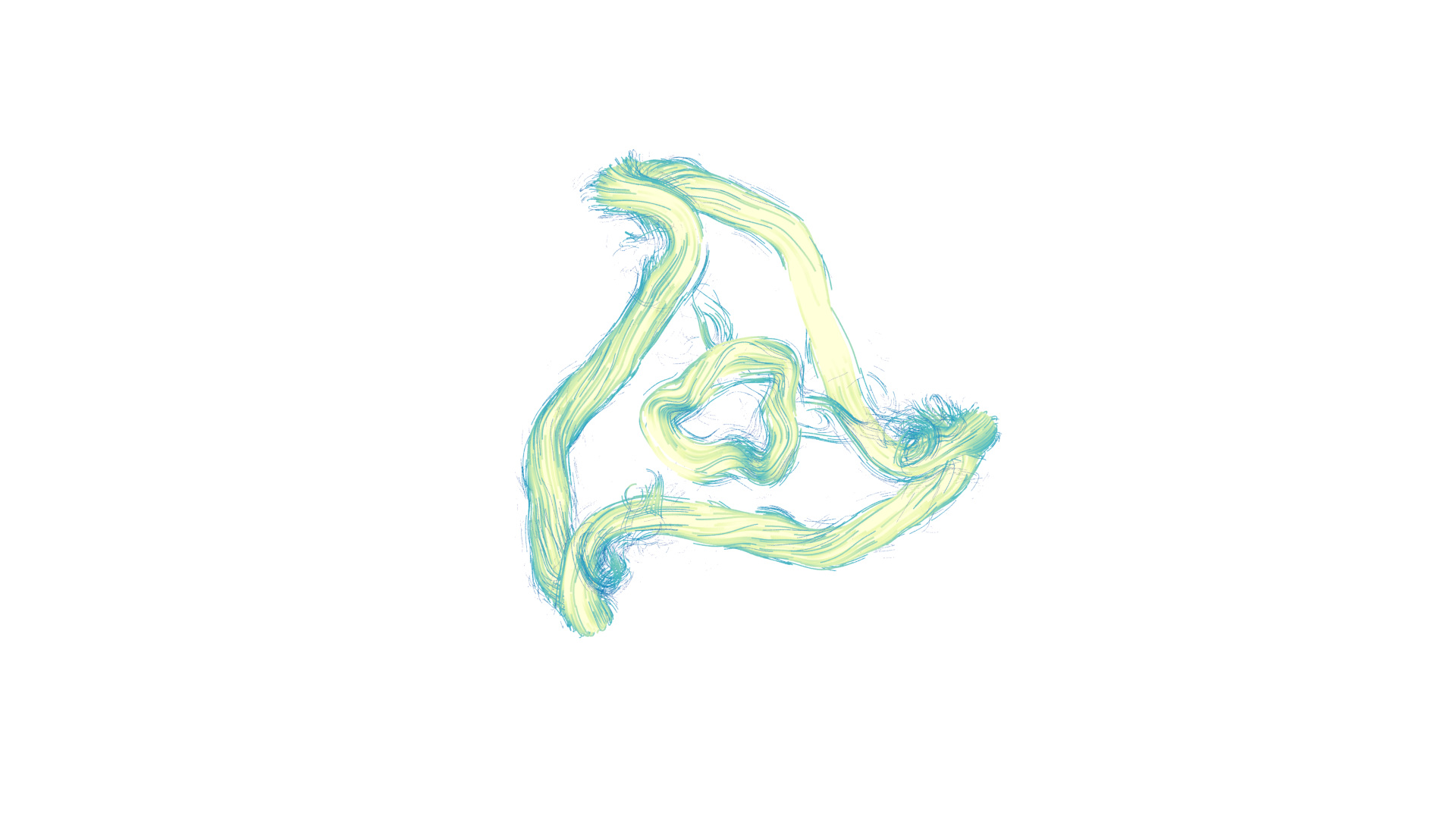}
    \includegraphics[trim={550px 210px 550px 90px},clip,width=0.16\linewidth]{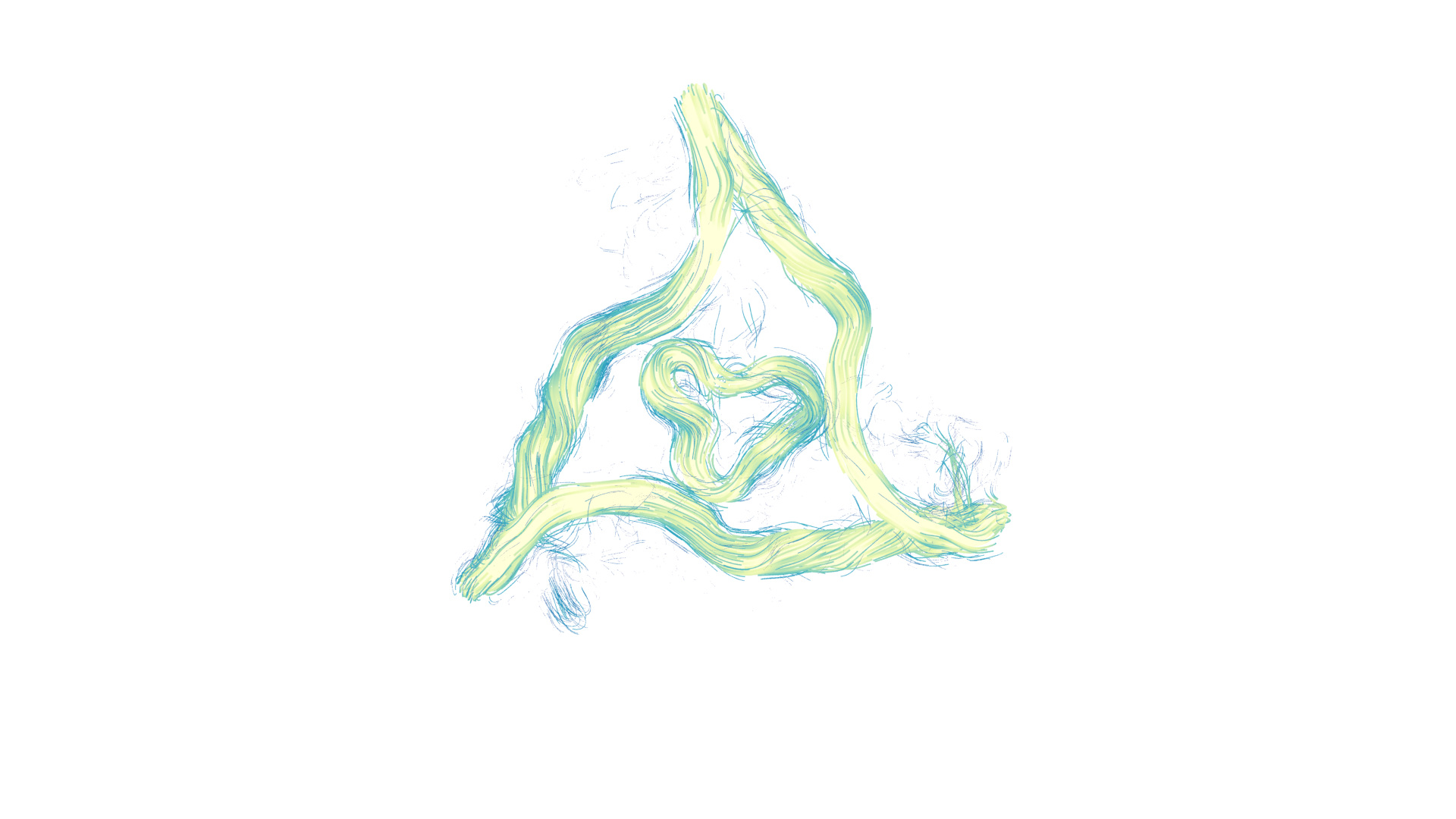}
    \\
    \begin{picture}(0,0)(0,0)
        \put(-265,67){\begin{tikzpicture}
\pgfplotscolorbardrawstandalone[ 
    colormap={myProteinColor}{
        rgb255=(8, 29, 88)
        rgb255=(37, 52, 148)
        rgb255=(34, 94, 168)
        rgb255=(29, 145, 192)
        rgb255=(65, 182, 196)
        rgb255=(127, 205, 187)
        rgb255=(199, 233, 180)
        rgb255=(237, 248, 177)
        rgb255=(255,255, 217)
    },
    point meta min=-1,
    point meta max=1,
    colorbar style={
        width=4pt,
        height=35pt,
        ytick={-1,-0.5,0,0.5,1},
        ytick style={draw=none},
        yticklabels={{0},{},{},{},{10}},
        yticklabel style={font=\scriptsize, xshift=-0.5ex},
        ylabel={\sffamily vorticity norm ($\nicefrac{1}{\text{s}}$)},
        ylabel style={font=\tiny, yshift=25pt}
        }]
\end{tikzpicture}}
        \put(-28,98){\sffamily\scriptsize (a) Side view.}
        \put(-29,0){\sffamily\scriptsize (b) Head-on view.}
        \put(-220,10){\sffamily \scriptsize Frame 0}
        \put(-140,10){\sffamily \scriptsize Frame 78}
        \put(-50,10){\sffamily \scriptsize Frame 156}
        \put(25,10){\sffamily \scriptsize Frame 234}
        \put(115,10){\sffamily \scriptsize Frame 312}
        \put(210,10){\sffamily \scriptsize Frame 390}
    \end{picture}
    \caption{Trefoil knot, also known as the (2,3)-torus knot.
    As time passes, the knot moves forward and stretches until the vortex reconnection even when it breaks into one big and one smaller part \cite{Kleckner:2013:CDK}.
    Our captures this behavior at a low resolution of $64\times64\times64$.}
    \label{fig:trefoilknot_evolution}
\end{figure*}

\subsubsection{2D and 3D Casimirs under refinement}
As discussed in \secref{sec:CasimirMeasurement}, we use a set of auxiliary measurement grids.
We progressively increase the resolution of the measurement grid while running the experiment based on a fixed-resolution simulation grid.
The Hamiltonian is measured on the simulation grid.

In 2D, we initialize the velocity field analytically with the Taylor-Green velocity field in an $8\times8$ simulation grid.
We measure enstrophy, which is the Casimir related to the second moment of vorticity.
As seen in \figref{fig:casimir_refinement}, enstrophy is preserved progressively better under finer measurement grids.
Additionally, \figref{fig:casimir_refinement} demonstrates that even without the curl consistent interpolation, i.e. using the traditional FLIP pressure force, the Casimirs are only changed by a small amount.

In 3D, we use an ABC flow velocity field pressure projected to a closed box.
Note that, after pressure projection the velocity field no longer exactly follows the ABC analytical equation, but provides a close-to-ABC flow; which is smooth velocity field.
Similar to the 2D experiment, \figref{fig:casimir_refinement} shows that helicity is preserved under fine enough measurements.
Note that the coarse measurement shows a larger error, and should only be taken as an approximate measurement for the helicity Casimir.

\subsubsection{2D Vortex leapfrogging vortices}
For the 2D vortex leapfrogging experiment, the vortices are expected to run indefinitely as they leap through one another.
Inspecting the final frame in \figref{fig:leapfrog_marathon}, we see that our method (CO-FLIP) maintains the core vortical structure of the vortices despite 500 seconds of simulation, while other methods have drastically lost their energy and vorticity.
As seen in \figref{fig:leapfrog_plots}, our method conserves the energy close to floating point accuracy, while all previous methods drastically lose their energy over time.
From this plot, we also see that the Casimirs for our method only change slightly.
For the second moment of vorticity $\mathcal{W}^2$, our method has a relative error of $2\%$, which is drastically lower than other methods (the best one being at $60\%$ error).
Similarly, for the fourth moment of vorticity $\mathcal{W}^4$, our method has a relative error of 17\% compared to 97\% (at best) for other methods.

\begin{figure}
    \centering
    \includegraphics[trim={550px 210px 550px 90px},clip,width=0.32\columnwidth]{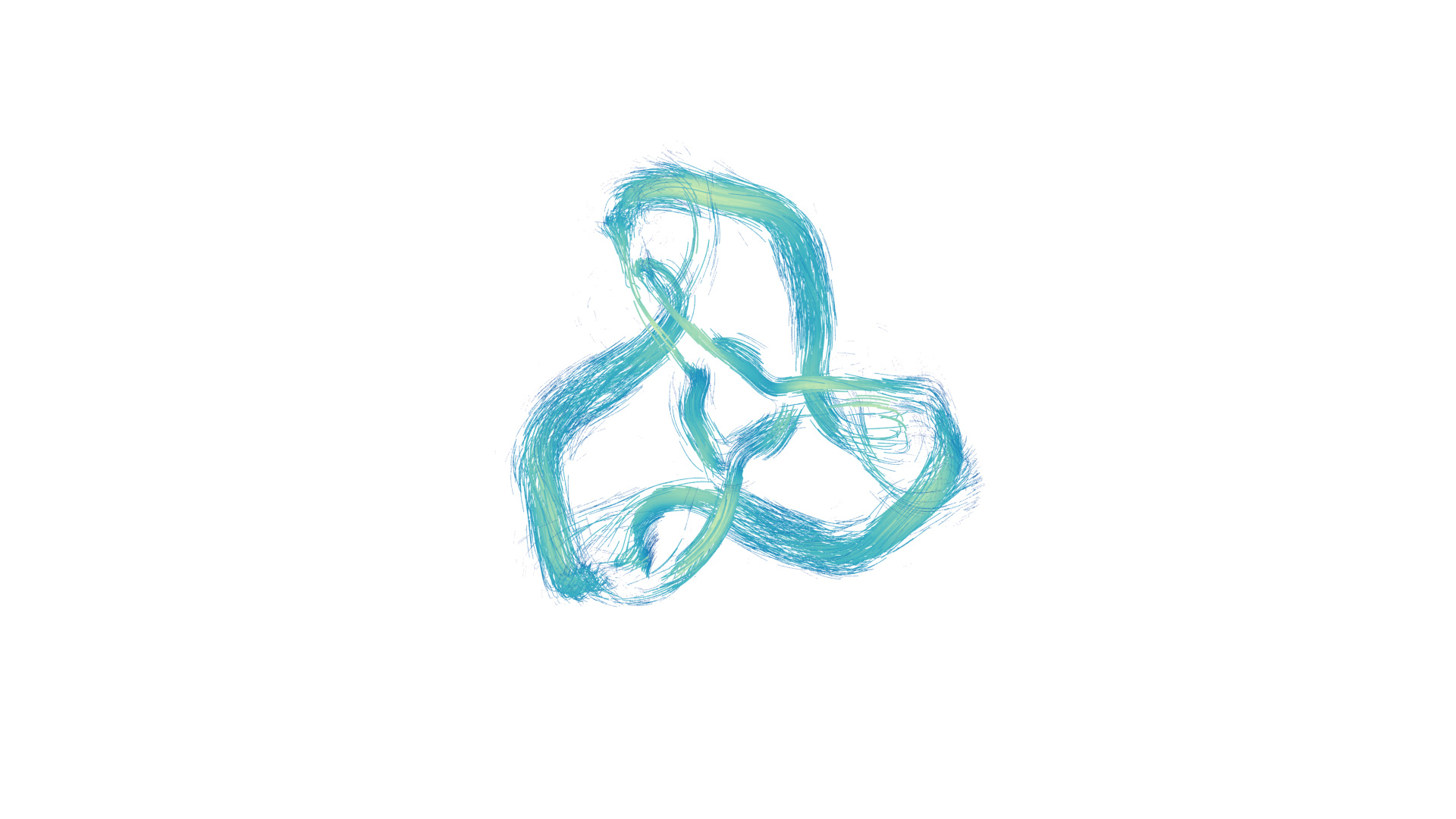}
    \includegraphics[trim={550px 210px 550px 90px},clip,width=0.32\columnwidth]{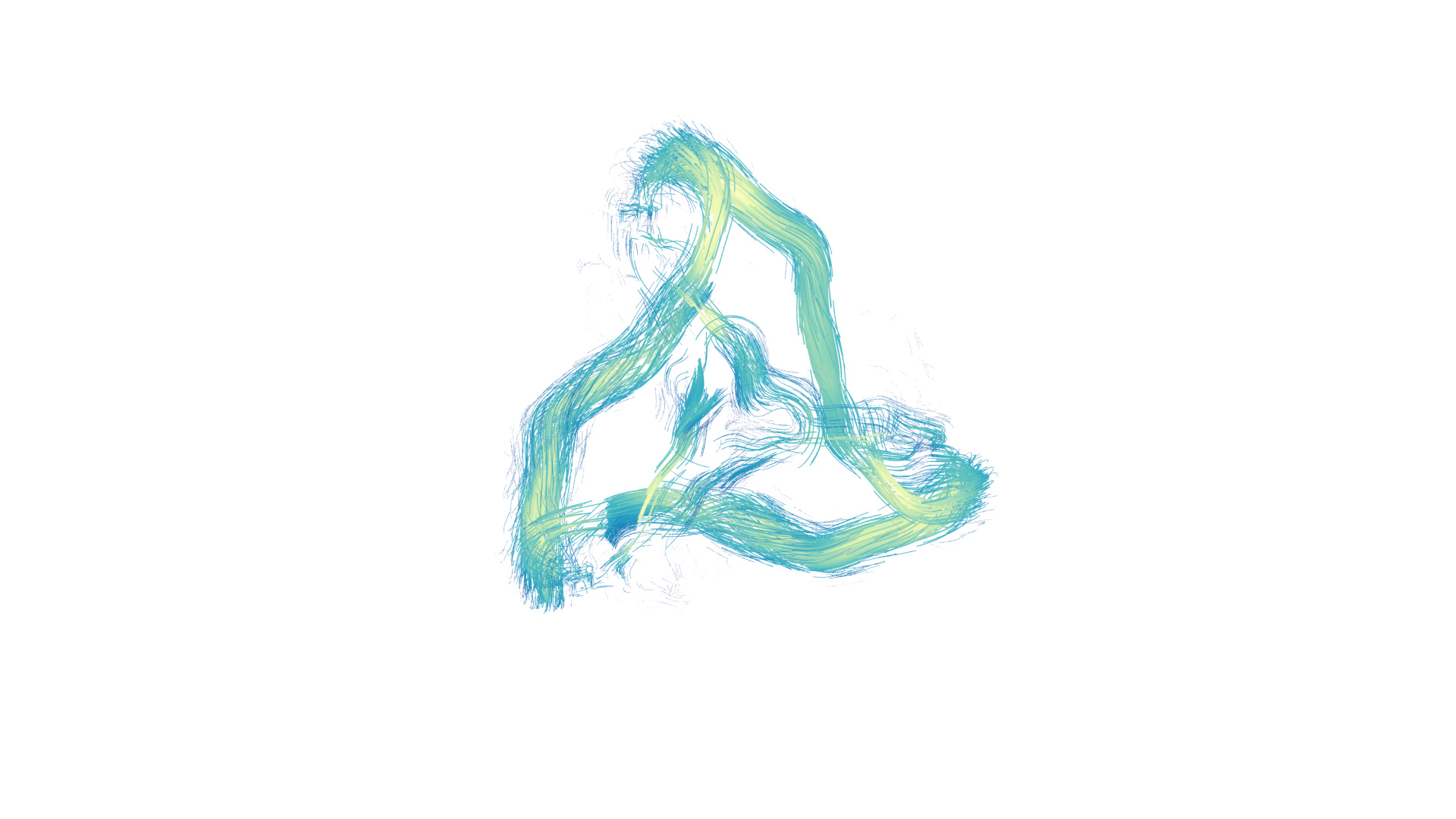}
    \includegraphics[trim={550px 210px 550px 90px},clip,width=0.32\columnwidth]{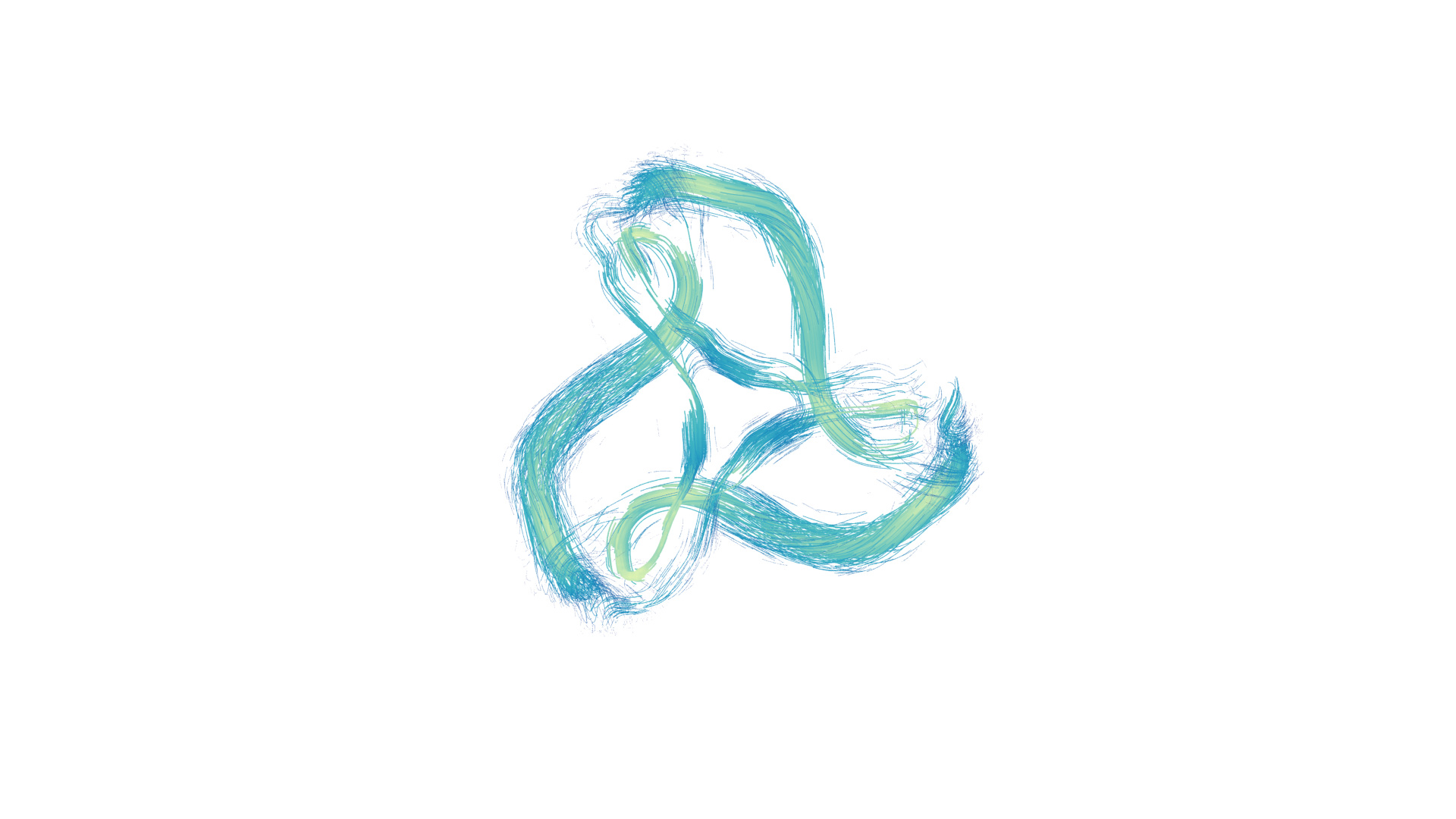}
    \includegraphics[trim={550px 210px 550px 90px},clip,width=0.32\columnwidth]{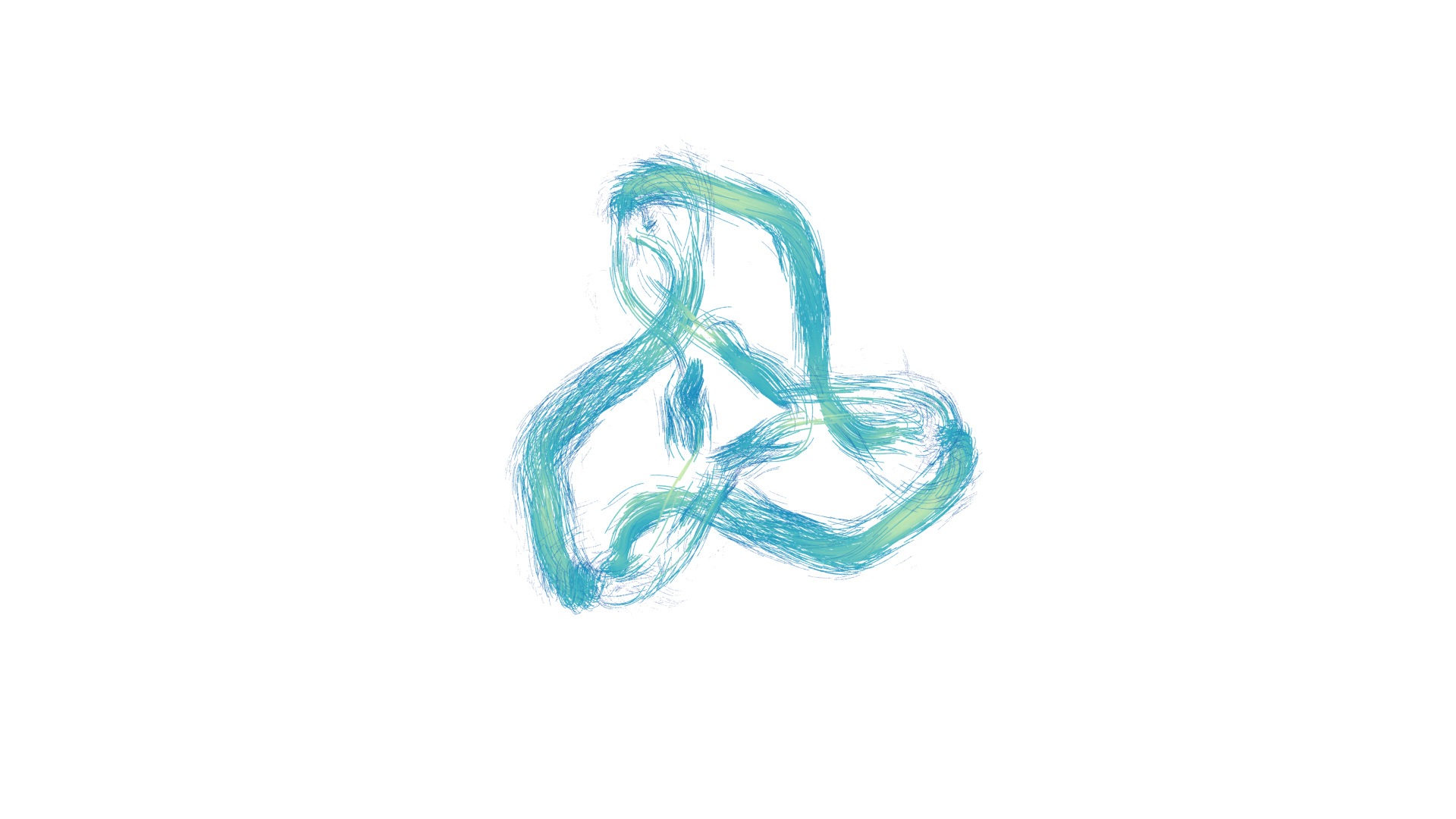}
    \includegraphics[trim={550px 210px 550px 90px},clip,width=0.32\columnwidth]{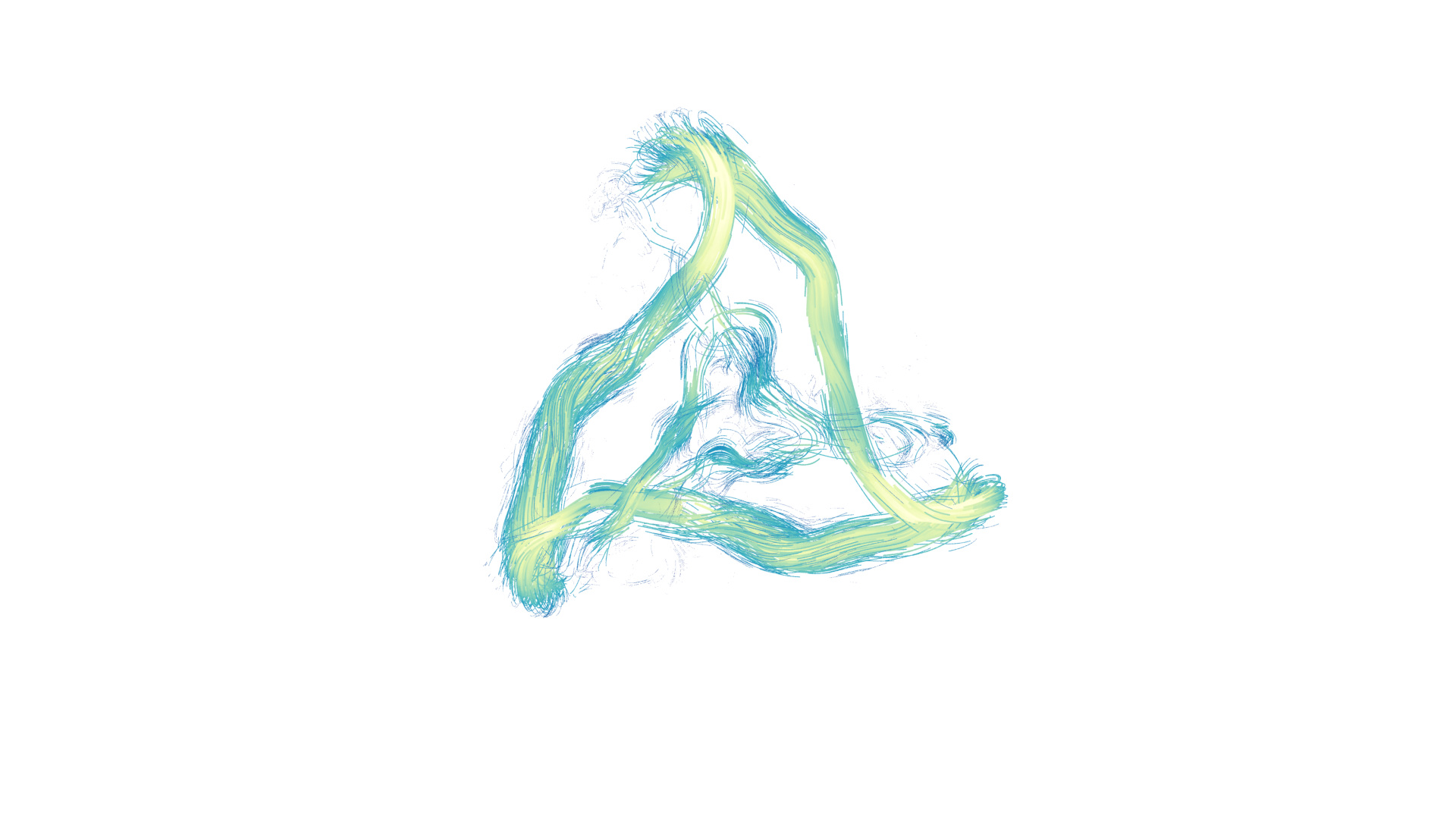}
    \includegraphics[trim={550px 210px 550px 90px},clip,width=0.32\columnwidth]{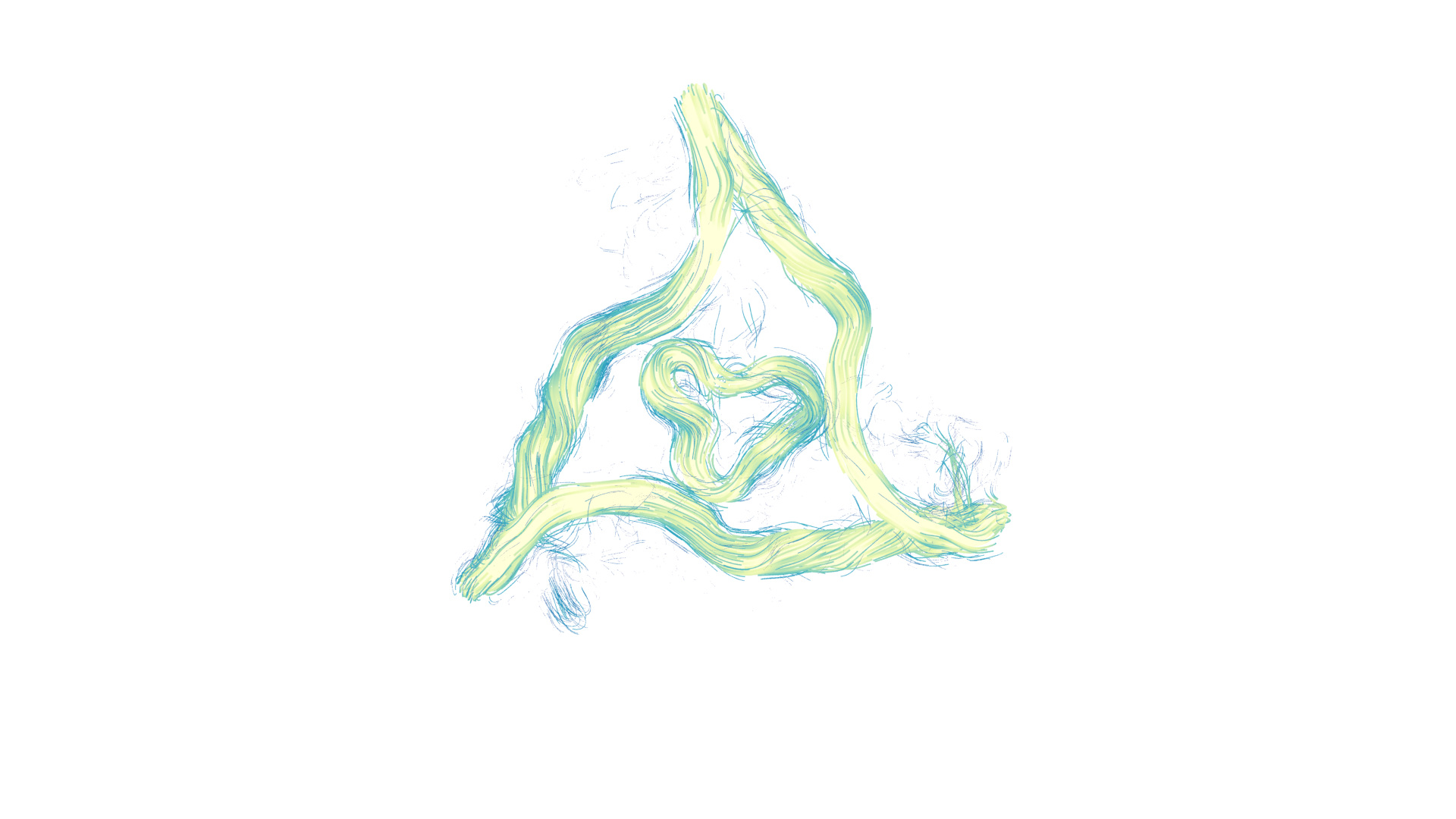}
    \\
    \begin{picture}(0,0)(0,0)
        \put(-130,55){\begin{tikzpicture}
\pgfplotscolorbardrawstandalone[ 
    colormap={myProteinColor}{
        rgb255=(8, 29, 88)
        rgb255=(37, 52, 148)
        rgb255=(34, 94, 168)
        rgb255=(29, 145, 192)
        rgb255=(65, 182, 196)
        rgb255=(127, 205, 187)
        rgb255=(199, 233, 180)
        rgb255=(237, 248, 177)
        rgb255=(255,255, 217)
    },
    point meta min=-1,
    point meta max=1,
    colorbar style={
        width=4pt,
        height=35pt,
        ytick={-1,-0.5,0,0.5,1},
        ytick style={draw=none},
        yticklabels={{0},{},{},{},{10}},
        yticklabel style={font=\scriptsize, xshift=-0.5ex},
        ylabel={\sffamily vorticity norm ($\nicefrac{1}{\text{s}}$)},
        ylabel style={font=\tiny, yshift=25pt}
        }]
\end{tikzpicture}}
        \put(-85,85){\sffamily \scriptsize PolyPIC}
        \put(-10,85){\sffamily \scriptsize CF+PolyFLIP}
        \put(73,85){\sffamily \scriptsize NFM}
        \put(-85,10){\sffamily \scriptsize PolyFLIP}
        \put(-10,10){\sffamily \scriptsize R+PolyFLIP}
        \put(66,10){\sffamily \scriptsize \textbf{CO-FLIP (Ours)}}
    \end{picture}
    \caption{Comparisons of trefoil knot between different methods.
    Note that our method shows a superior preservation of energy and vortical structures, while successfully showcasing the vortex reconnection event.}
    \label{fig:trefoilknot}
\end{figure}

\begin{figure*}
    \centering
    \scalebox{-1}[1]{\includegraphics[trim={50px 0 0 0},clip,width=0.18\linewidth]{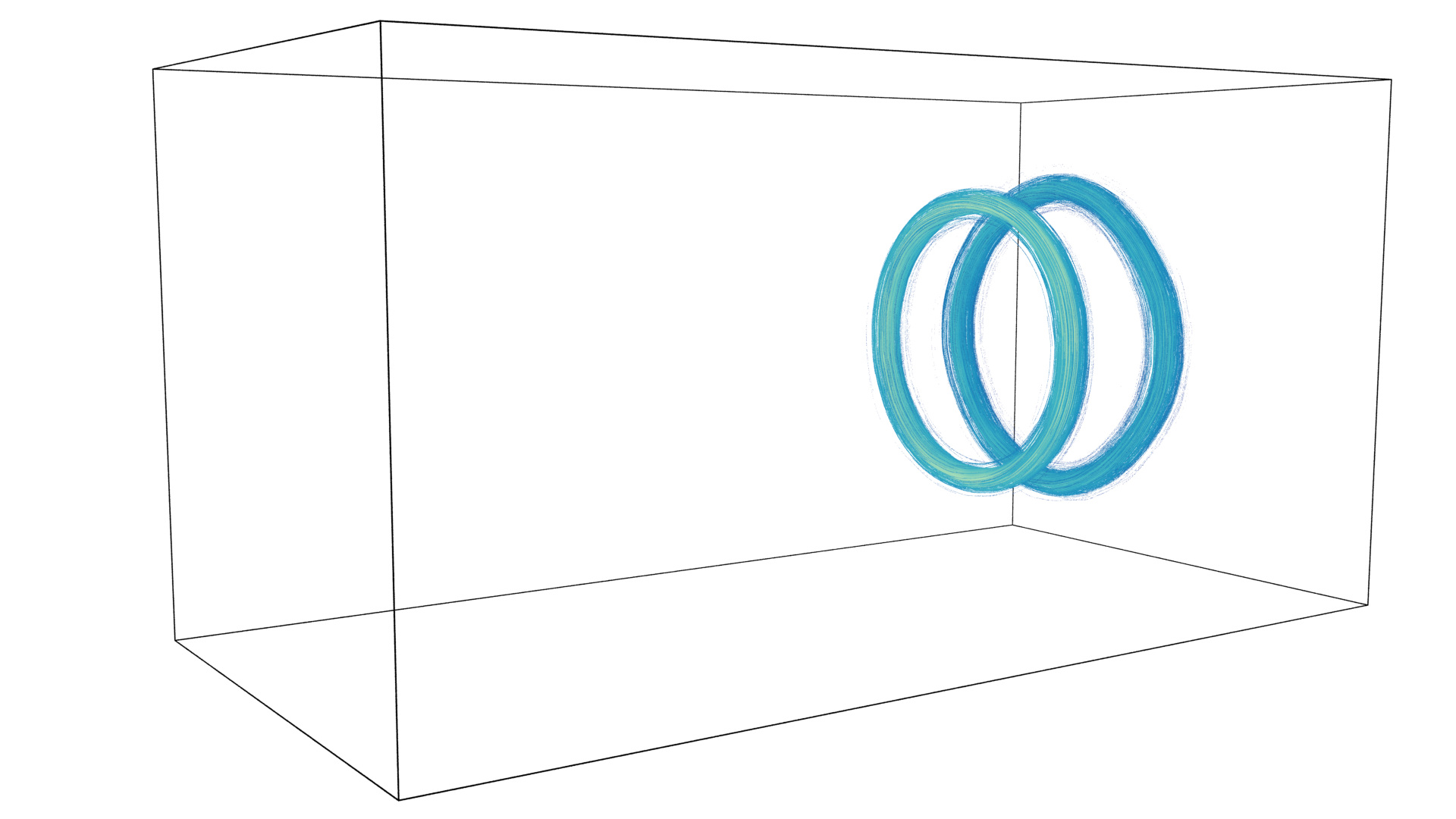}}
    \scalebox{-1}[1]{\includegraphics[trim={50px 0 0 0},clip,width=0.18\linewidth]{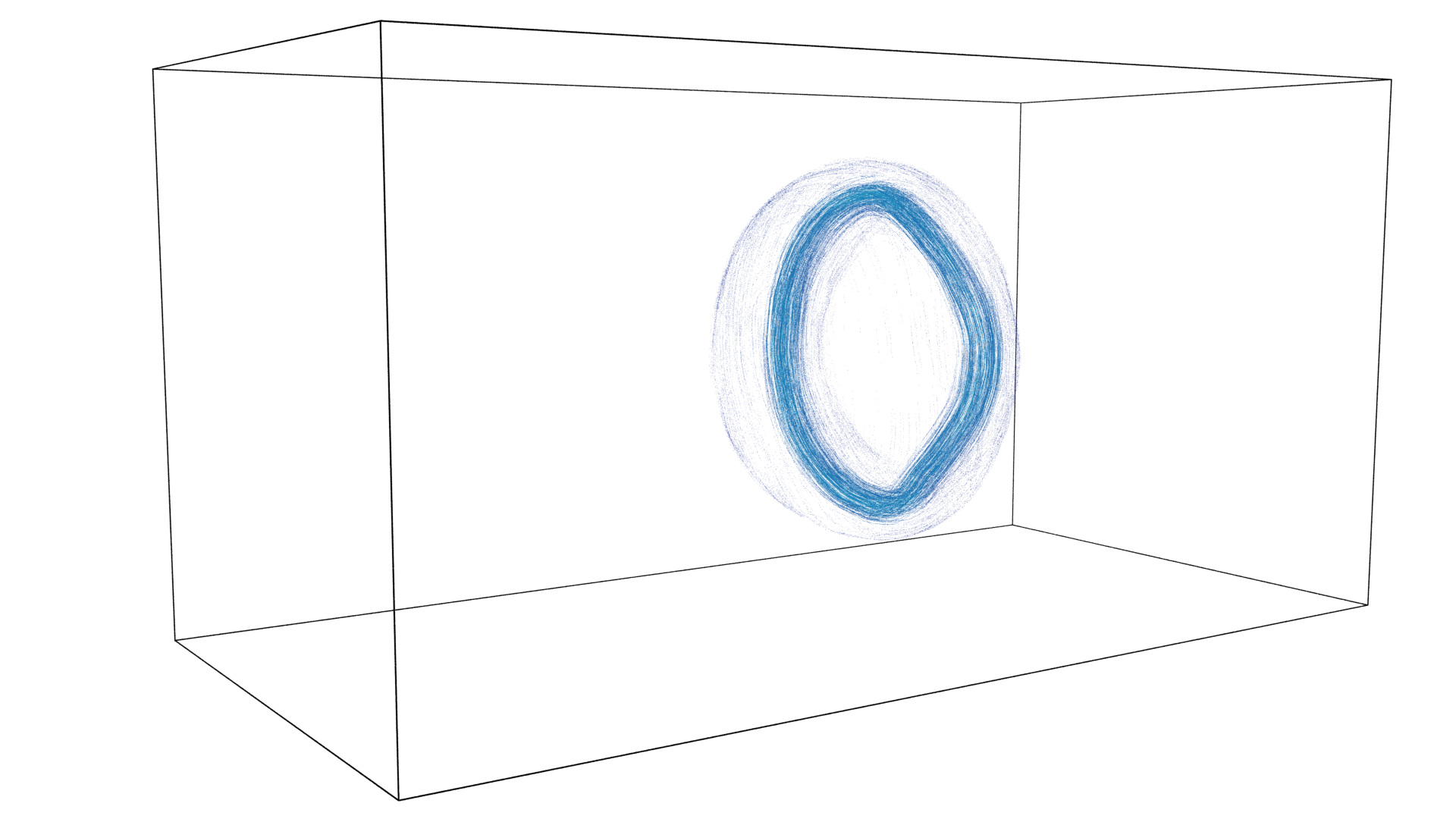}}
    \scalebox{-1}[1]{\includegraphics[trim={50px 0 0 0},clip,width=0.18\linewidth]{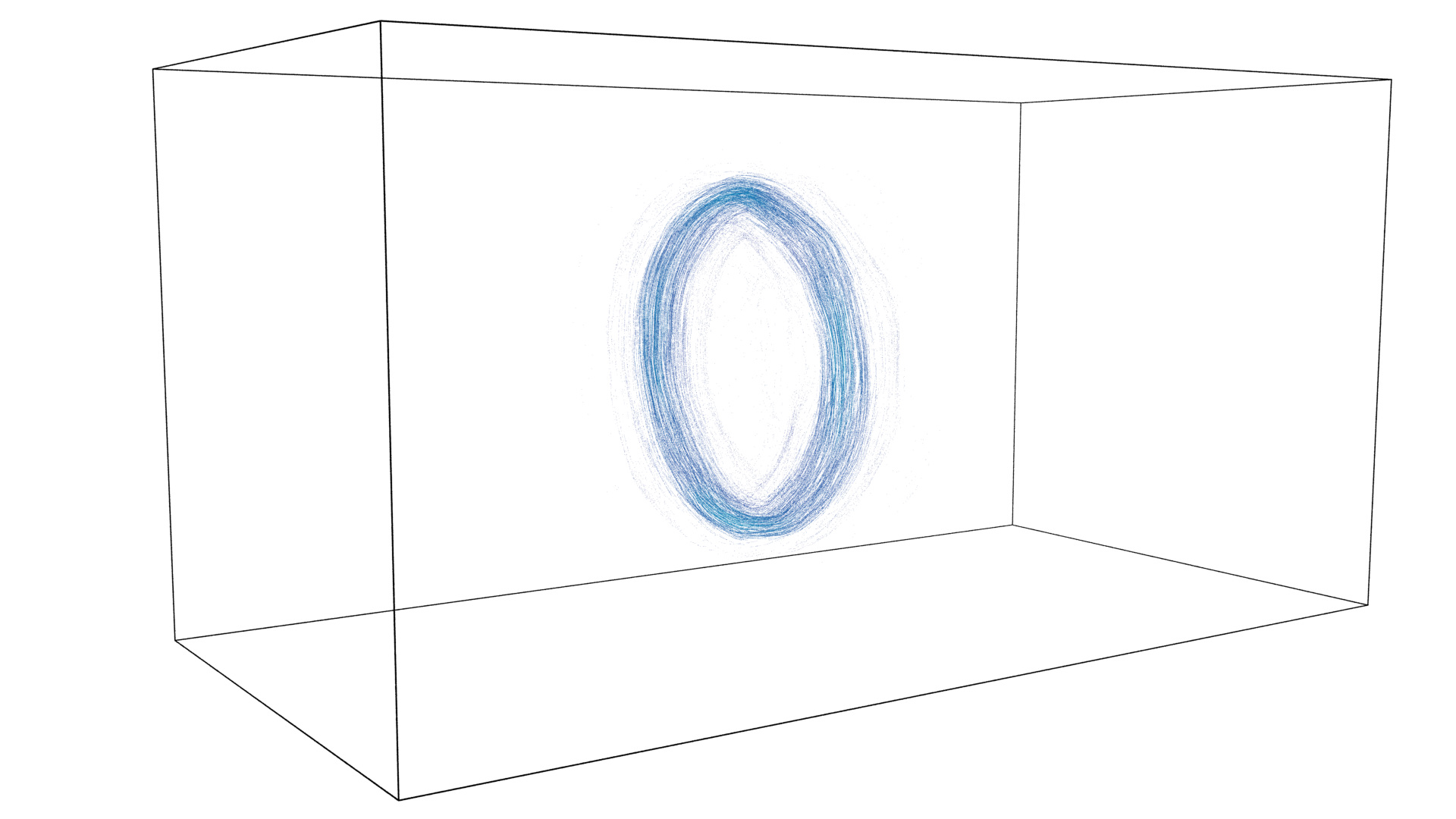}}
    \scalebox{-1}[1]{\includegraphics[trim={50px 0 0 0},clip,width=0.18\linewidth]{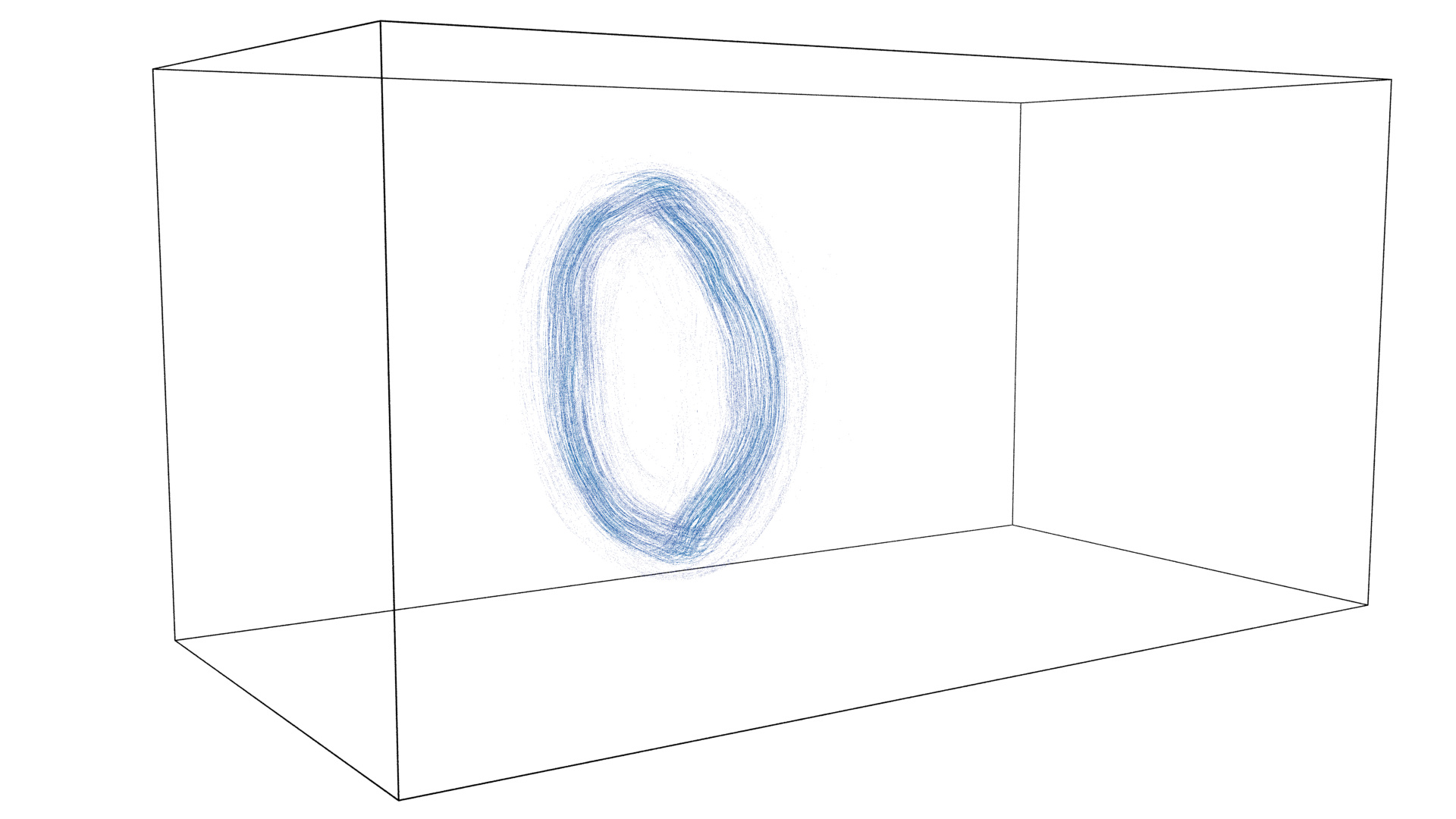}}
    \scalebox{-1}[1]{\includegraphics[trim={50px 0 0 0},clip,width=0.18\linewidth]{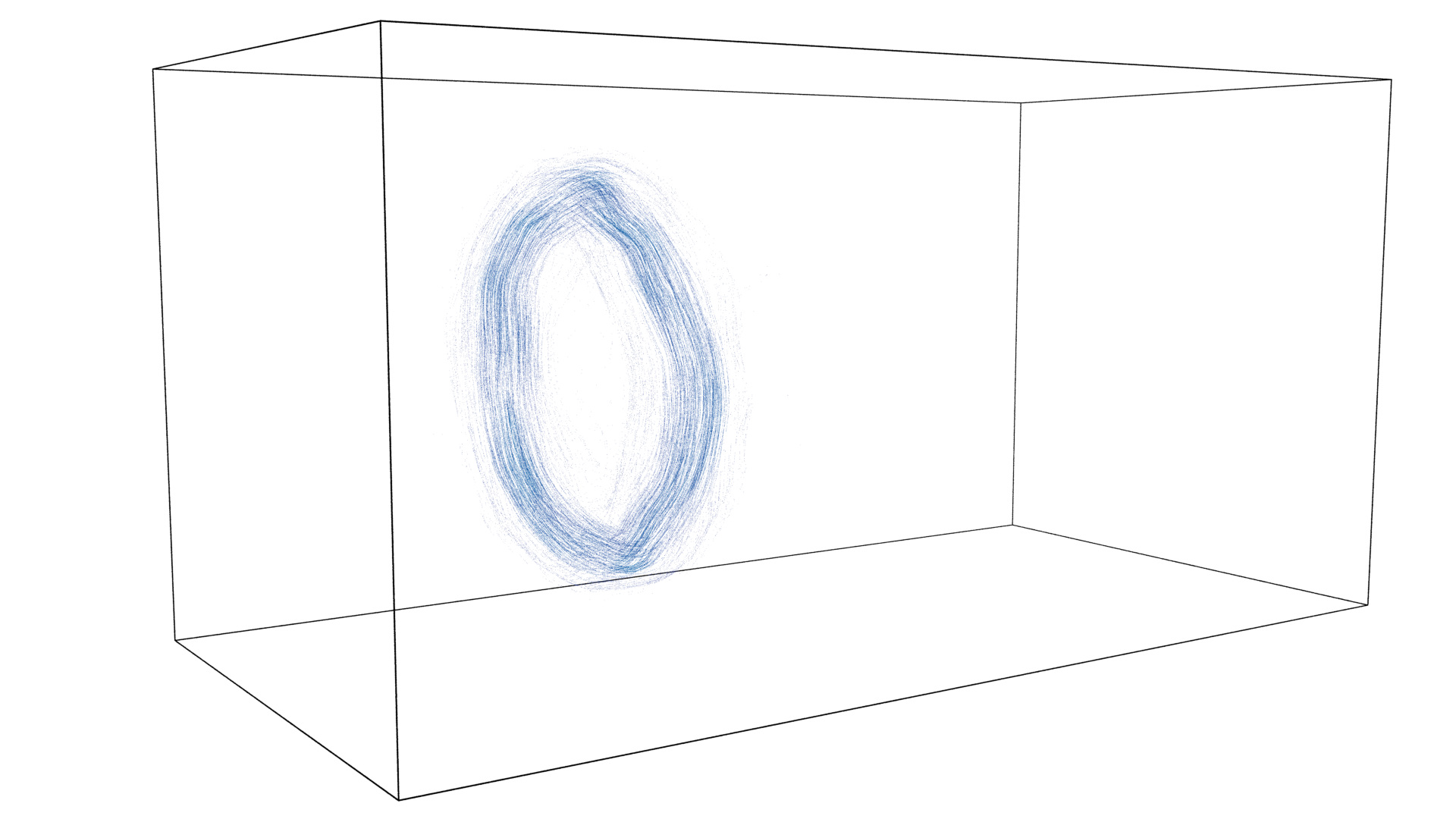}}
    \scalebox{-1}[1]{\includegraphics[trim={50px 0 0 0},clip,width=0.18\linewidth]{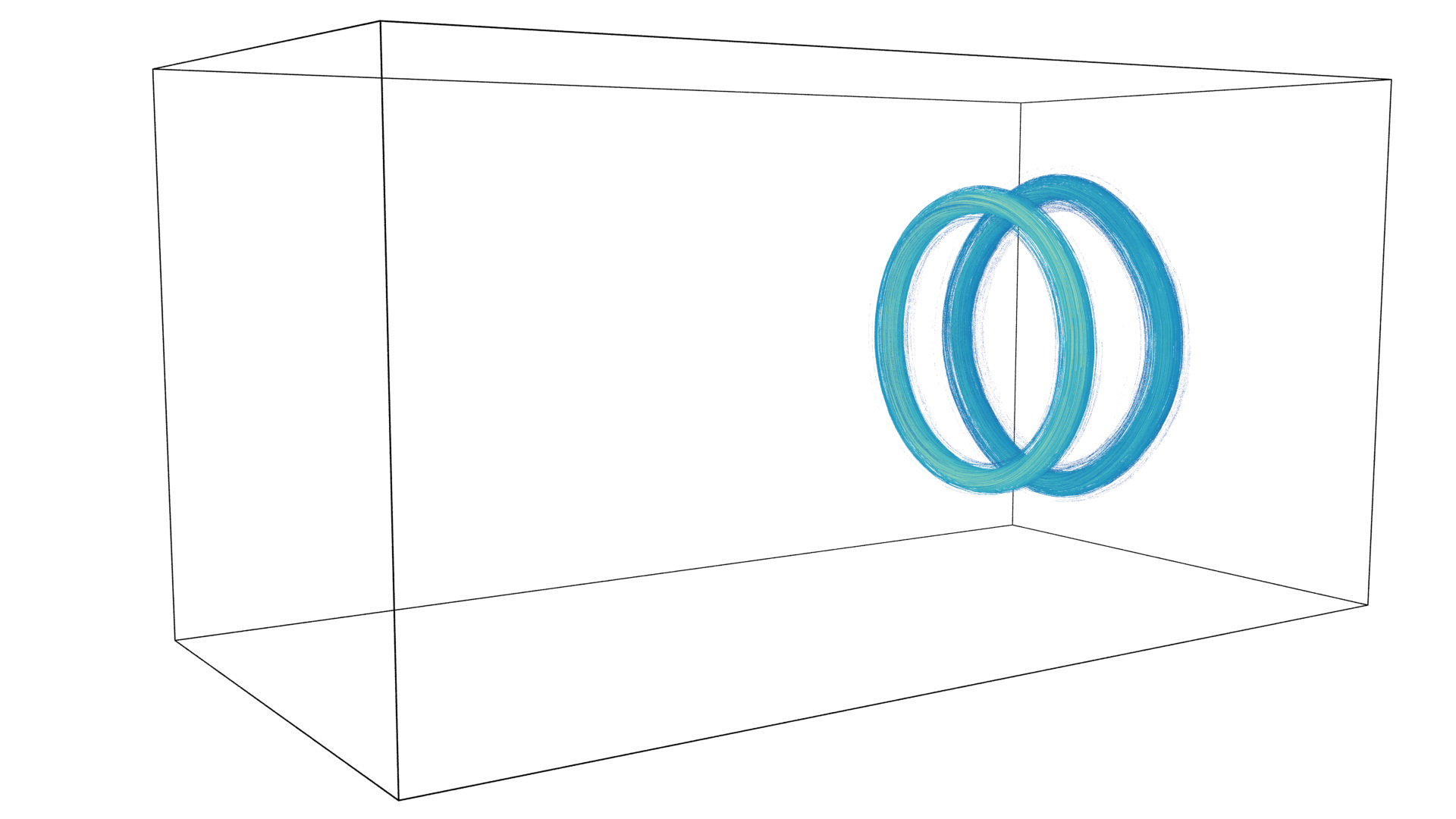}}
    \scalebox{-1}[1]{\includegraphics[trim={50px 0 0 0},clip,width=0.18\linewidth]{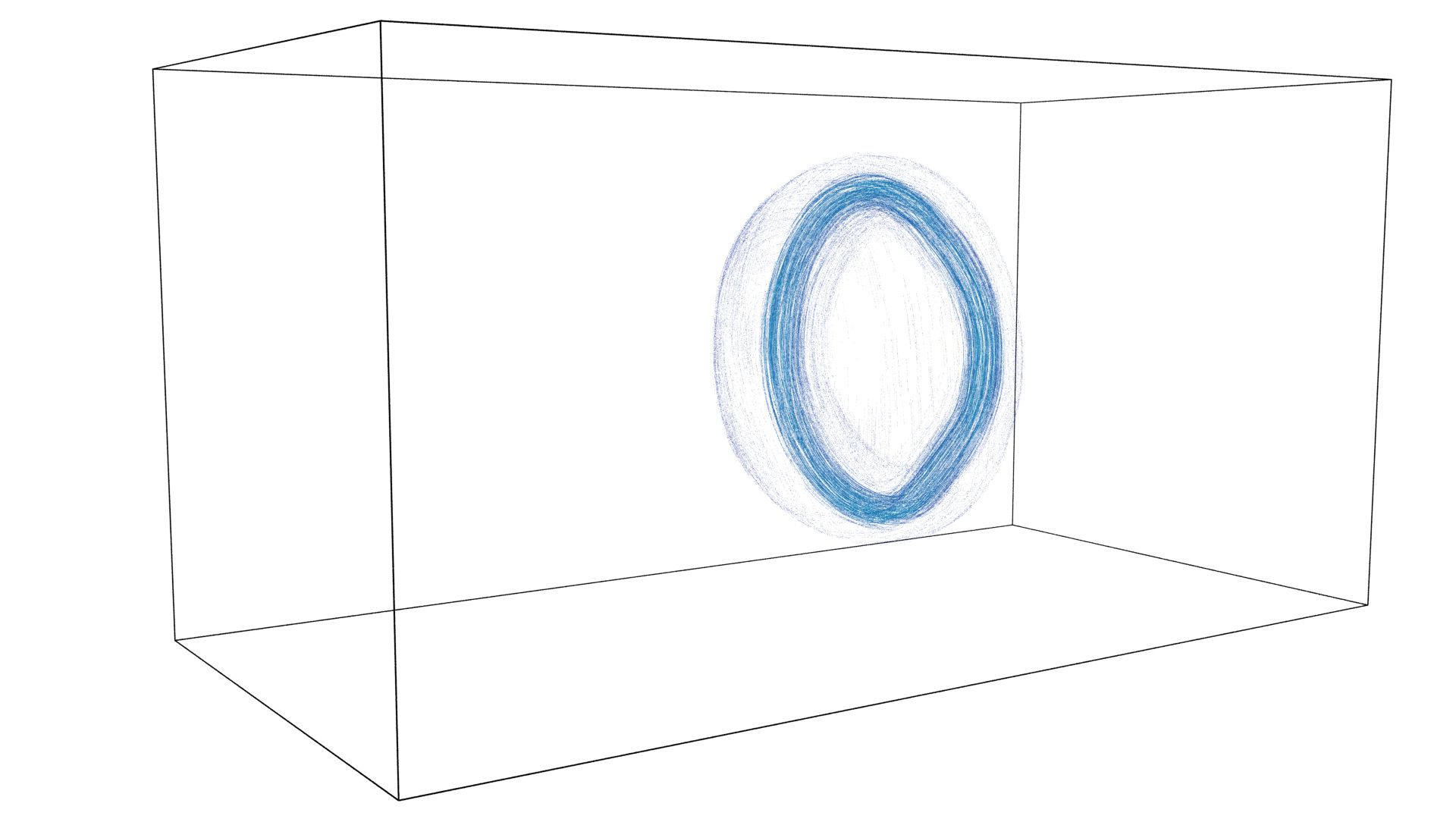}}
    \scalebox{-1}[1]{\includegraphics[trim={50px 0 0 0},clip,width=0.18\linewidth]{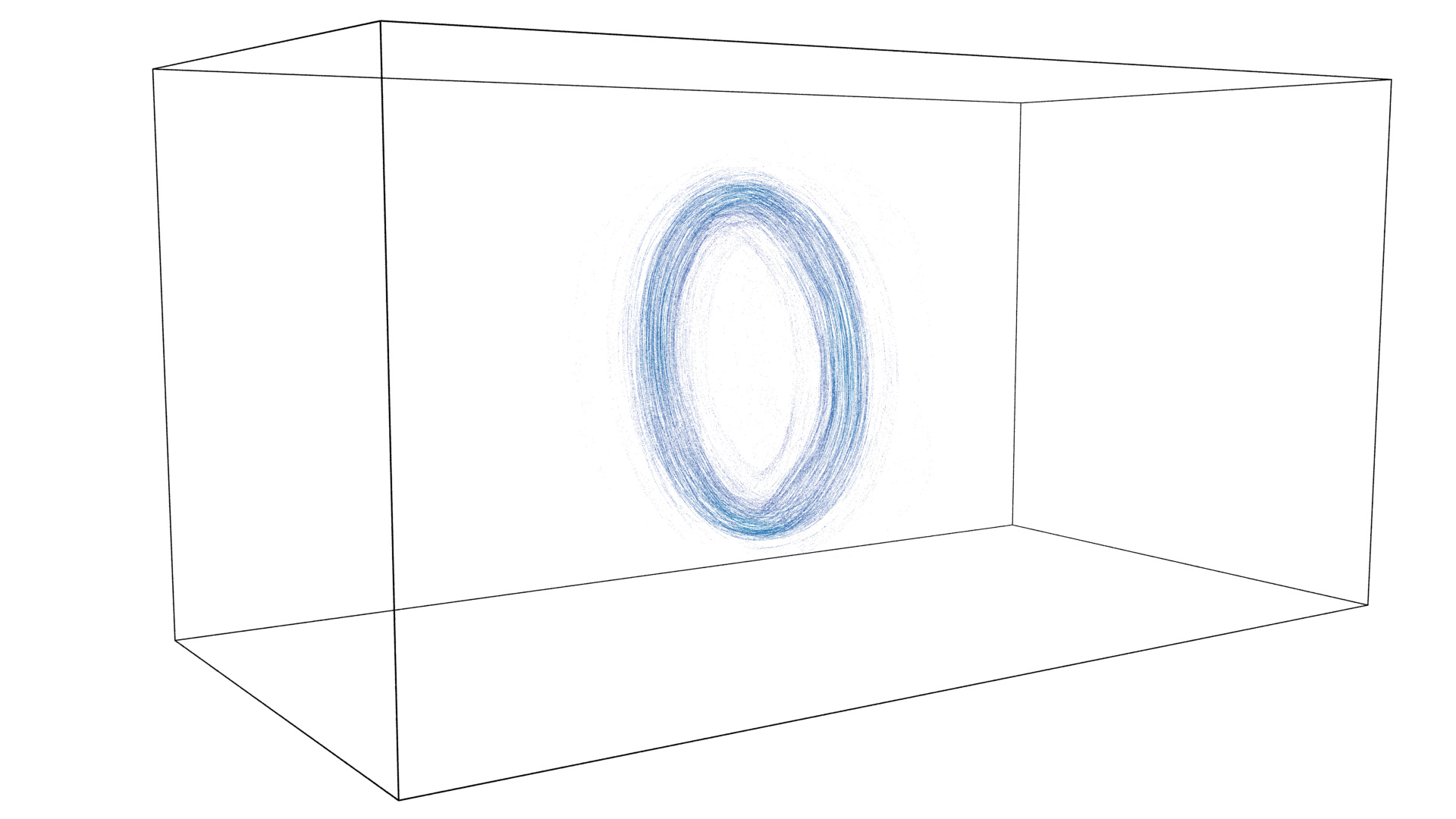}}
    \scalebox{-1}[1]{\includegraphics[trim={50px 0 0 0},clip,width=0.18\linewidth]{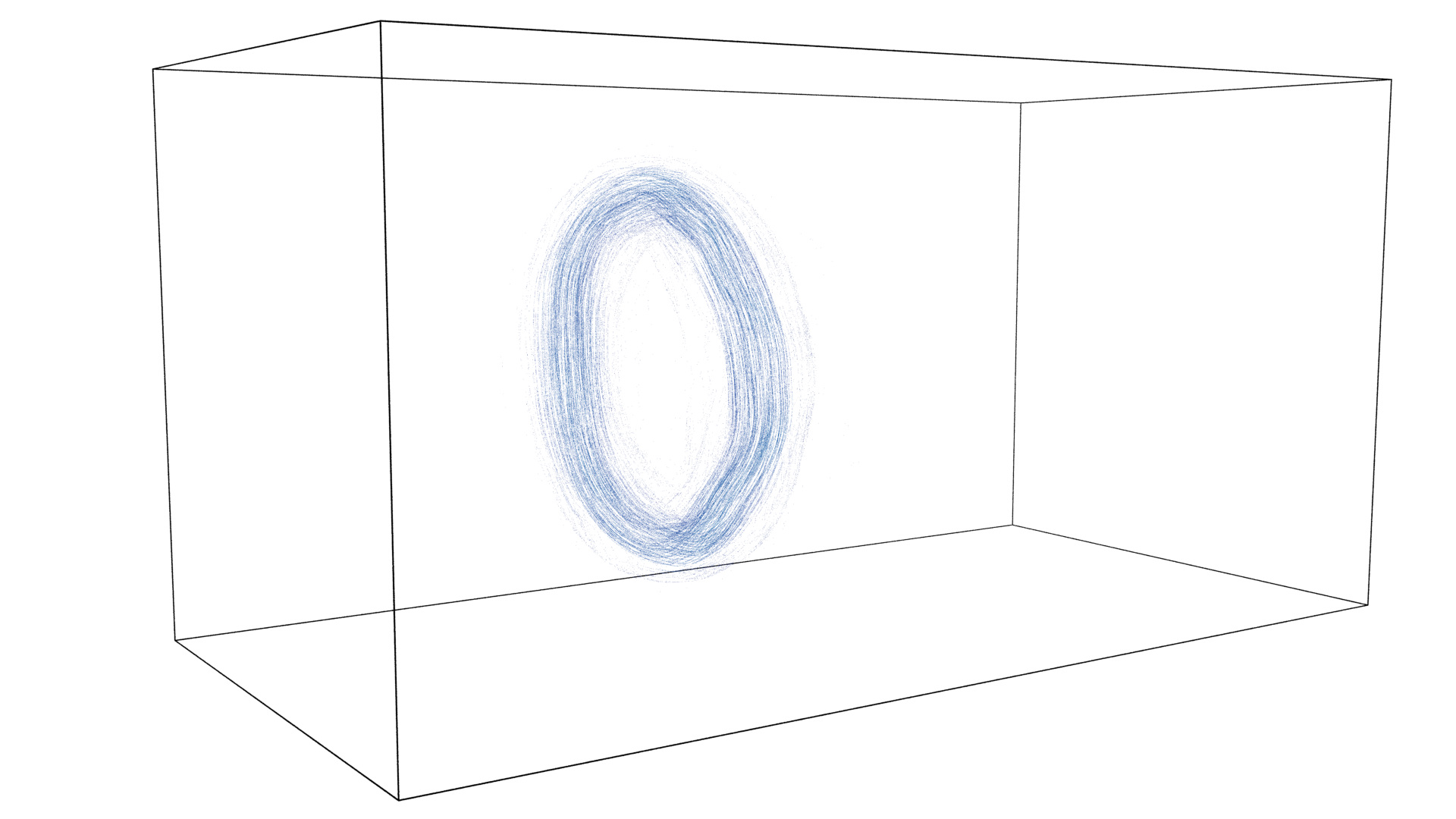}}
    \scalebox{-1}[1]{\includegraphics[trim={50px 0 0 0},clip,width=0.18\linewidth]{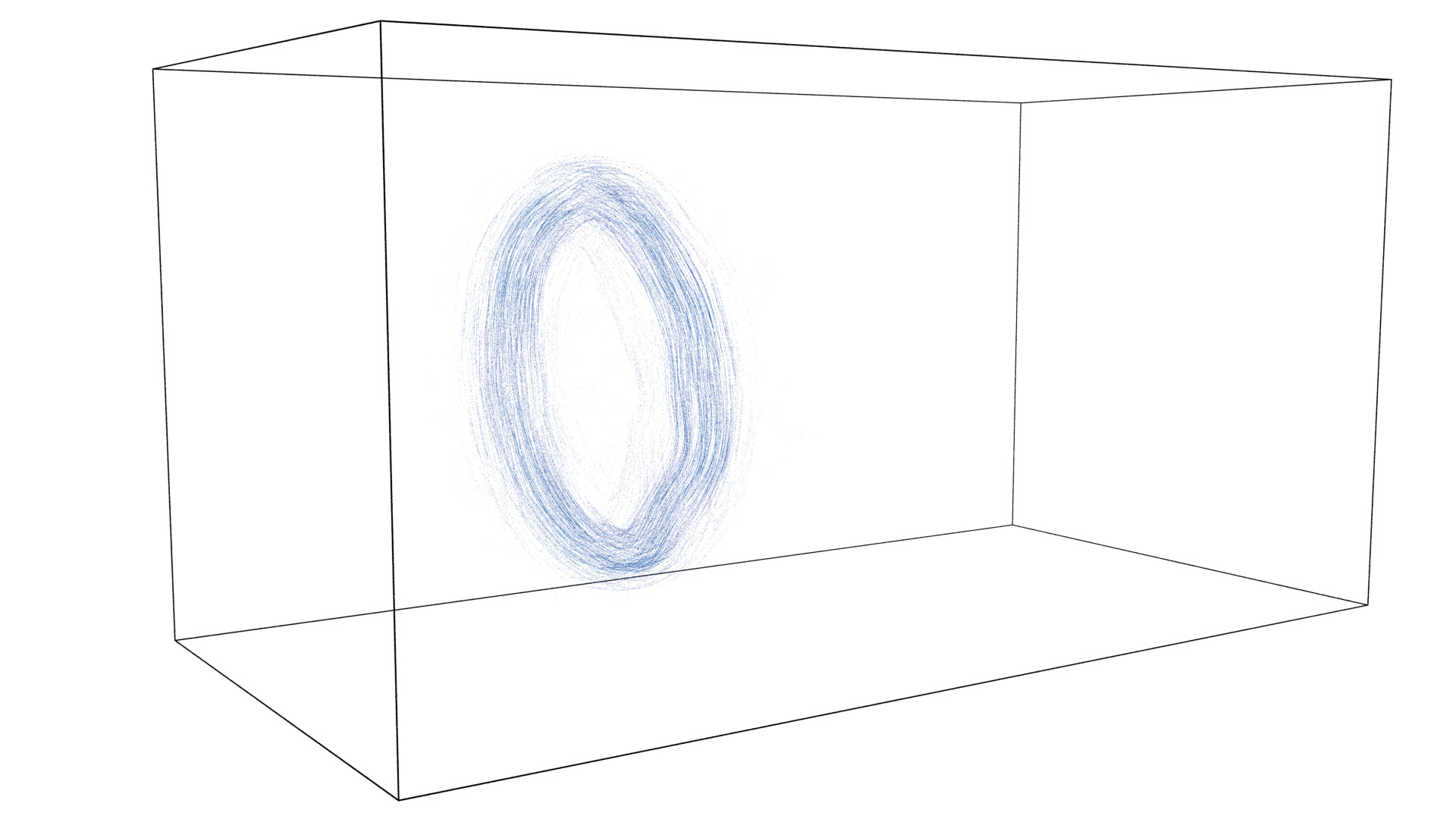}}
    \scalebox{-1}[1]{\includegraphics[trim={50px 0 0 0},clip,width=0.18\linewidth]{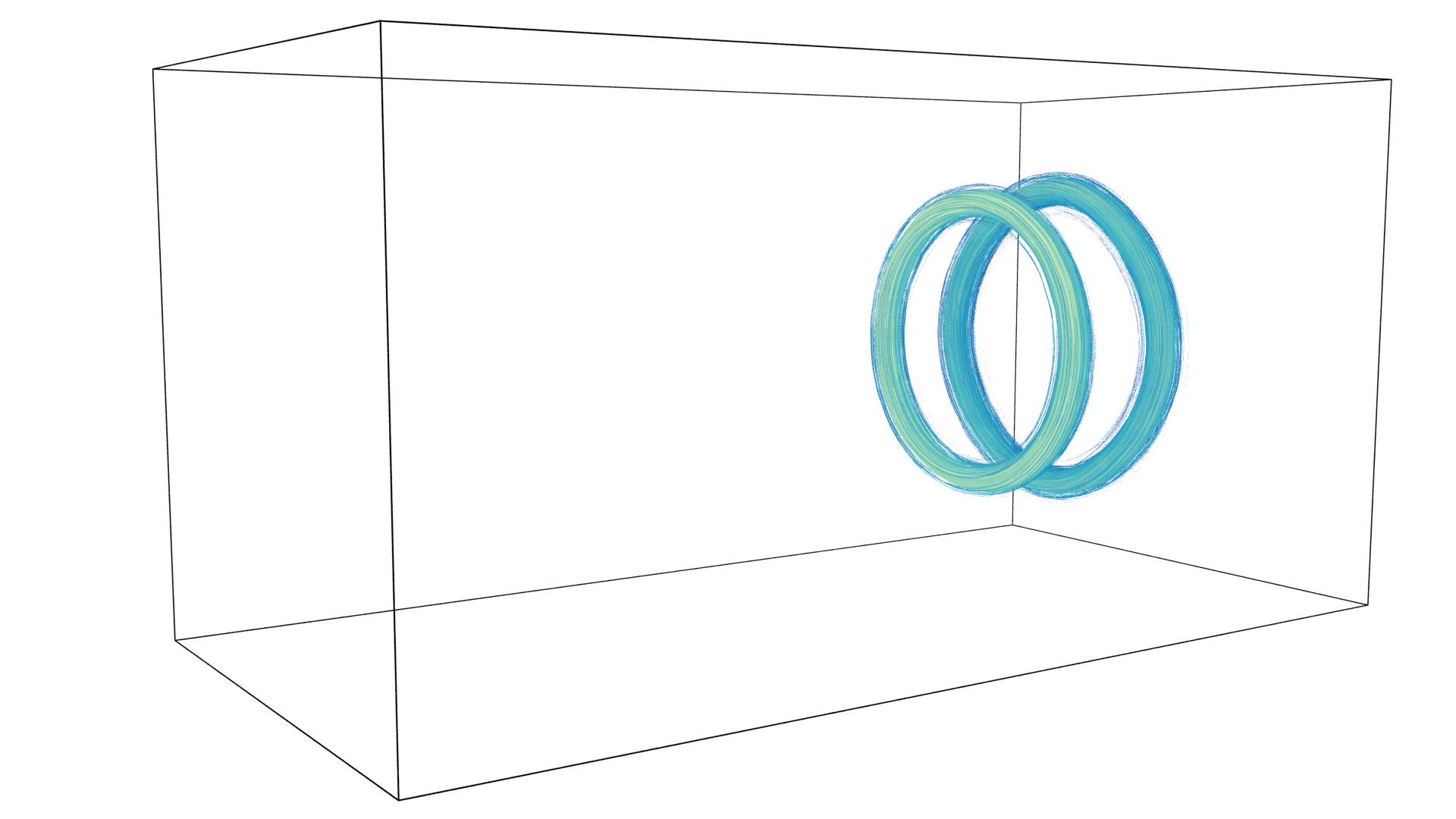}}
    \scalebox{-1}[1]{\includegraphics[trim={50px 0 0 0},clip,width=0.18\linewidth]{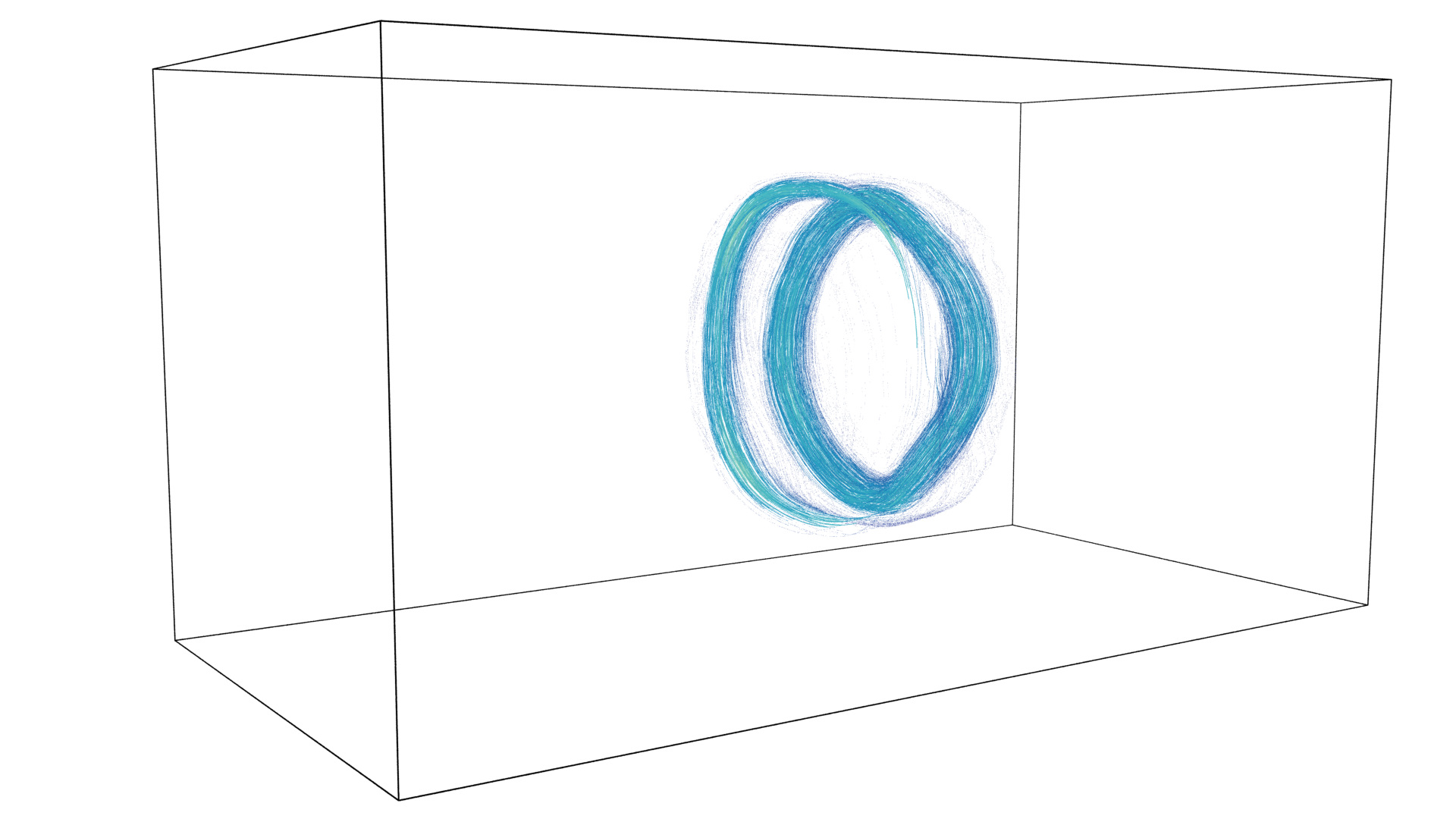}}
    \scalebox{-1}[1]{\includegraphics[trim={50px 0 0 0},clip,width=0.18\linewidth]{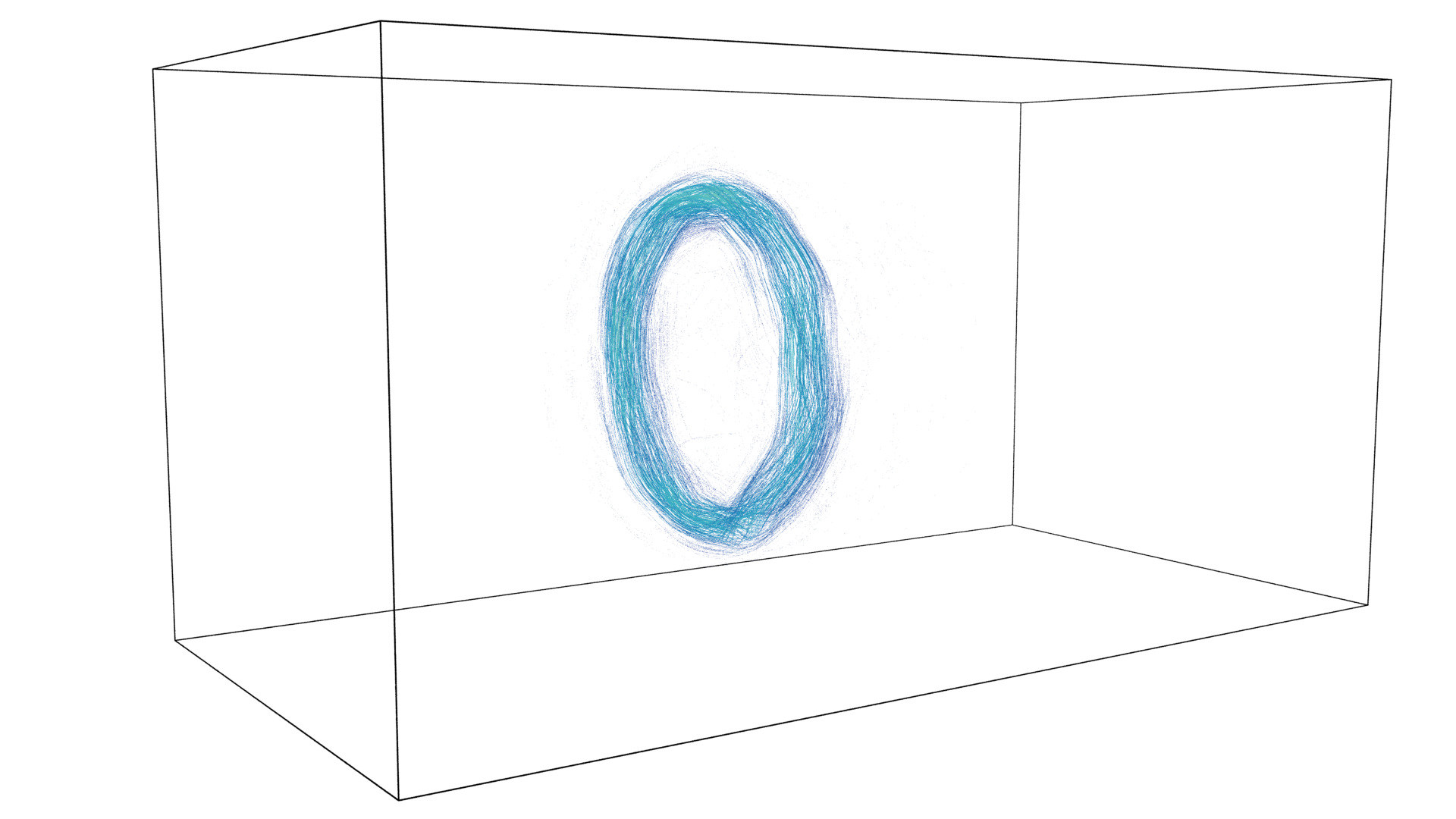}}
    \scalebox{-1}[1]{\includegraphics[trim={50px 0 0 0},clip,width=0.18\linewidth]{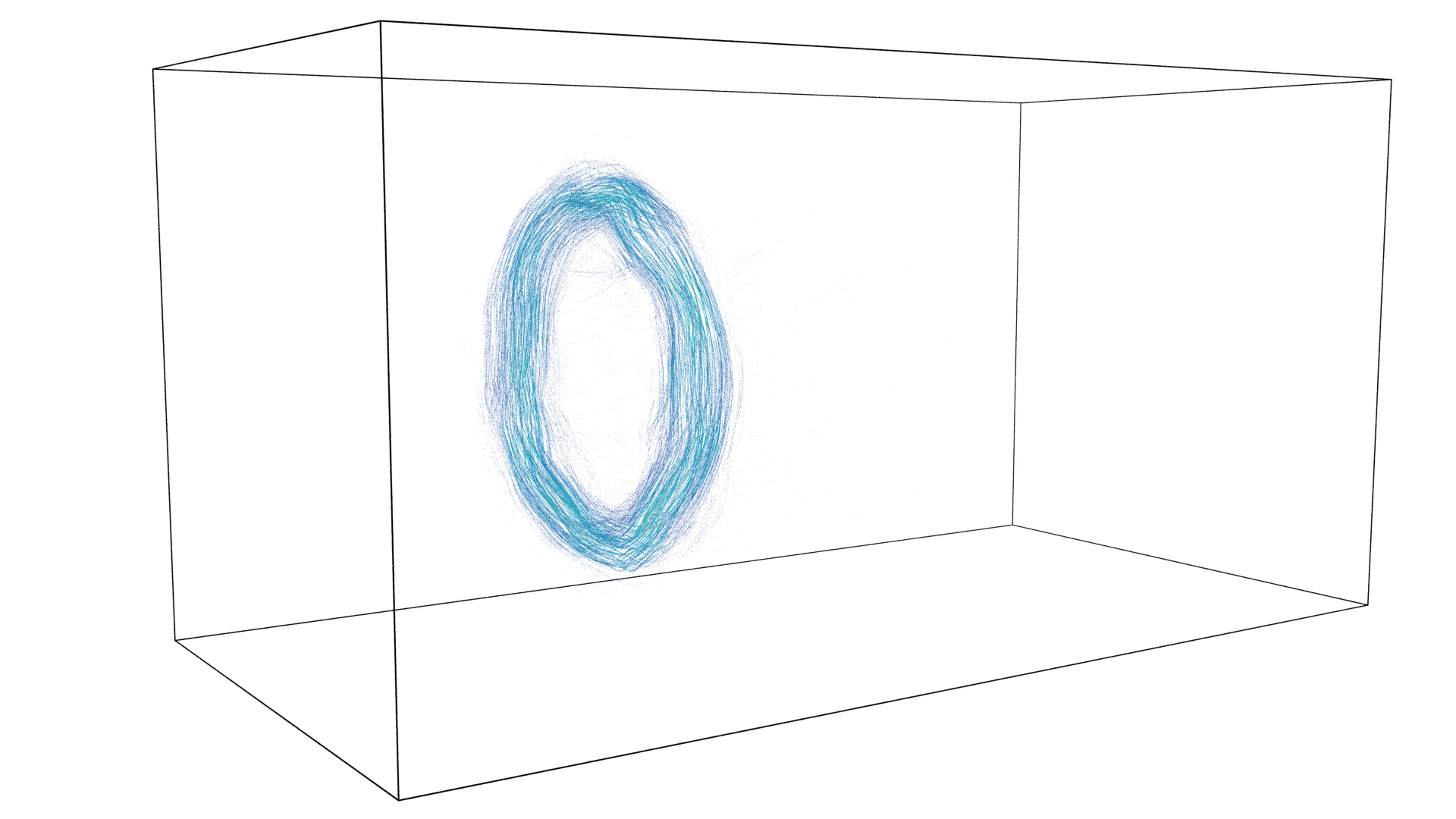}}
    \scalebox{-1}[1]{\includegraphics[trim={50px 0 0 0},clip,width=0.18\linewidth]{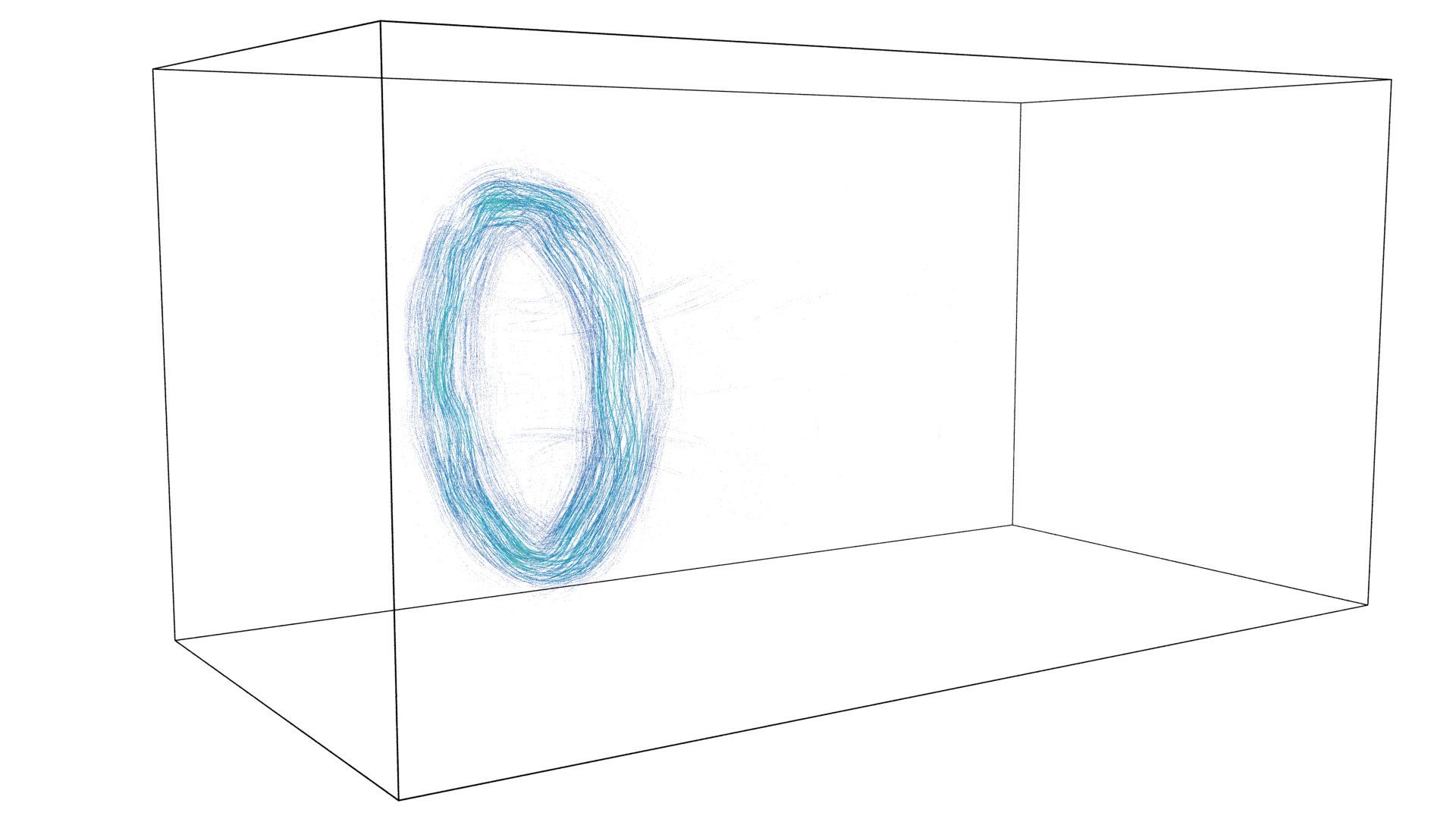}}
    \scalebox{-1}[1]{\includegraphics[trim={50px 0 0 0},clip,width=0.18\linewidth]{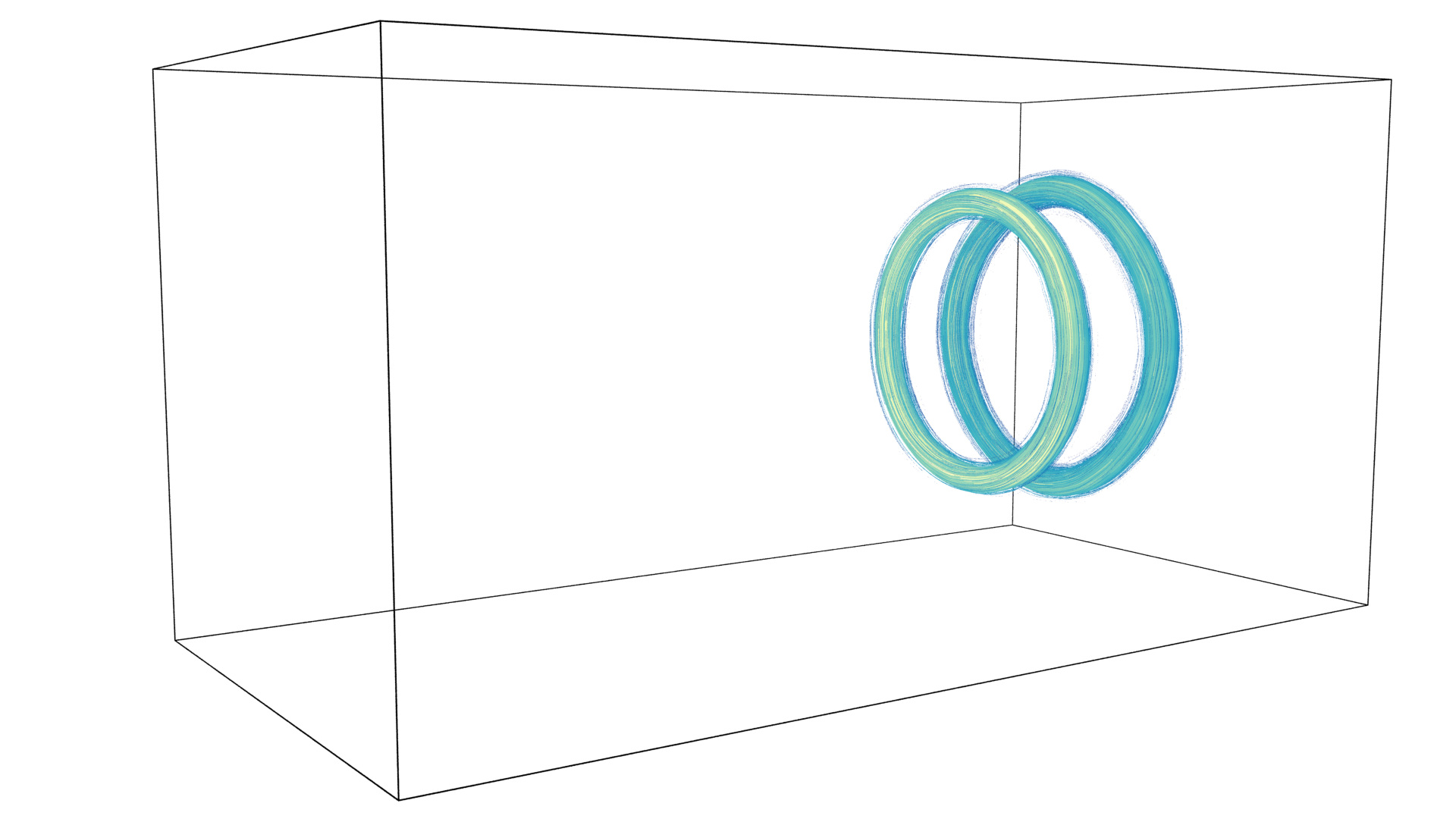}}
    \scalebox{-1}[1]{\includegraphics[trim={50px 0 0 0},clip,width=0.18\linewidth]{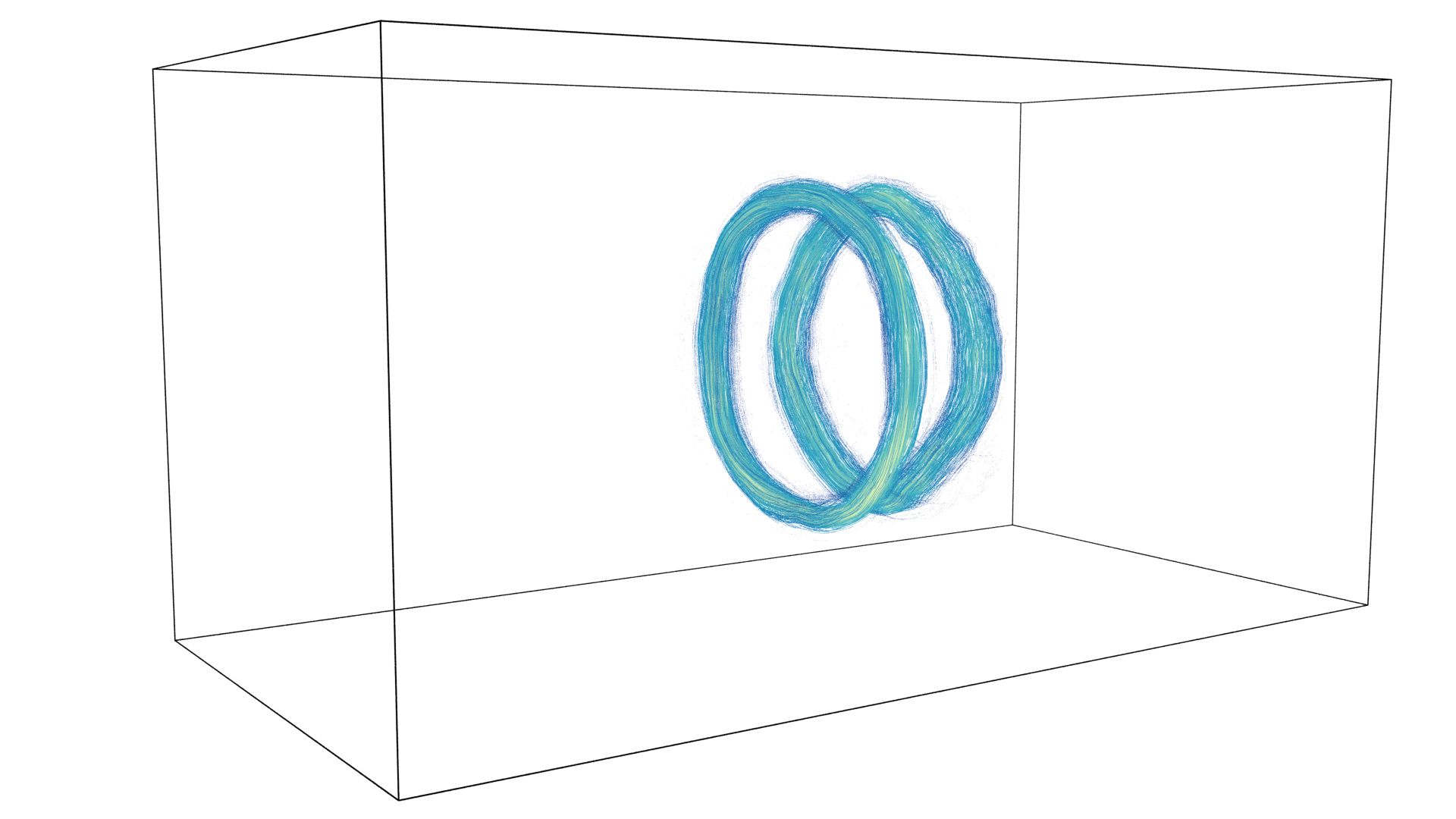}}
    \scalebox{-1}[1]{\includegraphics[trim={50px 0 0 0},clip,width=0.18\linewidth]{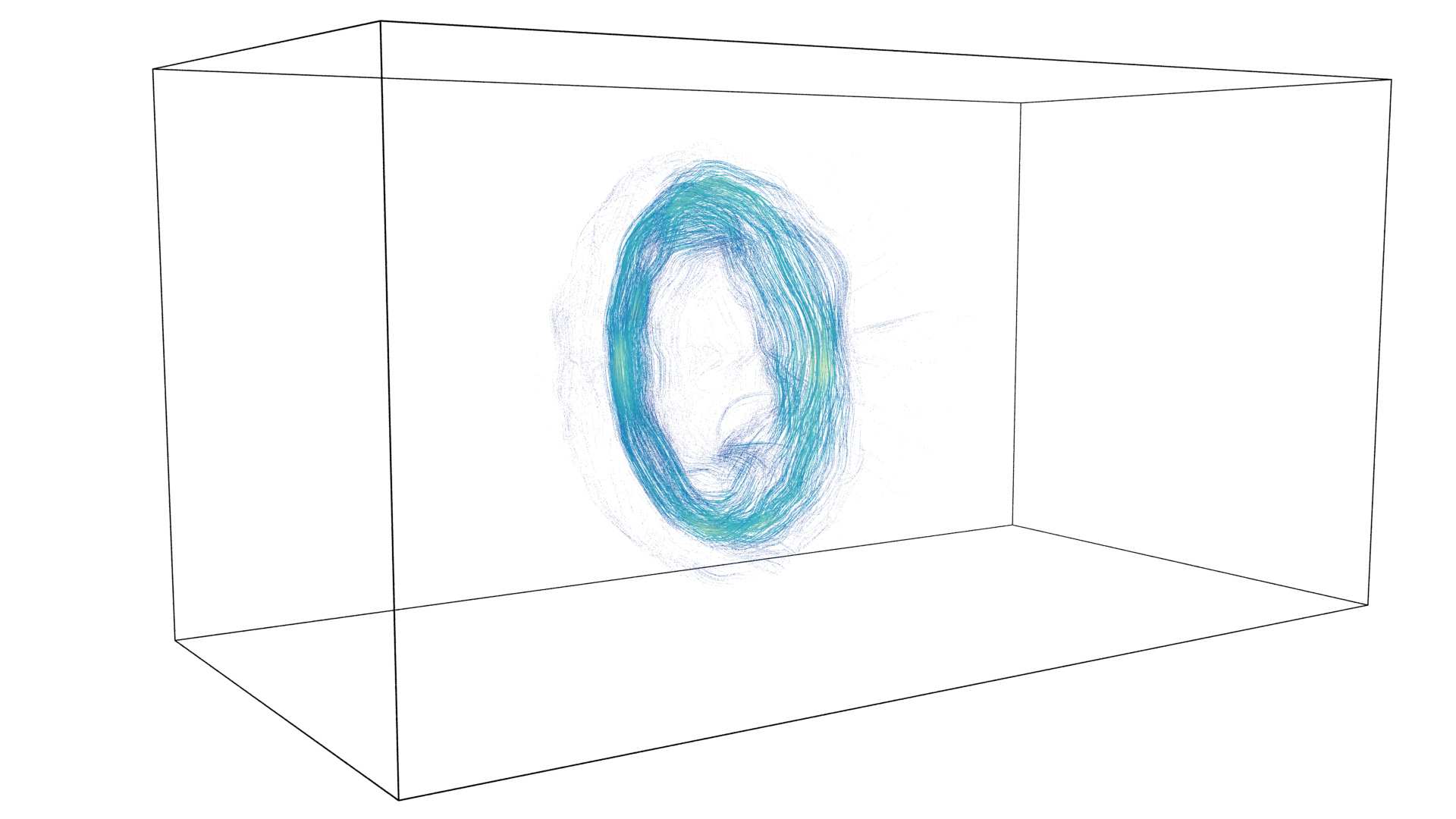}}
    \scalebox{-1}[1]{\includegraphics[trim={50px 0 0 0},clip,width=0.18\linewidth]{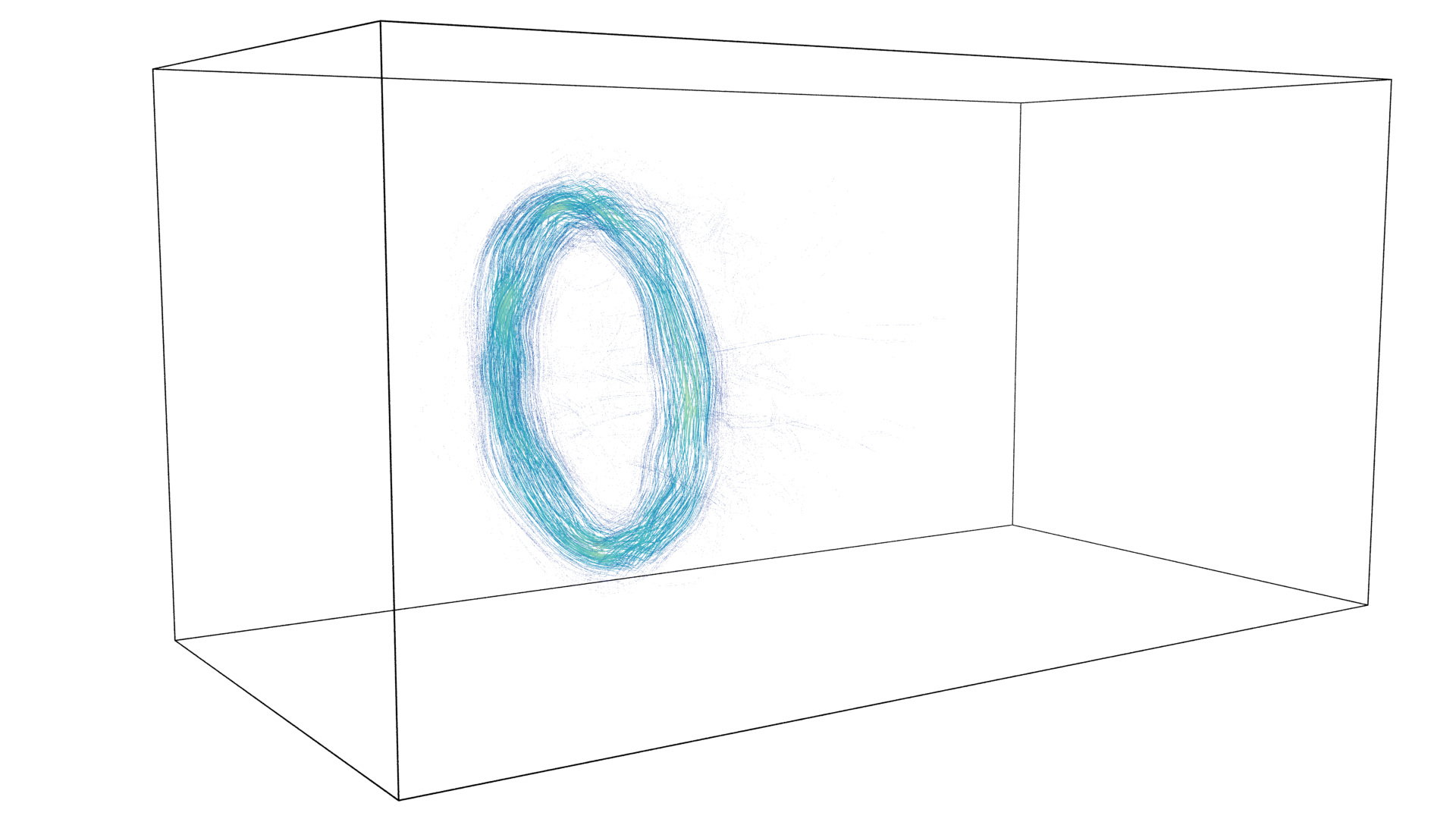}}
    \scalebox{-1}[1]{\includegraphics[trim={50px 0 0 0},clip,width=0.18\linewidth]{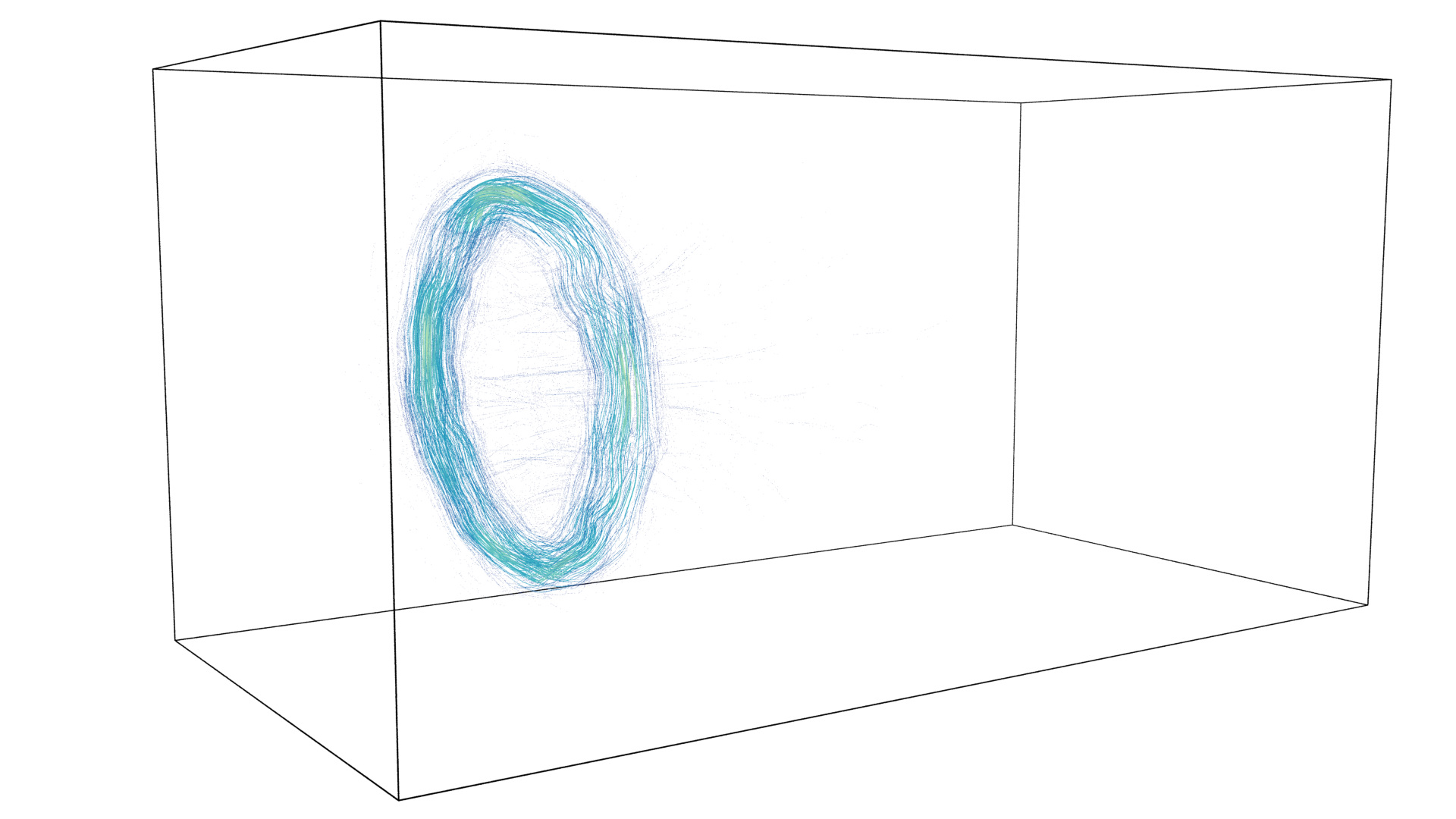}}
    \scalebox{-1}[1]{\includegraphics[trim={50px 0 0 0},clip,width=0.18\linewidth]{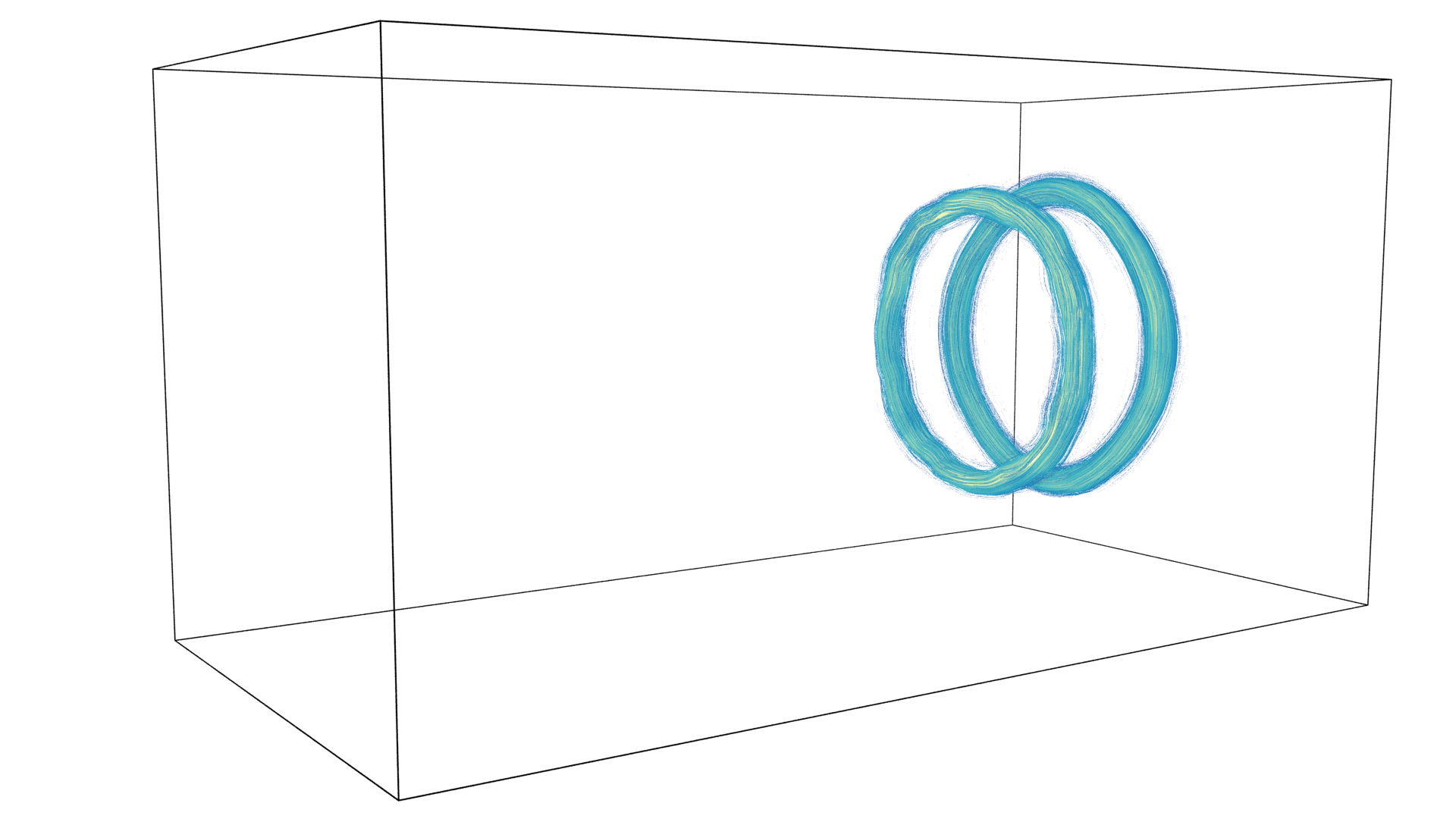}}
    \scalebox{-1}[1]{\includegraphics[trim={50px 0 0 0},clip,width=0.18\linewidth]{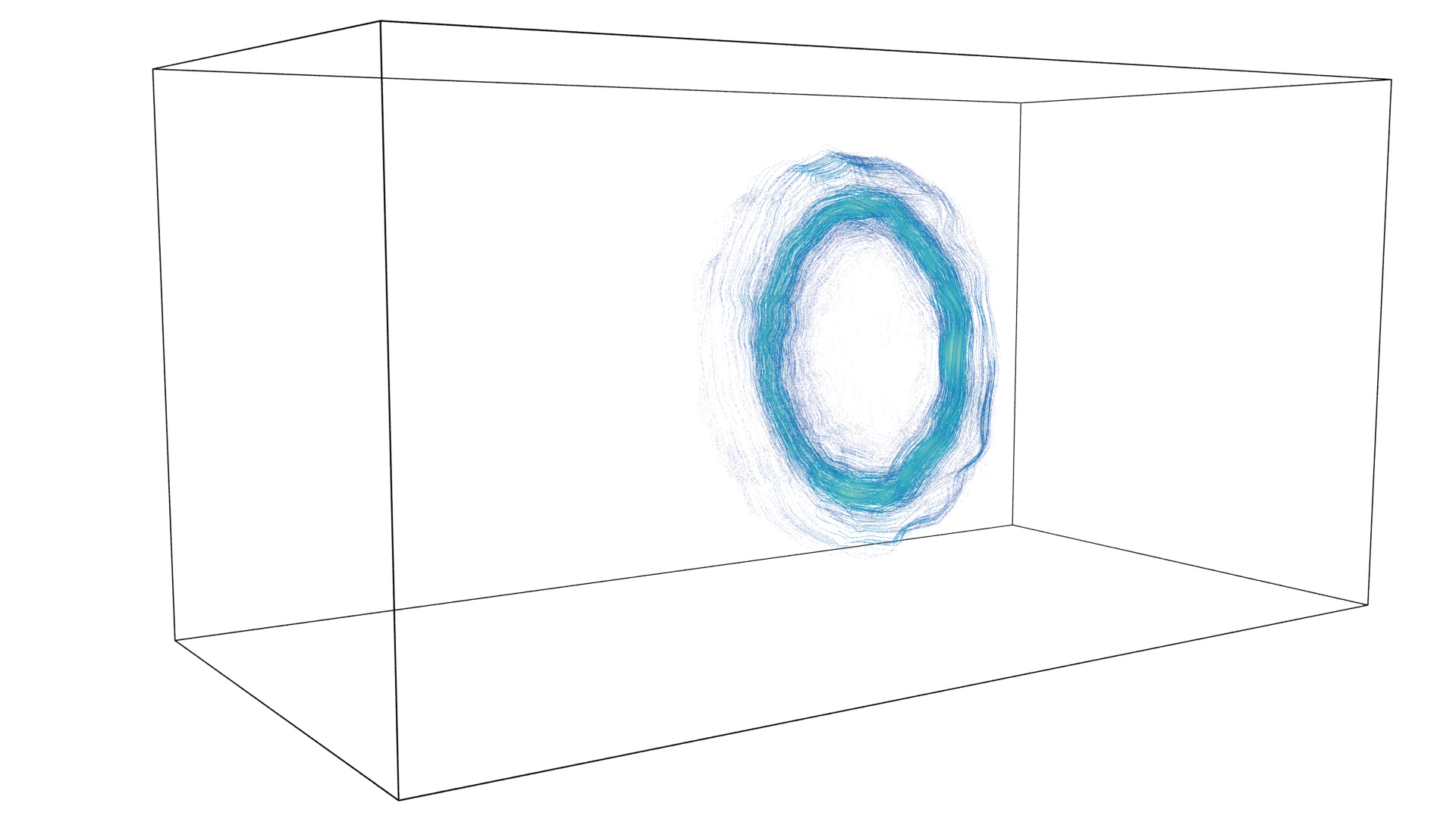}}
    \scalebox{-1}[1]{\includegraphics[trim={50px 0 0 0},clip,width=0.18\linewidth]{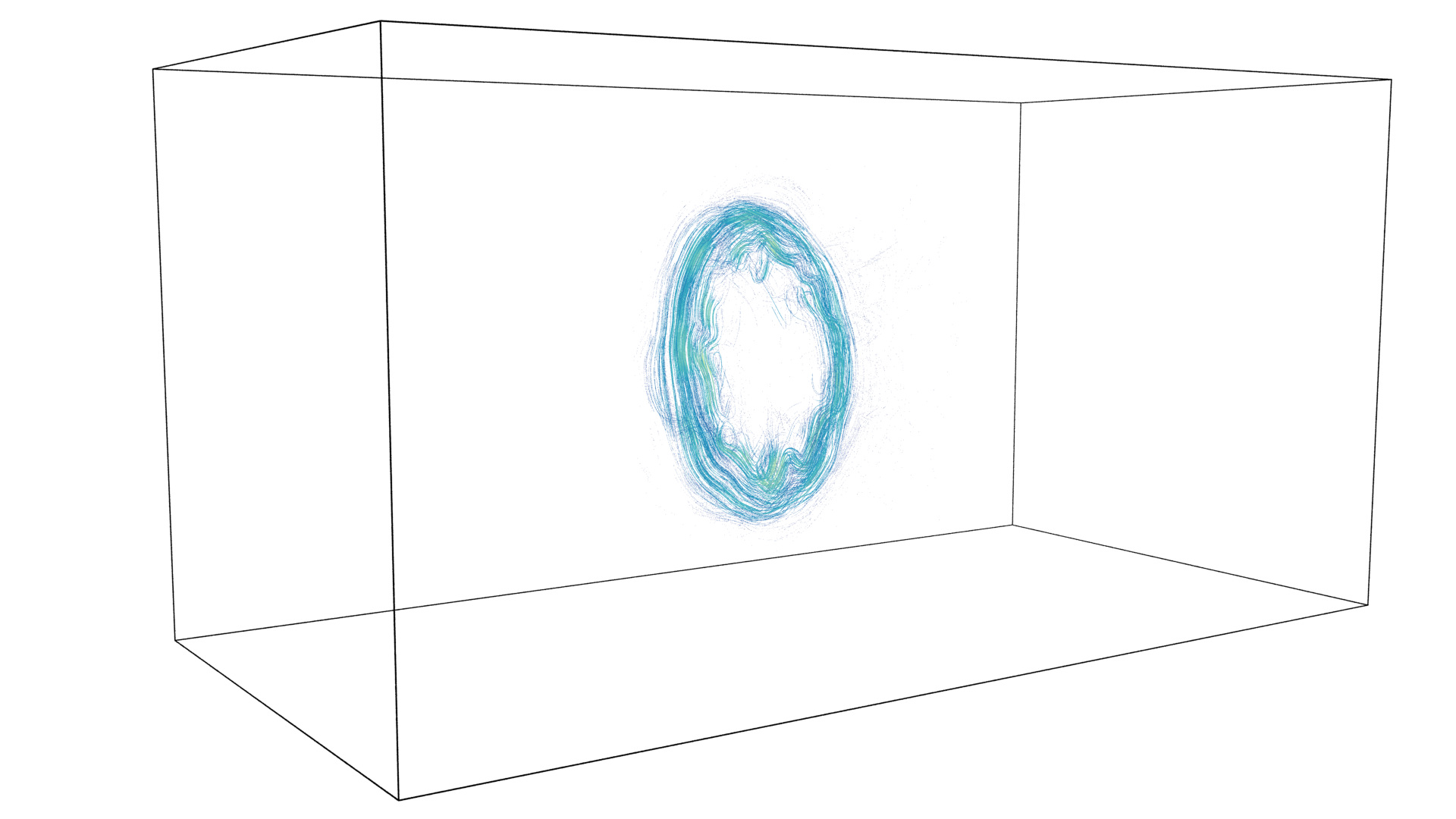}}
    \scalebox{-1}[1]{\includegraphics[trim={50px 0 0 0},clip,width=0.18\linewidth]{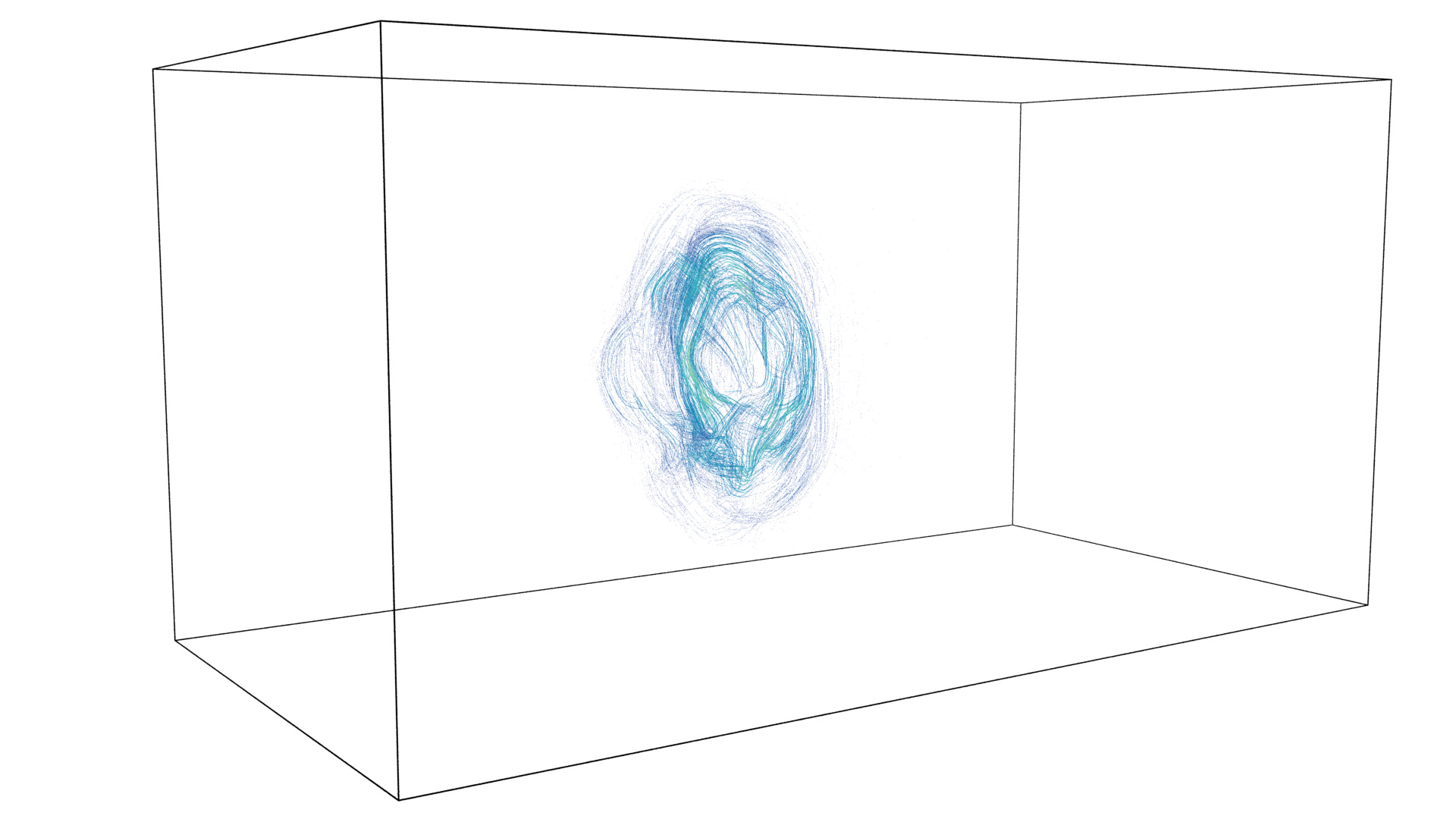}}
    \scalebox{-1}[1]{\includegraphics[trim={50px 0 0 0},clip,width=0.18\linewidth]{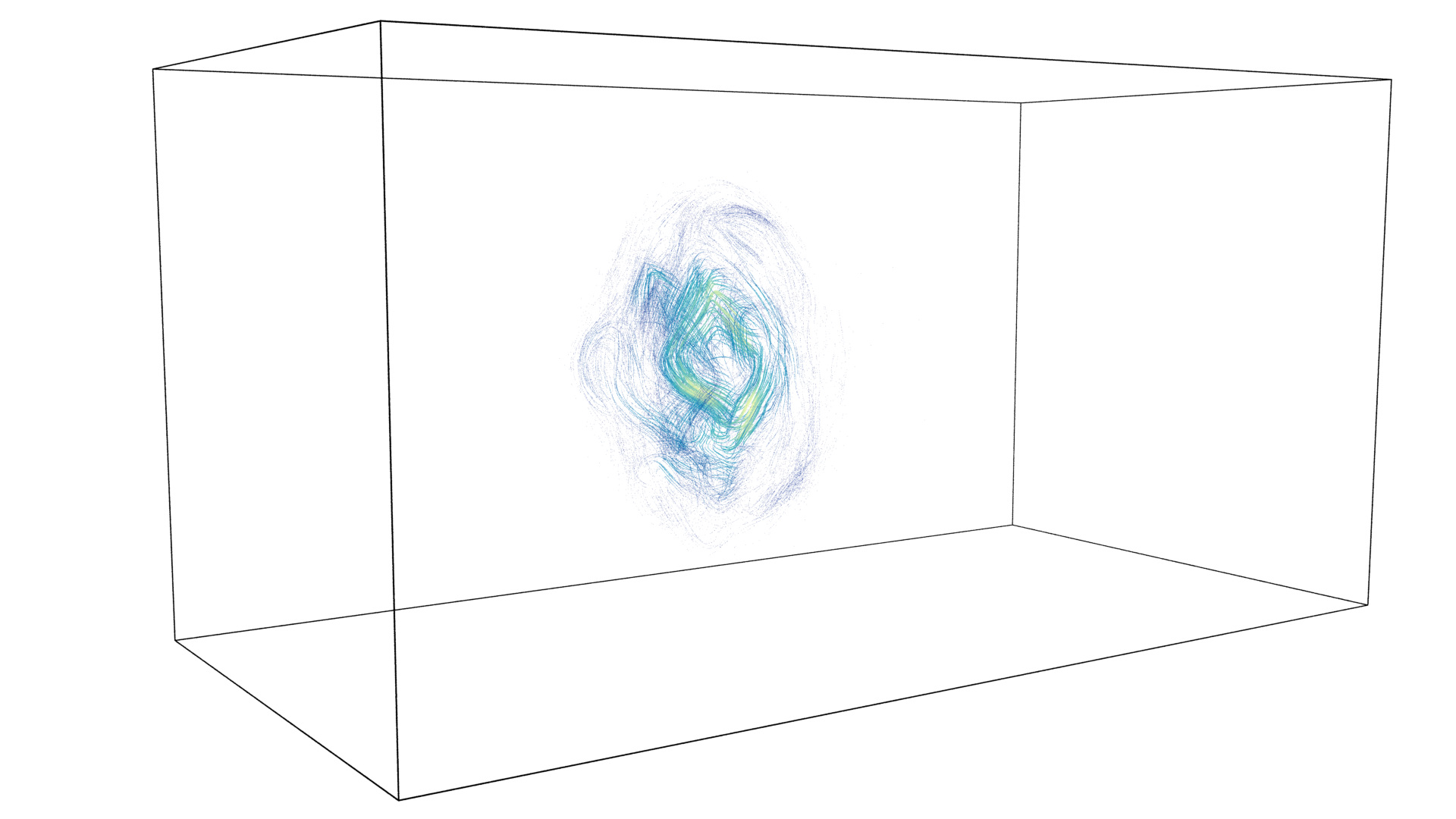}}
    \scalebox{-1}[1]{\includegraphics[trim={50px 0 0 0},clip,width=0.18\linewidth]{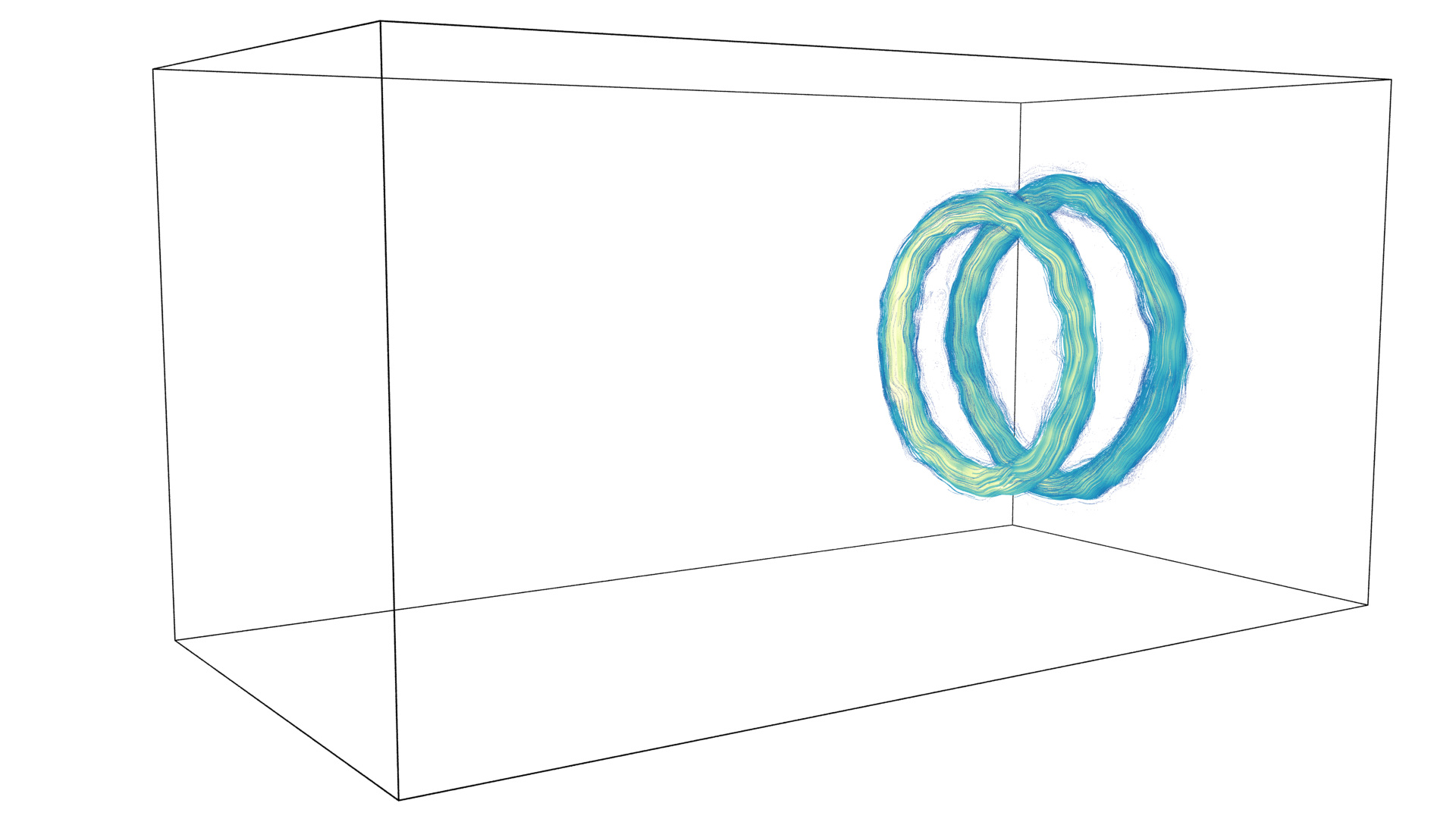}}
    \scalebox{-1}[1]{\includegraphics[trim={50px 0 0 0},clip,width=0.18\linewidth]{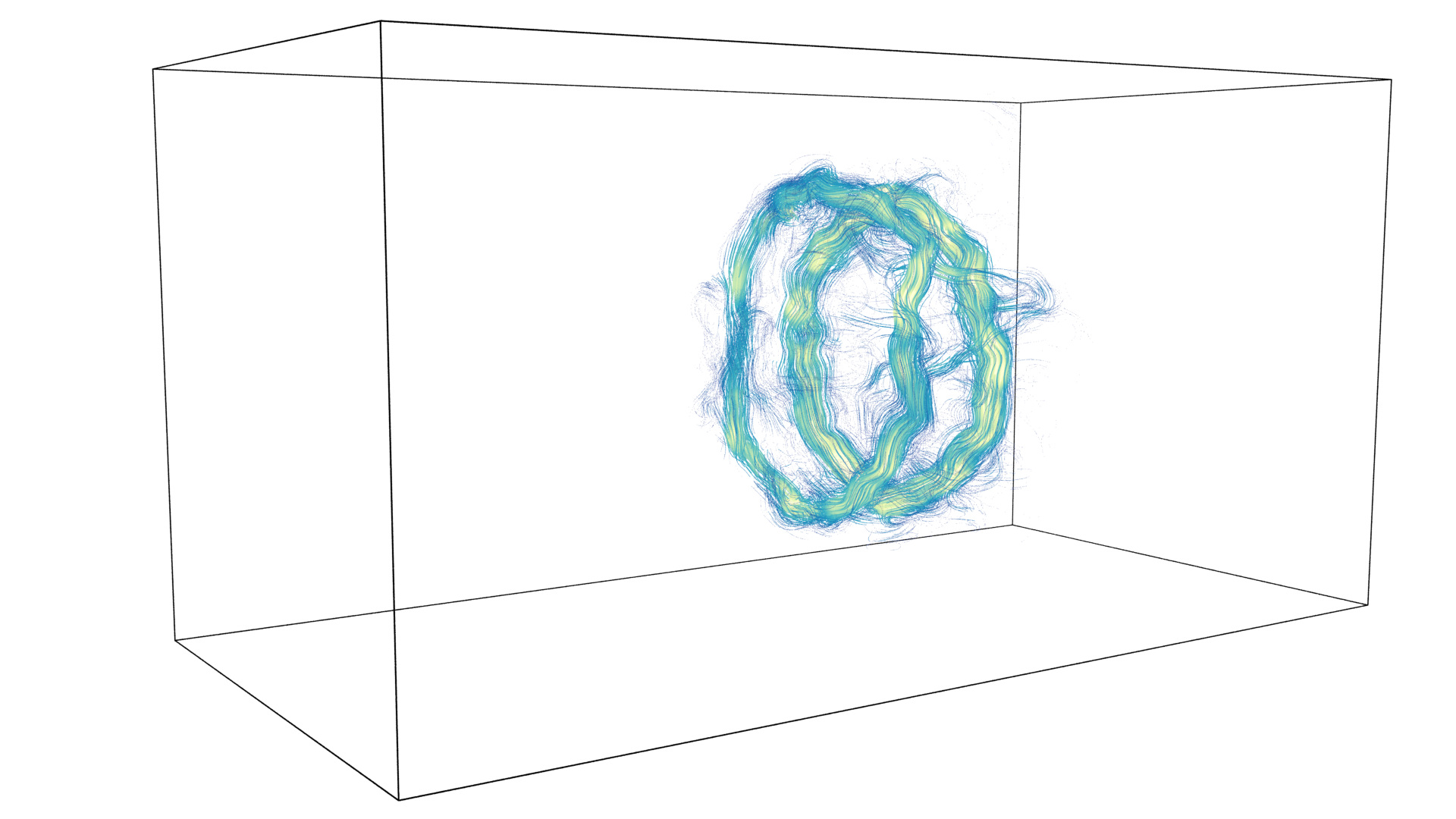}}
    \scalebox{-1}[1]{\includegraphics[trim={50px 0 0 0},clip,width=0.18\linewidth]{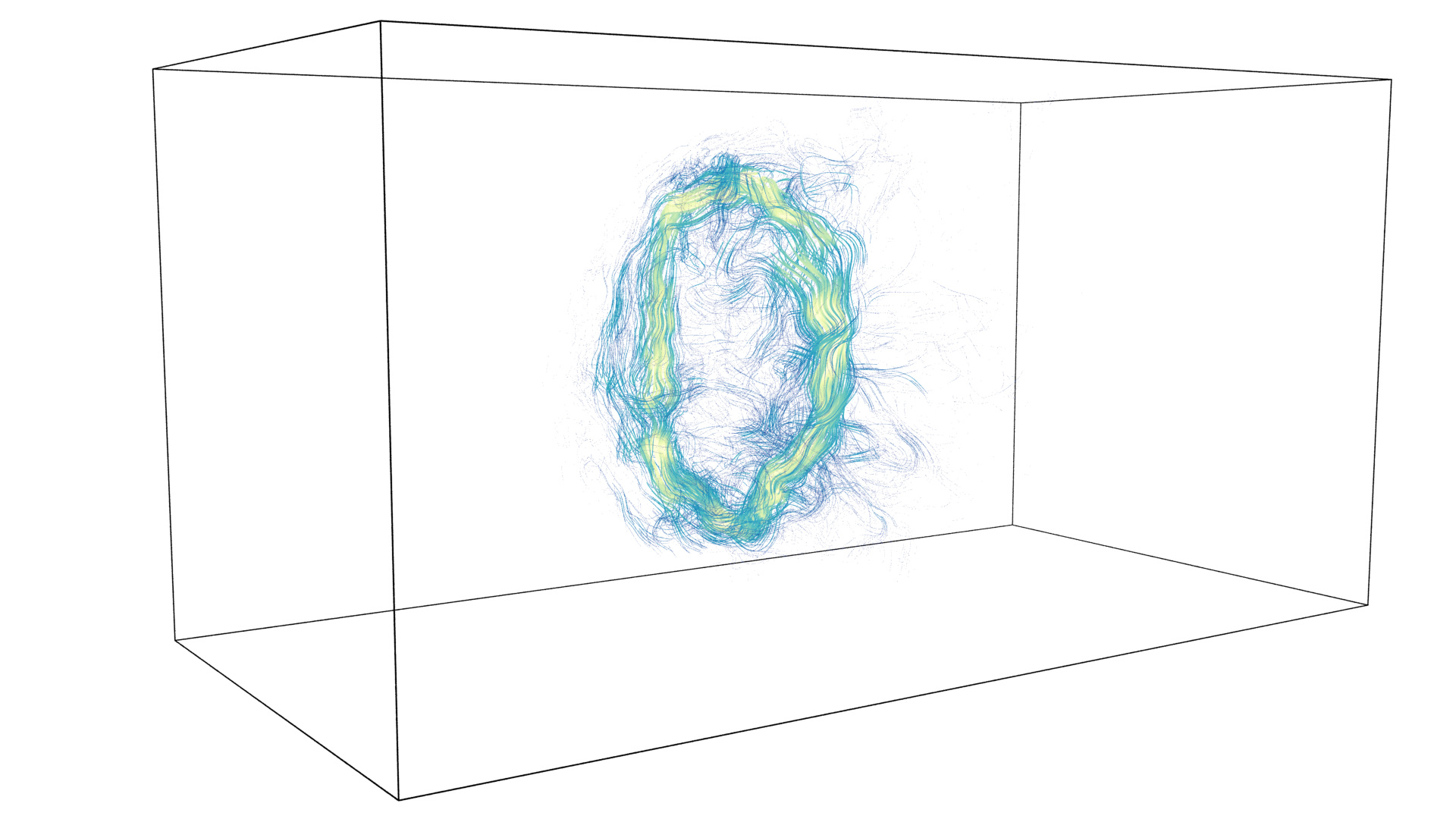}}
    \scalebox{-1}[1]{\includegraphics[trim={50px 0 0 0},clip,width=0.18\linewidth]{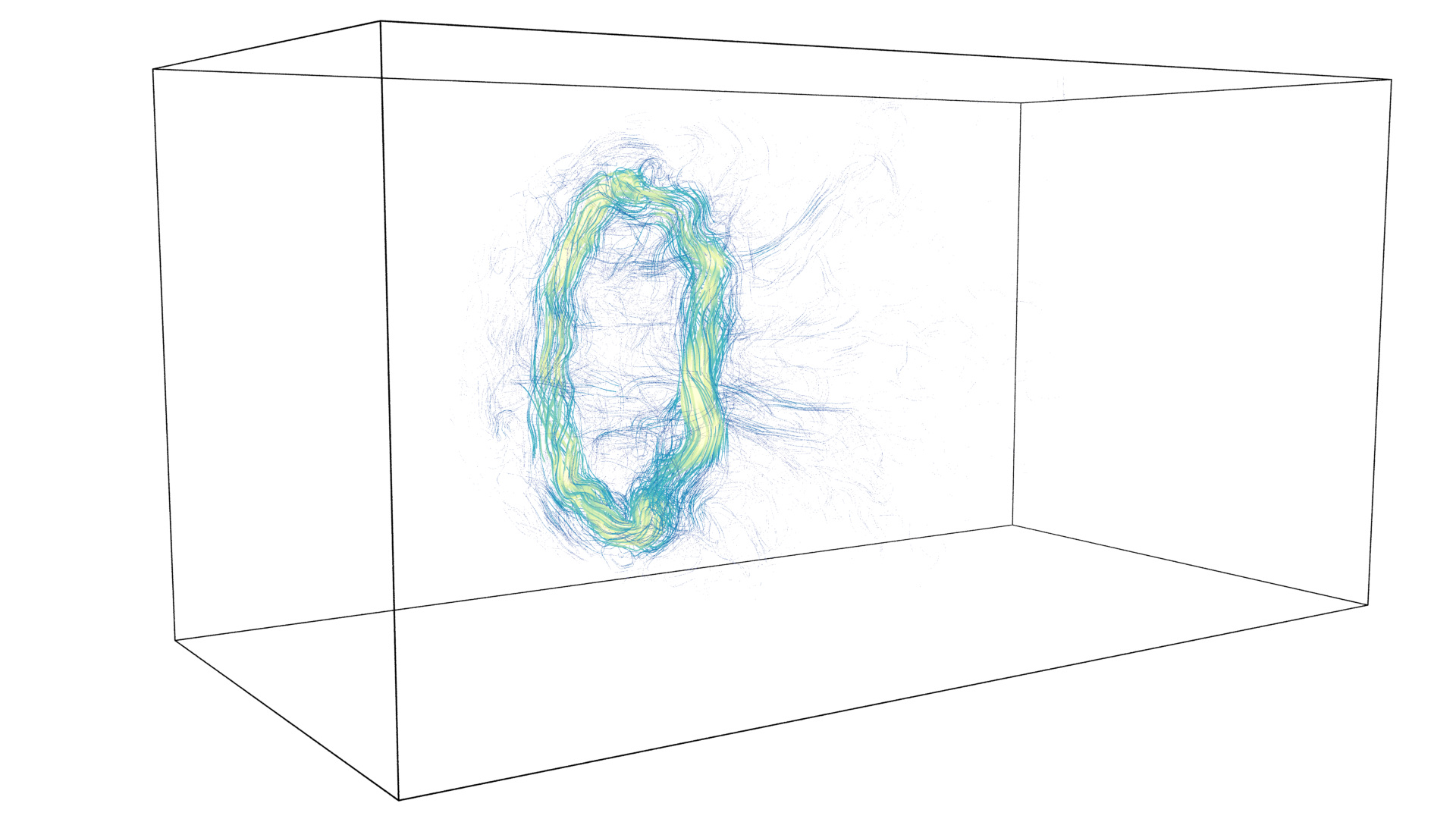}}
    \scalebox{-1}[1]{\includegraphics[trim={50px 0 0 0},clip,width=0.18\linewidth]{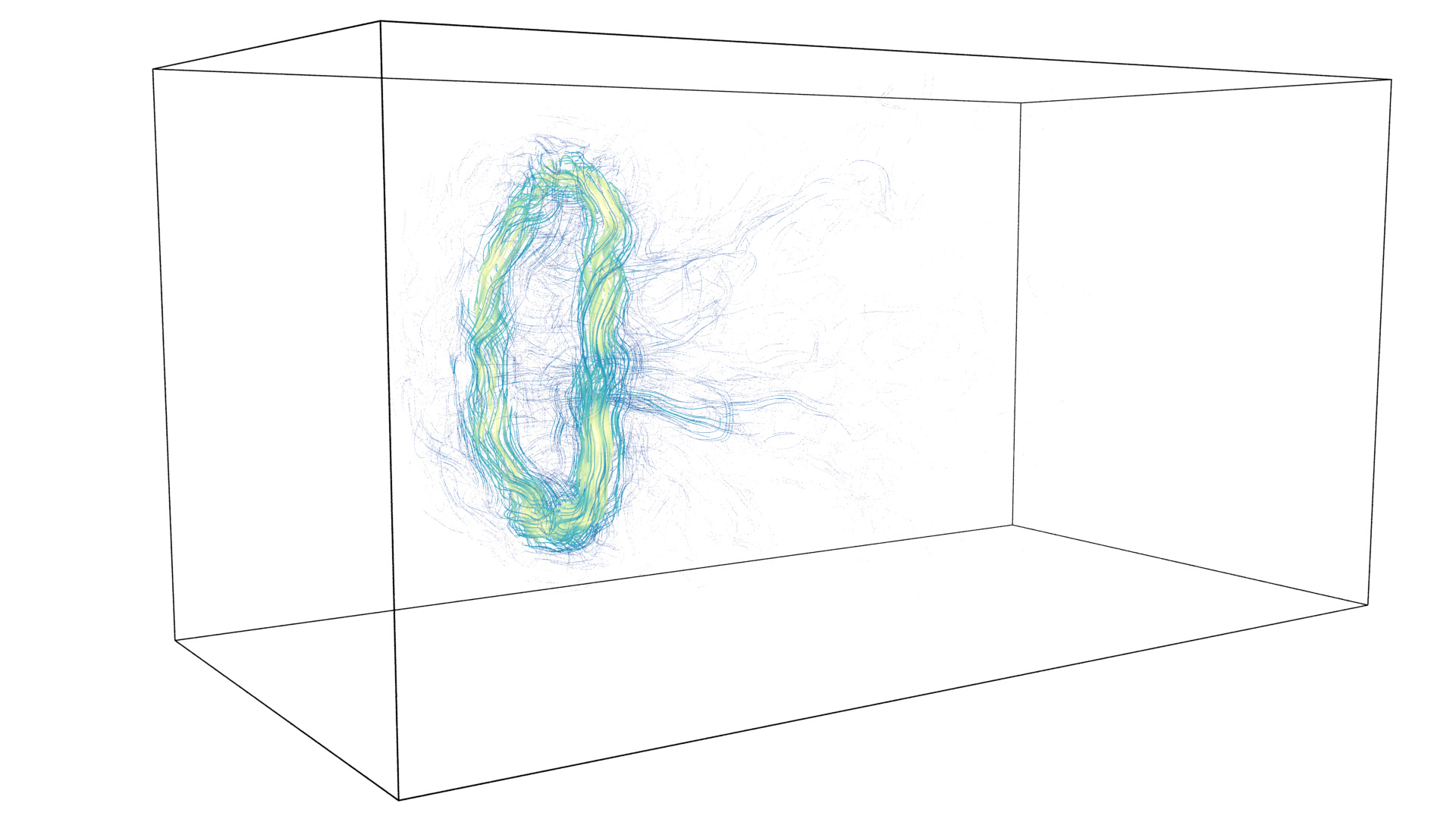}}
    \scalebox{-1}[1]{\includegraphics[trim={50px 0 0 0},clip,width=0.18\linewidth]{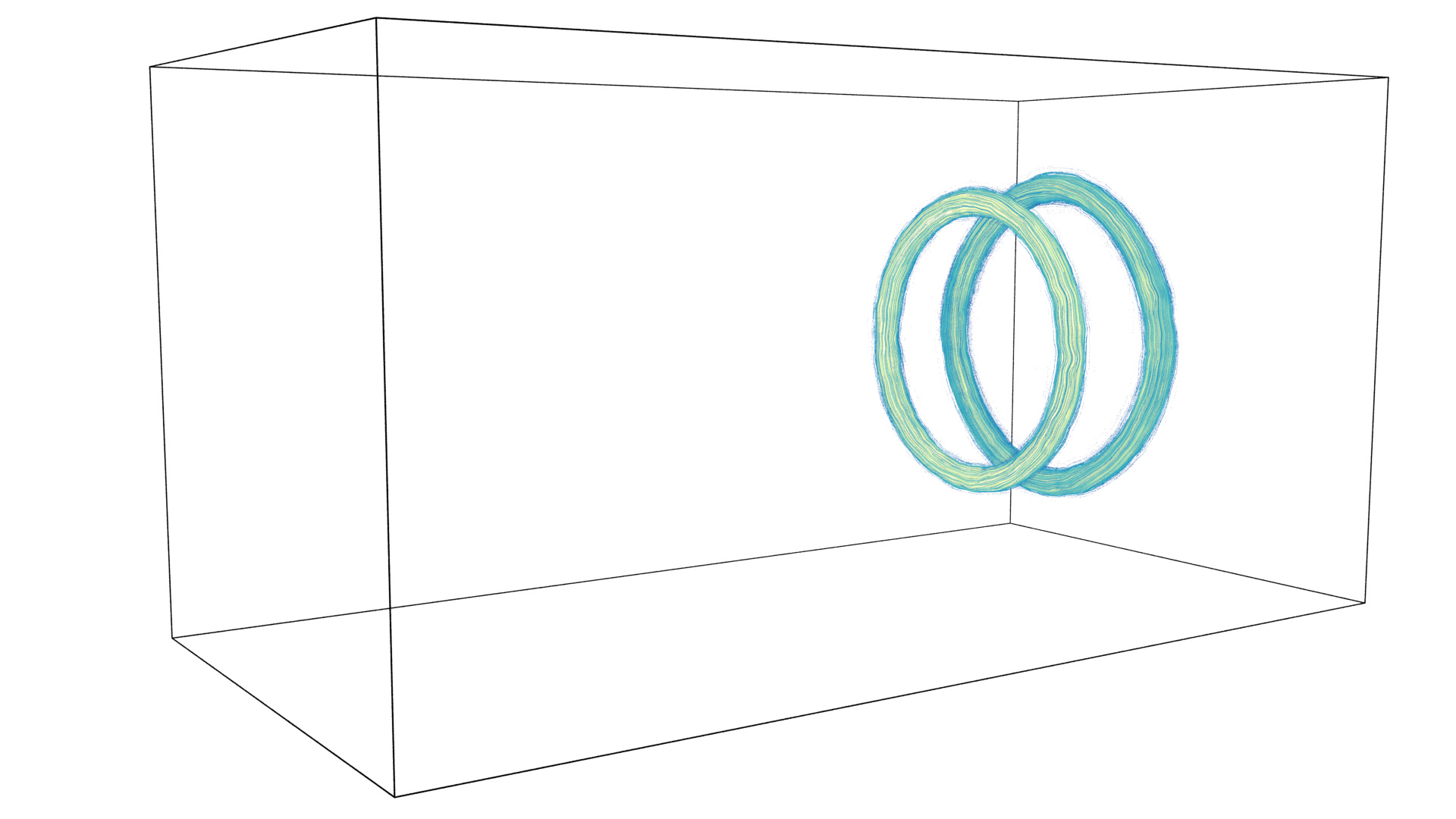}}
    \scalebox{-1}[1]{\includegraphics[trim={50px 0 0 0},clip,width=0.18\linewidth]{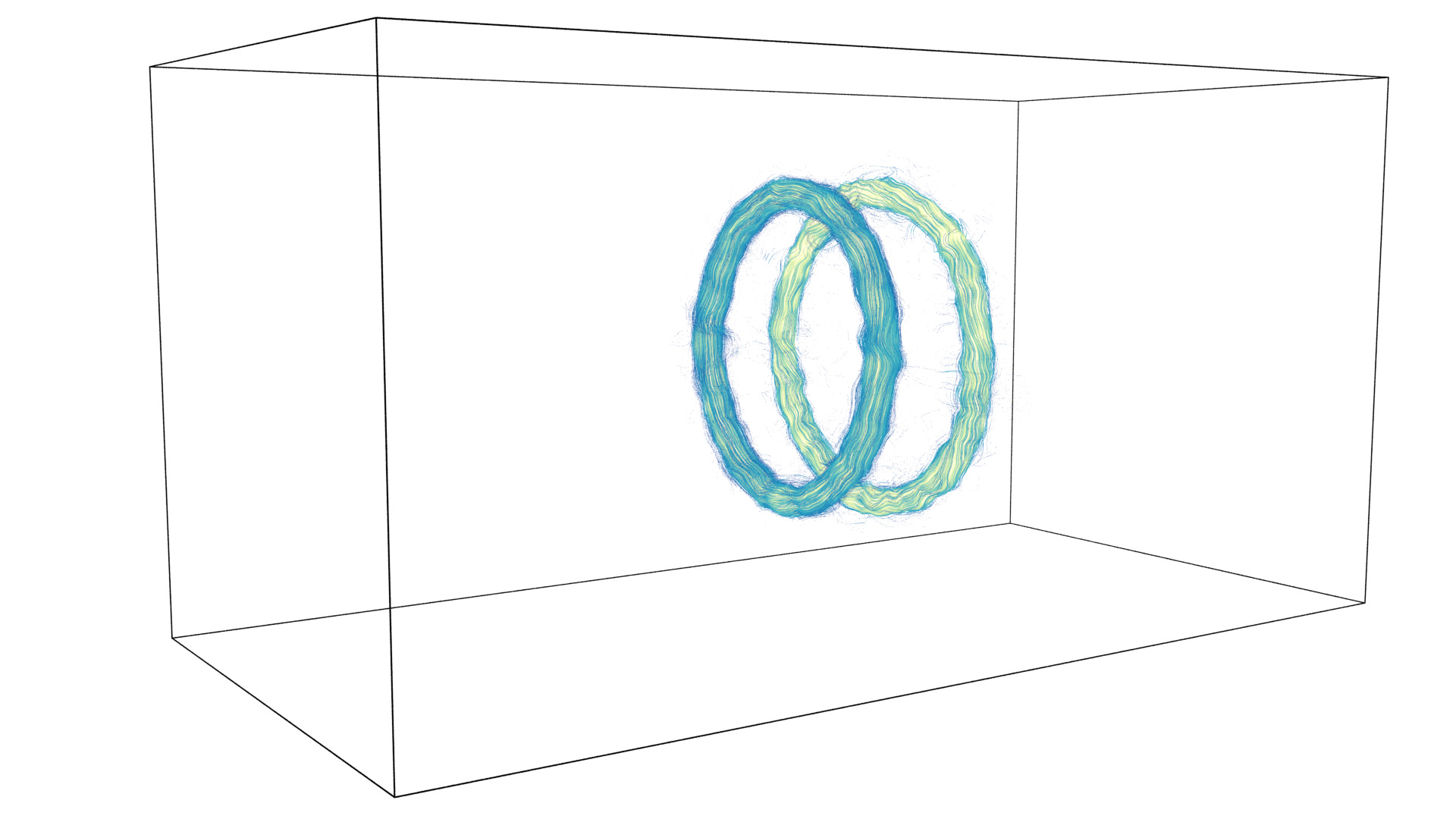}}
    \scalebox{-1}[1]{\includegraphics[trim={50px 0 0 0},clip,width=0.18\linewidth]{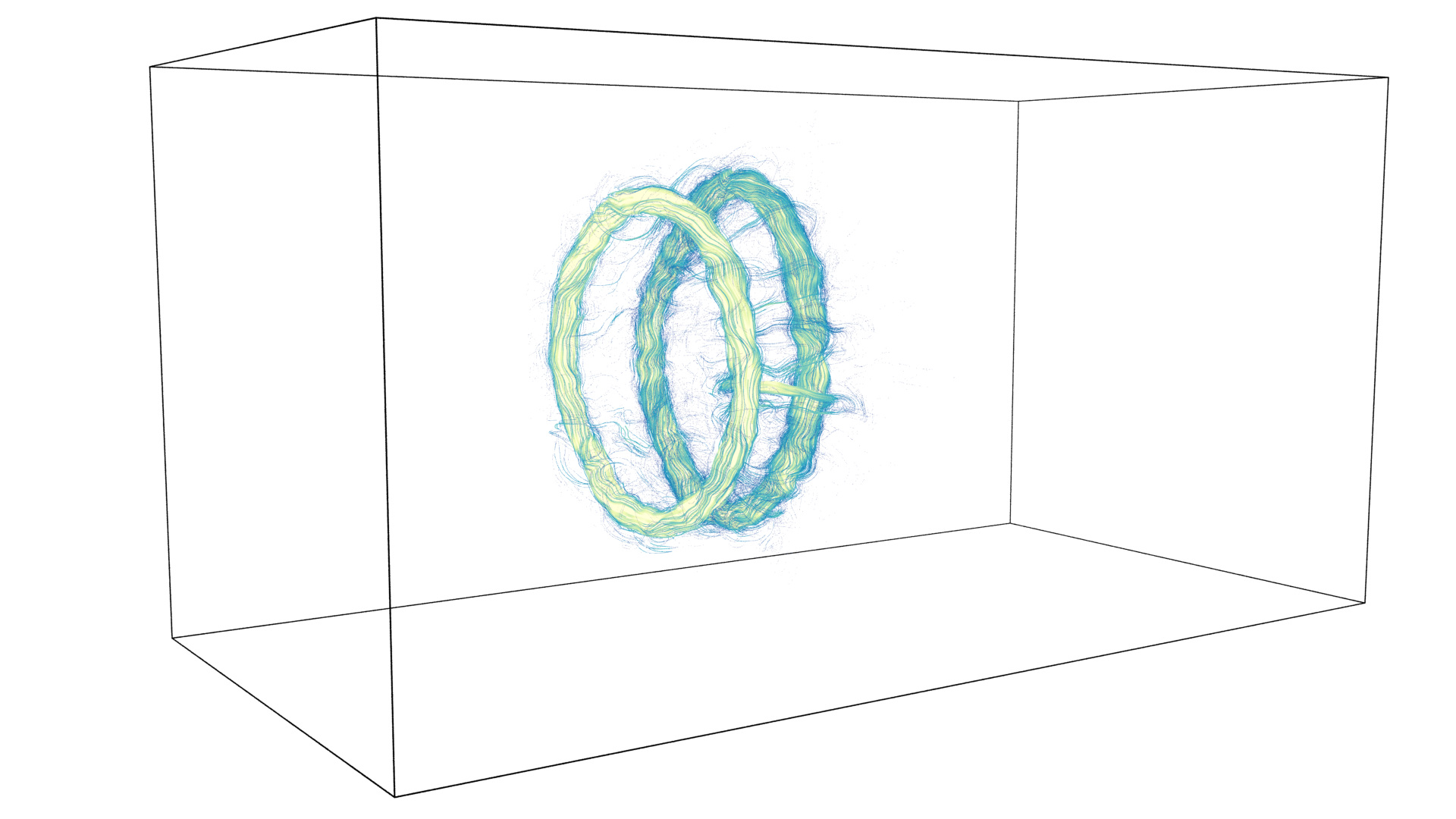}}
    \scalebox{-1}[1]{\includegraphics[trim={50px 0 0 0},clip,width=0.18\linewidth]{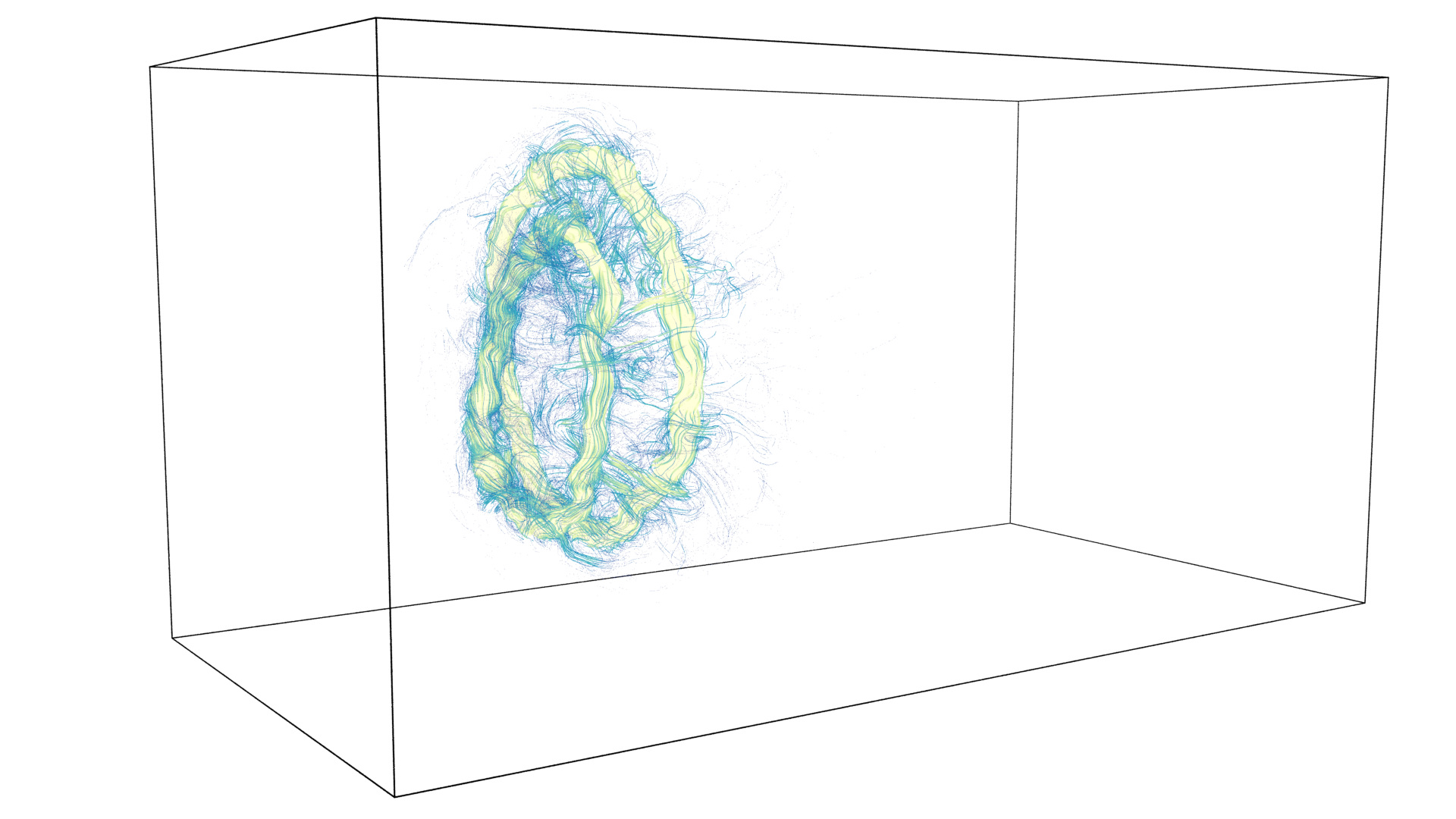}}
    \scalebox{-1}[1]{\includegraphics[trim={50px 0 0 0},clip,width=0.18\linewidth]{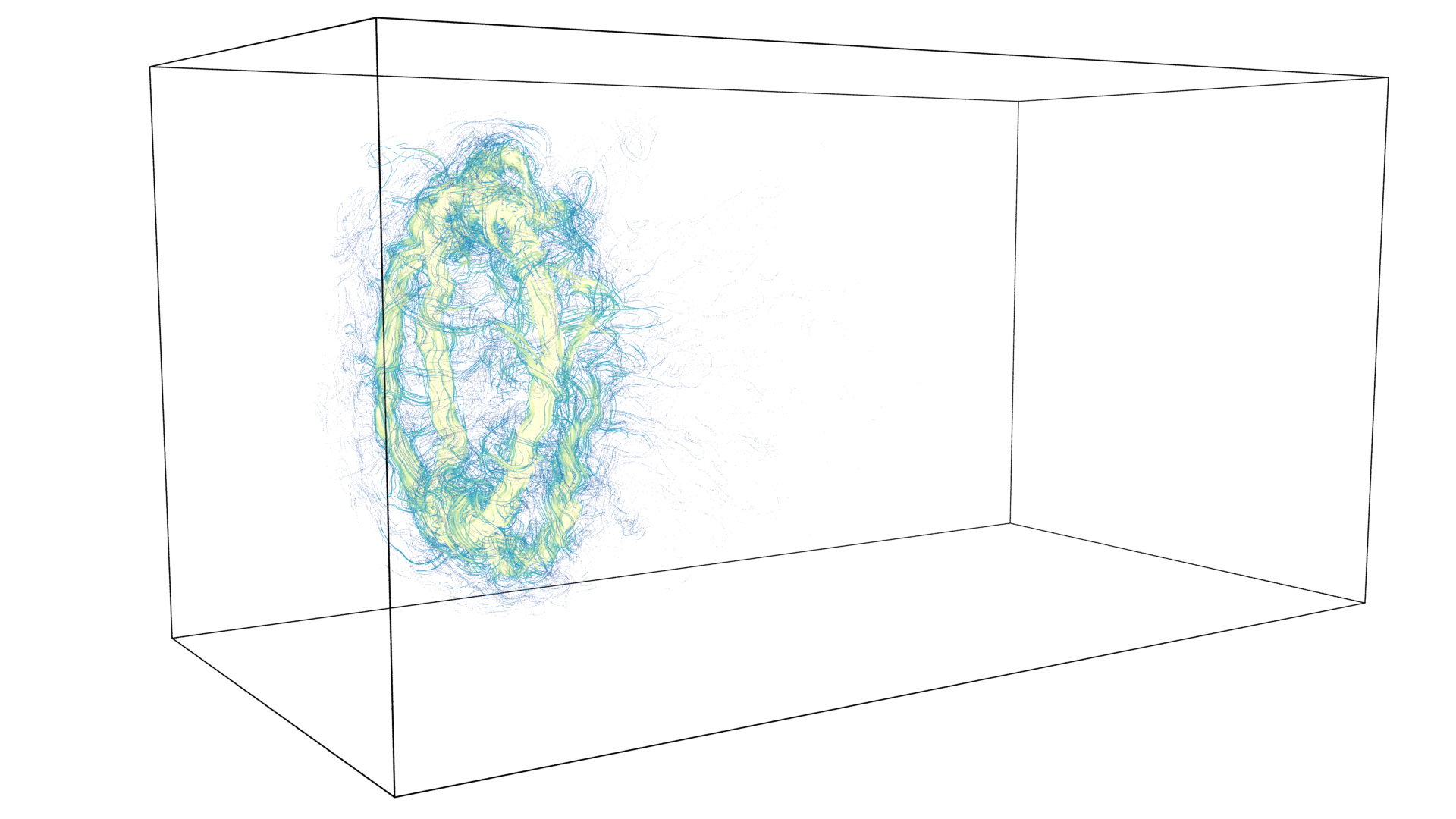}}
    \\
    \begin{picture}(0,0)(0,0)
        \put(-266,25){\begin{tikzpicture}
\pgfplotscolorbardrawstandalone[ 
    colormap={myProteinColor}{
        rgb255=(8, 29, 88)
        rgb255=(37, 52, 148)
        rgb255=(34, 94, 168)
        rgb255=(29, 145, 192)
        rgb255=(65, 182, 196)
        rgb255=(127, 205, 187)
        rgb255=(199, 233, 180)
        rgb255=(237, 248, 177)
        rgb255=(255,255, 217)
    },
    point meta min=-1,
    point meta max=1,
    colorbar style={
        width=4pt,
        height=35pt,
        ytick={-1,-0.5,0,0.5,1},
        ytick style={draw=none},
        yticklabels={{0},{},{},{},{15}},
        yticklabel style={font=\tiny, xshift=-0.5ex, scale=0.9},
        ylabel={\sffamily vorticity norm ($\nicefrac{1}{\text{s}}$)},
        ylabel style={font=\tiny, yshift=22pt, scale=0.9}
        }]
\end{tikzpicture}}
        \put(-225,10){\sffamily \scriptsize First leap}
        \put(-130,10){\sffamily \scriptsize Second leap}
        \put(-30,10){\sffamily \scriptsize Third leap}
        \put(60,10){\sffamily \scriptsize Fourth leap}
        \put(155,10){\sffamily \scriptsize Fifth leap}
        \put(-242,357){\sffamily \scriptsize \rotatebox{90}{PolyPIC}}
        \put(-242,305){\sffamily \scriptsize \rotatebox{90}{PolyFLIP}}
        \put(-242,245){\sffamily \scriptsize \rotatebox{90}{CF+PolyFLIP}}
        \put(-242,192){\sffamily \scriptsize \rotatebox{90}{R+PolyFLIP}}
        \put(-242,140){\sffamily \scriptsize \rotatebox{90}{NFM}}
        \put(-235,82){\sffamily \tiny \rotatebox{90}{$64\times32\times32$}}
        \put(-242,72){\sffamily \scriptsize \rotatebox{90}{\textbf{CO-FLIP (Ours)}}}
        \put(-235,24){\sffamily \tiny \rotatebox{90}{$128\times64\times64$}}
        \put(-242,18){\sffamily \scriptsize \rotatebox{90}{\textbf{CO-FLIP (Ours)}}}
    \end{picture}
    \caption{Leapfrogging vortex rings in 3D.
    At a low resolution of $128\times64\times64$, our method (CO-FLIP) leaps five times before the vortex rings mix.
    At an even lower resolution of $64\times32\times32$, CO-FLIP manages to leap almost three times, which shows better behavior than some traditional works at higher resolutions.
    Previous methods at best (i.e. R+PolyFLIP method) leap twice.}
    \label{fig:leapfrog_rings}
\end{figure*}

\subsubsection{2D Vortex sheet under refinement}
To study cascading vortical structures, we initialize the velocity field with a rigidly rotating disk with angular velocity $0.25\,\text{rad}/\text{s}$.  Due to the sudden drop in velocity at the border of the disk, a vortex sheet is created.  Together with the finite resolution of the grid, this would create an instability that produces cascading vortices.  See \figref{fig:vortex_sheet}. Note that with our method the energy remains constant regardless of the spatiotemporal resolution.  With increasing the resolution, our method consistently produces richer vortical structures, which qualitatively mimics the results of previous methods at a higher resolution. Additionally, the structures get progressively more turbulent.

\subsubsection{3D Twisted Torus}
We further validate our method in three dimensions.
When filaments in a closed vortex tube are twisted, they generate a velocity field with non-zero helicity \cite{Kleckner:2013:CDK}.
We use the codebase open-sourced by \cite{Chern:2017:IF} to generate such a vortex tube.
Concretely, we restrict the initial vorticity field to the inside of a torus with a major radius $1.5\,m$ and an aspect ratio of 6.
Given two variables $\hat\omega_1$ (for vorticity per unit length through the torus tube) and $\omega_2$ (for vorticity responsible for creating twists, which is perpendicular to the former component).
We set the total vorticity field as $\hat\omega_1(\vec n(\vec x) \times \vec d(\vec x)) + \omega_2\vec n(\vec x)$, where $\vec x$ is a location in the domain, $\vec d = \vec x - \vec r_{\text{core}}$, and $\vec n$ is the tangential unit vector to the torus at minor radius $\vec r_{\text{core}}$ from the core.
We refer to this tube as a \emph{twisted torus}.
As seen in \figref{fig:twistedtorus}, our method preserves the twisted vortical structures at the end of the simulation, while all other methods have lost their structure.
Solely inspecting the energy plot, one may conclude that such methods provide a high-quality result.
But when inspected visually, as seen in \figref{fig:twistedtorus}, we see a large loss of vorticity in other methods as time progresses.
This can be explained by complementing the analysis with the helicity plot.
Our method conserves energy exactly, and only shows a small error (5\%) in maintaining helicity, while other methods lose their helicity more significantly (around 13\% at best for R+PolyFLIP).

\subsubsection{3D Unknot}
Here, we focus on a (1,5)-torus knot, which we refer to as an ``unknot'' following the convention from \cite{Maggioni:2010:VEH}.
We set up our experiment using the following construction for a (p,q)-torus knot formula in $\mathbb{R}^3$:
\begin{equation}
    x = r \cos(p\theta),\\
    y = r \sin(p\theta),\\
    z = - \sin(q\theta),
\end{equation}
where $r=cos(q\theta)+2$ and $\theta \in [0,2\pi]$.
As the vortex unknot travels forward in time, the structure stretches to a point where a reconnection event occurs, and five other smaller rings are created from the bigger unknot (see \figref{fig:unknot_evolution}).
As seen in \figref{fig:unknot}, our method (CO-FLIP) captures this even on a low resolution of $64^3$, while all other methods fail to achieve a clean separation.
Additionally, note that the vortex strength is kept at the initial maximum for our results, while all other methods have lost most of their vorticity.

\subsubsection{3D Trefoil Knot}
We further validate our experiments by running the trefoil knot experiment, which sets up a (2,3)-torus knot in space.
We expect that the knot will evolve and break into two separate knots.
While other methods fail to fully undergo the vortex reconnection event, our method achieves this and maintains its vorticity strength, as seen from the color of the vortex filaments in \figref{fig:trefoilknot}.
This evolution is captured by our method in \figref{fig:trefoilknot_evolution} at a low resolution of $64^3$.

\subsubsection{3D Vortex Leapfrogging}
Similar to the two-dimensional case, the vortex leapfrogging is expected to continue indefinitely under no viscous forces.
However, the 3D experiment has proven to be more difficult to maintain this ever leaping behavior.
As shown in \figref{fig:leapfrog_rings}, at low resolutions such as $128\times64\times64$, previous methods fail to recover the correct behavior and, at best, leap twice (see R+PolyFLIP).
Our method (CO-FLIP) achieves five leaps before merging into one another.
Note that at even lower resolutions, previously not seen before, our method manages to leap once before merging together.
Based on this trend, we expect that higher resolutions would allow for higher leap counts.

\subsection{Demonstrations}

We also demonstrate that our method can create visually stunning simulations of phenomena seen in nature that are often of interest to the graphics community using relatively low resolutions (See \tabref{tab:Statistics}).

\subsubsection{2D Rayleigh Taylor Instability}

\figref{fig:RT_instability} shows that our method in simulating 2D Rayleigh Taylor instability.
We set up an experiment with one Rayleigh Taylor finger.
Heavy fluid sits above light fluid and travels downward, creating fractal-shaped vortices.
We set a small viscosity of $10^{-7}\,m^2/s$ and buoyancy acceleration $0.15\,m/s^2$, and run the simulation with a small sinusoidal perturbation at the interface given by $\frac{L}{2} + 0.05W(\cos (\frac{2\pi}{W})-1)$, where $L, W$ are the length and height of the domain.
Note that our method can create fractal-shaped vortical structures that better match the expected output at very high Reynolds numbers.
While other methods produce a modest amount of vorticities, our method produces abundantly more features at the same resolution.

\subsubsection{2D Smoke plume}
To additionally showcase the intricate vortical structures produced by our CO-FLIP algorithm, we set up an experiment where buoyant hot travels upward with acceleration $0.1\,m/s^2$ in two dimensions.
We apply a small viscosity of $10^{-6}\,m^2/s$. 
\figref{fig:2dsmokeplume} shows how over time, our method produces an increasing amount of turbulent fractal-shaped vortices.

\subsubsection{3D Smoke plume}
\figref{fig:smokeplumes} shows an experiment with a smoke plume under buoyancy $1\,m/s^2$ and a small viscosity of $10^{-4}\,m^2/s$.
Our method captures more detail and vortices than other methods at a low resolution of $64\times128\times64$.
Our results at this resolution are comparable to traditional methods running at double the resolution.
Even at a lower resolution of $32\times64\times32$, previously considered as too coarse, our method can create highly energetic results.
As the resolution drops, CO-FLIP provides a consistent qualitative visual result with a rich level of vortical structures.  This is because CO-FLIP does not have any numerical dissipation. Regardless of the resolution, the difference in the behavior is due to how far the $\cI$-discrete Euler equations are from the continuous Euler equations.

\subsubsection{3D Pyroclastic}
We further demonstrate the effectiveness of our high-order structure-preserving solver through the pyroclastic cloud \figref{fig:pyroclastic}.
By injecting hot, buoyant smoke in the shape of a flat slab under buoyancy acceleration of $0.1\,m/s^2$, our method produces results that are comparable in terms of the amount of vortical details to real footage of volcanic clouds emerging after their eruption. 
Note that our method runs at only $96\times192\times96$ simulation, which was previously considered too low for such effects.
As shown in \figref{fig:pyro_comparison}, even at a lower resolution of $64\times128\times64$, we see energetic dense vortical structures.

\begin{figure}
    \centering
    \includegraphics[{trim=95px 0 95px 0px},clip,width=0.49\columnwidth]{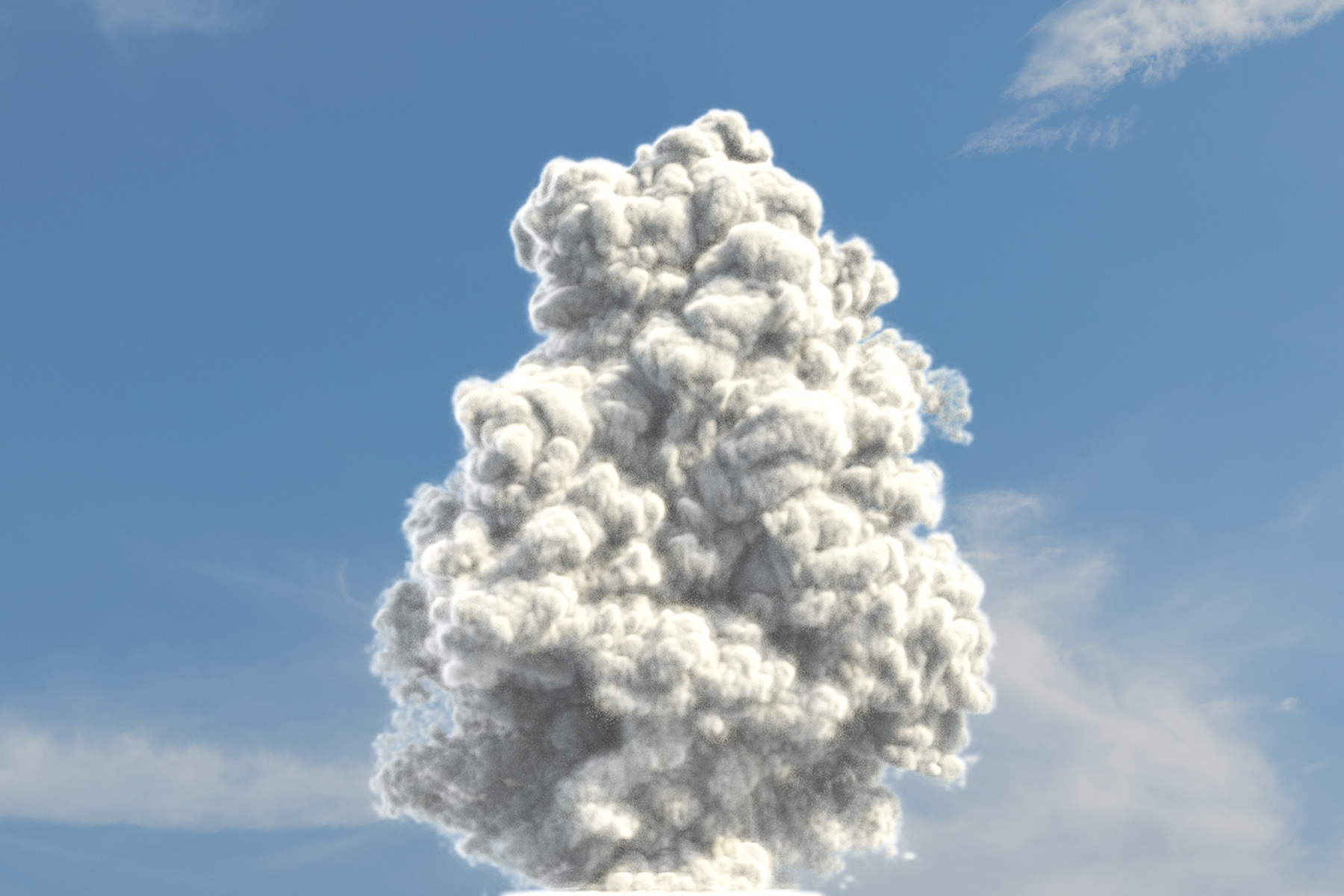}
    \includegraphics[{trim=95px 0 95px 0px},clip,width=0.49\columnwidth]{figures/Results/3D/pyroclastic/newcam/recolor_pyroclastic_res96.jpg}
    \begin{picture}(0,0)(0,0)
        \put(-241,2){\sffamily \scriptsize \rotatebox{90}{\textcolor{white}{$64\times128\times64$}}}
        \put(-120,2){\sffamily \scriptsize \rotatebox{90}{\textcolor{white}{$96\times192\times96$}}}
    \end{picture}
    \caption{A comparison of the pyroclastic experiment for different resolutions.  Note as the spatiotemporal resolution increases, at a constant CFL number, we see more dense and energertic vortical structures emerge.}
    \label{fig:pyro_comparison}
\end{figure}

\subsubsection{3D Inkjet}
We simulate the effect of injecting ink with a high speed of $2\,m/s$ into a medium in \figref{fig:inkjet}.
The ink evolves under a buoyancy acceleration of $0.1\,m/s^2$, and a small viscosity of $1\times10^{-4}\,m^2/s$.
Note the interplay of two smoke jets as they reach one another, mix, and produce many vortical structures while running at a low resolution of $64\times128\times64$.

\subsubsection{3D Rocket}
To further showcase our method, we run an experiment inspired by rocket launch footage.
We stage a nozzle that steadily emits hot smoke into the scene at a speed of  $0.9\,m/s$, where the effective buoyancy acceleration is $0.55\,m/s^2$ point upwards.
We additionally use a small amount of viscosity of $1\times10^{-4}\,m^2/s$.
\figref{fig:rocket} demonstrates the abundance of turbulent fuel and smoke left in the wake of the rocket flying upwards.
Note that the rocket model and the smoke are not coupled; this is purely for the final presentation of the simulation.
The simulation runs at a low resolution of $64^3$.

\subsubsection{3D Spot Obstacle}
We implement arbitrary boundary handling as described in \secref{sec:Obstacle}.
\figref{fig:spot} shows how our method produces intricate vortical structures at a low resolution of $128\times64\times64$, as the smoke jet coming out of the nozzle interacts with the obstacle shaped by the Spot model \cite{Crane:2013:RFC}.
The nozzle jet speed for this experiment is set to $2\,m/s$, and we use a small amount of viscosity of $1\times10^{-4}\,m^2/s$.

\section{Conclusion and future work}
This paper details the construction of a high-order energy- and Casimirs-preserving discretization of the Euler equations for the dynamics of inviscid incompressible fluids in a FLIP-based hybrid Eulerian--Lagrangian framework.
By a few key novel modifications to the FLIP algorithm, we establish a discrete Hamiltonian dynamical system.  The Hamiltonian flow is based on an energy defined on the discrete grid velocity field, which sets motion on a symplectic space represented by a particle system.
We employ a mimetic B-spline interpolation scheme which gives us pointwise divergence-free velocities as well as volume-preserving actions on the particles. 
This action is symplectomorphic on the particle system, and as such induces a momentum map back to the grid, which perfectly matches the P2G transfer operator in the FLIP algorithm.
Since momentum maps are Poisson maps, the actions preserve the coadjoint orbits, and thereby Casimirs, in the infinite-dimensional dual space of divergence-free velocity fields.
Finally, the pressure projection uses an inner product structure given by the induced metric through the interpolation operator.  This choice for the pressure projection is crucial for discrete energy conservation. 
As each component of the solver is canonically derived from the interpolation, the entire system can be made arbitrarily high-order by only specifying the order of accuracy of the interpolation operator.

We demonstrate that our high-order solver is capable of producing and maintaining high levels of detail in fluid flows.
Even using a low resolution grid, the state on the coadjoint orbit can still develop turbulent flow with high resolution vortical structures and maintain concentrated vortex filaments.
In contrast, competing numerical methods cannot capture the same level of detail at low resolutions.
The mathematical reason is that, in our method, the grid resolution only affects the Hamiltonian, whereas the Poisson structure that describes the advective nature for the momentum is still the same as the one in the continuous theory.  
We also elucidate that the Casimir preservation is expected to be at the infinite-dimensional dual space of continuous velocities instead of on the dual space of the finite dimensional grid velocities suggested by previous structure-preserving discretizations of Euler equations.  

The implicit time integration in the method facilitates long-term stability.
Previous stability issues in FLIP-based algorithms are minimized when the mimetic interpolation is consistently applied to pressure force and the P2G transfer.

While our method provides higher orders of accuracy and structure preservation in simulating fluids, it comes with a higher computational cost.
Currently, the performance bottleneck lies in the pseudoinverse solve in the P2G transfer.  
For future work, we aim to construct more efficient solvers that can tackle the bottlenecks of the algorithm.
Another avenue is to extend the CO-FLIP machinery to free-surface flows, requiring careful handing of the Galerkin Hodge star near the free-surface boundary.
We are also excited to explore other possible fluid solvers by following a similar derivation of CO-FLIP (\secref{sec:Theory1}) but with a different mimetic velocity representation and a different auxiliary symplectic space.

\begin{acks}
This work was funded in part by the Ronald L. Graham Chair, NSF Career Award 2239062, and the UC San Diego Center for Visual Computing.  Additional support was provided by SideFX Software and Activision Blizzard.
The work used compute resources provided by the National Research Platform (funded by NSF grants 2100237 and 2120019), which Khai Vu helped the authors to set up.
Moreover, the authors would like to thank Mukund Varma for sharing his desktop machine as an additional computational resource. 
Finally, the authors would like to thank Professors Melvin Leok and Craig Schroeder, and Dr. Alexey Stomakhin for insightful discussions.
\end{acks}

\bibliographystyle{ACM-Reference-Format}
\bibliography{CoadjointFLIP}


\providecommand{\noopsort}[1]{}
\begin{thebibliography}{126}


\ifx \showCODEN    \undefined \def \showCODEN     #1{\unskip}     \fi
\ifx \showDOI      \undefined \def \showDOI       #1{#1}\fi
\ifx \showISBNx    \undefined \def \showISBNx     #1{\unskip}     \fi
\ifx \showISBNxiii \undefined \def \showISBNxiii  #1{\unskip}     \fi
\ifx \showISSN     \undefined \def \showISSN      #1{\unskip}     \fi
\ifx \showLCCN     \undefined \def \showLCCN      #1{\unskip}     \fi
\ifx \shownote     \undefined \def \shownote      #1{#1}          \fi
\ifx \showarticletitle \undefined \def \showarticletitle #1{#1}   \fi
\ifx \showURL      \undefined \def \showURL       {\relax}        \fi
\providecommand\bibfield[2]{#2}
\providecommand\bibinfo[2]{#2}
\providecommand\natexlab[1]{#1}
\providecommand\showeprint[2][]{arXiv:#2}

\bibitem[Aanjaneya et~al\mbox{.}(2019)]%
        {Aanjaneya:2019:EGMVL}
\bibfield{author}{\bibinfo{person}{Mridul Aanjaneya}, \bibinfo{person}{Chengguizi Han}, \bibinfo{person}{Ryan Goldade}, {and} \bibinfo{person}{Christopher Batty}.} \bibinfo{year}{2019}\natexlab{}.
\newblock \showarticletitle{An Efficient Geometric Multigrid Solver for Viscous Liquids}.
\newblock \bibinfo{journal}{\emph{Proceedings of the ACM in Computer Graphics and Interactive Techniques}} \bibinfo{volume}{2}, \bibinfo{number}{2}, Article \bibinfo{articleno}{14} (\bibinfo{date}{July} \bibinfo{year}{2019}), \bibinfo{numpages}{21}~pages.
\newblock


\bibitem[Adams et~al\mbox{.}(2003)]%
        {Adams:2003:PMS}
\bibfield{author}{\bibinfo{person}{Mark Adams}, \bibinfo{person}{Marian Brezina}, \bibinfo{person}{Jonathan Hu}, {and} \bibinfo{person}{Ray Tuminaro}.} \bibinfo{year}{2003}\natexlab{}.
\newblock \showarticletitle{Parallel multigrid smoothing: polynomial versus Gauss--Seidel}.
\newblock \bibinfo{journal}{\emph{J. Comput. Phys.}} \bibinfo{volume}{188}, \bibinfo{number}{2} (\bibinfo{year}{2003}), \bibinfo{pages}{593--610}.
\newblock


\bibitem[Arnold et~al\mbox{.}(2006)]%
        {Arnold:2006:FEEC}
\bibfield{author}{\bibinfo{person}{Douglas~N Arnold}, \bibinfo{person}{Richard~S Falk}, {and} \bibinfo{person}{Ragnar Winther}.} \bibinfo{year}{2006}\natexlab{}.
\newblock \showarticletitle{Finite element exterior calculus, homological techniques, and applications}.
\newblock \bibinfo{journal}{\emph{Acta numerica}}  \bibinfo{volume}{15} (\bibinfo{year}{2006}), \bibinfo{pages}{1--155}.
\newblock


\bibitem[Arnold(1966)]%
        {Arnold:1966:GDG}
\bibfield{author}{\bibinfo{person}{Vladimir~I. Arnold}.} \bibinfo{year}{1966}\natexlab{}.
\newblock \showarticletitle{Sur la g{\'e}om{\'e}trie diff{\'e}rentielle des groupes de Lie de dimension infinie et ses applications {\`a} l'hydrodynamique des fluides parfaits}. In \bibinfo{booktitle}{\emph{Annales de l'institut Fourier}}, Vol.~\bibinfo{volume}{16}. \bibinfo{pages}{319--361}.
\newblock


\bibitem[Arnold and Khesin(1998)]%
        {Arnold:1998:TMH}
\bibfield{author}{\bibinfo{person}{Vladimir~I. Arnold} {and} \bibinfo{person}{Boris~A. Khesin}.} \bibinfo{year}{1998}\natexlab{}.
\newblock \bibinfo{booktitle}{\emph{\href{http://www.springer.com/us/book/9780387949475}{Topological Methods in Hydrodynamics}}}.
\newblock \bibinfo{publisher}{Springer}.
\newblock


\bibitem[Azencot et~al\mbox{.}(2013)]%
        {Azencot:2013:OAT}
\bibfield{author}{\bibinfo{person}{Omri Azencot}, \bibinfo{person}{Mirela Ben-Chen}, \bibinfo{person}{Fr{\'e}d{\'e}ric Chazal}, {and} \bibinfo{person}{Maks Ovsjanikov}.} \bibinfo{year}{2013}\natexlab{}.
\newblock \showarticletitle{An operator approach to tangent vector field processing}. In \bibinfo{booktitle}{\emph{Computer Graphics Forum}}, Vol.~\bibinfo{volume}{32}. Wiley Online Library, \bibinfo{pages}{73--82}.
\newblock


\bibitem[Azencot et~al\mbox{.}(2015)]%
        {Azencot:2015:DDV}
\bibfield{author}{\bibinfo{person}{Omri Azencot}, \bibinfo{person}{Maks Ovsjanikov}, \bibinfo{person}{Fr{\'e}d{\'e}ric Chazal}, {and} \bibinfo{person}{Mirela Ben-Chen}.} \bibinfo{year}{2015}\natexlab{}.
\newblock \showarticletitle{Discrete derivatives of vector fields on surfaces--an operator approach}.
\newblock \bibinfo{journal}{\emph{ACM Transactions on Graphics (TOG)}} \bibinfo{volume}{34}, \bibinfo{number}{3} (\bibinfo{year}{2015}), \bibinfo{pages}{1--13}.
\newblock


\bibitem[Azencot et~al\mbox{.}(2014)]%
        {Azencot:2014:FFS}
\bibfield{author}{\bibinfo{person}{Omri Azencot}, \bibinfo{person}{Steffen Wei\ss{}mann}, \bibinfo{person}{Maks Ovsjanikov}, \bibinfo{person}{Max Wardetzky}, {and} \bibinfo{person}{Mirela Ben-Chen}.} \bibinfo{year}{2014}\natexlab{}.
\newblock \showarticletitle{Functional Fluids on Surfaces}. In \bibinfo{booktitle}{\emph{Proceedings of the Symposium on Geometry Processing}} (Cardiff, United Kingdom) \emph{(\bibinfo{series}{SGP '14})}. \bibinfo{publisher}{Eurographics Association}, \bibinfo{address}{Goslar, DEU}, \bibinfo{pages}{237–246}.
\newblock
\urldef\tempurl%
\url{https://doi.org/10.1111/cgf.12449}
\showDOI{\tempurl}


\bibitem[Bell(2008)]%
        {Bell:2008:AMG}
\bibfield{author}{\bibinfo{person}{William~N Bell}.} \bibinfo{year}{2008}\natexlab{}.
\newblock \showarticletitle{Algebraic multigrid for discrete differential forms}.
\newblock  (\bibinfo{year}{2008}).
\newblock


\bibitem[Brackbill and Ruppel(1986)]%
        {Brackbill:1986:FLIP}
\bibfield{author}{\bibinfo{person}{J.U. Brackbill} {and} \bibinfo{person}{H.M. Ruppel}.} \bibinfo{year}{1986}\natexlab{}.
\newblock \showarticletitle{FLIP: A method for adaptively zoned, particle-in-cell calculations of fluid flows in two dimensions}.
\newblock \bibinfo{journal}{\emph{J. Comput. Phys.}} \bibinfo{volume}{65}, \bibinfo{number}{2} (\bibinfo{year}{1986}), \bibinfo{pages}{314--343}.
\newblock
\showISSN{0021-9991}
\urldef\tempurl%
\url{https://doi.org/10.1016/0021-9991(86)90211-1}
\showDOI{\tempurl}


\bibitem[Brackbill et~al\mbox{.}(1988)]%
        {Brackbill:1988:FLIP}
\bibfield{author}{\bibinfo{person}{Jeremiah~U Brackbill}, \bibinfo{person}{Douglas~B Kothe}, {and} \bibinfo{person}{Hans~M Ruppel}.} \bibinfo{year}{1988}\natexlab{}.
\newblock \showarticletitle{FLIP: a low-dissipation, particle-in-cell method for fluid flow}.
\newblock \bibinfo{journal}{\emph{Computer Physics Communications}} \bibinfo{volume}{48}, \bibinfo{number}{1} (\bibinfo{year}{1988}), \bibinfo{pages}{25--38}.
\newblock


\bibitem[Buffa et~al\mbox{.}(2011)]%
        {Buffa:2011:IGD}
\bibfield{author}{\bibinfo{person}{Annalisa Buffa}, \bibinfo{person}{Judith Rivas}, \bibinfo{person}{Giancarlo Sangalli}, {and} \bibinfo{person}{Rafael V{\'a}zquez}.} \bibinfo{year}{2011}\natexlab{}.
\newblock \showarticletitle{Isogeometric discrete differential forms in three dimensions}.
\newblock \bibinfo{journal}{\emph{SIAM J. Numer. Anal.}} \bibinfo{volume}{49}, \bibinfo{number}{2} (\bibinfo{year}{2011}), \bibinfo{pages}{818--844}.
\newblock


\bibitem[Buffa et~al\mbox{.}(2010)]%
        {Buffa:2010:IGA}
\bibfield{author}{\bibinfo{person}{Annalisa Buffa}, \bibinfo{person}{Giancarlo Sangalli}, {and} \bibinfo{person}{Rafael V{\'a}zquez}.} \bibinfo{year}{2010}\natexlab{}.
\newblock \showarticletitle{Isogeometric analysis in electromagnetics: B-splines approximation}.
\newblock \bibinfo{journal}{\emph{Computer Methods in Applied Mechanics and Engineering}} \bibinfo{volume}{199}, \bibinfo{number}{17-20} (\bibinfo{year}{2010}), \bibinfo{pages}{1143--1152}.
\newblock


\bibitem[Celledoni et~al\mbox{.}(2014)]%
        {Celledoni:2014:LGI}
\bibfield{author}{\bibinfo{person}{Elena Celledoni}, \bibinfo{person}{H{\aa}kon Marthinsen}, {and} \bibinfo{person}{Brynjulf Owren}.} \bibinfo{year}{2014}\natexlab{}.
\newblock \showarticletitle{An introduction to Lie group integrators--basics, new developments and applications}.
\newblock \bibinfo{journal}{\emph{J. Comput. Phys.}}  \bibinfo{volume}{257} (\bibinfo{year}{2014}), \bibinfo{pages}{1040--1061}.
\newblock


\bibitem[Chang et~al\mbox{.}(2022)]%
        {Chang:2022:CurlFlow}
\bibfield{author}{\bibinfo{person}{Jumyung Chang}, \bibinfo{person}{Ruben Partono}, \bibinfo{person}{Vinicius~C Azevedo}, {and} \bibinfo{person}{Christopher Batty}.} \bibinfo{year}{2022}\natexlab{}.
\newblock \showarticletitle{Curl-Flow: Boundary-Respecting Pointwise Incompressible Velocity Interpolation for Grid-Based Fluids}.
\newblock \bibinfo{journal}{\emph{ACM Transactions on Graphics (TOG)}} \bibinfo{volume}{41}, \bibinfo{number}{6} (\bibinfo{year}{2022}), \bibinfo{pages}{243:1--243:21}.
\newblock


\bibitem[Chen and Hou(2022)]%
        {Chen:2022:SSS}
\bibfield{author}{\bibinfo{person}{Jiajie Chen} {and} \bibinfo{person}{Thomas~Y Hou}.} \bibinfo{year}{2022}\natexlab{}.
\newblock \showarticletitle{Stable nearly self-similar blowup of the 2D Boussinesq and 3D Euler equations with smooth data I: Analysis}.
\newblock \bibinfo{journal}{\emph{arXiv preprint arXiv:2210.07191}} (\bibinfo{year}{2022}).
\newblock


\bibitem[Chern et~al\mbox{.}(2017)]%
        {Chern:2017:IF}
\bibfield{author}{\bibinfo{person}{Albert Chern}, \bibinfo{person}{Felix Kn{\"o}ppel}, \bibinfo{person}{Ulrich Pinkall}, {and} \bibinfo{person}{Peter Schr{\"o}der}.} \bibinfo{year}{2017}\natexlab{}.
\newblock \showarticletitle{Inside fluids: Clebsch maps for visualization and processing}.
\newblock \bibinfo{journal}{\emph{ACM Transactions on Graphics (TOG)}} \bibinfo{volume}{36}, \bibinfo{number}{4} (\bibinfo{year}{2017}), \bibinfo{pages}{142:1--142:11}.
\newblock


\bibitem[Chern et~al\mbox{.}(2016)]%
        {Chern:2016:SS}
\bibfield{author}{\bibinfo{person}{Albert Chern}, \bibinfo{person}{Felix Kn{\"o}ppel}, \bibinfo{person}{Ulrich Pinkall}, \bibinfo{person}{Peter Schr{\"o}der}, {and} \bibinfo{person}{Steffen Wei{\ss}mann}.} \bibinfo{year}{2016}\natexlab{}.
\newblock \showarticletitle{Schr{\"o}dinger's smoke}.
\newblock \bibinfo{journal}{\emph{ACM Transactions on Graphics (TOG)}} \bibinfo{volume}{35}, \bibinfo{number}{4} (\bibinfo{year}{2016}), \bibinfo{pages}{77:1--77:13}.
\newblock


\bibitem[Clebsch(1859)]%
        {Clebsch:1859:IHG}
\bibfield{author}{\bibinfo{person}{A. Clebsch}.} \bibinfo{year}{1859}\natexlab{}.
\newblock \showarticletitle{Ueber die Integration der hydrodynamischen Gleichungen}.
\newblock \bibinfo{journal}{\emph{Journal f\"ur die reine und angewandte Mathematik}}  \bibinfo{volume}{56} (\bibinfo{year}{1859}), \bibinfo{pages}{1--10}.
\newblock
\newblock
\shownote{English translation by {D.\ H.\ Delphenich}, \url{http://www.neo-classical-physics.info/uploads/3/4/3/6/34363841/clebsch_-_clebsch_variables.pdf}}.


\bibitem[Cornelis et~al\mbox{.}(2014)]%
        {Cornelius:2014:IISPH-FLIP}
\bibfield{author}{\bibinfo{person}{Jens Cornelis}, \bibinfo{person}{Markus Ihmsen}, \bibinfo{person}{Andreas Peer}, {and} \bibinfo{person}{Matthias Teschner}.} \bibinfo{year}{2014}\natexlab{}.
\newblock \showarticletitle{IISPH-FLIP for incompressible fluids}. In \bibinfo{booktitle}{\emph{Computer Graphics Forum}}, Vol.~\bibinfo{volume}{33}. Wiley Online Library, \bibinfo{pages}{255--262}.
\newblock


\bibitem[Cortez(1995)]%
        {Cortez:1995:IBM}
\bibfield{author}{\bibinfo{person}{Ricardo Cortez}.} \bibinfo{year}{1995}\natexlab{}.
\newblock \bibinfo{title}{Impulse-based methods for fluid flow}.
\newblock
\newblock
\urldef\tempurl%
\url{https://doi.org/10.2172/87798}
\showDOI{\tempurl}


\bibitem[Courant et~al\mbox{.}(1928)]%
        {Courant:1928:CFL}
\bibfield{author}{\bibinfo{person}{Richard Courant}, \bibinfo{person}{Kurt Friedrichs}, {and} \bibinfo{person}{Hans Lewy}.} \bibinfo{year}{1928}\natexlab{}.
\newblock \showarticletitle{{\"U}ber die partiellen Differenzengleichungen der mathematischen Physik}.
\newblock \bibinfo{journal}{\emph{Mathematische annalen}} \bibinfo{volume}{100}, \bibinfo{number}{1} (\bibinfo{year}{1928}), \bibinfo{pages}{32--74}.
\newblock


\bibitem[Cox(1972)]%
        {Cox:1972:NEB}
\bibfield{author}{\bibinfo{person}{Maurice~G Cox}.} \bibinfo{year}{1972}\natexlab{}.
\newblock \showarticletitle{The numerical evaluation of B-splines}.
\newblock \bibinfo{journal}{\emph{IMA Journal of Applied mathematics}} \bibinfo{volume}{10}, \bibinfo{number}{2} (\bibinfo{year}{1972}), \bibinfo{pages}{134--149}.
\newblock


\bibitem[Crane et~al\mbox{.}(2013a)]%
        {Crane:2013:DEC}
\bibfield{author}{\bibinfo{person}{Keenan Crane}, \bibinfo{person}{Fernando de Goes}, \bibinfo{person}{Mathieu Desbrun}, {and} \bibinfo{person}{Peter Schr\"{o}der}.} \bibinfo{year}{2013}\natexlab{a}.
\newblock \showarticletitle{Digital geometry processing with discrete exterior calculus}. In \bibinfo{booktitle}{\emph{ACM SIGGRAPH 2013 Courses}} (Anaheim, California) \emph{(\bibinfo{series}{SIGGRAPH '13})}. \bibinfo{publisher}{Association for Computing Machinery}, \bibinfo{address}{New York, NY, USA}, Article \bibinfo{articleno}{7}, \bibinfo{numpages}{126}~pages.
\newblock
\showISBNx{9781450323390}
\urldef\tempurl%
\url{https://doi.org/10.1145/2504435.2504442}
\showDOI{\tempurl}


\bibitem[Crane et~al\mbox{.}(2013b)]%
        {Crane:2013:RFC}
\bibfield{author}{\bibinfo{person}{Keenan Crane}, \bibinfo{person}{Ulrich Pinkall}, {and} \bibinfo{person}{Peter Schr{\"o}der}.} \bibinfo{year}{2013}\natexlab{b}.
\newblock \showarticletitle{Robust fairing via conformal curvature flow}.
\newblock \bibinfo{journal}{\emph{ACM Transactions on Graphics (TOG)}} \bibinfo{volume}{32}, \bibinfo{number}{4} (\bibinfo{year}{2013}), \bibinfo{pages}{1--10}.
\newblock


\bibitem[Curry and Schoenberg(1947)]%
        {Curry:1947:SDL}
\bibfield{author}{\bibinfo{person}{Haskell~Brooks Curry} {and} \bibinfo{person}{Isaac~Jacob Schoenberg}.} \bibinfo{year}{1947}\natexlab{}.
\newblock \showarticletitle{On spline distributions and their limits-the polya distribution functions}. In \bibinfo{booktitle}{\emph{Bulletin of the American Mathematical Society}}, Vol.~\bibinfo{volume}{53}. AMER MATHEMATICAL SOC 201 CHARLES ST, PROVIDENCE, RI 02940-2213, \bibinfo{pages}{1114--1114}.
\newblock


\bibitem[Curry and Schoenberg(1966)]%
        {Curry:1966:PFF}
\bibfield{author}{\bibinfo{person}{Haskell~Brooks Curry} {and} \bibinfo{person}{Isaac~J Schoenberg}.} \bibinfo{year}{1966}\natexlab{}.
\newblock \showarticletitle{On P{\'o}lya frequency functions IV: the fundamental spline functions and their limits}.
\newblock \bibinfo{journal}{\emph{J. Analyse math}} \bibinfo{volume}{17}, \bibinfo{number}{71} (\bibinfo{year}{1966}), \bibinfo{pages}{107}.
\newblock


\bibitem[De~Boor(1972)]%
        {DeBoor:1972:CBS}
\bibfield{author}{\bibinfo{person}{Carl De~Boor}.} \bibinfo{year}{1972}\natexlab{}.
\newblock \showarticletitle{On calculating with B-splines}.
\newblock \bibinfo{journal}{\emph{Journal of Approximation theory}} \bibinfo{volume}{6}, \bibinfo{number}{1} (\bibinfo{year}{1972}), \bibinfo{pages}{50--62}.
\newblock


\bibitem[De~Goes et~al\mbox{.}(2016)]%
        {DeGoes:2016:SEC}
\bibfield{author}{\bibinfo{person}{Fernando De~Goes}, \bibinfo{person}{Mathieu Desbrun}, \bibinfo{person}{Mark Meyer}, {and} \bibinfo{person}{Tony DeRose}.} \bibinfo{year}{2016}\natexlab{}.
\newblock \showarticletitle{Subdivision exterior calculus for geometry processing}.
\newblock \bibinfo{journal}{\emph{ACM Transactions on Graphics (TOG)}} \bibinfo{volume}{35}, \bibinfo{number}{4} (\bibinfo{year}{2016}), \bibinfo{pages}{1--11}.
\newblock


\bibitem[Deng et~al\mbox{.}(2023)]%
        {Deng:2023:FSN}
\bibfield{author}{\bibinfo{person}{Yitong Deng}, \bibinfo{person}{Hong-Xing Yu}, \bibinfo{person}{Diyang Zhang}, \bibinfo{person}{Jiajun Wu}, {and} \bibinfo{person}{Bo Zhu}.} \bibinfo{year}{2023}\natexlab{}.
\newblock \showarticletitle{Fluid Simulation on Neural Flow Maps}.
\newblock \bibinfo{journal}{\emph{ACM Transactions on Graphics (TOG)}} \bibinfo{volume}{42}, \bibinfo{number}{6} (\bibinfo{year}{2023}), \bibinfo{pages}{1--21}.
\newblock


\bibitem[Ding et~al\mbox{.}(2020)]%
        {Ding:2020:APIC-MAC}
\bibfield{author}{\bibinfo{person}{Ounan Ding}, \bibinfo{person}{Tamar Shinar}, {and} \bibinfo{person}{Craig Schroeder}.} \bibinfo{year}{2020}\natexlab{}.
\newblock \showarticletitle{Affine particle in cell method for MAC grids and fluid simulation}.
\newblock \bibinfo{journal}{\emph{J. Comput. Phys.}}  \bibinfo{volume}{408} (\bibinfo{year}{2020}), \bibinfo{pages}{109311}.
\newblock


\bibitem[Drake et~al\mbox{.}(2021)]%
        {Drake:2021:RBF}
\bibfield{author}{\bibinfo{person}{Kathryn~P Drake}, \bibinfo{person}{Edward~J Fuselier}, {and} \bibinfo{person}{Grady~B Wright}.} \bibinfo{year}{2021}\natexlab{}.
\newblock \showarticletitle{A partition of unity method for divergence-free or curl-free radial basis function approximation}.
\newblock \bibinfo{journal}{\emph{SIAM Journal on Scientific Computing}} \bibinfo{volume}{43}, \bibinfo{number}{3} (\bibinfo{year}{2021}), \bibinfo{pages}{A1950--A1974}.
\newblock


\bibitem[Dupont and Liu(2003)]%
        {Dupont:2003:BFECC}
\bibfield{author}{\bibinfo{person}{Todd~F Dupont} {and} \bibinfo{person}{Yingjie Liu}.} \bibinfo{year}{2003}\natexlab{}.
\newblock \showarticletitle{Back and forth error compensation and correction methods for removing errors induced by uneven gradients of the level set function}.
\newblock \bibinfo{journal}{\emph{J. Comput. Phys.}} \bibinfo{volume}{190}, \bibinfo{number}{1} (\bibinfo{year}{2003}), \bibinfo{pages}{311--324}.
\newblock


\bibitem[Ebin and Marsden(1970)]%
        {Ebin:1970:GOD}
\bibfield{author}{\bibinfo{person}{David~G Ebin} {and} \bibinfo{person}{Jerrold Marsden}.} \bibinfo{year}{1970}\natexlab{}.
\newblock \showarticletitle{Groups of diffeomorphisms and the motion of an incompressible fluid}.
\newblock \bibinfo{journal}{\emph{Annals of Mathematics}} (\bibinfo{year}{1970}), \bibinfo{pages}{102--163}.
\newblock


\bibitem[Edwards and Bridson(2012)]%
        {Edwards:2012:HOPIC}
\bibfield{author}{\bibinfo{person}{Essex Edwards} {and} \bibinfo{person}{Robert Bridson}.} \bibinfo{year}{2012}\natexlab{}.
\newblock \showarticletitle{A high-order accurate particle-in-cell method}.
\newblock \bibinfo{journal}{\emph{Internat. J. Numer. Methods Engrg.}} \bibinfo{volume}{90}, \bibinfo{number}{9} (\bibinfo{year}{2012}), \bibinfo{pages}{1073--1088}.
\newblock


\bibitem[Elcott et~al\mbox{.}(2007)]%
        {Elcott:2007:SCP}
\bibfield{author}{\bibinfo{person}{Sharif Elcott}, \bibinfo{person}{Yiying Tong}, \bibinfo{person}{Eva Kanso}, \bibinfo{person}{Peter Schr{\"o}der}, {and} \bibinfo{person}{Mathieu Desbrun}.} \bibinfo{year}{2007}\natexlab{}.
\newblock \showarticletitle{Stable, circulation-preserving, simplicial fluids}.
\newblock \bibinfo{journal}{\emph{ACM Transactions on Graphics (TOG)}} \bibinfo{volume}{26}, \bibinfo{number}{1} (\bibinfo{year}{2007}), \bibinfo{pages}{4--es}.
\newblock


\bibitem[Elgindi et~al\mbox{.}(2019)]%
        {Elgindi:2019:SSS}
\bibfield{author}{\bibinfo{person}{Tarek~M Elgindi}, \bibinfo{person}{Tej-Eddine Ghoul}, {and} \bibinfo{person}{Nader Masmoudi}.} \bibinfo{year}{2019}\natexlab{}.
\newblock \showarticletitle{On the stability of self-similar blow-up for $C^{1,\alpha}$ solutions to the incompressible Euler equations on $\mathbb{R}^3$}.
\newblock \bibinfo{journal}{\emph{arXiv preprint arXiv:1910.14071}} (\bibinfo{year}{2019}).
\newblock


\bibitem[Eng{\o} and Faltinsen(2001)]%
        {Engo:2001:NILP}
\bibfield{author}{\bibinfo{person}{Kenth Eng{\o}} {and} \bibinfo{person}{Stig Faltinsen}.} \bibinfo{year}{2001}\natexlab{}.
\newblock \showarticletitle{Numerical Integration of Lie--Poisson Systems While Preserving Coadjoint Orbits and Energy}.
\newblock \bibinfo{journal}{\emph{SIAM journal on numerical analysis}} \bibinfo{volume}{39}, \bibinfo{number}{1} (\bibinfo{year}{2001}), \bibinfo{pages}{128--145}.
\newblock


\bibitem[Evans and Hughes(2013)]%
        {Evans:2013:IGA}
\bibfield{author}{\bibinfo{person}{John~A Evans} {and} \bibinfo{person}{Thomas~JR Hughes}.} \bibinfo{year}{2013}\natexlab{}.
\newblock \showarticletitle{Isogeometric divergence-conforming B-splines for the unsteady Navier--Stokes equations}.
\newblock \bibinfo{journal}{\emph{J. Comput. Phys.}}  \bibinfo{volume}{241} (\bibinfo{year}{2013}), \bibinfo{pages}{141--167}.
\newblock


\bibitem[Fang et~al\mbox{.}(2018)]%
        {Fang:2018:TAM}
\bibfield{author}{\bibinfo{person}{Yu Fang}, \bibinfo{person}{Yuanming Hu}, \bibinfo{person}{Shi-Min Hu}, {and} \bibinfo{person}{Chenfanfu Jiang}.} \bibinfo{year}{2018}\natexlab{}.
\newblock \showarticletitle{A temporally adaptive material point method with regional time stepping}. In \bibinfo{booktitle}{\emph{Computer graphics forum}}, Vol.~\bibinfo{volume}{37}. Wiley Online Library, \bibinfo{pages}{195--204}.
\newblock


\bibitem[Fedkiw et~al\mbox{.}(2001)]%
        {Fedkiw:2001:VSS}
\bibfield{author}{\bibinfo{person}{Ronald Fedkiw}, \bibinfo{person}{Jos Stam}, {and} \bibinfo{person}{Henrik~Wann Jensen}.} \bibinfo{year}{2001}\natexlab{}.
\newblock \showarticletitle{Visual simulation of smoke}. In \bibinfo{booktitle}{\emph{Proceedings of the 28th annual conference on Computer graphics and interactive techniques}}. \bibinfo{pages}{15--22}.
\newblock


\bibitem[Fei et~al\mbox{.}(2021)]%
        {Fei:2021:ASFLIP}
\bibfield{author}{\bibinfo{person}{Yun~(Raymond) Fei}, \bibinfo{person}{Qi Guo}, \bibinfo{person}{Rundong Wu}, \bibinfo{person}{Li Huang}, {and} \bibinfo{person}{Ming Gao}.} \bibinfo{year}{2021}\natexlab{}.
\newblock \showarticletitle{Revisiting Integration in the Material Point Method: A Scheme for Easier Separation and Less Dissipation}.
\newblock \bibinfo{journal}{\emph{ACM Trans. Graph.}} \bibinfo{volume}{40}, \bibinfo{number}{4}, Article \bibinfo{articleno}{109} (\bibinfo{date}{Aug.} \bibinfo{year}{2021}), \bibinfo{numpages}{16}~pages.
\newblock
\urldef\tempurl%
\url{https://doi.org/10.1145/3450626.3459678}
\showDOI{\tempurl}


\bibitem[Feng et~al\mbox{.}(2022)]%
        {Feng:2022:IFS}
\bibfield{author}{\bibinfo{person}{Fan Feng}, \bibinfo{person}{Jinyuan Liu}, \bibinfo{person}{Shiying Xiong}, \bibinfo{person}{Shuqi Yang}, \bibinfo{person}{Yaorui Zhang}, {and} \bibinfo{person}{Bo Zhu}.} \bibinfo{year}{2022}\natexlab{}.
\newblock \showarticletitle{Impulse Fluid Simulation}.
\newblock \bibinfo{journal}{\emph{IEEE Transactions on Visualization and Computer Graphics}} (\bibinfo{year}{2022}).
\newblock


\bibitem[Fu et~al\mbox{.}(2017)]%
        {Fu:2017:PolyPIC}
\bibfield{author}{\bibinfo{person}{Chuyuan Fu}, \bibinfo{person}{Qi Guo}, \bibinfo{person}{Theodore Gast}, \bibinfo{person}{Chenfanfu Jiang}, {and} \bibinfo{person}{Joseph Teran}.} \bibinfo{year}{2017}\natexlab{}.
\newblock \showarticletitle{A polynomial particle-in-cell method}.
\newblock \bibinfo{journal}{\emph{ACM Transactions on Graphics (TOG)}} \bibinfo{volume}{36}, \bibinfo{number}{6} (\bibinfo{year}{2017}), \bibinfo{pages}{222:1--222:12}.
\newblock


\bibitem[Gao et~al\mbox{.}(2018)]%
        {Gao:2018:GPUMPM}
\bibfield{author}{\bibinfo{person}{Ming Gao}, \bibinfo{person}{Xinlei Wang}, \bibinfo{person}{Kui Wu}, \bibinfo{person}{Andre Pradhana}, \bibinfo{person}{Eftychios Sifakis}, \bibinfo{person}{Cem Yuksel}, {and} \bibinfo{person}{Chenfanfu Jiang}.} \bibinfo{year}{2018}\natexlab{}.
\newblock \showarticletitle{GPU optimization of material point methods}.
\newblock \bibinfo{journal}{\emph{ACM Transactions on Graphics (TOG)}} \bibinfo{volume}{37}, \bibinfo{number}{6} (\bibinfo{year}{2018}), \bibinfo{pages}{1--12}.
\newblock


\bibitem[Gassmann and Herzog(2008)]%
        {Gassmann:2008:TCN}
\bibfield{author}{\bibinfo{person}{Almut Gassmann} {and} \bibinfo{person}{Hans-Joachim Herzog}.} \bibinfo{year}{2008}\natexlab{}.
\newblock \showarticletitle{Towards a consistent numerical compressible non-hydrostatic model using generalized Hamiltonian tools}.
\newblock \bibinfo{journal}{\emph{Quarterly Journal of the Royal Meteorological Society}} \bibinfo{volume}{134}, \bibinfo{number}{635} (\bibinfo{year}{2008}), \bibinfo{pages}{1597--1613}.
\newblock


\bibitem[Gerritsma(2010)]%
        {Gerritsma:2010:EFS}
\bibfield{author}{\bibinfo{person}{Marc Gerritsma}.} \bibinfo{year}{2010}\natexlab{}.
\newblock \showarticletitle{Edge functions for spectral element methods}. In \bibinfo{booktitle}{\emph{Spectral and High Order Methods for Partial Differential Equations: Selected papers from the ICOSAHOM'09 conference, June 22-26, Trondheim, Norway}}. Springer, \bibinfo{pages}{199--207}.
\newblock


\bibitem[González-López et~al\mbox{.}(1992)]%
        {Gonzalez-Lopez:1992:LIE}
\bibfield{author}{\bibinfo{person}{Artemio González-López}, \bibinfo{person}{Niky Kamran}, {and} \bibinfo{person}{Peter~J. Olver}.} \bibinfo{year}{1992}\natexlab{}.
\newblock \showarticletitle{Lie {{Algebras}} of {{Vector Fields}} in the {{Real Plane}}}.
\newblock  \bibinfo{volume}{s3-64}, \bibinfo{number}{2} (\bibinfo{year}{1992}), \bibinfo{pages}{339--368}.
\newblock
\showISSN{0024-6115}
\urldef\tempurl%
\url{https://doi.org/10.1112/plms/s3-64.2.339}
\showDOI{\tempurl}


\bibitem[Gould and Scott(2017)]%
        {Gould:2017:SPS}
\bibfield{author}{\bibinfo{person}{Nicholas Gould} {and} \bibinfo{person}{Jennifer Scott}.} \bibinfo{year}{2017}\natexlab{}.
\newblock \showarticletitle{The state-of-the-art of preconditioners for sparse linear least-squares problems}.
\newblock \bibinfo{journal}{\emph{ACM Transactions on Mathematical Software (TOMS)}} \bibinfo{volume}{43}, \bibinfo{number}{4} (\bibinfo{year}{2017}), \bibinfo{pages}{1--35}.
\newblock


\bibitem[Guennebaud et~al\mbox{.}(2010)]%
        {Gael:2010:Eigen}
\bibfield{author}{\bibinfo{person}{Ga\"{e}l Guennebaud}, \bibinfo{person}{Beno\^{i}t Jacob}, {et~al\mbox{.}}} \bibinfo{year}{2010}\natexlab{}.
\newblock \bibinfo{title}{Eigen v3}.
\newblock \bibinfo{howpublished}{http://eigen.tuxfamily.org}.
\newblock


\bibitem[Hairer et~al\mbox{.}(2006)]%
        {Hairer:2006:GNI}
\bibfield{author}{\bibinfo{person}{Ernst Hairer}, \bibinfo{person}{Marlis Hochbruck}, \bibinfo{person}{Arieh Iserles}, {and} \bibinfo{person}{Christian Lubich}.} \bibinfo{year}{2006}\natexlab{}.
\newblock \showarticletitle{Geometric numerical integration}.
\newblock \bibinfo{journal}{\emph{Oberwolfach Reports}} \bibinfo{volume}{3}, \bibinfo{number}{1} (\bibinfo{year}{2006}), \bibinfo{pages}{805--882}.
\newblock


\bibitem[Hammerquist and Nairn(2017)]%
        {Hammerquist:2017:XPIC}
\bibfield{author}{\bibinfo{person}{Chad~C Hammerquist} {and} \bibinfo{person}{John~A Nairn}.} \bibinfo{year}{2017}\natexlab{}.
\newblock \showarticletitle{A new method for material point method particle updates that reduces noise and enhances stability}.
\newblock \bibinfo{journal}{\emph{Computer methods in applied mechanics and engineering}}  \bibinfo{volume}{318} (\bibinfo{year}{2017}), \bibinfo{pages}{724--738}.
\newblock


\bibitem[Han et~al\mbox{.}(2019)]%
        {Han:MPM-contact}
\bibfield{author}{\bibinfo{person}{Xuchen Han}, \bibinfo{person}{Theodore~F. Gast}, \bibinfo{person}{Qi Guo}, \bibinfo{person}{Stephanie Wang}, \bibinfo{person}{Chenfanfu Jiang}, {and} \bibinfo{person}{Joseph Teran}.} \bibinfo{year}{2019}\natexlab{}.
\newblock \showarticletitle{A Hybrid Material Point Method for Frictional Contact with Diverse Materials}.
\newblock \bibinfo{journal}{\emph{Proc. ACM Comput. Graph. Interact. Tech.}} \bibinfo{volume}{2}, \bibinfo{number}{2}, Article \bibinfo{articleno}{17} (\bibinfo{date}{jul} \bibinfo{year}{2019}), \bibinfo{numpages}{24}~pages.
\newblock
\urldef\tempurl%
\url{https://doi.org/10.1145/3340258}
\showDOI{\tempurl}


\bibitem[Harlow(1962)]%
        {Harlow:1962:PIC}
\bibfield{author}{\bibinfo{person}{Francis~H Harlow}.} \bibinfo{year}{1962}\natexlab{}.
\newblock \bibinfo{booktitle}{\emph{The particle-in-cell method for numerical solution of problems in fluid dynamics}}.
\newblock \bibinfo{type}{{T}echnical {R}eport}. \bibinfo{institution}{Los Alamos National Lab.(LANL), Los Alamos, NM (United States)}.
\newblock


\bibitem[Helmholz(1858)]%
        {Helmholtz:1858:VD}
\bibfield{author}{\bibinfo{person}{Hermann~von Helmholz}.} \bibinfo{year}{1858}\natexlab{}.
\newblock \showarticletitle{{\"U}ber Integrale der hydrodynamischen Gleichungen, welche den Wirbelbewegungen entsprechen}.
\newblock \bibinfo{journal}{\emph{Journal f{\"u}r die reine und angewandte Mathematik}}  \bibinfo{volume}{55} (\bibinfo{year}{1858}), \bibinfo{pages}{25--55}.
\newblock
\newblock
\shownote{English translation by P.G. Tait, 1867, \url{http://www.biodiversitylibrary.org/item/121849\#page/499/mode/1up}}.


\bibitem[Hiemstra et~al\mbox{.}(2014)]%
        {Hiemstra:2014:HOG}
\bibfield{author}{\bibinfo{person}{Ren{\'e}~R Hiemstra}, \bibinfo{person}{Deepesh Toshniwal}, \bibinfo{person}{RHM Huijsmans}, {and} \bibinfo{person}{Marc~I Gerritsma}.} \bibinfo{year}{2014}\natexlab{}.
\newblock \showarticletitle{High order geometric methods with exact conservation properties}.
\newblock \bibinfo{journal}{\emph{J. Comput. Phys.}}  \bibinfo{volume}{257} (\bibinfo{year}{2014}), \bibinfo{pages}{1444--1471}.
\newblock


\bibitem[Hirani(2003)]%
        {Hirani:2003:DEC}
\bibfield{author}{\bibinfo{person}{Anil~Nirmal Hirani}.} \bibinfo{year}{2003}\natexlab{}.
\newblock \bibinfo{booktitle}{\emph{Discrete exterior calculus}}.
\newblock \bibinfo{publisher}{California Institute of Technology}.
\newblock


\bibitem[Holm(2011)]%
        {Holm:2011:GMD}
\bibfield{author}{\bibinfo{person}{Darryl~D Holm}.} \bibinfo{year}{2011}\natexlab{}.
\newblock \bibinfo{booktitle}{\emph{Geometric mechanics-Part I: Dynamics and symmetry}}.
\newblock \bibinfo{publisher}{World Scientific Publishing Company}.
\newblock


\bibitem[Hu et~al\mbox{.}(2018)]%
        {Hu:2018:MLS}
\bibfield{author}{\bibinfo{person}{Yuanming Hu}, \bibinfo{person}{Yu Fang}, \bibinfo{person}{Ziheng Ge}, \bibinfo{person}{Ziyin Qu}, \bibinfo{person}{Yixin Zhu}, \bibinfo{person}{Andre Pradhana}, {and} \bibinfo{person}{Chenfanfu Jiang}.} \bibinfo{year}{2018}\natexlab{}.
\newblock \showarticletitle{A moving least squares material point method with displacement discontinuity and two-way rigid body coupling}.
\newblock \bibinfo{journal}{\emph{ACM Transactions on Graphics (TOG)}} \bibinfo{volume}{37}, \bibinfo{number}{4} (\bibinfo{year}{2018}), \bibinfo{pages}{1--14}.
\newblock


\bibitem[Hughes et~al\mbox{.}(2005)]%
        {Hughes:2005:IGA}
\bibfield{author}{\bibinfo{person}{Thomas~JR Hughes}, \bibinfo{person}{John~A Cottrell}, {and} \bibinfo{person}{Yuri Bazilevs}.} \bibinfo{year}{2005}\natexlab{}.
\newblock \showarticletitle{Isogeometric analysis: CAD, finite elements, NURBS, exact geometry and mesh refinement}.
\newblock \bibinfo{journal}{\emph{Computer methods in applied mechanics and engineering}} \bibinfo{volume}{194}, \bibinfo{number}{39-41} (\bibinfo{year}{2005}), \bibinfo{pages}{4135--4195}.
\newblock


\bibitem[Izosimov and Khesin(2017)]%
        {Izosimov:2017:CC2}
\bibfield{author}{\bibinfo{person}{Anton Izosimov} {and} \bibinfo{person}{Boris Khesin}.} \bibinfo{year}{2017}\natexlab{}.
\newblock \showarticletitle{Classification of Casimirs in 2D hydrodynamics}.
\newblock \bibinfo{journal}{\emph{arXiv preprint arXiv:1702.01843}} (\bibinfo{year}{2017}).
\newblock


\bibitem[Jiang et~al\mbox{.}(2017)]%
        {Jiang:MPM-cloth}
\bibfield{author}{\bibinfo{person}{Chenfanfu Jiang}, \bibinfo{person}{Theodore Gast}, {and} \bibinfo{person}{Joseph Teran}.} \bibinfo{year}{2017}\natexlab{}.
\newblock \showarticletitle{Anisotropic elastoplasticity for cloth, knit and hair frictional contact}.
\newblock \bibinfo{journal}{\emph{ACM Trans. Graph.}} \bibinfo{volume}{36}, \bibinfo{number}{4}, Article \bibinfo{articleno}{152} (\bibinfo{date}{jul} \bibinfo{year}{2017}), \bibinfo{numpages}{14}~pages.
\newblock
\showISSN{0730-0301}
\urldef\tempurl%
\url{https://doi.org/10.1145/3072959.3073623}
\showDOI{\tempurl}


\bibitem[Jiang et~al\mbox{.}(2015)]%
        {Jiang:2015:APIC}
\bibfield{author}{\bibinfo{person}{Chenfanfu Jiang}, \bibinfo{person}{Craig Schroeder}, \bibinfo{person}{Andrew Selle}, \bibinfo{person}{Joseph Teran}, {and} \bibinfo{person}{Alexey Stomakhin}.} \bibinfo{year}{2015}\natexlab{}.
\newblock \showarticletitle{The affine particle-in-cell method}.
\newblock \bibinfo{journal}{\emph{ACM Transactions on Graphics (TOG)}} \bibinfo{volume}{34}, \bibinfo{number}{4} (\bibinfo{year}{2015}), \bibinfo{pages}{51:1--51:10}.
\newblock


\bibitem[Jiang et~al\mbox{.}(2016)]%
        {Jiang:2016:MPM}
\bibfield{author}{\bibinfo{person}{Chenfanfu Jiang}, \bibinfo{person}{Craig Schroeder}, \bibinfo{person}{Joseph Teran}, \bibinfo{person}{Alexey Stomakhin}, {and} \bibinfo{person}{Andrew Selle}.} \bibinfo{year}{2016}\natexlab{}.
\newblock \showarticletitle{The material point method for simulating continuum materials}.
\newblock In \bibinfo{booktitle}{\emph{Acm siggraph 2016 courses}}. \bibinfo{pages}{1--52}.
\newblock


\bibitem[Jiang and Shu(1996)]%
        {Jiang:1996:WENO}
\bibfield{author}{\bibinfo{person}{Guang-Shan Jiang} {and} \bibinfo{person}{Chi-Wang Shu}.} \bibinfo{year}{1996}\natexlab{}.
\newblock \showarticletitle{Efficient implementation of weighted ENO schemes}.
\newblock \bibinfo{journal}{\emph{Journal of computational physics}} \bibinfo{volume}{126}, \bibinfo{number}{1} (\bibinfo{year}{1996}), \bibinfo{pages}{202--228}.
\newblock


\bibitem[Kapidani and Hernandez(2022)]%
        {Kapidani:2022:HOG}
\bibfield{author}{\bibinfo{person}{Bernard Kapidani} {and} \bibinfo{person}{Rafael~V{\'a}zquez Hernandez}.} \bibinfo{year}{2022}\natexlab{}.
\newblock \showarticletitle{High Order Geometric Methods With Splines: An Analysis of Discrete Hodge-Star Operators}.
\newblock \bibinfo{journal}{\emph{SIAM Journal on Scientific Computing}} \bibinfo{volume}{44}, \bibinfo{number}{6} (\bibinfo{year}{2022}), \bibinfo{pages}{A3673--A3699}.
\newblock


\bibitem[Khesin and Chekanov(1989)]%
        {Khesin:1989:IEE}
\bibfield{author}{\bibinfo{person}{BA Khesin} {and} \bibinfo{person}{Yu~V Chekanov}.} \bibinfo{year}{1989}\natexlab{}.
\newblock \showarticletitle{Invariants of the Euler equations for ideal or barotropic hydrodynamics and superconductivity in D dimensions}.
\newblock \bibinfo{journal}{\emph{Physica D: Nonlinear Phenomena}} \bibinfo{volume}{40}, \bibinfo{number}{1} (\bibinfo{year}{1989}), \bibinfo{pages}{119--131}.
\newblock


\bibitem[Khesin et~al\mbox{.}(2019)]%
        {Khesin:2019:GMT}
\bibfield{author}{\bibinfo{person}{Boris Khesin}, \bibinfo{person}{Gerard Misio{\l}ek}, {and} \bibinfo{person}{Klas Modin}.} \bibinfo{year}{2019}\natexlab{}.
\newblock \showarticletitle{Geometry of the Madelung transform}.
\newblock \bibinfo{journal}{\emph{Archive for Rational Mechanics and Analysis}}  \bibinfo{volume}{234} (\bibinfo{year}{2019}), \bibinfo{pages}{549--573}.
\newblock


\bibitem[Khesin et~al\mbox{.}(2022)]%
        {Khesin:2022:HUC}
\bibfield{author}{\bibinfo{person}{Boris Khesin}, \bibinfo{person}{Daniel Peralta-Salas}, {and} \bibinfo{person}{Cheng Yang}.} \bibinfo{year}{2022}\natexlab{}.
\newblock \showarticletitle{The helicity uniqueness conjecture in 3D hydrodynamics}.
\newblock \bibinfo{journal}{\emph{Trans. Amer. Math. Soc.}} \bibinfo{volume}{375}, \bibinfo{number}{02} (\bibinfo{year}{2022}), \bibinfo{pages}{909--924}.
\newblock


\bibitem[Kim et~al\mbox{.}(2005)]%
        {Kim:2005:Flowfixer}
\bibfield{author}{\bibinfo{person}{ByungMoon Kim}, \bibinfo{person}{Yingjie Liu}, \bibinfo{person}{Ignacio Llamas}, {and} \bibinfo{person}{Jaroslaw~R Rossignac}.} \bibinfo{year}{2005}\natexlab{}.
\newblock \bibinfo{booktitle}{\emph{Flowfixer: Using {BFECC} for fluid simulation}}.
\newblock \bibinfo{type}{{T}echnical {R}eport}. \bibinfo{institution}{Georgia Institute of Technology}.
\newblock


\bibitem[Kl{\'a}r et~al\mbox{.}(2016)]%
        {Klar:2016:DPE}
\bibfield{author}{\bibinfo{person}{Gergely Kl{\'a}r}, \bibinfo{person}{Theodore Gast}, \bibinfo{person}{Andre Pradhana}, \bibinfo{person}{Chuyuan Fu}, \bibinfo{person}{Craig Schroeder}, \bibinfo{person}{Chenfanfu Jiang}, {and} \bibinfo{person}{Joseph Teran}.} \bibinfo{year}{2016}\natexlab{}.
\newblock \showarticletitle{Drucker-prager elastoplasticity for sand animation}.
\newblock \bibinfo{journal}{\emph{ACM Transactions on Graphics (TOG)}} \bibinfo{volume}{35}, \bibinfo{number}{4} (\bibinfo{year}{2016}), \bibinfo{pages}{1--12}.
\newblock


\bibitem[Kleckner and Irvine(2013)]%
        {Kleckner:2013:CDK}
\bibfield{author}{\bibinfo{person}{Dustin Kleckner} {and} \bibinfo{person}{William~TM Irvine}.} \bibinfo{year}{2013}\natexlab{}.
\newblock \showarticletitle{Creation and dynamics of knotted vortices}.
\newblock \bibinfo{journal}{\emph{Nature physics}} \bibinfo{volume}{9}, \bibinfo{number}{4} (\bibinfo{year}{2013}), \bibinfo{pages}{253--258}.
\newblock


\bibitem[Kugelstadt et~al\mbox{.}(2019)]%
        {Kugelstadt:2019:IDP}
\bibfield{author}{\bibinfo{person}{Tassilo Kugelstadt}, \bibinfo{person}{Andreas Longva}, \bibinfo{person}{Nils Thuerey}, {and} \bibinfo{person}{Jan Bender}.} \bibinfo{year}{2019}\natexlab{}.
\newblock \showarticletitle{Implicit density projection for volume conserving liquids}.
\newblock \bibinfo{journal}{\emph{IEEE Transactions on Visualization and Computer Graphics}} \bibinfo{volume}{27}, \bibinfo{number}{4} (\bibinfo{year}{2019}), \bibinfo{pages}{2385--2395}.
\newblock


\bibitem[Lee(2013)]%
        {Lee:2013:SM}
\bibfield{author}{\bibinfo{person}{John~M Lee}.} \bibinfo{year}{2013}\natexlab{}.
\newblock \showarticletitle{Smooth manifolds}.
\newblock In \bibinfo{booktitle}{\emph{Introduction to smooth manifolds}}. \bibinfo{publisher}{Springer}, \bibinfo{pages}{1--31}.
\newblock


\bibitem[Lin and Mor{\'e}(1999)]%
        {Lin:1999:IC}
\bibfield{author}{\bibinfo{person}{Chih-Jen Lin} {and} \bibinfo{person}{Jorge~J Mor{\'e}}.} \bibinfo{year}{1999}\natexlab{}.
\newblock \showarticletitle{Incomplete Cholesky factorizations with limited memory}.
\newblock \bibinfo{journal}{\emph{SIAM Journal on Scientific computing}} \bibinfo{volume}{21}, \bibinfo{number}{1} (\bibinfo{year}{1999}), \bibinfo{pages}{24--45}.
\newblock


\bibitem[Liu et~al\mbox{.}(2015)]%
        {Liu:2015:MVFS}
\bibfield{author}{\bibinfo{person}{Beibei Liu}, \bibinfo{person}{Gemma Mason}, \bibinfo{person}{Julian Hodgson}, \bibinfo{person}{Yiying Tong}, {and} \bibinfo{person}{Mathieu Desbrun}.} \bibinfo{year}{2015}\natexlab{}.
\newblock \showarticletitle{Model-reduced variational fluid simulation}.
\newblock \bibinfo{journal}{\emph{ACM Transactions on Graphics (TOG)}} \bibinfo{volume}{34}, \bibinfo{number}{6} (\bibinfo{year}{2015}), \bibinfo{pages}{244:1--244: 12}.
\newblock


\bibitem[Lyu et~al\mbox{.}(2024)]%
        {Lyu:2024:WPE}
\bibfield{author}{\bibinfo{person}{Luan Lyu}, \bibinfo{person}{Xiaohua Ren}, \bibinfo{person}{Wei Cao}, \bibinfo{person}{Jian Zhu}, \bibinfo{person}{Enhua Wu}, {and} \bibinfo{person}{Zhi-Xin Yang}.} \bibinfo{year}{2024}\natexlab{}.
\newblock \showarticletitle{Wavelet Potentials: An Efficient Potential Recovery Technique for Pointwise Incompressible Fluids}. In \bibinfo{booktitle}{\emph{Computer Graphics Forum}}. Wiley Online Library, \bibinfo{pages}{e15023}.
\newblock


\bibitem[Maggioni et~al\mbox{.}(2010)]%
        {Maggioni:2010:VEH}
\bibfield{author}{\bibinfo{person}{Francesca Maggioni}, \bibinfo{person}{Sultan Alamri}, \bibinfo{person}{Carlo~F Barenghi}, {and} \bibinfo{person}{Renzo~L Ricca}.} \bibinfo{year}{2010}\natexlab{}.
\newblock \showarticletitle{Velocity, energy, and helicity of vortex knots and unknots}.
\newblock \bibinfo{journal}{\emph{Physical Review E}} \bibinfo{volume}{82}, \bibinfo{number}{2} (\bibinfo{year}{2010}), \bibinfo{pages}{026309}.
\newblock


\bibitem[Marsden and Weinstein(1983)]%
        {Marsden:1983:COV}
\bibfield{author}{\bibinfo{person}{Jerrold Marsden} {and} \bibinfo{person}{Alan Weinstein}.} \bibinfo{year}{1983}\natexlab{}.
\newblock \showarticletitle{Coadjoint Orbits, Vortices, and Clebsch Variables for Incompressible Fluids}.
\newblock \bibinfo{journal}{\emph{Physica D: Nonlinear Phenomena}} \bibinfo{volume}{7}, \bibinfo{number}{1} (\bibinfo{year}{1983}), \bibinfo{pages}{305--323}.
\newblock


\bibitem[Marsden et~al\mbox{.}(1999)]%
        {Marsden:1999:DEP}
\bibfield{author}{\bibinfo{person}{Jerrold~E Marsden}, \bibinfo{person}{Sergey Pekarsky}, {and} \bibinfo{person}{Steve Shkoller}.} \bibinfo{year}{1999}\natexlab{}.
\newblock \showarticletitle{Discrete euler-poincar{\'e} and lie-poisson equations}.
\newblock \bibinfo{journal}{\emph{Nonlinearity}} \bibinfo{volume}{12}, \bibinfo{number}{6} (\bibinfo{year}{1999}), \bibinfo{pages}{1647}.
\newblock


\bibitem[Marsden et~al\mbox{.}(1997)]%
        {Marsden:1997:IMS}
\bibfield{author}{\bibinfo{person}{Jerrold~E Marsden}, \bibinfo{person}{Tudor~S Ratiu}, {and} \bibinfo{person}{Robert Hermann}.} \bibinfo{year}{1997}\natexlab{}.
\newblock \showarticletitle{Introduction to mechanics and symmetry}.
\newblock \bibinfo{journal}{\emph{SIAM Rev.}} \bibinfo{volume}{39}, \bibinfo{number}{1} (\bibinfo{year}{1997}), \bibinfo{pages}{152--152}.
\newblock


\bibitem[Moffatt(1969)]%
        {Moffatt:1969:DKT}
\bibfield{author}{\bibinfo{person}{Henry~Keith Moffatt}.} \bibinfo{year}{1969}\natexlab{}.
\newblock \showarticletitle{The degree of knottedness of tangled vortex lines}.
\newblock \bibinfo{journal}{\emph{Journal of Fluid Mechanics}} \bibinfo{volume}{35}, \bibinfo{number}{1} (\bibinfo{year}{1969}), \bibinfo{pages}{117--129}.
\newblock


\bibitem[Moffatt(2014)]%
        {Moffatt:2014:HSS}
\bibfield{author}{\bibinfo{person}{H~Keith Moffatt}.} \bibinfo{year}{2014}\natexlab{}.
\newblock \showarticletitle{Helicity and singular structures in fluid dynamics}.
\newblock \bibinfo{journal}{\emph{Proceedings of the National Academy of Sciences}} \bibinfo{volume}{111}, \bibinfo{number}{10} (\bibinfo{year}{2014}), \bibinfo{pages}{3663--3670}.
\newblock


\bibitem[Moreau(1961)]%
        {Moreau:1961:CTF}
\bibfield{author}{\bibinfo{person}{Jean~Jacques Moreau}.} \bibinfo{year}{1961}\natexlab{}.
\newblock \showarticletitle{Constantes d'un {\^\i}lot tourbillonnaire en fluide parfait barotrope}.
\newblock \bibinfo{journal}{\emph{Comptes rendus hebdomadaires des s{\'e}ances de l'Acad{\'e}mie des sciences}}  \bibinfo{volume}{252} (\bibinfo{year}{1961}), \bibinfo{pages}{2810--2812}.
\newblock


\bibitem[Morrison(1998)]%
        {Morrison:1998:HD}
\bibfield{author}{\bibinfo{person}{Philip~J Morrison}.} \bibinfo{year}{1998}\natexlab{}.
\newblock \showarticletitle{Hamiltonian description of the ideal fluid}.
\newblock \bibinfo{journal}{\emph{Reviews of modern physics}} \bibinfo{volume}{70}, \bibinfo{number}{2} (\bibinfo{year}{1998}), \bibinfo{pages}{467}.
\newblock


\bibitem[Mullen et~al\mbox{.}(2009)]%
        {Mullen:2009:EPI}
\bibfield{author}{\bibinfo{person}{Patrick Mullen}, \bibinfo{person}{Keenan Crane}, \bibinfo{person}{Dmitry Pavlov}, \bibinfo{person}{Yiying Tong}, {and} \bibinfo{person}{Mathieu Desbrun}.} \bibinfo{year}{2009}\natexlab{}.
\newblock \showarticletitle{Energy-preserving integrators for fluid animation}.
\newblock \bibinfo{journal}{\emph{ACM Transactions on Graphics (TOG)}} \bibinfo{volume}{28}, \bibinfo{number}{3} (\bibinfo{year}{2009}), \bibinfo{pages}{38:1--38:8}.
\newblock


\bibitem[Nabizadeh et~al\mbox{.}(2022)]%
        {Nabizadeh:2022:CF}
\bibfield{author}{\bibinfo{person}{Mohammad~Sina Nabizadeh}, \bibinfo{person}{Stephanie Wang}, \bibinfo{person}{Ravi Ramamoorthi}, {and} \bibinfo{person}{Albert Chern}.} \bibinfo{year}{2022}\natexlab{}.
\newblock \showarticletitle{Covector fluids}.
\newblock \bibinfo{journal}{\emph{ACM Transactions on Graphics (TOG)}} \bibinfo{volume}{41}, \bibinfo{number}{4} (\bibinfo{year}{2022}), \bibinfo{pages}{113:1--113:16}.
\newblock


\bibitem[Nairn and Hammerquist(2021)]%
        {Nairn:2021:FMPM}
\bibfield{author}{\bibinfo{person}{John~A. Nairn} {and} \bibinfo{person}{Chad~C. Hammerquist}.} \bibinfo{year}{2021}\natexlab{}.
\newblock \showarticletitle{Material point method simulations using an approximate full mass matrix inverse}.
\newblock \bibinfo{journal}{\emph{Computer Methods in Applied Mechanics and Engineering}}  \bibinfo{volume}{377} (\bibinfo{year}{2021}), \bibinfo{pages}{113667}.
\newblock
\showISSN{0045-7825}
\urldef\tempurl%
\url{https://doi.org/10.1016/j.cma.2021.113667}
\showDOI{\tempurl}


\bibitem[Nambu(1973)]%
        {Nambu:1973:GHD}
\bibfield{author}{\bibinfo{person}{Yoichiro Nambu}.} \bibinfo{year}{1973}\natexlab{}.
\newblock \showarticletitle{Generalized hamiltonian dynamics}.
\newblock \bibinfo{journal}{\emph{Physical Review D}} \bibinfo{volume}{7}, \bibinfo{number}{8} (\bibinfo{year}{1973}), \bibinfo{pages}{2405}.
\newblock


\bibitem[Narain et~al\mbox{.}(2019)]%
        {Narain:2019:SAR}
\bibfield{author}{\bibinfo{person}{Rahul Narain}, \bibinfo{person}{Jonas Zehnder}, {and} \bibinfo{person}{Bernhard Thomaszewski}.} \bibinfo{year}{2019}\natexlab{}.
\newblock \showarticletitle{A second-order advection-reflection solver}.
\newblock \bibinfo{journal}{\emph{Proceedings of the ACM on Computer Graphics and Interactive Techniques}} \bibinfo{volume}{2}, \bibinfo{number}{2} (\bibinfo{year}{2019}), \bibinfo{pages}{16:1--16:14}.
\newblock


\bibitem[Narcowich and Ward(1994)]%
        {Narcowich:1994:GHI}
\bibfield{author}{\bibinfo{person}{Francis~J Narcowich} {and} \bibinfo{person}{Joseph~D Ward}.} \bibinfo{year}{1994}\natexlab{}.
\newblock \showarticletitle{Generalized Hermite interpolation via matrix-valued conditionally positive definite functions}.
\newblock \bibinfo{journal}{\emph{Math. Comp.}} \bibinfo{volume}{63}, \bibinfo{number}{208} (\bibinfo{year}{1994}), \bibinfo{pages}{661--687}.
\newblock


\bibitem[N{\'e}vir and Blender(1993)]%
        {Nevir:1993:NRI}
\bibfield{author}{\bibinfo{person}{P N{\'e}vir} {and} \bibinfo{person}{R Blender}.} \bibinfo{year}{1993}\natexlab{}.
\newblock \showarticletitle{A Nambu representation of incompressible hydrodynamics using helicity and enstrophy}.
\newblock \bibinfo{journal}{\emph{Journal of Physics A: Mathematical and General}} \bibinfo{volume}{26}, \bibinfo{number}{22} (\bibinfo{year}{1993}), \bibinfo{pages}{L1189}.
\newblock


\bibitem[{nLab authors}(2024)]%
        {nlab:derivations_of_smooth_functions_are_vector_fields}
\bibfield{author}{\bibinfo{person}{{nLab authors}}.} \bibinfo{year}{2024}\natexlab{}.
\newblock \bibinfo{title}{Derivations of smooth functions are vector fields}.
\newblock \bibinfo{howpublished}{\url{https://ncatlab.org/nlab/show/derivations+of+smooth+functions+are+vector+fields}}.
\newblock
\newblock
\shownote{\href{https://ncatlab.org/nlab/revision/derivations+of+smooth+functions+are+vector+fields/14}{Revision 14}}.


\bibitem[Oseledets(1989)]%
        {Oseledets:1989:NWW}
\bibfield{author}{\bibinfo{person}{VI Oseledets}.} \bibinfo{year}{1989}\natexlab{}.
\newblock \showarticletitle{On a new way of writing the Navier-Stokes equation. The Hamiltonian formalism}.
\newblock \bibinfo{journal}{\emph{Russ. Math. Surveys}}  \bibinfo{volume}{44} (\bibinfo{year}{1989}), \bibinfo{pages}{210--211}.
\newblock


\bibitem[Pavlov et~al\mbox{.}(2011)]%
        {Pavlov:2011:SPD}
\bibfield{author}{\bibinfo{person}{Dmitry Pavlov}, \bibinfo{person}{Patrick Mullen}, \bibinfo{person}{Yiying Tong}, \bibinfo{person}{Eva Kanso}, \bibinfo{person}{Jerrold~E Marsden}, {and} \bibinfo{person}{Mathieu Desbrun}.} \bibinfo{year}{2011}\natexlab{}.
\newblock \showarticletitle{Structure-preserving discretization of incompressible fluids}.
\newblock \bibinfo{journal}{\emph{Physica D: Nonlinear Phenomena}} \bibinfo{volume}{240}, \bibinfo{number}{6} (\bibinfo{year}{2011}), \bibinfo{pages}{443--458}.
\newblock


\bibitem[Pletzer and Fillmore(2015)]%
        {Pletzer:2015:CIE}
\bibfield{author}{\bibinfo{person}{Alexander Pletzer} {and} \bibinfo{person}{David Fillmore}.} \bibinfo{year}{2015}\natexlab{}.
\newblock \showarticletitle{Conservative interpolation of edge and face data on n dimensional structured grids using differential forms}.
\newblock \bibinfo{journal}{\emph{J. Comput. Phys.}}  \bibinfo{volume}{302} (\bibinfo{year}{2015}), \bibinfo{pages}{21--40}.
\newblock


\bibitem[Pletzer and Hayek(2019)]%
        {Pletzer:2019:MIV}
\bibfield{author}{\bibinfo{person}{Alexander Pletzer} {and} \bibinfo{person}{Wolfgang Hayek}.} \bibinfo{year}{2019}\natexlab{}.
\newblock \showarticletitle{Mimetic interpolation of vector fields on Arakawa C/D grids}.
\newblock \bibinfo{journal}{\emph{Monthly Weather Review}} \bibinfo{volume}{147}, \bibinfo{number}{1} (\bibinfo{year}{2019}), \bibinfo{pages}{3--16}.
\newblock


\bibitem[Qiu et~al\mbox{.}(2010)]%
        {Qiu:2022:Spectra}
\bibfield{author}{\bibinfo{person}{Yixuan Qiu} {et~al\mbox{.}}} \bibinfo{year}{2010}\natexlab{}.
\newblock \bibinfo{title}{Spectra}.
\newblock \bibinfo{howpublished}{https://spectralib.org/}.
\newblock


\bibitem[Qu et~al\mbox{.}(2022)]%
        {Qu:2022:PPIC}
\bibfield{author}{\bibinfo{person}{Ziyin Qu}, \bibinfo{person}{Minchen Li}, \bibinfo{person}{Fernando De~Goes}, {and} \bibinfo{person}{Chenfanfu Jiang}.} \bibinfo{year}{2022}\natexlab{}.
\newblock \showarticletitle{The power particle-in-cell method}.
\newblock \bibinfo{journal}{\emph{ACM Transactions on Graphics}} \bibinfo{volume}{41}, \bibinfo{number}{4} (\bibinfo{year}{2022}).
\newblock


\bibitem[Qu et~al\mbox{.}(2019)]%
        {Qu:2019:ECF}
\bibfield{author}{\bibinfo{person}{Ziyin Qu}, \bibinfo{person}{Xinxin Zhang}, \bibinfo{person}{Ming Gao}, \bibinfo{person}{Chenfanfu Jiang}, {and} \bibinfo{person}{Baoquan Chen}.} \bibinfo{year}{2019}\natexlab{}.
\newblock \showarticletitle{Efficient and conservative fluids using bidirectional mapping}.
\newblock \bibinfo{journal}{\emph{ACM Transactions on Graphics (TOG)}} \bibinfo{volume}{38}, \bibinfo{number}{4} (\bibinfo{year}{2019}), \bibinfo{pages}{128:1--128:12}.
\newblock


\bibitem[Ram et~al\mbox{.}(2015)]%
        {Ram:2015:MPM-VE}
\bibfield{author}{\bibinfo{person}{Daniel Ram}, \bibinfo{person}{Theodore Gast}, \bibinfo{person}{Chenfanfu Jiang}, \bibinfo{person}{Craig Schroeder}, \bibinfo{person}{Alexey Stomakhin}, \bibinfo{person}{Joseph Teran}, {and} \bibinfo{person}{Pirouz Kavehpour}.} \bibinfo{year}{2015}\natexlab{}.
\newblock \showarticletitle{A material point method for viscoelastic fluids, foams and sponges}. In \bibinfo{booktitle}{\emph{Proceedings of the 14th ACM SIGGRAPH/Eurographics Symposium on Computer Animation}}. \bibinfo{pages}{157--163}.
\newblock


\bibitem[Roy-Chowdhury et~al\mbox{.}(2024)]%
        {RoyChowdhury:2024:HOD}
\bibfield{author}{\bibinfo{person}{Ritoban Roy-Chowdhury}, \bibinfo{person}{Tamar Shinar}, {and} \bibinfo{person}{Craig Schroeder}.} \bibinfo{year}{2024}\natexlab{}.
\newblock \showarticletitle{Higher order divergence-free and curl-free interpolation on MAC grids}.
\newblock \bibinfo{journal}{\emph{J. Comput. Phys.}} (\bibinfo{year}{2024}), \bibinfo{pages}{112831}.
\newblock


\bibitem[Saad(2003)]%
        {Saad:2003:IMS}
\bibfield{author}{\bibinfo{person}{Yousef Saad}.} \bibinfo{year}{2003}\natexlab{}.
\newblock \bibinfo{booktitle}{\emph{Iterative methods for sparse linear systems}}.
\newblock \bibinfo{publisher}{SIAM}.
\newblock


\bibitem[Sato et~al\mbox{.}(2018)]%
        {Sato:2018:SALS}
\bibfield{author}{\bibinfo{person}{Takahiro Sato}, \bibinfo{person}{Christopher Batty}, \bibinfo{person}{Takeo Igarashi}, {and} \bibinfo{person}{Ryoichi Ando}.} \bibinfo{year}{2018}\natexlab{}.
\newblock \showarticletitle{Spatially adaptive long-term semi-Lagrangian method for accurate velocity advection}.
\newblock \bibinfo{journal}{\emph{Computational Visual Media}} \bibinfo{volume}{4}, \bibinfo{number}{3} (\bibinfo{year}{2018}), \bibinfo{pages}{223--230}.
\newblock


\bibitem[Saye(2016)]%
        {Saye:2016:IGM}
\bibfield{author}{\bibinfo{person}{Robert Saye}.} \bibinfo{year}{2016}\natexlab{}.
\newblock \showarticletitle{Interfacial gauge methods for incompressible fluid dynamics}.
\newblock \bibinfo{journal}{\emph{Science advances}} \bibinfo{volume}{2}, \bibinfo{number}{6} (\bibinfo{year}{2016}), \bibinfo{pages}{e1501869}.
\newblock


\bibitem[Schroeder et~al\mbox{.}(2022)]%
        {Schroeder:2022:LDP}
\bibfield{author}{\bibinfo{person}{Craig Schroeder}, \bibinfo{person}{Ritoban~Roy Chowdhury}, {and} \bibinfo{person}{Tamar Shinar}.} \bibinfo{year}{2022}\natexlab{}.
\newblock \showarticletitle{Local divergence-free polynomial interpolation on MAC grids}.
\newblock \bibinfo{journal}{\emph{J. Comput. Phys.}}  \bibinfo{volume}{468} (\bibinfo{year}{2022}), \bibinfo{pages}{111500}.
\newblock


\bibitem[Selle et~al\mbox{.}(2008)]%
        {Selle:2008:MCM}
\bibfield{author}{\bibinfo{person}{Andrew Selle}, \bibinfo{person}{Ronald Fedkiw}, \bibinfo{person}{Byungmoon Kim}, \bibinfo{person}{Yingjie Liu}, {and} \bibinfo{person}{Jarek Rossignac}.} \bibinfo{year}{2008}\natexlab{}.
\newblock \showarticletitle{An unconditionally stable MacCormack method}.
\newblock \bibinfo{journal}{\emph{Journal of Scientific Computing}} \bibinfo{volume}{35}, \bibinfo{number}{2} (\bibinfo{year}{2008}), \bibinfo{pages}{350--371}.
\newblock


\bibitem[Selle et~al\mbox{.}(2005)]%
        {Selle:2005:VPM}
\bibfield{author}{\bibinfo{person}{Andrew Selle}, \bibinfo{person}{Nick Rasmussen}, {and} \bibinfo{person}{Ronald Fedkiw}.} \bibinfo{year}{2005}\natexlab{}.
\newblock \showarticletitle{A Vortex Particle Method for Smoke, Water and Explosions}.
\newblock \bibinfo{journal}{\emph{ACM Transactions on Graphics (TOG)}} \bibinfo{volume}{24}, \bibinfo{number}{3} (\bibinfo{year}{2005}), \bibinfo{pages}{910–914}.
\newblock
\showISSN{0730-0301}
\urldef\tempurl%
\url{https://doi.org/10.1145/1073204.1073282}
\showDOI{\tempurl}


\bibitem[Stam(1999)]%
        {Stam:1999:SF}
\bibfield{author}{\bibinfo{person}{Jos Stam}.} \bibinfo{year}{1999}\natexlab{}.
\newblock \showarticletitle{Stable fluids}. In \bibinfo{booktitle}{\emph{Proceedings of the 26th annual conference on Computer graphics and interactive techniques}}. \bibinfo{pages}{121--128}.
\newblock


\bibitem[Stomakhin et~al\mbox{.}(2013)]%
        {Stomakhin:2013:SMPM}
\bibfield{author}{\bibinfo{person}{Alexey Stomakhin}, \bibinfo{person}{Craig Schroeder}, \bibinfo{person}{Lawrence Chai}, \bibinfo{person}{Joseph Teran}, {and} \bibinfo{person}{Andrew Selle}.} \bibinfo{year}{2013}\natexlab{}.
\newblock \showarticletitle{A material point method for snow simulation}.
\newblock \bibinfo{journal}{\emph{ACM Transactions on Graphics (TOG)}} \bibinfo{volume}{32}, \bibinfo{number}{4} (\bibinfo{year}{2013}), \bibinfo{pages}{1--10}.
\newblock


\bibitem[Stomakhin et~al\mbox{.}(2014)]%
        {Stomakhin:2014:AMPM}
\bibfield{author}{\bibinfo{person}{Alexey Stomakhin}, \bibinfo{person}{Craig Schroeder}, \bibinfo{person}{Chenfanfu Jiang}, \bibinfo{person}{Lawrence Chai}, \bibinfo{person}{Joseph Teran}, {and} \bibinfo{person}{Andrew Selle}.} \bibinfo{year}{2014}\natexlab{}.
\newblock \showarticletitle{Augmented MPM for phase-change and varied materials}.
\newblock \bibinfo{journal}{\emph{ACM Transactions on Graphics (TOG)}} \bibinfo{volume}{33}, \bibinfo{number}{4} (\bibinfo{year}{2014}), \bibinfo{pages}{1--11}.
\newblock


\bibitem[Sulsky et~al\mbox{.}(1995)]%
        {Sulsky:1995:MPM}
\bibfield{author}{\bibinfo{person}{Deborah Sulsky}, \bibinfo{person}{Shi-Jian Zhou}, {and} \bibinfo{person}{Howard~L Schreyer}.} \bibinfo{year}{1995}\natexlab{}.
\newblock \showarticletitle{Application of a particle-in-cell method to solid mechanics}.
\newblock \bibinfo{journal}{\emph{Computer physics communications}} \bibinfo{volume}{87}, \bibinfo{number}{1-2} (\bibinfo{year}{1995}), \bibinfo{pages}{236--252}.
\newblock


\bibitem[Tampubolon et~al\mbox{.}(2017)]%
        {Tampubolon:2017:MSS}
\bibfield{author}{\bibinfo{person}{Andre~Pradhana Tampubolon}, \bibinfo{person}{Theodore Gast}, \bibinfo{person}{Gergely Kl{\'a}r}, \bibinfo{person}{Chuyuan Fu}, \bibinfo{person}{Joseph Teran}, \bibinfo{person}{Chenfanfu Jiang}, {and} \bibinfo{person}{Ken Museth}.} \bibinfo{year}{2017}\natexlab{}.
\newblock \showarticletitle{Multi-species simulation of porous sand and water mixtures}.
\newblock \bibinfo{journal}{\emph{ACM Transactions on Graphics (TOG)}} \bibinfo{volume}{36}, \bibinfo{number}{4} (\bibinfo{year}{2017}), \bibinfo{pages}{1--11}.
\newblock


\bibitem[Taylor and Green(1937)]%
        {Taylor:1937:MPS}
\bibfield{author}{\bibinfo{person}{Geoffrey~Ingram Taylor} {and} \bibinfo{person}{Albert~Edward Green}.} \bibinfo{year}{1937}\natexlab{}.
\newblock \showarticletitle{Mechanism of the production of small eddies from large ones}.
\newblock \bibinfo{journal}{\emph{Proceedings of the Royal Society of London. Series A-Mathematical and Physical Sciences}} \bibinfo{volume}{158}, \bibinfo{number}{895} (\bibinfo{year}{1937}), \bibinfo{pages}{499--521}.
\newblock


\bibitem[Tessendorf and Pelfrey(2011)]%
        {Tessendorf:2011:CMF}
\bibfield{author}{\bibinfo{person}{Jerry Tessendorf} {and} \bibinfo{person}{Brandon Pelfrey}.} \bibinfo{year}{2011}\natexlab{}.
\newblock \showarticletitle{The characteristic map for fast and efficient vfx fluid simulations}. In \bibinfo{booktitle}{\emph{Computer Graphics International Workshop on VFX, Computer Animation, and Stereo Movies. Ottawa, Canada}}.
\newblock


\bibitem[Thomson(1868)]%
        {Thomson:1868:OVM}
\bibfield{author}{\bibinfo{person}{William Thomson}.} \bibinfo{year}{1868}\natexlab{}.
\newblock \showarticletitle{On Vortex Motion}.
\newblock \bibinfo{journal}{\emph{Earth and Environmental Science Transactions of the Royal Society of Edinburgh}} \bibinfo{volume}{25}, \bibinfo{number}{1} (\bibinfo{year}{1868}), \bibinfo{pages}{217--260}.
\newblock


\bibitem[Wang et~al\mbox{.}(2023)]%
        {Wang:2023:EC}
\bibfield{author}{\bibinfo{person}{Stephanie Wang}, \bibinfo{person}{Mohammad~Sina Nabizadeh}, {and} \bibinfo{person}{Albert Chern}.} \bibinfo{year}{2023}\natexlab{}.
\newblock \showarticletitle{Exterior Calculus in Graphics: Course Notes for a SIGGRAPH 2023 Course}. In \bibinfo{booktitle}{\emph{ACM SIGGRAPH 2023 Courses}} (Los Angeles, California) \emph{(\bibinfo{series}{SIGGRAPH '23})}. \bibinfo{publisher}{Association for Computing Machinery}, \bibinfo{address}{New York, NY, USA}, Article \bibinfo{articleno}{8}, \bibinfo{numpages}{126}~pages.
\newblock
\showISBNx{9798400701450}
\urldef\tempurl%
\url{https://doi.org/10.1145/3587423.3595525}
\showDOI{\tempurl}


\bibitem[Wolper et~al\mbox{.}(2019)]%
        {Wolpher:CDMPM}
\bibfield{author}{\bibinfo{person}{Joshuah Wolper}, \bibinfo{person}{Yu Fang}, \bibinfo{person}{Minchen Li}, \bibinfo{person}{Jiecong Lu}, \bibinfo{person}{Ming Gao}, {and} \bibinfo{person}{Chenfanfu Jiang}.} \bibinfo{year}{2019}\natexlab{}.
\newblock \showarticletitle{CD-MPM: continuum damage material point methods for dynamic fracture animation}.
\newblock \bibinfo{journal}{\emph{ACM Trans. Graph.}} \bibinfo{volume}{38}, \bibinfo{number}{4}, Article \bibinfo{articleno}{119} (\bibinfo{date}{jul} \bibinfo{year}{2019}), \bibinfo{numpages}{15}~pages.
\newblock
\showISSN{0730-0301}
\urldef\tempurl%
\url{https://doi.org/10.1145/3306346.3322949}
\showDOI{\tempurl}


\bibitem[Xiong et~al\mbox{.}(2022)]%
        {Xiong:2022:CMF}
\bibfield{author}{\bibinfo{person}{Shiying Xiong}, \bibinfo{person}{Zhecheng Wang}, \bibinfo{person}{Mengdi Wang}, {and} \bibinfo{person}{Bo Zhu}.} \bibinfo{year}{2022}\natexlab{}.
\newblock \showarticletitle{A Clebsch method for free-surface vortical flow simulation}.
\newblock \bibinfo{journal}{\emph{ACM Transactions on Graphics (TOG)}} \bibinfo{volume}{41}, \bibinfo{number}{4} (\bibinfo{year}{2022}), \bibinfo{pages}{1--13}.
\newblock


\bibitem[Yang et~al\mbox{.}(2021)]%
        {Yang:2021:CGF}
\bibfield{author}{\bibinfo{person}{Shuqi Yang}, \bibinfo{person}{Shiying Xiong}, \bibinfo{person}{Yaorui Zhang}, \bibinfo{person}{Fan Feng}, \bibinfo{person}{Jinyuan Liu}, {and} \bibinfo{person}{Bo Zhu}.} \bibinfo{year}{2021}\natexlab{}.
\newblock \showarticletitle{Clebsch gauge fluid}.
\newblock \bibinfo{journal}{\emph{ACM Transactions on Graphics (TOG)}} \bibinfo{volume}{40}, \bibinfo{number}{4} (\bibinfo{year}{2021}), \bibinfo{pages}{99:1--99:11}.
\newblock


\bibitem[Yin et~al\mbox{.}(2023)]%
        {Yin:2023:FC}
\bibfield{author}{\bibinfo{person}{Hang Yin}, \bibinfo{person}{Mohammad~Sina Nabizadeh}, \bibinfo{person}{Baichuan Wu}, \bibinfo{person}{Stephanie Wang}, {and} \bibinfo{person}{Albert Chern}.} \bibinfo{year}{2023}\natexlab{}.
\newblock \showarticletitle{Fluid Cohomology}.
\newblock \bibinfo{journal}{\emph{ACM Trans. Graph.}} \bibinfo{volume}{42}, \bibinfo{number}{4}, Article \bibinfo{articleno}{126} (\bibinfo{date}{jul} \bibinfo{year}{2023}), \bibinfo{numpages}{25}~pages.
\newblock
\showISSN{0730-0301}
\urldef\tempurl%
\url{https://doi.org/10.1145/3592402}
\showDOI{\tempurl}


\bibitem[Z{\"a}ngl(2013)]%
        {Zangl:2013:ICON}
\bibfield{author}{\bibinfo{person}{Gunther Z{\"a}ngl}.} \bibinfo{year}{2013}\natexlab{}.
\newblock \showarticletitle{ICON: The icosahedral nonhydrostatic modelling framework of DWD and MPI-M}. In \bibinfo{booktitle}{\emph{Proc. ECMWF Seminar on Numerical Methods for Atmosphere and Ocean Modelling}}.
\newblock


\bibitem[Zehnder et~al\mbox{.}(2018)]%
        {Zehnder:2018:ARS}
\bibfield{author}{\bibinfo{person}{Jonas Zehnder}, \bibinfo{person}{Rahul Narain}, {and} \bibinfo{person}{Bernhard Thomaszewski}.} \bibinfo{year}{2018}\natexlab{}.
\newblock \showarticletitle{An advection-reflection solver for detail-preserving fluid simulation}.
\newblock \bibinfo{journal}{\emph{ACM Transactions on Graphics (TOG)}} \bibinfo{volume}{37}, \bibinfo{number}{4} (\bibinfo{year}{2018}), \bibinfo{pages}{85:1--85:8}.
\newblock


\bibitem[Zhang et~al\mbox{.}(2015)]%
        {Zhang:2015:RMV}
\bibfield{author}{\bibinfo{person}{Xinxin Zhang}, \bibinfo{person}{Robert Bridson}, {and} \bibinfo{person}{Chen Greif}.} \bibinfo{year}{2015}\natexlab{}.
\newblock \showarticletitle{Restoring the missing vorticity in advection-projection fluid solvers}.
\newblock \bibinfo{journal}{\emph{ACM Transactions on Graphics (TOG)}} \bibinfo{volume}{34}, \bibinfo{number}{4} (\bibinfo{year}{2015}), \bibinfo{pages}{52:1--52:8}.
\newblock


\bibitem[Zhang et~al\mbox{.}(2022)]%
        {Zhang:2022:MEH}
\bibfield{author}{\bibinfo{person}{Yi Zhang}, \bibinfo{person}{Artur Palha}, \bibinfo{person}{Marc Gerritsma}, {and} \bibinfo{person}{Leo~G Rebholz}.} \bibinfo{year}{2022}\natexlab{}.
\newblock \showarticletitle{A mass-, kinetic energy-and helicity-conserving mimetic dual-field discretization for three-dimensional incompressible Navier-Stokes equations, part I: Periodic domains}.
\newblock \bibinfo{journal}{\emph{J. Comput. Phys.}}  \bibinfo{volume}{451} (\bibinfo{year}{2022}), \bibinfo{pages}{110868}.
\newblock


\bibitem[Zhu and Bridson(2005)]%
        {Zhu:2005:ASF}
\bibfield{author}{\bibinfo{person}{Yongning Zhu} {and} \bibinfo{person}{Robert Bridson}.} \bibinfo{year}{2005}\natexlab{}.
\newblock \showarticletitle{Animating sand as a fluid}.
\newblock \bibinfo{journal}{\emph{ACM Transactions on Graphics (TOG)}} \bibinfo{volume}{24}, \bibinfo{number}{3} (\bibinfo{year}{2005}), \bibinfo{pages}{965--972}.
\newblock


\end{thebibliography}
\appendix
\section{Preliminaries}
\label{app:Preliminaries}
In this appendix, we review the essentials of Lie algebra, symplectic structure, Poisson structure, and Hamiltonian systems in a self-contained manner.  A more complete introduction to these topics can be found in textbooks such as \cite{Marsden:1997:IMS, Holm:2011:GMD}.
Our introduction uses notations in exterior calculus [\citealt[Appendix A]{Yin:2023:FC}; \citealt{Wang:2023:EC}; \citealt{Crane:2013:DEC}] and terminologies in group theory.

Most of the technical terminologies in the group theoretic Hamiltonian fluid dynamics are borrowed from geometric Hamiltonian mechanics.
In this paradigm, a Hamiltonian dynamical system is described by an energy function defined on a \emph{phase space}.  
This energy function is called the \emph{Hamiltonian}. 

In the most general setting, the type of a phase space is a \emph{Poisson space}, which is a manifold or a vector space equipped with a \emph{Poisson structure}.
The Poisson structure allows turning the Hamiltonian into a vector field (analogous to how a metric turns a loss function into a gradient field in numerical optimization).  The vector field generated by the Hamiltonian and the Poisson structure is called the \emph{symplectic gradient field} of the Hamiltonian.  The Hamiltonian dynamical system is the negative symplectic gradient flow on the phase space (analogous to gradient descents).

Two most important examples of Poisson spaces are \emph{symplectic manifolds} and \emph{dual spaces of Lie algebras}.  Let us briefly review the definitions and notations of these concepts.

\subsection{Lie Algebra}

\begin{definition}[Lie algebra]
    A \textbf{Lie algebra} \((V,[\cdot,\cdot])\) is a vector space \(V\) equipped with a skew-symmetric bilinear function \([\cdot,\cdot]\colon V\times V\xrightarrow{\rm bilinear} V\) satisfying the Jacobi identity: \([\vec a,[\vec b,\vec c]] + [\vec b,[\vec c,\vec a]] + [\vec c, [\vec a,\vec b]] = 0\) for all \(\vec a,\vec b,\vec c\in V\).
\end{definition}

Elementary examples of Lie algebra products include the 3D cross product \((\RR^3,\times)\) and matrix commutator \((\RR^{n\times n}=\gl(\RR^n),[\bA,\bB]\coloneqq \bA\bB-\bB\bA)\).
\begin{definition}[Adjoint representation]
\label{def:ad}
    The currying\footnote{As the currying operation in functional programming.} of the Lie bracket \([\cdot,\cdot]\) of a Lie algebra \(V\) is called the \textbf{adjoint representation}
    \begin{align}
        \ad\colon V\to\underbrace{V\to V}_{\gl(V)},\quad \ad_{\vec u}(\vec v)\coloneqq[\vec u,\vec v].
    \end{align}
\end{definition}
As a direct consequence of the Jacobi identity for \([\cdot,\cdot]\):
\begin{proposition}
\label{prop:adIsLieAlgHom}
    \(\ad\colon (V,[\cdot,\cdot])\rightarrow(\gl(V),[\cdot,\cdot]_{\mathrm{commutator}})\) is a Lie algebra homomorphism.
\end{proposition}
\begin{proof}
    By currying the Jacobi identity \([\vec a,[\vec b,\vec c]] + [\vec b,[\vec c,\vec a]] + [\vec c, [\vec a,\vec b]] = 0\) we obtain \(\ad_{\vec a}\ad_{\vec b}\vec c - \ad_{\vec b}\ad_{\vec a}\vec c - \ad_{[\vec a,\vec b]}\vec c = 0\).  Therefore \(\ad_{[\vec a,\vec b]} = [\ad_{\vec a},\ad_{\vec b}]_{\mathrm{commutator}}\).
\end{proof}

Here, a \textbf{Lie algebra homomorphism} between two Lie algebras \((U,[\cdot,\cdot]), (V,[\cdot,\cdot])\) is referred to as a linear map \(A\colon U\to V\) so that \(A([\vec u_1,\vec u_2]) = [A\vec u_1, A\vec u_2]\).  A \textbf{Lie algebra anti-homomorphism} is a linear map satisfying  \(A([\vec u_1,\vec u_2]) = -[A\vec u_1, A\vec u_2]\) instead.

The classes of Lie algebras that will be relevant in our context are the following two examples.
\begin{example}
\label{ex:TeGIsLieAlgebra}
    The tangent space \(\frak g = T_eG\) to a Lie group \(G\) at the identity \(e\in G\) is a Lie algebra.
    The Lie bracket is given by the mixed 2nd derivative of the conjugation expression \((g\in G,h\in G)\mapsto ghg^{-1}\in G\) at the identity
    \begin{align}
        [\vec u,\vec v]\coloneqq {\partial\over\partial g}\left({\partial\over\partial h}\left(ghg^{-1}\right)_{h=e}\llbracket \vec v\rrbracket\right)_{g=e}\llbracket\vec u\rrbracket.
    \end{align}
\end{example}

The next example says that the space of vector fields on a domain is naturally a Lie algebra.  To concisely describe it, the space \(\Gamma(TM)\) of tangent vector fields on a manifold \(M\) is identified with the space \(\Der(\continuity^\infty(M))\) of \textbf{derivations} for scalar functions on \(M\).  Here, a derivation is any linear operator \(D\colon \continuity^\infty(M)\xrightarrow{\rm linear}\continuity^\infty(M)\) satisfying the Leibniz rule \(D(fg) = D(f) g + f D(g)\).  
Every tangent vector field \(\vec u\) can be taken as a derivation, being the directional derivative operator: \((\vec u\rhd{})\in\Der(\continuity^\infty(M))\), \(\vec u\rhd f\coloneqq df\llbracket \vec u\rrbracket\).
Conversely,
each derivation can be regarded as a directional derivative operator along a vector field \cite{nlab:derivations_of_smooth_functions_are_vector_fields}. In graphics, this representation of vector fields has been used for tangent vector field processing and solving PDEs on surfaces \cite{Azencot:2013:OAT,Azencot:2014:FFS,Azencot:2015:DDV}.

\begin{example}
\label{ex:DerIsLieAlgebra}
    The space \(\fX(M)=\Der(\continuity^\infty(M))\) of tangent vector fields on a manifold \(M\) is a Lie algebra, with \([\vec u,\vec v]\rhd\coloneqq(\vec u\rhd{})(\vec v\rhd{}) - (\vec v\rhd{})(\vec u\rhd{})\).
\end{example}

The following proposition combines \exref{ex:TeGIsLieAlgebra} and \exref{ex:DerIsLieAlgebra} in the context of group actions.
A \textbf{group action} of a Lie group \(G\) on a manifold \(M\) is a (anti-)homomorphism \(\Phi\) from \(G\) to smooth maps from \(M\) to \(M\) (with composition ``\(\circ\)'' of smooth maps as the group operation).  Spelling it out, we have type \(\Phi\colon G\xrightarrow{\text{hom}}(M\to M)\).  That is, \(\Phi(g_1g_2)(x) = \Phi(g_1)\Phi(g_2)(x)\), \(g_1,g_2\in G,x\in M\).  We call such \(\Phi\) a \textbf{left action}.
If \(\Phi\colon G\xrightarrow{\text{anti-hom}}(M\to M)\) is an anti-homomorphism, \ie\@ \(\Phi(g_1g_2)(x) = \Phi(g_2)\Phi(g_1)(x)\), we call \(\Phi\) a \textbf{right action}.

Note that if \(\Phi\colon G\xrightarrow{\text{(anti-)hom}}(M\to M)\) is a group action, then its differential at the identity is a linear map from \(\frak g=T_eG\) to the space of tangent vector fields \(\fX(M)\) on \(M\), representing an infinitesimal motion.
\begin{proposition}\label{prop:DifferentialOfAction}
    The differential \(d\Phi|_e\) of a left action \(\Phi\colon G\xrightarrow{\textnormal{hom}}(M\to M)\) at the identity is a Lie algebra anti-homomorphism \(d\Phi|_e\colon\frak g\xrightarrow{\substack{\textnormal{Lie alg}\\\textnormal{anti-hom}}}\fX(M)\).  Likewise, the differential of a right action \(\Phi\colon G\xrightarrow{\textnormal{anti-hom}}(M\to M)\) at the identity is a Lie algebra homomorphism \(d\Phi|_e\colon\frak g\xrightarrow{\substack{\textnormal{Lie alg}\\\textnormal{hom}}}\fX(M)\).
\end{proposition}
\begin{proof}
    \cite[Theorem~20.15]{Lee:2013:SM}.
\end{proof}

\subsection{Poisson Structure}

A Hamiltonian dynamical system is a 1st order ODE on a phase space.  This phase space has a structure that turns a Hamiltonian function into a vector field describing the ODE.  In the most general setting, this structure is the Poisson structure.
There are two special examples of spaces that are naturally endowed with Poisson structures.  One example is any symplectic space (\secref{sec:SymplecticManifoldsArePoisson}).  The second example is the dual vector space of any Lie algebra (\secref{sec:DualLieAlgArePoisson}).
\begin{definition}
    A \textbf{Poisson manifold} (or \textbf{Poisson space}) is a manifold (or vector space) \(P\) equipped with a Lie algebra structure on the scalar functions \(\{\cdot,\cdot\}\colon \continuity^\infty(P)\times \continuity^\infty(P)\xrightarrow{\textnormal{Lie bracket}} \continuity^\infty(P)\), satisfying that \(\{f,\cdot\}\colon \continuity^\infty(P)\to \continuity^\infty(P)\) is a derivation for every \(f\in \continuity^\infty(P)\).
\end{definition}
\begin{definition}\label{def:sgradPoisson}
    The vector field \(\vec X^f\in\Gamma(TP)\) associated with the derivation \(\{f,\cdot\}=\vec X^f\rhd(\cdot)\) is called the \textbf{symplectic gradient field} (or the \textbf{Hamiltonian vector field}) of \(f\in \continuity^\infty(P)\), denoted by \(\sgrad f\coloneqq\vec X^f\). 
\end{definition}
\begin{definition}[Hamiltonian system]\label{def:HamiltonianDynamicalSystem}
    A Hamiltonian dynamical system \((P,H)\) for a Hamiltonian function \(H\in \continuity^\infty(P)\) on a Poisson manifold \(P\) is a time-dependent position \(x(t)\in P\) governed by the ODE \({d\over dt}x(t) = -\sgrad H|_{x(t)}\).
\end{definition}

\subsubsection{Properties of \(\sgrad\)}
The operator \(\sgrad\) is derived directly from the Poisson structure \(\{\cdot,\cdot\}\).  It turns each Hamiltonian function into a Hamiltonian dynamical system.  The following is an important property of \(\sgrad\).  
\begin{proposition}\label{prop:sgradIsLieAlgHom}
    The symplectic gradient operator
    \begin{align}
        \sgrad\colon (\continuity^\infty(P),\{\cdot,\cdot\})\to (\fX(P),[\cdot,\cdot])
    \end{align}
    is a Lie algebra homomorphism. 
\end{proposition}
\begin{proof}
    Observe that \(\sgrad\) is the currying of the Poisson bracket (\(\{f,\cdot\} = \sgrad(f)\rhd\)) analogous to that \(\ad\) is the currying of a Lie bracket.  The proof follows by replicating the proof of \propref{prop:adIsLieAlgHom}.
\end{proof}
\begin{definition}
\label{def:hamP}
    The space \(\ham(P)\) of Hamiltonian vector fields on \(P\) is the Lie subalgebra \(\ham(P)\coloneqq\im(\sgrad)\subset\fX(P)\).
\end{definition}
Note that the ``integration'' operator (defined up to a Casimir function which we will introduce later in \defref{def:Casimir})
\begin{align}
\label{eq:sgradInverse}
    \sgrad^{-1}\colon (\ham(P),[\cdot,\cdot])\to (\continuity^\infty(P),\{\cdot,\cdot\})
\end{align}
is also a Lie algebra homomorphism.  This will be convenient when we discuss momentum maps.

\subsubsection{Poisson Maps}
A Poisson map is a Poisson-structure preserving map.  As such Hamiltonian systems map to each other.
\begin{definition}
\label{def:PoissonMap}
    A map \(\phi\colon P_1\to P_2\) between two Poisson manifolds \((P_1,\{\cdot,\cdot\}_1)\) and \((P_2,\{\cdot,\cdot\}_2)\) is called a \textbf{Poisson map} if \(\{f\circ \phi,g\circ \phi\}_1 = \{f,g\}_2\circ\phi\) for all \(f,g\in\continuity^\infty(P_2)\).
    The map is called an \textbf{anti-Poisson map} if \(\{f\circ \phi,g\circ \phi\}_1 = -\{f,g\}_2\circ\phi\).
\end{definition}
\begin{proposition}
\label{prop:PoissonMapHamiltonianSystem}
    Let \(\phi\colon (P_1,\{\cdot,\cdot\}_1)\to (P_2,\{\cdot,\cdot\}_2)\) be a Poisson map. 
    Let \(H\in\continuity^\infty(P_2)\) be a Hamiltonian function defining a Hamiltonian system \((P_2,H)\).
    If \(x(t)\in P_1\) is a solution to the dynamical system \((P_1,H\circ \phi)\) then \((\phi\circ x)(t)\) is a solution to the system \((P_2, H)\).
    If \(\phi\) is an anti-Poisson map instead, then \(\phi\circ x\) is a solution to the system \((P_2,H)\) but reversed in time.
\end{proposition}

\subsection{Symplectic Manifolds are Poisson}
\label{sec:SymplecticManifoldsArePoisson}
In elementary classical mechanics, the phase space of a mechanical system is introduced as the space of positions \(q\) and the corresponding momenta \(p\).  These position--momentum (\(qp\)) phase spaces are \emph{symplectic spaces} or, more generally, symplectic manifolds.  
Indeed, symplectic manifolds are special cases of Poisson manifolds.   

\begin{definition}\label{def:SymplecticManifold}
    A \textbf{symplectic manifold} is a manifold \(\Sigma\) equipped with a non-degenerate 2-form \(\sigma\in\Omega^2(\Sigma)\).
\end{definition}
By currying, the symplectic structure \(\sigma\) gives rise to a ``flat'' operator (analogous to the \(\flat,\sharp\) for an inner product structure)
\begin{align}
    \hat\flat_x\colon T_x\Sigma\xrightarrow{\rm linear} T_x^*\Sigma,\quad \hat\flat_x(\vec v)(\cdot) \coloneqq \sigma_x(\vec v,\cdot),\quad x\in\Sigma, \vec v\in T_x\Sigma.
\end{align}
The non-degeneracy of \(\sigma\) means that \(\hat\flat\) is invertible, defining \(\hat\sharp_x\coloneqq \hat\flat_x^{-1}\colon T_x^*\Sigma\to T_x\Sigma\).
Using a symplectic structure, one defines the symplectic gradient of a function \(f\in \continuity^\infty(\Sigma)\) by
\begin{align}
\label{eq:sgradForSymplectic}
    \sgrad f \coloneqq \hat\sharp(df),\quad\text{equivalently}\quad \ip_{\sgrad f}\sigma = df,
\end{align}
where \(\ip_{(\cdot)}(\cdot)\) is the interior product.
Every symplectic structure gives rise to a natural Poisson structure:
\begin{subequations}
\label{eq:PoissonFromSymplectic}
    \begin{align}
    \{f,g\}&\coloneqq -\sigma(\sgrad f,\sgrad g)\\& = -df\llbracket\sgrad g\rrbracket = dg\llbracket \sgrad f\rrbracket.
\end{align}
\end{subequations}

\begin{definition}\label{def:SymplectomorphicVectorField}
    A vector field \(\vec X\in\fX(\Sigma)\) is \textbf{symplectomorphic} if \(\LD_{\vec X}\sigma = 0\).  That is, the flow \(\varphi^t\colon\Sigma\to\Sigma\) generated by \(\vec X\) (\({\partial\over\partial t}\varphi^t = \vec X\circ\varphi^t, \varphi^0 = \id_\Sigma\)) preserves the symplectic form \(\varphi^{t*}\sigma = \sigma\).
\end{definition}
\begin{proposition}\label{prop:LiouvilleTheorem}
    All Hamiltonian vector fields are symplectomorphic (Liouville theorem).  On simply connected symplectic manifolds, all symplectomorphic vector fields are also Hamiltonian.
\end{proposition}
\begin{proof}
    By Cartan's formula \(\LD_{\sgrad f}\sigma \overset{d\sigma=0}{=} d\ip_{\sgrad f}\sigma \overset{\eqref{eq:sgradForSymplectic}}{=} ddf = 0\).
    Conversely, if \(\LD_{\vec X}\sigma = 0\), then \(d\ip_{\vec X}\sigma=0\).  On simply-connected \(\Sigma\) this implies that there exists \(f\) so that \(df = \ip_{\vec X}\sigma\); that is, \(\vec X = \sgrad f\).
\end{proof}
\begin{example}[\(\bq\bp\)-space]
\label{ex:qpSpace}
    The space of positions \(\bq\in\RR^m\) and momenta \(\bp\in\RR^m\) of a generic finite dimensional mechanical system form a symplectic space \(\Sigma = \RR^{2m}\ni(\bq,\bp)\) with symplectic form \(\sigma = \sum_{\sfi=1}^m dp_{\sfi}\wedge dq_{\sfi}\).  Note that \(\sigma = d\vartheta\) where \(\vartheta = \sum_{\sfi=1}^m p_{\sfi}dq_{\sfi}\) is called the Liouville 1-form.
    Now, given any function \(H\in\continuity^\infty(P)\), the Hamiltonian dynamical system (\defref{def:HamiltonianDynamicalSystem}) is given by \({d\over dt}q_{\sfi}(t) = {\partial H\over\partial p_{\sfi}}\), \({d\over dt}p_{\sfi}(t) = -{\partial H\over\partial q_{\sfi}}\).  One checks that this \(\sgrad H\) satisfies \(\ip_{\sgrad H}\sigma = dH = \sum_{\sfi=1}^m {\partial H\over\partial q_{\sfi}}dq_{\sfi} + {\partial H\over\partial p_{\sfi}}dp_{\sfi}\).
\end{example}
\begin{example}[Cotangent bundles]
\label{ex:CotangentBundles}
The coordinate-free version of \exref{ex:qpSpace} is the cotangent bundle \(\Sigma = T^*Q\) over a space \(Q\) of positions.  Each element \((q,p)\in T^*Q\) takes the form of \(q\in Q\) and \(p\in T_q^*Q\).  Let \(\pi\colon T^*Q\to Q\), \((q,p)\mapsto q\), denote the projection from the bundle to the base.  Then we can define the Liouville 1-form \(\vartheta\in\Omega^1(\Sigma)\) by \(\vartheta_{(q,p)}\llbracket \vec X\rrbracket \coloneqq \langle p| d\pi(\vec X) \rangle \) for \(\vec X\in T_{(q,p)}(\Sigma)\).
The symplectic form \(\sigma\) for \(\Sigma\) is given by \(\sigma = d\vartheta\).
\end{example}

\subsubsection{Lifted Map and Lifted Vector Field}
\label{sec:LiftedMapAndLiftedVectorField}

There is a general way for constructing a symplectomorphic map on a cotangent bundle \(T^*Q\).  Suppose \(\varphi\colon Q\to Q\) is a map on \(Q\).  Then its \textbf{lifted map}, given by \(\tilde\varphi \colon T^*Q\to T^*Q\),
\(\tilde\varphi(q,p) \coloneqq (\varphi(q),(\varphi^{-1})^*p)\), is a symplectomorphism.  That is, \(\tilde\varphi^*\sigma = \sigma\).

Similarly, each vector field \(\vec X\in \fX(Q)\) can be lifted into a symplectomorphic vector field \(\tilde{\vec X}\), whose flow map is the lifted map of the flow map generated by \(\vec X\).  Concretely, \(\tilde{\vec X} = -\sgrad H\) where \(H\colon T^*Q\to\RR, H(q,p)\coloneqq \langle p|\vec X\rangle\).

\subsection{Poisson Manifolds have Symplectic Foliations}
\label{sec:PoissionSymplectic}
We have seen that every symplectic manifold is Poisson.  What about the converse?  It turns out that every Poisson manifold is foliated into (\ie\@ decomposed into a disjoint union of) symplectic submanifolds, as discussed below.

If one is given a Poisson manifold \((P,\{\cdot,\cdot\})\), one can reconstruct the \(\hat\sharp\) operator so that the symplectic gradient of \defref{def:sgradPoisson} can be represented by \(\sgrad f = \hat\sharp df\).  If the \(\hat\sharp\) operator from a Poisson structure is invertible, then the Poisson manifold is symplectic with symplectic structure \(\hat\flat\coloneqq\hat\sharp^{-1}\).  A more general Poisson manifold may have a non-invertible \(\hat\sharp\).  Remarkably, even when \(\hat\sharp\) is not invertible, the subspace field \(\im(\hat\sharp_x)\subset T_xP\) is integrable; \ie\@ \(P\) is foliated into submanifolds with \(\im(\hat\sharp_x)\) being the tangent spaces of these submanifolds.
These submanifolds are symplectic since \(\hat\sharp\) is invertible on their tangent spaces.  Each of these submanifolds is called a \textbf{symplectic leaf}.

Note that since \(\sgrad f = \hat\sharp df\), we have \(\sgrad f\in \im(\hat\sharp)\) tangential to the symplectic leaves.  Therefore, every Hamiltonian dynamical system on a Poisson manifold stays in a symplectic leaf.

\begin{definition}[Casimir]
\label{def:Casimir}
    A function \(C\colon P\to\RR\) on a Poisson manifold \(P\) is called a \textbf{Casimir} if it is constant on each symplectic leaf.
\end{definition}
\begin{proposition}
    Each Casimir function is conserved under any Hamiltonian dynamical system on a Poisson manifold.
\end{proposition}
\begin{proof}
    Every Hamiltonian flow stays in a symplectic leaf, on which any Casimir function is constant.
\end{proof}

\subsection{Dual Lie Algebras are Poisson}
\label{sec:DualLieAlgArePoisson}
Lie algebra: \((V,[\cdot,\cdot])\).  Dual Lie algebra: its dual vector space \(V^*\).  The Lie--Poisson structure:
\begin{subequations}
\label{eq:LiePoissonBracket}
    \begin{align}
    \label{eq:LiePoissonBracketA}
    &\{\cdot,\cdot\}\colon \continuity^\infty(V^*)\times \continuity^\infty(V^*)\to \continuity^\infty(V^*)\\
    \label{eq:LiePoissonBracketB}
    &\{f,g\}_\alpha\coloneqq \left\langle\alpha\,\Big|\,\left[df_{\alpha},dg_\alpha\right] \right\rangle,\quad\alpha\in V^*, f,g\in \continuity^\infty(V^*).
\end{align}
\end{subequations}
Here \(df_\alpha\in T^*_{\alpha}V^*\) (and \(dg_\alpha\)) is the differential of \(f\) at \(\alpha\), which can be viewed as an element of \(V\) (by \(T_\alpha^*V^* = V^{**} = V\)).  Thus it makes sense to take \([\cdot,\cdot]\) between \(df_\alpha, dg_\alpha\), returning another element in \(V\), which is finally paired with the covector \(\alpha\in V^*\).

\begin{example}[Dual space of functions on a Poisson space]
\label{ex:DualSpaceOfFunctionsOnPoissonSpace}
Here is an example of a dual Lie algebra. Recall that \(\continuity^\infty(P)\) of a Poisson space \((P,\{\cdot,\cdot\})\) is a Lie algebra with Lie bracket being \(\{\cdot,\cdot\}\). Its dual space, \(\continuity^\infty(P)^*\), is the space of distributions.  A subset of distributions is the space of point measures, identified with \(P\) itself, \(P\subset \continuity^\infty(P)^*\) with the dual pairing being the point evaluation.  The Lie--Poisson structure \eqref{eq:LiePoissonBracket} on \(P\subset \continuity^\infty(P)^*\) as in the dual Lie algebra is just the original Poisson structure on \(P\).
\end{example}

\begin{proposition}
    The symplectic gradient \(\sgrad f\in\fX(V^*)\) of a function \(f\colon V^*\to\RR\) on the dual Lie algebra \(V^*\) with the Lie--Poisson bracket \eqref{eq:LiePoissonBracket} is given by
    \begin{align}
    \label{eq:coadjointsgrad}
        \sgrad f|_{\alpha} = \ad_{df_\alpha}^\adjoint \alpha,\quad\alpha\in V^*
    \end{align}
\end{proposition}
\begin{proof}
    Using the adjoint representation of \defref{def:ad}, rewrite \eqref{eq:LiePoissonBracketB} as \(\{f,g\}_\alpha = \langle \alpha | \ad_{df_\alpha}(dg_\alpha)\rangle = \langle \ad_{df_\alpha}^\intercal \alpha | dg_\alpha\rangle = \ad_{df_\alpha}^\intercal \rhd g\).  Eq.~\eqref{eq:coadjointsgrad} follows by \defref{def:sgradPoisson}.
\end{proof}

\begin{definition}[Coadjoint action]
    Let \(\vec v\in V\) be an element of a Lie algebra \(V\).  The operator \(\ad_{\vec v}^\intercal\colon V^*\to V^*\) is an \textbf{infinitesimal coadjoint action} by \(\vec v\).
\end{definition}
\begin{definition}[Coadjoint orbits]
\label{def:CoadjointOrbits}
    Let \(\alpha_0\in V^*\) be a point in a dual Lie algebra \(V^*\). Consider the set \(\cO_{\alpha_0}\) of all possible endpoints \(\gamma(1)\) of arbitrary paths \(\gamma\colon [0,1]\to V^*\) starting at \(\gamma(0) = \alpha_0\) so that \(\gamma'(t) = \ad_{\vec v(t)}^\intercal (\gamma(t))\) for some \(\vec v(t)\in V\).  That is, \(\cO_{\alpha_0}\) is the set of points that can be reached by infinitesimal coadjoint actions.  The set \(\cO_{\alpha_0}\subset V^*\) is called the \textbf{coadjoint orbit} of \(\alpha_0\).
\end{definition}

The following proposition is an immediate consequence of \eqref{eq:coadjointsgrad}.
\begin{proposition}
    Any path of a Hamiltonian dynamical system on any dual Lie algebra stays on a coadjoint orbit.
\end{proposition}
\begin{proof}
    Hamiltonian dynamical systems (\defref{def:HamiltonianDynamicalSystem}) is in the form of coadjoint actions by \eqref{eq:coadjointsgrad}.  Therefore they generate paths that lie in coadjoint orbits.
\end{proof}

In fact, coadjoint orbits are exactly the symplectic leaves.
\begin{proposition}
\label{prop:CoadjointOrbitsCoincideWithSymplecticLeaves}
    Coadjoint orbits coincide with the symplectic leaves in any dual Lie algebra.
\end{proposition}
\begin{proof}
    At each point \(\alpha\in V^*\), the subspace \(\{\sgrad f|_{\alpha}|f\in \continuity^\infty(V^*)\}\) coincides with \(\im(\hat\sharp_\alpha)\) and, by \eqref{eq:coadjointsgrad}, \(\{\ad_{\vec v}^\intercal\alpha|\vec v\in V\}\).  \(\im(\hat\sharp)\) integrates into symplectic leaves, and \(\{\ad_{\vec v}^\intercal\alpha|\vec v\in V\}\) integrates into coadjoint orbits.  
\end{proof}

Note that these Poisson structures and coadjoint orbit structures in a dual Lie algebra are canonically derived directly from the Lie algebra structure.  An immediate consequence is the following proposition.

\begin{proposition}
\label{prop:AdjointOfLAHomIsPoisson}
Let \(U,V\) be two Lie algebras.  If \(A\colon U\to V\) is a Lie algebra homomorphism, then \(A^\intercal\colon V^*\to U^*\) is a Poisson map.  Likewise if \(A\colon U\to V\) is a Lie algebra anti-homomorphism, then \(A^\intercal\colon V^*\to U^*\) is an anti-Poisson map.
\end{proposition}

\subsection{Momentum Maps}
Let \(P\) be a Poisson space.  Let \(V\) be a linear subspace of Hamiltonian vector fields (\defref{def:hamP}) on \(P\) with the inclusion map
\begin{align}
    I\colon V\xhookrightarrow{\rm linear}\ham(P).
\end{align}
\begin{definition}[Momentum map]
    A \textbf{momentum map} associated to \(I\colon V\hookrightarrow\ham(P)\) is a function \(J\colon P\to V^*\)
    so that \(\langle J(\cdot)|\vec v\rangle\colon P\to\RR\) is the Hamiltonian function that generates the Hamiltonian vector field \(I(\vec v)\) for all \(\vec v\).  That is,
    \begin{align}
    \label{eq:MomentumMapAssociatedToI}
        -\sgrad \langle J|\vec v\rangle = I(\vec v),\quad \vec v\in V. 
    \end{align}
\end{definition}
An explicit formula for a momentum map is
\begin{align}
\label{eq:MomentumMapFormula}
    J = -I^\intercal\circ \sgrad^{-\adjoint}.
\end{align}
To understand, recall \(\sgrad^{-1}\) is a map of type \(\sgrad^{-1}\colon \ham(P)\to \continuity^\infty(P)\), which is a Lie algebra homomorphism (\cf\@ \eqref{eq:sgradInverse}).   In particular, \(\sgrad^{-\adjoint}\) is a Poisson map from distributions \(\continuity^\infty(P)^*\) to \(\ham(P)^*\) (\propref{prop:AdjointOfLAHomIsPoisson}). Here we restrict \(\sgrad^{-\intercal}\) to \(P\subset \continuity^\infty(P)^*\); see \exref{ex:DualSpaceOfFunctionsOnPoissonSpace}.

\begin{proposition}[Momentum maps for Lie subalgebras are Poisson]
\label{prop:MomentumMapsArePoisson}
    If \(V\) is an (anti)-Lie subalgebra of \(\ham(P)\), \ie\@ \(I\) is a Lie algebra (anti-)homomorphism, then the momentum map \(J\) \eqref{eq:MomentumMapFormula} is an (anti-)Poisson map.
\end{proposition}
\begin{proof}
    \(J\) is a composition of a Poisson map \(\sgrad^{-\adjoint}\) and an (anti-)Poisson map \(I^\adjoint\) (\propref{prop:AdjointOfLAHomIsPoisson}).
\end{proof}

\begin{example}[Momentum maps for lifted vector fields on cotangent bundles]
\label{ex:MomentumMapsForLiftedVectorFields}
    Let \(P\) be a cotangent bundle \(T^*Q\).  Let \(V\subset\fX(Q)\) be a subspace of vector fields on \(Q\) and let \(I\colon V\to \ham(T^*Q)\), \(\vec v\mapsto\tilde{\vec v}\), be lifted vector fields (\secref{sec:LiftedMapAndLiftedVectorField}). Then we can write down the associated momentum map \(J\)  explicitly as \(\langle J(q,p)|\vec v\rangle =  \vartheta_{(q,p)}\llbracket I(\vec v)\rrbracket = \langle p|\vec v\rangle\), where \(\vartheta\) is the Liouville form (\exref{ex:CotangentBundles}).
\end{example}

\begin{example}[Momentum map associated to symplectomorphic group actions]
\label{ex:MomentumMapForSymplectomorphicGroupAction}
    Let \((\Sigma,\sigma)\) be a symplectic manifold.
    Let \(\Phi\colon G\xrightarrow{\rm hom}(\Sigma\to \Sigma)\) be a left action on \(\Sigma\) by a Lie group \(G\) so that \(\Phi_g^*\sigma = \sigma\) for all \(g\in G\).  That is, the group action produces symplectomorphic maps \(\Phi_g\colon \Sigma\to \Sigma\).
    Then \(d\Phi|_e\colon \frak g \to \fX(\Sigma)\) is a Lie algebra anti-homomorphism (\propref{prop:DifferentialOfAction}).
    Moreover, the image of \(d\Phi|_e\) is always a symplectomorphic vector field (\defref{def:SymplectomorphicVectorField}).
    Assuming \(\Sigma\) being simply-connected, by \propref{prop:LiouvilleTheorem}, these symplectomorphic vector fields are Hamiltonian vector fields.  
    That is, we have a Lie algebra anti-homomorphism \(I\coloneqq d\Phi|_e\colon \frak g \to \ham(\Sigma)\).
    Therefore the associated momentum map \(J\colon \Sigma\to\frak g^*\) is an anti-Poisson map (\propref{prop:MomentumMapsArePoisson}).

    Combining \eqref{eq:MomentumMapAssociatedToI} and \eqref{eq:sgradForSymplectic}, we have that the momentum maps for symplectomorphic actions satisfy
    \begin{align}
    \label{eq:MomentumMapForSymplectomorphicAction}
        -d\langle J|\vec v\rangle = \ip_{I(\vec v)}\sigma,\quad \vec v\in \frak g.
    \end{align}
    
\end{example}

\begin{example}[Momentum map associated to group actions on cotangent bundles]
\label{ex:MomentumMapForCotangentBundleAction}
    Let \(\Phi\colon G\xrightarrow{\rm hom}(Q\to Q)\) be a left action on \(Q\) by a Lie group \(G\).  Then \(d\Phi|_e\colon \frak g \to \fX(Q)\) is a Lie algebra anti-homomorphism (\propref{prop:DifferentialOfAction}).
    The vector fields can be lifted to the cotangent bundle using \secref{sec:LiftedMapAndLiftedVectorField} and \exref{ex:MomentumMapsForLiftedVectorFields}, giving us a Lie algebra anti-homomorphism \(I\coloneqq\widetilde{d\Phi|_e}\colon \frak g\to \ham(T^*Q)\).
    Therefore, the associated momentum map \(J\colon T^*Q\to \frak g^*\) is an anti-Poisson map (\propref{prop:MomentumMapsArePoisson}).
\end{example}

\section{Relation to Non-Holonomicity}
\label{app:RelationToNonHolonomicity}
Here we expand the discussion on the relation between the lack of Poisson structure for the dual space \(\fB_{\div}^*\) in \secref{sec:LackOfPoissonStructure} and the \emph{non-holonomicity} problem in the literature.

The problem here, that the space of discrete vector fields is not Lie algebraically closed, has been observed in previous work in computer graphics that searches for structure preserving discretizations of the Euler equation.
In the setup of \cite{Pavlov:2011:SPD,Mullen:2009:EPI}, the discretization is constructed by replacing the function space \(C^\infty(W)\) by a finite dimensional function space \(\RR^N\) (\eg\@ the piecewise linear functions over a simplicial mesh with \(N\) vertices).  This leads to that the Lie algebra \(\fX_{\div}(W)\) is replaced by a finite dimensional Lie algebra \(\so(N)\) of stochastic matrices as the generators of volume-preserving diffeomorphisms.   Although \(\so(N)\) is a Lie algebra, the only vector fields that are computationally meaningful (representing discrete directional derivative operators) are restricted to a subspace \(S\subset \so(N)\) characterized by an additional sparsity condition on the stochastic matrices. 
This sparsity constraint is referred to as a \emph{non-holonomic constraint} \cite[\S~2.4, \S~4.1]{Pavlov:2011:SPD}, which is equivalent to that the subspace \(S\) is not Lie algebraically closed.
The same non-holonomic constraint is present in the spectral variational integrator of \cite[\S~3.2]{Liu:2015:MVFS}.
In these works, the resulting equations of motion are no longer pure coadjoint actions on the circulation variables like \eqref{eq:HamiltonianFlowAsLieAdvection}.  Instead, these equations are in their \emph{weak form}, where the circulation variable is restricted to stay on \(S^*\) and the coadjoint equation is only tested against vectors in \(S\).
In particular, the coadjoint orbit conservation is not guaranteed, as remarked at the end of \S~4 of \cite{Pavlov:2011:SPD}.
This is consistent with our observation that the dual space \(S^*\) of a non-holonomic velocity space \(S\) is not a Poisson space, since otherwise \(S^*\) would have been foliated into coadjoint orbits conserved by the methods of \cite{Mullen:2009:EPI,Pavlov:2011:SPD,Liu:2015:MVFS}.

One structure-preserving method in computer graphics that is able to circumvent the non-holonomic projection is the 2D vorticity method of Azencot \etal\@  \shortcite{Azencot:2014:FFS}.  This method follows the setup of \cite{Pavlov:2011:SPD} that the space of velocity is a non-holonomic \(S\subset\so(N)\).  However, instead of formulating a dynamical system for elements in \(S^*\), the phase space is the space of vorticity function \(\RR^N\), which is acted by unitary matrices generated by elements in \(S\subset\so(N)\).  
In particular, we obtain exact conservation on the \(\ell^2\)-norm of the \(\RR^N\) vorticity, and indeed the \(\ell^2\)-norm on \(\RR^N\) is one of the Casimirs.  
Our method uses a similar technique (\secref{sec:AuxiliarySymplecticSpace}) of maintaining an advected object to obtain coadjoint orbit preservation, but this time on the original fluid's Lie algebra \(\fX_{\div}\) rather than on \(\so(N)\) and the method works in both 2D and 3D.
\section{Proofs}
\label{app:Proofs}

\subsection{Derivation of \eqref{eq:DiscretePressureProjection}}
\label{app:DiscretePressureProjection}
The orthogonal projection \(\bg  = \bP_{\fB_{\div}}\bff\) of \(\bff\in\fB\) to \(\fB_{\div} = \ker(\bd)\) is the solution to 
\begin{align}
    \textstyle{\rm minimize}_{\bg}\,\frac{1}{2}(\bg - \bff)^\intercal\star(\bg-\bff)\quad\text{subject to}\quad \bd\bg = \bzero
\end{align}
The minimizing condition is given by \(\star(\bg-\bff) + \bd^\intercal\bp = \bzero\) for some Lagrange multiplier \(\bp\).  Solving the system
\begin{align}
    \begin{cases}
        \star(\bg - \bff) + \bd^\intercal\bp = \bzero\\
        \bd\bg = \bzero
    \end{cases}
\end{align}
yields \(\bp = (\bd\star^{-1}\bd^\intercal)^{-1}\bd\bff\) and \(\bg =\bff - \star^{-1}\bd^{\intercal}\bp =  \bff - \star^{-1}\bd^{\intercal}(\bd\star^{-1}\bd^\intercal)^{-1}\bd\bff\).  Therefore, \(\bP_{\fB_{\div}} = \id {}-{} \star^{-1}\bd^\intercal (\bd\star^{-1}\bd^\intercal)^{-1} \bd\).
\qed

\subsection{Proof of \thmref{thm:PressureProjectionIsExact}}
\label{app:PressureProjectionIsExact}
In this proof we will repeatedly using the linear algebraic fact that \(\langle \bP\ba,\bb\rangle = \langle\ba,\bb\rangle\) if \(\bP\) is an orthogonal projection to a subspace \(B\) respecting the inner product structure \(\langle\cdot,\cdot\rangle\), and \(\bb\) lies in the subspace \(B\).

Let \(\cI\bg\in\cI(\fB_{\div})\)  be an arbitrary interpolated divergence-free vector field. Then
\begin{align}
    &\llangle\cI\bP_{\fB_{\div}}(\bff) - \bP_{\fX_{\div}}\cI(\bff),\cI(\bg)\rrangle \\
    &=
    \llangle \cI\bP_{\fB_{\div}}(\bff),\cI(\bg)\rrangle - \llangle \bP_{\fX_{\div}}\cI(\bff),\cI(\bg)\rrangle\\
    &=
    \langle \bP_{\fB_{\div}}(\bff),\bg\rangle_{\fB} - \llangle\cI(\bff),\cI(\bg)\rrangle\\
    &=\langle\bff,\bg\rangle_{\fB}- \langle\bff,\bg\rangle_{\fB} = 0,
\end{align}
where we have used that \(\bP_{\fX_{\div}}\) is an orthogonal projection respecting \(\llangle\cdot,\cdot\rrangle\), \(\bP_{\fB_{\div}}\) is an orthogonal projection with respect to \(\langle\cdot,\cdot\rangle_{\fB}\), and \(\langle\cdot,\cdot\rangle_{\fB} = \llangle\cI(\cdot),\cI(\cdot)\rrangle\).
Therefore, \(\cI\bP_{\fB_{\div}}(\bff) - \bP_{\fX_{\div}}\cI(\bff)\) is orthogonal to \(\cI(\fB_{\div})\).
\qed

\subsection{Proof of \thmref{thm:IDaggerIsLeftInverse}}
\label{app:IDaggerIsLeftInverse}
In \eqref{eq:ArgminProblemForIDagger} if \(\vec u_{\sfp}\) is a sampling of a vector field \(\vec u = \cI(\bff)\) in \(\cI(\fB)\) at the particles, then the minimization problem \(\min_{\bff'\in\fB}\sum_{\sfp\in\cP}|\cI(\bff')_{\vec x_\sfp}-\cI(\bff)_{\vec x_\sfp}|^2\)
 can attain zero by \(\bff' = \bff\) uniquely, assuming that the minimization is unique guaranteed by sufficiently many samples.  Therefore \(\hat\cI^+(\vec x,\cI(\bff))\) recovers \(\bff\).
 \qed

\subsection{Proof of \propref{prop:DualSpaceOfDivFree}}
\label{app:DualSpaceOfDivFree}
Represent \(\div\) as a map to the space \(\Omega^n(W)\) of measures \(\fX(W)\xrightarrow{\div}\Omega^n(W)\), by \(\div\vec v = \LD_{\vec v} (d\mu) = d\ip_{\vec v}(d\mu)\in\Omega^n(W)\) where \(d\mu\) is the volume form.
Then 
\(0\to \fX_{\div}(W)\hookrightarrow \fX(W)\xrightarrow{\div}\Omega^n(W)\) is an exact sequence.  Take the dual of this sequence yields another exact sequence \(C^0(W)\xrightarrow{d}\Omega^1(W)\to\fX_{\div}^*(W)\to 0\).  One can check that \(\pm d\) is the adjoint of \(\div\). By the first-isomorphism theorem we conclude \(\fX_{\div}^*(W) = \Omega^1(W)/d\Omega^0(W)\).
\qed

\subsection{Proof of \propref{prop:HamiltonianFlowOnXDivStar}}
\label{app:HamiltonianFlowOnXDivStar}
The (minus) Lie--Poisson bracket \eqref{eq:LiePoissonBracket} on \(\fX_{\div}^*(W)\) is given by that for each functionals \(H,G\colon \fX_{\div}^*(W)\to\RR\) and \([\eta]\in\fX_{\div}^*(W)\)
\begin{align}
    \{H,G\}_{[\eta]} &= 
    \textstyle
    -\left\llangle [\eta]\,\middle|\,\left[ {\delta H\over\delta [\eta]}, {\delta G\over\delta[\eta]} \right]
 \right\rrangle \\
 &=
 \textstyle
 -\int_W \left\langle \eta\, \middle|\,
 \left[ {\delta H\over\delta [\eta]}, {\delta G\over\delta[\eta]} \right]
 \right\rangle \, d\mu\\
 &=
 \textstyle
 \underbrace{-\int_W \LD_{\delta H\over\delta [\eta]}\left\langle \eta \,\middle|\, {\delta G\over\delta[\eta]} \right\rangle\, d\mu}_{=\int_W\langle \eta|{\delta G\over\delta[\eta]}\rangle\, \LD_{\delta H\over\delta [\eta]}(d\mu) = 0}
 +\int_W \left\langle \LD_{\delta H\over\delta [\eta]}\eta\,\middle|\,{\delta G\over\delta[\eta]}\right\rangle\, d\mu.
\end{align}
Therefore the symplectic gradient (\defref{def:sgradPoisson}) \(\sgrad H\) of \(H\) is
\begin{align}
    \sgrad H|_{[\eta]} = [\LD_{\delta H\over\delta[\eta]}] = \LD_{\delta H\over\delta[\eta]}[\eta].
\end{align}
Hence the Hamiltonian dynamical system (\defref{def:HamiltonianDynamicalSystem}) \({\partial \over\partial t}[\eta] = -\sgrad H|_{[\eta]}\) is given by
\({\partial\over \partial t}[\eta] + \LD_{\frac{\delta H}{\delta[\eta]}}[\eta] = [0].\)
\qed

\subsection{Proof of \propref{prop:CoadjointOrbitOnCirculations}}
\label{app:CoadjointOrbitOnCirculations}
Based on the calculation in the proof of \propref{prop:HamiltonianFlowOnXDivStar} ( \appref{app:HamiltonianFlowOnXDivStar}), we have that the coadjoint action \(\ad_{\vec v}^\adjoint\colon \fX_{\div}^*(W)\to\fX_{\div}^*(W)\) for each \(\vec v\in\fX_{\div}(W)\) is given by
\begin{align}
    \ad_{\vec v}^\adjoint[\eta] = -\LD_{\vec v}[\eta].
\end{align}
We want to show, by \defref{def:CoadjointOrbits}, that (i) the resulting \([\eta_1]\) that solves \({\partial\over\partial t}[\eta_t] + \LD_{\vec v(t)}[\eta_t] = [0]\) with given \([\eta_0]\) and \(\vec v\colon[0,1]\in\fX_{\div}(W)\) can be expressed by \([\eta_0]=[\varphi^*\eta_1]\) for some \(\varphi\colon W\to W\) isotopic to the identity map, and (ii) the converse statement.

For (i), consider \(\varphi_t\colon W\to W\), \({\partial\over\partial t}\varphi_t = \vec v(t)\circ\varphi_t\), \(\varphi_0 = \id_W\), be the flow map integrated by the given time-dependent \(\vec v(t)\) from \(t=0\) to \(t=1\).  Then \([\eta_t]\coloneqq [(\varphi_t^{-1})^*\eta_0]\) is the solution to \({\partial\over\partial t}[\eta_t] + \LD_{\vec v(t)}[\eta_t] = [0]\).  Therefore, at \(t=1\) we have \([\eta_1]=[(\varphi^{-1})^*\eta_0]\), which implies \([\eta_0]=[\varphi^*\eta_1]\).

Conversely (ii), suppose \(\varphi\colon W\to W\) is isotopic to \(\id_W\).  Then there exists a smooth family of maps \(\varphi_t \colon W\to W\) so that \(\varphi_0 = \id_W\) and \(\varphi_1 = \id_W\).  Let \(\vec v(t)\) be the family of vector fields so that \({\partial\over\partial t}\varphi_t = \vec v(t)\circ\varphi_t\).  Then \([\eta_t]\coloneqq [(\varphi_t^{-1})^*\eta_0]\), with \([\eta_0]=[\varphi^*\eta_1]\) as a result, is the solution to \({\partial\over\partial t}[\eta_t] + \LD_{\vec v(t)}[\eta_t] = [0]\).  Therefore \([\eta_0]\) can reach \([\eta_1]\) by a continuous sequence of coadjoint action, and thus they lie on the same coadjoint orbit.
\qed

\subsection{Proof of \propref{prop:EulerEquationHamiltonian}}
\label{app:EulerEquationHamiltonian}
The variation of \eqref{eq:H_Euler} is given by
\begin{align}
\label{eq:DerivativeOfH_Euler}
    {\delta H_{\rm Euler}\over\delta[\eta]} = \sharp_{\fX_{\div}}[\eta]\eqqcolon\vec v\in \fX_{\div}(W)
\end{align}
which is the pressure projected \(\eta^\sharp\) (\cf\@ \eqref{eq:SharpXDiv}).
Substituting \eqref{eq:DerivativeOfH_Euler} to \eqref{eq:HamiltonianFlowAsLieAdvection} yields the Euler equation \eqref{eq:EulerEquationCovector} or \eqref{eq:ImpulseEquation}.
\qed

\subsection{Proof of \thmref{thm:IDiscreteEulerFlow}}
\label{app:IDiscreteEulerFlow}
The Hamiltonian \eqref{eq:DiscreteHamiltonianIT} (with \eqref{eq:DiscreteHamiltonian}) is given by 
\begin{align}
    (H_{\rm D}\circ \cI^\adjoint)([\eta]) &= {1\over 2}\langle \sharp_{\fB_{\div}}\cI^\adjoint[\eta],\sharp_{\fB_{\div}}\cI^\adjoint[\eta]\rangle_{\fB} \\
    &=
    {1\over 2}(\bP_{\fB_{\div}}\star^{-1}\cI^\adjoint[\eta])^\intercal
    \star 
    (\bP_{\fB_{\div}}\star^{-1}\cI^\adjoint[\eta]).
\end{align}
Its variation is given by
\begin{align}
    {\delta (H_{\rm D}\circ\cI^\adjoint)\over \delta [\eta]} = 
    \cI\star^{-1}\bP_{\fB_{\div}}^\intercal\star \bP_{\fB_{\div}}\star^{-1}\cI^\adjoint[\eta].
\end{align}
Note that since \(\bP_{\fB_{\div}}\) is an orthogonal projection with respect to the inner product structure \(\star\), we have the self-adjoint property \(\star^{-1}\bP_{\fB_{\div}}^\intercal\star  = \bP_{\fB_{\div}}\).  Then by idempotence of projection \(\bP_{\fB_{\div}}^2 = \bP_{\fB_{\div}}\), we obtain
\begin{align}
\label{eq:VariationOfDiscreteHamiltonian}
    {\delta (H_{\rm D}\circ\cI^\adjoint)\over \delta [\eta]} = 
    \cI\bP_{\fB_{\div}}\star^{-1}\cI^\adjoint[\eta].
\end{align}
Substituting \eqref{eq:VariationOfDiscreteHamiltonian} to \eqref{eq:HamiltonianFlowAsLieAdvection} we obtain \eqref{eq:IDiscreteEulerEquation}.
\qed

\subsection{\thmref{thm:SigmaEmulates}}
\label{app:SigmaEmulates}
This is a direct application of \propref{prop:PoissonMapHamiltonianSystem} on \(J_{\adv}\colon\Sigma\to\fX_{\div}^*(W)\) being an (anti-)Poisson map.

\subsection{Proof of \thmref{thm:IDiscreteEulerEquationSigma}}
\label{app:IDiscreteEulerEquationSigma}
By \eqref{eq:MomentumMapForSymplectomorphicAction}, with \eqref{eq:advLieAlgHom} \(\adv\colon \fX_{\div}(W) \to \ham(\Sigma)\subset\fX(\Sigma)\) being the infinitesimal group action, we have
\begin{align}
\label{eq:JadvSigma}
    -d\llangle J_{\adv}(\cdot) | \vec w\rrangle = \ip_{\adv_{\vec w}(s)}\sigma\quad \forall\,\vec w\in \fX_{\div}(W).
\end{align}
Note that we may express the left-hand side as \(-\llangle dJ_{\adv}|\vec w\rrangle\).
Now let us take the derivative of the Hamiltonian \( H_\Sigma = H_{\rm D}\circ\cI^\adjoint \circ J_{\adv}\):
\begin{align}
    dH_\Sigma &= \textstyle{\delta (H_{\rm D}\circ \cI^\adjoint)\over\delta[\eta]}\big\vert_{[\eta] = J_{\adv}(s)}dJ_{\adv}|_s\\
    &=\llangle  dJ_{\adv} | \vec v\rrangle,\quad\vec v = \cI\bP_{\fB_{\div}}\star^{-1}\cI^\adjoint J_{\adv}(s)
\end{align}
using \eqref{eq:VariationOfDiscreteHamiltonian}. 
Now apply \eqref{eq:JadvSigma}, which yields
\begin{align}
    dH_\Sigma = -\ip_{\adv_{\vec v} s}\sigma.
\end{align}
Therefore by \eqref{eq:sgradForSymplectic} we have 
\begin{align}
    \sgrad H_\Sigma|_s = -\adv_{\vec v}s.
\end{align}
Hence the Hamiltonian flow (\defref{def:HamiltonianDynamicalSystem}) for \(s\) is given by \eqref{eq:IDiscreteEulerEquationSigma}. 
\qed

\subsection{Proofs of \lemref{lem:IDagger1} and \lemref{lem:IDagger2}}
\label{app:IDagger}

The minimizer \(\bff\in\fB\) for \(\int_W|\cI(\bff) - \eta^\sharp|^2\, d\mu\) satisfies the optimality condition \(\llangle \cI(\bff) - \eta^\sharp,\cI(\mathring\bff)\rrangle = 0\) for all \(\mathring\bff\in\fB\).
Since \(\cI\colon (\fB,\langle\cdot,\cdot\rangle_{\fB})\to (\fX(W),\llangle\cdot,\cdot\rrangle)\) is an isometry, and \(\sharp\colon \Omega^1(W)\to\fX(W)\) respects the \(L^2\) product \(\llangle \cdot,\cdot \rrangle\), the condition becomes that for all \(\mathring\bff\in\fB\)
\begin{align}
    0&=\langle\bff,\mathring\bff\rangle_{\fB} - \llangle \eta|\cI(\mathring\bff)\rrangle\\
    &=\mathring\bff^\intercal \star\bff - \langle \cI^\adjoint\eta|\mathring\bff\rangle \\
    &= \mathring\bff^\intercal(\star\bff - \cI^\adjoint\eta).
\end{align}
Therefore \(\bff = \star^{-1}\cI^\adjoint\eta\), showing \eqref{eq:IDaggerContinuum1}.

The proof for \eqref{eq:IDaggerContinuum} is the same as the one for \eqref{eq:IDaggerContinuum1} with the substitution \(\eta = J_{\adv}(\vec x,\vec u)\).
\qed

\subsection{Proof of \thmref{thm:EnergyPreservingSigma}}
\label{app:EnergyPreservingSigma}
The proof has two parts.  First we show that states \(J_{\adv}(s^{(\sfn)})\) stay on the same coadjoint orbit.  
Second, we show that \(\cE^{(\sfn)} = \cE(s^{(\sfn)})\coloneqq H_{\rm D}\circ\cI^\adjoint\circ J_{\adv}(s^{(\sfn)})\) is \(\sfn\)-invariant. 

\subsubsection{Coadjoint Orbit Conservation}
\label{sec:CoadjointOrbitConservationProof}
To show the coadjoint orbit conservation we need the following lemma.
\begin{lemma}
    For each \(\vec v_1,\vec v_2\in \fX_{\div}(W)\) and \(s\in\Sigma\), we have
    \begin{align}\label{eq:FormulaWithAdvSigmaJLieBracket}
        \ip_{\adv_{\vec v_2}s} \ip_{\adv_{\vec v_1}s} \sigma = \llangle J_{\adv}(s)|[\vec v_1,\vec v_2]\rrangle.
    \end{align}
\end{lemma}
\begin{proof}
    The formula \eqref{eq:FormulaWithAdvSigmaJLieBracket} we want to show is an equality between functions defined over \(s\in\Sigma\).  That is we want to show the functional relation \(\ip_{\adv(\vec v_2)}\ip_{\adv(\vec v_1)}\sigma = \llangle J_{\adv}|[\vec v_1,\vec v_2]\rrangle\), where \(\adv\colon \fX_{\div}(W)\to \fX(\Sigma)\) is a Lie algebra homomorphism that results in a vector field.
    To show this functional relation, it suffices to show that their differentials are equal: We claim that \(d\ip_{\adv(\vec v_2)}\ip_{\adv(\vec v_1)}\sigma = d\llangle J_{\adv}|[\vec v_1,\vec v_2]\rrangle\).
    Using \eqref{eq:MomentumMapForSymplectomorphicAction} and that \(\adv\) is a Lie algebra homomorphism, we have
    \begin{align}
        &d\llangle J_{\adv}|[\vec v_1,\vec v_2]\rrangle = -\ip_{\adv([\vec v_1,\vec v_2])}\sigma\\
        &=-\ip_{[\adv(\vec v_1),\adv(\vec v_2)]}\sigma \\
        &=-\LD_{\adv(\vec v_1)}\ip_{\adv(\vec v_2)}\sigma + \ip_{\adv(\vec v_2)}\LD_{\adv(\vec v_1)}\sigma\\
        &=-d\ip_{\adv(\vec v_1)}\ip_{\adv(\vec v_2)}\sigma,
    \end{align}
    where the last equality uses the Cartan formula, \(d\sigma = 0\), and that \(\adv(\vec v_1), \adv(\vec v_2)\) are symplectomorphic \(\LD_{\adv(\vec v_1)}\sigma =\LD_{\adv(\vec v_2)}\sigma = 0\).
\end{proof}

Now we are ready to show that \eqref{eq:IDiscreteEulerEquationSigma} preserves the coadjoint orbit.
The update \eqref{eq:IDiscreteEulerEquationA} \(s^{(\sfn+1)} = \Adv_{\varphi(\Deltait t)}s^{(\sfn)}\) is the solution to the following ODE for a fixed \(\vec v\)
\begin{align}\label{eq:s_t}
\textstyle
    {\partial\over\partial t}s_t = \adv_{\vec v}s_t,\quad
    s_0 = s^{(\sfn)},\quad s^{(\sfn+1)} = s_1.
\end{align}
We show that during this update, \([\eta_t]\coloneqq J_{\adv}(s_t)\) evolves tangential to the coadjoint orbit.
Taking the time derivative, we get
\begin{align}
    \textstyle{\partial\over\partial t}[\eta_t] &=
    \textstyle{\partial\over\partial t}J_{\adv}(s_t)
    = dJ_{\adv}|_{s_t}\llbracket {\partial s_t\over\partial t}\rrbracket = dJ_{\adv}|_{s_t}\llbracket\adv_{\vec v}s_t\rrbracket.
\end{align}
Take an arbitrary test vector field \(\vec\xi\in\fX_{\div}(W)\) to pair with this formula
\begin{align}
    \llangle \textstyle{\partial\over\partial t}[\eta_t]|\vec\xi\rrangle &= \left\llangle 
    dJ_{\adv}|_{s_t}\llbracket \adv_{\vec v}s_t \rrbracket\middle|\vec \xi\right\rrangle\\
    &=\ip_{\adv_{\vec v}s_t}d\llangle J_{\adv}|\vec \xi\rrangle \\
    &\overset{\eqref{eq:MomentumMapForSymplectomorphicAction}}{=}
    -\ip_{\adv_{\vec v}s_t}\ip_{\adv_{\vec \xi}s_t}\sigma\\
    &\overset{\eqref{eq:FormulaWithAdvSigmaJLieBracket}}{=}\llangle J_{\adv}(s_t)|[\vec v,\vec \xi]\rrangle = \llangle[\eta_t]|[\vec v,\vec\xi]\rrangle\\
    &=\llangle-\LD_{\vec v}[\eta_t]|\vec\xi\rrangle.
\end{align}
For this to hold for all \(\vec\xi\in\fX_{\div}(W)\) we conclude that \({\partial\over\partial t}[\eta_t] + \LD_{\vec v}[\eta_t] = [0]\).  Therefore \([\eta_t] = J_{\adv}(s_t)\) evolves by coadjoint action and therefore stays on the same coadjoint orbit.\qed

\subsubsection{Energy Conservation}
Our approach is similar to \citet{Engo:2001:NILP}. Continuing the setup of \secref{sec:CoadjointOrbitConservationProof}, we let \(s_t\) be evolved by \eqref{eq:s_t} that connects \(s^{(\sfn)}\) and \(s^{(\sfn+1)}\).  Let \([\eta_t] = J_{\adv}(s_t)\).  By \secref{sec:CoadjointOrbitConservationProof} we know \({\partial\over\partial t}[\eta_t] + \LD_{\vec v}[\eta_t] = [0]\), where \(\vec v\) is fixed.
Then \eqref{eq:IDiscreteEulerEquation} simplifies to 
\begin{subequations}\label{eq:IDiscreteEulerEta}
    \begin{numcases}{}
    \label{eq:IDiscreteEulerEtaA}
        \textstyle
        [\eta_1] = \exp(-\Deltait t\LD_{\vec v})[\eta_0]\\
    \label{eq:IDiscreteEulerEtaB}
        \textstyle
        \vec v = {1\over 2}\cI\bP_{\fB_{\div}}\cI^+([\eta_1]+[\eta_0]).
    \end{numcases}
\end{subequations}
The energies at frame \(\sfn\) and \(\sfn+1\) are respectively 
\begin{align}
    \cE^{(\sfn)} &= \textstyle {1\over 2}|\bP_{\fB_{\div}}\cI^+[\eta_0]|^2_{\fB} = {1\over 2}\|\cI\bP_{\fB_{\div}}\cI^+[\eta_0]\|^2\\
    \cE^{(\sfn+1)} &= \textstyle {1\over 2}|\bP_{\fB_{\div}}\cI^+[\eta_1]|^2_{\fB} = {1\over 2}\|\cI\bP_{\fB_{\div}}\cI^+[\eta_1]\|^2.
\end{align}
Our goal is to show that their difference is zero.  Their difference, by the difference-of-squares formula, is given by
\begin{align}
\nonumber
    \cE^{(\sfn+1)}-\cE^{(\sfn)} &= \textstyle {1\over 2}\llangle \cI\bP_{\fB_{\div}}\cI^+([\eta_1]+[\eta_0]), \cI\bP_{\fB_{\div}}\cI^+([\eta_1]-[\eta_0])\rrangle\\
    &\overset{\eqref{eq:IDiscreteEulerEtaB}}{=}\textstyle\llangle\vec v,\cI\bP_{\fB_{\div}}\cI^+([\eta_1]-[\eta_0])\rrangle.
\end{align}
Now, by \thmref{thm:PressureProjectionIsExact} and \(\vec v\in\cI(\fB_{\div})\), we can replace \(\cI\bP_{\fB_{\div}}\) by \(\bP_{\fX_{\div}}\cI\), giving us
\begin{align}
    \cE^{(\sfn+1)}-\cE^{(\sfn)} &= \llangle\vec v,\bP_{\fX_{\div}}\cI\cI^+([\eta_1]-[\eta_0])\rrangle\\
    &=\llangle\vec v,\sharp([\eta_1]-[\eta_0])\rrangle = \llangle\vec v|[\eta_1]-[\eta_0]\rrangle
\end{align}
The second equality uses the fact that \(\bP_{\fX_{\div}}\) and \(\cI\cI^+\) are \(L^2\) orthogonal projection operators respectively to \(\fX_{\div}(W)\subset\fX(W)\) and \(\cI(\fB)\subset\fX(W)\), and  \(\vec v\) is an element of both subspaces.

We are left to show \(\llangle\vec v|[\eta_1]-[\eta_0]\rrangle = 0\).
By \eqref{eq:IDiscreteEulerEtaA} we have
\begin{align}
    &[\eta_1]-[\eta_0] = (\exp(-\Deltait t\LD_{\vec v})-1)[\eta_0]\\
    &= \textstyle \left((1-\Deltait t\LD_{\vec v}+{\Deltait t^2\over 2!}\LD_{\vec v}^2 - {\Deltait t^3\over 3!}\LD_{\vec v}^3+\cdots) - 1\right)[\eta_0]\\
    &=
    \textstyle
    \left(-\Deltait t \LD_{\vec v} + {\Deltait t^2\over 2!}\LD_{\vec v}^2-{\Deltait t^3\over 3!}\LD_{\vec v}^3+\cdots\right)[\eta_0]\\
    &=
    \textstyle
    \LD_{\vec v}\underbrace{\left(-\Deltait t + {\Deltait t^2\over 2!}\LD_{\vec v}-{\Deltait t^3\over 3!}\LD_{\vec v}^2+\cdots\right)[\eta_0]}_{\eqqcolon [\tilde\eta]}.
\end{align}
Therefore,
\begin{align}
    \llangle \vec v|[\eta_1]-[\eta_0]\rrangle = \llangle \vec v|\LD_{\vec v}[\tilde\eta]\rrangle = -\llangle \underbrace{[\vec v,\vec v]}_{=0}|[\tilde\eta]\rrangle =0.
\end{align}
\qed

\subsection{Proof of \thmref{thm:COFLIPCoadjointPreserving}}
\label{app:COFLIPCoadjointPreserving}
The advection in \eqref{eq:COFLIP-ODEa} is a coadjoint action, and therefore it preserves the coadjoint orbit.  The advection in \eqref{eq:ImplicitMidpointA} preserves the coadjoint orbit by the same argument as \secref{sec:CoadjointOrbitConservationProof}.
\qed

\subsection{Proof of \thmref{thm:EnergyCorrection}}
\label{app:EnergyCorrection}
By construction \((\bff^{(\sfn+1)}-\bff^{(\sfn)})\bot (\bff^{(\sfn+1)}+\bff^{(\sfn)})\).  Therefore \(|\bff^{(\sfn+1)}|_{\fB}^2 - |\bff^{(\sfn)}|_{\fB}^2=0\).
\qed

\end{document}